\documentclass[12pt]{article}
\usepackage{amsfonts}
\usepackage{amsmath}
\usepackage{amssymb}
\usepackage{graphicx}
\usepackage{color}
\usepackage{mathtools}
\usepackage{cite}
\usepackage[colorlinks=true,linkcolor=blue,citecolor=red,plainpages=false,pdfpagelabels]
{hyperref}
\usepackage{tikz}
\usepackage[margin=1.25in]{geometry}
\usepackage{times}
\usepackage{subcaption}
\providecommand{\U}[1]{\protect\rule{.1in}{.1in}}
\newtheorem{theorem}{Theorem}

\newtheorem{conclusion}[theorem]{Conclusion}

\newtheorem{corollary}[theorem]{Corollary}

\newtheorem{definition}[theorem]{Definition}

\newtheorem{lemma}[theorem]{Lemma}

\newtheorem{proposition}[theorem]{Proposition}
\newtheorem{remark}[theorem]{Remark}

\newenvironment{proof}[1][Proof]{\noindent\textbf{#1.} }{\ \rule{0.5em}{0.5em}}

\allowdisplaybreaks
\numberwithin{equation}{section}
\begin{document}

\title{Geometric distinguishability measures limit quantum channel estimation and discrimination}
\author{Vishal Katariya\thanks{Hearne Institute for Theoretical Physics, Department of
Physics and Astronomy, and Center for Computation and Technology, Louisiana
State University, Baton Rouge, Louisiana 70803, USA}
\and Mark M. Wilde\footnotemark[1] \thanks{Stanford Institute for Theoretical
Physics, Stanford University, Stanford, California 94305, USA} \thanks{Email: mwilde@lsu.edu}}
\maketitle

\begin{abstract}
Quantum channel estimation and discrimination are fundamentally related information processing tasks of interest in quantum information science. In this paper, we analyze these tasks by employing the right logarithmic derivative Fisher information and the geometric R\'enyi relative entropy, respectively, and we also identify connections between these distinguishability measures. A key result of our paper is that a chain-rule property holds for the right logarithmic
derivative Fisher information and the geometric R\'enyi relative entropy for the interval $\alpha\in(0,1) $ of the R\'enyi parameter $\alpha$. 
 In channel estimation, these results imply a condition for the unattainability of Heisenberg scaling, while in channel discrimination, they lead to improved bounds on error rates in the Chernoff and Hoeffding error exponent settings.
 More generally, we introduce
the amortized quantum Fisher information as a conceptual framework for analyzing general
sequential protocols that estimate a parameter encoded in a quantum channel. We then use this framework, beyond the aforementioned application, to show that Heisenberg scaling is not possible when a parameter is encoded in a classical--quantum channel.
We then identify a number of other
conceptual and technical connections between the tasks of estimation and
discrimination and the distinguishability measures involved in analyzing each. As part of
this work, we present a detailed overview of the geometric R\'enyi relative
entropy of quantum states and channels, as well as its properties, which may
be of independent interest.
\end{abstract}
\tableofcontents

\section{Introduction}

Quantum channel discrimination and  estimation are fundamental tasks in quantum information science. Channel discrimination refers to the task of distinguishing 
two (or more) quantum channels, while quantum channel estimation is a
generalization of this scenario. Instead of determining an unknown channel selected from a 
finite set, the goal of channel estimation is to estimate a
particular member chosen from a continuously parameterized set of quantum channels. The simplest
channel discrimination task consists of discriminating two channels selected from a set
$\{\mathcal{N}_{\theta}\}_{\theta \in \{1,2\}}$, whereas the simplest estimation task
consists of identifying a particular member of a continuously parameterized set
of channels $\{ \mathcal{N}_{\theta} \}_{\theta \in \Theta}$, where $\Theta \subseteq  \mathbb{R}$. Theoretical studies in both the discrimination and estimation of quantum channels have been applied  in a variety of settings, including quantum illumination \cite{Lloyd2008}, phase estimation using optical interferometry \cite{Braunstein1992, Dowling1998, Demkowicz-Dobrzanski2015}, and gravitational wave detection \cite{Caves1981, Yurke1986, Berry2000, Demkowicz-Dobrzanksi2013}. 

In classical parameter estimation, the unknown parameter $\theta$ is encoded
in a probability distribution $p_{\theta}(x)$ with associated random variable $X$. One tries to guess its value from a realization $x$ of $X$ by calculating an estimator $\hat{\theta}(x)$ of the true value $\theta$.  The most common measure of performance employed in estimation theory is the mean-squared error, defined as $\mathbb{E}[(\hat{\theta}(X) - \theta)^2]$. For an unbiased estimator satisfying $\mathbb{E}[\hat{\theta}(X) ] = \theta$, the mean-squared error is equal to $\text{Var}(\hat{\theta}(X))$, and one of the fundamental results of classical estimation theory is the Cramer-Rao lower bound (CRB) on the mean-squared error of an unbiased estimator:
\begin{equation}
\text{Var}(\hat{\theta}(X)) \geq
\frac{1}{I_{F}(\theta; \{p_{\theta}\}_\theta)}.
\end{equation}
 The lower bound features  
the Fisher information, defined as the following function of the probability distribution family $\{p_{\theta}\}_\theta$:
\begin{equation}
I_{F}
(\theta; \{p_{\theta}\}_\theta) = \mathbb{E}[(\partial_\theta \ln p_\theta(X))^2]
= \int dx\, p_\theta(x) (\partial_\theta \ln p_\theta(x))^2,
\label{eq:classical-FI-intro}
\end{equation}
where we employ the shorthand $\partial_\theta(\cdot) \equiv \frac{\partial}{\partial \theta}(\cdot)$.
Recalling the interpretation of $-\ln p_\theta(x)$ as the surprisal of the realization $x$, it follows that $\partial_\theta [ - \ln p_\theta(x)]$ is the rate of change of the surprisal with the parameter $\theta$ (surprisal rate). After noticing that the expected surprisal rate vanishes, by applying the  conservation of probability, it follows that the Fisher information is equal to the variance of the surprisal rate, thus characterizing its fluctuations  \cite{PhysRevE.98.032106,NGCG20}. If one generates $n$ independent samples $x^n \equiv x_1, \ldots, x_n$ of $p_{\theta}(x)$, described by the random sequence $X^n \equiv X_1, \ldots, X_n$, and forms an unbiased estimator $\hat{\theta}(x^n)$, then the Fisher information increases linearly with $n$ and the CRB becomes as follows:
\begin{equation}
\text{Var}(\hat{\theta}(X^n)) \geq
\frac{1}{n I_{F}(\theta; \{p_{\theta}\}_\theta)},
\end{equation}
which is how it is commonly employed in applications. 

In quantum estimation, the parameter $\theta$ is encoded in a quantum state
$\rho_{\theta}$ or a quantum channel $\mathcal{N}_{\theta}$, and generally, it is possible to attain better-than-classical scaling in error by
using quantum resources such as entanglement and collective measurements. When  formulating a quantum generalization of the Cramer--Rao bound and Fisher information, it is necessary to find a quantum generalization of the logarithmic derivative $\partial_\theta \ln p_\theta(x)$ in \eqref{eq:classical-FI-intro}. However, the noncommutative nature of quantum mechanics yields an infinite number
of logarithmic derivatives of $\rho_{\theta}$. To demonstrate this point, consider that we can define a family of parameterized logarithmic derivative operators $\{ D^{(p)}_\theta\}_{p}$ with $p \in [0, 1/2]$ as follows: $\partial_\theta \rho_\theta \coloneqq p D^{(p)}_\theta \rho_\theta + (1-p) \rho_\theta D^{(p)}_\theta$. Each $D^{(p)}_\theta$ collapses to the scalar logarithmic derivative in the classical case. The two most studied logarithmic derivatives are specific instances of $D^{(p)}_\theta$: the symmetric logarithmic derivative
(SLD) corresponding to $p = 1/2$ \cite{Hel67} and the right logarithmic derivative (RLD) corresponding to $p=0$ \cite{YL73}. At least two quantum Fisher informations can be defined based on these specific possibilities. By far, the SLD Fisher information has been
the most studied, on account of it providing the tightest quantum Cramer-Rao bound
(QCRB) in single parameter estimation of quantum states, while also being achievable when many copies of the state are available. 
The recent review \cite{Sidhu2019} provides an in-depth study of these and other notions
in quantum estimation. 

In this paper, we focus on the task of estimating a single unknown parameter $\theta$ encoded in a
quantum channel $\mathcal{N}_{\theta}$. This task has been
studied extensively in prior work \cite{Sasaki2002, Fujiwara_2003,
Fujiwara2004, Ji2008, Fujiwara2008, Mat10, Hayashi2011,
Demkowicz-Dobrzanski2012, Kolodynski2013, Demkowicz-Dobrzanski2014,
Sekatski2017, Demkowicz-Dobrzanski2017, Zhou2018, Zhou2019, Zhou2019a,YCH20}, and the most general setting for this problem is known as the sequential setting \cite{Giovannetti2006,PhysRevLett.98.090501,Demkowicz-Dobrzanski2014,Yuan2017}, in which one can interact with the channel $n$ independent times in the most general way allowed by quantum mechanics. Heisenberg scaling refers to the quantum Fisher information scaling as $n^{2}$,
where $n$ is the number of channel uses, or as $t^{2}$, where $t$ is the total
probing time. One
fundamental question for channel estimation is whether Heisenberg scaling can
be achieved when estimating a particular quantum channel. % The Heisenberg limit is the ultimate limit to precision
%attainable by quantum-enhanced estimation \cite{Giovannetti2006}. 

Our approach to the channel  estimation problem involves defining the amortized Fisher information of a family of channels, which is in the same
spirit as the amortized channel divergence introduced in \cite{Berta2018c}. The amortized Fisher information
provides a compact mathematical framework for studying the difference between
sequential and parallel estimation strategies, just as the amortized channel divergence does for channel discrimination \cite{Berta2018c}. Specifically, we prove that the amortized Fisher information  is a generic bound for all channel estimation protocols (called the ``meta-converse'' for channel estimation).

One key result of our paper is a chain rule for the
RLD Fisher information, with a consequence being that
amortization does not increase the RLD Fisher information of quantum channels. Importantly, when combining this result with the aforementioned meta-converse, it follows that Heisenberg scaling is unattainable for a channel family if its RLD Fisher information is finite.  This latter result generalizes a finding of \cite{Hayashi2011} beyond parallel strategies for channel estimation to the more general sequential strategies. Let us also note that evaluating the finiteness condition for the RLD Fisher information is a simpler task than evaluating the RLD (or SLD) Fisher information itself.

Turning to the related task of channel discrimination, a key tool that we employ for this purpose is the geometric R\'enyi relative entropy.
This distinguishability measure  has its roots in \cite{PR98}, and it was further
developed in \cite{Mat13,Matsumoto2018} (see also \cite{T15book,HM17}). It was given the name ``geometric R\'enyi relative entropy'' in \cite{Fang2019a} because it is a function
of the matrix geometric mean of its arguments. It was also  used to great
effect in \cite{Fang2019a} to bound quantum channel capacities and error rates of
channel discrimination in the  asymmetric setting. We continue to use it in this vein, in particular, by
improving upper bounds on error rates of channel discrimination in the symmetric
setting (specifically, the Chernoff and Hoeffding error exponents). Due to the chain rule of the
geometric R\'enyi relative entropy (and hence amortization collapse of the related channel function), the bounds that we report here are both single-letter and
efficiently computable via semi-definite programs. Our bounds also improve upon those found recently in  \cite{Berta2018c, CE18}.

As mentioned earlier, channel estimation is a
 generalization of channel discrimination to the case in which the unknown parameter is continuous. We devote the last
section of our paper to bringing out connections between the two tasks. We
observe that the RLD Fisher information arises from taking the limit of the geometric R\'enyi relative entropy
 of two infinitesimally close elements of a family of channels. Therefore, in this sense, we see that
the QCRB arising from the RLD Fisher information has the geometric R\'enyi relative entropy
underlying it. Further, we connect properties of the SLD and RLD Fisher informations to the
corresponding properties of their underlying distance measures (fidelity and
geometric R\'enyi relative entropy, respectively).

Our paper is structured as follows. First, we present a more detailed, yet brief overview of our results in Section~\ref{sec:summary-results}. In Section~\ref{sec:q-info-preliminaries}, we review some notation and mathematical identities
used throughout our paper. In Section~\ref{sec:setting-q-ch-est-disc}, we present the
information-processing tasks of channel estimation and discrimination.
Section~\ref{sec:bounds-on-estimation} contains all of our results regarding bounds
on channel estimation. Section~\ref{sec:bounds-on-disc} introduces
the geometric R\'enyi relative entropy and contains our bounds on channel
discrimination. Section~\ref{sec:connections} brings out connections between
estimation and discrimination, building on our results from the previous two
sections. In Section~\ref{sec:conclusion}, we conclude by summarizing our results
and outlining future work. The appendices of our paper contain many detailed mathematical proofs, as well as a detailed overview of the geometric R\'enyi relative entropy of quantum states and channels (Appendices~\ref{app:geo-ren-props} and \ref{app:geo-renyi-channels-app}).

\section{Summary of Results} \label{sec:summary-results}

Here we summarize the main contributions and results of our paper:

\begin{enumerate}
    \item In Section~\ref{subsec:classical-q-fisher-info}, we provide definitions for the SLD and RLD Fisher informations of quantum state families. These definitions are accompanied by specific conditions that govern the finiteness of the quantities. We also prove that the SLD and RLD Fisher informations are physically consistent, i.e., that the definitions provided are consistent with a limiting procedure in which some additive noise vanishes.
    
    \item In Section~\ref{subsec:generalized-fisher-info}, we define the generalized Fisher information of quantum state and channel families, with the aim of establishing a number of properties that arise solely from data processing. We also provide finiteness conditions for the SLD and RLD Fisher informations of quantum channels, which are helpful for determining whether Heisenberg scaling can occur in channel estimation. In this same section, we also introduce the idea of and define the amortized Fisher information of quantum channel families, as a generalization of the amortized channel divergence introduced in \cite{Berta2018c}. We then establish a meta-converse for all channel estimation protocols, which demonstrates that amortized Fisher information is a generic bound for all such protocols. 
    
    \item In Section~\ref{subsec:qfi-optimization-formulae}, we cast the SLD and RLD Fisher informations as optimization problems. Specifically, we cast the SLD Fisher information of quantum states as a semi-definite program, the SLD Fisher information of quantum channels as a bilinear program, and the RLD Fisher information of both quantum states and channels as a semi-definite program. We also provide a quadratically constrained program for the root SLD Fisher information of quantum states, whose formulation is used to establish the chain rule property of the root SLD Fisher information. We provide duals to our semi-definite programs in all cases.
    
    \item In Section~\ref{subsec:SLD-limit-cq-channels}, we show that sequential estimation strategies provide no advantage over parallel estimation strategies for classical-quantum channel families.
    
    \item In Sections~\ref{sec:SLD-fisher-info-limits} and \ref{sec:RLD-Fish-limits}, we utilize the SLD and RLD Fisher information of quantum channels  to place lower bounds on the error of sequential parameter estimation protocols. We prove chain rule properties for the RLD Fisher information and the root SLD Fisher information, which imply an amortization collapse for these quantities.
    
    \item An important corollary of the amortization collapse of the RLD Fisher information is a condition for the unattainability of Heisenberg scaling. Specifically, we prove that if the RLD Fisher information of a channel family is finite, then Heisenberg scaling is unattainable for it. Thus, we provide an operational consequence of the finiteness condition for the RLD Fisher information of quantum channels.
    
    \item When estimating a single parameter, the RLD Fisher information is never smaller than the SLD Fisher information. We study an example in Section~\ref{subsec:estimating-gadc} regarding the effectiveness of the RLD Fisher information as a performance bound when estimating various parameters encoded in a generalized amplitude damping channel.
    
    \item In Sections~\ref{subsec:geo-rel-ent} and \ref{subsec:geo-rel-ent-properties}, we provide a limit-based formula for the geometric R\'enyi relative entropy, and then we establish consistency of this formula with more explicit formulas for the whole range $\alpha \in (0,1) \cup (1,\infty)$. We review existing and also establish new properties of the geometric R\'enyi relative entropy of quantum states and channels.
    
    \item In the rest of Section~\ref{sec:bounds-on-disc}, we use the geometric R\'enyi relative entropy to improve currently known upper bounds on error rates in quantum channel discrimination. We (a) use the geometric fidelity to place an upper bound on the error exponent in the symmetric Chernoff setting and (b) introduce the Belavkin--Staszewski divergence sphere as an upper bound on the Hoeffding error exponent. We also study a task called ``sequential channel discrimination with repetition'' and establish an upper bound on its Chernoff and Hoeffding error exponents.
    
    \item Finally, in Section~\ref{sec:connections}, we bring out a number of conceptual and technical connections between the tasks of channel estimation and discrimination.
    
\end{enumerate}

\section{Quantum information preliminaries}

\label{sec:q-info-preliminaries}

We begin by recalling some basic facts and identities that appear often in
this paper and more generally in quantum information. For further background,
we refer to the textbooks \cite{H06,H13book,Wat18,Wbook17}.

A quantum state is described by a density operator, which is a positive
semi-definite operator with trace equal to one and often denoted by $\rho$,
$\sigma$, $\tau$, etc. A quantum channel $\mathcal{N}_{A\rightarrow B}$ taking
an input quantum system $A$ to an output quantum system $B$ is described by a
completely positive, trace-preserving map. In this paper, we deal exclusively
with finite-dimensional systems, but it is clear that many of the concepts and
results should generalize to quantum states and channels acting on
separable Hilbert spaces.

Let $|\Gamma\rangle_{RA}$ denote the unnormalized maximally entangled vector:
\begin{equation}
|\Gamma\rangle_{RA}:=\sum_{i}|i\rangle_{R}|i\rangle_{A},
\end{equation}
where $\{|i\rangle_{R}\}_{i}$ and $\{|i\rangle_{A}\}_{i}$ are orthonormal
bases for the isomorphic Hilbert spaces $\mathcal{H}_{R}$ and $\mathcal{H}
_{A}$. We repeatedly use the fact that a pure bipartite state $|\psi
\rangle_{RA}$ can be written as $(X_{R}\otimes I_{A})|\Gamma\rangle_{RA}$
where $X_{R}$ is an operator satisfying $\operatorname{Tr}[X_{R}^{\dag}
X_{R}]=1$. For a linear operator~$M$, the following transpose trick identity
holds
\begin{equation}
\left(  I_{R}\otimes M_{A}\right)  |\Gamma\rangle_{RA}=\left(  M_{R}
^{T}\otimes I_{A}\right)  |\Gamma\rangle_{RA}, \label{eq:transpose-trick}
\end{equation}
where $M^{T}$ denotes the transpose of $M$ with respect to the orthonormal
basis $\{|i\rangle_{R}\}_{i}$. For a linear operator $K_{R}$, the following
identity holds
\begin{equation}
\langle\Gamma|_{RA}\left(  K_{R}\otimes I_{A}\right)  |\Gamma\rangle
_{RA}=\operatorname{Tr}[K_{R}]. \label{eq:max-ent-partial-trace}
\end{equation}

The Choi operator $\Gamma_{RB}^{\mathcal{N}}$ of a quantum channel
$\mathcal{N}_{A\rightarrow B}$ is defined as
\begin{equation}
\Gamma_{RB}^{\mathcal{N}}:=\mathcal{N}_{A\rightarrow B}(\Gamma_{RA}),
\end{equation}
where
\begin{equation}
\Gamma_{RA}:=|\Gamma\rangle\!\langle\Gamma|_{RA}.
\end{equation}
The Choi operator is positive semi-definite and satisfies the following
property as a consequence of $\mathcal{N}_{A\rightarrow B}$ being trace
preserving:
\begin{equation}
\operatorname{Tr}_{B}[\Gamma_{RB}^{\mathcal{N}}]=I_{R}.
\end{equation}
The following post-selected teleportation identity \cite{B05}\ allows for
writing the output of a quantum channel $\mathcal{N}_{A\rightarrow B}$\ on an
input quantum state $\rho_{RA}$ in the following way:
\begin{equation}
\mathcal{N}_{A\rightarrow B}(\rho_{RA})=\langle\Gamma|_{AS}\rho_{RA}
\otimes\Gamma_{SB}^{\mathcal{N}}|\Gamma\rangle_{AS}\text{,}
\label{eq:PS-TP-identity}
\end{equation}
where $S$ is a system isomorphic to the channel input system $A$.

\section{Setting of quantum channel parameter estimation and discrimination}

\label{sec:setting-q-ch-est-disc}

We now recall the two related tasks of channel parameter estimation and
discrimination. In the first task, one is interested in estimating an unknown
channel selected from a continuously parameterized family of channels, while
in the latter, the goal is the same but the unknown channel is selected from a
finite set. The metrics used to quantify performance are different and are
explained below. Also, in this paper, we focus exclusively on channel
discrimination of just two quantum channels.

\subsection{Quantum channel parameter estimation}

\label{sec:ch-param-est}Let us now discuss channel parameter estimation in
more detail. Let $\left\{  \mathcal{N}_{A\rightarrow B}^{\theta}\right\}
_{\theta}$ denote a family of quantum channels with input system $A$ and
output system $B$, such that each channel in the family is parameterized by a
single real parameter $\theta\in\Theta\subseteq\mathbb{R}$, where $\Theta$ is
the parameter set. The problem we consider is this: given a particular unknown
channel $\mathcal{N}^{\theta}_{A \to B}$, how well can we estimate $\theta$
when allowed to probe the channel $n$ times? There are various ways that one
can probe the quantum channel $n$ times, but each such procedure results in a
probability distribution $p_{\theta}(x)$ for a final measurement outcome $x$,
with corresponding random variable $X$. This distribution $p_{\theta}(x)$
depends on the unknown parameter $\theta$. Using the measurement outcome $x$,
one formulates an estimate $\hat{\theta}(x)$ of the unknown parameter. An
unbiased estimator satisfies $\mathbb{E}_{p_{\theta}}[\hat{\theta}(X)]=\theta
$. For an unbiased estimator (on which we focus exclusively here), the mean
squared error (MSE) is a commonly considered measure of performance:
\begin{equation}
\text{Var}(\hat{\theta}(X)):=\mathbb{E}[(\hat{\theta}(X)-\theta)^{2}]=\int
dx\ p_{\theta}(x)(\hat{\theta}(x)-\theta)^{2}.
\end{equation}

One major question of interest is to ascertain the optimal scaling of the MSE
with the number $n$ of channel uses. We note that much work has been done
on this topic, with an inexhaustive reference list given by \cite{Sasaki2002,
Fujiwara_2003, Fujiwara2004, Ji2008, Fujiwara2008, Mat10, Hayashi2011,
Demkowicz-Dobrzanski2012, Kolodynski2013, Demkowicz-Dobrzanski2014,
Sekatski2017, Zhou2018, Zhou2019, Zhou2019a}. We also clarify that our approach adopts the frequentist approach to parameter estimation. In general, the MSE and Cramer--Rao bounds may depend on the value of the unknown parameter, in contrast with the more general paradigm of Bayesian parameter estimation \cite{Li2018}. This is alleviated by enforcing the unbiasedness condition.

\begin{figure}[ptb]
\centering
\begin{tikzpicture}[scale=1.8]
	\draw (0.375, 0.75-0.375) -- (0.75,0);
	\draw (0.375, 0.75-0.375) -- (0.75, 1.5-0.75);
	\node at (0.15, 0.75-0.375) {$\rho_{R A_1}$};
	\draw (0.75,0) -- (1,0);
	\draw (1,-0.25) rectangle (1.5,0.25);
	\node at (1.25,0) {$\mathcal{N}^{\theta}$};
	\node at (0.5,0) {\small $A_1$};
	\node at (1.8,0.14) {\small$B_1$};
	\draw (1.5,0) -- (2,0);
	\draw (0.75,1.5-0.75) -- (2,1.5-0.75);
	\draw (2, 0.75+0.25) rectangle (2.5, 0.75-1); %S1
	\draw (2.5, 1.5-0.75) -- (4, 1.5-.75);
	\node at (2.25, 0.75-0.375) {$\mathcal{S}^1$};
	\draw (3,-0.25) rectangle (3.5,0.25);
	\node at (3.25,0) {$\mathcal{N}^{\theta}$};
	\node at (2.75,0.14) {\small$A_2$};
	\node at (3.75,0.14) {\small$B_2$};
	\draw (3.5,0) -- (4,0); %A2
	\draw (2.5,0) -- (3,0); %B2
	\draw (4.5, 0) -- (5, 0); %A3
	\node at (4.75, 0.14) {$A_3$};
	\draw (4, 0.75+0.25) rectangle (4.5, 0.75-1); %S2
	\node at (1.8, 1.5-0.6) {$R_1$};
	\node at (3.75, 1.5-0.6) {$R_2$};
	\node at (4.75, 1.5-0.6) {$R_3$};
	\draw (4.5, 1.5-0.75) -- (5, 1.5-0.75);
	\node at (5.75+1, 1.5-0.6) {$R_n$};
	\node at (4.25, 0.75-0.375) {$\mathcal{S}^2$}; %S2
	\node at (5.25,0.75-0.375) {$\cdots$};
	\draw (4.5+1, 1.5-0.75) -- (6+1, 1.5-0.75);
	\draw (4.5+1,0) -- (5+1,0);
	\node at (4.75+1, 0.14) {$A_n$};
	\node at (5.75+1, 0.14) {$B_n$};
	\draw (5.5+1, 0) -- (6+1, 0);
	\draw (5+1,-0.25) rectangle (5.5+1,0.25);
	\node at (5.25+1, 0) {$\mathcal{N}^{\theta}$};
	\draw (6+1,0.75-1) rectangle (6.75+1,0.75+0.25); %POVM
	\node at (6.375+1,0.75-0.375) {$\Lambda^{\hat{\theta}}$};
	\draw[double distance between line centers=0.2em] (7.75, 0.75-0.375) -- (8, 0.75-0.375);
	\node at (7.92, 0.75-0.2) {$\hat{\theta}$};
\end{tikzpicture}
\caption{Processing $n$ uses of channel $\mathcal{N}^{\theta}$ in a sequential
or adaptive manner is the most general approach to channel parameter
estimation or discrimination. The $n$ uses of the channel are interleaved with
$n$ quantum channels $\mathcal{S}^{1}$ through $\mathcal{S}^{n-1}$, which can
also share memory systems with each other. The final measurement's outcome is
then used to obtain an estimate of the unknown parameter $\theta$. If
${\theta}\in\Theta\subseteq\mathbb{R}$, then this is the task of parameter
estimation. If, on the other hand, ${\theta}\in\{1,2\}$, then this task
corresponds to channel discrimination. }
\label{fig:adaptive-scheme}
\end{figure}
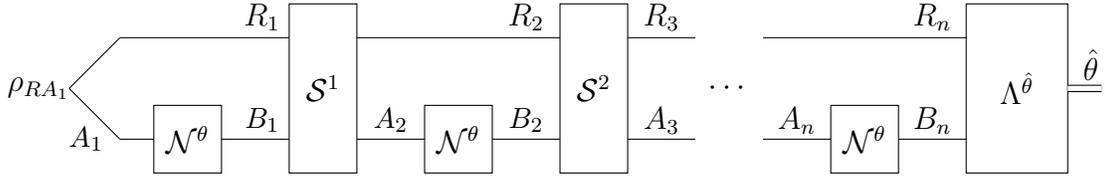The most general channel estimation procedure is depicted in
Figure~\ref{fig:adaptive-scheme}. A sequential or adaptive strategy that makes
$n$ calls to the channel is specified in terms of an input quantum state
$\rho_{R_{1}A_{1}}$, a set of interleaved channels $\{\mathcal{S}_{R_{i}
B_{i}\rightarrow R_{i+1}A_{i+1}}^{i}\}_{i=1}^{n-1}$, and a final quantum
measurement $\{\Lambda_{R_{n}B_{n}}^{\hat{\theta}}\}_{\hat{\theta}}$ that
outputs an estimate $\hat{\theta}$ of the unknown parameter (here we
incorporate any classical post-processing of a preliminary measurement outcome
$x$ to generate the estimate $\hat{\theta}$ as part of the final measurement).
Note that any particular strategy $\{\rho_{R_{1}A_{1}},\{\mathcal{S}
_{R_{i}B_{i}\rightarrow R_{i+1}A_{i+1}}^{i}\}_{i=1}^{n-1},\{\Lambda
_{R_{n}B_{n}}^{\hat{\theta}}\}_{\hat{\theta}}\}$\ employed does not depend on
the actual value of the unknown parameter $\theta$. We make the following
abbreviation for a fixed strategy in what follows:
\begin{equation}
\{\mathcal{S}^{(n)},\Lambda^{\hat{\theta}}\}\equiv\{\rho_{R_{1}A_{1}
},\{\mathcal{S}_{R_{i}B_{i}\rightarrow R_{i+1}A_{i+1}}^{i}\}_{i=1}
^{n-1},\{\Lambda_{R_{n}B_{n}}^{\hat{\theta}}\}_{\hat{\theta}}\}.
\end{equation}
The strategy begins with the estimator preparing the input quantum state
$\rho_{R_{1}A_{1}}$ and sending the $A_{1}$ system into the channel
$\mathcal{N}_{A_{1}\rightarrow B_{1}}^{\theta}$. The first channel
$\mathcal{N}_{A_{1}\rightarrow B_{1}}^{\theta}$ outputs the system $B_{1}$,
which is then available to the estimator. The resulting state is
\begin{equation}
\rho_{R_{1}B_{1}}^{\theta}:=\mathcal{N}_{A_{1}\rightarrow B_{1}}^{\theta}
(\rho_{R_{1}A_{1}}).
\end{equation}
The estimator adjoins the system $B_{1}$ to system $R_{1}$ and applies the
channel $\mathcal{S}_{R_{1}B_{1}\rightarrow R_{2}A_{2}}^{1}$, leading to the
state
\begin{equation}
\rho_{R_{2}A_{2}}^{\theta}:=\mathcal{S}_{R_{1}B_{1}\rightarrow R_{2}A_{2}}
^{1}(\rho_{R_{1}B_{1}}^{\theta}).
\end{equation}
The channel $\mathcal{S}_{R_{1}B_{1}\rightarrow R_{2}A_{2}}^{1}$ can take an
action conditioned on information in the system $B_{1}$, which itself might
contain some partial information about the unknown parameter~$\theta$. The
estimator then inputs the system $A_{2}$ into the second use of the channel
$\mathcal{N}_{A_{2}\rightarrow B_{2}}^{\theta}$, which outputs a system
$B_{2}$ and gives the state
\begin{equation}
\rho_{R_{2}B_{2}}^{\theta}:=\mathcal{N}_{A_{2}\rightarrow B_{2}}^{\theta}
(\rho_{R_{2}A_{2}}^{\theta}).
\end{equation}
This process repeats $n-2$ more times, for which we have the intermediate
states
\begin{align}
\rho_{R_{i}B_{i}}^{\theta}  &  :=\mathcal{N}_{A_{i}\rightarrow B_{i}}^{\theta
}(\rho_{R_{i}A_{i}}^{\theta}),\\
\rho_{R_{i}A_{i}}^{\theta}  &  :=\mathcal{S}_{R_{i-1}B_{i-1}\rightarrow
R_{i}A_{i}}^{i-1}(\rho_{R_{i-1}B_{i-1}}^{\theta}),
\end{align}
for $i\in\left\{  3,\ldots,n\right\}  $, and at the end, the estimator has
systems $R_{n}$ and $B_{n}$. We define $\omega_{R_{n}B_{n}}^{\theta}$ to be
the final state of the estimation protocol before the final measurement
$\{\Lambda_{R_{n}B_{n}}^{\hat{\theta}}\}_{\hat{\theta}}$:
\begin{equation}
\omega_{R_{n}B_{n}}^{\theta}:=(\mathcal{N}_{A_{n}\rightarrow B_{n}}^{\theta
}\circ\mathcal{S}_{R_{n-1}B_{n-1}\rightarrow R_{n}A_{n}}^{n-1}\circ\cdots
\circ\mathcal{S}_{R_{1}B_{1}\rightarrow R_{2}A_{2}}^{1}\circ\mathcal{N}
_{A_{1}\rightarrow B_{1}}^{\theta})(\rho_{R_{1}A_{1}}).
\label{eq:estimation-final-state}
\end{equation}
The estimator finally performs a measurement $\{\Lambda_{R_{n}B_{n}}
^{\hat{\theta}}\}_{\hat{\theta}}$ that outputs an estimate $\hat{\theta}$ of
the unknown parameter $\theta$. The conditional probability for the estimate
$\hat{\theta}$\ given the unknown parameter $\theta$ is given by the Born
rule:
\begin{equation}
p_{\theta}(\hat{\theta})=\operatorname{Tr}[\Lambda_{R_{n}B_{n}}^{\hat{\theta}
}\omega_{R_{n}B_{n}}^{\theta}]. \label{eq:cond-prob-adaptive}
\end{equation}
As we stated above, any particular strategy does not depend on the value of
the unknown parameter $\theta$, but the states at each step of the protocol do
depend on $\theta$ through the successive probings of the underlying channel
$\mathcal{N}_{A\rightarrow B}^{\theta}$.

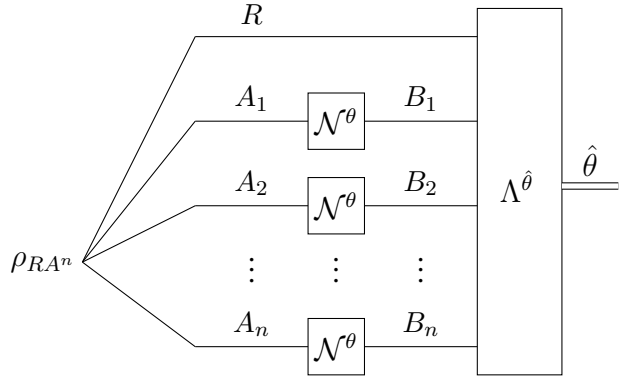
\begin{figure}[ptb]
\centering
\begin{tikzpicture}[scale=1.5]
	\draw (-1, -1.25) -- (0, 0);
	\draw (-1, -1.25) -- (0, -0.75);
	\draw (-1, -1.25) -- (0, -2);
	\draw (-1, -1.25) -- (0, 3.5-2.75);
	\node at (-1.35, -1.25) {$\rho_{R A^n}$};
	\node at (1-0.5, -1.25) {$\vdots$};
	\node at (2, -1.25) {$\vdots$};
	\draw (0,0) -- (1,0);
	\draw (1,-0.25) rectangle (1.5,0.25);
	\node at (1.25,0) {$\mathcal{N}^{\theta}$};
	\node at (0.5,0.2) {\small $A_1$};
	\node at (2,0.2) {\small$B_1$};
	\draw (1.5,0) -- (2.5,0);
	\draw (0,-0.75) -- (1,-0.75);
	\draw (1,-1) rectangle (1.5,-0.5);
	\node at (1.25,-0.75) {$\mathcal{N}^{\theta}$};
	\node at (0.5,0.2-0.75) {\small$A_2$};
	\node at (2,0.2-0.75) {\small$B_2$};
	\draw (1.5,-0.75) -- (2.5,-0.75);
	\draw (0,-2) -- (1,-2);
	\draw (1,-2.25) rectangle (1.5,-1.75);
	\node at (1.25,-2) {$\mathcal{N}^{\theta}$};
	\node at (0.5,-1.8) {\small$A_n$};
	\node at (2,-1.8) {\small$B_n$};
	\draw (1.5,-2) -- (2.5,-2);
	\node at (1.25,-1.25) {$\vdots$};
	\node at (0.5,3.5-2.55) {\small$R$};
	\draw (0,3.5-2.75) -- (2.5,3.5-2.75);
	\draw (2.5,0.75-3) rectangle (3.25,0.75+0.25);
	\node at (2.850,0.75-1.325) {$\Lambda^{\hat{\theta}}$};
	\draw[double distance between line centers=0.2em] (3.25,0.75-1.325) -- (3.75,0.75-1.325);
	\node at (3.5, 0.75-1.1) {$\hat{\theta}$};
\end{tikzpicture}
\caption{Processing $n$ uses of channel $\mathcal{N}^{\theta}$ in a parallel
manner. The $n$ channels are called in parallel, allowing for entanglement to
be shared among input systems $A_{1}$ through $A_{n}$, along with a quantum
memory system $R$. A collective measurement is made, with its outcome being an
estimate $\hat{\theta}$ for the unknown parameter $\theta$. Parallel
strategies form a special case of sequential ones, and therefore parallel
strategies are no more powerful than sequential ones.}
\label{fig:parallel-scheme}
\end{figure}

Note that such a sequential strategy contains a parallel or non-adaptive
strategy as a special case: the system $R_{1}$ can be arbitrarily large and
divided into subsystems, with the only role of the interleaved channels
$\mathcal{S}_{R_{i}B_{i}\rightarrow R_{i+1}A_{i+1}}^{i}$ being that they
redirect these subsystems to be the inputs of future calls to the channel (as
would be the case in any non-adaptive strategy for estimation or
discrimination). Figure~\ref{fig:parallel-scheme}\ depicts a parallel or
non-adaptive channel estimation strategy.
%When one restricts to the case of unitary channels, it
%is known that general sequential strategies for channel estimation provide the
%same scaling in error as parallel ones \cite{Giovannetti2006}.

One main goal of the present paper is to place a lower bound on the MSE\ of a
general sequential strategy for channel parameter estimation, such that the
lower bound is a function solely of the channel family $\left\{
\mathcal{N}_{A\rightarrow B}^{\theta}\right\}  _{\theta}$ and the number
$n$\ of channel uses. Such a bound indicates a fundamental limitation for
channel estimation that cannot be improved upon by any possible estimation strategy.

\subsection{Quantum channel discrimination}

\label{sec:ch-disc}

The operational setting for quantum channel discrimination is exactly as
described above, and the only difference is that $\theta\in\Theta= \left\{
1,\ldots,d\right\}  $ for some integer $d$. In this work, we focus exclusively
on the case $d=2$ for channel discrimination.

\subsubsection{Symmetric setting}

In this subsection, we recall the setting of symmetric or Bayesian channel
discrimination in which there is a prior probability distribution for $\theta
$: $\Pr[\theta=1]=p\in\left(  0,1\right)  $ and $\Pr[\theta=2]=1-p$. The
relevant measure of performance of a given channel discrimination strategy
$\{\mathcal{S}^{(n)} ,\Lambda^{\hat{\theta}}\}$ is the expected error
probability:
\begin{align}
&  p_{e}^{(n)}(\{\mathcal{N}^{\theta}\}_{\theta},\{\mathcal{S}^{(n)}
,\Lambda^{\hat{\theta}}\})\nonumber\\
&  =\Pr[\hat{\theta}\neq\theta]\\
&  =\Pr[\theta=1]\Pr[\hat{\theta}=2|\theta=1]+\Pr[\theta=2]\Pr[\hat{\theta
}=1|\theta=2]\\
&  =p\operatorname{Tr}[\Lambda_{R_{n}B_{n}}^{\hat{\theta}=2}\omega_{R_{n}
B_{n}}^{\theta=1}]+\left(  1-p\right)  \operatorname{Tr}[\Lambda_{R_{n}B_{n}
}^{\hat{\theta}=1}\omega_{R_{n}B_{n}}^{\theta=2}]\\
&  =p\operatorname{Tr}[(I_{R_{n}B_{n}}-\Lambda_{R_{n}B_{n}})\omega_{R_{n}
B_{n}}^{\theta=1}]+\left(  1-p\right)  \operatorname{Tr}[\Lambda_{R_{n}B_{n}
}\omega_{R_{n}B_{n}}^{\theta=2}],
\end{align}
where $\omega_{R_{n}B_{n}}^{\theta}$ is the state at the end of the protocol,
as defined in \eqref{eq:estimation-final-state}, and we made the abbreviation
$\Lambda_{R_{n}B_{n}}\equiv\Lambda_{R_{n}B_{n}}^{\hat{\theta}=1}$. We can also
write the error probability in conventional notation as follows:
\begin{equation}
p_{e}^{(n)}(\{\mathcal{N}^{\theta}\}_{\theta},\{\mathcal{S}^{(n)}
,\Lambda^{\hat{\theta}}\}):=p\alpha_{n}(\{\mathcal{S}^{(n)},\Lambda
^{\hat{\theta}}\})+\left(  1-p\right)  \beta_{n}(\{\mathcal{S}^{(n)}
,\Lambda^{\hat{\theta}}\}),
\end{equation}
where $\alpha_{n}$ is called the Type~I error probability and $\beta_{n}$ the
Type~II error probability:
\begin{align}
\alpha_{n}(\{\mathcal{S}^{(n)},\Lambda^{\hat{\theta}}\})  &
:=\operatorname{Tr}[(I_{R_{n}B_{n}}-\Lambda_{R_{n}B_{n}})\omega_{R_{n}B_{n}
}^{\theta=1}],\label{eq:type-I-err-prob}\\
\beta_{n}(\{\mathcal{S}^{(n)},\Lambda^{\hat{\theta}}\})  &
:=\operatorname{Tr}[\Lambda_{R_{n}B_{n}}\omega_{R_{n}B_{n}}^{\theta=2}].
\label{eq:type-II-err-prob}
\end{align}
By optimizing the final measurement, we arrive at the following optimized
error probability:
\begin{align}
p_{e}^{(n)}(\{\mathcal{N}^{\theta}\}_{\theta},\mathcal{S}^{(n)})  &
:=\inf_{\{\Lambda^{\hat{\theta}}\}_{\hat{\theta}}}p_{e}^{(n)}(\{\mathcal{N}
^{\theta}\}_{\theta},\mathcal{S}^{(n)},\Lambda^{\hat{\theta}})\\
&  =\frac{1}{2}\left(  1-\left\Vert p\omega_{R_{n}B_{n}}^{\theta=1}-\left(
1-p\right)  \omega_{R_{n}B_{n}}^{\theta=2}\right\Vert _{1}\right)  ,
\end{align}
where the last equality follows from a standard result in quantum state
discrimination theory \cite{Hel69,Hol72,Hel76}. We can perform a further
optimization over all discrimination strategies to arrive at the optimal
expected error probability:
\begin{align}
p_{e}^{(n)}(\{\mathcal{N}^{\theta}\}_{\theta})  &  :=\inf_{\mathcal{S}^{(n)}
}p_{e}^{(n)}(\{\mathcal{N}^{\theta}\}_{\theta},\mathcal{S}^{(n)})\\
&  =\frac{1}{2}\left(  1-\left\Vert p(\mathcal{N}^{\theta=1})^{(n)}-\left(
1-p\right)  (\mathcal{N}^{\theta=2})^{(n)}\right\Vert _{\diamond n}\right)  ,
\end{align}
where the quantum strategy distance \cite{Gutoski2007,Gutoski2010,Gutoski2012}
(see also \cite{Chiribella2008,Chiribella2009}) is defined as
\begin{equation}
\left\Vert p(\mathcal{N}^{\theta=1})^{(n)}-\left(  1-p\right)  (\mathcal{N}
^{\theta=2})^{(n)}\right\Vert _{\diamond n}:=\sup_{\mathcal{S}^{(n)}
}\left\Vert p\omega_{R_{n}B_{n}}^{\theta=1}-\left(  1-p\right)  \omega
_{R_{n}B_{n}}^{\theta=2}\right\Vert _{1}.
\end{equation}
Although the strategy distance can be computed by means of a semi-definite
program \cite{Gutoski2012}, this fact is only useful for small $n$ and
small-dimensional channels because the difficulty in calculating grows quickly
as $n$ becomes larger (see \cite{KW20} for explicit examples of the
calculation of the strategy distance).

As such, we are interested in the exponential rate at which the expected error
probability converges to zero in the limit as $n$ becomes larger:
\begin{equation}
\xi_{n}(p,\{\mathcal{N}^{\theta}\}_{\theta}):=-\frac{1}{n}\ln p_{e}
^{(n)}(\{\mathcal{N}^{\theta}\}_{\theta}).
\label{eq:non-asymptotic-Chernoff-exp}
\end{equation}
This quantity is called the non-asymptotic Chernoff exponent of quantum
channels \cite{Berta2018c}, and its asymptotic counterparts are defined as
\begin{equation}
\underline{\xi
}(\{\mathcal{N}^{\theta}\}_{\theta}):=\liminf_{n\rightarrow\infty}\xi
_{n}(p,\{\mathcal{N}^{\theta}\}_{\theta}),
\qquad
\overline{\xi}(\{\mathcal{N}^{\theta}\}_{\theta}):=\limsup_{n\rightarrow
\infty}\xi_{n}(p,\{\mathcal{N}^{\theta}\}_{\theta}).
\label{eq:asymptotic-Chernoff}
\end{equation}
The asymptotic quantities 
 $\underline{\xi}(\{\mathcal{N}^{\theta}\}_{\theta})$ and
$\overline{\xi}(\{\mathcal{N}^{\theta}\}_{\theta})$
are independent of
the particular value of $p\in(0,1)$.

Another goal of the present paper is to establish an improved upper bound on
$\xi_{n}(p,\{\mathcal{N}^{\theta}\}_{\theta})$ and thus on $\overline{\xi
}(\{\mathcal{N}^{\theta}\}_{\theta})$.

\subsubsection{Asymmetric setting -- Hoeffding error exponent}

Another setting of interest for channel discrimination is called the Hoeffding
error exponent setting (see, e.g., \cite{CMW14,Berta2018c}). In this case,
there is no assumed prior probability on the parameter $\theta$. In this
setting, the Type~II\ error probability $\beta_{n}$\ in
\eqref{eq:type-II-err-prob}\ is constrained to decrease exponentially at a
fixed rate $r>0$, and the objective is to determine the optimal exponential
rate of decay for the Type~I error probability $\alpha_{n}$ in
\eqref{eq:type-I-err-prob}, subject to this constraint. Formally, the
non-asymptotic Hoeffding error exponent is defined as follows
\cite{Berta2018c}:
\begin{equation}
B_{n}(r,\{\mathcal{N}^{\theta}\}_{\theta}):=\sup_{\{\mathcal{S}^{(n)}
,\Lambda^{\hat{\theta}}\}}\left\{  -\frac{1}{n}\ln\alpha_{n}(\{\mathcal{S}
^{(n)},\Lambda^{\hat{\theta}}\})\middle|-\frac{1}{n}\ln\beta_{n}
(\{\mathcal{S}^{(n)},\Lambda^{\hat{\theta}}\})\geq r\right\}  ,
\end{equation}
and its asymptotic variants as follows:
\begin{equation}
\underline{B}(r,\{\mathcal{N}^{\theta}\}_{\theta}):=\liminf_{n\rightarrow
\infty}B_{n}(r,\{\mathcal{N}^{\theta}\}_{\theta}),\qquad\overline
{B}(r,\{\mathcal{N}^{\theta}\}_{\theta}):=\limsup_{n\rightarrow\infty}
B_{n}(r,\{\mathcal{N}^{\theta}\}_{\theta}). \label{eq:asymptotic-Hoeffding}
\end{equation}

\subsubsection{Sequential channel discrimination with repetition}

\label{sec:seq-ch-w-rep}

As a variation of the general channel discrimination setting discussed in
Section~\ref{sec:ch-disc}, we can consider a more specialized setting that
we call sequential channel discrimination with repetition. In this setting,
the general, $n$-round channel discrimination protocol discussed in
Section~\ref{sec:ch-disc}\ is repeated $m$ times, such that the final state of
the protocol is $(\omega_{R_{n}B_{n}}^{\theta})^{\otimes m}$, where
$\omega_{R_{n}B_{n}}^{\theta}$ is defined in
\eqref{eq:estimation-final-state}. One can then perform a collective
measurement $\{\Lambda_{(R_{n}B_{n})^{m}}^{\hat{\theta}}\}$ on this final
state, where the notation $(R_{n}B_{n})^{m}$ is a shorthand for all of the
remaining systems at the end of the $nm$ calls to the channel. We abbreviate
such a protocol with the notation $\{\mathcal{S}^{(n)},\Lambda^{\hat{\theta
},(m)}\}$, which indicates that the protocol $\mathcal{S}^{(n)}$ is fixed, but
the final measurement is performed on $m$ systems. The two kinds of errors in
such a protocol are then defined as follows:
\begin{align}
\alpha_{n,m}(\{\mathcal{S}^{(n)},\Lambda^{\hat{\theta},(m)}\})  &
:=\operatorname{Tr}[(I_{R_{n}B_{n}}-\Lambda_{(R_{n}B_{n})^{m}}^{\hat{\theta
}=1})(\omega_{R_{n}B_{n}}^{\theta=1})^{\otimes m}],\\
\beta_{n,m}(\{\mathcal{S}^{(n)},\Lambda^{\hat{\theta},(m)}\})  &
:=\operatorname{Tr}[\Lambda_{(R_{n}B_{n})^{m}}^{\hat{\theta}=1}(\omega
_{R_{n}B_{n}}^{\theta=2})^{\otimes m}].
\end{align}
This somewhat specialized setting has been considered in the context of quantum
channel estimation \cite{Yuan2017}.\ We refer to such a protocol as an $(n,m)$
protocol for sequential channel discrimination with repetition.

Of course, sequential channel discrimination with repetition is special kind
of channel discrimination protocol of the form discussed in
Section~\ref{sec:ch-disc}, in which the channel is called $nm$ times. Thus,
the optimal error probabilities involved in $(n,m)$ sequential channel
discrimination with repetition cannot be smaller than the optimal error
probabilities in a general channel discrimination protocol that calls the
channel $nm$ times. At the same time, a general channel discrimination
protocol that calls the channel $n$ times is trivially an $(n,1)$ sequential
channel discrimination protocol with repetition (however, the phrase
\textquotedblleft with repetition\textquotedblright\ is not particular apt in
this specialized instance).

We can define the non-asymptotic Chernoff and Hoeffding error exponents in a
similar way to how they were defined in the previous section. The
non-asymptotic Chernoff exponent is defined as
\begin{equation}
\xi_{n,m}(p,\{\mathcal{N}^{\theta}\}_{\theta}):=-\frac{1}{nm}\ln p_{e}
^{(n,m)}(\{\mathcal{N}^{\theta}\}_{\theta}),
\end{equation}
where
\begin{equation}
p_{e}^{(n,m)}(\{\mathcal{N}^{\theta}\}_{\theta}):=\inf_{\{\mathcal{S}
^{(n)},\Lambda^{\hat{\theta},(m)}\}}p\alpha_{n,m}(\{\mathcal{S}^{(n)}
,\Lambda^{\hat{\theta},(m)}\})+\left(  1-p\right)  \beta_{n,m}(\{\mathcal{S}
^{(n)},\Lambda^{\hat{\theta},(m)}\}),
\end{equation}
and the non-asymptotic Hoeffding exponent as
\begin{multline}
B_{n,m}(r,\{\mathcal{N}^{\theta}\}_{\theta}):=\\
\sup_{\{\mathcal{S}^{(n)},\Lambda^{\hat{\theta},(m)}\}}\left\{  -\frac{1}
{nm}\ln\alpha_{n,m}(\{\mathcal{S}^{(n)},\Lambda^{\hat{\theta},(m)}
\})\middle|-\frac{1}{nm}\ln\beta_{n,m}(\{\mathcal{S}^{(n)},\Lambda
^{\hat{\theta},(m)}\})\geq r\right\}  .
\end{multline}
From these non-asymptotic quantities, one can then define asymptotic
quantities similar to \eqref{eq:asymptotic-Chernoff} and
\eqref{eq:asymptotic-Hoeffding}. However, note that they might possibly depend on
the order of limits (whether one takes $\lim_{m\rightarrow\infty}$ or
$\lim_{n\rightarrow\infty}$ first). Another contribution of our paper is to
establish upper bounds on the asymptotic versions of $\xi_{n,m}
(p,\{\mathcal{N}^{\theta}\}_{\theta})$\ and $B_{n,m}(r,\{\mathcal{N}^{\theta
}\}_{\theta})$ that hold in the case that we take the limit $\lim
_{m\rightarrow\infty}$ first, followed by the limit $\lim_{n\rightarrow\infty
}$.

\section{Limits on quantum channel parameter estimation}

\label{sec:bounds-on-estimation}

\subsection{Classical and quantum Fisher information} \label{subsec:classical-q-fisher-info}

\subsubsection{Classical Fisher information and its operational relevance}

Let us first recall some fundamental results well known in classical
estimation theory \cite{Cram46,Rao45} (see also \cite{Kay93}). Here, we
suppose that there is a family $\{p_{\theta}\}_{\theta}$\ of probability
distributions that are a function of the unknown parameter $\theta\in
\Theta\subseteq\mathbb{R} $, and the goal is to produce an estimate
$\hat{\theta}$ of $\theta$ from $n$ independent samples of the distribution
$p_{\theta}(x)$. It is clear that the estimate can improve as the number
$n$\ of samples becomes large, but we are interested in how the MSE\ scales
with $n$, as well as particular scaling factors.

Let $\{p_{\theta}(x)\}_{\theta}$\ denote a family of probability density
functions. Suppose that the family $\{p_{\theta}(x)\}_{\theta}$ is
differentiable with respect to the parameter $\theta$, so that $\partial
_{\theta}p_{\theta}(x)$ exists for all values of $\theta$ and $x$, where
$\partial_{\theta}\equiv\frac{\partial}{\partial\theta}$. The classical Fisher
information $I_{F}(\theta;\{p_{\theta}\}_{\theta})$ of the family
$\{p_{\theta}(x)\}_{\theta}$ is defined as follows:
\begin{equation}
I_{F}(\theta;\{p_{\theta}\}_{\theta}):=\left\{
\begin{array}
[c]{cc}
\int_{\Omega}dx\ \frac{1}{p_{\theta}(x)}\left(  \partial_{\theta}p_{\theta
}(x)\right)  ^{2} & \text{if }\operatorname{supp}(\partial_{\theta}p_{\theta
})\subseteq\operatorname{supp}(p_{\theta})\\
+\infty & \text{otherwise}
\end{array}
\right.  , \label{eq:CFI}
\end{equation}
where $\Omega$ is the sample space for the probability density function
$p_{\theta}(x)$. When the support condition
\begin{equation}
\operatorname{supp}(\partial_{\theta}p_{\theta})\subseteq\operatorname{supp}
(p_{\theta})
\end{equation}
is satisfied (understood as \textquotedblleft essential
support\textquotedblright), the classical Fisher information has the following
alternative expression:
\begin{equation}
I_{F}(\theta;\{p_{\theta}\}_{\theta})    =\int_{\Omega}dx\ p_{\theta
}(x)\left(  \partial_{\theta}\ln p_{\theta}(x)\right)  ^{2}
  =\mathbb{E}_{p_{\theta}}[\left(  \partial_{\theta}\ln p_{\theta}(X)\right)
^{2}], \label{eq:CFI-for-RLD}
\end{equation}
interpreted as the variance of the surprisal rate $\partial_{\theta}[- \ln p_{\theta}(x)]$.

One of the fundamental results of classical estimation theory
\cite{Cram46,Rao45,Kay93} is the Cramer--Rao lower bound on the MSE\ of an
unbiased estimator of $\theta$:
\begin{equation}
\text{Var}(\hat{\theta})\geq\frac{1}{nI_{F}(\theta;\{p_{\theta}\}_{\theta})}.
\label{eq:CRB-classical}
\end{equation}
The Cramer--Rao bound can be saturated, in the sense that there exists an
estimator, the maximum likelihood estimator, having an MSE\ that achieves the
lower bound in the large~$n$ limit of many independent trials
\cite{fisher_1925}.

\subsubsection{SLD Fisher information and its operational relevance}

\label{subsubsec:SLD-fish}

The classical Cramer--Rao bound (CRB) can be generalized to a quantum scenario
\cite{Hel67,Hel69,Hel76} (see also \cite{H11book}). Let $\{\rho_{\theta
}\}_{\theta}$ denote a family of quantum states into which the parameter
$\theta\in\Theta\subseteq\mathbb{R}$ is encoded. One can then subject $n$ copies of this
state $\rho_{\theta}$ to a quantum measurement $\{\Lambda_{x}\}_{x}$ to yield
a classical probability distribution according to the Born rule:
\begin{equation}
p_{\theta}(x)=\operatorname{Tr}[\Lambda_{x}\rho^{\otimes n}_{\theta}],
\end{equation}
from which one then forms an estimate $\hat{\theta}$.
Suppose that the family $\{\rho_{\theta}\}_{\theta}$ of quantum states is
differentiable with respect to $\theta$, so that $\partial_{\theta}
\rho_{\theta}$ exists for all values of $\theta$. We can then apply the
classical CRB\ as given in \eqref{eq:CRB-classical}, but it is desirable in
the quantum case to perform the best possible measurement in order to know the
scaling of any possible quantum estimation strategy. The optimal measurement
leads to the \textit{most informative} CRB, which is called the quantum CRB
(QCRB) and is given as the following bound on the variance of an unbiased estimator of $\theta$:
\begin{equation}
\text{Var}(\hat{\theta})\geq\frac{1}{nI_{F}(\theta;\{\rho_{\theta}\}_{\theta
})}, \label{eq:QCRB}
\end{equation}
where $I_{F}(\theta;\{\rho_{\theta}\}_{\theta})$ is the symmetric logarithmic
derivative (SLD) quantum Fisher information, given in
Definition~\ref{def:SLD-Fish-states}\ below, and we have applied the additivity relation $I_{F}(\theta;\{\rho^{\otimes n}_{\theta}\}_{\theta
}) = nI_{F}(\theta;\{\rho_{\theta}\}_{\theta
})$. The lower bound in
\eqref{eq:QCRB} is achievable in the large $n$ limit of many copies of the
state $\rho_{\theta}$ \cite{Nag89,Braunstein1994a}.

\begin{definition}
[SLD\ Fisher information]\label{def:SLD-Fish-states}Let $\{\rho_{\theta
}\}_{\theta}$ be a differentiable family of quantum states. Then the
SLD\ Fisher information is defined as follows:
\begin{equation}
I_{F}(\theta;\{\rho_{\theta}\}_{\theta})=\left\{
\begin{array}
[c]{cc}
2\left\Vert \left(  \rho_{\theta}\otimes I+I\otimes\rho_{\theta}^{T}\right)
^{-\frac{1}{2}}\left(  (\partial_{\theta}\rho_{\theta})\otimes I\right)
|\Gamma\rangle\right\Vert _{2}^{2} & \text{if }\Pi_{\rho_{\theta}}^{\perp
}(\partial_{\theta}\rho_{\theta})\Pi_{\rho_{\theta}}^{\perp}=0\\
+\infty & \text{otherwise}
\end{array}
\right.  , \label{eq:basis-independent-formula-SLD}
\end{equation}
where $\Pi_{\rho_{\theta}}^{\perp}$ denotes the projection onto the kernel of
$\rho_{\theta}$, $|\Gamma\rangle=\sum_{i}|i\rangle|i\rangle$ is the
unnormalized maximally entangled vector, $\{|i\rangle\}_{i}$ is any
orthonormal basis, the transpose in \eqref{eq:basis-independent-formula-SLD}
is with respect to this basis, and the inverse is taken on the support of
$\rho_{\theta}\otimes I+I\otimes\rho_{\theta}^{T}$.
\end{definition}

Let the spectral decomposition of $\rho_{\theta}$ be given as
\begin{equation}
\rho_{\theta}=\sum_{j}\lambda_{\theta}^{j}|\psi_{\theta}^{j}\rangle\!\langle
\psi_{\theta}^{j}|,
\end{equation}
which includes the indices for which $\lambda_{\theta}^{j}=0$. Then the
projection $\Pi_{\rho_{\theta}}^{\perp}$ onto the kernel of $\rho_{\theta}$ is
given by
\begin{equation}
\Pi_{\rho_{\theta}}^{\perp}:=\sum_{j:\lambda_{\theta}^{j}=0}|\psi_{\theta}
^{j}\rangle\!\langle\psi_{\theta}^{j}|.
\end{equation}
With this notation, the SLD\ quantum Fisher information can also be written as
follows, as discussed in Appendix~\ref{app:basis-ind-dep-formulas-SLD-FI}:
\begin{equation}
I_{F}(\theta;\{\rho_{\theta}\}_{\theta})=\left\{
\begin{array}
[c]{cc}
2\sum_{j,k:\lambda_{j}^{\theta}+\lambda_{k}^{\theta}>0}\frac{|\langle
\psi_{\theta}^{j}|(\partial_{\theta}\rho_{\theta})|\psi_{\theta}^{k}
\rangle|^{2}}{\lambda_{\theta}^{j}+\lambda_{\theta}^{k}} & \text{if }\Pi
_{\rho_{\theta}}^{\perp}(\partial_{\theta}\rho_{\theta})\Pi_{\rho_{\theta}
}^{\perp}=0\\
+\infty & \text{otherwise}
\end{array}
\right.  . \label{eq:SLD-Fish-info-formula}
\end{equation}
The formula in \eqref{eq:basis-independent-formula-SLD} has the advantage that
it is basis independent, with no need to perform a spectral decomposition in
order to calculate the SLD\ Fisher information. It also leads to a
semi-definite program for calculating the SLD Fisher information, as we show
in Section~\ref{sec:SDPs-for-FIs}.

As we discuss in more detail in
Appendix~\ref{app:phys-cons-SLD-RLD-Fish-states}, the finiteness condition
\begin{equation}
\Pi_{\rho_{\theta}}^{\perp}\partial_{\theta}\rho_{\theta}\Pi_{\rho_{\theta}
}^{\perp}=0 \label{eq:finiteness-condition-SLD-Fish-states}
\end{equation}
in \eqref{eq:basis-independent-formula-SLD} is equivalent to the following condition:
\begin{equation}
\forall j,k:\langle\psi_{j}^{\theta}|(\partial_{\theta}\rho_{\theta})|\psi
_{k}^{\theta}\rangle=0\text{ if }\lambda_{j}^{\theta}+\lambda_{k}^{\theta}=0,
\label{eq:finiteness-condition-SLD-Fish-alt}
\end{equation}
which is helpful for understanding the formula in \eqref{eq:SLD-Fish-info-formula}.

Note that the condition $\Pi_{\rho_{\theta}}^{\perp}\partial_{\theta}
\rho_{\theta}\Pi_{\rho_{\theta}}^{\perp}=0$ is not equivalent to
$\operatorname{supp}(\partial_{\theta}\rho_{\theta})\subseteq
\operatorname{supp}(\rho_{\theta})$. The latter condition $\operatorname{supp}
(\partial_{\theta}\rho_{\theta})\subseteq\operatorname{supp}(\rho_{\theta})$
is equivalent to $\Pi_{\rho_{\theta}}^{\perp}\partial_{\theta}\rho_{\theta
}=\partial_{\theta}\rho_{\theta}\Pi_{\rho_{\theta}}^{\perp}=0$ and implies
$\Pi_{\rho_{\theta}}^{\perp}\partial_{\theta}\rho_{\theta}\Pi_{\rho_{\theta}
}^{\perp}=0$, but the converse is not necessarily true. To elaborate on this
point, consider that we can write the operator $\partial_{\theta}\rho_{\theta
}$ with respect to the Hilbert space decomposition $\operatorname{supp}
(\rho_{\theta})\oplus\ker(\rho_{\theta})$ in the following matrix form:
\begin{equation}
\partial_{\theta}\rho_{\theta}=
\begin{bmatrix}
(\partial_{\theta}\rho_{\theta})_{0,0} & (\partial_{\theta}\rho_{\theta
})_{0,1}\\
(\partial_{\theta}\rho_{\theta})_{0,1}^{\dag} & (\partial_{\theta}\rho
_{\theta})_{1,1}
\end{bmatrix}
,
\end{equation}
where
\begin{equation}
(\partial_{\theta}\rho_{\theta})_{0,0}:=\Pi_{\rho_{\theta}}\partial_{\theta
}\rho_{\theta}\Pi_{\rho_{\theta}},\quad(\partial_{\theta}\rho_{\theta}
)_{0,1}:=\Pi_{\rho_{\theta}}\partial_{\theta}\rho_{\theta}\Pi_{\rho_{\theta}
}^{\perp},\quad(\partial_{\theta}\rho_{\theta})_{1,1}:=\Pi_{\rho_{\theta}
}^{\perp}\partial_{\theta}\rho_{\theta}\Pi_{\rho_{\theta}}^{\perp}.
\end{equation}
The constraint $\operatorname{supp}(\partial_{\theta}\rho_{\theta}
)\subseteq\operatorname{supp}(\rho_{\theta})$ implies that both $(\partial
_{\theta}\rho_{\theta})_{0,1}$ and $(\partial_{\theta}\rho_{\theta})_{1,1}$
are zero, whereas the constraint $\Pi_{\rho_{\theta}}^{\perp}\partial_{\theta
}\rho_{\theta}\Pi_{\rho_{\theta}}^{\perp}=0$ implies that $(\partial_{\theta
}\rho_{\theta})_{1,1}$ is zero.

As we also show in Appendix~\ref{app:phys-cons-SLD-RLD-Fish-states}, when the
finiteness condition in
\eqref{eq:finiteness-condition-SLD-Fish-states}\ holds, a formula alternative
but equal to \eqref{eq:SLD-Fish-info-formula}\ is as follows:
\begin{equation}
I_{F}(\theta;\{\rho_{\theta}\}_{\theta})=2\sum_{j,k:\lambda_{j}^{\theta
},\lambda_{k}^{\theta}>0}\frac{|\langle\psi_{\theta}^{j}|(\partial_{\theta
}\rho_{\theta})|\psi_{\theta}^{k}\rangle|^{2}}{\lambda_{\theta}^{j}
+\lambda_{\theta}^{k}}+4\sum_{j:\lambda_{j}^{\theta}>0}\frac{\langle
\psi_{\theta}^{j}|(\partial_{\theta}\rho_{\theta})\Pi_{\rho_{\theta}}^{\perp
}(\partial_{\theta}\rho_{\theta})|\psi_{\theta}^{j}\rangle}{\lambda_{\theta
}^{j}}. \label{eq:SLD-Fish-alt-formula-kernel-rho-theta}
\end{equation}

For a differentiable family $\{|\varphi_{\theta}\rangle\!\langle\varphi_{\theta
}|\}_{\theta}$ of pure states, we discuss in
Appendix~\ref{sec:pure-state-fam-examps}\ how the formula in
\eqref{eq:SLD-Fish-info-formula} reduces to the well known expression
\cite{Braunstein1994a,Fuji94}:
\begin{equation}
I_{F}(\theta;\{|\varphi_{\theta}\rangle\!\langle\varphi_{\theta}|\}_{\theta
})=4\left[  \langle\partial_{\theta}\phi_{\theta}|\partial_{\theta}
\phi_{\theta}\rangle-\left\vert \langle\partial_{\theta}\phi_{\theta}
|\phi_{\theta}\rangle\right\vert ^{2}\right]  .
\label{eq:pure-state-reduction-SLD-Fish}
\end{equation}
That is, for all pure-state differentiable families, the finiteness condition
in \eqref{eq:finiteness-condition-SLD-Fish-states} always holds and one can
employ the formula in \eqref{eq:SLD-Fish-info-formula} to arrive at the
expression in \eqref{eq:pure-state-reduction-SLD-Fish}.

The following proposition demonstrates that the definition in
\eqref{eq:SLD-Fish-info-formula} is physically consistent, in the sense that
it is the result of a limiting procedure in which some constant additive noise vanishes:

\begin{proposition}
\label{prop:physical-consistency-SLD-Fish-states}Let $\{\rho_{\theta
}\}_{\theta}$ be a differentiable family of quantum states. Then the
SLD\ Fisher information in \eqref{eq:SLD-Fish-info-formula}\ is given by the
following limit:
\begin{equation}
I_{F}(\theta;\{\rho_{\theta}\}_{\theta})=\lim_{\varepsilon\rightarrow0}
I_{F}(\theta;\{\rho_{\theta}^{\varepsilon}\}_{\theta}),
\label{eq:fish-limit-delta}
\end{equation}
where
\begin{equation}
\rho_{\theta}^{\varepsilon}:=\left(  1-\varepsilon\right)  \rho_{\theta
}+\varepsilon\pi_{d}, \label{eq:rho-delta-def}
\end{equation}
and $\pi_{d}:=I/d$ is the maximally mixed state, with $d$ large enough so that
$\operatorname{supp}(\rho_{\theta})\subseteq\operatorname{supp}(\pi)$ for all
$\theta$.
\end{proposition}

\begin{proof}
See Appendix~\ref{app:phys-cons-SLD-RLD-Fish-states}.
\end{proof}

\medskip

In the case that the condition $\Pi_{\rho_{\theta}}^{\perp}(\partial_{\theta
}\rho_{\theta})\Pi_{\rho_{\theta}}^{\perp}=0$ holds, we can also write the
SLD\ Fisher information as follows:
\begin{equation}
I_{F}(\theta;\{\rho_{\theta}\}_{\theta})=\mathrm{Tr}[ L_{\theta}^{2}
\rho_{\theta}] =\operatorname{Tr}[L_{\theta}(\partial_{\theta}\rho_{\theta})],
\label{eq:SLD-FI}
\end{equation}
where the operator $L_{\theta}$ is the symmetric logarithmic derivative (SLD)
\cite{Hel67}, defined through the following differential equation:
\begin{equation}
\partial_{\theta}\rho_{\theta}:=\frac{1}{2}\left(  \rho_{\theta}L_{\theta
}+L_{\theta}\rho_{\theta}\right)  . \label{eq:SLD-op-def}
\end{equation}
In Appendix~\ref{app:basis-ind-dep-formulas-SLD-FI}, we revisit the derivation
of \cite{Saf18}\ and show how \eqref{eq:SLD-FI} is a consequence of
\eqref{eq:basis-independent-formula-SLD} when the finiteness condition in
\eqref{eq:finiteness-condition-SLD-Fish-states}\ holds. By sandwiching
\eqref{eq:SLD-op-def} on the left and right by $\langle\psi_{k}^{\theta}|$ and
$|\psi_{j}^{\theta}\rangle$, with $|\psi_{j}^{\theta}\rangle,|\psi_{k}
^{\theta}\rangle\in\operatorname{supp}(\rho_{\theta})$, one can check that the
SLD\ has the following unique and explicit form on the subspace
$\operatorname{span}\{|\psi\rangle\!\langle\varphi|:|\psi\rangle,|\varphi
\rangle\in\operatorname{supp}(\rho_{\theta})\}$:
\begin{equation}
L_{\theta}=2\sum_{j,k:\lambda_{j}^{\theta}+\lambda_{k}^{\theta}>0}
\frac{\langle\psi_{j}^{\theta}|(\partial_{\theta}\rho_{\theta})|\psi
_{k}^{\theta}\rangle}{\lambda_{j}^{\theta}+\lambda_{k}^{\theta}}|\psi
_{j}^{\theta}\rangle\!\langle\psi_{k}^{\theta}|.
\end{equation}
Then, in the case that the finiteness condition in
\eqref{eq:finiteness-condition-SLD-Fish-states} holds, after evaluating
\eqref{eq:SLD-FI}, we arrive at the explicit formula for the SLD Fisher
information in \eqref{eq:SLD-Fish-info-formula}.

As indicated above, in the case that
\eqref{eq:finiteness-condition-SLD-Fish-states} holds, the following equality
holds between the basis-independent formula in
\eqref{eq:basis-independent-formula-SLD} and the basis-dependent formula in
\eqref{eq:SLD-Fish-info-formula}:
\begin{align}
I_{F}(\theta;\{\rho_{\theta}\}_{\theta})  &  =2\sum_{j,k:\lambda_{j}^{\theta
}+\lambda_{k}^{\theta}>0}\frac{|\langle\psi_{\theta}^{j}|(\partial_{\theta
}\rho_{\theta})|\psi_{\theta}^{k}\rangle|^{2}}{\lambda_{\theta}^{j}
+\lambda_{\theta}^{k}}\label{eq:basis-dependent-SLD-formula}\\
&  =2\langle\Gamma|\left(  (\partial_{\theta}\rho_{\theta})\otimes I\right)
\left(  \rho_{\theta}\otimes I+I\otimes\rho_{\theta}^{T}\right)  ^{-1}\left(
(\partial_{\theta}\rho_{\theta})\otimes I\right)  |\Gamma\rangle
\label{eq:basis-independent-SLD-formula-extra}\\
&  =2\left\Vert \left(  \rho_{\theta}\otimes I+I\otimes\rho_{\theta}
^{T}\right)  ^{-\frac{1}{2}}\left(  (\partial_{\theta}\rho_{\theta})\otimes
I\right)  |\Gamma\rangle\right\Vert _{2}^{2}.
\label{eq:basis-independent-formula-SLD-2}
\end{align}
This basis-independent formula was explicitly given in \cite{Saf18}. Arguably,
it is implicitly given in \cite{Petz96,Jen12}, being a consequence of (a)\ the
general theory presented in \cite{Petz96}\ in terms of monotone metrics and
the relative modular operator formalism \cite{ArakiMasuda82}\ and (b)\ the well
known isomorphism connecting the Hilbert--Schmidt inner product to an extended
vector-space inner product \cite{Ando79}, which is called Ando's identity in
\cite{Carlen09}\ (see also\ \cite{P86,TCR09,HMPB11,W18opt}). The formula in
\eqref{eq:basis-independent-SLD-formula-extra}\ was presented in
\cite[Remark~4]{Jen12} in the relative modular operator formalism and in
\cite{Saf18} in the extended Hilbert space formalism (as given above). As
indicated above, we discuss this equality in more detail in
Appendix~\ref{app:basis-ind-dep-formulas-SLD-FI}.

The explicit formula in \eqref{eq:SLD-Fish-info-formula} can be difficult to
evaluate in practice because it requires performing a spectral decomposition
of $\rho_{\theta}$. The same is true for the formula in
\eqref{eq:basis-independent-formula-SLD} due to the presence of a matrix
inverse. To get around these problems, we show in
Section~\ref{sec:SDPs-for-FIs}\ how the SLD Fisher information can be
evaluated by means of a semi-definite program that takes $\rho_{\theta}$ and
$\partial_{\theta}\rho_{\theta}$ as input (that is, with this approach, there
is no need to perform a diagonalization of $\rho_{\theta}$ or a matrix
inverse). See \cite{BV04,Wat18} for general background on semi-definite programming.

\subsubsection{RLD\ Fisher information}

\label{subsubsec:RLD-fish}

The quantum Cramer--Rao bound (QCRB) provides a technique to bound the MSE in
estimating a parameter by using the SLD Fisher information. As mentioned previously, there is in fact
an infinite number of QCRBs, with each of them arising from a particular
noncommutative generalization of the classical Fisher information in
\eqref{eq:CFI}. Another noncommutative generalization of the classical Fisher
information is the \textit{right} logarithmic derivative (RLD) Fisher information:

\begin{definition}
[RLD\ Fisher information]\label{def:RLD-Fish-info-states}Let $\{\rho_{\theta
}\}_{\theta}$ be a differentiable family of quantum states. Then the
RLD\ Fisher information is defined as follows:
\begin{equation}
\widehat{I}_{F}(\theta;\{\rho_{\theta}\}_{\theta})=\left\{
\begin{array}
[c]{cc}
\operatorname{Tr}[(\partial_{\theta}\rho_{\theta})^{2}\rho_{\theta}^{-1}] &
\text{if }\operatorname{supp}(\partial_{\theta}\rho_{\theta})\subseteq
\operatorname{supp}(\rho_{\theta})\\
+\infty & \text{otherwise}
\end{array}
\right.  , \label{eq:RLD-FI}
\end{equation}
where the inverse $\rho_{\theta}^{-1}$ is taken on the support of
$\rho_{\theta}$.
\end{definition}

Note that the support condition $\operatorname{supp}(\partial_{\theta}\rho_{\theta})\subseteq
\operatorname{supp}(\rho_{\theta})$ is equivalent to $\Pi^{\perp}_{\rho_\theta} \partial_\theta \rho_\theta =  \partial_\theta \rho_\theta \Pi^{\perp}_{\rho_\theta} = 0$, which implies that $\Pi^{\perp}_{\rho_\theta}\partial_\theta \rho_\theta \Pi^{\perp}_{\rho_\theta} = 0$.

For a differentiable family $\{|\varphi_{\theta}\rangle\!\langle\varphi_{\theta
}|\}_{\theta}$ of pure states, the RLD\ Fisher information has trivial
behavior due to the finiteness condition in \eqref{eq:RLD-FI}. If the family
is constant, such that $|\varphi_{\theta}\rangle=|\varphi\rangle$ for all
$\theta$, then the RLD\ Fisher information is finite and equal to zero.
Otherwise, the RLD\ Fisher information is infinite. We show this in more
detail in Appendix~\ref{sec:pure-state-fam-examps}.\ Thus, the RLD\ Fisher
information is a degenerate and uninteresting information measure for
pure-state families.

Similar to Proposition~\ref{prop:physical-consistency-SLD-Fish-states}, the
following proposition demonstrates that the definition in \eqref{eq:RLD-FI} is
physically consistent, in the sense that it is the result of a limiting
procedure in which some constant additive noise vanishes:

\begin{proposition}
\label{prop:physical-consistency-RLD-Fish-states}Let $\{\rho_{\theta
}\}_{\theta}$ be a differentiable family of quantum states. Then the
RLD\ Fisher information in \eqref{eq:RLD-FI}\ is given by the following limit:
\begin{equation}
\widehat{I}_{F}(\theta;\{\rho_{\theta}\}_{\theta})=\lim_{\varepsilon
\rightarrow0}\widehat{I}_{F}(\theta;\{\rho_{\theta}^{\varepsilon}\}_{\theta}),
\end{equation}
where
\begin{equation}
\rho_{\theta}^{\varepsilon}:=\left(  1-\varepsilon\right)  \rho_{\theta
}+\varepsilon\pi_{d},
\end{equation}
and $\pi_{d}:=I/d$ is the maximally mixed state, with $d$ large enough so that
$\operatorname{supp}(\rho_{\theta})\subseteq\operatorname{supp}(\pi)$ for all
$\theta$.
\end{proposition}

\begin{proof}
See Appendix~\ref{app:phys-cons-SLD-RLD-Fish-states}.
\end{proof}

\medskip

In the case that the following support condition holds
\begin{equation}
\operatorname{supp}(\partial_{\theta}\rho_{\theta})\subseteq
\operatorname{supp}(\rho_{\theta}), \label{eq:support-condition-RLD}
\end{equation}
then the RLD\ Fisher information can also be defined in the following way:
\begin{equation}
\widehat{I}_{F}(\theta;\{\rho_{\theta}\}_{\theta}):=\operatorname{Tr}[R_{\theta
}R_{\theta}^{\dag}\rho_{\theta}]=\operatorname{Tr}[(\partial_{\theta}
\rho_{\theta})R_{\theta}^{\dag}],
\end{equation}
where the RLD\ operator \cite{YL73}\ is defined through the following
differential equation:
\begin{equation}
\partial_{\theta}\rho_{\theta}=\rho_{\theta}R_{\theta}. \label{eq:RLD-diff-eq}
\end{equation}
By observing from \eqref{eq:RLD-diff-eq} that $\Pi_{\rho_{\theta}}R_{\theta
}=\rho_{\theta}^{-1}\partial_{\theta}\rho_{\theta}$, where $\Pi_{\rho_{\theta
}}$ is the projection onto the support of $\rho_{\theta}$, the RLD\ Fisher
information can be written explicitly as $\widehat{I}_{F}(\theta
;\{\rho_{\theta}\}_{\theta}):=\operatorname{Tr}[(\partial_{\theta}\rho
_{\theta})^{2}\rho_{\theta}^{-1}]$, consistent with
Definition~\ref{def:RLD-Fish-info-states}. This formula is thus a more direct
quantum generalization of the classical formula in \eqref{eq:CFI-for-RLD}.

The SLD Fisher information never exceeds the RLD Fisher information:
\begin{equation}
I_{F}(\theta;\{\rho_{\theta}\}_{\theta})\leq\widehat{I}_{F}(\theta
;\{\rho_{\theta}\}_{\theta}), \label{eq:SLD<=RLD}
\end{equation}
which can be seen from the operator convexity of the function $x^{-1}$ for
$x>0$. That is, for full-rank $\rho_{\theta}$, we have that
\begin{align}
2\left(  \rho_{\theta}\otimes I+I\otimes\rho_{\theta}^{T}\right)  ^{-1}  &
=\left(  \frac{1}{2}\rho_{\theta}\otimes I+\frac{1}{2}I\otimes\rho_{\theta
}^{T}\right)  ^{-1}\\
&  \leq\frac{1}{2}\left(  \rho_{\theta}\otimes I\right)  ^{-1}+\frac{1}
{2}\left(  I\otimes\rho_{\theta}^{T}\right)  ^{-1}\\
&  =\frac{1}{2}\left(  \rho_{\theta}^{-1}\otimes I\right)  +\frac{1}{2}\left(
I\otimes\rho_{\theta}^{-T}\right)  ,
\end{align}
and then \eqref{eq:basis-independent-formula-SLD}, \eqref{eq:transpose-trick},
\eqref{eq:max-ent-partial-trace}, and the limit formulas in
Propositions~\ref{prop:physical-consistency-SLD-Fish-states} and
\ref{prop:physical-consistency-RLD-Fish-states}\ lead to \eqref{eq:SLD<=RLD}.
Thus, as a consequence of \eqref{eq:QCRB}\ and \eqref{eq:SLD<=RLD}, the
RLD\ Fisher information leads to another lower bound on the MSE of an unbiased
estimator:
\begin{equation}
\text{Var}(\hat{\theta})\geq\frac{1}{n\widehat{I}_{F}(\theta;\{\rho_{\theta
}\}_{\theta})}. \label{eq:RLD-QCRB}
\end{equation}
Although the inequality above is not generally achievable, the RLD\ Fisher
information possesses an operational meaning in terms of a task called reverse
estimation \cite{Matsumoto2005}.

The formula in \eqref{eq:RLD-FI} may be difficult to evaluate in practice due
to the presence of a matrix inverse. In Section~\ref{sec:SDPs-for-FIs}, we
show how this quantity can be evaluated by means of a semi-definite program
that takes $\rho_{\theta}$ and $\partial_{\theta}\rho_{\theta}$ as input, thus
obviating the need to perform the inverse.

\subsection{Basic properties of SLD\ and RLD\ Fisher information of quantum
states}

Here we collect some basic properties of SLD\ and RLD\ Fisher information of
quantum states, which include faithfulness, data processing, additivity, and
decomposition on classical--quantum states.

\subsubsection{Faithfulness}

\begin{proposition}
[Faithfulness]\label{prop:faithfulness-SLD-RLD-Fish}For a differentiable
family $\{\rho_{A}^{\theta}\}_{\theta}$ of quantum states, the SLD\ and
RLD\ Fisher informations are equal to zero:
\begin{equation}
I_{F}(\theta;\{\rho_{A}\}_{\theta})=\widehat{I}_{F}(\theta;\{\rho
_{A}\}_{\theta})=0\qquad\forall\theta\in\Theta,
\label{eq:SLD-RLD-Fish-faithful}
\end{equation}
if and only if $\rho_{A}^{\theta}$ has no dependence on the parameter $\theta$
(i.e., $\rho_{A}^{\theta}=\rho_{A}$ for all $\theta$).
\end{proposition}

\begin{proof}
The if-part follows directly from plugging into the definitions after
observing that $\partial_{\theta}\rho_{\theta}=0$ for a constant family. So we
now prove the only-if part. If $I_{F}(\theta;\{\rho_{A}\}_{\theta})=0$, then
it is necessary for the finiteness condition in
\eqref{eq:finiteness-condition-SLD-Fish-states}\ to hold (otherwise we would
have a contradiction). Then this means that
\begin{align}
\Pi_{\rho_{\theta}}^{\perp}(\partial_{\theta}\rho_{\theta})\Pi_{\rho_{\theta}
}^{\perp}  &  =0,\\
2\sum_{j,k:\lambda_{j}^{\theta}+\lambda_{k}^{\theta}>0}\frac{|\langle
\psi_{\theta}^{j}|(\partial_{\theta}\rho_{\theta})|\psi_{\theta}^{k}
\rangle|^{2}}{\lambda_{\theta}^{j}+\lambda_{\theta}^{k}}  &  =0\qquad
\forall\theta.
\end{align}
By sandwiching the first equation by $\langle\psi_{\theta}^{j}|$ and
$|\psi_{\theta}^{k}\rangle$ for which $\lambda_{\theta}^{j},\lambda_{\theta
}^{k}=0$, we find that these matrix elements $\langle\psi_{\theta}
^{j}|(\partial_{\theta}\rho_{\theta})|\psi_{\theta}^{k}\rangle$\ of
$\partial_{\theta}\rho_{\theta}$ are equal to zero. Since $\lambda_{\theta
}^{j}+\lambda_{\theta}^{k}>0$ in the latter expression, the latter equality
implies the following
\begin{equation}
|\langle\psi_{\theta}^{j}|(\partial_{\theta}\rho_{\theta})|\psi_{\theta}
^{k}\rangle|^{2}=0
\end{equation}
for all $\lambda_{\theta}^{j}$ and $\lambda_{\theta}^{k}$ satisfying
$\lambda_{\theta}^{j}+\lambda_{\theta}^{k}>0$. This implies that these matrix
elements $\langle\psi_{\theta}^{j}|(\partial_{\theta}\rho_{\theta}
)|\psi_{\theta}^{k}\rangle$ of $\partial_{\theta}\rho_{\theta}$ are equal to
zero. These are all possible matrix elements, and so we conclude that
$\partial_{\theta}\rho_{\theta}=0$. This in turn implies that $\rho_{\theta}$
is a constant family (i.e., $\rho_{A}^{\theta}=\rho_{A}$ for all $\theta$). If
$\widehat{I}_{F}(\theta;\{\rho_{A}\}_{\theta})=0$, then by the inequality in
\eqref{eq:SLD<=RLD}, $I_{F}(\theta;\{\rho_{A}\}_{\theta})=0$. Then by what we
have just shown, $\rho_{\theta}$ is a constant family in this case also.
\end{proof}

\subsubsection{Data processing}

The SLD\ and RLD\ Fisher informations obey the following data-processing
inequalities:
\begin{align}
I_{F}(\theta;\{\rho_{A}^{\theta}\}_{\theta})  &  \geq I_{F}(\theta
;\{\mathcal{N}_{A\rightarrow B}(\rho_{A}^{\theta})\}_{\theta}
),\label{eq:DP-SLD}\\
\widehat{I}_{F}(\theta;\{\rho_{A}^{\theta}\}_{\theta})  &  \geq\widehat{I}
_{F}(\theta;\{\mathcal{N}_{A\rightarrow B}(\rho_{A}^{\theta})\}_{\theta}),
\label{eq:DP-RLD}
\end{align}
where $\mathcal{N}_{A\rightarrow B}$ is a quantum channel independent of the
parameter $\theta$ (more generally, these hold if $\mathcal{N}_{A\rightarrow
B}$ is a two-positive, trace-preserving map). The data-processing inequalities
for $I_{F}$ and $\widehat{I}_{F}$ were established in \cite{Petz96}. In fact,
the inequality in \eqref{eq:DP-RLD}\ is an immediate consequence of
\cite[Proposition~4.1]{Choi80}.

\subsubsection{Additivity}

\begin{proposition}
\label{prop:additivity-SLD-RLD-states}Let $\{\rho_{A}^{\theta}\}_{\theta}$ and
$\{\sigma_{A}^{\theta}\}_{\theta}$ be differentiable families of quantum
states. Then the SLD\ and RLD\ Fisher informations are additive in the
following sense:
\begin{align}
I_{F}(\theta;\{\rho_{A}^{\theta}\otimes\sigma_{B}^{\theta}\}_{\theta})  &
=I_{F}(\theta;\{\rho_{A}^{\theta}\}_{\theta})+I_{F}(\theta;\{\sigma
_{B}^{\theta}\}_{\theta}),\label{eq:additivity-SLD-states}\\
\widehat{I}_{F}(\theta;\{\rho_{A}^{\theta}\otimes\sigma_{B}^{\theta}
\}_{\theta})  &  =\widehat{I}_{F}(\theta;\{\rho_{A}^{\theta}\}_{\theta
})+\widehat{I}_{F}(\theta;\{\sigma_{B}^{\theta}\}_{\theta}).
\label{eq:additivity-RLD-states}
\end{align}

\end{proposition}

\begin{proof}
See Appendix~\ref{app:additivity-SLD-RLD}.
\end{proof}

\subsubsection{Decomposition for classical--quantum families}

\begin{proposition}
\label{prop:cq-decomp-SLD-RLD}Let $\left\{  \rho_{XB}^{\theta}\right\}
_{\theta}$ be a differentiable family of classical--quantum states, where
\begin{equation}
\rho_{XB}^{\theta}:=\sum_{x}p_{\theta}(x)|x\rangle\!\langle x|_{X}\otimes
\rho_{\theta}^{x}.
\end{equation}
Then the following decompositions hold for the SLD\ and RLD\ Fisher
informations:
\begin{align}
I_{F}(\theta;\{\rho_{XB}^{\theta}\}_{\theta})  &  =I_{F}(\theta;\{p_{\theta
}\}_{\theta})+\sum_{x}p_{\theta}(x)I_{F}(\theta;\{\rho_{\theta}^{x}\}_{\theta
}),\label{eq:decomp-SLD-cq}\\
\widehat{I}_{F}(\theta;\{\rho_{XB}^{\theta}\}_{\theta})  &  =I_{F}
(\theta;\{p_{\theta}\}_{\theta})+\sum_{x}p_{\theta}(x)\widehat{I}_{F}
(\theta;\{\rho_{\theta}^{x}\}_{\theta}).
\end{align}

\end{proposition}

\begin{proof}
See Appendix~\ref{app:decomp-fisher-cqq}.
\end{proof}

\medskip
We note here that the extended convexity inequality reported in \cite[Eq.~(4)]{PhysRevA.91.042104} is a consequence of \eqref{eq:decomp-SLD-cq}. That is, one recovers the extended convexity inequality of \cite{PhysRevA.91.042104} by performing a partial trace over the classical register $X$ on the left-hand side of  \eqref{eq:decomp-SLD-cq} and applying the  data-processing inequality in \eqref{eq:DP-SLD}.

\subsection{Generalized Fisher information and a meta-converse for channel
parameter estimation} \label{subsec:generalized-fisher-info}

Since data processing is such a fundamental and powerful tool, it can be
fruitful to define and develop a generalized distinguishability measure based
on this property alone (this is also called generalized divergence
\cite{PV10,SW12}). This approach has been employed for some time now in
quantum communication
\cite{SW12,WWY14,GW15,TWW17,WTB16,Led16,KW17,PhysRevA.101.012344,PhysRevLett.123.070502,Fang2019a,WWW19}
and distinguishability \cite{TW16,LKDW18,Berta2018c,PhysRevResearch.1.033169}
theory. Here we extend the approach to quantum estimation theory.

\subsubsection{Generalized Fisher information of states}

Let $\mathcal{D}$ denote the set of density operators and $\Theta$ the
parameter set. We define the generalized Fisher information of quantum states
as follows:

\begin{definition}
[Generalized Fisher information of quantum states]
\label{def:gen-fish-info-states}The generalized Fisher information
$\mathbf{I}_{F}(\theta;\{\rho_{A}^{\theta}\}_{\theta})$ of a family
$\{\rho_{A}^{\theta}\}_{\theta}$ of quantum states is a function
$\mathbf{I}_{F} : \Theta\times\mathcal{D} \to\mathbb{R}$ that does not
increase under the action of a parameter-independent quantum channel~$\mathcal{N}_{A\rightarrow B}$:
\begin{equation}
\mathbf{I}_{F}(\theta;\{\rho_{A}^{\theta}\}_{\theta})\geq\mathbf{I}_{F}
(\theta;\{\mathcal{N}_{A\rightarrow B}(\rho_{A}^{\theta})\}_{\theta}).
\label{eq:gen-fisher-states}
\end{equation}

\end{definition}

It follows from \eqref{eq:DP-SLD} and \eqref{eq:DP-RLD} that the SLD\ and
RLD\ Fisher informations in \eqref{eq:basis-independent-formula-SLD} and \eqref{eq:RLD-FI} are
particular examples because they possess this basic property. Furthermore, the
generalized divergence of \cite{PV10,SW12}\ is a special case of generalized
Fisher information when the parameter $\theta$ takes on only two values.

An immediate consequence of Definition~\ref{def:gen-fish-info-states} is that
the generalized Fisher information is equal to a constant, minimal value for a
state family that has no dependence on the parameter $\theta$:
\begin{equation}
\mathbf{I}_{F}(\theta;\{\rho_{A}\}_{\theta})=c.
\label{eq:constant-value-gen-Fish-states}
\end{equation}
This follows because one can get from one fixed family $\{\rho_{A}\}_{\theta}$
to another $\{\sigma_{A}\}_{\theta}$ by means of a trace and replace channel
$(\cdot)\rightarrow\operatorname{Tr}[(\cdot)]\sigma_{A}$, and then we apply
the data-processing inequality. If this constant $c$\ is equal to zero, then
we say that the generalized Fisher information is \textit{weakly faithful}.

A generalized Fisher information obeys the \textit{direct-sum property} if the
following equality holds
\begin{equation}
\mathbf{I}_{F}\!\left(  \theta;\left\{  \sum_{x}p(x)|x\rangle\!\langle
x|\otimes\rho_{\theta}^{x}\right\}  _{\theta}\right)  =\sum_{x}p(x)\mathbf{I}
_{F}(\theta;\left\{  \rho_{\theta}^{x}\right\}  _{\theta}),
\label{eq:direct-sum-prop-gen-fish}
\end{equation}
where, for each $x$, the family $\left\{  \rho_{\theta}^{x}\right\}  _{\theta
}$ of quantum states is differentiable. Observe that the probability
distribution $p(x)$ has no dependence on the parameter $\theta$. If a
generalized Fisher information obeys the direct-sum property, then it is also
convex in the following sense:
\begin{equation}
\sum_{x}p(x)\mathbf{I}_{F}(\theta;\left\{  \rho_{\theta}^{x}\right\}
_{\theta})\geq\mathbf{I}_{F}(\theta;\left\{  \overline{\rho}_{\theta}\right\}
_{\theta}), \label{eq:gen-fish-convex}
\end{equation}
where $\overline{\rho}_{\theta}:=\sum_{x}p(x)\rho_{\theta}^{x}$. This follows
by applying \eqref{eq:direct-sum-prop-gen-fish} and the data-processing inequality with
a partial trace over the classical register. Thus, due to \eqref{eq:DP-SLD},
\eqref{eq:DP-RLD}, and Proposition~\ref{prop:cq-decomp-SLD-RLD}, the SLD and
RLD Fisher informations are convex.

\subsubsection{Generalized Fisher information of channels}

From the generalized Fisher information of states, we can define the
generalized Fisher information of channels:

\begin{definition}
[Generalized Fisher information of quantum channels]
\label{def:gen-fish-channels}The generalized Fisher information of a family
$\{\mathcal{N}_{A\rightarrow B}^{\theta}\}_{\theta}$ of quantum channels is
defined in terms of the following optimization:
\begin{equation}
\mathbf{I}_{F}(\theta;\{\mathcal{N}_{A\rightarrow B}^{\theta}\}_{\theta
}):=\sup_{\rho_{RA}}\mathbf{I}_{F}(\theta;\{\mathcal{N}_{A\rightarrow
B}^{\theta}(\rho_{RA})\}_{\theta}). \label{eq:gen-fisher-channels}
\end{equation}
In the above definition, we take the supremum over arbitrary states $\rho
_{RA}$\ with unbounded reference system $R$.
\end{definition}

The SLD\ Fisher information of quantum channels was defined in \cite{fuji2001}
and the RLD\ Fisher information of quantum channels in \cite{Hayashi2011}; these are special cases of \eqref{eq:gen-fisher-channels}. The
generalized channel divergence of \cite{CMW14,LKDW18}\ is a special case of
generalized Fisher information of channels when the parameter $\theta$ takes
on only two values.

\begin{remark}
\label{rem:restrict-to-pure-bipartite}As is the case for all information
measures that obey the data-processing inequality, we can employ the
data-processing inequality in \eqref{eq:gen-fisher-states} with respect to
partial trace and the Schmidt decomposition theorem to conclude that it
suffices to perform the optimization in \eqref{eq:gen-fisher-channels} with
respect to pure bipartite states $\psi_{RA}$ with system $R$ isomorphic to
system $A$, so that
\begin{equation}
\mathbf{I}_{F}(\theta;\{\mathcal{N}_{A\rightarrow B}^{\theta}\}_{\theta}
)=\sup_{\psi_{RA}}\mathbf{I}_{F}(\theta;\{\mathcal{N}_{A\rightarrow B}
^{\theta}(\psi_{RA})\}_{\theta}).
\end{equation}

\end{remark}

Some basic properties of the generalized Fisher information of quantum
channels are as follows:

\begin{proposition}
Let $\{\mathcal{N}_{A\rightarrow B}\}_{\theta}$ be a family of quantum
channels that has no dependence on the parameter $\theta$, and suppose that
the underlying generalized Fisher information is weakly faithful. Then
\begin{equation}
\mathbf{I}_{F}(\theta;\{\mathcal{N}_{A\rightarrow B}\}_{\theta})=0.
\end{equation}

\end{proposition}

\begin{proof}
This follows as an immediate consequence of the definition,
\eqref{eq:constant-value-gen-Fish-states}, and the weak faithfulness assumption.
\end{proof}

\begin{proposition}
[Reduction to states]Let $\{\rho_{B}^{\theta}\}_{\theta}$ be a family of
quantum states, and define the family $\{\mathcal{R}_{A\rightarrow B}^{\theta
}\}_{\theta}$ of replacer channels as
\begin{equation}
\mathcal{R}_{A\rightarrow B}^{\theta}(\omega_{A})=\operatorname{Tr}[\omega
_{A}]\rho_{B}^{\theta}.
\end{equation}
Then
\begin{equation}
\mathbf{I}_{F}(\theta;\{\mathcal{R}_{A\rightarrow B}^{\theta}\}_{\theta
})=\mathbf{I}_{F}(\theta;\{\rho_{B}^{\theta}\}_{\theta}).
\end{equation}

\end{proposition}

\begin{proof}
This follows from the definition and the data-processing inequality. Consider
that
\begin{align}
\mathbf{I}_{F}(\theta;\{\mathcal{R}_{A\rightarrow B}^{\theta}\}_{\theta})  &
=\sup_{\psi_{RA}}\mathbf{I}_{F}(\theta;\{\mathcal{R}_{A\rightarrow B}^{\theta
}(\psi_{RA})\}_{\theta})\\
&  =\sup_{\psi_{RA}}\mathbf{I}_{F}(\theta;\{\psi_{R}\otimes\rho_{B}^{\theta
}\}_{\theta})\\
&  =\mathbf{I}_{F}(\theta;\{\rho_{B}^{\theta}\}_{\theta}).
\end{align}
The last equality follows because
\begin{align}
\mathbf{I}_{F}(\theta;\{\rho_{B}^{\theta}\}_{\theta})  &  \geq\mathbf{I}
_{F}(\theta;\{\psi_{R}\otimes\rho_{B}^{\theta}\}_{\theta}),\\
\mathbf{I}_{F}(\theta;\{\rho_{B}^{\theta}\}_{\theta})  &  \leq\mathbf{I}
_{F}(\theta;\{\psi_{R}\otimes\rho_{B}^{\theta}\}_{\theta}),
\end{align}
with the first inequality following from the fact that there is a
parameter-independent preparation channel such that $\rho_{B}^{\theta
}\rightarrow\psi_{R}\otimes\rho_{B}^{\theta}$, while the second inequality
follows from data-processing under partial trace over the reference system $R$.
\end{proof}

\begin{proposition}
\label{prop:choi-state-bound-gen-fish}Let $\{\mathcal{N}_{A\rightarrow
B}^{\theta}\}_{\theta}$ be a family of quantum channels, and suppose that the
underlying generalized Fisher information is weakly faithful and obeys the
direct-sum property. Then the following inequalities hold
\begin{equation}
\mathbf{I}_{F}(\theta;\{\mathcal{N}_{A\rightarrow B}^{\theta}(\Phi
_{RA})\}_{\theta})\leq\mathbf{I}_{F}(\theta;\{\mathcal{N}_{A\rightarrow
B}^{\theta}\}_{\theta})\leq d\cdot\mathbf{I}_{F}(\theta;\{\mathcal{N}
_{A\rightarrow B}^{\theta}(\Phi_{RA})\}_{\theta}),
\label{eq:choi-state-bound-gen-fish}
\end{equation}
where $\Phi_{RA}$ is the maximally entangled state and $d$ is the dimension of
the channel input system $A$.
\end{proposition}

\begin{proof}
The first inequality is trivial, following from the definition in
\eqref{eq:gen-fisher-channels}. So we prove the second one and note that it
follows from a quantum steering or remote state preparation argument. Let
$\psi_{RA}$ be an arbitrary pure bipartite input state. To each such state,
there exists an operator $Z_{R}$ satisfying
\begin{align}
\psi_{RA}  &  =d\cdot Z_{R}\Phi_{RA}Z_{R}^{\dag},\\
\operatorname{Tr}[Z_{R}^{\dag}Z_{R}]  &  =1.
\end{align}
Let $\mathcal{P}_{R\rightarrow XR}$ denote the following steering quantum
channel:
\begin{equation}
\mathcal{P}_{R\rightarrow XR}(\omega_{R}):=|0\rangle\!\langle0|_{X}\otimes
Z_{R}\omega_{R}Z_{R}^{\dag}+|1\rangle\!\langle1|_{X}\otimes\sqrt{I_{R}
-Z_{R}^{\dag}Z_{R}}\omega_{R}\sqrt{I_{R}-Z_{R}^{\dag}Z_{R}},
\end{equation}
and consider that
\begin{equation}
\mathcal{P}_{R\rightarrow XR}(\Phi_{RA})=\frac{1}{d}|0\rangle\!\langle
0|_{X}\otimes\psi_{RA}+\left(  1-\frac{1}{d}\right)  |1\rangle\!\langle
1|_{X}\otimes\sigma_{RA},
\end{equation}
where
\begin{equation}
\sigma_{RA}:=\left(  1-\frac{1}{d}\right)  ^{-1}\sqrt{I_{R}-Z_{R}^{\dag}Z_{R}
}\Phi_{RA}\sqrt{I_{R}-Z_{R}^{\dag}Z_{R}}.
\end{equation}
This implies that
\begin{align}
&  \mathcal{P}_{R\rightarrow XR}(\mathcal{N}_{A\rightarrow B}^{\theta}
(\Phi_{RA}))\nonumber\\
&  =\mathcal{N}_{A\rightarrow B}^{\theta}(\mathcal{P}_{R\rightarrow XR}
(\Phi_{RA}))\\
&  =\frac{1}{d}|0\rangle\!\langle0|_{X}\otimes\mathcal{N}_{A\rightarrow
B}^{\theta}(\psi_{RA})+\left(  1-\frac{1}{d}\right)  |1\rangle\!\langle
1|_{X}\otimes\mathcal{N}_{A\rightarrow B}^{\theta}(\sigma_{RA}).
\label{eq:steering-ch-dim-bound}
\end{align}
Then we find that
\begin{align}
&  \mathbf{I}_{F}(\theta;\{\mathcal{N}_{A\rightarrow B}^{\theta}(\Phi
_{RA})\}_{\theta})\nonumber\\
&  \geq\mathbf{I}_{F}(\theta;\{\mathcal{P}_{R\rightarrow XR}(\mathcal{N}
_{A\rightarrow B}^{\theta}(\Phi_{RA}))\}_{\theta})\\
&  =\frac{1}{d}\mathbf{I}_{F}(\theta;\{\mathcal{N}_{A\rightarrow B}^{\theta
}(\psi_{RA})\}_{\theta})+\left(  1-\frac{1}{d}\right)  \mathbf{I}_{F}
(\theta;\{\mathcal{N}_{A\rightarrow B}^{\theta}(\sigma_{RA})\}_{\theta})\\
&  \geq\frac{1}{d}\mathbf{I}_{F}(\theta;\{\mathcal{N}_{A\rightarrow B}
^{\theta}(\psi_{RA})\}_{\theta}).
\end{align}
The first inequality follows from data processing. The equality follows from
\eqref{eq:steering-ch-dim-bound}\ and the direct-sum property in
\eqref{eq:direct-sum-prop-gen-fish}. The last inequality follows from the
assumption that $\mathbf{I}_{F}$ is weakly faithful, so that $\mathbf{I}
_{F}(\theta;\{\mathcal{N}_{A\rightarrow B}^{\theta}(\sigma_{RA})\}_{\theta
})\geq0$. Since the inequality holds for all pure bipartite states $\psi_{RA}$, we conclude the second inequality in \eqref{eq:choi-state-bound-gen-fish}.
\end{proof}

\begin{remark}
Note that a special case of \eqref{eq:choi-state-bound-gen-fish} occurs when
the parameter $\theta$ takes on only two values. So the argument above applies
to all generalized channel divergences \cite{LKDW18}\ that are weakly faithful
and obey the direct-sum property, which includes diamond distance, relative
entropy, negative root fidelity, and Petz-,  sandwiched, and geometric R\'{e}nyi relative quasi-entropies.
\end{remark}

\begin{remark}
\label{rem:finiteness-cond-gen-fish}Supposing that a generalized Fisher
information is weakly faithful and obeys the direct-sum property, a
consequence of Proposition~\ref{prop:choi-state-bound-gen-fish}\ is that, in
order to determine whether the corresponding generalized Fisher information of
channels is finite, it is only necessary to check the value of the quantity on
the maximally entangled input state.
\end{remark}

Particular generalized Fisher informations of channels of interest include the
SLD\ and RLD\ ones. Due to \eqref{eq:DP-SLD}--\eqref{eq:DP-RLD},
Propositions~\ref{prop:faithfulness-SLD-RLD-Fish},
\ref{prop:cq-decomp-SLD-RLD}, and \ref{prop:choi-state-bound-gen-fish}, and
Remark~\ref{rem:finiteness-cond-gen-fish}, we can write them respectively as
follows:
\begin{equation}
I_{F}(\theta;\{\mathcal{N}_{A\rightarrow B}^{\theta}\}_{\theta})=\left\{
\begin{array}
[c]{cc}
\sup_{\psi_{RA}}I_{F}(\theta;\{\mathcal{N}_{A\rightarrow B}^{\theta}(\psi
_{RA})\}_{\theta}) & \text{if }\Pi_{\Gamma_{RB}^{\mathcal{N}^{\theta}}}
^{\perp}(\partial_{\theta}\Gamma_{RB}^{\mathcal{N}^{\theta}})\Pi_{\Gamma
_{RB}^{\mathcal{N}^{\theta}}}^{\perp}=0\\
+\infty & \text{otherwise.}
\end{array}
\right.  , \label{eq:finiteness-condition-SLD-fish-ch}
\end{equation}
\begin{multline}
\widehat{I}_{F}(\theta;\{\mathcal{N}_{A\rightarrow B}^{\theta}\}_{\theta
})=\label{eq:finiteness-condition-RLD-fish-ch}\\
\left\{
\begin{array}
[c]{cc}
\left\Vert \operatorname{Tr}_{B}[(\partial_{\theta}\Gamma_{RB}^{\mathcal{N}
^{\theta}})(\Gamma_{RB}^{\mathcal{N}^{\theta}})^{-1}(\partial_{\theta}
\Gamma_{RB}^{\mathcal{N}^{\theta}})]\right\Vert _{\infty} & \text{if
}\operatorname{supp}(\partial_{\theta}\Gamma_{RB}^{\mathcal{N}^{\theta}
})\subseteq\operatorname{supp}(\Gamma_{RB}^{\mathcal{N}^{\theta}})\\
+\infty & \text{otherwise.}
\end{array}
\right.  ,
\end{multline}
where $\Gamma^{\mathcal{N}^{\theta}}_{RB}$ is the Choi operator of the channel
$\mathcal{N}^{\theta}_{A\to B}$. The explicit expression above for
$\widehat{I}_{F}(\theta;\{\mathcal{N}_{A\rightarrow B}^{\theta}\}_{\theta})$
was given in \cite{Hayashi2011}\ and is recalled in
Proposition~\ref{prop:geo-fish-explicit-formula}\ below. It is unclear to us
at the moment how to obtain a more explicit form for $I_{F}(\theta
;\{\mathcal{N}_{A\rightarrow B}^{\theta}\}_{\theta})$ in terms of its Choi operator.

The finiteness conditions in \eqref{eq:finiteness-condition-SLD-fish-ch} and \eqref{eq:finiteness-condition-RLD-fish-ch} have interesting implications for a differentiable family $\{\mathcal{U}_\theta\}_\theta$ of isometric or unitary channels. When such a family acts on one share of a maximally entangled state, it induces a differentiable family of pure states. Now applying what was stated previously in Sections~\ref{subsubsec:SLD-fish} and \ref{subsubsec:RLD-fish} for such families, it follows that the SLD Fisher information of  $\{\mathcal{U}_\theta\}_\theta$ is always finite, whereas the RLD Fisher information of  $\{\mathcal{U}_\theta\}_\theta$ is finite if and only if it is equal to zero (i.e., when the family $\{\mathcal{U}_\theta\}_\theta$ is a constant family $\{\mathcal{U}\}_\theta$ independent of the parameter $\theta$). So in this sense, the RLD Fisher information of isometric or unitary channels is a degenerate and uninteresting information measure.

\subsubsection{Amortized Fisher information}

The generalized Fisher information of quantum channels is motivated by channel
parameter estimation, and in particular, by the parallel setting of channel
estimation. Now motivated by the more general sequential setting of channel
parameter estimation, we define the following amortized Fisher information of
quantum channels:

\begin{definition}
[Amortized Fisher information of quantum channels]
\label{def:amort-Fish-ch}
The amortized Fisher
information of a family $\{\mathcal{N}_{A\rightarrow B}^{\theta}\}_{\theta}$
of quantum channels is defined as follows:
\begin{equation}
\mathbf{I}_{F}^{\mathcal{A}}(\theta;\{\mathcal{N}_{A\rightarrow B}^{\theta
}\}_{\theta}):=\sup_{\{\rho_{RA}^{\theta}\}_{\theta}}\left[  \mathbf{I}
_{F}(\theta;\{\mathcal{N}_{A\rightarrow B}^{\theta}(\rho_{RA}^{\theta
})\}_{\theta})-\mathbf{I}_{F}(\theta;\{\rho_{RA}^{\theta})\}_{\theta})\right]
, \label{eq:amortized-fisher-info}
\end{equation}
where the supremum is with respect to arbitrary state families $\{\rho
_{RA}^{\theta}\}_{\theta}$\ with unbounded reference system $R$.
\end{definition}

The idea behind this quantity is the same as that of the amortized channel
divergence of \cite{Berta2018c}.\ We allow for a resource at the channel input
in order to help with the estimation task, but then we subtract off the value
of this resource in order to account for the amount of resource that is strictly present in the channel family. In this case, the resource is estimability, as proposed in
\cite{Matsumoto2005}. This kind of idea has been useful in the analysis of
feedback-assisted or sequential protocols in other areas of quantum
information science
\cite{BHLS03,BGMW17,RKBKMA17,KW17,BW17,DW17,PhysRevResearch.1.033169,Fang2019a,Wang_2019}
, and here we see how it is useful in the context of channel parameter
estimation. Also, we should indicate here that the amortized channel
divergence of \cite{Berta2018c} is a special case of the amortized Fisher
information in which the parameter $\theta$ takes on only two values.

\begin{proposition}
\label{prop:amort->=-ch-Fish-gen}Let $\{\mathcal{N}_{A\rightarrow B}^{\theta
}\}_{\theta}$ be a family of quantum channels, and suppose that the underlying
generalized Fisher information is weakly faithful. Then the generalized Fisher
information does not exceed the amortized one:
\begin{equation}
\mathbf{I}_{F}^{\mathcal{A}}(\theta;\{\mathcal{N}_{A\rightarrow B}^{\theta
}\}_{\theta})\geq\mathbf{I}_{F}(\theta;\{\mathcal{N}_{A\rightarrow B}^{\theta
}\}_{\theta}). \label{eq:general-amort-ineq-obvi-dir}
\end{equation}

\end{proposition}

\begin{proof}
This follows because we can always pick the input family $\{\rho_{RA}^{\theta
}\}_{\theta}$ in \eqref{eq:amortized-fisher-info} to have no dependence on the
parameter $\theta$. Then we find that
\begin{align}
\mathbf{I}_{F}^{\mathcal{A}}(\theta;\{\mathcal{N}_{A\rightarrow B}^{\theta
}\}_{\theta})  &  \geq\mathbf{I}_{F}(\theta;\{\mathcal{N}_{A\rightarrow
B}^{\theta}(\rho_{RA})\}_{\theta})-\mathbf{I}_{F}(\theta;\{\rho_{RA}
)\}_{\theta})\\
&  =\mathbf{I}_{F}(\theta;\{\mathcal{N}_{A\rightarrow B}^{\theta}(\rho
_{RA})\}_{\theta}),
\end{align}
where we applied the weak faithfulness assumption to arrive at the equality.
Since the inequality holds for all input states $\rho_{RA}$, we conclude \eqref{eq:general-amort-ineq-obvi-dir}.
\end{proof}

\medskip

We now connect the amortized Fisher information to sequential channel
estimation through the following meta-converse, which generalizes the related
meta-converse of \cite{Berta2018c}:

\begin{theorem}
\label{thm:meta-converse}Consider a general sequential channel estimation
protocol of the form discussed in Section~\ref{sec:ch-param-est}. Suppose that
the generalized Fisher information $\mathbf{I}_{F}$ is weakly faithful. Then
the following inequality holds
\begin{equation}
\mathbf{I}_{F}(\theta;\{\omega_{R_{n}B_{n}}^{\theta}\}_{\theta})\leq
n\cdot \mathbf{I}_{F}^{\mathcal{A}}(\theta;\{\mathcal{N}_{A\rightarrow B}^{\theta
}\}_{\theta}),
\end{equation}
where $\omega_{R_{n}B_{n}}^{\theta}$ is the final state of the estimation
protocol, as given in \eqref{eq:estimation-final-state}.
\end{theorem}

\begin{proof}
Consider that
\begin{align}
&  \mathbf{I}_{F}(\theta;\{\omega_{R_{n}B_{n}}^{\theta}\}_{\theta})\nonumber\\
&  =\mathbf{I}_{F}(\theta;\{\omega_{R_{n}B_{n}}^{\theta}\}_{\theta
})-\mathbf{I}_{F}(\theta;\{\rho_{R_{1}A_{1}}\}_{\theta})\\
&  =\mathbf{I}_{F}(\theta;\{\omega_{R_{n}B_{n}}^{\theta}\}_{\theta
})-\mathbf{I}_{F}(\theta;\{\rho_{R_{1}A_{1}}\}_{\theta})+\sum_{i=2}^{n}\left(
\mathbf{I}_{F}(\theta;\{\rho_{R_{i}A_{i}}^{\theta}\}_{\theta})-\mathbf{I}
_{F}(\theta;\{\rho_{R_{i}A_{i}}^{\theta}\}_{\theta})\right) \\
&  =\mathbf{I}_{F}(\theta;\{\omega_{R_{n}B_{n}}^{\theta}\}_{\theta
})-\mathbf{I}_{F}(\theta;\{\rho_{R_{1}A_{1}}\}_{\theta})\nonumber\\
&  \qquad+\sum_{i=2}^{n}\left(  \mathbf{I}_{F}(\theta;\{\mathcal{S}
_{R_{i-1}B_{i-1}\rightarrow R_{i}A_{i}}^{i-1}(\rho_{R_{i-1}B_{i-1}}^{\theta
})\}_{\theta})-\mathbf{I}_{F}(\theta;\{\rho_{R_{i}A_{i}}^{\theta}\}_{\theta
})\right) \\
&  \leq\mathbf{I}_{F}(\theta;\{\omega_{R_{n}B_{n}}^{\theta}\}_{\theta
})-\mathbf{I}_{F}(\theta;\{\rho_{R_{1}A_{1}}\}_{\theta})\nonumber\\
&  \qquad+\sum_{i=2}^{n}\left(  \mathbf{I}_{F}(\theta;\{\rho_{R_{i-1}B_{i-1}
}^{\theta}\}_{\theta})-\mathbf{I}_{F}(\theta;\{\rho_{R_{i}A_{i}}^{\theta
}\}_{\theta})\right) \\
&  =\sum_{i=1}^{n}\left(  \mathbf{I}_{F}(\theta;\{\rho_{R_{i}B_{i}}^{\theta
}\}_{\theta})-\mathbf{I}_{F}(\theta;\{\rho_{R_{i}A_{i}}^{\theta}\}_{\theta
}\right) \\
&  =\sum_{i=1}^{n}\left(  \mathbf{I}_{F}(\theta;\{\mathcal{N}_{A_{i}
\rightarrow B_{i}}^{\theta}(\rho_{R_{i}A_{i}}^{\theta})\}_{\theta}
)-\mathbf{I}_{F}(\theta;\{\rho_{R_{i}A_{i}}^{\theta}\}_{\theta}\right) \\
&  \leq n\cdot\sup_{\{\rho_{RA}^{\theta}\}_{\theta}}\left[  \mathbf{I}
_{F}(\theta;\{\mathcal{N}_{A\rightarrow B}^{\theta}(\rho_{RA}^{\theta
})\}_{\theta})-\mathbf{I}_{F}(\theta;\{\rho_{RA}^{\theta})\}_{\theta})\right]
\\
&  =n\cdot \mathbf{I}_{F}^{\mathcal{A}}(\theta;\{\mathcal{N}_{A\rightarrow
B}^{\theta}\}_{\theta}).
\end{align}
The first equality follows from the weak faithfulness assumption and because
the initial state of the protocol has no dependence on the parameter $\theta$.
The inequality follows from the data-processing inequality. The other steps
are straightforward manipulations.
\end{proof}

\bigskip

For some particular choices of the generalized Fisher information, the
inequality in \eqref{eq:general-amort-ineq-obvi-dir}\ can be reversed, which is called an ``amortization collapse.'' 
Theorem~\ref{thm:meta-converse}\ makes such a collapse useful for establishing limits on the
performance of sequential estimation protocols if the underlying Fisher
information has a relation to the MSE\ through a CRB. We show later that the
following equalities hold for the root SLD and RLD Fisher informations for all differentiable
families $\{\mathcal{N}_{A\rightarrow B}^{\theta}\}_{\theta}$\ of quantum
channels:
\begin{align}
\sqrt{I_{F}}^{\mathcal{A}}(\theta;\{\mathcal{N}_{A\rightarrow B}^{\theta
}\}_{\theta}) & =\sqrt{I_{F}}(\theta;\{\mathcal{N}_{A\rightarrow B}^{\theta
}\}_{\theta}),
\\
\widehat{I}_{F}^{\mathcal{A}}(\theta;\{\mathcal{N}_{A\rightarrow B}^{\theta
}\}_{\theta}) & =\widehat{I}_{F}(\theta;\{\mathcal{N}_{A\rightarrow B}^{\theta
}\}_{\theta}).
\end{align}
Also, for differentiable families $\{\mathcal{N}_{X\rightarrow B}^{\theta
}\}_{\theta}$\ of classical--quantum channels, the following equality holds
for the SLD Fisher information:
\begin{equation}
I_{F}^{\mathcal{A}}(\theta;\{\mathcal{N}_{X\rightarrow B}^{\theta}\}_{\theta
})=I_{F}(\theta;\{\mathcal{N}_{X\rightarrow B}^{\theta}\}_{\theta}).
\end{equation}

\subsubsection{Environment-parameterized and environment-seizable channel families}

In this section, we recall the notion of environment-parameterized and
environment-seizable channel families, as discussed in \cite{TW16,Berta2018c,DW19}, and we show
that the amortized Fisher information collapses for environment-seizable
channel families. Environment-parameterized channel families are also known as programmable
channel families \cite{DP05}.

\begin{definition}
[Environment-parameterized family]A family $\{\mathcal{N}_{A\rightarrow
B}^{\theta}\}_{\theta}$ is called  environment-parameterized if
there exists a family $\{\rho_{E}^{\theta}\}_{\theta}$ of states and a
parameter-independent quantum channel $\mathcal{M}_{AE\rightarrow B}$ such
that the action of $\mathcal{N}_{A\rightarrow B}^{\theta}$ on any channel
input $\omega_{A}$ can be written as follows:%
\begin{equation}
\mathcal{N}_{A\rightarrow B}^{\theta}(\omega_{A})=\mathcal{M}_{AE\rightarrow
B}(\omega_{A}\otimes\rho_{E}^{\theta}).\label{eq:env-param-def}%
\end{equation}

\end{definition}

It is important to highlight that \textit{every} channel family is environment
parameterized in a trivial way, as discussed in \cite{DW19} for a finite set.
Indeed, set $\rho_{E}^{\theta}=|\theta\rangle\!\langle\theta|_{E}$, where the
vectors $\{|\theta\rangle_{E}\}_{\theta}$ are an orthonormal family, and set%
\begin{equation}
\mathcal{M}_{AE\rightarrow B}(\tau_{AE})=\int d\theta\ \mathcal{N}%
_{A\rightarrow B}^{\theta}(\langle\theta|_{E}\tau_{AE}|\theta\rangle_{E}).
\end{equation}
This simulation can be thought of as preparing a classical register $E$ with
the parameter value $\theta$, and then the parameter-independent channel
$\mathcal{M}_{AE\rightarrow B}$ observes the value $\theta$ in the classical
register and performs the channel $\mathcal{N}_{A\rightarrow B}^{\theta}$ on
the input system $A$. However, this construction is not useful for obtaining
upper bounds on the performance of channel families for quantum estimation,
because the classical Fisher information of the classical background family
$\{|\theta\rangle\!\langle\theta|_{E}\}_{\theta}$ is equal to infinity.

The notion of environment-parameterized channels only becomes interesting or
useful for obtaining bounds on the performance of channel estimation in the
case that the background environment states $\rho_{E}^{\theta}$ are not
perfectly distinguishable, as considered in \cite{Ji2008,Demkowicz-Dobrzanski2014,TW16}. That is, this concept is
only useful for obtaining bounds if the Fisher information of the state family
$\{\rho_{E}^{\theta}\}_{\theta}$ is finite. In a general sense, performance
bounds in the general sequential setting can be understood as being a
consequence of the following proposition:

\begin{proposition}
\label{prop:env-param-bound}Let $\{\mathcal{N}_{A\rightarrow B}^{\theta
}\}_{\theta}$ be an environment-parameterized channel family with associated
environment state family $\{\rho_{E}^{\theta}\}_{\theta}$. Suppose that the
underlying generalized Fisher information is subadditive on product-state
families. Then the amortized Fisher information obeys the following bound:%
\begin{equation}
\mathbf{I}_{F}^{\mathcal{A}}(\theta;\{\mathcal{N}_{A\rightarrow B}^{\theta
}\}_{\theta})\leq\mathbf{I}_{F}(\theta;\{\rho_{E}^{\theta}\}_{\theta}).
\label{eq:env-param-bound}
\end{equation}

\end{proposition}

\begin{proof}
Let $\{\omega_{RA}^{\theta}\}_{\theta}$ be an arbitrary input state family.
Then the following chain of inequalities holds%
\begin{align}
\mathbf{I}_{F}(\theta;\{\mathcal{N}_{A\rightarrow B}^{\theta}(\omega
_{RA}^{\theta})\}_{\theta})  & =\mathbf{I}_{F}(\theta;\{\mathcal{M}%
_{AE\rightarrow B}(\omega_{RA}^{\theta}\otimes\rho_{E}^{\theta})\}_{\theta
})\\
& \leq\mathbf{I}_{F}(\theta;\{\omega_{RA}^{\theta}\otimes\rho_{E}^{\theta
}\}_{\theta})\\
& \leq\mathbf{I}_{F}(\theta;\{\omega_{RA}^{\theta}\}_{\theta})+\mathbf{I}%
_{F}(\theta;\{\rho_{E}^{\theta}\}_{\theta}).
\end{align}
The equality follows by applying \eqref{eq:env-param-def}. The first
inequality follows from data processing, and the second inequality follows
from the assumption of subadditivity of $\mathbf{I}_{F}$ on product-state families. Since the inequality holds for an arbitrary state family $\{\omega_{RA}^{\theta}\}_{\theta}$, we conclude~\eqref{eq:env-param-bound}.
\end{proof}

\bigskip

Perhaps the most interesting case of environment-parameterized channel
families is when the environment states are seizable by a pre- and
post-processing of the channel \cite{Berta2018c,DW19}:

\begin{definition}
[Environment-seizable family]\label{def:env-seize-ch-fams}An
environment-parameterized channel family $\{\mathcal{N}_{A\rightarrow
B}^{\theta}\}_{\theta}$ with associated environment state family $\{\rho
_{E}^{\theta}\}_{\theta}$ is called environment seizable if there exists a
parameter-independent input state $\zeta_{RA}$ and post-processing channel
$\mathcal{D}_{RB\rightarrow E}$ that can be used to seize the background state
$\rho_{E}^{\theta}$ in the following sense:%
\begin{equation}
\mathcal{D}_{RB\rightarrow E}(\mathcal{N}_{A\rightarrow B}^{\theta}(\zeta
_{RA}))=\rho_{E}^{\theta}.
\end{equation}

\end{definition}

Simple examples of these channel families, along with simple
environment-seizing procedures, were discussed in \cite{Berta2018c}. These examples
include erasure and dephasing channels, with the underlying parameter being
the noise parameter of the channel.

As indicated by Definition~\ref{def:env-seize-ch-fams}, environment-seizable
channel families are fully identified with their background environment
states. That is, for such channel families, the most powerful procedure for
estimating them is to seize the background states first and then perform
processing on these background environment states. One way to formalize this
is with the following proposition:

\begin{proposition}
Let $\{\mathcal{N}_{A\rightarrow B}^{\theta}\}_{\theta}$ be an
environment-seizable channel family with associated environment state family
$\{\rho_{E}^{\theta}\}_{\theta}$. Suppose that the underlying generalized
Fisher information is subadditive on product-state families and weakly
faithful. Then the amortized Fisher information is equal to the generalized
Fisher information of the environment state family:%
\begin{equation}
\mathbf{I}_{F}^{\mathcal{A}}(\theta;\{\mathcal{N}_{A\rightarrow B}^{\theta
}\}_{\theta})=\mathbf{I}_{F}(\theta;\{\rho_{E}^{\theta}\}_{\theta}).
\end{equation}

\end{proposition}

\begin{proof}
The inequality $\leq$ was established by
Proposition~\ref{prop:env-param-bound}. To see the opposite inequality, pick
$\{\rho_{RA}^{\theta}\}_{\theta}$ in the definition of $\mathbf{I}%
_{F}^{\mathcal{A}}(\theta;\{\mathcal{N}_{A\rightarrow B}^{\theta}\}_{\theta})$
to be the parameter-independent family $\{\zeta_{RA}\}_{\theta}$. Then it
follows that%
\begin{align}
\mathbf{I}_{F}^{\mathcal{A}}(\theta;\{\mathcal{N}_{A\rightarrow B}^{\theta
}\}_{\theta})  & \geq\mathbf{I}_{F}(\theta;\{\mathcal{N}_{A\rightarrow
B}^{\theta}(\zeta_{RA})\}_{\theta})-\mathbf{I}_{F}(\theta;\{\zeta
_{RA}\}_{\theta})\\
& =\mathbf{I}_{F}(\theta;\{\mathcal{N}_{A\rightarrow B}^{\theta}(\zeta
_{RA})\}_{\theta})\\
& \geq\mathbf{I}_{F}(\theta;\{\mathcal{D}_{RB\rightarrow E}(\mathcal{N}%
_{A\rightarrow B}^{\theta}(\zeta_{RA}))\}_{\theta})\\
& =\mathbf{I}_{F}(\theta;\{\rho_{E}^{\theta}\}_{\theta}).
\end{align}
The first inequality follows from Definition~\ref{def:amort-Fish-ch}. The first equality
follows from the weak faithfulness assumption. The second inequality follows
from data processing. The final equality follows from
Definition~\ref{def:env-seize-ch-fams}.
\end{proof}

\bigskip

For these channel families, we can then employ the SLD\ Fisher information to
arrive at the following conclusion, the first part of which was already given
in \cite{Demkowicz-Dobrzanski2014}:

\begin{conclusion}
Let $\{\mathcal{N}_{A\rightarrow B}^{\theta}\}_{\theta}$ be an
environment-parameterized channel family with associated environment state
family $\{\rho_{E}^{\theta}\}_{\theta}$. As a direct consequence of the QCRB in
\eqref{eq:QCRB}, the meta-converse from Theorem~\ref{thm:meta-converse}, and the bound in
Proposition~\ref{prop:env-param-bound}, we conclude the following bound on the
MSE\ of an unbiased estimator $\hat{\theta}$ of $\theta$ that results from an
$n$-round sequential estimation protocol:%
\begin{equation}
\operatorname{Var}(\hat{\theta})\geq\frac{1}{nI_{F}(\theta;\{\rho_{E}^{\theta
}\}_{\theta})}.
\end{equation}
If the channel family is environment seizable as well, then this bound is
achievable in the large $n$ limit.
\end{conclusion}

\subsection{Optimizing the SLD\ and RLD\ Fisher information of quantum states
and channels} \label{subsec:qfi-optimization-formulae}

Particular generalized Fisher informations of interest in applications, due to
the bounds in \eqref{eq:QCRB}, \eqref{eq:SLD<=RLD}, and \eqref{eq:RLD-QCRB}, are
the SLD\ and RLD ones. In this section, we show how these quantities, along
with their dynamic channel versions, can be cast as optimization problems. In
some cases, we find semi-definite programs, which implies that these
quantities can be efficiently computed \cite{AHK05,AK07,AHK12, LSW15} (we
should clarify that, by ``efficient,'' we mean the computational run time is
polynomial in the dimension of the states or channels under consideration).
Thus, in these cases, there is no need to compute spectral decompositions or
matrix inverses in order to evaluate the Fisher information quantities.

\subsubsection{Semi-definite program for SLD\ Fisher information of quantum
states}

\label{sec:SDPs-for-FIs}We begin with the SLD\ Fisher information,
establishing that it can be evaluated by means of a semi-definite program.

\begin{proposition}
\label{prop:SLD-Fish-states-SDP}The SLD\ Fisher information of a
differentiable family $\{\rho_{\theta}\}_{\theta}$ of states satisfying the
finiteness condition in \eqref{eq:finiteness-condition-SLD-Fish-states} can be
evaluated by means of the following semi-definite program:
\begin{equation}
I_{F}(\theta;\{\rho_{\theta}\}_{\theta})=2\cdot\inf\left\{  \mu\in\mathbb{R}:
\begin{bmatrix}
\mu & \langle\Gamma|\left(  \partial_{\theta}\rho_{\theta}\otimes I\right) \\
\left(  \partial_{\theta}\rho_{\theta}\otimes I\right)  |\Gamma\rangle &
\rho_{\theta}\otimes I+I\otimes\rho_{\theta}^{T}
\end{bmatrix}
\geq0\right\}  .
\end{equation}
The dual semi-definite program is as follows:
\begin{equation}
2\cdot\sup_{\lambda,|\varphi\rangle,Z}2\operatorname{Re}[\langle
\varphi|\left(  \partial_{\theta}\rho_{\theta}\otimes I\right)  |\Gamma
\rangle]- \operatorname{Tr}[(\rho_{\theta}\otimes I+I\otimes\rho_{\theta}
^{T})Z],
\end{equation}
subject to $\lambda\in\mathbb{R}$, $|\varphi\rangle$ an arbitrary complex
vector, $Z$ Hermitian, and
\begin{equation}
\lambda\leq1,\qquad
\begin{bmatrix}
\lambda & \langle\varphi|\\
|\varphi\rangle & Z
\end{bmatrix}
\geq0.
\end{equation}

\end{proposition}

\begin{proof}
The primal semi-definite program is a direct consequence of the formula in
\eqref{eq:basis-independent-SLD-formula-extra}\ and Lemma~\ref{lem:min-XYinvX}. The dual program is a consequence of
Lemma~\ref{lem:freq-used-SDP-primal-dual}.
\end{proof}

\subsubsection{Root SLD Fisher information of quantum states as a
quadratically constrained optimization}

In this section, we find that the root SLD\ Fisher information of quantum
states can be computed by means of a quadratically constrained optimization.
These optimization problems are difficult to solve in general, but heuristic
methods are available \cite{PB17}. In any case, the particular optimization
formula in Proposition~\ref{prop:root-SLD-opt-formula}\ is helpful for
establishing the chain rule property of the root SLD\ Fisher information,
which we discuss in Section~\ref{sec:root-SLD-chain-rule-limits}.

\begin{proposition}
\label{prop:root-SLD-opt-formula}Let $\{\rho_{\theta}\}_{\theta}$ be a
differentiable family of quantum states. Then the root SLD\ Fisher information
can be written as the following optimization:
\begin{equation}
\sqrt{I_{F}}(\theta;\{\rho_{\theta}\}_{\theta})=\sqrt{2}\sup_{X}\left\{
\left\vert \operatorname{Tr}[X(\partial_{\theta}\rho_{\theta})]\right\vert
:\operatorname{Tr}[(XX^{\dag}+X^{\dag}X)\rho_{\theta}]\leq1\right\}  .
\label{eq:root-SLD-opt-formula}
\end{equation}
If the finiteness condition in \eqref{eq:finiteness-condition-SLD-Fish-states}
is not satisfied, then the optimization formula evaluates to $+\infty$.
\end{proposition}

\begin{proof}
Let us begin by supposing that the finiteness condition in
\eqref{eq:finiteness-condition-SLD-Fish-states} is satisfied (i.e., $\Pi
_{\rho_{\theta}}^{\perp}(\partial_{\theta}\rho_{\theta})\Pi_{\rho_{\theta}
}^{\perp}=0$). Recall from \eqref{eq:basis-independent-SLD-formula-extra} the
following formula for SLD\ Fisher information:
\begin{equation}
I_{F}(\theta;\{\rho_{\theta}\}_{\theta})=2\langle\Gamma|\left(  \partial
_{\theta}\rho_{\theta}\otimes I\right)  \left(  \rho_{\theta}\otimes
I+I\otimes\rho_{\theta}^{T}\right)  ^{-1}\left(  \partial_{\theta}\rho
_{\theta}\otimes I\right)  |\Gamma\rangle,
\end{equation}
so that
\begin{align}
&  \frac{1}{\sqrt{2}}\sqrt{I_{F}}(\theta;\{\rho_{\theta}\}_{\theta
})\nonumber\\
&  =\sqrt{\langle\Gamma|\left(  \partial_{\theta}\rho_{\theta}\otimes
I\right)  \left(  \rho_{\theta}\otimes I+I\otimes\rho_{\theta}^{T}\right)
^{-1}\left(  \partial_{\theta}\rho_{\theta}\otimes I\right)  |\Gamma\rangle}\\
&  =\left\Vert \left(  \rho_{\theta}\otimes I+I\otimes\rho_{\theta}
^{T}\right)  ^{-\frac{1}{2}}\left(  \partial_{\theta}\rho_{\theta}\otimes
I\right)  |\Gamma\rangle\right\Vert _{2}\\
&  =\sup_{|\psi\rangle:\left\Vert |\psi\rangle\right\Vert _{2}=1}\left\vert
\langle\psi|\left(  \rho_{\theta}\otimes I+I\otimes\rho_{\theta}^{T}\right)
^{-\frac{1}{2}}\left(  \partial_{\theta}\rho_{\theta}\otimes I\right)
|\Gamma\rangle\right\vert . \label{eq:sqrt-SLD-Fish-opt-alt}
\end{align}
Observe that the projection onto the support of $\rho_{\theta}\otimes
I+I\otimes\rho_{\theta}^{T}$ is
\begin{equation}
\Pi_{\rho_{\theta}}\otimes\Pi_{\rho_{\theta}^{T}}+\Pi_{\rho_{\theta}}^{\perp
}\otimes\Pi_{\rho_{\theta}^{T}}+\Pi_{\rho_{\theta}}\otimes\Pi_{\rho_{\theta
}^{T}}^{\perp}=I\otimes I-\Pi_{\rho_{\theta}}^{\perp}\otimes\Pi_{\rho_{\theta
}^{T}}^{\perp}.
\end{equation}
Thus, it suffices to optimize over $|\psi\rangle$ satisfying
\begin{equation}
|\psi\rangle=(I\otimes I-\Pi_{\rho_{\theta}}^{\perp}\otimes\Pi_{\rho_{\theta
}^{T}}^{\perp})|\psi\rangle
\end{equation}
because
\begin{multline}
\left(  \rho_{\theta}\otimes I+I\otimes\rho_{\theta}^{T}\right)  ^{-\frac
{1}{2}}\left(  \partial_{\theta}\rho_{\theta}\otimes I\right)  |\Gamma
\rangle\\
=(I\otimes I-\Pi_{\rho_{\theta}}^{\perp}\otimes\Pi_{\rho_{\theta}^{T}}^{\perp
})\left(  \rho_{\theta}\otimes I+I\otimes\rho_{\theta}^{T}\right)  ^{-\frac
{1}{2}}\left(  \partial_{\theta}\rho_{\theta}\otimes I\right)  |\Gamma\rangle.
\end{multline}
Now define
\begin{equation}
|\psi^{\prime}\rangle:=\left(  \rho_{\theta}\otimes I+I\otimes\rho_{\theta
}^{T}\right)  ^{-\frac{1}{2}}|\psi\rangle,
\end{equation}
which implies that
\begin{equation}
|\psi\rangle=(I\otimes I-\Pi_{\rho_{\theta}}^{\perp}\otimes\Pi_{\rho_{\theta
}^{T}}^{\perp})|\psi\rangle=\left(  \rho_{\theta}\otimes I+I\otimes
\rho_{\theta}^{T}\right)  ^{\frac{1}{2}}|\psi^{\prime}\rangle,
\end{equation}
because $I\otimes I-\Pi_{\rho_{\theta}}^{\perp}\otimes\Pi_{\rho_{\theta}^{T}
}^{\perp}$ is the projection onto the support of $\rho_{\theta}\otimes
I+I\otimes\rho_{\theta}^{T}$. Thus, the following equivalence holds
\begin{align}
\left\Vert |\psi\rangle\right\Vert _{2}=1\quad &  \Longleftrightarrow
\quad\left\Vert \left(  \rho_{\theta}\otimes I+I\otimes\rho_{\theta}
^{T}\right)  ^{\frac{1}{2}}|\psi^{\prime}\rangle\right\Vert _{2}=1\\
&  \Longleftrightarrow\quad\langle\psi^{\prime}|\left(  \rho_{\theta}\otimes
I+I\otimes\rho_{\theta}^{T}\right)  |\psi^{\prime}\rangle=1.
\end{align}
Now fix the operator $X$ such that
\begin{equation}
|\psi^{\prime}\rangle=\left(  X\otimes I\right)  |\Gamma\rangle.
\end{equation}
Then the last condition above is the same as the following:
\begin{align}
1  &  =\langle\Gamma|\left(  X^{\dag}\otimes I\right)  \left(  \rho_{\theta
}\otimes I+I\otimes\rho_{\theta}^{T}\right)  \left(  X\otimes I\right)
|\Gamma\rangle\\
&  =\langle\Gamma|\left(  X^{\dag}\rho_{\theta}X\otimes I+X^{\dag}X\otimes
\rho_{\theta}^{T}\right)  |\Gamma\rangle\\
&  =\langle\Gamma|\left(  X^{\dag}\rho_{\theta}X\otimes I+X^{\dag}
X\rho_{\theta}\otimes I\right)  |\Gamma\rangle\\
&  =\operatorname{Tr}[X^{\dag}\rho_{\theta}X]+\operatorname{Tr}[X^{\dag}
X\rho_{\theta}]\\
&  =\operatorname{Tr}[(XX^{\dag}+X^{\dag}X)\rho_{\theta}],
\end{align}
where we used \eqref{eq:transpose-trick}\ and
\eqref{eq:max-ent-partial-trace}. So then the optimization problem in
\eqref{eq:sqrt-SLD-Fish-opt-alt}\ is equal to the following:
\begin{align}
&  \sup_{X:\operatorname{Tr}[(XX^{\dag}+X^{\dag}X)\rho_{\theta}]=1}\left\vert
\langle\Gamma|\left(  X\otimes I\right)  \left(  \partial_{\theta}\rho
_{\theta}\otimes I\right)  |\Gamma\rangle\right\vert \nonumber\\
&  =\sup_{X:\operatorname{Tr}[(XX^{\dag}+X^{\dag}X)\rho_{\theta}]=1}\left\vert
\langle\Gamma|\left(  X(\partial_{\theta}\rho_{\theta})\otimes I\right)
|\Gamma\rangle\right\vert \\
&  =\sup_{X}\left\{  \left\vert \operatorname{Tr}[X(\partial_{\theta}
\rho_{\theta})]\right\vert :\operatorname{Tr}[(XX^{\dag}+X^{\dag}
X)\rho_{\theta}]=1\right\}  ,
\end{align}
where again we used \eqref{eq:max-ent-partial-trace}. Now suppose that
$\operatorname{Tr}[(XX^{\dag}+X^{\dag}X)\rho_{\theta}]=c$, with $c\in(0,1)$.
Then we can multiply $X$ by $\sqrt{1/c}$, and the new operator satisfies the
equality constraint while the value of the objective function increases. So we
can write
\begin{equation}
\sqrt{I_{F}}(\theta;\{\rho_{\theta}\}_{\theta})=\sqrt{2}\sup_{X}\left\{
\left\vert \operatorname{Tr}[X(\partial_{\theta}\rho_{\theta})]\right\vert
:\operatorname{Tr}[(XX^{\dag}+X^{\dag}X)\rho_{\theta}]\leq1\right\}  .
\end{equation}
Finally, in this form, note that we can trivially include $X=0$ as part of the
optimization because it leads to a generally suboptimal value of zero for the
objective function.

Suppose that $\Pi_{\rho_{\theta}}^{\perp}(\partial_{\theta}\rho_{\theta}
)\Pi_{\rho_{\theta}}^{\perp}\neq0$. Then we can pick $X=c\Pi_{\rho_{\theta}
}^{\perp}+dI$ where $c,d>0$ and $2d^{2}=1$. We find that
\begin{align}
\operatorname{Tr}[(XX^{\dag}+X^{\dag}X)\rho_{\theta}]  &  =2\operatorname{Tr}
[\left(  c\Pi_{\rho_{\theta}}^{\perp}+dI\right)  ^{2}\rho_{\theta}]\\
&  =2\operatorname{Tr}[\left(  \left[  c^{2}+2cd\right]  \Pi_{\rho_{\theta}
}^{\perp}+d^{2}I\right)  \rho_{\theta}]\\
&  =2d^{2}=1.
\end{align}
for this case, so that the constraint in \eqref{eq:root-SLD-opt-formula}\ is
satisfied. The objective function then evaluates to
\begin{align}
\left\vert \operatorname{Tr}[X(\partial_{\theta}\rho_{\theta})]\right\vert  &
=\left\vert \operatorname{Tr}[\left(  c\Pi_{\rho_{\theta}}^{\perp}+dI\right)
(\partial_{\theta}\rho_{\theta})]\right\vert \\
&  =\left\vert c\operatorname{Tr}[\Pi_{\rho_{\theta}}^{\perp}(\partial
_{\theta}\rho_{\theta})]+d\operatorname{Tr}[\partial_{\theta}\rho_{\theta
}]\right\vert \\
&  =c\left\vert \operatorname{Tr}[\Pi_{\rho_{\theta}}^{\perp}(\partial
_{\theta}\rho_{\theta})]\right\vert .
\end{align}
Then we can pick $c>0$ arbitrarily large to get that
\eqref{eq:root-SLD-opt-formula} evaluates to $+\infty$ in the case that
$\Pi_{\rho_{\theta}}^{\perp}(\partial_{\theta}\rho_{\theta})\Pi_{\rho_{\theta
}}^{\perp}\neq0$.
\end{proof}

\bigskip

We can use the optimization formula in
Proposition~\ref{prop:root-SLD-opt-formula} to conclude that the
data-processing inequality holds for all two-positive, trace-preserving maps,
which includes quantum channels as a special case. This was already observed
in \cite{Petz96}, but here we give a different proof based on the optimization
formula in Proposition~\ref{prop:root-SLD-opt-formula}.

\begin{proposition}
Let $\{\rho_{\theta}\}_{\theta}$ be a differentiable family of quantum states,
and let $\mathcal{P}$ be a two-positive, trace-preserving map. Then the
following data-processing inequality holds
\begin{equation}
I_{F}(\theta;\{\rho_{\theta}\}_{\theta})\geq I_{F}(\theta;\{\mathcal{P}
(\rho_{\theta})\}_{\theta}).
\end{equation}

\end{proposition}

\begin{proof}
Let $X$ be an operator satisfying
\begin{equation}
\operatorname{Tr}[(XX^{\dag}+X^{\dag}X)\mathcal{P}(\rho_{\theta})]\leq1.
\label{eq:constraint-DP-processed-state}
\end{equation}
Then it follows that
\begin{align}
1  &  \geq\operatorname{Tr}[(XX^{\dag}+X^{\dag}X)\mathcal{P}(\rho_{\theta})]\\
&  =\operatorname{Tr}[\mathcal{P}^{\dag}(XX^{\dag}+X^{\dag}X)\rho_{\theta}]\\
&  \geq\operatorname{Tr}[(\mathcal{P}^{\dag}(X)\mathcal{P}^{\dag}(X^{\dag
})+\mathcal{P}^{\dag}(X^{\dag})\mathcal{P}^{\dag}(X))\rho_{\theta}],
\end{align}
where the last inequality follows because $\rho_\theta \geq 0$ and
\begin{equation}
\mathcal{P}^{\dag}(XX^{\dag})\geq\mathcal{P}^{\dag}(X)\mathcal{P}^{\dag
}(X^{\dag}),\qquad\mathcal{P}^{\dag}(X^{\dag}X)\geq\mathcal{P}^{\dag}(X^{\dag
})\mathcal{P}^{\dag}(X).
\end{equation}
The latter inequalities are a consequence of the Schwarz inequality, which
holds for two-positive, unital maps \cite[Eq.~(3.14)]{Bhat07}. (Note that
two-positive, unital maps are the Hilbert--Schmidt adjoints of two-positive,
trace-preserving maps). Furthermore,
\begin{align}
\left\vert \operatorname{Tr}[X(\partial_{\theta}\mathcal{P}(\rho_{\theta
}))]\right\vert  &  =\left\vert \operatorname{Tr}[X\mathcal{P}(\partial
_{\theta}\rho_{\theta})]\right\vert \\
&  =\left\vert \operatorname{Tr}[\mathcal{P}^{\dag}(X)(\partial_{\theta}
\rho_{\theta})]\right\vert \\
&  \leq\sup_{Z}\left\{  \left\vert \operatorname{Tr}[Z(\partial_{\theta}
\rho_{\theta})]\right\vert :\operatorname{Tr}[(ZZ^{\dag}+Z^{\dag}
Z)\rho_{\theta}]\leq1\right\} \\
&  =\frac{1}{\sqrt{2}}\sqrt{I_{F}}(\theta;\{\rho_{\theta}\}_{\theta}).
\end{align}
Since the inequality holds for all $X$ satisfying
\eqref{eq:constraint-DP-processed-state}, we conclude that
\begin{equation}
\sqrt{I_{F}}(\theta;\{\rho_{\theta}\}_{\theta})\geq\sqrt{I_{F}}(\theta
;\{\mathcal{P}(\rho_{\theta})\}_{\theta}).
\end{equation}
This concludes the proof.
\end{proof}

\subsubsection{Bilinear program for SLD\ Fisher information of quantum
channels}

We can exploit Proposition~\ref{prop:SLD-Fish-states-SDP} and a number of
manipulations to arrive at a bilinear program for the SLD Fisher information
of channels:

\begin{proposition}
\label{prop:SDP-SLD-Fish}The SLD Fisher information of a differentiable family
$\{\mathcal{N}_{A\rightarrow B}^{\theta}\}_{\theta}$ of channels satisfying
the finiteness condition in \eqref{eq:finiteness-condition-SLD-fish-ch} can be
evaluated by means of the following bilinear program:
\begin{multline}
I_{F}(\theta;\{\mathcal{N}_{A\rightarrow B}^{\theta}
\})=\label{eq:SDP-SLD-ch-Fish}\\
2\sup_{\substack{\lambda,|\varphi\rangle_{RBR^{\prime}B^{\prime}
},\\W_{RBR^{\prime}B^{\prime}},Y_{R},\sigma_{R}}}\left(  2\operatorname{Re}
[\langle\varphi|_{RBR^{\prime}B^{\prime}}(\partial_{\theta}\Gamma
_{RB}^{\mathcal{N}^{\theta}})|\Gamma\rangle_{RR^{\prime}BB^{\prime}
}]-\operatorname{Tr}[Y_{R}\Phi(W_{RBR^{\prime}B^{\prime}})]\right)
\end{multline}
subject to
\begin{equation}
\sigma_{R}\geq0,\quad\operatorname{Tr}[\sigma_{R}]=1,\quad\lambda\leq1,\quad
\begin{bmatrix}
\lambda & \langle\varphi|_{RBR^{\prime}B^{\prime}}\\
|\varphi\rangle_{RBR^{\prime}B^{\prime}} & W_{RBR^{\prime}B^{\prime}}
\end{bmatrix}
\geq0,\quad
\begin{bmatrix}
\sigma_{R} & I_{R}\\
I_{R} & Y_{R}
\end{bmatrix}
\geq0.
\end{equation}
where
\begin{align}
|\Gamma\rangle_{RR^{\prime}BB^{\prime}}  &  :=|\Gamma\rangle_{RR^{\prime}
}\otimes|\Gamma\rangle_{BB^{\prime}},\\
\Phi(W_{RBR^{\prime}B^{\prime}})  &  :=(\operatorname{Tr}_{BR^{\prime
}B^{\prime}}[\Gamma_{R^{\prime}B^{\prime}}^{\mathcal{N}^{\theta}}\left(
F_{RR^{\prime}}\otimes F_{BB^{\prime}}\right)  W_{RBR^{\prime}B^{\prime}
}\left(  F_{RR^{\prime}}\otimes F_{BB^{\prime}}\right)  ^{\dag}])^{T}
\nonumber\\
&  \qquad\qquad+\operatorname{Tr}_{BR^{\prime}B^{\prime}}[(\Gamma_{R^{\prime
}B^{\prime}}^{\mathcal{N}^{\theta}})^{T}W_{RBR^{\prime}B^{\prime}}],
\end{align}
and $F_{RR^{\prime}}$ is the flip or swap operator that swaps systems $R$ and
$R^{\prime}$, with a similar definition for $F_{BB^{\prime}}$ but for $B$ and
$B^{\prime}$.
\end{proposition}

\begin{proof}
See Appendix~\ref{app:SDP-SLD-Fish-ch}.
\end{proof}

\bigskip

The optimization above is a jointly constrained semi-definite bilinear program
\cite{HKT18}\ because the variables $Y_{R}$ and $W_{RBR^{\prime}B^{\prime}}$
are operators involved in the optimization and they multiply each other in the
last expression in \eqref{eq:SDP-SLD-ch-Fish}. This kind of optimization can
be approached with a heuristic \textquotedblleft seesaw\textquotedblright
\ method, but more advanced methods are available in~\cite{HKT18}.

\subsubsection{Semi-definite programs for RLD\ Fisher information of quantum
states and channels}

We now give semi-definite programs for the RLD\ Fisher information of quantum states:

\begin{proposition}
The RLD\ Fisher information of a differentiable family $\{\rho_{\theta
}\}_{\theta}$ of states satisfying the support condition in
\eqref{eq:support-condition-RLD} can be evaluated by means of the following
semi-definite program:
\begin{equation}
\widehat{I}_{F}(\theta;\{\rho_{A}^{\theta}\})=\inf\left\{  \operatorname{Tr}
[M]:M\geq0,
\begin{bmatrix}
M & \partial_{\theta}\rho_{\theta}\\
\partial_{\theta}\rho_{\theta} & \rho_{\theta}
\end{bmatrix}
\geq0\right\}  .
\end{equation}
The dual semi-definite program is as follows:
\begin{equation}
\sup_{X,Y,Z}2\operatorname{Re}[\operatorname{Tr}[Y(\partial_{\theta}
\rho_{\theta})]]-\operatorname{Tr}[Z\rho_{\theta}],
\end{equation}
subject to $X$ and $Y$ being Hermitian and
\begin{equation}
X\leq I,\qquad
\begin{bmatrix}
X & Y^{\dag}\\
Y & Z
\end{bmatrix}
\geq0.
\end{equation}

\end{proposition}

\begin{proof}
The primal semi-definite program is a direct consequence of the RLD formula in
\eqref{eq:RLD-FI}\ and Lemma~\ref{lem:min-XYinvX}. The dual program is found
by applying Lemma~\ref{lem:freq-used-SDP-primal-dual}.
\end{proof}

\bigskip

The following formula for the RLD\ Fisher information of quantum channels is
known from \cite{Hayashi2011}. It comes about by manipulating the RLD\ formula
in \eqref{eq:RLD-FI}\ by means of Lemma~\ref{lem:transformer-ineq-basic}. We
review its proof in Appendix~\ref{app:proofs-RLD-fish-props}.

\begin{proposition}
\label{prop:geo-fish-explicit-formula}Let $\{\mathcal{N}_{A\rightarrow
B}^{\theta}\}_{\theta}$ be a differentiable family of quantum channels such
that the support condition in \eqref{eq:finiteness-condition-RLD-fish-ch}
holds. Then the RLD\ Fisher information of quantum channels has the following
explicit form:
\begin{equation}
\widehat{I}_{F}(\theta;\{\mathcal{N}_{A\rightarrow B}^{\theta}\}_{\theta
})=\left\Vert \operatorname{Tr}_{B}[(\partial_{\theta}\Gamma_{RB}
^{\mathcal{N}^{\theta}})(\Gamma_{RB}^{\mathcal{N}^{\theta}})^{-1}
(\partial_{\theta}\Gamma_{RB}^{\mathcal{N}^{\theta}})]\right\Vert _{\infty},
\label{eq:RLD-Fish-ch}
\end{equation}
where $\Gamma_{RB}^{\mathcal{N}^{\theta}}$ is the Choi operator of the channel
$\mathcal{N}_{A\rightarrow B}^{\theta}$.
\end{proposition}

We then find the following semi-definite program for the RLD\ Fisher
information of quantum channels:

\begin{proposition}
Let $\{\mathcal{N}_{A\rightarrow B}^{\theta}\}_{\theta}$ be a differentiable
family of quantum channels such that the support condition in
\eqref{eq:finiteness-condition-RLD-fish-ch} holds. Then the RLD Fisher
information of quantum channels can be calculated by means of the following
semi-definite program:
\begin{equation}
\widehat{I}_{F}(\theta;\{\mathcal{N}_{A\rightarrow B}^{\theta}\}_{\theta
})=\inf\lambda\in\mathbb{R}^{+}, \label{eq:SDP-SLD-fish-channels}
\end{equation}
subject to
\begin{equation}
\lambda I_{R}\geq\operatorname{Tr}_{B}[M_{RB}],\qquad
\begin{bmatrix}
M_{RB} & \partial_{\theta}\Gamma_{RB}^{\mathcal{N}^{\theta}}\\
\partial_{\theta}\Gamma_{RB}^{\mathcal{N}^{\theta}} & \Gamma_{RB}
^{\mathcal{N}^{\theta}}
\end{bmatrix}
\geq0. \label{eq:SDP-SLD-fish-channels-2}
\end{equation}
The dual program is given by
\begin{equation}
\sup_{\rho_{R}\geq0,P_{RB},Z_{RB},Q_{RB}}2\operatorname{Re}[\operatorname{Tr}
[Z_{RB}(\partial_{\theta}\Gamma_{RB}^{\mathcal{N}^{\theta}}
)]]-\operatorname{Tr}[Q_{RB}\Gamma_{RB}^{\mathcal{N}^{\theta}}],
\end{equation}
subject to
\begin{equation}
\operatorname{Tr}[\rho_{R}]\leq1,\quad
\begin{bmatrix}
P_{RB} & Z_{RB}^{\dag}\\
Z_{RB} & Q_{RB}
\end{bmatrix}
\geq0,\quad P_{RB}\leq\rho_{R}\otimes I_{B}.
\end{equation}

\end{proposition}

\begin{proof}
The form of the primal program follows directly from \eqref{eq:RLD-Fish-ch},
Lemma~\ref{lem:min-XYinvX}, and from the following characterization of the
infinity norm of a positive semi-definite operator $W$:
\begin{equation}
\left\Vert W\right\Vert _{\infty}=\inf\left\{  \lambda\geq0:W\leq\lambda
I\right\}  .
\end{equation}

To arrive at the dual program, we use the standard forms of primal and dual
semi-definite programs for Hermitian operators $A$ and $B$ and a
Hermiticity-preserving map $\Phi$ \cite{Wat18}:
\begin{equation}
\sup_{X\geq0}\left\{  \operatorname{Tr}[AX]:\Phi(X)\leq B\right\}  ,
\qquad\inf_{Y\geq0}\left\{  \operatorname{Tr}[BY]:\Phi^{\dag}(Y)\geq
A\right\}  . \label{eq:standard-SDP-form-RLD-ch-helper}
\end{equation}
From \eqref{eq:SDP-SLD-fish-channels}--\eqref{eq:SDP-SLD-fish-channels-2}, we
identify
\begin{align}
B  &  =
\begin{bmatrix}
1 & 0\\
0 & 0
\end{bmatrix}
,\quad Y=
\begin{bmatrix}
\lambda & 0\\
0 & M_{RB}
\end{bmatrix}
,\quad\Phi^{\dag}(Y)=
\begin{bmatrix}
\lambda I_{R}-\operatorname{Tr}_{B}[M_{RB}] & 0 & 0\\
0 & M_{RB} & 0\\
0 & 0 & 0
\end{bmatrix}
,\\
A  &  =
\begin{bmatrix}
0 & 0 & 0\\
0 & 0 & -\partial_{\theta}\Gamma_{RB}^{\mathcal{N}^{\theta}}\\
0 & -\partial_{\theta}\Gamma_{RB}^{\mathcal{N}^{\theta}} & -\Gamma
_{RB}^{\mathcal{N}^{\theta}}
\end{bmatrix}
.
\end{align}
Setting
\begin{equation}
X=
\begin{bmatrix}
\rho_{R} & 0 & 0\\
0 & P_{RB} & Z_{RB}^{\dag}\\
0 & Z_{RB} & Q_{RB}
\end{bmatrix}
,
\end{equation}
we find that
\begin{align}
\operatorname{Tr}[X\Phi^{\dag}(Y)]  &  =\operatorname{Tr}\left[
\begin{bmatrix}
\rho_{R} & 0 & 0\\
0 & P_{RB} & Z_{RB}^{\dag}\\
0 & Z_{RB} & Q_{RB}
\end{bmatrix}
\begin{bmatrix}
\lambda I_{R}-\operatorname{Tr}_{B}[M_{RB}] & 0 & 0\\
0 & M_{RB} & 0\\
0 & 0 & 0
\end{bmatrix}
\right] \\
&  =\operatorname{Tr}[\rho_{R}(\lambda I_{R}-\operatorname{Tr}_{B}
[M_{RB}])]+\operatorname{Tr}[P_{RB}M_{RB}]\\
&  =\lambda\operatorname{Tr}[\rho_{R}]+\operatorname{Tr}[(P_{RB}-\rho
_{R}\otimes I_{B})M_{RB}]\\
&  =\operatorname{Tr}\left[
\begin{bmatrix}
\lambda & 0\\
0 & M_{RB}
\end{bmatrix}
\begin{bmatrix}
\operatorname{Tr}[\rho_{R}] & 0\\
0 & P_{RB}-\rho_{R}\otimes I_{B}
\end{bmatrix}
\right]  ,
\end{align}
which implies that
\begin{equation}
\Phi(X)=
\begin{bmatrix}
\operatorname{Tr}[\rho_{R}] & 0\\
0 & P_{RB}-\rho_{R}\otimes I_{B}
\end{bmatrix}
.
\end{equation}
Then plugging into the left-hand side of
\eqref{eq:standard-SDP-form-RLD-ch-helper}, we find that the dual is given by
\begin{equation}
\sup_{\rho_{R},P_{RB},Z_{RB},Q_{RB}}\operatorname{Tr}\left[
\begin{bmatrix}
0 & 0 & 0\\
0 & 0 & -\partial_{\theta}\Gamma_{RB}^{\mathcal{N}^{\theta}}\\
0 & -\partial_{\theta}\Gamma_{RB}^{\mathcal{N}^{\theta}} & -\Gamma
_{RB}^{\mathcal{N}^{\theta}}
\end{bmatrix}
\begin{bmatrix}
W_{R} & 0 & 0\\
0 & P_{RB} & Z_{RB}^{\dag}\\
0 & Z_{RB} & Q_{RB}
\end{bmatrix}
\right]  ,
\end{equation}
subject to
\begin{equation}
\begin{bmatrix}
\rho_{R} & 0 & 0\\
0 & P_{RB} & Z_{RB}^{\dag}\\
0 & Z_{RB} & Q_{RB}
\end{bmatrix}
\geq0,\qquad
\begin{bmatrix}
\operatorname{Tr}[\rho_{R}] & 0\\
0 & P_{RB}-\rho_{R}\otimes I_{B}
\end{bmatrix}
\leq
\begin{bmatrix}
1 & 0\\
0 & 0
\end{bmatrix}
.
\end{equation}
Upon making the swap $Z_{RB}\rightarrow-Z_{RB}$, which does not change the
optimal value, and simplifying, we find the following form:
\begin{equation}
\sup_{\rho_{R}\geq0,P_{RB},Z_{RB},Q_{RB}}2\operatorname{Re}[\operatorname{Tr}
[Z_{RB}(\partial_{\theta}\Gamma_{RB}^{\mathcal{N}^{\theta}}
)]]-\operatorname{Tr}[Q_{RB}\Gamma_{RB}^{\mathcal{N}^{\theta}}],
\end{equation}
subject to
\begin{equation}
\operatorname{Tr}[\rho_{R}]\leq1,\quad
\begin{bmatrix}
P_{RB} & -Z_{RB}^{\dag}\\
-Z_{RB} & Q_{RB}
\end{bmatrix}
\geq0,\quad P_{RB}\leq\rho_{R}\otimes I_{B}.
\end{equation}
Then we note that
\begin{equation}
\begin{bmatrix}
P_{RB} & -Z_{RB}^{\dag}\\
-Z_{RB} & Q_{RB}
\end{bmatrix}
\geq0 \quad\Longleftrightarrow\quad
\begin{bmatrix}
P_{RB} & Z_{RB}^{\dag}\\
Z_{RB} & Q_{RB}
\end{bmatrix}
\geq0
\end{equation}
This concludes the proof.
\end{proof}

\subsection{SLD Fisher information limits on quantum channel parameter
estimation}

\label{sec:SLD-fisher-info-limits}

\subsubsection{SLD Fisher information limit on parameter estimation of
classical--quantum channels} \label{subsec:SLD-limit-cq-channels}

We first consider the special case of a family $\{\mathcal{N}_{X\rightarrow
B}^{\theta}\}_{\theta}$ of classical--quantum channels of the following form:
\begin{equation}
\mathcal{N}_{X\rightarrow B}^{\theta}(\sigma_{X}):=\sum_{x}\langle
x|_{X}\sigma_{X}|x\rangle_{X}\omega_{B}^{x,\theta}, \label{eq:cq-channel-fams}
\end{equation}
where $\{|x\rangle\}_{x}$ is an orthonormal basis and $\{\omega_{B}^{x,\theta
}\}_{x}$ is a collection of states prepared at the channel output conditioned
on the value of the unknown parameter $\theta$ and on the result of the
measurement of the channel input. The key aspect of these channels is that the
measurement at the input is the same regardless of the value of the parameter
$\theta$. We find the following amortization collapse for these channels:

\begin{theorem}
\label{thm:amort-collapse-cq}Let $\{\mathcal{N}_{X\rightarrow B}^{\theta
}\}_{\theta}$ be a family of differentiable classical--quantum channels. Then
the following amortization collapse occurs
\begin{equation}
I_{F}(\theta;\{\mathcal{N}_{X\rightarrow B}^{\theta}\}_{\theta})=I_{F}
^{\mathcal{A}}(\theta;\{\mathcal{N}_{X\rightarrow B}^{\theta}\}_{\theta}
)=\sup_{x}I_{F}(\theta;\{\omega_{B}^{x,\theta}\}_{\theta}).
\label{eq:amort-collapse}
\end{equation}

\end{theorem}

\begin{proof}
If the finiteness condition in
\eqref{eq:finiteness-condition-SLD-fish-ch}\ does not hold, then all
quantities are trivially equal to $+\infty$. So let us suppose that the
finiteness condition in \eqref{eq:finiteness-condition-SLD-fish-ch}\ holds.
Note that the finiteness condition is equivalent to
\begin{equation}
\Pi_{\omega_{B}^{x,\theta}}^{\perp}(\partial_{\theta}\omega_{B}^{x,\theta}
)\Pi_{\omega_{B}^{x,\theta}}^{\perp}=0\qquad\forall x.
\label{eq:finiteness-condition-cq-channels}
\end{equation}

First, consider that the following inequality holds
\begin{equation}
I_{F}(\theta;\{\mathcal{N}_{X\rightarrow B}^{\theta}\}_{\theta})\geq\sup
_{x}I_{F}(\theta;\{\omega_{B}^{x,\theta}\}_{\theta}) \label{eq:simple-way-cq}
\end{equation}
because we can input the state $|x\rangle\!\langle x|_{X}$ to the channel
$\mathcal{N}_{X\rightarrow B}^{\theta}$ and obtain the output state
$\mathcal{N}_{X\rightarrow B}^{\theta}(|x\rangle\!\langle x|_{X})=\omega
_{B}^{x,\theta}$. Then we can optimize over $x\in\mathcal{X}$ and obtain the
bound above.

We now prove the less trivial inequality
\begin{equation}
I_{F}^{\mathcal{A}}(\theta;\{\mathcal{N}_{X\rightarrow B}^{\theta}\}_{\theta
})\leq\sup_{x}I_{F}(\theta;\{\omega_{B}^{x,\theta}\}_{\theta}).
\label{eq:amort-dont-help-cq}
\end{equation}
Let $\{\rho_{RA}^{\theta}\}_{\theta}$ be a differentiable family of quantum
states. If the classical--quantum channel $\mathcal{N}_{X\rightarrow
B}^{\theta}$ acts on $\rho_{RA}^{\theta}$ (identifying $X=A$), the output
state is as follows:
\begin{equation}
\mathcal{N}_{X\rightarrow B}^{\theta}(\rho_{RA}^{\theta})=\sum_{x}p_{\theta
}(x)\rho_{R}^{x,\theta}\otimes\omega_{B}^{x,\theta},
\end{equation}
where
\begin{equation}
\rho_{R}^{x,\theta}:=\frac{1}{p_{\theta}(x)}\langle x|_{X}\rho_{RA}^{\theta
}|x\rangle_{X},\qquad p_{\theta}(x):=\operatorname{Tr}[\langle x|_{X}\rho
_{RA}^{\theta}|x\rangle_{X}].
\end{equation}
Then consider that
\begin{align}
&  I_{F}(\theta;\{\mathcal{N}_{X\rightarrow B}^{\theta}(\rho_{RA}^{\theta
})\}_{\theta})\nonumber\\
&  =I_{F}\!\left(  \theta;\left\{  \sum_{x}p_{\theta}(x)\rho_{R}^{x,\theta
}\otimes\omega_{B}^{x,\theta}\right\}  _{\theta}\right) \\
&  \leq I_{F}\!\left(  \theta;\left\{  \sum_{x}p_{\theta}(x)|x\rangle\!\langle
x|_{X}\otimes\rho_{R}^{x,\theta}\otimes\omega_{B}^{x,\theta}\right\}
_{\theta}\right) \\
&  =I_{F}(\theta;\{p_{\theta}\}_{\theta})+\sum_{x}p_{\theta}(x)I_{F}
(\theta;\{\rho_{R}^{x,\theta}\otimes\omega_{B}^{x,\theta}\}_{\theta})\\
&  =I_{F}(\theta;\{p_{\theta}\}_{\theta})+\sum_{x}p_{\theta}(x)I_{F}
(\theta;\{\rho_{R}^{x,\theta}\}_{\theta})+\sum_{x}p_{\theta}(x)I_{F}
(\theta;\{\omega_{B}^{x,\theta}\}_{\theta})\\
&  \leq I_{F}(\theta;\{p_{\theta}\}_{\theta})+\sum_{x}p_{\theta}
(x)I_{F}(\theta;\{\rho_{R}^{x,\theta}\}_{\theta})+\sup_{x}I_{F}(\theta
;\{\omega_{B}^{x,\theta}\}_{\theta})\\
&  =I_{F}\!\left(  \theta;\left\{  \sum_{x}p_{\theta}(x)|x\rangle\!\langle
x|_{X}\otimes\rho_{R}^{x,\theta}\right\}  _{\theta}\right)  +\sup_{x}
I_{F}(\theta;\{\omega_{B}^{x,\theta}\}_{\theta})\\
&  \leq I_{F}(\theta;\{\rho_{RA}^{\theta}\}_{\theta})+\sup_{x}I_{F}
(\theta;\{\omega_{B}^{x,\theta}\}_{\theta}).
\end{align}
The first inequality follows from the data-processing inequality for Fisher
information with respect to partial trace over the $X$ system. The second
equality follows from Proposition~\ref{prop:cq-decomp-SLD-RLD}. The third
equality follows from the additivity of SLD\ Fisher information for product
states (Proposition~\ref{prop:additivity-SLD-RLD-states}). The second
inequality follows from the fact that the average cannot exceed the maximum.
The last equality follows again from Proposition~\ref{prop:cq-decomp-SLD-RLD}.
The final inequality follows from the data-processing inequality under the
action of the measurement channel $(\cdot)\rightarrow\sum_{x}|x\rangle\!\langle
x|_{X}(\cdot)|x\rangle\!\langle x|_{X}$ on the state $\rho_{RA}$. Thus, the
following inequality holds for an arbitrary family $\{\rho_{RA}^{\theta
}\}_{\theta}$ of states:
\begin{equation}
I_{F}(\theta;\{\mathcal{N}_{X\rightarrow B}^{\theta}(\rho_{RA}^{\theta
})\}_{\theta})-I_{F}(\theta;\{\rho_{RA}^{\theta}\}_{\theta})\leq\sup_{x}
I_{F}(\theta;\{\omega_{B}^{x,\theta}\}_{\theta}).
\label{eq:amort-dont-help-cq-2}
\end{equation}
Since the inequality in \eqref{eq:amort-dont-help-cq-2} holds for an arbitrary
family $\{\rho_{RA}^{\theta}\}_{\theta}$ of states, we conclude
\eqref{eq:amort-dont-help-cq}. Combining \eqref{eq:simple-way-cq} and
\eqref{eq:amort-dont-help-cq}, along with the general inequality in
\eqref{eq:general-amort-ineq-obvi-dir}, we conclude \eqref{eq:amort-collapse}.
\end{proof}

\begin{conclusion}
As a direct consequence of the QCRB\ in \eqref{eq:QCRB}, the meta-converse
from Theorem~\ref{thm:meta-converse}, and the amortization collapse from
Theorem~\ref{thm:amort-collapse-cq}, we conclude the following bound on the
MSE\ of an unbiased estimator $\hat{\theta}$ for classical--quantum channel
families defined in \eqref{eq:cq-channel-fams} and for which the finiteness
condition in \eqref{eq:finiteness-condition-cq-channels} holds:
\begin{equation}
\operatorname{Var}(\hat{\theta})\geq\frac{1}{n\sup_{x}I_{F}(\theta
;\{\omega_{B}^{x,\theta}\}_{\theta})}.
\end{equation}
Thus, there is no advantage that sequential estimation strategies bring over
parallel estimation strategies for this class of channels. In fact, an optimal
parallel estimation strategy consists of picking the same optimal input letter $x$
to each channel use in order to estimate~$\theta$.
\end{conclusion}

\subsubsection{Root SLD\ Fisher information limit for quantum channel
parameter estimation}

\label{sec:root-SLD-chain-rule-limits}We begin by showing that the root
SLD\ Fisher information obeys the following chain rule:

\begin{proposition}
[Chain rule]\label{prop:chain-rule-root-SLD}Let $\{\rho_{\theta}\}_{\theta}$
be a differentiable family of quantum states, and let $\{\mathcal{N}
_{A\rightarrow B}^{\theta}\}_{\theta}$ be a differentiable family of quantum
channels. Then the following chain rule holds for the root SLD Fisher
information:
\begin{equation}
\sqrt{I_{F}}(\theta;\{\mathcal{N}_{A\rightarrow B}^{\theta}(\rho_{RA}^{\theta
})\}_{\theta})\leq\sqrt{I_{F}}(\theta;\{\mathcal{N}_{A\rightarrow B}^{\theta
}\}_{\theta})+\sqrt{I_{F}}(\theta;\{\rho_{RA}^{\theta}\}_{\theta}).
\label{eq:chain-rule-root-SLD}
\end{equation}

\end{proposition}

\begin{proof}
If the finiteness conditions in
\eqref{eq:finiteness-condition-SLD-Fish-states}\ and\ \eqref{eq:finiteness-condition-SLD-fish-ch}
do not hold, then the inequality is trivially satisfied. So let us suppose
that the finiteness conditions
\eqref{eq:finiteness-condition-SLD-Fish-states}\ and\ \eqref{eq:finiteness-condition-SLD-fish-ch}
hold.

By invoking Proposition~\ref{prop:root-SLD-opt-formula}\ and
Remark~\ref{rem:restrict-to-pure-bipartite}, first consider that the root
SLD\ Fisher information of channels has the following representation as an
optimization:
\begin{align}
&  \frac{1}{\sqrt{2}}\sqrt{I_{F}}(\theta;\{\mathcal{N}_{A\rightarrow
B}^{\theta}\}_{\theta})\nonumber\\
&  =\frac{1}{\sqrt{2}}\sup_{\rho_{RA}}\sqrt{I_{F}}(\theta;\{\mathcal{N}
_{A\rightarrow B}^{\theta}(\rho_{RA})\}_{\theta})\\
&  =\sup_{\rho_{RA}}\sup_{X_{RB}}\left\{
\begin{array}
[c]{c}
\left\vert \operatorname{Tr}[X_{RB}(\partial_{\theta}\mathcal{N}_{A\rightarrow
B}^{\theta}(\rho_{RA}))]\right\vert :\\
\operatorname{Tr}[(X_{RB}X_{RB}^{\dag}+X_{RB}^{\dag}X_{RB})\mathcal{N}
_{A\rightarrow B}^{\theta}(\rho_{RA})]\leq1
\end{array}
\right\} \\
&  =\sup_{\rho_{RA},X_{RB}}\left\{
\begin{array}
[c]{c}
\left\vert \operatorname{Tr}[X_{RB}(\partial_{\theta}\mathcal{N}_{A\rightarrow
B}^{\theta})(\rho_{RA})]\right\vert :\\
\operatorname{Tr}[(X_{RB}X_{RB}^{\dag}+X_{RB}^{\dag}X_{RB})\mathcal{N}
_{A\rightarrow B}^{\theta}(\rho_{RA})]\leq1
\end{array}
\right\}  , \label{eq:root-SLD-opt-formula-channels}
\end{align}
where the distinction between the third and last line is that $\partial
_{\theta}\mathcal{N}_{A\rightarrow B}^{\theta}(\rho_{RA})=(\partial_{\theta
}\mathcal{N}_{A\rightarrow B}^{\theta})(\rho_{RA})$ (i.e., for fixed
$\rho_{RA}$, the state $\rho_{RA}$ is constant with respect to the partial derivative.

Now recall the post-selected teleportation identity from
\eqref{eq:PS-TP-identity}:
\begin{equation}
\mathcal{N}_{A\rightarrow B}^{\theta}(\rho_{RA}^{\theta})=\langle\Gamma
|_{AS}\rho_{RA}^{\theta}\otimes\Gamma_{SB}^{\mathcal{N}^{\theta}}
|\Gamma\rangle_{AS}\text{.}
\end{equation}
This implies that
\begin{align}
&  \partial_{\theta}(\mathcal{N}_{A\rightarrow B}^{\theta}(\rho_{RA}^{\theta
}))\nonumber\\
&  =\partial_{\theta}(\langle\Gamma|_{AS}\rho_{RA}^{\theta}\otimes\Gamma
_{SB}^{\mathcal{N}^{\theta}}|\Gamma\rangle_{AS})\\
&  =\langle\Gamma|_{AS}\partial_{\theta}(\rho_{RA}^{\theta}\otimes\Gamma
_{SB}^{\mathcal{N}^{\theta}})|\Gamma\rangle_{AS}\\
&  =\langle\Gamma|_{AS}[(\partial_{\theta}\rho_{RA}^{\theta})\otimes
\Gamma_{SB}^{\mathcal{N}^{\theta}}+\rho_{RA}^{\theta}\otimes(\partial_{\theta
}\Gamma_{SB}^{\mathcal{N}^{\theta}})]|\Gamma\rangle_{AS}\\
&  =\langle\Gamma|_{AS}[(\partial_{\theta}\rho_{RA}^{\theta})\otimes
\Gamma_{SB}^{\mathcal{N}^{\theta}}|\Gamma\rangle_{AS}+\langle\Gamma|_{AS}
\rho_{RA}^{\theta}\otimes(\partial_{\theta}\Gamma_{SB}^{\mathcal{N}^{\theta}
})|\Gamma\rangle_{AS}\\
&  =\mathcal{N}_{A\rightarrow B}^{\theta}(\partial_{\theta}\rho_{RA}^{\theta
})+(\partial_{\theta}\mathcal{N}_{A\rightarrow B}^{\theta})(\rho_{RA}^{\theta
}).
\end{align}
Let $X_{RB}$ be an arbitrary operator satisfying
\begin{equation}
\operatorname{Tr}[(X_{RB}X_{RB}^{\dag}+X_{RB}^{\dag}X_{RB})\mathcal{N}
_{A\rightarrow B}^{\theta}(\rho_{RA}^{\theta})]\leq1.
\label{eq:root-SLD-chain-arb-op}
\end{equation}
Working with the left-hand side of the inequality, we find that
\begin{align}
&  \operatorname{Tr}[(X_{RB}X_{RB}^{\dag}+X_{RB}^{\dag}X_{RB})\mathcal{N}
_{A\rightarrow B}^{\theta}(\rho_{RA}^{\theta})]\nonumber\\
&  =\operatorname{Tr}[(\mathcal{N}_{A\rightarrow B}^{\theta})^{\dag}
(X_{RB}X_{RB}^{\dag}+X_{RB}^{\dag}X_{RB})(\rho_{RA}^{\theta})]\\
&  \geq\operatorname{Tr}[(Z_{RA}Z_{RA}^{\dag}+Z_{RA}^{\dag}Z_{RA})(\rho
_{RA}^{\theta})],
\end{align}
where we set
\begin{equation}
Z_{RA}:=(\mathcal{N}_{A\rightarrow B}^{\theta})^{\dag}(X_{RB}).
\end{equation}
The equality follows because $(\mathcal{N}_{A\rightarrow B}^{\theta})^{\dag}$
is the Hilbert--Schmidt adjoint of $\mathcal{N}_{A\rightarrow B}^{\theta}$,
and the inequality follows because $\rho_{RA}^{\theta}\geq 0$ and
\begin{align}
(\mathcal{N}_{A\rightarrow B}^{\theta})^{\dag}(X_{RB}^{\dag})(\mathcal{N}
_{A\rightarrow B}^{\theta})^{\dag}(X_{RB})  &  \leq(\mathcal{N}_{A\rightarrow
B}^{\theta})^{\dag}(X_{RB}^{\dag}X_{RB}),\\
(\mathcal{N}_{A\rightarrow B}^{\theta})^{\dag}(X_{RB})(\mathcal{N}
_{A\rightarrow B}^{\theta})^{\dag}(X_{RB}^{\dag})  &  \leq(\mathcal{N}
_{A\rightarrow B}^{\theta})^{\dag}(X_{RB}X_{RB}^{\dag}),
\end{align}
which themselves follow from the Schwarz inequality for completely positive
unital maps \cite[Eq.~(3.14)]{Bhat07}. So we conclude that
\begin{equation}
\operatorname{Tr}[(Z_{RA}Z_{RA}^{\dag}+Z_{RA}^{\dag}Z_{RA})(\rho_{RA}^{\theta
})]\leq1. \label{eq:condition-on-adjoint-op-root-SLD-chain}
\end{equation}
Then consider that
\begin{align}
&  \left\vert \operatorname{Tr}[X_{RB}(\partial_{\theta}(\mathcal{N}
_{A\rightarrow B}^{\theta}(\rho_{RA}^{\theta})))]\right\vert \nonumber\\
&  =\left\vert \operatorname{Tr}[X_{RB}((\partial_{\theta}\mathcal{N}
_{A\rightarrow B}^{\theta})(\rho_{RA}^{\theta}))]+\operatorname{Tr}
[X_{RB}\mathcal{N}_{A\rightarrow B}^{\theta}(\partial_{\theta}
\rho_{RA}^{\theta})]\right\vert \\
&  =\left\vert \operatorname{Tr}[X_{RB}((\partial_{\theta}\mathcal{N}
_{A\rightarrow B}^{\theta})(\rho_{RA}^{\theta}))]+\operatorname{Tr}
[(\mathcal{N}_{A\rightarrow B}^{\theta})^{\dag}(X_{RB})(\partial_{\theta}
\rho_{RA}^{\theta})]\right\vert \\
&  \leq\left\vert \operatorname{Tr}[X_{RB}((\partial_{\theta}\mathcal{N}
_{A\rightarrow B}^{\theta})(\rho_{RA}^{\theta}))]\right\vert +\left\vert
\operatorname{Tr}[(\mathcal{N}_{A\rightarrow B}^{\theta})^{\dag}
(X_{RB})(\partial_{\theta}\rho_{RA}^{\theta})]\right\vert .
\end{align}
By applying \eqref{eq:root-SLD-opt-formula-channels}, we find that
\begin{equation}
\sqrt{2}\left\vert \operatorname{Tr}[X_{RB}((\partial_{\theta}\mathcal{N}
_{A\rightarrow B}^{\theta})(\rho_{RA}^{\theta}))]\right\vert \leq\sqrt{I_{F}
}(\theta;\{\mathcal{N}_{A\rightarrow B}^{\theta}\}_{\theta}).
\end{equation}
Since the operator $(\mathcal{N}_{A\rightarrow B}^{\theta})^{\dag}
(X_{RB})=Z_{RA}$ satisfies \eqref{eq:condition-on-adjoint-op-root-SLD-chain},
by applying the optimization in \eqref{eq:root-SLD-opt-formula}, we find that
\begin{equation}
\sqrt{2}\left\vert \operatorname{Tr}[(\mathcal{N}_{A\rightarrow B}^{\theta
})^{\dag}(X_{RB})(\partial_{\theta}\rho_{RA}^{\theta})]\right\vert \leq
\sqrt{I_{F}}(\theta;\{\rho_{RA}^{\theta}\}_{\theta}).
\end{equation}
So we conclude that
\begin{equation}
\sqrt{2}\left\vert \operatorname{Tr}[X_{RB}(\partial_{\theta}(\mathcal{N}
_{A\rightarrow B}^{\theta}(\rho_{RA}^{\theta})))]\right\vert \leq\sqrt{I_{F}
}(\theta;\{\mathcal{N}_{A\rightarrow B}^{\theta}\}_{\theta})+\sqrt{I_{F}
}(\theta;\{\rho_{RA}^{\theta}\}_{\theta}).
\end{equation}
Since $X_{RB}$ is an arbitrary operator satisfying
\eqref{eq:root-SLD-chain-arb-op}, we can optimize over all such operators to
conclude the chain rule inequality in \eqref{eq:chain-rule-root-SLD}.
\end{proof}

\begin{corollary}
\label{cor:amort-collapse-root-SLD}Let $\{\mathcal{N}_{A\rightarrow B}
^{\theta}\}_{\theta}$ be a family of differentiable quantum channels. Then the
following amortization collapse occurs for the root SLD\ Fisher information of
quantum channels:
\begin{equation}
\sqrt{I_{F}}^{\mathcal{A}}(\theta;\{\mathcal{N}_{A\rightarrow B}^{\theta
}\}_{\theta})=\sqrt{I_{F}}(\theta;\{\mathcal{N}_{A\rightarrow B}^{\theta
}\}_{\theta}),
\end{equation}
where
\begin{equation}
\sqrt{I_{F}}^{\mathcal{A}}(\theta;\{\mathcal{N}_{A\rightarrow B}^{\theta
}\}_{\theta}):=\sup_{\{\rho_{RA}^{\theta}\}_{\theta}}\left[  \sqrt{I_{F}
}(\theta;\{\mathcal{N}_{A\rightarrow B}^{\theta}(\rho_{RA}^{\theta}
)\}_{\theta})-\sqrt{I_{F}}(\theta;\{\rho_{RA}^{\theta}\}_{\theta})\right]  .
\end{equation}

\end{corollary}

\begin{proof}
If the finiteness condition in\ \eqref{eq:finiteness-condition-SLD-fish-ch}
does not hold, then the equality trivially holds. So let us suppose that the
finiteness condition in \eqref{eq:finiteness-condition-SLD-fish-ch}\ holds.
The inequality $\geq$ follows from Proposition~\ref{prop:amort->=-ch-Fish-gen}
and the fact that the root SLD\ Fisher information is faithful (see
\eqref{eq:SLD-RLD-Fish-faithful}). The opposite inequality $\leq$ is a
consequence of the chain rule from Proposition~\ref{prop:chain-rule-root-SLD}.
Let $\{\rho_{RA}^{\theta}\}_{\theta}$ be a family of quantum states on a
systems $RA$. Then it follows from Proposition~\ref{prop:chain-rule-root-SLD}
that
\begin{equation}
\sqrt{I_{F}}(\theta;\{\mathcal{N}_{A\rightarrow B}^{\theta}(\rho_{RA}^{\theta
})\}_{\theta})-\sqrt{I_{F}}(\theta;\{\rho_{RA}^{\theta}\}_{\theta})\leq
\sqrt{I_{F}}(\theta;\{\mathcal{N}_{A\rightarrow B}^{\theta}\}_{\theta}).
\end{equation}
Since the family $\{\rho_{RA}^{\theta}\}_{\theta}$ is arbitrary, we can take a
supremum of the left-hand side over all such families, and conclude that
\begin{equation}
\sqrt{I_{F}}^{\mathcal{A}}(\theta;\{\mathcal{N}_{A\rightarrow B}^{\theta
}\}_{\theta})\leq\sqrt{I_{F}}(\theta;\{\mathcal{N}_{A\rightarrow B}^{\theta
}\}_{\theta}).
\end{equation}
This concludes the proof.
\end{proof}

\begin{corollary}
\label{cor:subadd-serial-concat-root-SLD-Fish}Let $\{\mathcal{N}_{A\rightarrow
B}^{\theta}\}_{\theta}$ and $\{\mathcal{M}_{B\rightarrow C}^{\theta}
\}_{\theta}$ be differentiable families of quantum channels. Then the root
SLD\ Fisher information of quantum channels is subadditive with respect to
serial composition, in the following sense:
\begin{equation}
\sqrt{I_{F}}(\theta;\{\mathcal{M}_{B\rightarrow C}^{\theta}\circ
\mathcal{N}_{A\rightarrow B}^{\theta}\}_{\theta})\leq\sqrt{I_{F}}
(\theta;\{\mathcal{N}_{A\rightarrow B}^{\theta}\}_{\theta})+\sqrt{I_{F}
}(\theta;\{\mathcal{M}_{B\rightarrow C}^{\theta}\}_{\theta}).
\label{eq:serial-composition-root-SLD-ch}
\end{equation}

\end{corollary}

\begin{proof}
If the finiteness condition in\ \eqref{eq:finiteness-condition-SLD-fish-ch}
does not hold for either channel, then the inequality trivially holds. So let
us suppose that the finiteness condition in
\eqref{eq:finiteness-condition-SLD-fish-ch}\ holds for both channels. Pick an
arbitrary input state $\omega_{RA}$. Now apply
Proposition~\ref{prop:chain-rule-root-SLD}\ to find that
\begin{align}
&  \sqrt{I_{F}}(\theta;\{\mathcal{M}_{B\rightarrow C}^{\theta}(\mathcal{N}
_{A\rightarrow B}^{\theta}(\omega_{RA}))\}_{\theta})\nonumber\\
&  \leq\sqrt{I_{F}}(\theta;\{\mathcal{N}_{A\rightarrow B}^{\theta}(\omega
_{RA})\}_{\theta})+\sqrt{I_{F}}(\theta;\{\mathcal{M}_{B\rightarrow C}^{\theta
}\}_{\theta})\\
&  \leq\sup_{\omega_{RA}}\sqrt{I_{F}}(\theta;\{\mathcal{N}_{A\rightarrow
B}^{\theta}(\omega_{RA})\}_{\theta})+\sqrt{I_{F}}(\theta;\{\mathcal{M}
_{B\rightarrow C}^{\theta}\}_{\theta})\\
&  =\sqrt{I_{F}}(\theta;\{\mathcal{N}_{A\rightarrow B}^{\theta}\}_{\theta
})+\sqrt{I_{F}}(\theta;\{\mathcal{M}_{B\rightarrow C}^{\theta}\}_{\theta}).
\end{align}
Since the inequality holds for all input states, we conclude that
\begin{equation}
\sup_{\omega_{RA}}\sqrt{I_{F}}(\theta;\{\mathcal{M}_{B\rightarrow C}^{\theta
}(\mathcal{N}_{A\rightarrow B}^{\theta}(\omega_{RA}))\}_{\theta})\leq
\sqrt{I_{F}}(\theta;\{\mathcal{N}_{A\rightarrow B}^{\theta}\}_{\theta}
)+\sqrt{I_{F}}(\theta;\{\mathcal{M}_{B\rightarrow C}^{\theta}\}_{\theta}),
\end{equation}
which implies \eqref{eq:serial-composition-root-SLD-ch}.
\end{proof}

\bigskip

The following bound in \eqref{eq:Heisenberg-bnd-SLD-channels} was reported
recently in \cite{Yuan2017}. Here, we see how it is a consequence of the
QCRB\ in \eqref{eq:QCRB}, the meta-converse from
Theorem~\ref{thm:meta-converse}, and the amortization collapse from
Corollary~\ref{cor:amort-collapse-root-SLD}. At the same time, our approach
offers a technical improvement over the result of \cite{Yuan2017}, in that the
families of quantum channels to which the bound applies need only be
differentiable rather than second-order differentiable, the latter being
required by the approach of \cite{Yuan2017}.

\begin{conclusion}
As a direct consequence of the QCRB\ in \eqref{eq:QCRB}, the meta-converse
from Theorem~\ref{thm:meta-converse}, and the amortization collapse from
Corollary~\ref{cor:amort-collapse-root-SLD}, we conclude the following bound
on the MSE\ of an unbiased estimator $\hat{\theta}$ for all differentiable
quantum channel families:
\begin{equation}
\operatorname{Var}(\hat{\theta})\geq\frac{1}{n^{2}I_{F}(\theta;\{\mathcal{N}
_{A\rightarrow B}^{\theta}\}_{\theta})}.
\label{eq:Heisenberg-bnd-SLD-channels}
\end{equation}
This bound thus poses a ``Heisenberg'' limitation on sequential estimation protocols for all
differentiable quantum channel families satisfying the finiteness condition in \eqref{eq:finiteness-condition-SLD-fish-ch}.
\end{conclusion}

\subsection{RLD\ Fisher information limit on quantum channel parameter
estimation}

\label{sec:RLD-Fish-limits}

\subsubsection{RLD\ Fisher information of quantum channels and its properties}

We now recall and establish some properties of the RLD\ Fisher information of
quantum channels. Following \cite{Hayashi2011} and the general prescription in
Definition~\ref{def:gen-fish-channels}, it is defined as follows:
\begin{equation}
\widehat{I}_{F}(\theta;\{\mathcal{N}_{A\rightarrow B}^{\theta}\}_{\theta
}):=\sup_{\rho_{RA}}\widehat{I}_{F}(\theta;\{\mathcal{N}_{A\rightarrow
B}^{\theta}(\rho_{RA})\}_{\theta}), \label{eq:def-RLD-channel}
\end{equation}
but note that the optimization can be restricted to pure bipartite states, due
to Remark~\ref{rem:restrict-to-pure-bipartite}. Recall that the RLD\ Fisher
information of quantum channels has an explicit formula, as given in \eqref{eq:RLD-Fish-ch}.

The following additivity relation was established in \cite{Hayashi2011}, and
we review its proof in Appendix~\ref{app:proofs-RLD-fish-props}.

\begin{proposition}
\label{prop:add-RLD-Fish-ch}Let $\{\mathcal{N}_{A\rightarrow B}^{\theta
}\}_{\theta}$ and $\{\mathcal{M}_{C\rightarrow D}^{\theta}\}_{\theta}$ be
differentiable families of quantum channels. Then the RLD\ Fisher information
of quantum channels is additive in the following sense:
\begin{equation}
\widehat{I}_{F}(\theta;\{\mathcal{N}_{A\rightarrow B}^{\theta}\otimes
\mathcal{M}_{C\rightarrow D}^{\theta}\}_{\theta})=\widehat{I}_{F}
(\theta;\{\mathcal{N}_{A\rightarrow B}^{\theta}\}_{\theta})+\widehat{I}
_{F}(\theta;\{\mathcal{M}_{C\rightarrow D}^{\theta}\}_{\theta}).
\end{equation}

\end{proposition}

The RLD\ Fisher information of quantum states and channels obeys the following
chain rule:

\begin{proposition}
[Chain rule]\label{prop:chain-rule-RLD}Let $\{\mathcal{N}_{A\rightarrow
B}^{\theta}\}_{\theta}$ be a differentiable family of quantum channels, and
let $\{\rho_{RA}^{\theta}\}_{\theta}$ be a differentiable family of quantum
states on systems $RA$, with the system $R$ of arbitrary size. Then the
following chain rule holds
\begin{equation}
\widehat{I}_{F}(\theta;\{\mathcal{N}_{A\rightarrow B}^{\theta}(\rho
_{RA}^{\theta})\}_{\theta})\leq\widehat{I}_{F}(\theta;\{\mathcal{N}
_{A\rightarrow B}^{\theta}\}_{\theta})+\widehat{I}_{F}(\theta;\{\rho
_{RA}^{\theta}\}_{\theta}).
\end{equation}

\end{proposition}

\begin{proof}
If the finiteness conditions in
\eqref{eq:support-condition-RLD}\ and\ \eqref{eq:finiteness-condition-RLD-fish-ch}
do not hold, then the inequality is trivially satisfied. So let us suppose
that the finiteness conditions
\eqref{eq:support-condition-RLD}\ and\ \eqref{eq:finiteness-condition-RLD-fish-ch}
hold. Recall the following post-selected teleportation identity from
\eqref{eq:PS-TP-identity}:
\begin{equation}
\mathcal{N}_{A\rightarrow B}^{\theta}(\rho_{RA}^{\theta})=\langle\Gamma
|_{AS}\rho_{RA}^{\theta}\otimes\Gamma_{SB}^{\mathcal{N}^{\theta}}
|\Gamma\rangle_{AS}. \label{eq:TP-identity}
\end{equation}
Then we can write
\begin{align}
&  \widehat{I}_{F}(\theta;\{\mathcal{N}_{A\rightarrow B}^{\theta}(\rho
_{RA}^{\theta})\}_{\theta})\nonumber\\
&  =\operatorname{Tr}[(\partial_{\theta}\mathcal{N}_{A\rightarrow B}^{\theta
}(\rho_{RA}^{\theta}))^{2}(\mathcal{N}_{A\rightarrow B}^{\theta}(\rho
_{RA}^{\theta}))^{-1}]\\
&  =\operatorname{Tr}[(\partial_{\theta}(\langle\Gamma|_{AS}\rho_{RA}^{\theta
}\otimes\Gamma_{SB}^{\mathcal{N}^{\theta}}|\Gamma\rangle_{AS}))^{2}
(\langle\Gamma|_{AS}\rho_{RA}^{\theta}\otimes\Gamma_{SB}^{\mathcal{N}^{\theta
}}|\Gamma\rangle_{AS})^{-1}]\\
&  =\operatorname{Tr}[((\langle\Gamma|_{AS}\partial_{\theta}(\rho_{RA}
^{\theta}\otimes\Gamma_{SB}^{\mathcal{N}^{\theta}})|\Gamma\rangle_{AS}
))^{2}(\langle\Gamma|_{AS}\rho_{RA}^{\theta}\otimes\Gamma_{SB}^{\mathcal{N}
^{\theta}}|\Gamma\rangle_{AS})^{-1}]\\
&  \leq\operatorname{Tr}[\langle\Gamma|_{AS}(\partial_{\theta}(\rho
_{RA}^{\theta}\otimes\Gamma_{SB}^{\mathcal{N}^{\theta}}))(\rho_{RA}^{\theta
}\otimes\Gamma_{SB}^{\mathcal{N}^{\theta}})^{-1}(\partial_{\theta}(\rho
_{RA}^{\theta}\otimes\Gamma_{SB}^{\mathcal{N}^{\theta}}))|\Gamma\rangle
_{AS}]\\
&  =\operatorname{Tr}_{RB}[\langle\Gamma|_{AS}(\partial_{\theta}(\rho
_{RA}^{\theta}\otimes\Gamma_{SB}^{\mathcal{N}^{\theta}}))(\rho_{RA}^{\theta
}\otimes\Gamma_{SB}^{\mathcal{N}^{\theta}})^{-1}(\partial_{\theta}(\rho
_{RA}^{\theta}\otimes\Gamma_{SB}^{\mathcal{N}^{\theta}}))|\Gamma\rangle
_{AS}]\\
&  =\langle\Gamma|_{AS}\operatorname{Tr}_{RB}[(\partial_{\theta}(\rho
_{RA}^{\theta}\otimes\Gamma_{SB}^{\mathcal{N}^{\theta}}))(\rho_{RA}^{\theta
}\otimes\Gamma_{SB}^{\mathcal{N}^{\theta}})^{-1}(\partial_{\theta}(\rho
_{RA}^{\theta}\otimes\Gamma_{SB}^{\mathcal{N}^{\theta}}))]|\Gamma\rangle_{AS}.
\label{eq:proof-chain-rule-final-line}
\end{align}
The second equality follows from applying \eqref{eq:TP-identity}, and the
inequality is a consequence of the transformer inequality in
Lemma~\ref{lem:transformer-ineq-basic}, with
\begin{align}
L  &  =\langle\Gamma|_{AS}\otimes I_{RB},\\
X  &  =\partial_{\theta}(\rho_{RA}^{\theta}\otimes\Gamma_{SB}^{\mathcal{N}
^{\theta}}),\\
Y  &  =\rho_{RA}^{\theta}\otimes\Gamma_{SB}^{\mathcal{N}^{\theta}}.
\end{align}
Now consider that
\[
\partial_{\theta}(\rho_{RA}^{\theta}\otimes\Gamma_{SB}^{\mathcal{N}^{\theta}
})=(\partial_{\theta}\rho_{RA}^{\theta})\otimes\Gamma_{SB}^{\mathcal{N}
^{\theta}}+\rho_{RA}^{\theta}\otimes(\partial_{\theta}\Gamma_{SB}
^{\mathcal{N}^{\theta}}).
\]
Right multiplying this by $(\rho_{RA}^{\theta}\otimes\Gamma_{SB}
^{\mathcal{N}^{\theta}})^{-1}$ gives
\begin{align}
&  (\partial_{\theta}(\rho_{RA}^{\theta}\otimes\Gamma_{SB}^{\mathcal{N}
^{\theta}}))(\rho_{RA}^{\theta}\otimes\Gamma_{SB}^{\mathcal{N}^{\theta}}
)^{-1}\nonumber\\
&  =(\partial_{\theta}\rho_{RA}^{\theta})(\rho_{RA}^{\theta})^{-1}
\otimes\Gamma_{SB}^{\mathcal{N}^{\theta}}(\Gamma_{SB}^{\mathcal{N}^{\theta}
})^{-1}+\rho_{RA}^{\theta}(\rho_{RA}^{\theta})^{-1}\otimes(\partial_{\theta
}\Gamma_{SB}^{\mathcal{N}^{\theta}})(\Gamma_{SB}^{\mathcal{N}^{\theta}}
)^{-1}\\
&  =(\partial_{\theta}\rho_{RA}^{\theta})(\rho_{RA}^{\theta})^{-1}\otimes
\Pi_{\Gamma^{\mathcal{N}^{\theta}}}+\Pi_{\rho_{RA}^{\theta}}\otimes
(\partial_{\theta}\Gamma_{SB}^{\mathcal{N}^{\theta}})(\Gamma_{SB}
^{\mathcal{N}^{\theta}})^{-1}.
\end{align}
Right multiplying the last line by $(\partial_{\theta}(\rho_{RA}^{\theta
}\otimes\Gamma_{SB}^{\mathcal{N}^{\theta}}))$ gives
\begin{align}
&  \left[  (\partial_{\theta}\rho_{RA}^{\theta})(\rho_{RA}^{\theta}
)^{-1}\otimes\Pi_{\Gamma^{\mathcal{N}^{\theta}}}+\Pi_{\rho_{RA}^{\theta}
}\otimes(\partial_{\theta}\Gamma_{SB}^{\mathcal{N}^{\theta}})(\Gamma
_{SB}^{\mathcal{N}^{\theta}})^{-1}\right]  (\partial_{\theta}(\rho
_{RA}^{\theta}\otimes\Gamma_{SB}^{\mathcal{N}^{\theta}}))\nonumber\\
&  =\left[  (\partial_{\theta}\rho_{RA}^{\theta})(\rho_{RA}^{\theta}
)^{-1}\otimes\Pi_{\Gamma^{\mathcal{N}^{\theta}}}+\Pi_{\rho_{RA}^{\theta}
}\otimes(\partial_{\theta}\Gamma_{SB}^{\mathcal{N}^{\theta}})(\Gamma
_{SB}^{\mathcal{N}^{\theta}})^{-1}\right] \nonumber\\
&  \qquad\times\left[  (\partial_{\theta}\rho_{RA}^{\theta})\otimes\Gamma
_{SB}^{\mathcal{N}^{\theta}}+\rho_{RA}^{\theta}\otimes(\partial_{\theta}
\Gamma_{SB}^{\mathcal{N}^{\theta}})\right] \\
&  =(\partial_{\theta}\rho_{RA}^{\theta})(\rho_{RA}^{\theta})^{-1}
(\partial_{\theta}\rho_{RA}^{\theta})\otimes\Gamma_{SB}^{\mathcal{N}^{\theta}
}+(\partial_{\theta}\rho_{RA}^{\theta})(\rho_{RA}^{\theta})^{-1}\rho
_{RA}^{\theta}\otimes\Pi_{\Gamma^{\mathcal{N}^{\theta}}}(\partial_{\theta
}\Gamma_{SB}^{\mathcal{N}^{\theta}})\nonumber\\
&  \qquad+\Pi_{\rho_{RA}^{\theta}}(\partial_{\theta}\rho_{RA}^{\theta}
)\otimes(\partial_{\theta}\Gamma_{SB}^{\mathcal{N}^{\theta}})(\Gamma
_{SB}^{\mathcal{N}^{\theta}})^{-1}\Gamma_{SB}^{\mathcal{N}^{\theta}}+\rho
_{RA}^{\theta}\otimes(\partial_{\theta}\Gamma_{SB}^{\mathcal{N}^{\theta}
})(\Gamma_{SB}^{\mathcal{N}^{\theta}})^{-1}(\partial_{\theta}\Gamma
_{SB}^{\mathcal{N}^{\theta}})\nonumber\\
&  =(\partial_{\theta}\rho_{RA}^{\theta})(\rho_{RA}^{\theta})^{-1}
(\partial_{\theta}\rho_{RA}^{\theta})\otimes\Gamma_{SB}^{\mathcal{N}^{\theta}
}+(\partial_{\theta}\rho_{RA}^{\theta})\Pi_{\rho_{RA}^{\theta}}\otimes
\Pi_{\Gamma^{\mathcal{N}^{\theta}}}(\partial_{\theta}\Gamma_{SB}
^{\mathcal{N}^{\theta}})\nonumber\\
&  \qquad+\Pi_{\rho_{RA}^{\theta}}(\partial_{\theta}\rho_{RA}^{\theta}
)\otimes(\partial_{\theta}\Gamma_{SB}^{\mathcal{N}^{\theta}})\Pi
_{\Gamma^{\mathcal{N}^{\theta}}}+\rho_{RA}^{\theta}\otimes(\partial_{\theta
}\Gamma_{SB}^{\mathcal{N}^{\theta}})(\Gamma_{SB}^{\mathcal{N}^{\theta}}
)^{-1}(\partial_{\theta}\Gamma_{SB}^{\mathcal{N}^{\theta}}).
\end{align}
Since the finiteness conditions $\Pi_{\rho_{RA}^{\theta}}^{\perp}
(\partial_{\theta}\rho_{RA}^{\theta})=(\partial_{\theta}\rho_{RA}^{\theta}
)\Pi_{\rho_{RA}^{\theta}}^{\perp}=0$ and $\Pi_{\Gamma^{\mathcal{N}^{\theta}}
}^{\perp}(\partial_{\theta}\Gamma_{SB}^{\mathcal{N}^{\theta}})=(\partial
_{\theta}\Gamma_{SB}^{\mathcal{N}^{\theta}})\Pi_{\Gamma^{\mathcal{N}^{\theta}
}}^{\perp}=0$ hold, we can \textquotedblleft add in\textquotedblright\ extra
zero terms to the two middle terms above to conclude that
\begin{multline}
(\partial_{\theta}(\rho_{RA}^{\theta}\otimes\Gamma_{SB}^{\mathcal{N}^{\theta}
}))(\rho_{RA}^{\theta}\otimes\Gamma_{SB}^{\mathcal{N}^{\theta}})^{-1}
(\partial_{\theta}(\rho_{RA}^{\theta}\otimes\Gamma_{SB}^{\mathcal{N}^{\theta}
}))=(\partial_{\theta}\rho_{RA}^{\theta})(\rho_{RA}^{\theta})^{-1}
(\partial_{\theta}\rho_{RA}^{\theta})\otimes\Gamma_{SB}^{\mathcal{N}^{\theta}
}\\
+2(\partial_{\theta}\rho_{RA}^{\theta})\otimes(\partial_{\theta}\Gamma
_{SB}^{\mathcal{N}^{\theta}})+\rho_{RA}^{\theta}\otimes(\partial_{\theta
}\Gamma_{SB}^{\mathcal{N}^{\theta}})(\Gamma_{SB}^{\mathcal{N}^{\theta}}
)^{-1}(\partial_{\theta}\Gamma_{SB}^{\mathcal{N}^{\theta}}).
\end{multline}
Now taking the partial trace over $RB$, we find the following for each term:
\begin{align}
\operatorname{Tr}_{RB}[(\partial_{\theta}\rho_{RA}^{\theta})(\rho_{RA}
^{\theta})^{-1}(\partial_{\theta}\rho_{RA}^{\theta})\otimes\Gamma
_{SB}^{\mathcal{N}^{\theta}}]  &  =\operatorname{Tr}_{R}[(\partial_{\theta
}\rho_{RA}^{\theta})(\rho_{RA}^{\theta})^{-1}(\partial_{\theta}\rho
_{RA}^{\theta})]\otimes I_{S},\\
\operatorname{Tr}_{RB}[2(\partial_{\theta}\rho_{RA}^{\theta})\otimes
(\partial_{\theta}\Gamma_{SB}^{\mathcal{N}^{\theta}})]  &  =2\operatorname{Tr}
_{R}[(\partial_{\theta}\rho_{RA}^{\theta})]\otimes\operatorname{Tr}
_{B}[(\partial_{\theta}\Gamma_{SB}^{\mathcal{N}^{\theta}})]\\
&  =2\operatorname{Tr}_{R}[(\partial_{\theta}\rho_{RA}^{\theta})]\otimes
(\partial_{\theta}\operatorname{Tr}_{B}[\Gamma_{SB}^{\mathcal{N}^{\theta}}])\\
&  =2\operatorname{Tr}_{R}[(\partial_{\theta}\rho_{RA}^{\theta})]\otimes
(\partial_{\theta}(I_{S})])\\
&  =0,\\
\operatorname{Tr}_{RB}[\rho_{RA}^{\theta}\otimes(\partial_{\theta}\Gamma
_{SB}^{\mathcal{N}^{\theta}})(\Gamma_{SB}^{\mathcal{N}^{\theta}}
)^{-1}(\partial_{\theta}\Gamma_{SB}^{\mathcal{N}^{\theta}})]  &  =\rho
_{A}^{\theta}\otimes\operatorname{Tr}_{B}[(\partial_{\theta}\Gamma
_{SB}^{\mathcal{N}^{\theta}})(\Gamma_{SB}^{\mathcal{N}^{\theta}}
)^{-1}(\partial_{\theta}\Gamma_{SB}^{\mathcal{N}^{\theta}})].
\end{align}
Now applying the sandwich $\langle\Gamma|_{AS}(\cdot)|\Gamma\rangle_{AS}$, the
first and last terms become as follows:
\begin{align}
&  \langle\Gamma|_{AS}\operatorname{Tr}_{R}[(\partial_{\theta}\rho
_{RA}^{\theta})(\rho_{RA}^{\theta})^{-1}(\partial_{\theta}\rho_{RA}^{\theta
})]\otimes I_{S}|\Gamma\rangle_{AS}\nonumber\\
&  =\operatorname{Tr}[(\partial_{\theta}\rho_{RA}^{\theta})(\rho_{RA}^{\theta
})^{-1}(\partial_{\theta}\rho_{RA}^{\theta})]\\
&  =\operatorname{Tr}[(\partial_{\theta}\rho_{RA}^{\theta})^{2}(\rho
_{RA}^{\theta})^{-1}],
\end{align}
and
\begin{multline}
\langle\Gamma|_{AS}\rho_{A}^{\theta}\otimes\operatorname{Tr}_{B}
[(\partial_{\theta}\Gamma_{SB}^{\mathcal{N}^{\theta}})(\Gamma_{SB}
^{\mathcal{N}^{\theta}})^{-1}(\partial_{\theta}\Gamma_{SB}^{\mathcal{N}
^{\theta}})]|\Gamma\rangle_{AS}\\
=\operatorname{Tr}[(\rho_{S}^{\theta})^{T}\operatorname{Tr}_{B}[(\partial
_{\theta}\Gamma_{SB}^{\mathcal{N}^{\theta}})(\Gamma_{SB}^{\mathcal{N}^{\theta
}})^{-1}(\partial_{\theta}\Gamma_{SB}^{\mathcal{N}^{\theta}})]].
\end{multline}
Plugging back into \eqref{eq:proof-chain-rule-final-line}, we find that
\begin{align}
&  \langle\Gamma|_{AS}\operatorname{Tr}_{RB}[(\partial_{\theta}(\rho
_{RA}^{\theta}\otimes\Gamma_{SB}^{\mathcal{N}^{\theta}}))(\rho_{RA}^{\theta
}\otimes\Gamma_{SB}^{\mathcal{N}^{\theta}})^{-1}(\partial_{\theta}(\rho
_{RA}^{\theta}\otimes\Gamma_{SB}^{\mathcal{N}^{\theta}}))]|\Gamma\rangle
_{AS}\nonumber\\
&  =\operatorname{Tr}[(\partial_{\theta}\rho_{RA}^{\theta})^{2}(\rho
_{RA}^{\theta})^{-1}]+\operatorname{Tr}[(\rho_{S}^{\theta})^{T}
\operatorname{Tr}_{B}[(\partial_{\theta}\Gamma_{SB}^{\mathcal{N}^{\theta}
})(\Gamma_{SB}^{\mathcal{N}^{\theta}})^{-1}(\partial_{\theta}\Gamma
_{SB}^{\mathcal{N}^{\theta}})]]\\
&  \leq\operatorname{Tr}[(\partial_{\theta}\rho_{RA}^{\theta})^{2}(\rho
_{RA}^{\theta})^{-1}]+\left\Vert \operatorname{Tr}_{B}[(\partial_{\theta
}\Gamma_{SB}^{\mathcal{N}^{\theta}})(\Gamma_{SB}^{\mathcal{N}^{\theta}}
)^{-1}(\partial_{\theta}\Gamma_{SB}^{\mathcal{N}^{\theta}})]\right\Vert
_{\infty}\\
&  =\widehat{I}_{F}(\theta;\{\rho_{RA}^{\theta}\}_{\theta})+\widehat{I}
_{F}(\theta;\{\mathcal{N}_{A\rightarrow B}^{\theta}\}_{\theta}).
\end{align}
This concludes the proof.
\end{proof}

\begin{corollary}
\label{cor:amort-collapse-RLD-fish}Let $\{\mathcal{N}_{A\rightarrow B}
^{\theta}\}_{\theta}$ be a differentiable family of quantum channels. Then
amortization does not increase the RLD\ Fisher information of quantum
channels, in the following sense:
\begin{equation}
\widehat{I}_{F}^{\mathcal{A}}(\theta;\{\mathcal{N}_{A\rightarrow B}^{\theta
}\}_{\theta})=\widehat{I}_{F}(\theta;\{\mathcal{N}_{A\rightarrow B}^{\theta
}\}_{\theta}).
\end{equation}

\end{corollary}

\begin{proof}
If the finiteness condition in \eqref{eq:finiteness-condition-RLD-fish-ch}
does not hold, then the equality trivially holds. So let us suppose that the
finiteness condition in \eqref{eq:finiteness-condition-RLD-fish-ch} holds. The
inequality $\geq$ follows from Proposition~\ref{prop:amort->=-ch-Fish-gen} and
the fact that the RLD\ Fisher information is faithful (see
\eqref{eq:SLD-RLD-Fish-faithful}). The opposite inequality $\leq$ is a
consequence of the chain rule from Proposition~\ref{prop:chain-rule-RLD}. Let
$\{\rho_{RA}^{\theta}\}_{\theta}$ be a family of quantum states on systems
$RA$. Then it follows from Proposition~\ref{prop:chain-rule-RLD} that
\begin{equation}
\widehat{I}_{F}(\theta;\{\mathcal{N}_{A\rightarrow B}^{\theta}(\rho
_{RA}^{\theta})\}_{\theta})-\widehat{I}_{F}(\theta;\{\rho_{RA}^{\theta
}\}_{\theta})\leq\widehat{I}_{F}(\theta;\{\mathcal{N}_{A\rightarrow B}
^{\theta}\}_{\theta}).
\end{equation}
Since the family $\{\rho_{RA}^{\theta}\}_{\theta}$ is arbitrary, we can take a
supremum over the left-hand side over all such families, and conclude that
\begin{equation}
\widehat{I}_{F}^{\mathcal{A}}(\theta;\{\mathcal{N}_{A\rightarrow B}^{\theta
}\}_{\theta})\leq\widehat{I}_{F}(\theta;\{\mathcal{N}_{A\rightarrow B}
^{\theta}\}_{\theta}).
\end{equation}
This concludes the proof.
\end{proof}

\begin{corollary}
\label{cor:subadd-serial-concat-RLD-Fish}Let $\{\mathcal{N}_{A\rightarrow
B}^{\theta}\}_{\theta}$ and $\{\mathcal{M}_{B\rightarrow C}^{\theta}
\}_{\theta}$ be differentiable families of quantum channels. Then the
RLD\ Fisher information of quantum channels is subadditive with respect to
serial composition, in the following sense:
\begin{equation}
\widehat{I}_{F}(\theta;\{\mathcal{M}_{B\rightarrow C}^{\theta}\circ
\mathcal{N}_{A\rightarrow B}^{\theta}\}_{\theta})\leq\widehat{I}_{F}
(\theta;\{\mathcal{N}_{A\rightarrow B}^{\theta}\}_{\theta})+\widehat{I}
_{F}(\theta;\{\mathcal{M}_{B\rightarrow C}^{\theta}\}_{\theta}).
\label{eq:serial-composition-RLD-ch}
\end{equation}

\end{corollary}

\begin{proof}
If the finiteness condition in \eqref{eq:finiteness-condition-RLD-fish-ch}
does not hold for both channels, then the inequality is trivially satisfied.
So let us suppose that the finiteness condition in
\eqref{eq:finiteness-condition-RLD-fish-ch} holds for both channels. Pick an
arbitrary input state $\omega_{RA}$. Now apply
Proposition~\ref{prop:chain-rule-RLD}\ to find that
\begin{align}
&  \widehat{I}_{F}(\theta;\{\mathcal{M}_{B\rightarrow C}^{\theta}
(\mathcal{N}_{A\rightarrow B}^{\theta}(\omega_{RA}))\}_{\theta})\nonumber\\
&  \leq\widehat{I}_{F}(\theta;\{\mathcal{N}_{A\rightarrow B}^{\theta}
(\omega_{RA})\}_{\theta})+\widehat{I}_{F}(\theta;\{\mathcal{M}_{B\rightarrow
C}^{\theta}\}_{\theta})\\
&  \leq\sup_{\omega_{RA}}\widehat{I}_{F}(\theta;\{\mathcal{N}_{A\rightarrow
B}^{\theta}(\omega_{RA})\}_{\theta})+\widehat{I}_{F}(\theta;\{\mathcal{M}
_{B\rightarrow C}^{\theta}\}_{\theta})\\
&  =\widehat{I}_{F}(\theta;\{\mathcal{N}_{A\rightarrow B}^{\theta}\}_{\theta
})+\widehat{I}_{F}(\theta;\{\mathcal{M}_{B\rightarrow C}^{\theta}\}_{\theta}).
\end{align}
Since the inequality holds for all input states, we conclude that
\begin{equation}
\sup_{\omega_{RA}}\widehat{I}_{F}(\theta;\{\mathcal{M}_{B\rightarrow
C}^{\theta}(\mathcal{N}_{A\rightarrow B}^{\theta}(\omega_{RA}))\}_{\theta
})\leq\widehat{I}_{F}(\theta;\{\mathcal{N}_{A\rightarrow B}^{\theta}
\}_{\theta})+\widehat{I}_{F}(\theta;\{\mathcal{M}_{B\rightarrow C}^{\theta
}\}_{\theta}),
\end{equation}
which implies \eqref{eq:serial-composition-RLD-ch}.
\end{proof}

\subsubsection{RLD Fisher information bound for general channel parameter
estimation}

\begin{conclusion}
\label{conc:RLD-bnd}As a direct consequence of the QCRB\ in
\eqref{eq:RLD-QCRB}, the meta-converse from Theorem~\ref{thm:meta-converse},
and the amortization collapse from Corollary~\ref{cor:amort-collapse-RLD-fish}
, we conclude the following bound on the MSE\ of an unbiased estimator
$\hat{\theta}$ for all quantum channel families $\{\mathcal{N}_{A\rightarrow
B}^{\theta}\}_{\theta}$:
\begin{equation}
\operatorname{Var}(\hat{\theta})\geq\frac{1}{n\widehat{I}_{F}(\theta
;\{\mathcal{N}_{A\rightarrow B}^{\theta}\}_{\theta})}. \label{eq:RLD-bnd-ch}
\end{equation}
This bound thus poses a strong limitation on sequential estimation protocols
for all differentiable quantum channel families satisfying the finiteness
condition in \eqref{eq:finiteness-condition-RLD-fish-ch}.
\end{conclusion}

Conclusion~\ref{conc:RLD-bnd} strengthens one of the results of
\cite{Hayashi2011}. There, it was proved that the RLD\ Fisher information of
quantum channels is a limitation for parallel estimation protocols, but
Conclusion~\ref{conc:RLD-bnd} establishes it as a limitation for the more general
sequential estimation protocols.

Conclusion~\ref{conc:RLD-bnd} establishes \eqref{eq:finiteness-condition-RLD-fish-ch} as a sufficient condition for the unattainability of Heisenberg scaling. 
In a recent paper \cite[Theorem 1]{Zhou2020} concurrent to ours, a necessary and sufficient condition for the unattainability of Heisenberg scaling with general sequential estimation protocols has been established.

\subsection{Example: Estimating the parameters of the generalized amplitude
damping channel} \label{subsec:estimating-gadc}

We now apply the bound in \eqref{eq:RLD-bnd-ch} to a\ particular example, the
generalized amplitude damping channel \cite{NC00}. This channel has been
studied previously in the context of quantum estimation theory
\cite{Fujiwara_2003, Fujiwara2004}, where the SLD\ Fisher information of
quantum channels was studied. Our goal now is to compute the RLD\ Fisher
information of this channel with respect to its parameters.

Recall that a generalized amplitude damping channel is defined in terms of its
loss $\gamma\in(0,1)$ and noise $N\in\left(  0,1\right)  $ as
\begin{equation}
\mathcal{A}_{\gamma,N}(\rho):=K_{1}\rho K_{1}^{\dag}+K_{2}\rho K_{2}^{\dag
}+K_{3}\rho K_{3}^{\dag}+K_{4}\rho K_{4}^{\dag}, \label{eq:GADC}
\end{equation}
where
\begin{align}
K_{1}  &  :=\sqrt{1-N}\left(  |0\rangle\!\langle0|+\sqrt{1-\gamma}
|1\rangle\!\langle1|\right)  ,\\
K_{2}  &  :=\sqrt{\gamma\left(  1-N\right)  }|0\rangle\!\langle1|,\\
K_{3}  &  :=\sqrt{N}\left(  \sqrt{1-\gamma}|0\rangle\!\langle0|+|1\rangle
\langle1|\right)  ,\\
K_{4}  &  :=\sqrt{\gamma N}|1\rangle\!\langle0|.
\end{align}
The Choi operator of the channel is then given by
\begin{align}
\Gamma_{RB}^{\mathcal{A}_{\gamma,N}}  &  :=(\operatorname{id}_{R}
\otimes\mathcal{A}_{\gamma,N})(\Gamma_{RA})\\
&  =\left(  1-\gamma N\right)  |00\rangle\!\langle00|+\sqrt{1-\gamma}\left(
|00\rangle\!\langle11|+|11\rangle\!\langle00|\right)  +\gamma N|01\rangle
\langle01|\nonumber\\
&  \qquad+\gamma\left(  1-N\right)  |10\rangle\!\langle10|+\left(
1-\gamma\left(  1-N\right)  \right)  |11\rangle\!\langle11|\\
&  =
\begin{bmatrix}
1-\gamma N & 0 & 0 & \sqrt{1-\gamma}\\
0 & \gamma N & 0 & 0\\
0 & 0 & \gamma\left(  1-N\right)  & 0\\
\sqrt{1-\gamma} & 0 & 0 & 1-\gamma\left(  1-N\right)
\end{bmatrix}
.
\end{align}

\subsubsection{Estimating loss}

Let us apply this approach to the generalized amplitude damping channel, and
in particular, with the goal of finding limits on estimating the loss
parameter $\gamma\in\left(  0,1\right)  $. By direct evaluation, we find that
\begin{equation}
\partial_{\gamma}\Gamma_{RB}^{\mathcal{A}_{\gamma,N}}=
\begin{bmatrix}
-N & 0 & 0 & -\frac{1}{2\sqrt{1-\gamma}}\\
0 & N & 0 & 0\\
0 & 0 & 1-N & 0\\
-\frac{1}{2\sqrt{1-\gamma}} & 0 & 0 & -\left(  1-N\right)
\end{bmatrix}
.
\end{equation}
Then we evaluate the expression in \eqref{eq:RLD-Fish-ch}, which for our case
is as follows:
\begin{equation}
\widehat{I}_{F}(\gamma;\{\mathcal{A}_{\gamma,N}\}_{\gamma})=\left\Vert
\operatorname{Tr}_{B}\left[  \left(  \partial_{\gamma}\Gamma_{RB}
^{\mathcal{A}_{\gamma,N}}\right)  \left(  \Gamma_{RB}^{\mathcal{A}_{\gamma,N}
}\right)  ^{-1}\left(  \partial_{\gamma}\Gamma_{RB}^{\mathcal{A}_{\gamma,N}
}\right)  \right]  \right\Vert _{\infty}.
\end{equation}
Using the fact that
\begin{equation}
\left(  \Gamma_{RB}^{\mathcal{A}_{\gamma,N}}\right)  ^{-1}=
\begin{bmatrix}
\frac{1-\gamma\left(  1-N\right)  }{\left(  1-N\right)  N\gamma^{2}} & 0 & 0 &
\frac{-\sqrt{1-\gamma}}{\left(  1-N\right)  N\gamma^{2}}\\
0 & \frac{1}{\gamma N} & 0 & 0\\
0 & 0 & \frac{1}{\gamma\left(  1-N\right)  } & 0\\
\frac{-\sqrt{1-\gamma}}{\left(  1-N\right)  N\gamma^{2}} & 0 & 0 &
\frac{1-\gamma N}{\left(  1-N\right)  N\gamma^{2}}
\end{bmatrix}
, \label{eq:GADC-inverse}
\end{equation}
we find that
\begin{equation}
\operatorname{Tr}_{B}\left[  \left(  \partial_{\gamma}\Gamma_{RB}
^{\mathcal{A}_{\gamma,N}}\right)  \left(  \Gamma_{RB}^{\mathcal{A}_{\gamma,N}
}\right)  ^{-1}\left(  \partial_{\gamma}\Gamma_{RB}^{\mathcal{A}_{\gamma,N}
}\right)  \right]  =
\begin{bmatrix}
f_{1}(\gamma,N) & 0\\
0 & f_{2}(\gamma,N)
\end{bmatrix}
,
\end{equation}
where
\begin{align}
f_{1}(\gamma,N)  &  :=\frac{\frac{1}{N-\gamma N}+\frac{1}{1-N}-4}{4 \gamma^2},\\
f_{2}(\gamma,N)  &  :=\frac{\frac{1}{(1-\gamma) (1-N)}+\frac{1}{N}-4}{4 \gamma^2}.
\end{align}
Note that if $N\leq1/2$, then $f_{1}(\gamma,N)\geq f_{2}(\gamma,N)$, while if
$N>1/2$, then $f_{1}(\gamma,N)<f_{2}(\gamma,N)$. It then follows that
\begin{equation}
\widehat{I}_{F}(\gamma;\{\mathcal{A}_{\gamma,N}\}_{\gamma})=\left\{
\begin{array}
[c]{cc}
f_{1}(\gamma,N) & N\leq1/2\\
f_{2}(\gamma,N) & N>1/2
\end{array}
\right.  . \label{eq:estimate-gamma-GADC-RLD}
\end{equation}

Thus, it follows from \eqref{eq:RLD-bnd-ch} that the formula in
\eqref{eq:estimate-gamma-GADC-RLD} provides a fundamental limitation on any
protocol that attempts to estimate the loss parameter $\gamma$. For the noise
parameter $N$ equal to $0.2$ and $0.45$, Figure~\ref{fig:estimate-loss}
depicts the logarithm of this bound, as well as the logarithm of the
achievable bound from the SLD\ Fisher information of channels, corresponding
to a parallel strategy that estimates $\gamma$. The RLD\ bound becomes better
as $N$ approaches $1/2$, and we find numerically that the RLD\ and SLD\ bounds
coincide at $N=1/2$.

\begin{figure}[ptb]
\begin{subfigure}{.5\textwidth}
\centering
\includegraphics[width=.9\linewidth]{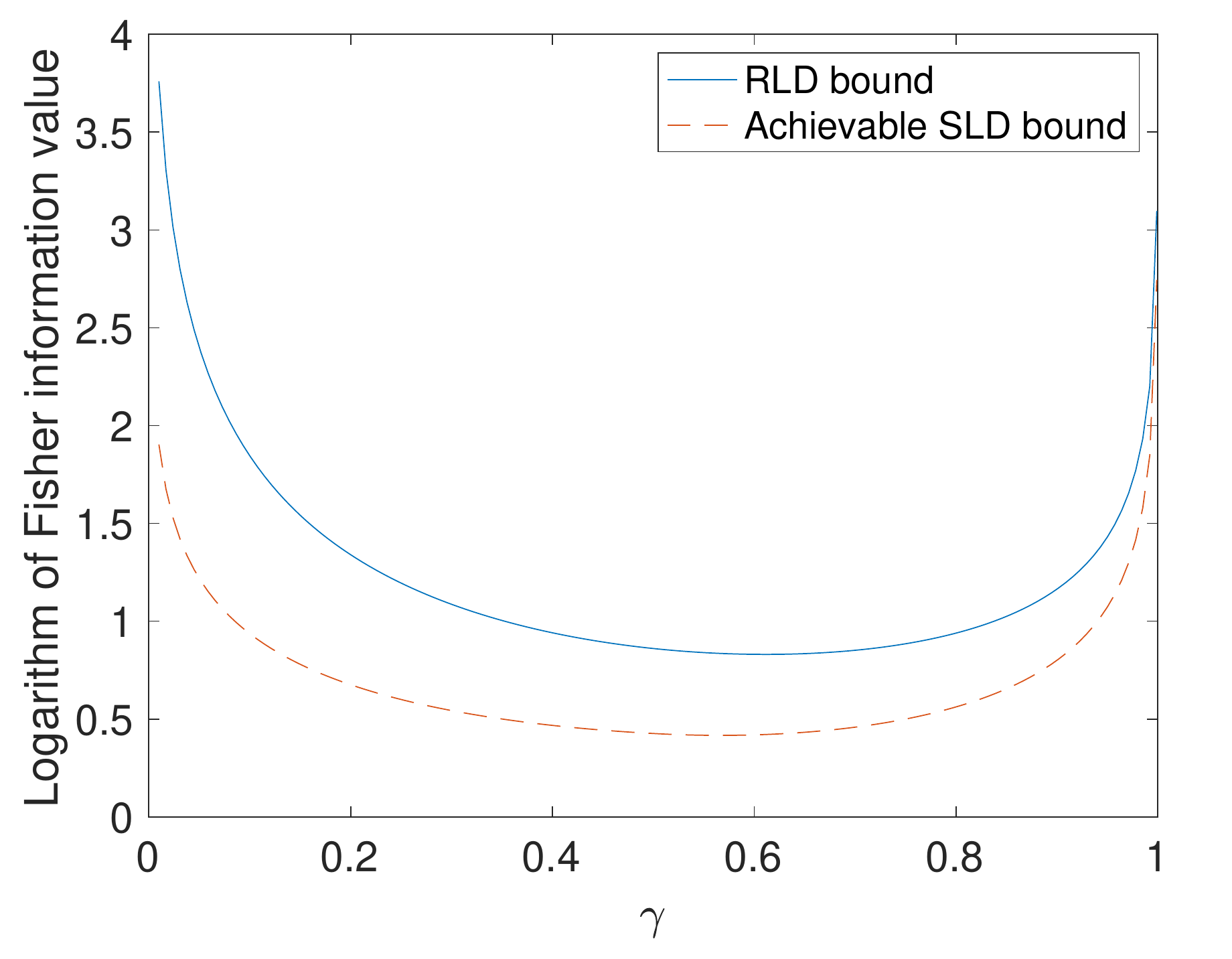}
\caption{}
\label{fig:estimate-loss-N-0.2}
\end{subfigure}
\begin{subfigure}{.5\textwidth}
\centering
\includegraphics[width=.9\linewidth]{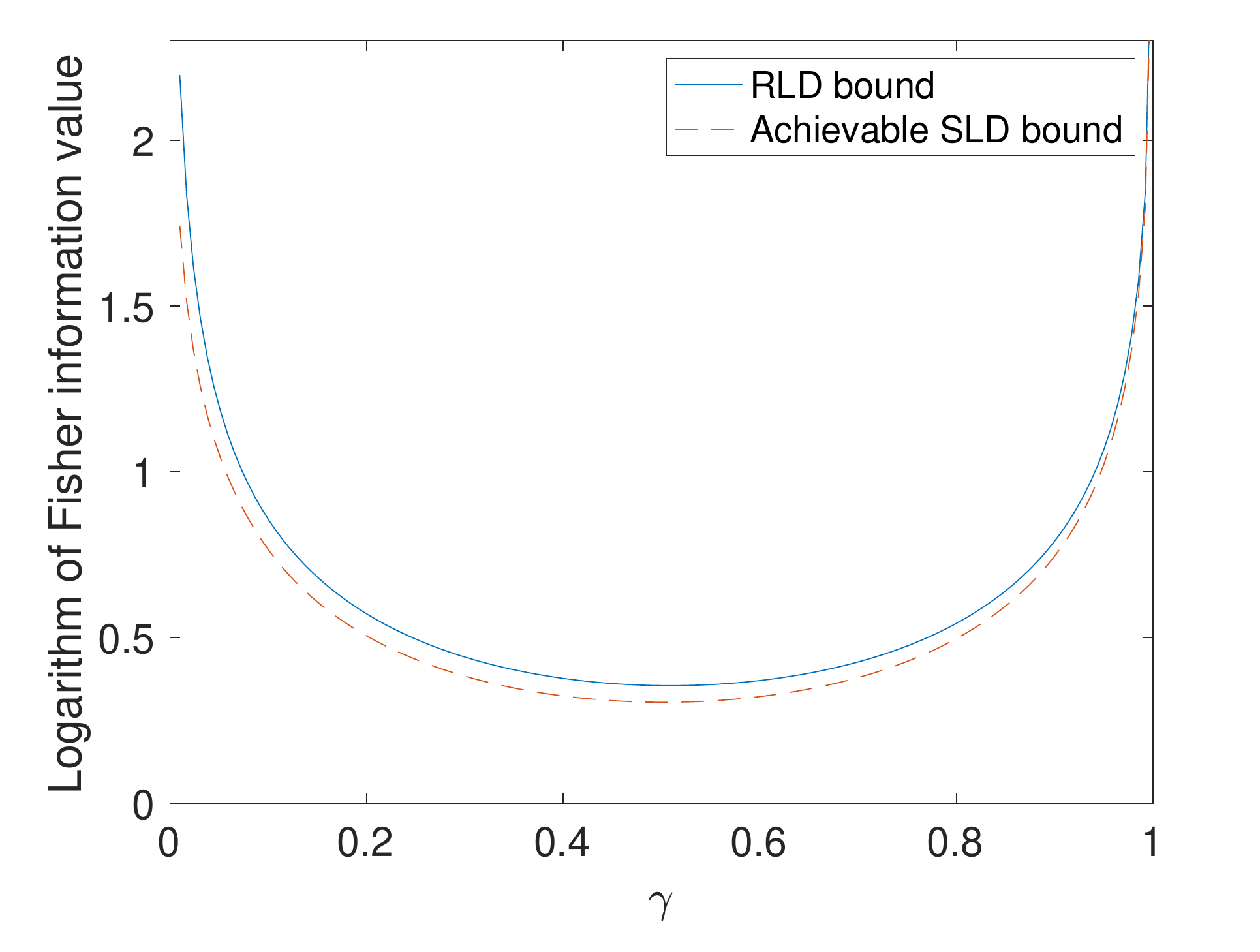}
\caption{}
\label{fig:estimate-loss-N-0.45}
\end{subfigure}
\caption{(a)~Logarithm of RLD bound and achievable SLD bound versus loss
$\gamma$ for noise $N=0.2$, when estimating the loss $\gamma$. (b)~Logarithm
of the RLD bound and achievable SLD bound versus loss $\gamma$ for noise
$N=0.45$, when estimating the loss $\gamma$.} 
\label{fig:estimate-loss} 
\end{figure}

\subsubsection{Estimating noise}

Now suppose that we are interested in estimating the noise parameter $N$\ of a
generalized amplitude damping channel. We find that
\begin{equation}
\partial_{N}\Gamma_{RB}^{\mathcal{A}_{\gamma,N}}=-\gamma\left(  I_{2}
\otimes\sigma_{Z}\right)  .
\end{equation}
Then by exploiting \eqref{eq:GADC-inverse}, we find that
\begin{equation}
\operatorname{Tr}_{B}\left[  \left(  \partial_{N}\Gamma_{RB}^{\mathcal{A}
_{\gamma,N}}\right)  \left(  \Gamma_{RB}^{\mathcal{A}_{\gamma,N}}\right)
^{-1}\left(  \partial_{N}\Gamma_{RB}^{\mathcal{A}_{\gamma,N}}\right)  \right]
=
\begin{bmatrix}
\frac{1}{N(1-N)} & 0\\
0 & \frac{1}{N(1-N)}
\end{bmatrix}
.
\end{equation}
%Thus, if $N>1/2$, then the first entry is the maximum, whereas if $N<1/2$,
%then the second one is the maximum. We can summarize this as
%\begin{equation}
%\left\Vert \operatorname{Tr}_{B}\left[  \left(  \partial_{N}\Gamma_{RB}^{\mathcal{A}_{\gamma,N}}\right)  \left(  \Gamma_{RB}^{\mathcal{A}_{\gamma,N}}\right)  ^{-1}\left(  \partial_{N}\Gamma_{RB}^{\mathcal{A}_{\gamma,N}}\right)  \right]  \right\Vert _{\infty}=\frac{1+\left\vert1-2N\right\vert \gamma}{\left(  1-N\right)  N},
%\end{equation}
Thus we have
\begin{equation}
\widehat{I}_{F}(N;\{\mathcal{A}_{\gamma,N}\}_{N})=\frac{1}{N(1-N)}.
\end{equation}

For the loss parameter $\gamma$ equal to $0.5$ and $0.8$,
Figure~\ref{fig:estimate-noise} depicts the logarithm of the RLD bound, as
well as the logarithm of the achievable bound from the SLD\ Fisher information
of channels, corresponding to a parallel strategy that estimates $N$. The
RLD\ bound becomes better as $\gamma$ approaches $1$.

\begin{figure}[ptb]
\begin{subfigure}{.5\textwidth}
\centering
\includegraphics[width=\linewidth]{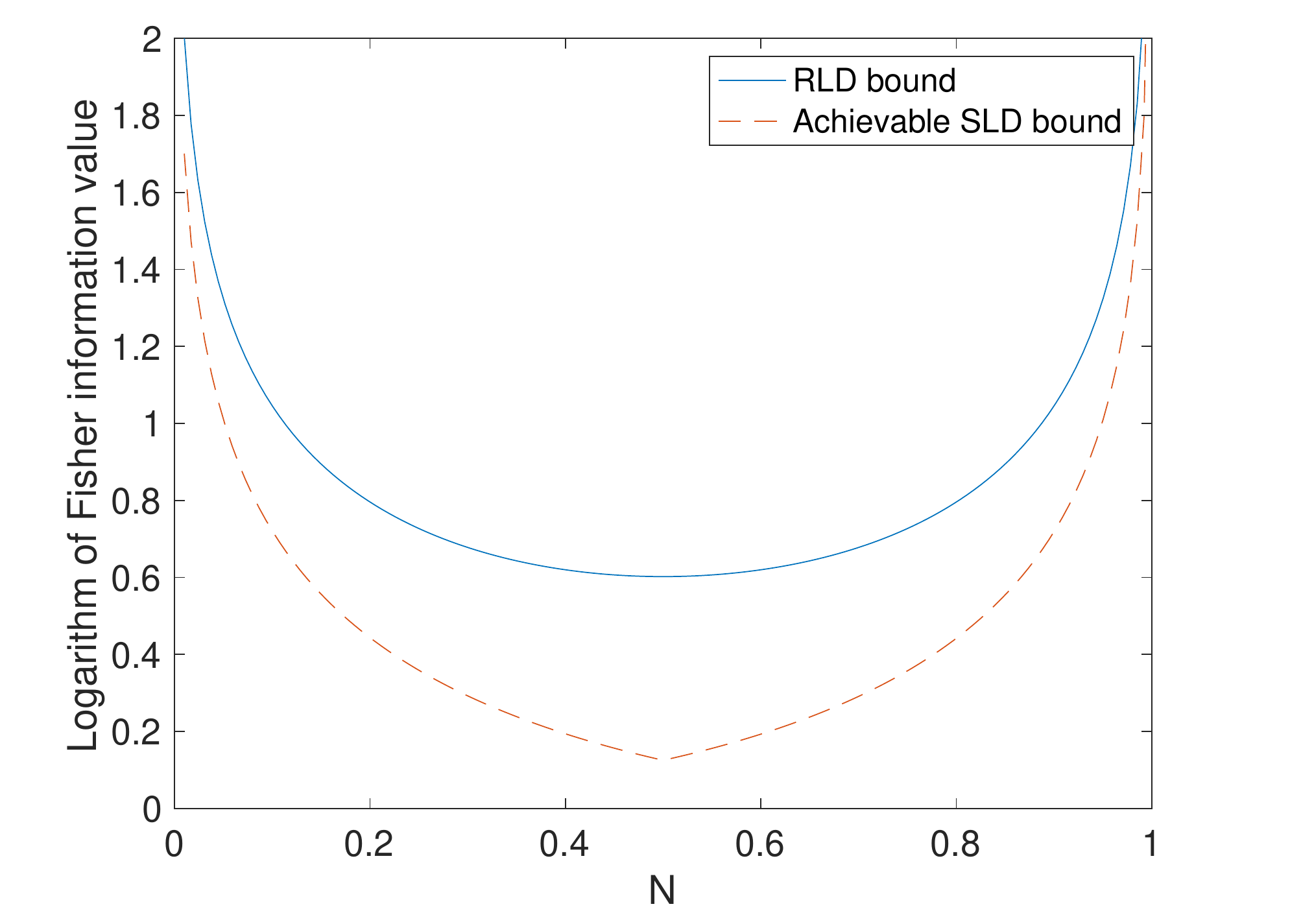}
\caption{}
\label{fig:estimate-noise-g-0.5}
\end{subfigure}
\begin{subfigure}{.5\textwidth}
\centering
\includegraphics[width=\linewidth]{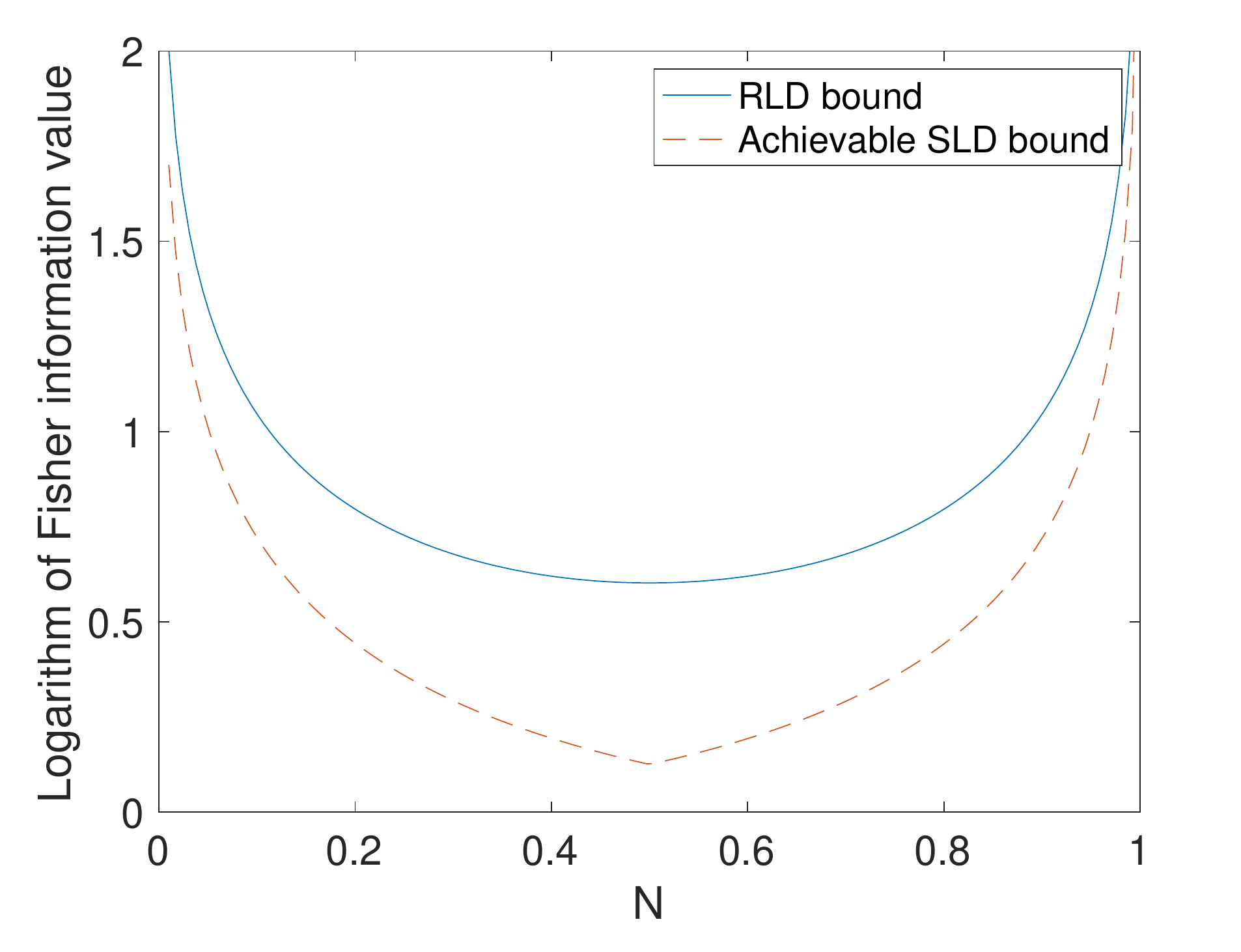}
\caption{}
\label{fig:estimate-loss-g-0.8}
\end{subfigure}
\caption{(a)~Logarithm of RLD bound and achievable SLD bound versus noise $N$
for loss $\gamma=0.5$, when estimating the noise $N$. (b)~Logarithm of the RLD
bound and achievable SLD bound versus noise $N$ for loss $\gamma=0.8$, when
estimating the noise $N$.}
\label{fig:estimate-noise}
\end{figure}

\subsubsection{Estimating a phase in loss and noise}

Now let us suppose that we have a combination of a coherent process and the
generalized amplitude damping channel. In particular, let us suppose that a
phase $\phi$ is encoded in a unitary $e^{-i\phi\sigma_{Z}}$, and this is
followed by the generalized amplitude damping channel. Then this process is
\begin{equation}
\mathcal{A}_{\phi,\gamma,N}(\rho):=\mathcal{A}_{\gamma,N}(e^{-i\phi\sigma_{Z}
}\rho e^{i\phi\sigma_{Z}}).
\end{equation}
The goal is to estimate the phase $\phi$.

The Choi operator is given by
\begin{equation}
\Gamma_{RB}^{\mathcal{A}_{\phi,\gamma,N}}:=
\begin{bmatrix}
1-\gamma N & 0 & 0 & e^{-i2\phi}\sqrt{1-\gamma}\\
0 & \gamma N & 0 & 0\\
0 & 0 & \gamma\left(  1-N\right)  & 0\\
e^{i2\phi}\sqrt{1-\gamma} & 0 & 0 & 1-\gamma\left(  1-N\right)
\end{bmatrix}
,
\end{equation}
and we find that
\begin{equation}
\partial_{\phi}\Gamma_{RB}^{\mathcal{A}_{\phi,\gamma,N}}=
\begin{bmatrix}
0 & 0 & 0 & -2ie^{-i2\phi}\sqrt{1-\gamma}\\
0 & 0 & 0 & 0\\
0 & 0 & 0 & 0\\
2ie^{i2\phi}\sqrt{1-\gamma} & 0 & 0 & 0
\end{bmatrix}
\end{equation}
Using the fact that
\begin{equation}
\left(  \Gamma_{RB}^{\mathcal{A}_{\phi}}\right)  ^{-1}=
\begin{bmatrix}
\frac{1-\gamma\left(  1-N\right)  }{\left(  1-N\right)  N\gamma^{2}} & 0 & 0 &
\frac{-e^{-2i\phi}\sqrt{1-\gamma}}{\left(  1-N\right)  N\gamma^{2}}\\
0 & \frac{1}{\gamma N} & 0 & 0\\
0 & 0 & \frac{1}{\gamma\left(  1-N\right)  } & 0\\
\frac{-e^{2i\phi}\sqrt{1-\gamma}}{\left(  1-N\right)  N\gamma^{2}} & 0 & 0 &
\frac{1-\gamma N}{\left(  1-N\right)  N\gamma^{2}}
\end{bmatrix}
,
\end{equation}
we find that
\begin{equation}
\operatorname{Tr}_{B}\left[  \left(  \partial_{\phi}\Gamma_{RB}^{\mathcal{A}
_{\phi,\gamma,N}}\right)  \left(  \Gamma_{RB}^{\mathcal{A}_{\phi,\gamma,N}
}\right)  ^{-1}\left(  \partial_{\phi}\Gamma_{RB}^{\mathcal{A}_{\phi,\gamma
,N}}\right)  \right]  =
\begin{bmatrix}
\frac{4\left(  1-\gamma\right)  \left(  1-\gamma N\right)  }{\left(
1-N\right)  N\gamma^{2}} & 0\\
0 & \frac{4\left(  1-\gamma\right)  \left(  1-\gamma\left(  1-N\right)
\right)  }{\left(  1-N\right)  N\gamma^{2}}
\end{bmatrix}
.
\end{equation}
Then if $N>1/2$, we have that
\begin{equation}
\left\Vert \operatorname{Tr}_{B}\left[  \left(  \partial_{\phi}\Gamma
_{RB}^{\mathcal{A}_{\phi,\gamma,N}}\right)  \left(  \Gamma_{RB}^{\mathcal{A}
_{\phi,\gamma,N}}\right)  ^{-1}\left(  \partial_{\phi}\Gamma_{RB}
^{\mathcal{A}_{\phi,\gamma,N}}\right)  \right]  \right\Vert _{\infty}
=\frac{4\left(  1-\gamma\right)  \left(  1-\gamma\left(  1-N\right)  \right)
}{\left(  1-N\right)  N\gamma^{2}},
\end{equation}
while if $N\leq1/2$, then
\begin{equation}
\left\Vert \operatorname{Tr}_{B}\left[  \left(  \partial_{\phi}\Gamma
_{RB}^{\mathcal{A}_{\phi,\gamma,N}}\right)  \left(  \Gamma_{RB}^{\mathcal{A}
_{\phi,\gamma,N}}\right)  ^{-1}\left(  \partial_{\phi}\Gamma_{RB}
^{\mathcal{A}_{\phi,\gamma,N}}\right)  \right]  \right\Vert _{\infty}
=\frac{4\left(  1-\gamma\right)  \left(  1-\gamma N\right)  }{\left(
1-N\right)  N\gamma^{2}}.
\end{equation}
So we conclude that
\begin{equation}
\widehat{I}_{F}(\phi;\{\mathcal{A}_{\phi,\gamma,N}\}_{\phi})=\frac{4\left(
1-\gamma\right)  \left(  1-\gamma\left(  N+\left(  1-2N\right)
u(2N-1)\right)  \right)  }{\left(  1-N\right)  N\gamma^{2}},
\end{equation}
where
\begin{equation}
u(x)=\left\{
\begin{array}
[c]{cc}
1 & x>0\\
0 & x\leq0
\end{array}
\right.  .
\end{equation}

For the noise parameter $N$ equal to $0.2$ and $0.45$,
Figure~\ref{fig:estimate-phase} depicts the logarithm of the RLD bound, as
well as the logarithm of the achievable bound from the SLD\ Fisher information
of channels, corresponding to a parallel strategy that estimates the phase
$\phi$ at $\phi=0.1$. The RLD\ bound becomes better as $\gamma$ approaches $1$.

\begin{figure}[ptb]
\begin{subfigure}{.5\textwidth}
\centering
\includegraphics[width=.9\linewidth]{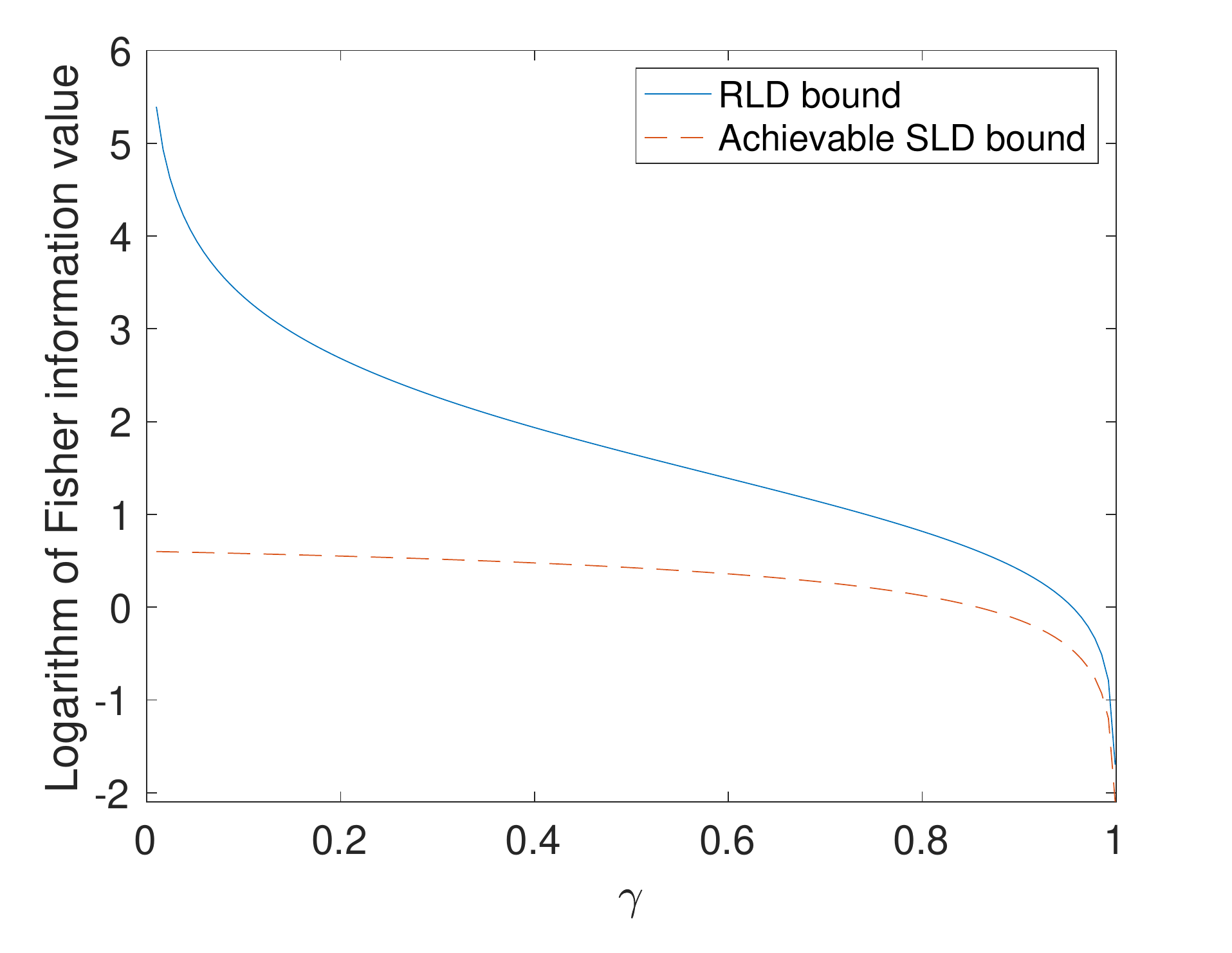}
\caption{}
\label{fig:estimate-phase-N-0.2}
\end{subfigure}
\begin{subfigure}{.5\textwidth}
\centering
\includegraphics[width=.9\linewidth]{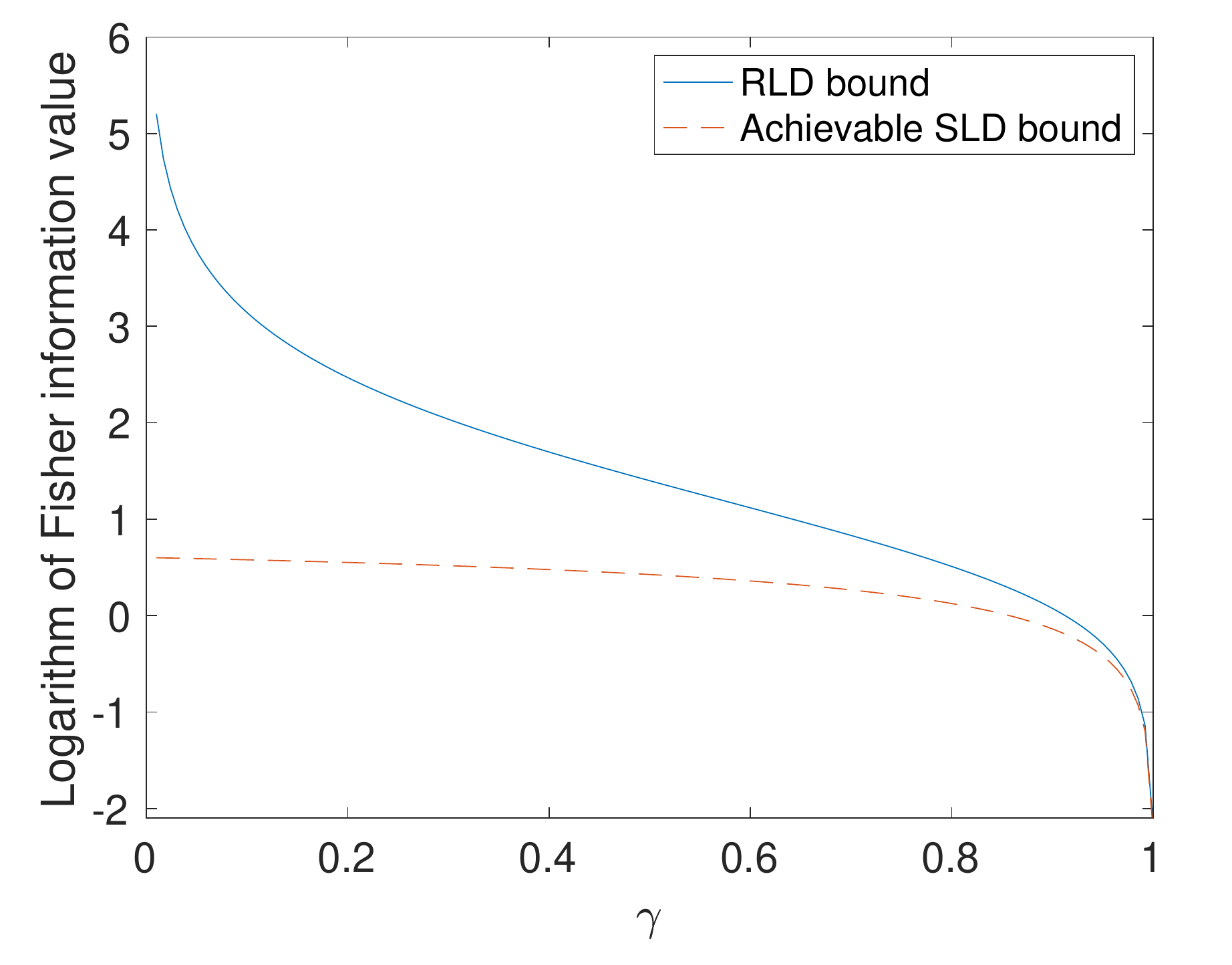}
\caption{}
\label{fig:estimate-phase-N-0.45}
\end{subfigure}
\caption{(a)~Logarithm of RLD bound and achievable SLD bound versus loss
$\gamma$ for noise $N=0.2$, when estimating the phase $\phi= 0.1$.
(b)~Logarithm of the RLD bound and achievable SLD bound versus loss $\gamma$
for noise $N=0.45$, when estimating the phase $\phi= 0.1$.}
\label{fig:estimate-phase}
\end{figure}

\section{Limits on quantum channel discrimination}

\label{sec:bounds-on-disc}

In this section, we shift to quantum channel discrimination, which has some
close ties to the theory of quantum channel estimation, as discussed in
Section~\ref{sec:setting-q-ch-est-disc}. The main tool that we use for the
analysis here is the geometric R\'enyi relative entropy, which we review in
what follows and in more detail in Appendix~\ref{app:geo-ren-props}.

\subsection{Geometric R\'enyi relative entropy} \label{subsec:geo-rel-ent}

The geometric R\'{e}nyi relative entropy is a key distinguishability measure
that we employ in the context of quantum channel discrimination, and it is
even connected to the RLD\ Fisher information, as we discuss in the
forthcoming Section~\ref{sec:connections}. The geometric R\'{e}nyi relative
entropy has its roots in the early work \cite{PR98}, and the specific form
given below was introduced by \cite{Mat13,Matsumoto2018}. It has been reviewed
briefly in \cite{T15book}\ and in more detail in \cite{HM17} (in particular,
see \cite[Example~4.5]{HM17}). See also \cite{KW20book} for a more recent review. It has been used effectively in recent work to
obtain upper bounds on quantum channel capacities \cite{Fang2019a} and rates
of channel discrimination in the asymmetric setting \cite[Appendix~D]
{Fang2019a}. This latter paper has thus established the geometric R\'enyi
relative entropy as a useful tool in bounding rates of operational tasks.

We define the geometric R\'enyi relative entropy as follows:

\begin{definition}
[Geometric R\'{e}nyi relative entropy]
\label{def:geometric-renyi-rel-ent-main-text}Let $\rho$ be a state, $\sigma$ a
positive semi-definite operator, and $\alpha\in(0,1)\cup(1,\infty)$. The
geometric R\'{e}nyi relative quasi-entropy is defined as
\begin{equation}
\widehat{Q}_{\alpha}(\rho\Vert\sigma):=\lim_{\varepsilon\rightarrow0^{+}
}\operatorname{Tr}\!\left[  \sigma_{\varepsilon}\!\left(  \sigma_{\varepsilon
}^{-\frac{1}{2}}\rho\sigma_{\varepsilon}^{-\frac{1}{2}}\right)  ^{\alpha
}\right]  , \label{eq:def-geometric-renyi-rel-quasi-ent-main-text}
\end{equation}
where $\sigma_{\varepsilon}:=\sigma+\varepsilon I$. The geometric R\'{e}nyi
relative entropy is then defined as
\begin{equation}
\widehat{D}_{\alpha}(\rho\Vert\sigma):=\frac{1}{\alpha-1} \ln\widehat
{Q}_{\alpha}(\rho\Vert\sigma).
\label{eq:def-geometric-renyi-rel-ent-main-text}
\end{equation}

\end{definition}

It is called the geometric R\'{e}nyi relative entropy \cite{Fang2019a}
\ because it can be written in terms of the weighted operator geometric mean
as
\begin{equation}
\widehat{Q}_{\alpha}(\rho\Vert\sigma)=\operatorname{Tr}[G_{\alpha}(\sigma
,\rho)],
\end{equation}
where the weighted operator geometric mean is defined as
\begin{equation}
G_{\alpha}(\sigma,\rho):=\lim_{\varepsilon\rightarrow0^{+}}\sigma
_{\varepsilon}^{\frac{1}{2}}\!\left(  \sigma_{\varepsilon}^{-\frac{1}{2}}
\rho\sigma_{\varepsilon}^{-\frac{1}{2}}\right)  ^{\alpha}\sigma_{\varepsilon
}^{\frac{1}{2}}.
\end{equation}
See, e.g., \cite{LL01} for a review of operator geometric means.

When the condition $\operatorname{supp}(\rho)\subseteq\operatorname{supp}
(\sigma)$ holds, the geometric R\'{e}nyi relative entropy can be written for
$\alpha\in(0,1)\cup(1,\infty)$ as
\begin{equation}
\widehat{Q}_{\alpha}(\rho\Vert\sigma):=\operatorname{Tr}\!\left[
\sigma\!\left(  \sigma^{-\frac{1}{2}}\rho\sigma^{-\frac{1}{2}}\right)
^{\alpha}\right]  . \label{eq:geo-quasi-ent-main-txt}
\end{equation}
For $\alpha\in(0,1)$, if the condition $\operatorname{supp}(\rho
)\subseteq\operatorname{supp}(\sigma)$ does not hold, then the explicit
formula for it is more complicated, given by \cite{Mat14,Mat14condconv}
\begin{equation}
\widehat{Q}_{\alpha}(\rho\Vert\sigma):=\operatorname{Tr}\!\left[
\sigma\!\left(  \sigma^{-\frac{1}{2}}\tilde{\rho}\sigma^{-\frac{1}{2}}\right)
^{\alpha}\right]  , \label{eq:geo-ren-exp-form-weird-case}
\end{equation}
where
\begin{align}
\tilde{\rho}  &  :=\rho_{0,0}-\rho_{0,1}\rho_{1,1}^{-1}\rho_{0,1}^{\dag},\\
\rho_{0,0}  &  :=\Pi_{\sigma}\rho\Pi_{\sigma},\quad\rho_{0,1}:=\Pi_{\sigma
}\rho\Pi_{\sigma}^{\perp},\quad\rho_{1,1}:=\Pi_{\sigma}^{\perp}\rho\Pi
_{\sigma}^{\perp},
\end{align}
$\Pi_{\sigma}$ is the projection onto the support of $\sigma$, $\Pi_{\sigma
}^{\perp}$ the projection onto the kernel of $\sigma$, and all inverses are
evaluated on the supports of the operators. We detail how this explicit
formula follows from \eqref{eq:def-geometric-renyi-rel-quasi-ent-main-text} in
Appendix~\ref{app:geo-ren-props}. For $\alpha\in(1,\infty)$, if the condition
$\operatorname{supp}(\rho)\subseteq\operatorname{supp}(\sigma)$ does not hold,
then it is equal to $+\infty$.

A special case of the geometric R\'{e}nyi relative entropy of interest to us
here, for $\alpha=1/2$, involves the geometric fidelity \cite{Mat10fid,Mat14}:
\begin{equation}
\widehat{D}_{1/2}(\rho\Vert\sigma):=-2\ln\widehat{Q}_{1/2}(\rho\Vert
\sigma)=-\ln\widehat{F}(\rho,\sigma),
\end{equation}
where the geometric fidelity of $\rho$ and $\sigma$ is defined as
\begin{equation}
\widehat{F}(\rho,\sigma):=\left(  \lim_{\varepsilon\rightarrow0^{+}
}\operatorname{Tr}\!\left[  \sigma_{\varepsilon}\!\left(  \sigma_{\varepsilon
}^{-\frac{1}{2}}\rho\sigma_{\varepsilon}^{-\frac{1}{2}}\right)  ^{\frac{1}{2}}\right]
\right)  ^{2} .
\end{equation}
A recent paper has explored the geometric fidelity (therein called Matsumoto fidelity) and its relation to semi-definite programming \cite{CS20}.

The geometric R\'{e}nyi relative entropy has a number of fundamental
properties that make it a worthwhile quantity to study. Although it is not
known to have an information-theoretic interpretation on its own, it is an upper bound
on other information quantities that are connected to operational tasks. The
important properties of geometric R\'{e}nyi relative entropy are as follows:

\begin{itemize}
\item Convergence to the Belavkin--Staszewski relative entropy
\cite{Belavkin1982}\ in the limit $\alpha\rightarrow1$:
\begin{equation}
\lim_{\alpha\rightarrow1}\widehat{D}_{\alpha}(\rho\Vert\sigma)=\widehat
{D}(\rho\Vert\sigma),
\end{equation}
where the Belavkin--Staszewski relative entropy $\widehat{D}(\rho\Vert\sigma)$
is defined as
\begin{equation}
\widehat{D}(\rho\Vert\sigma):=\left\{
\begin{array}
[c]{cc}
\operatorname{Tr}[\rho\ln(\rho^{\frac{1}{2}}\sigma^{-1}\rho^{\frac{1}{2}})] &
\text{if }\operatorname{supp}(\rho)\subseteq\operatorname{supp}(\sigma)\\
+\infty & \text{otherwise}
\end{array}
\right.  .
\end{equation}

\item Convergence to the max-relative entropy \cite{Datta2009b}\ in the limit
$\alpha\rightarrow\infty$:
\begin{equation}
\lim_{\alpha\rightarrow\infty}\widehat{D}_{\alpha}(\rho\Vert\sigma)=D_{\max
}(\rho\Vert\sigma), \label{eq:geo-renyi-to-dmax-main-text}
\end{equation}
as proven in Appendix~\ref{app:geo-ren-props}, where
\begin{equation}
D_{\max}(\rho\Vert\sigma):=\left\{
\begin{array}
[c]{cc}
\ln\inf\{\lambda\geq0:\rho\leq\lambda\sigma\} & \text{if }\operatorname{supp}
(\rho)\subseteq\operatorname{supp}(\sigma)\\
+\infty & \text{otherwise}
\end{array}
\right.  .
\end{equation}

\item For all $\alpha\in(0,1)\cup(1,2]$, data-processing inequality
\cite{Mat13,Matsumoto2018}:
\begin{equation}
\widehat{D}_{\alpha}(\rho\Vert\sigma)\geq\widehat{D}_{\alpha}(\mathcal{N}
(\rho)\Vert\mathcal{N}(\sigma)), \label{eq:DP-geo-Renyi-main-txt}
\end{equation}
where $\rho$ is quantum state, $\sigma$ is a positive semi-definite operator,
and $\mathcal{N}$ is a quantum channel.

\item Monotonicity in $\alpha$ for $\alpha\in(0,1)\cup(1,\infty)$. That is,
\begin{equation}
\widehat{D}_{\alpha}(\rho\Vert\sigma)\leq\widehat{D}_{\beta}(\rho\Vert\sigma),
\label{eq:mono-geo-renyi-main-text}
\end{equation}
for $0<\alpha\leq\beta$, as proven in Appendix~\ref{app:geo-ren-props}.

\item Not smaller than the sandwiched R\'{e}nyi relative entropy for all
$\alpha\in(0,1)\cup(1,\infty)$ \cite{T15book,WWW19}:
\begin{equation}
\widetilde{D}_{\alpha}(\rho\Vert\sigma)\leq\widehat{D}_{\alpha}(\rho
\Vert\sigma), \label{eq:sandwiched-to-geo-bound-main-text}
\end{equation}
where the sandwiched R\'{e}nyi relative entropy is defined as
\cite{MDSFT13,WWY14}
\begin{equation}
\widetilde{D}_{\alpha}(\rho\Vert\sigma):=\lim_{\varepsilon\rightarrow0^{+}
}\frac{1}{\alpha-1}\ln\operatorname{Tr}[(\sigma_{\varepsilon}^{\frac{1-\alpha
}{2\alpha}}\rho\sigma_{\varepsilon}^{\frac{1-\alpha}{2\alpha}})^{\alpha}].
\end{equation}
Special case for $\alpha=1/2$: geometric fidelity is not larger than the
fidelity \cite{Uhl76}:
\begin{equation}
\widehat{F}(\rho,\sigma) \leq F(\rho,\sigma) := \left\Vert \sqrt{\rho} \sqrt{\sigma} \right\Vert _{1}^{2}
.
\end{equation}

\item The geometric R\'{e}nyi relative entropy of two quantum states $\rho$
and $\sigma$ is computable via a semi-definite program \cite{Fang2019a}. 
\end{itemize}

As indicated above, we provide a detailed review of the geometric R\'{e}nyi
relative entropy and its properties in Appendix~\ref{app:geo-ren-props}.

\subsection{Properties of geometric R\'{e}nyi relative entropy of quantum
channels} \label{subsec:geo-rel-ent-properties}

In this section, we discuss some properties of the geometric R\'{e}nyi
relative entropy of quantum channels. These properties were established in
\cite{Fang2019a}\ for the interval $\alpha\in(1,2]$ and implicitly under suitable support
conditions on the Choi operators of the channels, but the interval $\alpha
\in(0,1)$ was not discussed in \cite{Fang2019a}, nor the case when the support
conditions do not hold. Our main observation here is that the same properties
hold for the full interval $\alpha\in(0,1)\cup(1,2]$ and without support
conditions, by following essentially the same proofs from \cite{Fang2019a}.
For completeness, we provide proofs in
Appendix~\ref{app:geo-renyi-channels-app}.

As observed in \cite{CMW14,LKDW18}, any state distinguishability measure can
be generalized to quantum channels by optimizing over all input states to the
channel. Thus, the geometric R\'{e}nyi relative entropy of quantum channels is
defined as follows:

\begin{definition}
For a quantum channel $\mathcal{N}_{A\rightarrow B}$ and a completely positive
map $\mathcal{M}_{A\rightarrow B}$, their geometric R\'{e}nyi relative entropy
is defined for $\alpha\in(0,1)\cup(1,\infty)$ as
\begin{equation}
\widehat{D}_{\alpha}(\mathcal{N}\Vert\mathcal{M}):=\sup_{\rho_{RA}}\widehat
{D}_{\alpha}(\mathcal{N}_{A\rightarrow B}(\rho_{RA})\Vert\mathcal{M}
_{A\rightarrow B}(\rho_{RA})).
\end{equation}

\end{definition}

By applying Remark~\ref{rem:restrict-to-pure-bipartite}, the formula
simplifies as follows for the data-processing interval $\alpha\in
(0,1)\cup(1,2]$:
\begin{equation}
\widehat{D}_{\alpha}(\mathcal{N}\Vert\mathcal{M})=\sup_{\psi_{RA}}\widehat
{D}_{\alpha}(\mathcal{N}_{A\rightarrow B}(\psi_{RA})\Vert\mathcal{M}
_{A\rightarrow B}(\psi_{RA})),
\end{equation}
where the supremum is with respect to all pure bipartite states $\psi_{RA}$
with system $R$ isomorphic to system $A$.

In fact, the formula simplifies further:

\begin{proposition}
\label{prop:explicit-formula-geo-renyi-ch}Let $\mathcal{N}_{A\rightarrow B}$
and $\mathcal{M}_{A\rightarrow B}$ be quantum channels, and let $\Gamma
_{RB}^{\mathcal{N}}$ and $\Gamma_{RB}^{\mathcal{M}}$ be their respective Choi
operators. For $\alpha\in(0,1)\cup(1,2]$, the geometric R\'{e}nyi relative
entropy of quantum channels $\mathcal{N}_{A\rightarrow B}$ and $\mathcal{M}
_{A\rightarrow B}$ has the following explicit form:
\begin{equation}
\widehat{D}_{\alpha}(\mathcal{N}\Vert\mathcal{M})=\frac{1}{\alpha-1}
\ln\widehat{Q}_{\alpha}(\mathcal{N}\Vert\mathcal{M}),
\end{equation}
where
\begin{equation}
\widehat{Q}_{\alpha}(\mathcal{N}\Vert\mathcal{M}):=\left\{
\begin{array}
[c]{cc}
\lambda_{\min}(\operatorname{Tr}_{B}[G_{\alpha}(\Gamma_{RB}^{\mathcal{M}
},\Gamma_{RB}^{\mathcal{N}})]) &
\begin{array}
[c]{c}
\text{if }\alpha\in(0,1)\\
\text{and }\operatorname{supp}(\Gamma_{RB}^{\mathcal{N}})\subseteq
\operatorname{supp}(\Gamma_{RB}^{\mathcal{M}})
\end{array}
\\
& \\
\left\Vert \operatorname{Tr}_{B}[G_{\alpha}(\Gamma_{RB}^{\mathcal{M}}
,\Gamma_{RB}^{\mathcal{N}})]\right\Vert _{\infty} &
\begin{array}
[c]{c}
\text{if }\alpha\in(1,2]\\
\text{and }\operatorname{supp}(\Gamma_{RB}^{\mathcal{N}})\subseteq
\operatorname{supp}(\Gamma_{RB}^{\mathcal{M}})
\end{array}
\\
& \\
\lambda_{\min}(\operatorname{Tr}_{B}[G_{\alpha}(\Gamma_{RB}^{\mathcal{M}
},\widetilde{\Gamma_{RB}^{\mathcal{N}}})]) &
\begin{array}
[c]{c}
\text{if }\alpha\in(0,1)\\
\text{and }\operatorname{supp}(\Gamma_{RB}^{\mathcal{N}})\not \subseteq
\operatorname{supp}(\Gamma_{RB}^{\mathcal{M}})
\end{array}
\\
& \\
+\infty &
\begin{array}
[c]{c}
\text{if }\alpha\in(1,2]\\
\text{and }\operatorname{supp}(\Gamma_{RB}^{\mathcal{N}})\not \subseteq
\operatorname{supp}(\Gamma_{RB}^{\mathcal{M}})
\end{array}
\end{array}
\right.  , \label{eq:geo-ren-ch-exp-form-all-cases}
\end{equation}
$\lambda_{\min}$ denotes the minimum eigenvalue of its argument,
\begin{align}
G_{\alpha}(X,Y)  &  :=X^{1/2}(X^{-1/2}YX^{-1/2})^{\alpha}X^{1/2},\\
\widetilde{\Gamma_{RB}^{\mathcal{N}}}  &  :=(\Gamma_{RB}^{\mathcal{N}}
)_{0,0}-(\Gamma_{RB}^{\mathcal{N}})_{0,1}(\Gamma_{RB}^{\mathcal{N}}
)_{1,1}^{-1}[(\Gamma_{RB}^{\mathcal{N}})_{0,1}]^{\dag},\\
(\Gamma_{RB}^{\mathcal{N}})_{0,0}  &  :=\Pi_{\Gamma^{\mathcal{M}}}\Gamma
_{RB}^{\mathcal{N}}\Pi_{\Gamma^{\mathcal{M}}},\\
(\Gamma_{RB}^{\mathcal{N}})_{0,1}  &  :=\Pi_{\Gamma^{\mathcal{M}}}\Gamma
_{RB}^{\mathcal{N}}\Pi_{\Gamma^{\mathcal{M}}}^{\perp},\\
(\Gamma_{RB}^{\mathcal{N}})_{1,1}  &  :=\Pi_{\Gamma^{\mathcal{M}}}^{\perp
}\Gamma_{RB}^{\mathcal{N}}\Pi_{\Gamma^{\mathcal{M}}}^{\perp},
\end{align}
$\Pi_{\Gamma^{\mathcal{M}}}$ is the projection onto the support of
$\Gamma_{RB}^{\mathcal{M}}$, $\Pi_{\Gamma^{\mathcal{M}}}^{\perp}$ is the
projection onto its kernel, and all inverses are taken on the support. For
$\alpha\in(0,1)$, we have the following alternative form:
\begin{equation}
\widehat{Q}_{\alpha}(\mathcal{N}\Vert\mathcal{M})=\lim_{\varepsilon
\rightarrow0^{+}}\lambda_{\min}\left(  \operatorname{Tr}_{B}[G_{\alpha}
(\Gamma_{RB}^{\mathcal{M}_{\varepsilon}},\Gamma_{RB}^{\mathcal{N}})]\right)  ,
\end{equation}
where $\Gamma_{RB}^{\mathcal{M}_{\varepsilon}}:=\Gamma_{RB}^{\mathcal{M}
}+\varepsilon I_{RB}$.
\end{proposition}

\begin{proof}
See Appendix~\ref{app:geo-renyi-channels-app}.
\end{proof}

\bigskip
It is known from \cite{Fang2019a} that the geometric R\'enyi relative entropy of quantum channels $\mathcal{N}_{A \to B}$ and $\mathcal{M}_{A \to B}$ converges to the Belavkin--Staszewski relative entropy of channels in the limit as $\alpha \to 1$:
\begin{equation}
    \lim_{\alpha \to 1} \widehat{D}_\alpha(\mathcal{N}_{A \to B} \Vert  \mathcal{M}_{A \to B}) = \widehat{D}(\mathcal{N}_{A \to B} \Vert  \mathcal{M}_{A \to B}) ,
\end{equation}
and the Belavkin--Staszewski relative entropy of channels has the following explicit expression:
\begin{equation}
    \widehat{D}(\mathcal{N}_{A \to B} \Vert  \mathcal{M}_{A \to B}) := \left \Vert \operatorname{Tr}_B \left[(\Gamma^{\mathcal{N}}_{RB})^{1/2}
    \log_2 \left(
    (\Gamma^{\mathcal{N}}_{RB})^{1/2}
    (\Gamma^{\mathcal{M}}_{RB})^{-1}
    (\Gamma^{\mathcal{N}}_{RB})^{1/2}
    \right)
    (\Gamma^{\mathcal{N}}_{RB})^{1/2}
    \right] \right \Vert_{\infty}
\end{equation}
if $\operatorname{supp}(\Gamma^{\mathcal{N}}_{RB}) \subseteq \operatorname{supp}(\Gamma^{\mathcal{M}}_{RB})$ and $\widehat{D}(\mathcal{N}_{A \to B} \Vert  \mathcal{M}_{A \to B}) := +\infty$ otherwise.

\begin{proposition}
[Chain rule]\label{prop:chain-rule-geo-renyi-ch}For $\rho_{RA}$ a quantum
state, $\sigma_{RA}$ a positive semi-definite operator, $\mathcal{N}
_{A\rightarrow B}$ a quantum channel, and $\mathcal{M}_{A\rightarrow B}$ a
completely positive map, the following chain rule holds for $\alpha
\in(0,1)\cup(1,2]$:
\begin{equation}
\widehat{D}_{\alpha}(\mathcal{N}_{A\rightarrow B}(\rho_{RA})\Vert
\mathcal{M}_{A\rightarrow B}(\sigma_{RA}))\leq\widehat{D}_{\alpha}
(\mathcal{N}\Vert\mathcal{M})+\widehat{D}_{\alpha}(\rho_{RA}\Vert\sigma_{RA}).
\end{equation}

\end{proposition}

\begin{proof}
See Appendix~\ref{app:geo-renyi-channels-app}.
\end{proof}

\begin{corollary}
\label{cor:amort-collapse-geo-renyi-ch}The geometric R\'{e}nyi relative
entropy does not increase under amortization for all $\alpha\in(0,1)\cup
(1,2]$:
\begin{equation}
\widehat{D}_{\alpha}(\mathcal{N}\Vert\mathcal{M})=\widehat{D}_{\alpha
}^{\mathcal{A}}(\mathcal{N}\Vert\mathcal{M}),
\end{equation}
where the amortized geometric R\'{e}nyi relative entropy is defined from the
general approach given in \cite{Berta2018c}:
\begin{equation}
\widehat{D}_{\alpha}^{\mathcal{A}}(\mathcal{N}\Vert\mathcal{M}):=\sup
_{\rho_{RA},\sigma_{RA}}\left[  \widehat{D}_{\alpha}(\mathcal{N}_{A\rightarrow
B}(\rho_{RA})\Vert\mathcal{M}_{A\rightarrow B}(\sigma_{RA}))-\widehat
{D}_{\alpha}(\rho_{RA}\Vert\sigma_{RA})\right]  .
\end{equation}

\end{corollary}

\begin{proof}
The proof is the same as that given for
Corollaries~\ref{cor:amort-collapse-root-SLD}\ and
\ref{cor:amort-collapse-RLD-fish}.
\end{proof}

\begin{proposition}
\label{prop:subadd-serial-comp-geo-renyi}The geometric R\'{e}nyi relative
entropy is subadditive under serial concatenation of quantum channels for
$\alpha\in(0,1)\cup(1,2]$, in the following sense:
\begin{equation}
\widehat{D}_{\alpha}(\mathcal{N}_{2}\circ\mathcal{N}_{1}\Vert\mathcal{M}
_{2}\circ\mathcal{M}_{1})\leq\widehat{D}_{\alpha}(\mathcal{N}_{2}
\Vert\mathcal{M}_{2})+\widehat{D}_{\alpha}(\mathcal{N}_{1}\Vert\mathcal{M}
_{1}),
\end{equation}
where $\mathcal{N}_{1}$ and $\mathcal{N}_{2}$ are quantum channels and
$\mathcal{M}_{1}$ and $\mathcal{M}_{2}$ are completely positive maps.
\end{proposition}

\begin{proof}
The proof is the same as that given for
Corollaries~\ref{cor:subadd-serial-concat-root-SLD-Fish}\ and
\ref{cor:subadd-serial-concat-RLD-Fish}.
\end{proof}

\medskip Just as the geometric fidelity of quantum states is a special case of
geometric R\'{e}nyi relative entropy, so is the geometric fidelity of quantum
channels:
\begin{equation}
\widehat{F}(\mathcal{N},\mathcal{M}):=\inf_{\psi_{RA}}\widehat{F}
(\mathcal{N}_{A\rightarrow B}(\psi_{RA}),\mathcal{M}_{A\rightarrow B}
(\psi_{RA})),
\end{equation}
where $\mathcal{N}_{A\rightarrow B}$ is a quantum channel and $\mathcal{M}
_{A\rightarrow B}$ is a completely positive map. By employing
Proposition~\ref{prop:explicit-formula-geo-renyi-ch}, we find the following
formula for the geometric fidelity of channels:
\begin{equation}
\widehat{F}(\mathcal{N},\mathcal{M})=\left[  \lim_{\varepsilon\rightarrow
0^{+}}\lambda_{\min}\left(  \operatorname{Tr}_{B}[(\Gamma_{RB}^{\mathcal{M}
_{\varepsilon}})^{1/2}((\Gamma_{RB}^{\mathcal{M}_{\varepsilon}})^{-1/2}
\Gamma_{RB}^{\mathcal{N}}(\Gamma_{RB}^{\mathcal{M}_{\varepsilon}}
)^{-1/2})^{1/2}(\Gamma_{RB}^{\mathcal{M}_{\varepsilon}})^{1/2}\right)
\right]  ^{2}. \label{eq:explicit-formula-geo-fid-channels}
\end{equation}
By exploiting this formula, we arrive at the following semi-definite program
for the geometric fidelity of quantum channels:

\begin{proposition}
The geometric channel fidelity of a quantum channel $\mathcal{N}$ and a
full-rank completely positive map $\mathcal{M}$ can be calculated by means of
the following semi-definite program:
\begin{equation}
\sqrt{\widehat{F}}(\mathcal{N},\mathcal{M})=\sup_{\mu\geq0,X_{RB}\geq0}\mu,
\label{eq:sdp-primal-geo-fid-1}
\end{equation}
subject to
\begin{equation}
\begin{bmatrix}
\Gamma_{RB}^{\mathcal{N}} & X_{RB}\\
X_{RB} & \Gamma_{RB}^{\mathcal{M}}
\end{bmatrix}
\geq0,\qquad\mu I_{R}\leq\operatorname{Tr}_{B}[X_{RB}].
\label{eq:sdp-primal-geo-fid-2}
\end{equation}
The dual program is given by
\begin{equation}
\inf_{\rho_{R}\geq0,Y_{RB},W_{RB},Z_{RB}}\operatorname{Tr}[\Gamma
_{RB}^{\mathcal{N}}Y_{RB}]+\operatorname{Tr}[\Gamma_{RB}^{\mathcal{M}}Z_{RB}]
\end{equation}
subject to
\begin{equation}
\begin{bmatrix}
Y_{RB} & W_{RB}^{\dag}\\
W_{RB} & Z_{RB}
\end{bmatrix}
\geq0,\qquad W_{RB}+W_{RB}^{\dag}\geq\rho_{R}\otimes I_{B},\qquad
\operatorname{Tr}[\rho_{R}]=1.
\end{equation}

\end{proposition}

\begin{proof}
As argued above, the geometric fidelity of quantum channels is given by the
expression in \eqref{eq:explicit-formula-geo-fid-channels}, which involves the
standard operator geometric mean of $\Gamma_{RB}^{\mathcal{M}}$ and
$\Gamma_{RB}^{\mathcal{N}}$:
\begin{equation}
G_{\frac{1}{2}}(\Gamma_{RB}^{\mathcal{M}},\Gamma_{RB}^{\mathcal{N}}) :=(\Gamma_{RB}
^{\mathcal{M}})^{1/2}((\Gamma_{RB}^{\mathcal{M}})^{-1/2}\Gamma_{RB}
^{\mathcal{N}}(\Gamma_{RB}^{\mathcal{M}})^{-1/2})^{1/2}(\Gamma_{RB}
^{\mathcal{M}})^{1/2}
\end{equation}
and the minimum eigenvalue of its partial trace over system $B$. The following
characterization of $G_{\frac{1}{2}}(\Gamma_{RB}^{\mathcal{M}},\Gamma_{RB}^{\mathcal{N}})$ is
well known \cite{Bhat07}
\begin{equation}
G_{\frac{1}{2}}(\Gamma_{RB}^{\mathcal{M}},\Gamma_{RB}^{\mathcal{N}}) = \sup\left\{  X_{RB}\geq0:
\begin{bmatrix}
\Gamma_{RB}^{\mathcal{N}} & X_{RB}\\
X_{RB} & \Gamma_{RB}^{\mathcal{M}}
\end{bmatrix}
\geq0\right\}  , \label{eq:geometric-mean-opt-char}
\end{equation}
where the ordering is with respect to the operator order (L\"{o}wner order).
Additionally, the minimum eigenvalue of a positive semi-definite operator $L$
is given by
\begin{equation}
\lambda_{\min}(L)=\sup\left\{  \mu\geq0:\mu I\leq L\right\}  .
\label{eq:min-eig-opt-char}
\end{equation}
Putting together \eqref{eq:geometric-mean-opt-char} and
\eqref{eq:min-eig-opt-char}, we conclude \eqref{eq:sdp-primal-geo-fid-1}--\eqref{eq:sdp-primal-geo-fid-2}.

To find the dual program, consider that the dual characterization of the
minimum eigenvalue $\lambda_{\min}(L)$ of an operator $L$ is as follows:
\begin{equation}
\lambda_{\min}(L)=\inf_{\rho\geq0,\operatorname{Tr}[\rho]=1}\operatorname{Tr}
[L\rho],
\end{equation}
so that
\begin{equation}
\sqrt{\widehat{F}}(\mathcal{N},\mathcal{M})=\inf_{\rho_{R}\geq
0,\operatorname{Tr}[\rho_{R}]=1}\sup_{X_{RB}\geq0}\operatorname{Tr}[\rho
_{R}\operatorname{Tr}_{B}[X_{RB}]] \label{eq:simplified-opt-geo-fid-ch}
\end{equation}
subject to
\begin{equation}
\begin{bmatrix}
\Gamma_{RB}^{\mathcal{N}} & X_{RB}\\
X_{RB} & \Gamma_{RB}^{\mathcal{M}}
\end{bmatrix}
\geq0.
\end{equation}
For fixed $\rho_{R}$, we can then consider finding the dual of the following program:
\begin{equation}
\sup_{X_{RB}\geq0}\operatorname{Tr}[\rho_{R}\operatorname{Tr}_B[X_{RB}]]
\label{eq:intermediary-SDP}
\end{equation}
subject to
\begin{equation}
\begin{bmatrix}
\Gamma_{RB}^{\mathcal{N}} & X_{RB}\\
X_{RB} & \Gamma_{RB}^{\mathcal{M}}
\end{bmatrix}
\geq0.
\end{equation}
Considering that
\begin{align}
\begin{bmatrix}
\Gamma_{RB}^{\mathcal{N}} & X_{RB}\\
X_{RB} & \Gamma_{RB}^{\mathcal{M}}
\end{bmatrix}
\geq0\qquad &  \Longleftrightarrow\qquad
\begin{bmatrix}
\Gamma_{RB}^{\mathcal{N}} & -X_{RB}\\
-X_{RB} & \Gamma_{RB}^{\mathcal{M}}
\end{bmatrix}
\geq0\\
&  \Longleftrightarrow\qquad
\begin{bmatrix}
\Gamma_{RB}^{\mathcal{N}} & 0\\
0 & \Gamma_{RB}^{\mathcal{M}}
\end{bmatrix}
\geq
\begin{bmatrix}
0 & X_{RB}\\
X_{RB} & 0
\end{bmatrix}
,
\end{align}
the standard form of the SDP\ is
\begin{equation}
\sup_{X\geq0}\left\{  \operatorname{Tr}[AX]:\Phi(X)\leq B\right\}  ,
\end{equation}
with
\begin{equation}
A=\rho_{R}\otimes I_{B},\quad\Phi(X_{RB})=
\begin{bmatrix}
0 & X_{RB}\\
X_{RB} & 0
\end{bmatrix}
,\quad B=
\begin{bmatrix}
\Gamma_{RB}^{\mathcal{N}} & 0\\
0 & \Gamma_{RB}^{\mathcal{M}}
\end{bmatrix}
.
\end{equation}
Then the dual map $\Phi^{\dag}$ is given by
\begin{equation}
\operatorname{Tr}[Y\Phi(X)]=\operatorname{Tr}[\Phi^{\dag}(Y)X],
\end{equation}
so that
\begin{align}
\operatorname{Tr}\left[
\begin{bmatrix}
Y_{RB} & W_{RB}^{\dag}\\
W_{RB} & Z_{RB}
\end{bmatrix}
\Phi(X_{RB})\right]   &  =\operatorname{Tr}\left[
\begin{bmatrix}
Y_{RB} & W_{RB}^{\dag}\\
W_{RB} & Z_{RB}
\end{bmatrix}
\begin{bmatrix}
0 & X_{RB}\\
X_{RB} & 0
\end{bmatrix}
\right] \\
&  =\operatorname{Tr}\left[  \left(  W_{RB}+W_{RB}^{\dag}\right)
X_{RB}\right]  ,
\end{align}
and we thus identify
\begin{equation}
\Phi^{\dag}(Y)=W_{RB}+W_{RB}^{\dag}.
\end{equation}
Then plugging in to the standard form of the dual program
\begin{equation}
\inf_{Y\geq0}\left\{  \operatorname{Tr}[BY]:\Phi^{\dag}(Y)\geq A\right\}  ,
\end{equation}
we find that it is given by
\begin{equation}
\inf\operatorname{Tr}[\Gamma_{RB}^{\mathcal{N}}Y_{RB}]+\operatorname{Tr}
[\Gamma_{RB}^{\mathcal{M}}Z_{RB}] \label{eq:intermediary-SDP-dual}
\end{equation}
subject to
\begin{equation}
\begin{bmatrix}
Y_{RB} & W_{RB}^{\dag}\\
W_{RB} & Z_{RB}
\end{bmatrix}
\geq0,\qquad W_{RB}+W_{RB}^{\dag}\geq\rho_{R}\otimes I_{B}.
\end{equation}
So applying strong duality to assert equality of \eqref{eq:intermediary-SDP}
and \eqref{eq:intermediary-SDP-dual}\ and combining this with
\eqref{eq:simplified-opt-geo-fid-ch}, the geometric fidelity of quantum
channels $\mathcal{N}$ and $\mathcal{M}$ can be computed as
\begin{equation}
\inf_{\rho_{R}\geq0}\operatorname{Tr}[\Gamma_{RB}^{\mathcal{N}}Y_{RB}
]+\operatorname{Tr}[\Gamma_{RB}^{\mathcal{M}}Z_{RB}]
\end{equation}
subject to
\begin{equation}
\begin{bmatrix}
Y_{RB} & W_{RB}^{\dag}\\
W_{RB} & Z_{RB}
\end{bmatrix}
\geq0,\qquad W_{RB}+W_{RB}^{\dag}\geq\rho_{R}\otimes I_{B},\qquad
\operatorname{Tr}[\rho_{R}]=1.
\end{equation}
Strong duality holds because we can choose $\rho_{R}=\pi_{R}$ (maximally mixed
state) $W_{RB}=I_{RB}$ and $Y_{RB}=Z_{RB}=2I_{RB}$ so that all constraints in
the dual program are strict. This concludes the proof.
\end{proof}

\subsection{Geometric fidelity of quantum channels as a limit on symmetric
channel discrimination}

One main use of the geometric fidelity of quantum channels is as a limit on
the error exponent of symmetric channel discrimination:

\begin{conclusion}
As a direct consequence of Eq.~(159)\ of \cite{Berta2018c}, the inequality in
\eqref{eq:sandwiched-to-geo-bound-main-text}, the meta-converse from
\cite[Lemma~14]{Berta2018c}, and the amortization collapse in
Corollary~\ref{cor:amort-collapse-geo-renyi-ch}, the following bound holds for
the non-asymptotic Chernoff error exponent $\xi_{n}(p,\mathcal{N}
,\mathcal{M})$ of symmetric channel discrimination of quantum channels
$\mathcal{N}_{A\rightarrow B}$ and $\mathcal{M}_{A\rightarrow B}$:
\begin{equation}
\xi_{n}(p,\mathcal{N},\mathcal{M})\leq\widehat{D}_{1/2}(\mathcal{N}
\Vert\mathcal{M})-\frac{1}{n}\ln[p(1-p)],
\end{equation}
where $\xi_{n}(p,\mathcal{N},\mathcal{M})$ is defined in
\eqref{eq:non-asymptotic-Chernoff-exp}. Thus, we conclude the following bound
on the asymptotic exponent:
\begin{equation}
\overline{\xi}(\mathcal{N},\mathcal{M})\leq\widehat{D}_{1/2}(\mathcal{N}
\Vert\mathcal{M}).
\end{equation}

\end{conclusion}

This result is a significant improvement over the bound from
\cite[Proposition~21]{Berta2018c} because $\widehat{D}_{1/2}(\mathcal{N}
\Vert\mathcal{M})\leq\min\{D_{\max}(\mathcal{N}\Vert\mathcal{M}),D_{\max
}(\mathcal{M}\Vert\mathcal{N})\}$, due to
\eqref{eq:geo-renyi-to-dmax-main-text}\ and
\eqref{eq:mono-geo-renyi-main-text}. It is also efficiently computable, so
that it improves as well upon the amortized fidelity bound from
\cite[Proposition~21]{Berta2018c} and \cite{CE18}.

An achievable rate for symmetric channel discrimination is given by the
Chernoff information of quantum channels \cite{Berta2018c}, defined as
\begin{equation}
C(\mathcal{N}\Vert\mathcal{M}):=\sup_{\psi_{RA},\alpha\in\left(  0,1\right)
}\left(  1-\alpha\right)  D_{\alpha}(\mathcal{N}_{A\rightarrow B}(\psi
_{RA})\Vert\mathcal{M}_{A\rightarrow B}(\psi_{RA})),
\end{equation}
where the Petz--R\'{e}nyi relative entropy $D_{\alpha}(\rho\Vert\sigma)$ of a
quantum state and a positive semi-definite operator $\sigma$ is defined for
$\alpha\in(0,1)\cup(1,\infty)$ as \cite{P85,P86}
\begin{equation}
D_{\alpha}(\rho\Vert\sigma)=\left\{
\begin{array}
[c]{cc}
\frac{1}{\alpha-1}\ln\operatorname{Tr}[\rho^{\alpha}\sigma^{1-\alpha}] &
\begin{array}
[c]{c}
\text{if }\alpha\in(0,1)\text{ or}\\
\alpha\in(1,\infty)\text{ and }\operatorname{supp}(\rho)\subseteq
\operatorname{supp}(\sigma)
\end{array}
\\
+\infty & \text{otherwise}
\end{array}
\right.  .
\end{equation}
This corresponds to a parallel discrimination strategy in which we feed in one
share of a state $\psi_{RA}$ to each use of the channel and then perform a
collective measurement on all of the output systems (this is even a special
case of what is depicted in Figure~\ref{fig:parallel-scheme}). That is, we
have that
\begin{equation}
C(\mathcal{N}\Vert\mathcal{M})\leq\underline{\xi}(\mathcal{N},\mathcal{M}
)\leq\overline{\xi}(\mathcal{N},\mathcal{M})\leq\widehat{D}_{1/2}
(\mathcal{N}\Vert\mathcal{M}).
\end{equation}

In Figures~\ref{fig:ch-disc-loss} and \ref{fig:ch-disc-noise}, we compare the
achievable lower bound given by $C(\mathcal{N}\Vert\mathcal{M})$ with the
general upper bound set by $\widehat{D}_{1/2}(\mathcal{N}\Vert\mathcal{M})$
for the case of the generalized amplitude damping channel defined in
\eqref{eq:GADC}, for various values of the loss and noise parameters.

\begin{figure}[ptb]
\begin{subfigure}{.5\textwidth}
\centering
\includegraphics[width=1.0\linewidth]{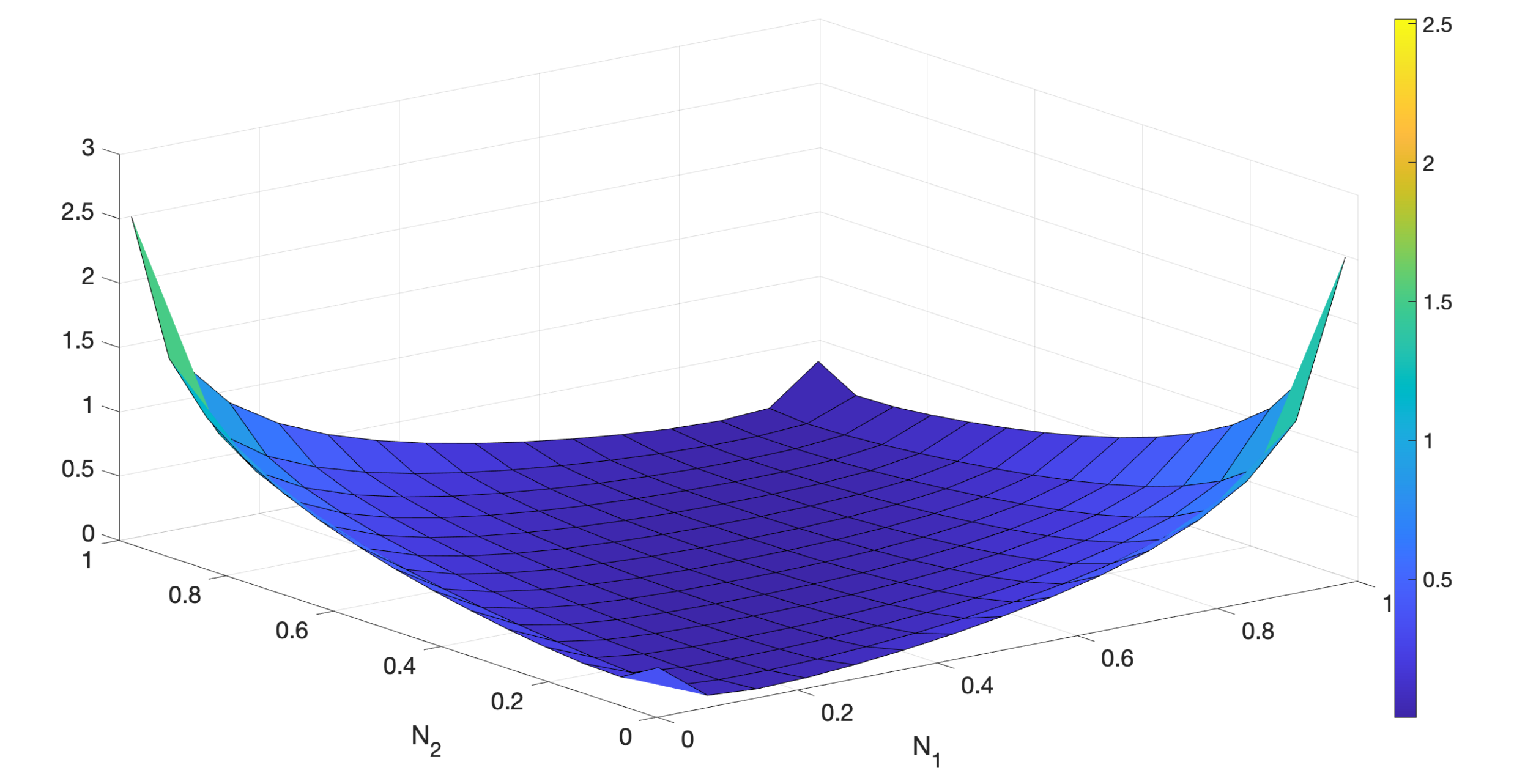}
\caption{}
\label{fig:ch-disc-loss-fixed-0.8-0.7}
\end{subfigure}
\begin{subfigure}{.5\textwidth}
\centering
\includegraphics[width=1.0\linewidth]{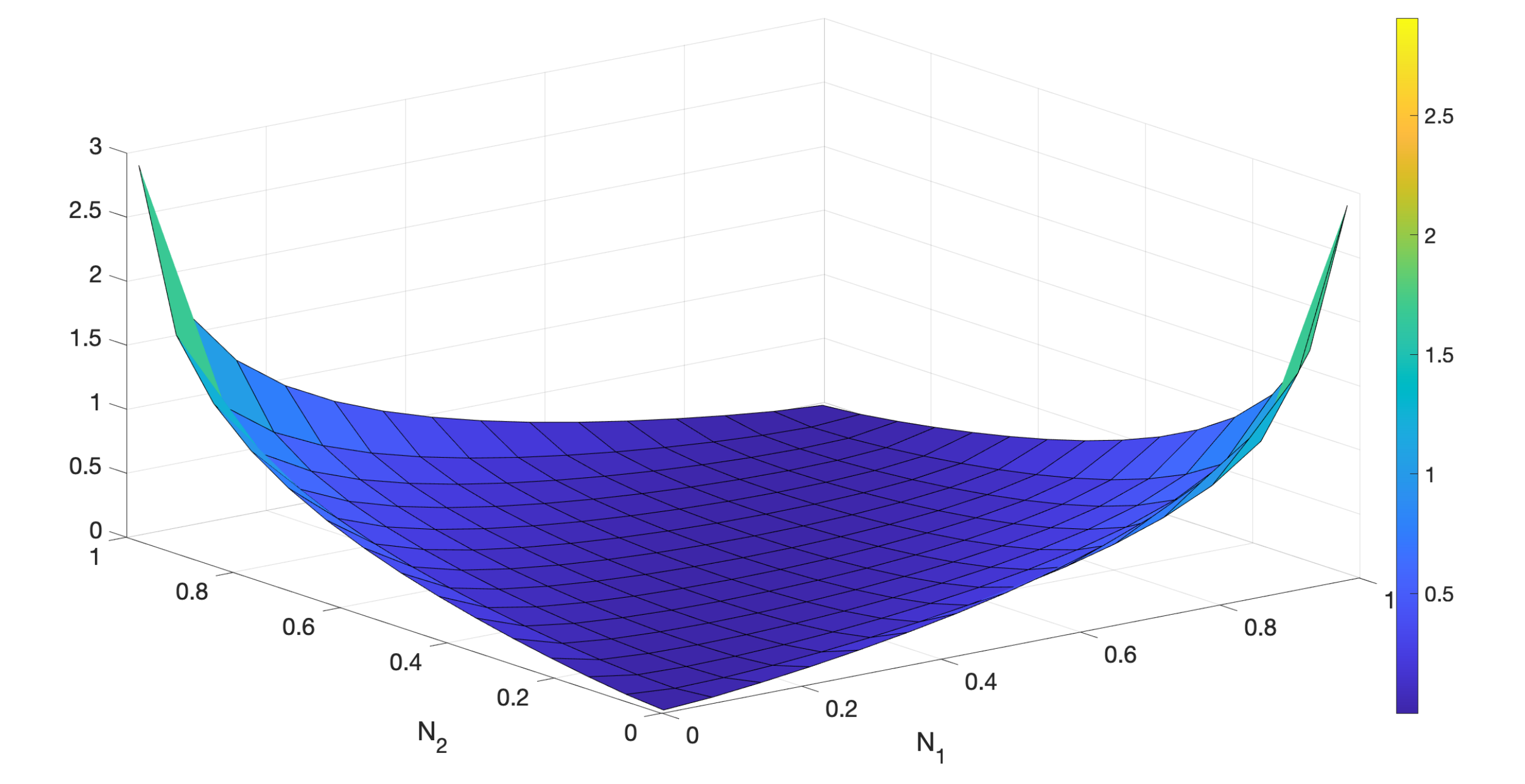}
\caption{}
\label{fig:ch-disc-loss-fixed-0.5-0.5}
\end{subfigure}
\caption{(a)~Difference of the geometric fidelity upper bound and Chernoff
information lower bound for generalized amplitude damping channels with fixed
loss $\gamma_{1} = 0.8$ and $\gamma_{2} = 0.7$. (b)~Same plot but fixed loss
$\gamma_{1} = 0.5$ and $\gamma_{2} = 0.5$.}
\label{fig:ch-disc-loss}
\end{figure}

\begin{figure}[ptb]
\begin{subfigure}{.5\textwidth}
\centering
\includegraphics[width=1.0\linewidth]{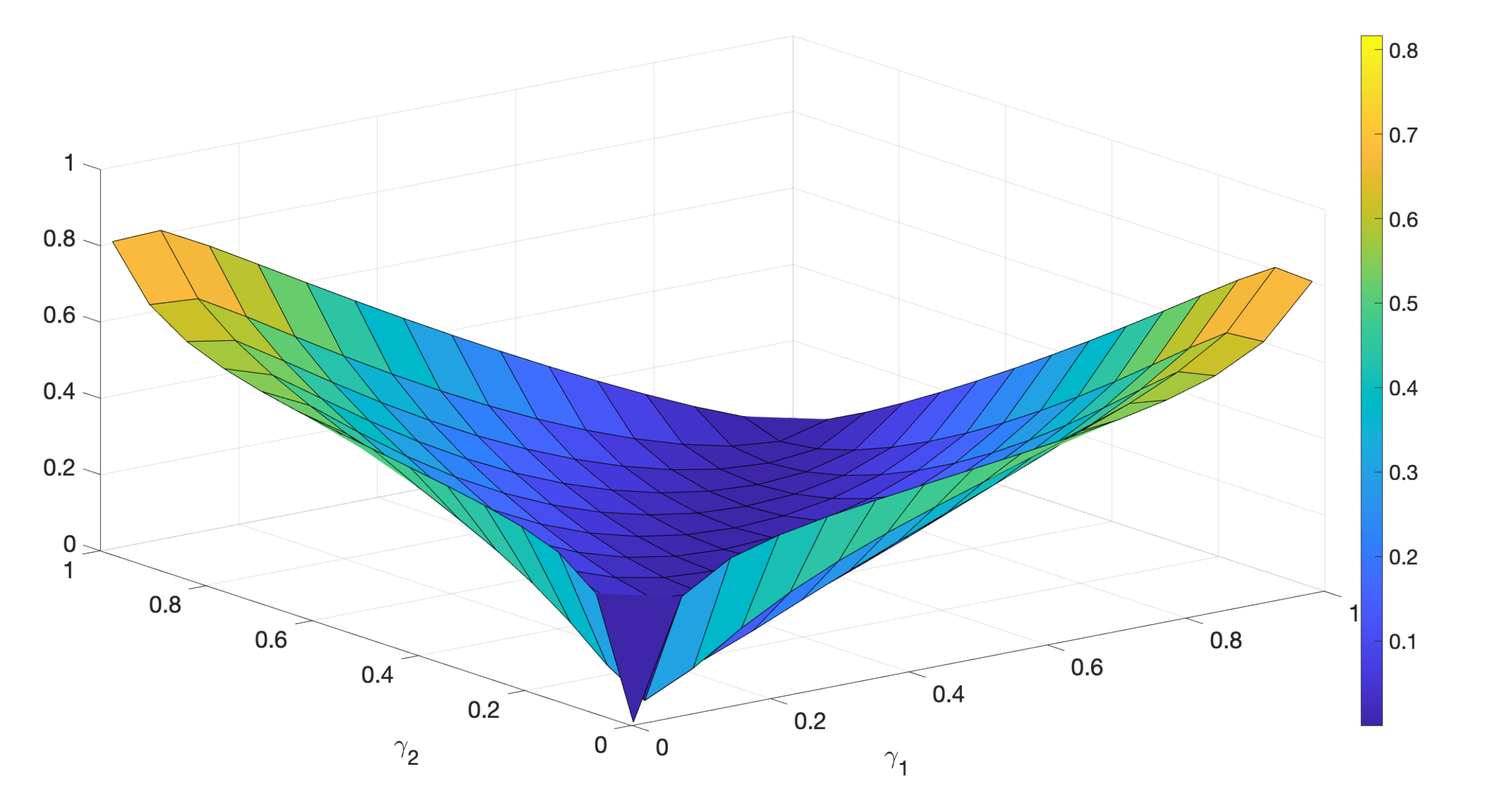}
\caption{}
\label{fig:ch-disc-loss-N-0.2}
\end{subfigure}
\begin{subfigure}{.5\textwidth}
\centering
\includegraphics[width=1.0\linewidth]{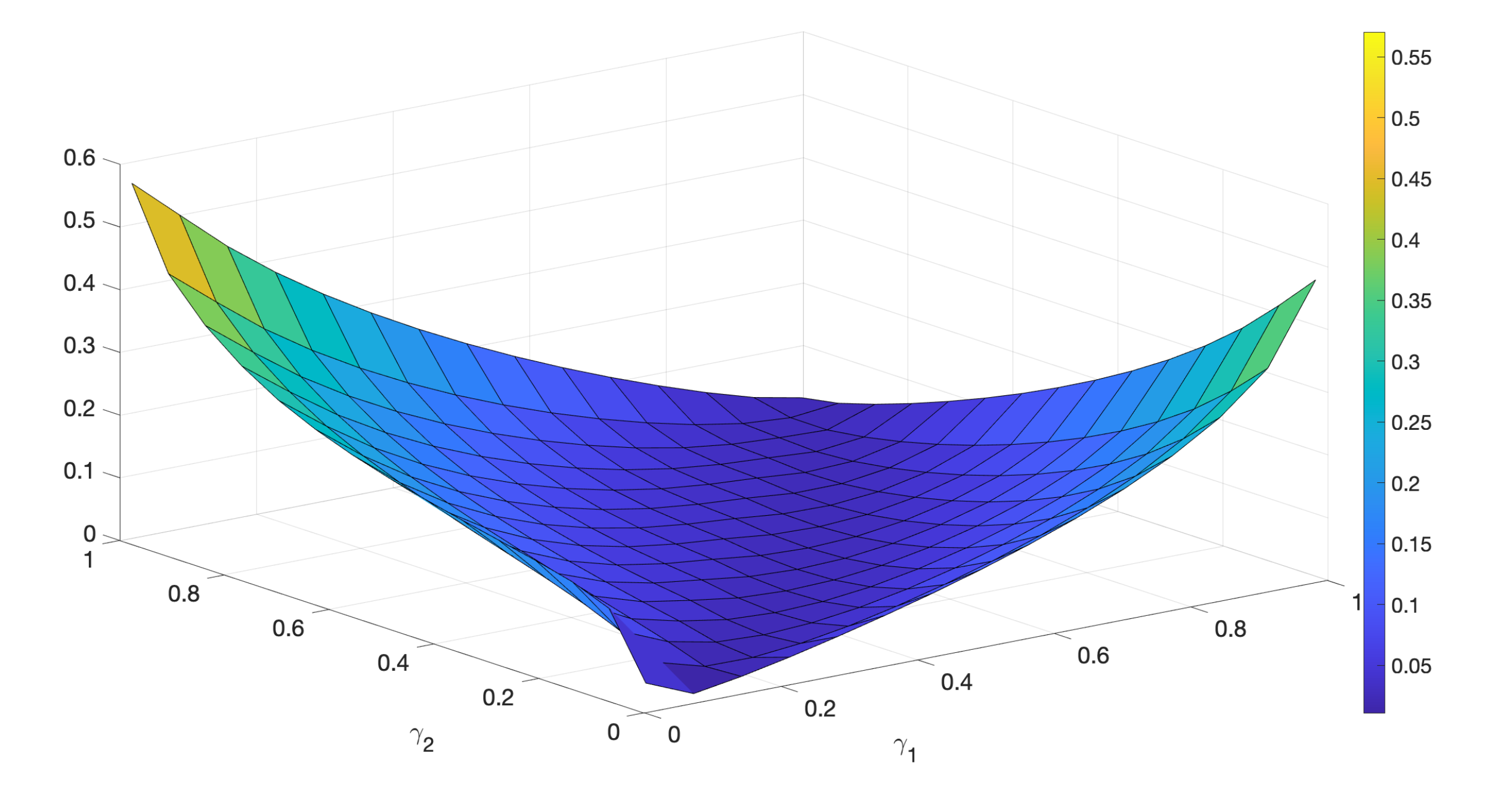}
\caption{}
\label{fig:ch-disc-loss-N-0.3}
\end{subfigure}
\caption{a)~Difference of the geometric fidelity upper bound and Chernoff
information lower bound for generalized amplitude damping channels with fixed
noise $N_{1} = 0.2$ and $N_{2} = 0.2$. (b)~Same plot but fixed noise $N_{1} =
0.3$ and $N_{2} = 0.5$.}
\label{fig:ch-disc-noise}
\end{figure}

\subsection{Belavkin--Staszewski divergence sphere as a limit on error
exponent of quantum channel discrimination}

In this section, we establish a limit on the asymptotic Hoeffding error
exponent for quantum channel discrimination. We find a generic upper bound for
arbitrary quantum channels in terms of what we call the Belavkin--Staszewski
divergence sphere formula.

\begin{proposition}
For quantum channels $\mathcal{N}$ and $\mathcal{M}$, the Belavkin--Staszewski
divergence sphere is an upper bound on their asymptotic Hoeffding error
exponent for quantum channel discrimination:
\begin{equation}
\overline{B}(r,\mathcal{N},\mathcal{M})\leq\inf_{\mathcal{T}:\widehat
{D}(\mathcal{T}\Vert\mathcal{M})\leq r}\widehat{D}(\mathcal{T}\Vert
\mathcal{N}). \label{eq:hoeffding-bound-BS-rel-ent-opt}
\end{equation}

\end{proposition}

\begin{proof}
The argument is the same as that given in \cite[Exercise~3.15]{H06},
\cite[Eq.~(16)]{Hayashi09}, and \cite[Proposition~30]{Berta2018c}, but here we
use the fact that the Belavkin--Staszewski relative entropy of quantum
channels is a strong converse upper bound for asymmetric quantum channel
discrimination \cite[Theorem~49]{Fang2019a}. Fix $\delta>0$. Let $\mathcal{T}$
be a quantum channel such that $\widehat{D}(\mathcal{T}\Vert\mathcal{M})\leq
r-\delta$. By construction, it follows that $r>\widehat{D}(\mathcal{T}
\Vert\mathcal{M})$. Let $(\{\mathcal{S}^{(n)},\Lambda^{\hat{\theta}}\})_{n}$
denote a sequence of channel discrimination strategies for $\mathcal{T}$ and
$\mathcal{M}$, and let us denote the associated Type~I and II error
probabilities by
\begin{equation}
\alpha_{n}^{\mathcal{T}\Vert\mathcal{M}}(\{\mathcal{S}^{(n)},\Lambda
^{\hat{\theta}}\}),\qquad\beta_{n}^{\mathcal{T}\Vert\mathcal{M}}
(\{\mathcal{S}^{(n)},\Lambda^{\hat{\theta}}\}),
\end{equation}
respectively. By applying \cite[Theorem~49]{Fang2019a}, that the
Belavkin--Staszewski relative entropy is a strong converse upper bound for
asymmetric channel discrimination of $\mathcal{T}$ and $\mathcal{M}$, if
$(\{\mathcal{S}^{(n)},\Lambda^{\hat{\theta}}\})_{n}$ is a sequence of channel
discrimination strategies for these channels such that
\begin{equation}
\limsup_{n\rightarrow\infty}-\frac{1}{n}\ln\beta_{n}^{\mathcal{T}
\Vert\mathcal{M}}(\{\mathcal{S}^{(n)},\Lambda^{\hat{\theta}}\})\geq r,
\end{equation}
then necessarily, we have that
\begin{equation}
\limsup_{n\rightarrow\infty}\alpha_{n}^{\mathcal{T}\Vert\mathcal{M}
}(\{\mathcal{S}^{(n)},\Lambda^{\hat{\theta}}\})=1.
\label{eq:bs-hoeffding-err-to-1}
\end{equation}
However, this implies that $\{\mathcal{S}^{(n)},I-\Lambda^{\hat{\theta}}\}$
can be used as a channel discrimination strategy for the channels
$\mathcal{T}$ and $\mathcal{M}$, and let us denote the associated Type~I and
II error probabilities by
\begin{equation}
\alpha_{n}^{\mathcal{T}\Vert\mathcal{N}}(\{\mathcal{S}^{(n)},I-\Lambda
^{\hat{\theta}}\}),\qquad\beta_{n}^{\mathcal{T}\Vert\mathcal{N}}
(\{\mathcal{S}^{(n)},I-\Lambda^{\hat{\theta}}\}).
\end{equation}
By applying \eqref{eq:bs-hoeffding-err-to-1}, we conclude that
\begin{equation}
\limsup_{n\rightarrow\infty}\alpha_{n}^{\mathcal{T}\Vert\mathcal{N}
}(\{\mathcal{S}^{(n)},\Lambda^{\hat{\theta}}\})=0,
\end{equation}
and by again invoking the strong converse from \cite[Theorem~49]{Fang2019a},
it is necessary that
\begin{equation}
\limsup_{n\rightarrow\infty}-\frac{1}{n}\ln\beta_{n}^{\mathcal{T}
\Vert\mathcal{N}}(\{\mathcal{S}^{(n)},\Lambda^{\hat{\theta}}\})\leq\widehat
{D}(\mathcal{T}\Vert\mathcal{N}).
\end{equation}
Thus, we find the following bound holding for an arbitrary quantum channel
$\mathcal{T}$ for which $r>\widehat{D}(\mathcal{T}\Vert\mathcal{M})$:
\begin{equation}
\overline{B}(r,\mathcal{N},\mathcal{M})\leq\inf_{\mathcal{T}:\widehat
{D}(\mathcal{T}\Vert\mathcal{M})\leq r-\delta}\widehat{D}(\mathcal{T}
\Vert\mathcal{N}).
\end{equation}
Since $\delta>0$ is arbitrary in the above argument, we can employ the facts
that the Belavkin--Staszewski relative entropy is continuous in its first
argument to arrive at the bound stated in \eqref{eq:hoeffding-bound-BS-rel-ent-opt}.
\end{proof}

\subsection{Bounds for sequential channel discrimination with repetition}

In this section, we establish upper bounds on the asymptotic error exponents
for sequential channel discrimination with repetition, as defined in
Section~\ref{sec:seq-ch-w-rep}. The main idea is to exploit the amortization
collapse for the geometric R\'{e}nyi channel divergence from
Corollary~\ref{cor:amort-collapse-geo-renyi-ch}, the meta-converse from
\cite[Lemma~14]{Berta2018c}, and the finite-sample bounds from \cite{AMV12}.

\begin{proposition}
For quantum channels $\mathcal{N}$ and $\mathcal{M}$, the following asymptotic
Chernoff exponent for sequential channel discrimination with repetition is
bounded for all $p\in(0,1)$ as follows:
\begin{equation}
\limsup_{n\rightarrow\infty}\limsup_{m\rightarrow\infty}\xi_{n,m}
(p,\mathcal{N},\mathcal{M})\leq\widehat{C}(\mathcal{N}\Vert\mathcal{M}),
\label{eq:seq-ch-disc-w-rep-chernoff-up-bnd}
\end{equation}
where the upper bound in \eqref{eq:seq-ch-disc-w-rep-chernoff-up-bnd}\ holds
only for the particular order of limits of $n$ and $m$ given and the geometric
Chernoff information of quantum channels is defined as follows:
\begin{equation}
\widehat{C}(\mathcal{N}\Vert\mathcal{M}):=\sup_{\alpha\in\left(  0,1\right)
}\left(  1-\alpha\right)  \widehat{D}_{\alpha}(\mathcal{N}\Vert\mathcal{M}).
\label{eq:seq-ch-disc-w-rep-hoeffding-up-bnd}
\end{equation}
The following asymptotic Hoeffding exponent for sequential channel
discrimination with repetition is bounded as follows:
\begin{equation}
\limsup_{n\rightarrow\infty}\limsup_{m\rightarrow\infty}B_{n,m}(r,\mathcal{N}
,\mathcal{M})\leq\sup_{\alpha\in(0,1)}\frac{\alpha-1}{\alpha}\left(
r-\widehat{D}_{\alpha}(\mathcal{N}\Vert\mathcal{M})\right)  .
\end{equation}

\end{proposition}

\begin{proof}
The method for establishing both bounds is the same. By applying
\cite[Theorem~4.7]{AMV12}, the following upper bound holds for the error
exponent in the Chernoff setting:
\begin{multline}
-\frac{1}{nm}\ln p_{e}^{(n,m)}(\mathcal{S}^{(n)},\Lambda^{\hat{\theta},m}
)\leq\frac{1}{m}C(\omega_{R_{n}B_{n}}^{\theta=1}\Vert\omega_{R_{n}B_{n}
}^{\theta=2})\\
+\frac{3\left(  d^{2}-1\right)  }{2}\frac{\ln n}{nm}+\frac{c}{nm}+\frac
{1}{nm\left(  12n+1\right)  },
\end{multline}
where $C(\omega_{R_{n}B_{n}}^{\theta=1}\Vert\omega_{R_{n}B_{n}}^{\theta=2})$
is the Chernoff information of the final states of the discrimination
protocol, $d$ is the dimension of this output state, and $c$ is a constant
that depends on the final output states. Now applying the meta-converse from
\cite[Lemma~14]{Berta2018c}, we find that
\begin{align}
C(\omega_{R_{n}B_{n}}^{\theta=1}\Vert\omega_{R_{n}B_{n}}^{\theta=2})  &
\leq\widehat{C}(\omega_{R_{n}B_{n}}^{\theta=1}\Vert\omega_{R_{n}B_{n}}
^{\theta=2})\\
&  =\sup_{\alpha\in\left(  0,1\right)  }\left(  1-\alpha\right)  \widehat
{D}_{\alpha}(\omega_{R_{n}B_{n}}^{\theta=1}\Vert\omega_{R_{n}B_{n}}^{\theta
=2})\\
&  \leq m\sup_{\alpha\in\left(  0,1\right)  }\left(  1-\alpha\right)
\widehat{D}_{\alpha}^{\mathcal{A}}(\mathcal{N}\Vert\mathcal{M})\\
&  =m\sup_{\alpha\in\left(  0,1\right)  }\left(  1-\alpha\right)  \widehat
{D}_{\alpha}(\mathcal{N}\Vert\mathcal{M})\\
&  =m\widehat{C}(\mathcal{N}\Vert\mathcal{M}).
\end{align}
The second-to-last equality follows from
Corollary~\ref{cor:amort-collapse-geo-renyi-ch}. Combining with the above, we
find the following bound
\begin{equation}
-\frac{1}{nm}\ln p_{e}^{(n,m)}(\mathcal{S}^{(n)},\Lambda^{\hat{\theta},m}
)\leq\widehat{C}(\mathcal{N}\Vert\mathcal{M})+\frac{3\left(  d^{2}-1\right)
}{2}\frac{\ln n}{nm}+\frac{c}{nm}+\frac{1}{nm\left(  12n+1\right)  }.
\end{equation}
By taking the limit as $m\rightarrow\infty$, we get the following uniform
bound:
\begin{equation}
\limsup_{m\rightarrow\infty}\left[  -\frac{1}{nm}\ln p_{e}^{(n,m)}
(\mathcal{S}^{(n)},\Lambda^{\hat{\theta},m})\right]  \leq\widehat
{C}(\mathcal{N}\Vert\mathcal{M}).
\end{equation}
Then taking the limit as $n\rightarrow\infty$, we arrive at
\eqref{eq:seq-ch-disc-w-rep-chernoff-up-bnd}. The proof of
\eqref{eq:seq-ch-disc-w-rep-hoeffding-up-bnd}\ is essentially the same, except
that we start from the other bound in \cite[Theorem~4.7]{AMV12} (having to do
with the Hoeffding exponent).
\end{proof}

\section{Connections between estimation and discrimination of quantum
channels}

\label{sec:connections}In this section, we outline connections between channel
estimation and discrimination, which indicate how one could derive many of the
results in Sections~\ref{sec:SLD-fisher-info-limits} and
\ref{sec:RLD-Fish-limits} based on properties of the quantum fidelity and
geometric R\'{e}nyi relative entropy. To do so, one however needs the stronger
assumption that the family of states or channels is second-order
differentiable with respect to the parameter $\theta$. This is the main reason
that we have avoided this approach in our earlier developments, because we
have shown that it is possible to develop them under the assumption of
first-order differentiability only. Nevertheless, the connections are
interesting and so we go through them here.

\subsection{Limit formulas for SLD\ and RLD\ Fisher informations}

The starting point is the following limit formula for the SLD\ Fisher information:

\begin{proposition}
\label{prop:SLD-fish-limit-fidelity}Let $\{\rho_{\theta}\}_{\theta}$ be a
second-order differentiable family of quantum states. Then the following holds
\begin{align}
I_{F}(\theta;\{\rho_{\theta}\}_{\theta})  &  =\lim_{\varepsilon\rightarrow
0}\lim_{\delta\rightarrow0}\frac{8}{\delta^{2}}\left(  1-\sqrt{F}(\rho
_{\theta}^{\varepsilon},\rho_{\theta+\delta}^{\varepsilon})\right)
,\label{eq:SLD-fish-bures-connection}\\
&  =\lim_{\varepsilon\rightarrow0}\lim_{\delta\rightarrow0}\frac{4}{\delta
^{2}}\left(  -\ln F(\rho_{\theta}^{\varepsilon},\rho_{\theta+\delta
}^{\varepsilon})\right)  , \label{eq:SLD-fish-bures-connection-2}
\end{align}
where
\begin{equation}
\rho_{\theta}^{\varepsilon}:=\left(  1-\varepsilon\right)  \rho_{\theta
}+\varepsilon\pi_{d},
\end{equation}
with $\pi_{d}$ the maximally mixed state.
\end{proposition}

The first expression without the $\varepsilon\rightarrow0$ limit was given in
\cite{Hub92}, where it was assumed that the family $\{\rho_{\theta}\}_{\theta
}$ is full rank. A different proof was then given in \cite{Sommers_2003}, in
which the full rank assumption is made as well. We can then apply these former
results and Proposition~\ref{prop:physical-consistency-SLD-Fish-states}\ to
arrive at the limiting expression in \eqref{eq:SLD-fish-bures-connection}. The
limit in \eqref{eq:SLD-fish-bures-connection-2} is also well known (see, e.g.,
\cite[Section~6]{Hayashi_2002}\ and \cite{TW16}), and we recall a proof of
this due to \cite{Mos17}\ in Appendix~\ref{app:fish-info-fid-limit}.

The exchange of limits in \eqref{eq:SLD-fish-bures-connection}\ has implicitly
been the subject of more recent investigations \cite{Liu_2014,Saf2017,Seveso_2019,Zhou2019a}, starting with \cite{Liu_2014}
and concluding with \cite{Zhou2019a}. The main claim of \cite{Zhou2019a} is
that the limit exchange is possible for any second-order differentiable family
if one modifies \eqref{eq:SLD-fish-bures-connection} from a forward shift to a
central shift:
\begin{align}
I_{F}(\theta;\{\rho_{\theta}\}_{\theta})  &  =\lim_{\delta\rightarrow0}
\lim_{\varepsilon\rightarrow0}\frac{8}{\delta^{2}}\left(  1-\sqrt{F}
(\rho_{\theta-\delta/2}^{\varepsilon},\rho_{\theta+\delta/2}^{\varepsilon
})\right) \\
&  =\lim_{\delta\rightarrow0}\frac{8}{\delta^{2}}\left(  1-\sqrt{F}
(\rho_{\theta-\delta/2},\rho_{\theta+\delta/2})\right)  .
\end{align}
Implicitly the finiteness condition in
\eqref{eq:finiteness-condition-SLD-Fish-states} has been assumed in the
derivation of \cite{Zhou2019a}.

The RLD\ Fisher information has been connected to the geometric R\'{e}nyi
relative entropy via a limit formula of the form in
\eqref{eq:SLD-fish-bures-connection} (see \cite[Section~11]{Mat10fid} and
\cite[Section~6.4]{Mat13,Matsumoto2018}). In this case, we have the following:

\begin{proposition}
\label{prop:RLD-fish-limit-geo-renyi}Let $\{\rho_{\theta}\}_{\theta}$ be a
second-order differentiable family of quantum states. Then the following
equalities hold for all $\alpha\in(0,1)\cup(1,\infty)$:
\begin{align}
\widehat{I}_{F}(\theta;\{\rho_{\theta}\}_{\theta})  &  =\lim_{\varepsilon
\rightarrow0}\lim_{\delta\rightarrow0}\frac{2}{\alpha\left(  \alpha-1\right)
\delta^{2}}\left(  \widehat{Q}_{\alpha}(\rho_{\theta+\delta}^{\varepsilon
}\Vert\rho_{\theta}^{\varepsilon})-1\right)
,\label{eq:geo-renyi-limit-to-RLD}\\
&  =\lim_{\varepsilon\rightarrow0}\lim_{\delta\rightarrow0}\frac{2}{\delta
^{2}\alpha}\widehat{D}_{\alpha}(\rho_{\theta+\delta}^{\varepsilon}\Vert
\rho_{\theta}^{\varepsilon}), \label{eq:geo-renyi-limit-to-RLD-2}
\end{align}
where
\begin{equation}
\rho_{\theta}^{\varepsilon}:=\left(  1-\varepsilon\right)  \rho_{\theta
}+\varepsilon\pi_{d},
\end{equation}
with $\pi_{d}$ the maximally mixed state. Additionally, we have that
\begin{equation}
\widehat{I}_{F}(\theta;\{\rho_{\theta}\}_{\theta})=\lim_{\varepsilon
\rightarrow0}\lim_{\delta\rightarrow0}\frac{2}{\delta^{2}}\widehat{D}
(\rho_{\theta+\delta}^{\varepsilon}\Vert\rho_{\theta}^{\varepsilon}).
\label{eq:RLD-from-BS-rel-ent}
\end{equation}

\end{proposition}

\begin{proof}
Due to the particular order of limits given above, we can assume that
$\rho_{\theta}$ is full rank. Let us define
\begin{equation}
d\rho_{\theta}:=\rho_{\theta+\delta}-\rho_{\theta},
\end{equation}
and observe that
\begin{equation}
\operatorname{Tr}[d\rho_{\theta}]=0.
\end{equation}
Then by plugging into \eqref{eq:geo-quasi-ent-main-txt}, we find that
\begin{align}
\widehat{Q}_{\alpha}(\rho_{\theta+\delta}\Vert\rho_{\theta})  &
=\operatorname{Tr}\!\left[  \rho_{\theta}\left(  \rho_{\theta}^{-1/2}
\rho_{\theta+\delta}\rho_{\theta}^{-1/2}\right)  ^{\alpha}\right] \\
&  =\operatorname{Tr}\!\left[  \rho_{\theta}\left(  \rho_{\theta}^{-1/2}
(\rho_{\theta}+d\rho_{\theta})\rho_{\theta}^{-1/2}\right)  ^{\alpha}\right] \\
&  =\operatorname{Tr}\!\left[  \rho_{\theta}\left(  I+\rho_{\theta}^{-1/2}
d\rho_{\theta}\rho_{\theta}^{-1/2}\right)  ^{\alpha}\right]  .
\label{eq:geometric-fidelity-fisher-explicit-states}
\end{align}
Now, by using the expansion
\begin{equation}
\left(  1+x\right)  ^{\alpha}=1+\alpha x+\frac{1}{2}\left(  \alpha-1\right)
\alpha x^{2}+O(x^{3}),
\end{equation}
we evaluate the innermost expression of
\eqref{eq:geometric-fidelity-fisher-explicit-states}:
\begin{multline}
\left(  I+\rho_{\theta}^{-1/2}d\rho_{\theta}\rho_{\theta}^{-1/2}\right)
^{\alpha}=I+\alpha\rho_{\theta}^{-1/2}d\rho_{\theta}\rho_{\theta}^{-1/2}\\
+\frac{1}{2}\left(  \alpha-1\right)  \alpha\left(  \rho_{\theta}^{-1/2}
d\rho_{\theta}\rho_{\theta}^{-1/2}\right)  ^{2}+O\left(  \left(  d\rho
_{\theta}\right)  ^{3}\right)  .
\end{multline}
Now left-multiplying by $\rho_{\theta}$ and taking the trace gives
\begin{align}
&  \widehat{Q}_{\alpha}(\rho_{\theta+\delta}\Vert\rho_{\theta})\nonumber\\
&  =\operatorname{Tr}\!\left[  \rho_{\theta}\left(  I+\rho_{\theta}^{-1/2}
d\rho_{\theta}\rho_{\theta}^{-1/2}\right)  ^{\alpha}\right] \\
&  =\operatorname{Tr}\!\left[  \rho_{\theta}\left(  I+\alpha\rho_{\theta}
^{-1/2}d\rho_{\theta}\rho_{\theta}^{-1/2}+\frac{1}{2}\left(  \alpha-1\right)
\alpha\left(  \rho_{\theta}^{-1/2}d\rho_{\theta}\rho_{\theta}^{-1/2}\right)
^{2}+O\left(  \left(  d\rho_{\theta}\right)  ^{3}\right)  \right)  \right] \\
&  =\operatorname{Tr}[\rho_{\theta}]+\alpha\operatorname{Tr}[d\rho_{\theta
}]+\frac{1}{2}\left(  \alpha-1\right)  \alpha\operatorname{Tr}\!\left[
\rho_{\theta}\left(  \rho_{\theta}^{-1/2}d\rho_{\theta}\rho_{\theta}
^{-1/2}\right)  ^{2}\right]  +O\left(  \left(  d\rho_{\theta}\right)
^{3}\right) \\
&  =1+\frac{1}{2}\left(  \alpha-1\right)  \alpha\operatorname{Tr}\left[
d\rho_{\theta}\rho_{\theta}^{-1}d\rho_{\theta}\right]  +O\left(  \left(
d\rho_{\theta}\right)  ^{3}\right)  .
\end{align}
So then
\[
\frac{2}{\alpha\left(  \alpha-1\right)  \delta^{2}}\left(  \widehat{Q}
_{\alpha}(\rho_{\theta+\delta}\Vert\rho_{\theta})-1\right)  =\operatorname{Tr}
\!\left[  \frac{d\rho_{\theta}}{\delta}\rho_{\theta}^{-1}\frac{d\rho_{\theta}
}{\delta}\right]  +\frac{2}{\alpha\left(  \alpha-1\right)  \delta^{2}}O\left(
\left(  d\rho_{\theta}\right)  ^{3}\right)  .
\]
For a second-order differentiable family, the following limit holds
\begin{equation}
\lim_{\delta\rightarrow0}\frac{1}{\delta^{2}}O\left(  \left\Vert d\rho
_{\theta}\right\Vert _{\infty}^{3}\right)  =\lim_{\delta\rightarrow0}
\delta\ O\left(  \left[  \left\Vert d\rho_{\theta}/\delta\right\Vert _{\infty
}\right]  ^{3}\right)  =0.
\end{equation}
Then we find that
\begin{align}
\lim_{\delta\rightarrow0}\frac{2}{\alpha\left(  \alpha-1\right)  \delta^{2}
}\left(  \widehat{Q}_{\alpha}(\rho_{\theta+\delta}\Vert\rho_{\theta
})-1\right)   &  =\lim_{\delta\rightarrow0}\operatorname{Tr}\!\left[
\frac{d\rho_{\theta}}{\delta}\rho_{\theta}^{-1}\frac{d\rho_{\theta}}{\delta
}\right] \\
&  =\operatorname{Tr}\!\left[  (\partial_{\theta}\rho_{\theta})^{2}\rho_{\theta
}^{-1}\right] \\
&  =\widehat{I}_{F}(\theta;\{\rho_{\theta}\}_{\theta}),
\end{align}
as claimed.

The equality between \eqref{eq:geo-renyi-limit-to-RLD} and
\eqref{eq:geo-renyi-limit-to-RLD-2} is similar to the equality between
\eqref{eq:SLD-fish-bures-connection} and
\eqref{eq:SLD-fish-bures-connection-2} and is shown in
Appendix~\ref{app:fish-info-fid-limit}. Defining $\eta(x)=x\ln x$, the last
equality in \eqref{eq:RLD-from-BS-rel-ent} follows because
\begin{align}
  \widehat{D}_{\alpha}(\rho_{\theta+\delta}\Vert\rho_{\theta})
&  =\operatorname{Tr}[\rho_{\theta+\delta}\ln\rho_{\theta+\delta}^{1/2}
\rho_{\theta}^{-1}\rho_{\theta+\delta}^{1/2}]\\
&  =\operatorname{Tr}[\rho_{\theta}\eta(\rho_{\theta}^{-1/2}\rho
_{\theta+\delta}\rho_{\theta}^{-1/2})]\\
&  =\operatorname{Tr}[\rho_{\theta}\eta(\rho_{\theta}^{-1/2}(\rho_{\theta
}+d\rho_{\theta})\rho_{\theta}^{-1/2})]\\
&  =\operatorname{Tr}[\rho_{\theta}\eta(I+\rho_{\theta}^{-1/2}d\rho_{\theta
}\rho_{\theta}^{-1/2})]\\
&  =\operatorname{Tr}[\rho_{\theta}(\rho_{\theta}^{-1/2}d\rho_{\theta}
\rho_{\theta}^{-1/2}+[\rho_{\theta}^{-1/2}d\rho_{\theta}\rho_{\theta}
^{-1/2}]^{2}/2)]+O((d\rho_{\theta})^{3})\\
&  =\operatorname{Tr}[d\rho_{\theta}]+\operatorname{Tr}[d\rho_{\theta}
\rho_{\theta}^{-1}d\rho_{\theta}]/2+O((d\rho_{\theta})^{3})\\
&  =\operatorname{Tr}[d\rho_{\theta}\rho_{\theta}^{-1}d\rho_{\theta
}]/2+O((d\rho_{\theta})^{3}),
\end{align}
where we used that $\eta(1+x)=x+x^{2}/2+O(x^{3})$. The reasoning to arrive at
\eqref{eq:RLD-from-BS-rel-ent}\ is similar to what was given previously.
\end{proof}

\subsection{Linking properties of Fisher informations and R\'{e}nyi relative
entropies}

The limit formulas in Propositions~\ref{prop:SLD-fish-limit-fidelity} and
\ref{prop:RLD-fish-limit-geo-renyi}\ allow us to connect properties of the
SLD\ and RLD\ Fisher informations to the fidelity and geometric R\'{e}nyi
relative entropy, respectively. This only occurs when the family of states or
channels is second-order differentiable, because the limit formulas in
Propositions~\ref{prop:SLD-fish-limit-fidelity} and
\ref{prop:RLD-fish-limit-geo-renyi} only apply under such a circumstance.

We list the connections now:

\begin{itemize}
\item Data processing for the SLD\ and RLD\ Fisher informations in
\eqref{eq:DP-SLD}--\eqref{eq:DP-RLD} follows from data processing for the
fidelity and the geometric R\'{e}nyi relative entropy in
\eqref{eq:DP-geo-Renyi-main-txt}, respectively.

\item Additivity of the SLD\ and RLD\ Fisher informations in
\eqref{eq:additivity-SLD-states} and \eqref{eq:additivity-RLD-states} follows
from the limit formulas in \eqref{eq:SLD-fish-bures-connection-2} and
\eqref{eq:geo-renyi-limit-to-RLD-2}, respectively, and additivity of these quantities.

\item The decomposition of SLD\ and RLD\ Fisher informations for
classical--quantum states in Proposition~\ref{prop:cq-decomp-SLD-RLD}\ follows
from the limit formulas in \eqref{eq:SLD-fish-bures-connection} and
\eqref{eq:geo-renyi-limit-to-RLD}, respectively, and also because the
underlying quantities have the same decomposition for classical--quantum states.

\item The amortization collapse in Theorem~\ref{thm:amort-collapse-cq}\ for
the SLD\ Fisher information of classical--quantum channels is a consequence of
the amortization collapse for the sandwiched R\'{e}nyi relative entropy given
in \cite[Lemma~26]{Berta2018c}.

\item The chain rule for the root SLD\ Fisher information in
Proposition~\ref{prop:chain-rule-root-SLD}\ is a consequence of the limit
formula in \eqref{eq:SLD-fish-bures-connection}, the triangle inequality for
the Bures distance, and the related chain-rule inequality given in
\cite[Lemma~44]{Berta2018c}.

\item The additivity of the RLD\ Fisher information in
Proposition~\ref{prop:add-RLD-Fish-ch} is a consequence of the limit formula
in \eqref{eq:geo-renyi-limit-to-RLD-2}\ and the additivity of the geometric
R\'{e}nyi relative entropy of quantum channels (the latter can either be shown
directly or as a consequence of
Proposition~\ref{prop:subadd-serial-comp-geo-renyi}).

\item The simple formula for the RLD\ Fisher information in
Proposition~\ref{prop:geo-fish-explicit-formula}\ can be seen as a consequence
of the limit formula in \eqref{eq:geo-renyi-limit-to-RLD-2} and the simple formula
for the geometric R\'{e}nyi relative entropy of quantum channels. This is
shown explicitly in Appendix~\ref{app:geo-renyi-limit-to-RLD-fish}.

\item The chain rule for the RLD\ Fisher information in
Proposition~\ref{prop:chain-rule-RLD}\ is a consequence of the limit formula
in \eqref{eq:geo-renyi-limit-to-RLD-2}\ and the chain rule for the geometric
R\'{e}nyi relative entropy from Proposition~\ref{prop:chain-rule-geo-renyi-ch}.
\end{itemize}

\subsection{Semi-definite programs for channel fidelity and SLD\ Fisher
information of quantum channels}

In this section, we show how the fidelity of quantum channels can be computed
by means of a semi-definite program. This was already shown in \cite{Yuan2017}, but here we arrive at semi-definite programs that are functions of the Choi
operators of the channels involved. Once the semi-definite program for
fidelity of channels is established, one can then use it and generalizations of the limit
formulas from Proposition~\ref{prop:SLD-fish-limit-fidelity} to approximate the
SLD\ Fisher information of quantum channels.

Our starting point is the following semi-definite program and its dual for the
root fidelity of quantum states \cite{Wat13}:

\begin{proposition}
Let $\rho$ and $\sigma$ be quantum states. Then their root fidelity $\sqrt
{F}(\rho,\sigma)=\left\Vert \sqrt{\rho}\sqrt{\sigma}\right\Vert _{1}$ can be
calculated by means of the following semi-definite program
\begin{equation}
\sup_{Q}\left\{  \operatorname{Re}[\operatorname{Tr}[Q]]:
\begin{bmatrix}
\rho & Q^{\dag}\\
Q & \sigma
\end{bmatrix}
\geq0\right\}  ,
\end{equation}
and its dual is given by
\begin{equation}
\frac{1}{2}\inf_{W,Z}\left\{  \operatorname{Tr}[\rho W]+\operatorname{Tr}
[\sigma Z]:
\begin{bmatrix}
W & I\\
I & Z
\end{bmatrix}
\geq0\right\}  . \label{eq:dual-root-fid-states-SDP}
\end{equation}

\end{proposition}

Using this semi-definite program, we then find the following for the root
fidelity of quantum channels:

\begin{proposition}
\label{prop:SDP-ch-fid}Let $\mathcal{N}_{A\rightarrow B}$ and $\mathcal{M}
_{A\rightarrow B}$ be quantum channels with respective Choi operators
$\Gamma_{RB}^{\mathcal{N}}$ and $\Gamma_{RB}^{\mathcal{M}}$. Then their root
channel fidelity
\begin{equation}
\sqrt{F}(\mathcal{N}_{A\rightarrow B},\mathcal{M}_{A\rightarrow B}
):=\inf_{\psi_{RA}}\sqrt{F}(\mathcal{N}_{A\rightarrow B}(\psi_{RA}
),\mathcal{M}_{A\rightarrow B}(\psi_{RA}))
\end{equation}
can be calculated by means of the following semi-definite program:
\begin{equation}
\sup_{\lambda\geq0,Q_{RB}}\left\{  \lambda:\lambda I_{R}\leq\operatorname{Re}
[\operatorname{Tr}_{B}[Q_{RB}]],\quad
\begin{bmatrix}
\Gamma_{RB}^{\mathcal{N}} & Q_{RB}^{\dag}\\
Q_{RB} & \Gamma_{RB}^{\mathcal{M}}
\end{bmatrix}
\geq0\right\}  . \label{eq:SDP-root-fid-ch}
\end{equation}
and its dual is given by
\begin{equation}
\frac{1}{2}\inf_{\rho_{R},W_{RB},Z_{RB}}\operatorname{Tr}[\Gamma
_{RB}^{\mathcal{N}}W_{RB}]+\operatorname{Tr}[\Gamma_{RB}^{\mathcal{M}}Z_{RB}],
\end{equation}
subject to
\begin{equation}
\rho_{R}\geq0,\quad\operatorname{Tr}[\rho_{R}]=1,\quad
\begin{bmatrix}
W_{RB} & \rho_{R}\otimes I_{B}\\
\rho_{R}\otimes I_{B} & Z_{RB}
\end{bmatrix}
\geq0.
\end{equation}
The expression in \eqref{eq:SDP-root-fid-ch}\ is equal to
\begin{equation}
\sup_{Q_{RB}}\left\{  \lambda_{\min}\left(  \operatorname{Re}
[\operatorname{Tr}_{B}[Q_{RB}]]\right)  :
\begin{bmatrix}
\Gamma_{RB}^{\mathcal{N}} & Q_{RB}^{\dag}\\
Q_{RB} & \Gamma_{RB}^{\mathcal{M}}
\end{bmatrix}
\geq0\right\}  ,
\end{equation}
where $\lambda_{\min}$ denotes the minimum eigenvalue of its argument.
\end{proposition}

\begin{proof}
See Appendix~\ref{app:SDP-root-fid-ch}.
\end{proof}

\begin{remark}
Now combining Proposition~\ref{prop:SDP-ch-fid}\ with the following limit
formula for SLD\ Fisher information of quantum channels
\begin{equation}
I_{F}(\theta;\{\mathcal{N}_{A\rightarrow B}^{\theta}\}_{\theta})=\lim
_{\delta\rightarrow0}\frac{8}{\delta^{2}}(1-\sqrt{F}(\mathcal{N}_{A\rightarrow
B}^{\theta},\mathcal{N}_{A\rightarrow B}^{\theta+\delta})),
\end{equation}
we can approximate $I_{F}(\theta;\{\mathcal{N}_{A\rightarrow B}^{\theta
}\}_{\theta})$ numerically by picking $\delta\approx10^{-3}$ or $\delta\approx10^{-4}$ and
calculating $\sqrt{F}(\mathcal{N}_{A\rightarrow B}^{\theta},\mathcal{N}
_{A\rightarrow B}^{\theta+\delta})$ by means of the semi-definite program in
Proposition~\ref{prop:SDP-ch-fid}.
\end{remark}

\section{Conclusion}

\label{sec:conclusion}

In this paper, we have used geometric distinguishability
measures to place limits on the related tasks of quantum channel estimation
and discrimination. By proving chain rules for the RLD Fisher information, as well as the
 root SLD Fisher information, we have established single-letter quantum Cramer--Rao
bounds on the performance of estimating a parameter encoded in a quantum channel.
In particular, the chain rule for the RLD Fisher information implies a simple condition to
determine if a particular family of channels can admit Heisenberg scaling in
error, complementing other conditions that have been presented previously in various settings
\cite{Fujiwara2008, Mat10, Hayashi2011, Demkowicz-Dobrzanski2017, Zhou2018}.

We have also used the geometric R\'enyi relative entropy to improve the
bounds of \cite{Berta2018c, CE18} in the realm of quantum channel discrimination,
particularly in both the Chernoff and  Hoeffding settings. Finally, we have detailed some conceptual and technical connections between estimation and discrimination. The conceptual connections are due to the fact that one task can be seen as a generalization of the other. The technical connections are due to the divergence measures that underlie each Fisher information quantity, whenever the family under question is second-order differentiable.

Extending our results to multiparameter estimation has been accomplished in \cite{UpcomingWork}. In future work, we will include energy constraints in our formalism and study the behavior of QFI quantities in the presence of energy constraints on the probe state. That is, the operational quantity to be developed further in future work is the energy-constrained generalized Fisher information of a quantum channel family, defined as follows:
\begin{equation}
    \mathbf{I}_{F, E} ( \theta; \{ \mathcal{N}^{\theta}_{A \rightarrow B} \}_{\theta} ) = \sup_{\rho_{RA}: \mathrm{Tr}[ H_A \rho_A ] \leq E} \mathbf{I}_{F} ( \theta; \{ \mathcal{N}^{\theta}_{A \rightarrow B} (\rho_{RA}) \}_{\theta} )
\end{equation}
where $H_A$ is a Hamiltonian acting on the input system of the channel $\mathcal{N}_{A \rightarrow B}^{\theta}$. This definition generalizes the energy-constrained channel divergence introduced in \cite{SWAT18}.
Furthermore, a relevant information quantity for sequential channel estimation with energy constraints is the following energy-constrained amortized Fisher information:
\begin{equation}
    \mathbf{I}_{F, E}^{\mathcal{A}} ( \theta; \{ \mathcal{N}^{\theta}_{A \rightarrow B} \}_{\theta} ) = \sup_{\{\rho^{\theta}_{RA}\}_\theta: \mathrm{Tr}[ H_A \rho^{\theta}_A ] \leq E} \mathbf{I}_{F} ( \theta; \{ \mathcal{N}^{\theta}_{A \rightarrow B} (\rho^{\theta}_{RA}) \}_{\theta} ) - \mathbf{I}_{F} ( \theta; \{  \rho^{\theta}_{RA} \}_{\theta} )
\end{equation}
We will study properties of these energy-constrained Fisher informations analogous to their corresponding unconstrained versions.

It is an interesting open question to determine whether sequential channel discrimination strategies offer any benefit over parallel discrimination strategies in the limit of a large number of channel uses and in the Chernoff and Hoeffding error exponent settings.  It is known that, in asymmetric quantum channel discrimination,  sequential strategies offer no advantage over parallel ones in the limit of a large number channel uses \cite{Hayashi09,Berta2018c,PhysRevResearch.1.033169,PhysRevLett.124.100501}. In a recent paper \cite{Zhou2020} concurrent to ours, it was established that sequential estimation strategies offer no advantage over parallel ones in the limit of a large number of channel uses whenever Heisenberg scaling is unattainable. What remains open is to determine whether sequential strategies can outperform parallel strategies in the case when Heisenberg scaling is attainable.  

We also leave open the question of determining an operational interpretation of  the RLD Fisher information
of channels as the optimal classical Fisher information needed to simulate the channel family in a local way (inspired by the question addressed in 
\cite{Matsumoto2005} for quantum state families). This task connects to coherence distillation of quantum channels from a resource-theoretic perspective \cite{Marvian2018}.

We acknowledge discussions with Sam Cree and Sumeet Khatri about geometric
R\'{e}nyi relative entropy. %We thank Mankei Tsang for feedback on our paper.
We also thank Sisi Zhou for discussions related to our paper. VK acknowledges support from the LSU Economic
Development Assistantship. VK and MMW acknowledge support from the US
National Science Foundation via grant number 1907615. MMW acknowledges support from Stanford QFARM and
AFOSR (FA9550-19-1-0369).

\phantomsection
\addcontentsline{toc}{section}{References}

\bibliographystyle{unsrt}
\bibliography{estimation-refs}

\begin{thebibliography}{100}

\bibitem{Lloyd2008}
Seth Lloyd.
\newblock Enhanced sensitivity of photodetection via quantum illumination.
\newblock {\em Science}, 321(5895):1463--1465, 2008.
\newblock arXiv:0803.2022.

\bibitem{Braunstein1992}
Samuel~L. {Braunstein}.
\newblock {Quantum limits on precision measurements of phase}.
\newblock {\em Physical Review Letters}, 69(25):3598--3601, December 1992.

\bibitem{Dowling1998}
Jonathan~P. Dowling.
\newblock Correlated input-port, matter-wave interferometer: Quantum-noise
  limits to the atom-laser gyroscope.
\newblock {\em Physical Review A}, 57(6):4736--4746, June 1998.

\bibitem{Demkowicz-Dobrzanski2015}
Rafal {Demkowicz-Dobrzanski}, Marcin Jarzyna, and Jan Kolodynski.
\newblock Quantum limits in optical interferometry.
\newblock {\em Progress in Optics}, 60:345--435, 2015.
\newblock arXiv:1405.7703.

\bibitem{Caves1981}
Carlton~M. Caves.
\newblock Quantum mechanical noise in an interferometer.
\newblock {\em Physical Review D}, 23(8):1693--1708, April 1981.

\bibitem{Yurke1986}
Bernard {Yurke}, Samuel~L. {McCall}, and John~R. {Klauder}.
\newblock {SU}(2) and {SU}(1,1) interferometers.
\newblock {\em Physical Review A}, 33(6):4033--4054, June 1986.

\bibitem{Berry2000}
Dominic~W. {Berry} and Howard~M. {Wiseman}.
\newblock Optimal states and almost optimal adaptive measurements for quantum
  interferometry.
\newblock {\em Physical Review Letters}, 85(24):5098--5101, December 2000.
\newblock arXiv:quant-ph/0009117.

\bibitem{Demkowicz-Dobrzanksi2013}
Rafal {Demkowicz-Dobrzanski}, Konrad Banaszek, and Roman Schnabel.
\newblock Fundamental quantum interferometry bound for the
  squeezed-light-enhanced gravitational wave detector {GEO} 600.
\newblock {\em Physical Review A}, 88(4):041802, October 2013.
\newblock arXiv:1305.7268.

\bibitem{PhysRevE.98.032106}
Schuyler~B. Nicholson, Adolfo del Campo, and Jason~R. Green.
\newblock Nonequilibrium uncertainty principle from information geometry.
\newblock {\em Physical~Review~E}, 98(3):032106, September 2018.
\newblock arXiv:1801.02242.

\bibitem{NGCG20}
Schuyler~B. Nicholson, Luis~Pedro Garcia-Pintos, Adolfo del Campo, and Jason~R.
  Green.
\newblock Time-information uncertainty relations in thermodynamics.
\newblock {\em Nature Physics}, 16:1211--1215, September 2020.
\newblock arXiv:2001.05418.

\bibitem{Hel67}
Carl~W. Helstrom.
\newblock Minimum mean-squared error of estimates in quantum statistics.
\newblock {\em Physics Letters A}, 25(2):101--102, July 1967.

\bibitem{YL73}
Horace {Yuen} and Melvin {Lax}.
\newblock Multiple-parameter quantum estimation and measurement of
  nonselfadjoint observables.
\newblock {\em IEEE Transactions on Information Theory}, 19(6):740--750,
  November 1973.

\bibitem{Sidhu2019}
Jasminder~S. Sidhu and Pieter Kok.
\newblock A geometric perspective on quantum parameter estimation.
\newblock {\em AVS Quantum Science}, 2(1):014701, February 2020.
\newblock arXiv:1907.06628.

\bibitem{Sasaki2002}
Masahide Sasaki, Masashi Ban, and Stephen~M. Barnett.
\newblock Optimal parameter estimation of a depolarizing channel.
\newblock {\em Physical Review A}, 66(2):022308, August 2002.
\newblock arXiv:quant-ph/0203113.

\bibitem{Fujiwara_2003}
Akio Fujiwara and Hiroshi Imai.
\newblock Quantum parameter estimation of a generalized {{Pauli}} channel.
\newblock {\em Journal of Physics A: Mathematical and General},
  36(29):8093--8103, July 2003.

\bibitem{Fujiwara2004}
Akio Fujiwara.
\newblock Estimation of a generalized amplitude-damping channel.
\newblock {\em Physical Review A}, 70(1):012317, July 2004.

\bibitem{Ji2008}
Zhengfeng Ji, Guoming Wang, Runyao Duan, Yuan Feng, and Mingsheng Ying.
\newblock Parameter estimation of quantum channels.
\newblock {\em IEEE Transactions on Information Theory}, 54(11):5172--5185,
  November 2008.
\newblock arXiv:quant-ph/0610060.

\bibitem{Fujiwara2008}
Akio Fujiwara and Hiroshi Imai.
\newblock A fibre bundle over manifolds of quantum channels and its application
  to quantum statistics.
\newblock {\em Journal of Physics A: Mathematical and Theoretical},
  41(25):255304, June 2008.

\bibitem{Mat10}
Keiji Matsumoto.
\newblock On metric of quantum channel spaces.
\newblock June 2010.
\newblock arXiv:1006.0300.

\bibitem{Hayashi2011}
Masahito Hayashi.
\newblock Comparison between the {C}ramer-{Rao} and the mini-max approaches in
  quantum channel estimation.
\newblock {\em Communications in Mathematical Physics}, 304(3):689--709, June
  2011.
\newblock arXiv:1003.4575.

\bibitem{Demkowicz-Dobrzanski2012}
Rafal {Demkowicz-Dobrzanski}, Jan Kolodynski, and Madalin Guta.
\newblock The elusive {{Heisenberg}} limit in quantum enhanced metrology.
\newblock {\em Nature Communications}, 3(1):1063, January 2012.
\newblock arXiv:1201.3940.

\bibitem{Kolodynski2013}
Jan Ko{\l}ody{\'n}ski and Rafa{\l} {Demkowicz-Dobrza{\'n}ski}.
\newblock Efficient tools for quantum metrology with uncorrelated noise.
\newblock {\em New Journal of Physics}, 15(7):073043, July 2013.
\newblock arXiv:1303.7271.

\bibitem{Demkowicz-Dobrzanski2014}
Rafal {Demkowicz-Dobrzanski} and Lorenzo Maccone.
\newblock Using entanglement against noise in quantum metrology.
\newblock {\em Physical Review Letters}, 113(25):250801, December 2014.
\newblock arXiv:1407.2934.

\bibitem{Sekatski2017}
Pavel Sekatski, Michalis Skotiniotis, Janek Ko{\l}ody{\'n}ski, and Wolfgang
  D{\"u}r.
\newblock Quantum metrology with full and fast quantum control.
\newblock {\em Quantum}, 1:27, September 2017.
\newblock arXiv:1603.08944.

\bibitem{Demkowicz-Dobrzanski2017}
Rafal {Demkowicz-Dobrzanski}, Jan Czajkowski, and Pavel Sekatski.
\newblock Adaptive quantum metrology under general {{Markovian}} noise.
\newblock {\em Physical Review X}, 7(4):041009, October 2017.
\newblock arXiv:1704.06280.

\bibitem{Zhou2018}
Sisi Zhou, Mengzhen Zhang, John Preskill, and Liang Jiang.
\newblock Achieving the {{Heisenberg}} limit in quantum metrology using quantum
  error correction.
\newblock {\em Nature Communications}, 9(1):78, December 2018.
\newblock arXiv:1706.02445.

\bibitem{Zhou2019}
Sisi Zhou and Liang Jiang.
\newblock Optimal approximate quantum error correction for quantum metrology.
\newblock {\em Physical Review Research}, 2(1):013235, March 2020.
\newblock arXiv:1910.08472.

\bibitem{Zhou2019a}
Sisi Zhou and Liang Jiang.
\newblock An exact correspondence between the quantum {{Fisher}} information
  and the {{Bures}} metric.
\newblock October 2019.
\newblock arXiv:1910.08473.

\bibitem{YCH20}
Yuxiang Yang, Giulio Chiribella, and Masahito Hayashi.
\newblock Communication cost of quantum processes.
\newblock {\em IEEE Journal on Selected Areas in Information Theory},
  1(2):387--400, August 2020.
\newblock arXiv:2002.06840.

\bibitem{Giovannetti2006}
Vittorio Giovannetti, Seth Lloyd, and Lorenzo Maccone.
\newblock Quantum {{Metrology}}.
\newblock {\em Physical Review Letters}, 96(1):010401, January 2006.
\newblock arXiv:quant-ph/0509179.

\bibitem{PhysRevLett.98.090501}
Wim van Dam, G.~Mauro D'Ariano, Artur Ekert, Chiara Macchiavello, and Michele
  Mosca.
\newblock Optimal quantum circuits for general phase estimation.
\newblock {\em Physical Review Letters}, 98(9):090501, March 2007.
\newblock arXiv:quant-ph/0609160.

\bibitem{Yuan2017}
Haidong Yuan and Chi-Hang~Fred Fung.
\newblock Fidelity and {{Fisher}} information on quantum channels.
\newblock {\em New Journal of Physics}, 19(11):113039, November 2017.
\newblock arXiv:1506.00819.

\bibitem{Berta2018c}
Mark~M. Wilde, Mario Berta, Christoph Hirche, and Eneet Kaur.
\newblock Amortized channel divergence for asymptotic quantum channel
  discrimination.
\newblock {\em Letters in Mathematical Physics}, 100:2277--2336, August 2020.
\newblock arXiv:1808.01498.

\bibitem{PR98}
D\'enes Petz and Mary~Beth Ruskai.
\newblock Contraction of generalized relative entropy under stochastic mappings
  on matrices.
\newblock {\em Infinite Dimensional Analysis, Quantum Probability and Related
  Topics}, 1(1):83--89, January 1998.

\bibitem{Mat13}
Keiji Matsumoto.
\newblock A new quantum version of f-divergence.
\newblock November 2013.
\newblock arXiv:1311.4722.

\bibitem{Matsumoto2018}
Keiji Matsumoto.
\newblock A new quantum version of f-divergence.
\newblock In Masanao Ozawa, Jeremy Butterfield, Hans Halvorson, Mikl{\'o}s
  R{\'e}dei, Yuichiro Kitajima, and Francesco Buscemi, editors, {\em Reality
  and Measurement in Algebraic Quantum Theory}, volume 261, pages 229--273,
  Singapore, 2018. Springer Singapore.
\newblock Series Title: Springer Proceedings in Mathematics \& Statistics.

\bibitem{T15book}
Marco Tomamichel.
\newblock {\em Quantum Information Processing with Finite Resources:
  Mathematical Foundations}, volume~5.
\newblock Springer, 2015.
\newblock arXiv:1504.00233.

\bibitem{HM17}
Fumio Hiai and Mil\'an Mosonyi.
\newblock Different quantum $f$-divergences and the reversibility of quantum
  operations.
\newblock {\em Reviews in Mathematical Physics}, 29(07):1750023, August 2017.
\newblock arXiv:1604.03089.

\bibitem{Fang2019a}
Kun Fang and Hamza Fawzi.
\newblock Geometric {R}\'enyi divergence and its applications in quantum
  channel capacities.
\newblock September 2019.
\newblock arXiv:1909.05758v1.

\bibitem{CE18}
Giulio Chiribella and Daniel Ebler.
\newblock Quantum speedup in the identification of cause-effect relations.
\newblock {\em Nature Communications}, 10:1472, April 2019.
\newblock arXiv:1806.06459.

\bibitem{H06}
Masahito Hayashi.
\newblock {\em Quantum Information: An Introduction}.
\newblock Berlin Heidelberg: Springer Verlag, 2006.

\bibitem{H13book}
Alexander~S. Holevo.
\newblock {\em Quantum Systems, Channels, Information: A Mathematical
  Introduction}, volume~16.
\newblock Walter de Gruyter, 2013.

\bibitem{Wat18}
John Watrous.
\newblock {\em The Theory of Quantum Information}.
\newblock Cambridge University Press, 2018.

\bibitem{Wbook17}
Mark~M. Wilde.
\newblock {\em Quantum Information Theory}.
\newblock Cambridge University Press, 2nd edition, 2017.
\newblock arXiv:1106.1445.

\bibitem{B05}
Charles~H. Bennett.
\newblock Simulated time travel, teleportation without communication, and how
  to conduct a romance with someone who has fallen into a black hole.
\newblock \url{https://www.research.ibm.com/people/b/bennetc/QUPONBshort.pdf},
  May 2005.

\bibitem{Li2018}
Yan Li, Luca Pezz{\`{e}}, Manuel Gessner, Zhihong Ren, Weidong Li, and Augusto
  Smerzi.
\newblock Frequentist and bayesian quantum phase estimation.
\newblock {\em Entropy}, 20(9):628, August 2018.
\newblock arXiv:1804.10048.

\bibitem{Hel69}
Carl~W. Helstrom.
\newblock Quantum detection and estimation theory.
\newblock {\em Journal of Statistical Physics}, 1:231--252, 1969.

\bibitem{Hol72}
Alexander~S. Holevo.
\newblock An analogue of statistical decision theory and noncommutative
  probability theory.
\newblock {\em Trudy Moskovskogo Matematicheskogo Obshchestva}, 26:133--149,
  1972.

\bibitem{Hel76}
Carl~W. Helstrom.
\newblock {\em Quantum Detection and Estimation Theory}.
\newblock Academic Press, 1976.

\bibitem{Gutoski2007}
Gus Gutoski and John Watrous.
\newblock Toward a general theory of quantum games.
\newblock {\em Proceedings of the thirty-ninth annual ACM symposium on Theory
  of computing}, pages 565--574, 2007.
\newblock arXiv:quant-ph/0611234.

\bibitem{Gutoski2010}
Gus Gutoski.
\newblock {\em Quantum strategies and local operations}.
\newblock PhD thesis, University of Waterloo, 2009.
\newblock arXiv:1003.0038.

\bibitem{Gutoski2012}
Gus Gutoski.
\newblock On a measure of distance for quantum strategies.
\newblock {\em Journal of Mathematical Physics}, 53(3):032202, March 2012.
\newblock arXiv:1008.4636.

\bibitem{Chiribella2008}
Giulio Chiribella, Giacomo~M. D'Ariano, and Paolo Perinotti.
\newblock Memory effects in quantum channel discrimination.
\newblock {\em Physical Review Letters}, 101(18):180501, October 2008.
\newblock arXiv:0803.3237.

\bibitem{Chiribella2009}
Giulio Chiribella, Giacomo~M. D'Ariano, and Paolo Perinotti.
\newblock Theoretical framework for quantum networks.
\newblock {\em Physical Review A}, 80(2):022339, August 2009.
\newblock arXiv:0904.4483.

\bibitem{KW20}
Vishal Katariya and Mark~M. Wilde.
\newblock Evaluating the advantage of adaptive strategies for quantum channel
  distinguishability.
\newblock January 2020.
\newblock arXiv:2001.05376.

\bibitem{CMW14}
Tom Cooney, Mil\'an Mosonyi, and Mark~M. Wilde.
\newblock Strong converse exponents for a quantum channel discrimination
  problem and quantum-feedback-assisted communication.
\newblock {\em Communications in Mathematical Physics}, 344(3):797--829, June
  2016.
\newblock arXiv:1408.3373.

\bibitem{Cram46}
Harald Cram\'er.
\newblock {\em Mathematical Methods of Statistics}.
\newblock Princeton University Press, Princeton, NJ, USA, 1946.

\bibitem{Rao45}
Calyampudi~Radakrishna Rao.
\newblock Information and the accuracy attainable in the estimation of
  statistical parameters.
\newblock {\em Bulletin of the Calcutta Mathematical Society}, 37:81--89, 1945.

\bibitem{Kay93}
Steven~M. Kay.
\newblock {\em Fundamentals of Statistical Signal Processing, Volume I:
  Estimation Theory}.
\newblock Prentice Hall, 1993.

\bibitem{fisher_1925}
Ronald~A. Fisher.
\newblock Theory of statistical estimation.
\newblock {\em Mathematical Proceedings of the Cambridge Philosophical
  Society}, 22(5):700--725, July 1925.

\bibitem{H11book}
Alexander~S. Holevo.
\newblock {\em Probabilistic and statistical aspects of quantum theory},
  volume~1.
\newblock Springer Science \& Business Media, 2011.

\bibitem{Nag89}
Hiroshi Nagaoka.
\newblock A new approach to {Cramer-Rao} bounds for quantum state estimation.
\newblock {\em Journal of the Institute of Electronics, Information, and
  Communication Engineers}, (Report No. IT 89-42):9--14, 1989.

\bibitem{Braunstein1994a}
Samuel~L. Braunstein and Carlton~M. Caves.
\newblock Statistical distance and the geometry of quantum states.
\newblock {\em Physical Review Letters}, 72(22):3439--3443, May 1994.

\bibitem{Fuji94}
Akio Fujiwara.
\newblock One-parameter pure state estimation based on the symmetric
  logarithmic derivative.
\newblock Mathematical Engineering Technical Report 94-8, University of Tokyo,
  July 1994.
\newblock Research Organization Report.

\bibitem{Saf18}
Dominik \v{S}afr\'anek.
\newblock Simple expression for the quantum {F}isher information matrix.
\newblock {\em Physical Review A}, 97(4):042322, April 2018.
\newblock arXiv:1801.00945.

\bibitem{Petz96}
D\'enes Petz.
\newblock Monotone metrics on matrix spaces.
\newblock {\em Linear Algebra and its Applications}, 244:81--96, 1996.

\bibitem{Jen12}
Anna Jencova.
\newblock Reversibility conditions for quantum operations.
\newblock {\em Reviews in Mathematical Physics}, 24(07):1250016, August 2012.
\newblock arXiv:1107.0453.

\bibitem{ArakiMasuda82}
Huzihiro Araki and Tetsuya Masuda.
\newblock Positive cones and $\ell_p$-spaces for von {N}eumann algebras.
\newblock {\em Publications of the Research Institute for Mathematical
  Sciences}, 18(2):339--411, August 1982.

\bibitem{Ando79}
Tsuyoshi Ando.
\newblock Concavity of certain maps on positive definite matrices and
  applications to {H}adamard products.
\newblock {\em Linear Algebra and its Applications}, 26:203--241, August 1979.

\bibitem{Carlen09}
Eric~A. Carlen.
\newblock Trace inequalities and quantum entropy: An introductory course.
\newblock {\em Contemporary Mathematics}, 529:73--140, 2010.

\bibitem{P86}
D\'enes Petz.
\newblock Quasi-entropies for finite quantum systems.
\newblock {\em Reports in Mathematical Physics}, 23:57--65, 1986.

\bibitem{TCR09}
Marco Tomamichel, Roger Colbeck, and Renato Renner.
\newblock A fully quantum asymptotic equipartition property.
\newblock {\em IEEE Transactions on Information Theory}, 55(12):5840--5847,
  December 2009.
\newblock arXiv:0811.1221.

\bibitem{HMPB11}
Fumio Hiai, Mil\'an Mosonyi, D\'enes Petz, and Cedric Beny.
\newblock Quantum $f$-divergences and error correction.
\newblock {\em Reviews in Mathematical Physics}, 23(7):691--747, August 2011.
\newblock arXiv:1008.2529.

\bibitem{W18opt}
Mark~M. Wilde.
\newblock Optimized quantum f-divergences and data processing.
\newblock {\em Journal of Physics A}, 51(37):374002, September 2018.
\newblock arXiv:1710.10252.

\bibitem{BV04}
Stephen Boyd and Lieven Vandenberghe.
\newblock {\em Convex Optimization}.
\newblock Cambridge University Press, Cambridge, UK, 2004.

\bibitem{Matsumoto2005}
Keiji Matsumoto.
\newblock Reverse estimation theory, complementality between {{RLD}} and
  {{SLD}}, and monotone distances.
\newblock November 2005.
\newblock arXiv:quant-ph/0511170.

\bibitem{Choi80}
Man-Duen Choi.
\newblock Some assorted inequalities for positive linear maps on {C}*-algebras.
\newblock {\em Journal of Operator Theory}, 4(2):271--285, 1980.

\bibitem{PhysRevA.91.042104}
S.~Alipour and A.~T. Rezakhani.
\newblock Extended convexity of quantum fisher information in quantum
  metrology.
\newblock {\em Physical Review A}, 91(4):042104, April 2015.
\newblock arXiv:1403.803.

\bibitem{PV10}
Yury Polyanskiy and Sergio Verd\'u.
\newblock Arimoto channel coding converse and {R\'enyi} divergence.
\newblock In {\em Proceedings of the 48th Annual Allerton Conference on
  Communication, Control, and Computation}, pages 1327--1333, September 2010.

\bibitem{SW12}
Naresh Sharma and Naqueeb~Ahmad Warsi.
\newblock On the strong converses for the quantum channel capacity theorems.
\newblock May 2012.
\newblock arXiv:1205.1712.

\bibitem{WWY14}
Mark~M. Wilde, Andreas Winter, and Dong Yang.
\newblock Strong converse for the classical capacity of entanglement-breaking
  and {Hadamard} channels via a sandwiched {R\'enyi} relative entropy.
\newblock {\em Communications in Mathematical Physics}, 331(2):593--622,
  October 2014.
\newblock arXiv:1306.1586.

\bibitem{GW15}
Manish Gupta and Mark~M. Wilde.
\newblock Multiplicativity of completely bounded $p$-norms implies a strong
  converse for entanglement-assisted capacity.
\newblock {\em Communications in Mathematical Physics}, 334(2):867--887, March
  2015.
\newblock arXiv:1310.7028.

\bibitem{TWW17}
Marco Tomamichel, Mark~M. Wilde, and Andreas Winter.
\newblock Strong converse rates for quantum communication.
\newblock {\em {IEEE} Transactions on Information Theory}, 63(1):715--727,
  January 2017.
\newblock arXiv:1406.2946.

\bibitem{WTB16}
Mark~M. Wilde, Marco Tomamichel, and Mario Berta.
\newblock Converse bounds for private communication over quantum channels.
\newblock {\em IEEE Transactions on Information Theory}, 63(3):1792--1817,
  March 2017.
\newblock arXiv:1602.08898.

\bibitem{Led16}
Felix Leditzky.
\newblock {\em Relative entropies and their use in quantum information theory}.
\newblock PhD thesis, University of Cambridge, November 2016.
\newblock arXiv:1611.08802.

\bibitem{KW17}
Eneet Kaur and Mark~M. Wilde.
\newblock Amortized entanglement of a quantum channel and approximately
  teleportation-simulable channels.
\newblock {\em Journal of Physics A: Mathematical and Theoretical},
  51(3):035303, January 2018.
\newblock arXiv:1707.07721.

\bibitem{PhysRevA.101.012344}
Siddhartha Das, Stefan B\"auml, and Mark~M. Wilde.
\newblock Entanglement and secret-key-agreement capacities of bipartite quantum
  interactions and read-only memory devices.
\newblock {\em Physical Review A}, 101(1):012344, January 2020.
\newblock arXiv:1712.00827.

\bibitem{PhysRevLett.123.070502}
Eneet Kaur, Siddhartha Das, Mark~M. Wilde, and Andreas Winter.
\newblock Extendibility limits the performance of quantum processors.
\newblock {\em Physical Review Letters}, 123(7):070502, August 2019.
\newblock arXiv:1803.10710.

\bibitem{WWW19}
Kun Wang, Xin Wang, and Mark~M. Wilde.
\newblock Quantifying the unextendibility of entanglement.
\newblock November 2019.
\newblock arXiv:1911.07433.

\bibitem{TW16}
Masahiro Takeoka and Mark~M. Wilde.
\newblock Optimal estimation and discrimination of excess noise in thermal and
  amplifier channels.
\newblock November 2016.
\newblock arXiv:1611.09165.

\bibitem{LKDW18}
Felix Leditzky, Eneet Kaur, Nilanjana Datta, and Mark~M. Wilde.
\newblock Approaches for approximate additivity of the {Holevo} information of
  quantum channels.
\newblock {\em Physical Review A}, 97(1):012332, January 2018.
\newblock arXiv:1709.01111.

\bibitem{PhysRevResearch.1.033169}
Xin Wang and Mark~M. Wilde.
\newblock Resource theory of asymmetric distinguishability for quantum
  channels.
\newblock {\em Physical Review Research}, 1(3):033169, December 2019.
\newblock arXiv:1907.06306.

\bibitem{fuji2001}
Akio Fujiwara.
\newblock Quantum channel identification problem.
\newblock {\em Physical Review A}, 63(4):042304, March 2001.

\bibitem{BHLS03}
Charles~H. Bennett, Aram~W. Harrow, Debbie~W. Leung, and John~A. Smolin.
\newblock On the capacities of bipartite {Hamiltonians} and unitary gates.
\newblock {\em IEEE Transactions on Information Theory}, 49(8):1895--1911,
  August 2003.
\newblock arXiv:quant-ph/0205057.

\bibitem{BGMW17}
Khaled Ben~Dana, Mar\'{\i}a Garc\'{\i}a~D\'{\i}az, Mohamed Mejatty, and Andreas
  Winter.
\newblock Resource theory of coherence: Beyond states.
\newblock {\em Physical Review A}, 95(6):062327, June 2017.
\newblock arXiv:1704.03710.

\bibitem{RKBKMA17}
Luca Rigovacca, Go~Kato, Stefan Baeuml, M.~S. Kim, W.~J. Munro, and Koji Azuma.
\newblock Versatile relative entropy bounds for quantum networks.
\newblock {\em New Journal of Physics}, 20:013033, January 2018.
\newblock arXiv:1707.05543.

\bibitem{BW17}
Mario Berta and Mark~M. Wilde.
\newblock Amortization does not enhance the max-{Rains} information of a
  quantum channel.
\newblock {\em New Journal of Physics}, 20(5):053044, May 2018.
\newblock arXiv:1709.00200.

\bibitem{DW17}
Siddhartha Das and Mark~M. Wilde.
\newblock Quantum reading capacity: General definition and bounds.
\newblock {\em IEEE Transactions on Information Theory}, 65(11):7566--7583,
  November 2019.
\newblock arXiv:1703.03706.

\bibitem{Wang_2019}
Xin Wang, Mark~M. Wilde, and Yuan Su.
\newblock Quantifying the magic of quantum channels.
\newblock {\em New Journal of Physics}, 21(10):103002, October 2019.
\newblock arXiv:1903.04483.

\bibitem{DW19}
Siddhartha Das and Mark~M. Wilde.
\newblock Quantum rebound capacity.
\newblock {\em Physical Review A}, 100(3):030302, September 2019.
\newblock arXiv:1904.10344.

\bibitem{DP05}
Giacomo~Mauro D'Ariano and Paolo Perinotti.
\newblock Programmable quantum channels and measurements.
\newblock In {\em Workshop on Quantum Information Theory and Quantum
  Statistical Inference}, Tokyo, ERATO Quantum Computation and Information
  Project, November 2005.
\newblock arXiv:quant-ph/0510033.

\bibitem{AHK05}
Sanjeev {Arora}, Elad {Hazan}, and Satyen {Kale}.
\newblock Fast algorithms for approximate semidefinite programming using the
  multiplicative weights update method.
\newblock In {\em 46th Annual IEEE Symposium on Foundations of Computer
  Science}, pages 339--348, 2005.

\bibitem{AK07}
Sanjeev Arora and Satyen Kale.
\newblock A combinatorial, primal-dual approach to semidefinite programs.
\newblock In {\em Proceedings of the Thirty-Ninth Annual ACM Symposium on
  Theory of Computing}, pages 227--236, New York, NY, USA, June 2007.
  Association for Computing Machinery.

\bibitem{AHK12}
Sanjeev Arora, Elad Hazan, and Satyen Kale.
\newblock The multiplicative weights update method: a meta-algorithm and
  applications.
\newblock {\em Theory of Computing}, 8(6):121--164, 2012.

\bibitem{LSW15}
Yin~Tat Lee, Aaron Sidford, and Sam Chiu~Wai Wong.
\newblock A faster cutting plane method and its implications for combinatorial
  and convex optimization.
\newblock In {\em IEEE 56th Annual Symposium on the Foundations of Computer
  Science}, pages 1049--1065, October 2015.
\newblock arXiv:1508.04874.

\bibitem{PB17}
Jaehyun Park and Stephen Boyd.
\newblock General heuristics for nonconvex quadratically constrained quadratic
  programming.
\newblock March 2017.
\newblock arXiv:1703.07870.

\bibitem{Bhat07}
Rajendra Bhatia.
\newblock {\em Positive Definite Matrices}.
\newblock Princeton University Press, Princeton, NJ, USA, 2007.

\bibitem{HKT18}
Stefan Huber, Robert K\"onig, and Marco Tomamichel.
\newblock Jointly constrained semidefinite bilinear programming with an
  application to {D}obrushin curves.
\newblock August 2018.
\newblock arXiv:1808.03182.

\bibitem{Zhou2020}
Sisi Zhou and Liang Jiang.
\newblock Asymptotic theory of quantum channel estimation.
\newblock March 2020.
\newblock arXiv:2003.10559.

\bibitem{NC00}
Michael~A. Nielsen and Isaac~L. Chuang.
\newblock {\em Quantum Computation and Quantum Information}.
\newblock Cambridge University Press, 2000.

\bibitem{KW20book}
Sumeet Khatri and Mark~M. Wilde.
\newblock {\em Principles of quantum communication theory: A modern approach}.
\newblock November 2020.
\newblock arXiv:2011.04672.

\bibitem{LL01}
Jimmie~D. Lawson and Yongdo Lim.
\newblock The geometric mean, matrices, metrics, and more.
\newblock {\em The American Mathematical Monthly}, 108(9):797--812, November
  2001.

\bibitem{Mat14}
Keiji Matsumoto.
\newblock Quantum fidelities, their duals, and convex analysis.
\newblock August 2014.
\newblock arXiv:1408.3462.

\bibitem{Mat14condconv}
Keiji Matsumoto.
\newblock On the condition of conversion of classical probability distribution
  families into quantum families.
\newblock December 2014.
\newblock arXiv:1412.3680.

\bibitem{Mat10fid}
Keiji Matsumoto.
\newblock Reverse test and quantum analogue of classical fidelity and
  generalized fidelity.
\newblock June 2010.
\newblock arXiv:1006.0302.

\bibitem{CS20}
Samuel~S. Cree and Jamie Sikora.
\newblock A fidelity measure for quantum states based on the matrix geometric
  mean.
\newblock June 2020.
\newblock arXiv:2006.06918.

\bibitem{Belavkin1982}
V.~P. Belavkin and P.~Staszewski.
\newblock C*-algebraic generalization of relative entropy and entropy.
\newblock {\em Annales de l'I.H.P. Physique th\'eorique}, 37(1):51--58, 1982.

\bibitem{Datta2009b}
Nilanjana Datta.
\newblock Min- and max-relative entropies and a new entanglement monotone.
\newblock {\em IEEE Transactions on Information Theory}, 55(6):2816--2826, June
  2009.
\newblock arXiv:0803.2770.

\bibitem{MDSFT13}
Martin {M\"uller}-Lennert, Fr\'ed\'eric Dupuis, Oleg Szehr, Serge Fehr, and
  Marco Tomamichel.
\newblock On quantum {R\'enyi} entropies: a new generalization and some
  properties.
\newblock {\em Journal of Mathematical Physics}, 54(12):122203, December 2013.
\newblock arXiv:1306.3142.

\bibitem{Uhl76}
Armin Uhlmann.
\newblock The ``transition probability'' in the state space of a *-algebra.
\newblock {\em Reports on Mathematical Physics}, 9(2):273--279, April 1976.

\bibitem{P85}
D\'enes Petz.
\newblock Quasi-entropies for states of a von {Neumann} algebra.
\newblock {\em Publ. RIMS, Kyoto University}, 21:787--800, 1985.

\bibitem{Hayashi09}
Masahito Hayashi.
\newblock {Discrimination of two channels by adaptive methods and its
  application to quantum system}.
\newblock {\em IEEE Transactions on Information Theory}, 55(8):3807--3820,
  August 2009.
\newblock arXiv:0804.0686.

\bibitem{AMV12}
Koenraad M.~R. Audenaert, Mil\'an Mosonyi, and Frank Verstraete.
\newblock Quantum state discrimination bounds for finite sample size.
\newblock {\em Journal of Mathematical Physics}, 53(12):122205, December 2012.
\newblock arXiv:1204.0711.

\bibitem{Hub92}
Matthias {{Hubner}}.
\newblock Explicit computation of the {B}ures distance for density matrices.
\newblock {\em Physics Letters A}, 163(4):239--242, March 1992.

\bibitem{Sommers_2003}
Hans-J\"urgen Sommers and Karol Zyczkowski.
\newblock Bures volume of the set of mixed quantum states.
\newblock {\em Journal of Physics A: Mathematical and General},
  36(39):10083--10100, September 2003.
\newblock arXiv:quant-ph/0304041.

\bibitem{Hayashi_2002}
Masahito Hayashi.
\newblock Two quantum analogues of {Fisher} information from a large deviation
  viewpoint of quantum estimation.
\newblock {\em Journal of Physics A: Mathematical and General},
  35(36):7689--7727, August 2002.
\newblock arXiv:quant-ph/0202003.

\bibitem{Mos17}
Mil\'an Mosonyi.
\newblock private communication.
\newblock May 2017.

\bibitem{Liu_2014}
Jing Liu, Xiao-Xing Jing, Wei Zhong, and Xiao-Guang Wang.
\newblock Quantum {F}isher information for density matrices with arbitrary
  ranks.
\newblock {\em Communications in Theoretical Physics}, 61(1):45--50, January
  2014.
\newblock arXiv:1312.6910.

\bibitem{Saf2017}
Dominik \v{S}afr\'anek.
\newblock Discontinuities of the quantum {F}isher information and the {B}ures
  metric.
\newblock {\em Physical Review A}, 95(5):052320, May 2017.
\newblock arXiv:1612.04581.

\bibitem{Seveso_2019}
Luigi Seveso, Francesco Albarelli, Marco~G. Genoni, and Matteo G.~A. Paris.
\newblock On the discontinuity of the quantum {F}isher information for quantum
  statistical models with parameter dependent rank.
\newblock {\em Journal of Physics A: Mathematical and Theoretical},
  53(2):02LT01, December 2019.
\newblock arXiv:1906.06185.

\bibitem{Wat13}
John Watrous.
\newblock Simpler semidefinite programs for completely bounded norms.
\newblock {\em Chicago Journal of Theoretical Computer Science}, July 2013.
\newblock arXiv:1207.5726.

\bibitem{UpcomingWork}
Vishal Katariya and Mark~M. Wilde.
\newblock {RLD F}isher information bound for multiparameter estimation of
  quantum channels.
\newblock August 2020.
\newblock arXiv:2008.11178.

\bibitem{SWAT18}
Kunal Sharma, Mark~M. Wilde, Sushovit Adhikari, and Masahiro Takeoka.
\newblock Bounding the energy-constrained quantum and private capacities of
  bosonic thermal channels.
\newblock {\em New Journal of Physics}, 20:063025, June 2018.
\newblock arXiv:1708.07257.

\bibitem{PhysRevLett.124.100501}
Kun Fang, Omar Fawzi, Renato Renner, and David Sutter.
\newblock Chain rule for the quantum relative entropy.
\newblock {\em Physical Review Letters}, 124(10):100501, March 2020.
\newblock arXiv:1909.05826.

\bibitem{Marvian2018}
Iman Marvian.
\newblock Coherence distillation machines are impossible in quantum
  thermodynamics.
\newblock {\em Nature Communications}, 11:25, January 2020.
\newblock arXiv:1805.01989.

\bibitem{MH11}
Mil\'an Mosonyi and Fumio Hiai.
\newblock On the quantum {R\'enyi} relative entropies and related capacity
  formulas.
\newblock {\em IEEE Transactions on Information Theory}, 57(4):2474--2487,
  April 2011.
\newblock arXiv:0912.1286.

\bibitem{U62}
Hisaharu Umegaki.
\newblock Conditional expectations in an operator algebra {IV} (entropy and
  information).
\newblock {\em Kodai Mathematical Seminar Reports}, 14(2):59--85, 1962.

\bibitem{Araki1990}
Huzihiro Araki.
\newblock On an inequality of {Lieb and Thirring}.
\newblock {\em Letters in Mathematical Physics}, 19(2):167--170, February 1990.

\bibitem{LT76}
Elliott~H. Lieb and Walter Thirring.
\newblock {\em Studies in Mathematical Physics}, chapter Inequalities for the
  moments of the eigenvalues of the Schroedinger Hamiltonian and their relation
  to Sobolev inequalities, pages 269--297.
\newblock Princeton University Press, Princeton, 1976.

\bibitem{HP03}
Frank Hansen and Gert~K. Pedersen.
\newblock Jensen's operator inequality.
\newblock {\em Bulletin of the London Mathematical Society}, 35(4):553--564,
  July 2003.
\newblock arXiv:math/0204049.

\bibitem{HP91}
Fumio Hiai and D{\'e}nes Petz.
\newblock The proper formula for relative entropy and its asymptotics in
  quantum probability.
\newblock {\em Communications in Mathematical Physics}, 143(1):99--114, 1991.

\bibitem{Sti55}
William~F. Stinespring.
\newblock Positive functions on {C*}-algebras.
\newblock {\em Proceedings of the American Mathematical Society}, 6:211--216,
  1955.

\bibitem{Petz1986}
D\'enes Petz.
\newblock Sufficient subalgebras and the relative entropy of states of a von
  {Neumann} algebra.
\newblock {\em Communications in Mathematical Physics}, 105(1):123--131, March
  1986.

\bibitem{Petz1988}
D\'enes Petz.
\newblock Sufficiency of channels over von {Neumann} algebras.
\newblock {\em Quarterly Journal of Mathematics}, 39(1):97--108, 1988.

\bibitem{P77}
Eduard {Prugove\v{c}ki}.
\newblock Information-theoretical aspects of quantum measurement.
\newblock {\em International Journal of Theoretical Physics}, 16:321--331, May
  1977.

\bibitem{B91}
Paul Busch.
\newblock Informationally complete sets of physical quantities.
\newblock {\em International Journal of Theoretical Physics}, 30(9):1217--1227,
  September 1991.

\bibitem{RBSC04}
Joseph~M. Renes, Robin Blume-Kohout, A.~J. Scott, and Carlton~M. Caves.
\newblock Symmetric informationally complete quantum measurements.
\newblock {\em Journal of Mathematical Physics}, 45:2171--2180, 2004.
\newblock arXiv:quant-ph/0310075.

\bibitem{DL14limit}
Nilanjana Datta and Felix Leditzky.
\newblock A limit of the quantum {R\'enyi} divergence.
\newblock {\em Journal of Physics A: Mathematical and Theoretical},
  47(4):045304, January 2014.
\newblock arXiv:1308.5961.

\bibitem{MO15}
Mil\'an Mosonyi and Tomohiro Ogawa.
\newblock Two approaches to obtain the strong converse exponent of quantum
  hypothesis testing for general sequences of quantum states.
\newblock {\em IEEE Transactions on Information Theory}, 61(12):6975--6994,
  December 2015.
\newblock arXiv:1407.3567.

\bibitem{KA80}
Fumio Kubo and Tsuyoshi Ando.
\newblock Means of positive linear operators.
\newblock {\em Mathematische Annalen}, 246(3):205--224, October 1980.

\end{thebibliography}

\appendix

\section{Technical lemmas}

Here we collect some technical lemmas used throughout the paper.

\begin{lemma}
\label{lem:min-XYinvX}Let $X$ be a linear operator and let $Y$ be a positive
definite operator. Then
\begin{equation}
X^{\dag}Y^{-1}X=\min\left\{  M:
\begin{bmatrix}
M & X^{\dag}\\
X & Y
\end{bmatrix}
\geq0\right\}  , \label{eq:G_-1_opt}
\end{equation}
where the ordering for the minimization is understood in the operator interval
sense (L\"{o}wner order).
\end{lemma}

\begin{proof}
This is a direct consequence of the Schur complement lemma, which states that
\begin{equation}
\begin{bmatrix}
M & X^{\dag}\\
X & Y
\end{bmatrix}
\geq0\qquad\Longleftrightarrow\qquad Y\geq0,\quad M\geq X^{\dag}Y^{-1}X.
\end{equation}
This concludes the proof.
\end{proof}

\begin{lemma}
\label{lem:freq-used-SDP-primal-dual}Let $K$ and $Z$ be Hermitian operators,
and let $W$ be a linear operator. Then the dual of the following semi-definite
program
\begin{equation}
\inf_{M}\left\{  \operatorname{Tr}[KM]:
\begin{bmatrix}
M & W^{\dag}\\
W & Z
\end{bmatrix}
\geq0\right\}  ,
\end{equation}
with $M$ Hermitian, is given by
\begin{equation}
\sup_{P,Q,R}\left\{  2\operatorname{Re}(\operatorname{Tr}[W^{\dag
}Q])-\operatorname{Tr}[ZR]:P\leq K,
\begin{bmatrix}
P & Q^{\dag}\\
Q & R
\end{bmatrix}
\geq0\right\}  ,
\end{equation}
where $Q$ is a linear operator and $P$ and $R$ are Hermitian.
\end{lemma}

\begin{proof}
The standard forms of a\ primal and dual semi-definite program, for $A$ and
$B$ Hermitian and $\Phi$ a Hermiticity-preserving map, are respectively as
follows \cite{Wat18}:
\begin{align}
&  \inf_{Y\geq0}\left\{  \operatorname{Tr}[BY]:\Phi^{\dag}(Y)\geq A\right\}
,\\
&  \sup_{X\geq0}\left\{  \operatorname{Tr}[AX]:\Phi(X)\leq B\right\}  ,
\end{align}
where $\Phi^{\dag}$ is the Hilbert--Schmidt adjoint of $\Phi$. Noting that
\begin{equation}
\begin{bmatrix}
M & W^{\dag}\\
W & Z
\end{bmatrix}
\geq0\quad\Longleftrightarrow\quad
\begin{bmatrix}
M & -W^{\dag}\\
-W & Z
\end{bmatrix}
\geq0\quad\Longleftrightarrow\quad
\begin{bmatrix}
M & 0\\
0 & 0
\end{bmatrix}
\geq
\begin{bmatrix}
0 & W^{\dag}\\
W & -Z
\end{bmatrix}
,
\end{equation}
we conclude the statement of the lemma after making the following
identifications:
\begin{align}
B  &  =K,\quad Y=M,\quad\Phi^{\dag}(M)=
\begin{bmatrix}
M & 0\\
0 & 0
\end{bmatrix}
,\\
A  &  =
\begin{bmatrix}
0 & W^{\dag}\\
W & -Z
\end{bmatrix}
,\quad X=
\begin{bmatrix}
P & Q^{\dag}\\
Q & R
\end{bmatrix}
,\quad\Phi(X)=P.
\end{align}
This concludes the proof.
\end{proof}

\begin{lemma}
\label{lem:transformer-ineq-basic}Let $X$ be a linear square operator, let $Y$
be a positive definite operator, and let $L$ be a linear operator. Then
\begin{equation}
LX^{\dag}L^{\dag}(LYL^{\dag})^{-1}LXL^{\dag}\leq LX^{\dag}Y^{-1}XL^{\dag},
\label{eq:transformer-ineq}
\end{equation}
where the inverse on the left hand side is taken on the image of $L$. If $L$
is invertible, then the following equality holds
\begin{equation}
LX^{\dag}L^{\dag}(LYL^{\dag})^{-1}LXL^{\dag}=LX^{\dag}Y^{-1}XL^{\dag}.
\label{eq:transformer-eq}
\end{equation}

\end{lemma}

\begin{proof}
Fix an operator $M\geq0$ satisfying
\begin{equation}
\begin{bmatrix}
M & X^{\dag}\\
X & Y
\end{bmatrix}
\geq0. \label{eq:condition-on-M-Schur-comp}
\end{equation}
Since the maps $(\cdot)\rightarrow L(\cdot)L^{\dag}$ and $(\cdot
)\rightarrow\left(  I_{2}\otimes L\right)  (\cdot)\left(  I_{2}\otimes
L\right)  ^{\dag}$ are positive, the condition $M\geq0$ and that in
\eqref{eq:condition-on-M-Schur-comp}\ imply the following conditions:
\begin{align}
LML^{\dag}  &  \geq0,\\
\begin{bmatrix}
LML^{\dag} & LX^{\dag}L^{\dag}\\
LXL^{\dag} & LYL^{\dag}
\end{bmatrix}
&  =\left(  I_{2}\otimes L\right)
\begin{bmatrix}
M & X^{\dag}\\
X & Y
\end{bmatrix}
\left(  I_{2}\otimes L\right)  ^{\dag}\geq0.
\end{align}
Applying \eqref{eq:G_-1_opt}, we conclude that
\begin{align}
LML^{\dag}  &  \geq\min\left\{  W\geq0:
\begin{bmatrix}
W & LX^{\dag}L^{\dag}\\
LXL^{\dag} & LYL^{\dag}
\end{bmatrix}
\geq0\right\} \\
&  =LX^{\dag}L^{\dag}\left(  LYL^{\dag}\right)  ^{-1}LXL^{\dag}.
\end{align}
Since $M$ is an arbitrary operator that satisfies $M\geq0$ and
\eqref{eq:condition-on-M-Schur-comp}, we can pick it to be the smallest and
set it to $X^{\dag}Y^{-1}X$. Thus we conclude \eqref{eq:transformer-ineq}.

If $L$ is invertible, then consider that
\begin{align}
LX^{\dag}L^{\dag}(LYL^{\dag})^{-1}LXL^{\dag}  &  =LX^{\dag}L^{\dag}L^{-\dag
}Y^{-1}L^{-1}LXL^{\dag}\\
&  =LX^{\dag}Y^{-1}XL^{\dag},
\end{align}
so that \eqref{eq:transformer-eq} follows.
\end{proof}

\begin{lemma}
\label{lem:additivity-op-norm} For positive semi-definite operators $X$ and
$Y$,
\begin{equation}
\left\Vert X\otimes I+I\otimes Y\right\Vert _{\infty}=\left\Vert X\right\Vert
_{\infty}+\left\Vert Y\right\Vert _{\infty}.
\label{eq-app-a:op-norm-identity-additive}
\end{equation}

\end{lemma}

\begin{proof}
This follows because
\begin{align}
&  \left\Vert X\otimes I+I\otimes Y\right\Vert _{\infty}\nonumber\\
&  =\sup_{|\psi\rangle:\left\Vert |\psi\rangle\right\Vert _{2}=1}\langle
\psi|\left(  X\otimes I+I\otimes Y\right)  |\psi\rangle\\
&  \geq\sup_{|\phi\rangle,|\varphi\rangle:\left\Vert |\phi\rangle\right\Vert
_{2}=\left\Vert |\varphi\rangle\right\Vert _{2}=1}\left(  \langle\phi
|\otimes\langle\varphi|\right)  \left(  X\otimes I+I\otimes Y\right)  \left(
|\phi\rangle\otimes|\varphi\rangle\right) \\
&  =\sup_{|\phi\rangle,|\varphi\rangle:\left\Vert |\phi\rangle\right\Vert
_{2}=\left\Vert |\varphi\rangle\right\Vert _{2}=1}\langle\phi|X|\phi
\rangle+\langle\varphi|Y|\varphi\rangle\\
&  =\sup_{|\phi\rangle:\left\Vert |\phi\rangle\right\Vert _{2}=1}\langle
\phi|X|\phi\rangle+\sup_{|\varphi\rangle:\left\Vert |\varphi\rangle\right\Vert
_{2}=1}\langle\varphi|Y|\varphi\rangle\\
&  =\left\Vert X\right\Vert _{\infty}+\left\Vert Y\right\Vert _{\infty}.
\end{align}
On the other hand, from the triangle inequality for the infinity norm, we have
that
\begin{align}
\left\Vert X\otimes I+I\otimes Y\right\Vert _{\infty}  &  \leq\left\Vert
X\otimes I\right\Vert _{\infty}+\left\Vert I\otimes Y\right\Vert _{\infty}\\
&  =\left\Vert X\right\Vert _{\infty}+\left\Vert Y\right\Vert _{\infty},
\end{align}
thus establishing \eqref{eq-app-a:op-norm-identity-additive}.
\end{proof}

\begin{lemma}
\label{lem:sing-val-lemma-pseudo-commute}Let $L$ be a square operator and $f$
a function such that the squares of the singular values of $L$ are in the
domain of $f$. Then
\begin{equation}
Lf(L^{\dag}L)=f(LL^{\dag})L.
\end{equation}

\end{lemma}

\begin{proof}
This is a direct consequence of the singular value decomposition theorem. Let
$L=UDV$ be a singular value decomposition of $L$, where $U$ and $V$ are
unitary operators and $D$ is a diagonal, positive semi-definite operator. Then
\begin{align}
Lf(L^{\dag}L)  &  =UDVf(\left(  UDV\right)  ^{\dag}UDV)\\
&  =UDVf(V^{\dag}DU^{\dag}UDV)\\
&  =UDVV^{\dag}f(D^{2})V\\
&  =UDf(D^{2})V\\
&  =Uf(D^{2})DV\\
&  =Uf(DVV^{\dag}D)U^{\dag}UDV\\
&  =f(UDVV^{\dag}DU^{\dag})UDV\\
&  =f(LL^{\dag})L.
\end{align}
This concludes the proof.
\end{proof}

\bigskip

The following lemma builds upon \cite[Lemma~3]{Zhou2019a}, wherein the
essential proof ideas are given.

\begin{lemma}
\label{lem:sisi-zhou-lem}Let $A$ be an invertible Hermitian operator, $B$ a
linear operator, $C$ a Hermitian operator, and let $\varepsilon>0$.\ Then with
\begin{align}
M(\varepsilon)  &  :=
\begin{bmatrix}
A & \varepsilon B\\
\varepsilon B^{\dag} & \varepsilon^{2}C
\end{bmatrix}
,\\
D(\varepsilon)  &  :=
\begin{bmatrix}
A+\varepsilon^{2}\operatorname{Re}[A^{-1}BB^{\dag}] & 0\\
0 & \varepsilon^{2}\left(  C-B^{\dag}A^{-1}B\right)
\end{bmatrix}
,\\
G  &  :=
\begin{bmatrix}
0 & -iA^{-1}B\\
iB^{\dag}A^{-1} & 0
\end{bmatrix}
,
\end{align}
the following inequality holds
\begin{equation}
\left\Vert M(\varepsilon)-e^{-i\varepsilon G}D(\varepsilon)e^{i\varepsilon
G}\right\Vert _{\infty}\leq o(\varepsilon^{2}). \label{eq:sis-zhou-lemm}
\end{equation}

\end{lemma}

\begin{proof}
Observe that $G$ is Hermitian and consider that
\begin{equation}
e^{i\varepsilon G}M(\varepsilon)e^{-i\varepsilon G}=\left(  I+i\varepsilon
G-\frac{\varepsilon^{2}}{2}G^{2}\right)  M(\varepsilon)\left(  I-i\varepsilon
G-\frac{\varepsilon^{2}}{2}G^{2}\right)  +o(\varepsilon^{2}).
\end{equation}
Then we find that
\begin{multline}
\left(  I+i\varepsilon G-\frac{\varepsilon^{2}}{2}G^{2}\right)  M(\varepsilon
)\left(  I-i\varepsilon G-\frac{\varepsilon^{2}}{2}G^{2}\right)
=M(\varepsilon)+i\varepsilon\left[  GM(\varepsilon)-M(\varepsilon)G\right] \\
+\varepsilon^{2}\left[  GM(\varepsilon)G-\frac{1}{2}G^{2}M(\varepsilon
)-\frac{1}{2}M(\varepsilon)G^{2}\right]  +o(\varepsilon^{2}).
\end{multline}
Now observe that
\begin{align}
GM(\varepsilon)  &  =
\begin{bmatrix}
0 & -iA^{-1}B\\
iB^{\dag}A^{-1} & 0
\end{bmatrix}
\begin{bmatrix}
A & \varepsilon B\\
\varepsilon B^{\dag} & \varepsilon^{2}C
\end{bmatrix}
\\
&  =
\begin{bmatrix}
-i\varepsilon A^{-1}BB^{\dag} & -i\varepsilon^{2}A^{-1}BC\\
iB^{\dag} & i\varepsilon B^{\dag}A^{-1}B
\end{bmatrix}
\\
&  =
\begin{bmatrix}
-i\varepsilon A^{-1}BB^{\dag} & o(\varepsilon)\\
iB^{\dag} & i\varepsilon B^{\dag}A^{-1}B
\end{bmatrix}
,\\
M(\varepsilon)G  &  =\left[  GM(\varepsilon)\right]  ^{\dag}\\
&  =
\begin{bmatrix}
i\varepsilon BB^{\dag}A^{-1} & -iB\\
o(\varepsilon) & -i\varepsilon B^{\dag}A^{-1}B
\end{bmatrix}
,
\end{align}
which implies that
\begin{align}
&  i\varepsilon\left[  GM(\varepsilon)-M(\varepsilon)G\right] \nonumber\\
&  =i\varepsilon\left(
\begin{bmatrix}
-i\varepsilon A^{-1}BB^{\dag} & o(\varepsilon)\\
iB^{\dag} & i\varepsilon B^{\dag}A^{-1}B
\end{bmatrix}
-
\begin{bmatrix}
i\varepsilon BB^{\dag}A^{-1} & -iB\\
o(\varepsilon) & -i\varepsilon B^{\dag}A^{-1}B
\end{bmatrix}
\right) \\
&  =
\begin{bmatrix}
2\varepsilon^{2}\operatorname{Re}[A^{-1}BB^{\dag}] & -\varepsilon
B+o(\varepsilon^{2})\\
-\varepsilon B^{\dag}+o(\varepsilon^{2}) & -2\varepsilon^{2}B^{\dag}A^{-1}B
\end{bmatrix}
.
\end{align}
Also, observe that
\begin{align}
GM(\varepsilon)G  &  =
\begin{bmatrix}
o(1) & o(\varepsilon)\\
iB^{\dag} & o(1)
\end{bmatrix}
\begin{bmatrix}
0 & -iA^{-1}B\\
iB^{\dag}A^{-1} & 0
\end{bmatrix}
\\
&  =
\begin{bmatrix}
o(\varepsilon) & o(1)\\
o(1) & B^{\dag}A^{-1}B
\end{bmatrix}
,\\
G^{2}M(\varepsilon)  &  =G[GM(\varepsilon)]\\
&  =
\begin{bmatrix}
0 & -iA^{-1}B\\
iB^{\dag}A^{-1} & 0
\end{bmatrix}
\begin{bmatrix}
o(1) & o(\varepsilon)\\
iB^{\dag} & o(1)
\end{bmatrix}
\\
&  =
\begin{bmatrix}
A^{-1}BB^{\dag} & o(1)\\
o(1) & o(\varepsilon)
\end{bmatrix}
,\\
M(\varepsilon)G^{2}  &  =\left[  G^{2}M(\varepsilon)\right]  ^{\dag}\\
&  =
\begin{bmatrix}
BB^{\dag}A^{-1} & o(1)\\
o(1) & o(\varepsilon)
\end{bmatrix}
.
\end{align}
So then we find that
\begin{align}
&  \varepsilon^{2}\left[  GM(\varepsilon)G-\frac{1}{2}G^{2}M(\varepsilon
)-\frac{1}{2}M(\varepsilon)G^{2}\right] \nonumber\\
&  =\varepsilon^{2}\left(
\begin{bmatrix}
o(\varepsilon) & o(1)\\
o(1) & B^{\dag}A^{-1}B
\end{bmatrix}
-\frac{1}{2}
\begin{bmatrix}
A^{-1}BB^{\dag} & o(1)\\
o(1) & o(\varepsilon)
\end{bmatrix}
-\frac{1}{2}
\begin{bmatrix}
BB^{\dag}A^{-1} & o(1)\\
o(1) & o(\varepsilon)
\end{bmatrix}
\right) \\
&  =
\begin{bmatrix}
-\varepsilon^{2}\operatorname{Re}[A^{-1}BB^{\dag}]+o(\varepsilon^{3}) &
o(\varepsilon^{2})\\
o(\varepsilon^{2}) & \varepsilon^{2}B^{\dag}A^{-1}B+o(\varepsilon^{3})
\end{bmatrix}
.
\end{align}
So then
\begin{align}
&  \left(  I+i\varepsilon G-\frac{\varepsilon^{2}}{2}G^{2}\right)
M(\varepsilon)\left(  I-i\varepsilon G-\frac{\varepsilon^{2}}{2}G^{2}\right)
\nonumber\\
&  =M(\varepsilon)+i\varepsilon\left[  GM(\varepsilon)-M(\varepsilon)G\right]
\nonumber\\
&  \qquad+\varepsilon^{2}\left[  GM(\varepsilon)G-\frac{1}{2}G^{2}
M(\varepsilon)-\frac{1}{2}M(\varepsilon)G^{2}\right]  +o(\varepsilon^{2})\\
&  =
\begin{bmatrix}
A & \varepsilon B\\
\varepsilon B^{\dag} & \varepsilon^{2}C
\end{bmatrix}
+
\begin{bmatrix}
2\varepsilon^{2}\operatorname{Re}[A^{-1}BB^{\dag}] & -\varepsilon
B+o(\varepsilon^{2})\\
-\varepsilon B^{\dag}+o(\varepsilon^{2}) & -2\varepsilon^{2}B^{\dag}A^{-1}B
\end{bmatrix}
\nonumber\\
&  \qquad+
\begin{bmatrix}
-\varepsilon^{2}\operatorname{Re}[A^{-1}BB^{\dag}]+o(\varepsilon^{3}) &
o(\varepsilon^{2})\\
o(\varepsilon^{2}) & \varepsilon^{2}B^{\dag}A^{-1}B+o(\varepsilon^{3})
\end{bmatrix}
+o(\varepsilon^{2})\\
&  =
\begin{bmatrix}
A+\varepsilon^{2}\operatorname{Re}[A^{-1}BB^{\dag}] & 0\\
0 & \varepsilon^{2}\left(  C-B^{\dag}A^{-1}B\right)
\end{bmatrix}
+o(\varepsilon^{2})\\
&  =D(\varepsilon)+o(\varepsilon^{2}).
\end{align}
So we conclude that
\begin{equation}
e^{i\varepsilon G}M(\varepsilon)e^{-i\varepsilon G}=D(\varepsilon
)+o(\varepsilon^{2}),
\end{equation}
which in turn implies that
\begin{equation}
M(\varepsilon)=e^{-i\varepsilon G}D(\varepsilon)e^{i\varepsilon G}
+o(\varepsilon^{2}),
\end{equation}
from which we conclude the claim in \eqref{eq:sis-zhou-lemm}.
\end{proof}

\section{Basis-dependent and basis-independent formulas for SLD\ Fisher
information}

\label{app:basis-ind-dep-formulas-SLD-FI}

Here we review the proof of the following equality, mentioned in
\eqref{eq:basis-dependent-SLD-formula}--\eqref{eq:basis-independent-SLD-formula-extra},
which was reported in \cite{Saf18} and holds when $\Pi_{\rho_{\theta}}^{\perp
}(\partial_{\theta}\rho_{\theta})\Pi_{\rho_{\theta}}^{\perp}=0$:
\begin{align}
\frac{1}{2}I_{F}(\theta;\{\rho_{\theta}\}_{\theta})  &  =\sum_{j,k:\lambda
_{\theta}^{j}+\lambda_{\theta}^{k}>0}\frac{|\langle\psi_{\theta}^{j}
|(\partial_{\theta}\rho_{\theta})|\psi_{\theta}^{k}\rangle|^{2}}
{\lambda_{\theta}^{j}+\lambda_{\theta}^{k}}\\
&  =\langle\Gamma|((\partial_{\theta}\rho_{\theta})\otimes I)(\rho_{\theta
}\otimes I+I\otimes\rho_{\theta}^{T})^{-1}((\partial_{\theta}\rho_{\theta
})\otimes I)|\Gamma\rangle. \label{eq:app-fish-info-basis-ind-SLD}
\end{align}
Consider that
\begin{align}
&  \rho_{\theta}\otimes I+I\otimes\rho_{\theta}^{T}\nonumber\\
&  =\sum_{j:\lambda_{j}>0}\lambda_{\theta}^{j}|\psi_{\theta}^{j}\rangle
\langle\psi_{\theta}^{j}|\otimes I+I\otimes\left(  \sum_{k:\lambda_{k}
>0}\lambda_{\theta}^{k}|\psi_{\theta}^{k}\rangle\!\langle\psi_{\theta}
^{k}|\right)  ^{T}\\
&  =\sum_{j:\lambda_{j}>0}\lambda_{\theta}^{j}|\psi_{\theta}^{j}\rangle
\langle\psi_{\theta}^{j}|\otimes I+I\otimes\sum_{k:\lambda_{k}>0}
\lambda_{\theta}^{k}|\overline{\psi_{\theta}^{k}}\rangle\!\langle\overline
{\psi_{\theta}^{k}}|\\
&  =\sum_{j:\lambda_{j}>0,k}\lambda_{\theta}^{j}|\psi_{\theta}^{j}
\rangle\!\langle\psi_{\theta}^{j}|\otimes|\overline{\psi_{\theta}^{k}}
\rangle\!\langle\overline{\psi_{\theta}^{k}}|+\sum_{j,k:\lambda_{k}>0}
\lambda_{\theta}^{k}|\psi_{\theta}^{j}\rangle\!\langle\psi_{\theta}^{j}
|\otimes|\overline{\psi_{\theta}^{k}}\rangle\!\langle\overline{\psi_{\theta}
^{k}}|\\
&  =\sum_{j,k:\lambda_{j}+\lambda_{k}>0}\left(  \lambda_{\theta}^{j}
+\lambda_{\theta}^{k}\right)  |\psi_{\theta}^{j}\rangle\!\langle\psi_{\theta
}^{j}|\otimes|\overline{\psi_{\theta}^{k}}\rangle\!\langle\overline{\psi
_{\theta}^{k}}|,
\end{align}
where $|\overline{\psi_{\theta}^{k}}\rangle$ denotes the complex conjugate of
$|\psi_{\theta}^{k}\rangle$ with respect to the orthonormal basis
$\{|i\rangle\}_{i}$ for the unnormalized maximally entangled vector
$|\Gamma\rangle$. Then it follows that
\begin{equation}
\left(  \rho_{\theta}\otimes I+I\otimes\rho_{\theta}^{T}\right)  ^{-1}
=\sum_{j,k:\lambda_{j}+\lambda_{k}>0}\frac{1}{\lambda_{\theta}^{j}
+\lambda_{\theta}^{k}}|\psi_{\theta}^{j}\rangle\!\langle\psi_{\theta}
^{j}|\otimes|\overline{\psi_{\theta}^{k}}\rangle\!\langle\overline{\psi_{\theta
}^{k}}|,
\end{equation}
and we find that
\begin{align}
&  \langle\Gamma|((\partial_{\theta}\rho_{\theta})\otimes I)(\rho_{\theta
}\otimes I+I\otimes\rho_{\theta}^{T})^{-1}((\partial_{\theta}\rho_{\theta
})\otimes I)|\Gamma\rangle\nonumber\\
&  =\langle\Gamma|((\partial_{\theta}\rho_{\theta})\otimes I)\left(
\sum_{j,k:\lambda_{j}+\lambda_{k}>0}\frac{1}{\lambda_{\theta}^{j}
+\lambda_{\theta}^{k}}|\psi_{\theta}^{j}\rangle\!\langle\psi_{\theta}
^{j}|\otimes|\overline{\psi_{\theta}^{k}}\rangle\!\langle\overline{\psi_{\theta
}^{k}}|\right)  ((\partial_{\theta}\rho_{\theta})\otimes I)|\Gamma\rangle\\
&  =\sum_{j,k:\lambda_{j}+\lambda_{k}>0}\frac{1}{\lambda_{\theta}^{j}
+\lambda_{\theta}^{k}}\langle\Gamma|((\partial_{\theta}\rho_{\theta})\otimes
I)\left(  |\psi_{\theta}^{j}\rangle\!\langle\psi_{\theta}^{j}|\otimes
|\overline{\psi_{\theta}^{k}}\rangle\!\langle\overline{\psi_{\theta}^{k}
}|\right)  ((\partial_{\theta}\rho_{\theta})\otimes I)|\Gamma\rangle\\
&  =\sum_{j,k:\lambda_{j}+\lambda_{k}>0}\frac{1}{\lambda_{\theta}^{j}
+\lambda_{\theta}^{k}}\langle\Gamma|((\partial_{\theta}\rho_{\theta}
)|\psi_{\theta}^{j}\rangle\!\langle\psi_{\theta}^{j}|(\partial_{\theta}
\rho_{\theta})\otimes|\overline{\psi_{\theta}^{k}}\rangle\!\langle\overline
{\psi_{\theta}^{k}}|)|\Gamma\rangle\\
&  =\sum_{j,k:\lambda_{j}+\lambda_{k}>0}\frac{1}{\lambda_{\theta}^{j}
+\lambda_{\theta}^{k}}\langle\Gamma|((\partial_{\theta}\rho_{\theta}
)|\psi_{\theta}^{j}\rangle\!\langle\psi_{\theta}^{j}|(\partial_{\theta}
\rho_{\theta})|\psi_{\theta}^{k}\rangle\!\langle\psi_{\theta}^{k}|\otimes
I)|\Gamma\rangle\\
&  =\sum_{j,k:\lambda_{j}+\lambda_{k}>0}\frac{1}{\lambda_{\theta}^{j}
+\lambda_{\theta}^{k}}\operatorname{Tr}[(\partial_{\theta}\rho_{\theta}
)|\psi_{\theta}^{j}\rangle\!\langle\psi_{\theta}^{j}|(\partial_{\theta}
\rho_{\theta})|\psi_{\theta}^{k}\rangle\!\langle\psi_{\theta}^{k}|]\\
&  =\sum_{j,k:\lambda_{j}+\lambda_{k}>0}\frac{|\langle\psi_{\theta}
^{j}|(\partial_{\theta}\rho_{\theta})|\psi_{\theta}^{k}\rangle|^{2}}
{\lambda_{\theta}^{j}+\lambda_{\theta}^{k}},
\end{align}
where we used \eqref{eq:transpose-trick} and \eqref{eq:max-ent-partial-trace}.

Following the approach given in \cite{Saf18}, we can also see how the formula
in \eqref{eq:SLD-FI} arises from the differential equation in
\eqref{eq:SLD-op-def} and the formula in
\eqref{eq:app-fish-info-basis-ind-SLD}.\ Again, this development is only
relevant when the finiteness condition $\Pi_{\rho_{\theta}}^{\perp}
(\partial_{\theta}\rho_{\theta})\Pi_{\rho_{\theta}}^{\perp}=0$ holds. Consider
that the SLD\ operator $L_{\theta}$ is defined from the following differential
equation:
\begin{equation}
\partial_{\theta}\rho_{\theta}=\frac{1}{2}\left(  \rho_{\theta}L_{\theta
}+L_{\theta}\rho_{\theta}\right)  .
\end{equation}
Then this is equivalent to the following vectorized form:
\begin{align}
\left(  \partial_{\theta}\rho_{\theta}\otimes I\right)  |\Gamma\rangle &
=\left(  \frac{1}{2}\left(  \rho_{\theta}L_{\theta}+L_{\theta}\rho_{\theta
}\right)  \otimes I\right)  |\Gamma\rangle\\
&  =\frac{1}{2}\left(  \rho_{\theta}L_{\theta}\otimes I+L_{\theta}\rho
_{\theta}\otimes I\right)  |\Gamma\rangle\\
&  =\frac{1}{2}\left(  \rho_{\theta}L_{\theta}\otimes I+L_{\theta}\otimes
\rho_{\theta}^{T}\right)  |\Gamma\rangle\\
&  =\frac{1}{2}\left(  \rho_{\theta}\otimes I+I\otimes\rho_{\theta}
^{T}\right)  \left(  L_{\theta}\otimes I\right)  |\Gamma\rangle.
\end{align}
Consider that
\begin{equation}
(\Pi_{\rho_{\theta}}^{\perp}\otimes\Pi_{\rho_{\theta}^{T}}^{\perp})\left(
\partial_{\theta}\rho_{\theta}\otimes I\right)  |\Gamma\rangle=0
\end{equation}
because
\begin{align}
(\Pi_{\rho_{\theta}}^{\perp}\otimes\Pi_{\rho_{\theta}^{T}}^{\perp})\left(
\partial_{\theta}\rho_{\theta}\otimes I\right)  |\Gamma\rangle &  =\left(
\Pi_{\rho_{\theta}}^{\perp}(\partial_{\theta}\rho_{\theta})\otimes\Pi
_{\rho_{\theta}^{T}}^{\perp}\right)  |\Gamma\rangle\\
&  =\left(  \Pi_{\rho_{\theta}}^{\perp}(\partial_{\theta}\rho_{\theta}
)\otimes(\Pi_{\rho_{\theta}}^{\perp})^{T}\right)  |\Gamma\rangle\\
&  =\left(  \Pi_{\rho_{\theta}}^{\perp}(\partial_{\theta}\rho_{\theta}
)\Pi_{\rho_{\theta}}^{\perp}\otimes I\right)  |\Gamma\rangle\\
&  =0.
\end{align}
Thus, $\left(  \partial_{\theta}\rho_{\theta}\otimes I\right)  |\Gamma\rangle$
is only non-zero on the space onto which $I\otimes I-\Pi_{\rho_{\theta}
}^{\perp}\otimes\Pi_{\rho_{\theta}^{T}}^{\perp}$ projects, i.e.,
\begin{equation}
\left(  I\otimes I-\Pi_{\rho_{\theta}}^{\perp}\otimes\Pi_{\rho_{\theta}^{T}
}^{\perp}\right)  \left(  \partial_{\theta}\rho_{\theta}\otimes I\right)
|\Gamma\rangle=\left(  \partial_{\theta}\rho_{\theta}\otimes I\right)
|\Gamma\rangle. \label{eq:diff-vec-projection-app}
\end{equation}
Furthermore, note that the support of the operator $\rho_{\theta}\otimes
I+I\otimes\rho_{\theta}^{T}$ is given by
\begin{equation}
\Pi_{\rho_{\theta}}\otimes\Pi_{\rho_{\theta}^{T}}+\Pi_{\rho_{\theta}}^{\perp
}\otimes\Pi_{\rho_{\theta}^{T}}+\Pi_{\rho_{\theta}}\otimes\Pi_{\rho_{\theta
}^{T}}^{\perp}=I\otimes I-\Pi_{\rho_{\theta}}^{\perp}\otimes\Pi_{\rho_{\theta
}^{T}}^{\perp}.
\end{equation}
Thus, by applying the inverse of the operator $\frac{1}{2}\left(  \rho
_{\theta}\otimes I+I\otimes\rho_{\theta}^{T}\right)  $ on its support on both
sides, we find that
\begin{equation}
2\left(  \rho_{\theta}\otimes I+I\otimes\rho_{\theta}^{T}\right)  ^{-1}\left(
\partial_{\theta}\rho_{\theta}\otimes I\right)  |\Gamma\rangle=(I\otimes
I-\Pi_{\rho_{\theta}}^{\perp}\otimes\Pi_{\rho_{\theta}^{T}}^{\perp})\left(
L_{\theta}\otimes I\right)  |\Gamma\rangle. \label{eq:diff-eq-rewritten-app}
\end{equation}
Next, we use the fact that
\begin{equation}
\operatorname{Tr}[X^{\dag}Y]=\langle\Gamma|\left(  X\otimes I\right)  ^{\dag
}\left(  Y\otimes I\right)  |\Gamma\rangle,
\end{equation}
and we find that
\begin{align}
\operatorname{Tr}[L_{\theta}(\partial_{\theta}\rho_{\theta})]  &
=\langle\Gamma|\left(  \partial_{\theta}\rho_{\theta}\otimes I\right)  \left(
L_{\theta}\otimes I\right)  |\Gamma\rangle\\
&  =\langle\Gamma|\left(  \partial_{\theta}\rho_{\theta}\otimes I\right)
\left(  I\otimes I-\Pi_{\rho_{\theta}}^{\perp}\otimes\Pi_{\rho_{\theta}^{T}
}^{\perp}\right)  \left(  L_{\theta}\otimes I\right)  |\Gamma\rangle\\
&  =2\langle\Gamma|\left(  \partial_{\theta}\rho_{\theta}\otimes I\right)
\left(  \rho_{\theta}\otimes I+I\otimes\rho_{\theta}^{T}\right)  ^{-1}\left(
\partial_{\theta}\rho_{\theta}\otimes I\right)  |\Gamma\rangle,
\end{align}
where we used \eqref{eq:diff-vec-projection-app} and
\eqref{eq:diff-eq-rewritten-app}. This concludes the proof that
\begin{equation}
I_{F}(\theta;\{\rho_{A}^{\theta}\})=\operatorname{Tr}[L_{\theta}
(\partial_{\theta}\rho_{\theta})].
\end{equation}

\section{Physical consistency of SLD\ and RLD\ Fisher informations of quantum
states}

\label{app:phys-cons-SLD-RLD-Fish-states}

We begin by establishing the equivalence of the conditions in
\eqref{eq:finiteness-condition-SLD-Fish-states} and
\eqref{eq:finiteness-condition-SLD-Fish-alt}. Suppose that $\Pi_{\rho_{\theta
}}^{\perp}(\partial_{\theta}\rho_{\theta})\Pi_{\rho_{\theta}}^{\perp}=0$
holds. Then consider that
\begin{equation}
\Pi_{\rho_{\theta}}^{\perp}=\sum_{j:\lambda_{\theta}^{j}=0}|\psi_{\theta}
^{j}\rangle\!\langle\psi_{\theta}^{j}|,
\end{equation}
so that
\begin{align}
0  &  =\Pi_{\rho_{\theta}}^{\perp}(\partial_{\theta}\rho_{\theta})\Pi
_{\rho_{\theta}}^{\perp}\\
&  =\left(  \sum_{j:\lambda_{\theta}^{j}=0}|\psi_{\theta}^{j}\rangle
\langle\psi_{\theta}^{j}|\right)  (\partial_{\theta}\rho_{\theta})\left(
\sum_{k:\lambda_{\theta}^{k}=0}|\psi_{\theta}^{k}\rangle\!\langle\psi_{\theta
}^{k}|\right) \\
&  =\sum_{j:\lambda_{\theta}^{j}=0}\sum_{k:\lambda_{\theta}^{k}=0}
|\psi_{\theta}^{j}\rangle\!\langle\psi_{\theta}^{j}|(\partial_{\theta}
\rho_{\theta})|\psi_{\theta}^{k}\rangle\!\langle\psi_{\theta}^{k}|\\
&  =\sum_{j,k:\lambda_{\theta}^{j}+\lambda_{\theta}^{k}=0}\langle\psi_{\theta
}^{j}|(\partial_{\theta}\rho_{\theta})|\psi_{\theta}^{k}\rangle\ |\psi
_{\theta}^{j}\rangle\!\langle\psi_{\theta}^{k}|.
\end{align}
The last equality follows because $\lambda_{\theta}^{j}\geq0$ for all $j$, so
that $\lambda_{\theta}^{j}+\lambda_{\theta}^{k}=0$ is equivalent to
$\lambda_{\theta}^{j}=0\wedge\lambda_{\theta}^{k}=0$. Then it follows that
$\langle\psi_{\theta}^{j}|(\partial_{\theta}\rho_{\theta})|\psi_{\theta}
^{k}\rangle=0$ if $\lambda_{\theta}^{j}+\lambda_{\theta}^{k}=0$. This
establishes \eqref{eq:finiteness-condition-SLD-Fish-states} $\Rightarrow$
\eqref{eq:finiteness-condition-SLD-Fish-alt}. The opposite implication follows
from running the proof above backwards.

The equality in \eqref{eq:SLD-Fish-alt-formula-kernel-rho-theta}\ is
established in
\eqref{eq:alt-formula-SLD-step-1}--\eqref{eq:alt-formula-SLD-step-3}\ of the
proof given below.

\bigskip

\begin{proof}
[Proof of Proposition~\ref{prop:physical-consistency-SLD-Fish-states}]First,
it is helpful to write the spectral decomposition of $\rho_{\theta}$ as
follows:
\begin{equation}
\rho_{\theta}=\sum_{j\in\mathcal{S}}\lambda_{\theta}^{j}|\psi_{\theta}
^{j}\rangle\!\langle\psi_{\theta}^{j}|+\sum_{j\in\mathcal{K}}\lambda_{\theta
}^{j}|\psi_{\theta}^{j}\rangle\!\langle\psi_{\theta}^{j}|,
\end{equation}
where $\mathcal{S}$ is the set of indices for which $\lambda_{\theta}^{j}>0$
and $\mathcal{K}$ is the set of indices for which $\lambda_{\theta}^{j}=0$
($\mathcal{S}$ and $\mathcal{K}$ are meant to refer to support and kernel,
respectively).\ Let us define
\begin{equation}
\Pi_{\rho_{\theta}}:=\sum_{j\in\mathcal{S}}|\psi_{\theta}^{j}\rangle
\langle\psi_{\theta}^{j}|,\qquad\Pi_{\rho_{\theta}}^{\perp}:=I-\Pi
_{\rho_{\theta}}=\sum_{j\in\mathcal{K}}|\psi_{\theta}^{j}\rangle\!\langle
\psi_{\theta}^{j}|.
\end{equation}
Then
\begin{align}
\rho_{\theta}^{\varepsilon}  &  =\left(  1-\varepsilon\right)  \rho_{\theta
}+\varepsilon\pi\\
&  =\sum_{j\in\mathcal{S}}\left[  \left(  1-\varepsilon\right)  \lambda
_{\theta}^{j}\right]  |\psi_{\theta}^{j}\rangle\!\langle\psi_{\theta}^{j}
|+\frac{\varepsilon}{d}\sum_{j}|\psi_{\theta}^{j}\rangle\!\langle\psi_{\theta
}^{j}|\\
&  =\sum_{j\in\mathcal{S}}\left[  \left(  1-\varepsilon\right)  \lambda
_{\theta}^{j}+\frac{\varepsilon}{d}\right]  |\psi_{\theta}^{j}\rangle
\langle\psi_{\theta}^{j}|+\frac{\varepsilon}{d}\sum_{j\in\mathcal{K}}
|\psi_{\theta}^{j}\rangle\!\langle\psi_{\theta}^{j}|\\
&  =\sum_{j\in\mathcal{S}}\left[  \left(  1-\varepsilon\right)  \lambda
_{\theta}^{j}+\frac{\varepsilon}{d}\right]  |\psi_{\theta}^{j}\rangle
\langle\psi_{\theta}^{j}|+\frac{\varepsilon}{d}\Pi_{\rho_{\theta}}^{\perp}.
\end{align}
Let $\{\lambda_{\theta,\varepsilon}^{j}\}_{j}$ denote the eigenvalues of
$\rho_{\theta}^{\varepsilon}$, so that $\lambda_{\theta,\varepsilon}
^{j}=\left(  1-\varepsilon\right)  \lambda_{\theta}^{j}+\frac{\varepsilon}{d}$
for $j\in\mathcal{S}$ and $\lambda_{\theta,\varepsilon}^{j}=\frac{\varepsilon
}{d}$ for $j\in\mathcal{K}$. Observe that the state $\rho_{\theta
}^{\varepsilon}$ has full support. Also, observe that
\begin{equation}
\partial_{\theta}\rho_{\theta}^{\varepsilon}=\left(  1-\varepsilon\right)
\partial_{\theta}\rho_{\theta}.
\end{equation}
Plugging into the formula in \eqref{eq:SLD-Fish-info-formula}, we find that
\begin{align}
&  \frac{1}{2\left(  1-\varepsilon\right)  ^{2}}I_{F}(\theta;\{\rho_{\theta
}^{\varepsilon}\}_{\theta})\nonumber\\
&  =\frac{1}{\left(  1-\varepsilon\right)  ^{2}}\sum_{j,k}\frac{|\langle
\psi_{\theta}^{j}|(\partial_{\theta}\rho_{\theta}^{\varepsilon})|\psi_{\theta
}^{k}\rangle|^{2}}{\lambda_{\theta,\varepsilon}^{j}+\lambda_{\theta
,\varepsilon}^{k}}\\
&  =\sum_{j,k}\frac{|\langle\psi_{\theta}^{j}|(\partial_{\theta}\rho_{\theta
})|\psi_{\theta}^{k}\rangle|^{2}}{\lambda_{\theta,\varepsilon}^{j}
+\lambda_{\theta,\varepsilon}^{k}}\\
&  =\sum_{j\in\mathcal{S},k\in\mathcal{S}}\frac{|\langle\psi_{\theta}
^{j}|(\partial_{\theta}\rho_{\theta})|\psi_{\theta}^{k}\rangle|^{2}}
{\lambda_{\theta,\varepsilon}^{j}+\lambda_{\theta,\varepsilon}^{k}}+\sum
_{j\in\mathcal{S},k\in\mathcal{K}}\frac{|\langle\psi_{\theta}^{j}
|(\partial_{\theta}\rho_{\theta})|\psi_{\theta}^{k}\rangle|^{2}}
{\lambda_{\theta,\varepsilon}^{j}+\lambda_{\theta,\varepsilon}^{k}}\nonumber\\
&  \qquad+\sum_{j\in\mathcal{K},k\in\mathcal{S}}\frac{|\langle\psi_{\theta
}^{j}|(\partial_{\theta}\rho_{\theta})|\psi_{\theta}^{k}\rangle|^{2}}
{\lambda_{\theta,\varepsilon}^{j}+\lambda_{\theta,\varepsilon}^{k}}+\sum
_{j\in\mathcal{K},k\in\mathcal{K}}\frac{|\langle\psi_{\theta}^{j}
|(\partial_{\theta}\rho_{\theta})|\psi_{\theta}^{k}\rangle|^{2}}
{\lambda_{\theta,\varepsilon}^{j}+\lambda_{\theta,\varepsilon}^{k}}
\end{align}
Let us consider the terms one at a time, starting with the first one:
\begin{equation}
\sum_{j\in\mathcal{S},k\in\mathcal{S}}\frac{|\langle\psi_{\theta}
^{j}|(\partial_{\theta}\rho_{\theta})|\psi_{\theta}^{k}\rangle|^{2}}
{\lambda_{\theta,\varepsilon}^{j}+\lambda_{\theta,\varepsilon}^{k}}=\sum
_{j\in\mathcal{S},k\in\mathcal{S}}\frac{|\langle\psi_{\theta}^{j}
|(\partial_{\theta}\rho_{\theta})|\psi_{\theta}^{k}\rangle|^{2}}{\left(
1-\varepsilon\right)  \left[  \lambda_{\theta}^{j}+\lambda_{\theta}
^{k}\right]  +\frac{2\varepsilon}{d}}
\end{equation}
The second term simplifies as follows:
\begin{align}
&  \sum_{j\in\mathcal{S},k\in\mathcal{K}}\frac{|\langle\psi_{\theta}
^{j}|(\partial_{\theta}\rho_{\theta})|\psi_{\theta}^{k}\rangle|^{2}}
{\lambda_{\theta,\varepsilon}^{j}+\lambda_{\theta,\varepsilon}^{k}}\nonumber\\
&  =\sum_{j\in\mathcal{S},k\in\mathcal{K}}\frac{|\langle\psi_{\theta}
^{j}|(\partial_{\theta}\rho_{\theta})|\psi_{\theta}^{k}\rangle|^{2}}{\left(
1-\varepsilon\right)  \lambda_{\theta}^{j}+\frac{2\varepsilon}{d}}\\
&  =\sum_{j\in\mathcal{S},k\in\mathcal{K}}\frac{\langle\psi_{\theta}
^{j}|(\partial_{\theta}\rho_{\theta})|\psi_{\theta}^{k}\rangle\!\langle
\psi_{\theta}^{k}|(\partial_{\theta}\rho_{\theta})|\psi_{\theta}^{j}\rangle
}{\left(  1-\varepsilon\right)  \lambda_{\theta}^{j}+\frac{2\varepsilon}{d}}\\
&  =\sum_{j\in\mathcal{S}}\frac{1}{\left(  1-\varepsilon\right)
\lambda_{\theta}^{j}+\frac{2\varepsilon}{d}}\langle\psi_{\theta}^{j}
|(\partial_{\theta}\rho_{\theta})\left(  \sum_{k\in\mathcal{K}}|\psi_{\theta
}^{k}\rangle\!\langle\psi_{\theta}^{k}|\right)  (\partial_{\theta}\rho_{\theta
})|\psi_{\theta}^{j}\rangle\\
&  =\sum_{j\in\mathcal{S}}\frac{\langle\psi_{\theta}^{j}|(\partial_{\theta
}\rho_{\theta})\Pi_{\rho_{\theta}}^{\perp}(\partial_{\theta}\rho_{\theta
})|\psi_{\theta}^{j}\rangle}{\left(  1-\varepsilon\right)  \lambda_{\theta
}^{j}+\frac{2\varepsilon}{d}}.
\end{align}
Similarly, due to symmetry, we find the following for the third term:
\begin{equation}
\sum_{j\in\mathcal{K},k\in\mathcal{S}}\frac{|\langle\psi_{\theta}
^{j}|(\partial_{\theta}\rho_{\theta})|\psi_{\theta}^{k}\rangle|^{2}}
{\lambda_{\theta,\varepsilon}^{j}+\lambda_{\theta,\varepsilon}^{k}}=\sum
_{j\in\mathcal{S}}\frac{\langle\psi_{\theta}^{j}|(\partial_{\theta}
\rho_{\theta})\Pi_{\rho_{\theta}}^{\perp}(\partial_{\theta}\rho_{\theta}
)|\psi_{\theta}^{j}\rangle}{\left(  1-\varepsilon\right)  \lambda_{\theta}
^{j}+\frac{2\varepsilon}{d}}.
\end{equation}
For the last term, we find that
\begin{align}
&  \sum_{j\in\mathcal{K},k\in\mathcal{K}}\frac{|\langle\psi_{\theta}
^{j}|(\partial_{\theta}\rho_{\theta})|\psi_{\theta}^{k}\rangle|^{2}}
{\lambda_{\theta,\varepsilon}^{j}+\lambda_{\theta,\varepsilon}^{k}}\nonumber\\
&  =\sum_{j\in\mathcal{K},k\in\mathcal{K}}\frac{|\langle\psi_{\theta}
^{j}|(\partial_{\theta}\rho_{\theta})|\psi_{\theta}^{k}\rangle|^{2}}
{\frac{2\varepsilon}{d}}\\
&  =\frac{d}{2\varepsilon}\sum_{j\in\mathcal{K},k\in\mathcal{K}}|\langle
\psi_{\theta}^{j}|(\partial_{\theta}\rho_{\theta})|\psi_{\theta}^{k}
\rangle|^{2}\\
&  =\frac{d}{2\varepsilon}\sum_{j\in\mathcal{K},k\in\mathcal{K}}\langle
\psi_{\theta}^{j}|(\partial_{\theta}\rho_{\theta})|\psi_{\theta}^{k}
\rangle\!\langle\psi_{\theta}^{k}|(\partial_{\theta}\rho_{\theta})|\psi_{\theta
}^{j}\rangle\\
&  =\frac{d}{2\varepsilon}\sum_{j\in\mathcal{K},k\in\mathcal{K}}
\operatorname{Tr}[|\psi_{\theta}^{j}\rangle\!\langle\psi_{\theta}^{j}
|(\partial_{\theta}\rho_{\theta})|\psi_{\theta}^{k}\rangle\!\langle\psi_{\theta
}^{k}|(\partial_{\theta}\rho_{\theta})]\\
&  =\frac{d}{2\varepsilon}\operatorname{Tr}\left[  \left(  \sum_{j\in
\mathcal{K}}|\psi_{\theta}^{j}\rangle\!\langle\psi_{\theta}^{j}|\right)
(\partial_{\theta}\rho_{\theta})\left(  \sum_{k\in\mathcal{K}}|\psi_{\theta
}^{k}\rangle\!\langle\psi_{\theta}^{k}|\right)  (\partial_{\theta}\rho_{\theta
})\right] \\
&  =\frac{d}{2\varepsilon}\operatorname{Tr}\left[  \Pi_{\rho_{\theta}}^{\perp
}(\partial_{\theta}\rho_{\theta})\Pi_{\rho_{\theta}}^{\perp}(\partial_{\theta
}\rho_{\theta})\right] \\
&  =\frac{d}{2\varepsilon}\operatorname{Tr}\left[  \left(  \Pi_{\rho_{\theta}
}^{\perp}(\partial_{\theta}\rho_{\theta})\Pi_{\rho_{\theta}}^{\perp}\right)
^{2}\right] \\
&  =\frac{d}{2\varepsilon}\left\Vert \Pi_{\rho_{\theta}}^{\perp}
(\partial_{\theta}\rho_{\theta})\Pi_{\rho_{\theta}}^{\perp}\right\Vert
_{2}^{2},
\end{align}
where $\left\Vert A\right\Vert _{2}:=\sqrt{\operatorname{Tr}[A^{\dag}A]}$ is
the Hilbert--Schmidt norm of an operator $A$. Putting everything together, we
find that
\begin{multline}
I_{F}(\theta;\{\rho_{\theta}^{\varepsilon}\}_{\theta})=2\left(  1-\varepsilon
\right)  ^{2}\sum_{j\in\mathcal{S},k\in\mathcal{S}}\frac{|\langle\psi_{\theta
}^{j}|(\partial_{\theta}\rho_{\theta})|\psi_{\theta}^{k}\rangle|^{2}}{\left(
1-\varepsilon\right)  \left[  \lambda_{\theta}^{j}+\lambda_{\theta}
^{k}\right]  +\frac{2\varepsilon}{d}}\\
+4\left(  1-\varepsilon\right)  ^{2}\sum_{j\in\mathcal{S}}\frac{\langle
\psi_{\theta}^{j}|(\partial_{\theta}\rho_{\theta})\Pi_{\rho_{\theta}}^{\perp
}(\partial_{\theta}\rho_{\theta})|\psi_{\theta}^{j}\rangle}{\left(
1-\varepsilon\right)  \lambda_{\theta}^{j}+\frac{2\varepsilon}{d}}
+\frac{d\left(  1-\varepsilon\right)  ^{2}}{\varepsilon}\left\Vert \Pi
_{\rho_{\theta}}^{\perp}(\partial_{\theta}\rho_{\theta})\Pi_{\rho_{\theta}
}^{\perp}\right\Vert _{2}^{2}.
\end{multline}
Now consider that
\begin{equation}
\left\Vert \Pi_{\rho_{\theta}}^{\perp}(\partial_{\theta}\rho_{\theta}
)\Pi_{\rho_{\theta}}^{\perp}\right\Vert _{2}^{2}=0\qquad\Longleftrightarrow
\qquad\Pi_{\rho_{\theta}}^{\perp}(\partial_{\theta}\rho_{\theta})\Pi
_{\rho_{\theta}}^{\perp}=0.
\end{equation}
If this condition holds, then the last term vanishes and we find that
\begin{equation}
\lim_{\varepsilon\rightarrow0}I_{F}(\theta;\{\rho_{\theta}^{\varepsilon
}\}_{\theta})=2\sum_{j\in\mathcal{S},k\in\mathcal{S}}\frac{|\langle
\psi_{\theta}^{j}|(\partial_{\theta}\rho_{\theta})|\psi_{\theta}^{k}
\rangle|^{2}}{\lambda_{\theta}^{j}+\lambda_{\theta}^{k}}+4\sum_{j\in
\mathcal{S}}\frac{\langle\psi_{\theta}^{j}|(\partial_{\theta}\rho_{\theta}
)\Pi_{\rho_{\theta}}^{\perp}(\partial_{\theta}\rho_{\theta})|\psi_{\theta}
^{j}\rangle}{\lambda_{\theta}^{j}}.
\end{equation}
However, if this condition does not hold, then $\left\Vert \Pi_{\rho_{\theta}
}^{\perp}(\partial_{\theta}\rho_{\theta})\Pi_{\rho_{\theta}}^{\perp
}\right\Vert _{2}>0$ and the following limit holds
\begin{equation}
\lim_{\varepsilon\rightarrow0}I_{F}(\theta;\{\rho_{\theta}^{\varepsilon
}\}_{\theta})=+\infty.
\end{equation}

Now consider that
\begin{align}
&  2\sum_{j,k:\lambda_{j}+\lambda_{k}>0}\frac{|\langle\psi_{\theta}
^{j}|(\partial_{\theta}\rho_{\theta})|\psi_{\theta}^{k}\rangle|^{2}}
{\lambda_{\theta}^{j}+\lambda_{\theta}^{k}}\nonumber\\
&  =2\sum_{j,k:\left(  j\notin\mathcal{K}\wedge k\notin\mathcal{K}\right)
}\frac{|\langle\psi_{\theta}^{j}|(\partial_{\theta}\rho_{\theta})|\psi
_{\theta}^{k}\rangle|^{2}}{\lambda_{\theta}^{j}+\lambda_{\theta}^{k}
}\label{eq:alt-formula-SLD-step-1}\\
&  =2\left[  \sum_{j\in\mathcal{S},k\in\mathcal{S}}\frac{|\langle\psi_{\theta
}^{j}|(\partial_{\theta}\rho_{\theta})|\psi_{\theta}^{k}\rangle|^{2}}
{\lambda_{\theta}^{j}+\lambda_{\theta}^{k}}+\sum_{j\in\mathcal{S}
,k\in\mathcal{K}}\frac{|\langle\psi_{\theta}^{j}|(\partial_{\theta}
\rho_{\theta})|\psi_{\theta}^{k}\rangle|^{2}}{\lambda_{\theta}^{j}
+\lambda_{\theta}^{k}}+\sum_{j\in\mathcal{K},k\in\mathcal{S}}\frac
{|\langle\psi_{\theta}^{j}|(\partial_{\theta}\rho_{\theta})|\psi_{\theta}
^{k}\rangle|^{2}}{\lambda_{\theta}^{j}+\lambda_{\theta}^{k}}\right] \\
&  =2\left[  \sum_{j\in\mathcal{S},k\in\mathcal{S}}\frac{|\langle\psi_{\theta
}^{j}|(\partial_{\theta}\rho_{\theta})|\psi_{\theta}^{k}\rangle|^{2}}
{\lambda_{\theta}^{j}+\lambda_{\theta}^{k}}+2\sum_{j\in\mathcal{S}}
\frac{\langle\psi_{\theta}^{j}|(\partial_{\theta}\rho_{\theta})\Pi
_{\rho_{\theta}}^{\perp}(\partial_{\theta}\rho_{\theta})|\psi_{\theta}
^{j}\rangle}{\lambda_{\theta}^{j}}\right]  , \label{eq:alt-formula-SLD-step-3}
\end{align}
where we arrived at the last line by applying the previous reasoning. Thus, we
find that if $\Pi_{\rho_{\theta}}^{\perp}(\partial_{\theta}\rho_{\theta}
)\Pi_{\rho_{\theta}}^{\perp}=0$, then
\begin{equation}
\lim_{\varepsilon\rightarrow0}I_{F}(\theta;\{\rho_{\theta}^{\varepsilon
}\}_{\theta})=2\sum_{j,k:\lambda_{j}+\lambda_{k}>0}\frac{|\langle\psi_{\theta
}^{j}|(\partial_{\theta}\rho_{\theta})|\psi_{\theta}^{k}\rangle|^{2}}
{\lambda_{\theta}^{j}+\lambda_{\theta}^{k}}.
\end{equation}
This concludes the proof.
\end{proof}

\bigskip

\begin{proof}
[Proof of Proposition~\ref{prop:physical-consistency-RLD-Fish-states}] Following the notation from the previous proof, it follows that
\begin{align}
(\partial_{\theta}\rho_{\theta}^{\varepsilon})^{2}  &  =\left(  1-\varepsilon
\right)  ^{2}(\partial_{\theta}\rho_{\theta})^{2}\\
(\rho_{\theta}^{\varepsilon})^{-1}  &  =\sum_{j\in\mathcal{S}}\frac{1}{\left(
1-\varepsilon\right)  \lambda_{\theta}^{j}+\frac{\varepsilon}{d}}|\psi
_{\theta}^{j}\rangle\!\langle\psi_{\theta}^{j}|+\frac{d}{\varepsilon}\Pi
_{\rho_{\theta}}^{\perp},
\end{align}
so that
\begin{align}
\widehat{I}_{F}(\theta;\{\rho_{\theta}^{\varepsilon}\}_{\theta})  &
=\operatorname{Tr}[(\partial_{\theta}\rho_{\theta}^{\varepsilon})^{2}
(\rho_{\theta}^{\varepsilon})^{-1}]\\
&  =\left(  1-\varepsilon\right)  ^{2}\operatorname{Tr}\left[  (\partial
_{\theta}\rho_{\theta})^{2}\left(  \sum_{j\in\mathcal{S}}
\frac{1}{\left(  1-\varepsilon\right)  \lambda_{\theta}^{j}+\frac{\varepsilon
}{d}}|\psi_{\theta}^{j}\rangle\!\langle\psi_{\theta}^{j}|\right)  \right]
\nonumber\\
&  \qquad+\frac{d\left(  1-\varepsilon\right)  ^{2}}{\varepsilon
}\operatorname{Tr}[(\partial_{\theta}\rho_{\theta})^{2}\Pi_{\rho_{\theta}
}^{\perp}].
\end{align}
The condition $\operatorname{Tr}[(\partial_{\theta}\rho_{\theta})^{2}\Pi
_{\rho_{\theta}}^{\perp}]=0$ is equivalent to the condition $(\partial
_{\theta}\rho_{\theta})^{2}\Pi_{\rho_{\theta}}^{\perp}=0$ because both
$(\partial_{\theta}\rho_{\theta})^{2}$ and $\Pi_{\rho_{\theta}}^{\perp}$ are
positive semi-definite. The condition $(\partial_{\theta}\rho_{\theta})^{2}
\Pi_{\rho_{\theta}}^{\perp}=0$ is equivalent to the condition
$\operatorname{supp}((\partial_{\theta}\rho_{\theta})^{2})\subseteq
\operatorname{supp}(\rho_{\theta})$. Since $\operatorname{supp}((\partial
_{\theta}\rho_{\theta})^{2})=\operatorname{supp}(\partial_{\theta}\rho
_{\theta})$, this condition is in turn equivalent to $\operatorname{supp}
(\partial_{\theta}\rho_{\theta})\subseteq\operatorname{supp}(\rho_{\theta})$.
Thus,
\begin{equation}
\operatorname{supp}(\partial_{\theta}\rho_{\theta})\subseteq
\operatorname{supp}(\rho_{\theta})\qquad\Longleftrightarrow\qquad
\operatorname{Tr}[(\partial_{\theta}\rho_{\theta})^{2}\Pi_{\rho_{\theta}
}^{\perp}]=0,
\end{equation}
and we find that if $\operatorname{supp}(\partial_{\theta}\rho_{\theta
})\subseteq\operatorname{supp}(\rho_{\theta})$, then
\begin{align}
\lim_{\varepsilon\rightarrow0}\widehat{I}_{F}(\theta;\{\rho_{\theta
}^{\varepsilon}\}_{\theta})  &  =\lim_{\varepsilon\rightarrow0}\left(
1-\varepsilon\right)  ^{2}\operatorname{Tr}\left[  (\partial_{\theta}
\rho_{\theta})^{2}\left(  \sum_{j\in\mathcal{S}}\frac{1}{\left(
1-\varepsilon\right)  \lambda_{\theta}^{j}+\frac{\varepsilon}{d}}|\psi
_{\theta}^{j}\rangle\!\langle\psi_{\theta}^{j}|\right)  \right] \\
&  =\operatorname{Tr}\left[  (\partial_{\theta}\rho_{\theta})^{2}\left(
\sum_{j\in\mathcal{S}}\frac{1}{\lambda_{\theta}^{j}}|\psi_{\theta}^{j}
\rangle\!\langle\psi_{\theta}^{j}|\right)  \right] \\
&  =\operatorname{Tr}[(\partial_{\theta}\rho_{\theta})^{2}\rho_{\theta}^{-1}].
\end{align}
On the other hand, if $\operatorname{supp}(\partial_{\theta}\rho_{\theta
})\not \subseteq \operatorname{supp}(\rho_{\theta})$, then $\operatorname{Tr}
[(\partial_{\theta}\rho_{\theta})^{2}\Pi_{\rho_{\theta}}^{\perp}]>0$, and
$\lim_{\varepsilon\rightarrow0}\widehat{I}_{F}(\theta;\{\rho_{\theta
}^{\varepsilon}\}_{\theta})=+\infty$.
\end{proof}

\subsection{Pure-state family examples}

\label{sec:pure-state-fam-examps}

\begin{proposition}
\label{prop:SLD-Fish-pure-state}Let $\{|\phi_{\theta}\rangle\!\langle
\phi_{\theta}|\}_{\theta}$ be a differentiable family of pure states. Then the
SLD\ Fisher information is as follows:
\begin{equation}
I_{F}(\theta;\{|\phi_{\theta}\rangle\!\langle\phi_{\theta}|\}_{\theta})=4\left[
\langle\partial_{\theta}\phi_{\theta}|\partial_{\theta}\phi_{\theta}
\rangle-\left\vert \langle\partial_{\theta}\phi_{\theta}|\phi_{\theta}
\rangle\right\vert ^{2}\right]  .
\end{equation}

\end{proposition}

\begin{proof}
First, observe that
\begin{equation}
\partial_{\theta}(|\phi_{\theta}\rangle\!\langle\phi_{\theta}|)=|\partial
_{\theta}\phi_{\theta}\rangle\!\langle\phi_{\theta}|+|\phi_{\theta}
\rangle\!\langle\partial_{\theta}\phi_{\theta}|,
\end{equation}
which, when combined with $\operatorname{Tr}[\partial_{\theta}(|\phi_{\theta
}\rangle\!\langle\phi_{\theta}|)]=\partial_{\theta}(\operatorname{Tr}
[|\phi_{\theta}\rangle\!\langle\phi_{\theta}|])=0$, implies that
\begin{equation}
0=\langle\phi_{\theta}|\partial_{\theta}\phi_{\theta}\rangle+\langle
\partial_{\theta}\phi_{\theta}|\phi_{\theta}\rangle=2\operatorname{Re}
[\langle\partial_{\theta}\phi_{\theta}|\phi_{\theta}\rangle].
\label{eq:real-part-pure-states-zero}
\end{equation}

Now consider that the finiteness condition $\Pi_{\phi_{\theta}}^{\perp
}(\partial_{\theta}|\phi_{\theta}\rangle\!\langle\phi_{\theta}|)\Pi
_{\phi_{\theta}}^{\perp}=0$ holds for all differentiable pure-state families,
where $\Pi_{\phi_{\theta}}^{\perp}=I-|\phi_{\theta}\rangle\!\langle\phi_{\theta
}|$. This is because $|\phi_{\theta}\rangle\!\langle\phi_{\theta}|\Pi
_{\phi_{\theta}}^{\perp}=\Pi_{\phi_{\theta}}^{\perp}|\phi_{\theta}
\rangle\!\langle\phi_{\theta}|=0$, so that
\begin{align}
\Pi_{\phi_{\theta}}^{\perp}(\partial_{\theta}|\phi_{\theta}\rangle\!\langle
\phi_{\theta}|)\Pi_{\phi_{\theta}}^{\perp}  &  =\Pi_{\phi_{\theta}}^{\perp
}(|\partial_{\theta}\phi_{\theta}\rangle\!\langle\phi_{\theta}|+|\phi_{\theta
}\rangle\!\langle\partial_{\theta}\phi_{\theta}|)\Pi_{\phi_{\theta}}^{\perp}\\
&  =\Pi_{\phi_{\theta}}^{\perp}|\partial_{\theta}\phi_{\theta}\rangle
\langle\phi_{\theta}|\Pi_{\phi_{\theta}}^{\perp}+\Pi_{\phi_{\theta}}^{\perp
}|\phi_{\theta}\rangle\!\langle\partial_{\theta}\phi_{\theta}|\Pi_{\phi_{\theta
}}^{\perp}\\
&  =0.
\end{align}
Then we can apply the general expression for the SLD\ Fisher information in
\eqref{eq:SLD-Fish-alt-formula-kernel-rho-theta}:
\begin{align}
&  I_{F}(\theta;\{|\phi_{\theta}\rangle\!\langle\phi_{\theta}|\}_{\theta
})\nonumber\\
&  =|\langle\phi_{\theta}|(\partial_{\theta}(|\phi_{\theta}\rangle\!\langle
\phi_{\theta}|))|\phi_{\theta}\rangle|^{2}+4\langle\phi_{\theta}
|(\partial_{\theta}(|\phi_{\theta}\rangle\!\langle\phi_{\theta}|))\Pi
_{\rho_{\theta}}^{\perp}(\partial_{\theta}(|\phi_{\theta}\rangle\!\langle
\phi_{\theta}|))|\phi_{\theta}\rangle\\
&  =|\langle\phi_{\theta}|(\partial_{\theta}(|\phi_{\theta}\rangle\!\langle
\phi_{\theta}|))|\phi_{\theta}\rangle|^{2}+4\langle\phi_{\theta}
|(\partial_{\theta}(|\phi_{\theta}\rangle\!\langle\phi_{\theta}|))\left(
I-|\phi_{\theta}\rangle\!\langle\phi_{\theta}|\right)  (\partial_{\theta}
(|\phi_{\theta}\rangle\!\langle\phi_{\theta}|))|\phi_{\theta}\rangle\\
&  =4\langle\phi_{\theta}|(\partial_{\theta}(|\phi_{\theta}\rangle\!\langle
\phi_{\theta}|))^{2}|\phi_{\theta}\rangle-3|\langle\phi_{\theta}
|(\partial_{\theta}(|\phi_{\theta}\rangle\!\langle\phi_{\theta}|))|\phi_{\theta
}\rangle|^{2}.
\end{align}
Then we find that
\begin{align}
\langle\phi_{\theta}|(\partial_{\theta}(|\phi_{\theta}\rangle\!\langle
\phi_{\theta}|))|\phi_{\theta}\rangle &  =\langle\phi_{\theta}|(|\partial
_{\theta}\phi_{\theta}\rangle\!\langle\phi_{\theta}|+|\phi_{\theta}
\rangle\!\langle\partial_{\theta}\phi_{\theta}|)|\phi_{\theta}\rangle\\
&  =\langle\phi_{\theta}|\partial_{\theta}\phi_{\theta}\rangle+\langle
\partial_{\theta}\phi_{\theta}|\phi_{\theta}\rangle\\
&  =0,
\end{align}
where we applied \eqref{eq:real-part-pure-states-zero}\ to get the last line.
This implies that
\begin{equation}
I_{F}(\theta;\{|\phi_{\theta}\rangle\!\langle\phi_{\theta}|\}_{\theta}
)=4\langle\phi_{\theta}|(\partial_{\theta}(|\phi_{\theta}\rangle\!\langle
\phi_{\theta}|))^{2}|\phi_{\theta}\rangle.
\label{eq:fisher-SLD-pure-states-almost-done-app}
\end{equation}
Now consider that
\begin{align}
&  \langle\phi_{\theta}|(\partial_{\theta}(|\phi_{\theta}\rangle\!\langle
\phi_{\theta}|))^{2}|\phi_{\theta}\rangle\nonumber\\
&  =\langle\phi_{\theta}|(|\partial_{\theta}\phi_{\theta}\rangle\!\langle
\phi_{\theta}|+|\phi_{\theta}\rangle\!\langle\partial_{\theta}\phi_{\theta
}|)(|\partial_{\theta}\phi_{\theta}\rangle\!\langle\phi_{\theta}|+|\phi_{\theta
}\rangle\!\langle\partial_{\theta}\phi_{\theta}|)|\phi_{\theta}\rangle\\
&  =\langle\phi_{\theta}|\partial_{\theta}\phi_{\theta}\rangle\!\langle
\phi_{\theta}|\partial_{\theta}\phi_{\theta}\rangle\!\langle\phi_{\theta}
|\phi_{\theta}\rangle+\langle\phi_{\theta}|\partial_{\theta}\phi_{\theta
}\rangle\!\langle\phi_{\theta}|\phi_{\theta}\rangle\!\langle\partial_{\theta}
\phi_{\theta}|\phi_{\theta}\rangle\nonumber\\
&  \qquad+\langle\phi_{\theta}|\phi_{\theta}\rangle\!\langle\partial_{\theta
}\phi_{\theta}|\partial_{\theta}\phi_{\theta}\rangle\!\langle\phi_{\theta}
|\phi_{\theta}\rangle+\langle\phi_{\theta}|\phi_{\theta}\rangle\!\langle
\partial_{\theta}\phi_{\theta}|\phi_{\theta}\rangle\!\langle\partial_{\theta
}\phi_{\theta}|\phi_{\theta}\rangle\\
&  =\left(  \langle\phi_{\theta}|\partial_{\theta}\phi_{\theta}\rangle\right)
^{2}+\left\vert \langle\phi_{\theta}|\partial_{\theta}\phi_{\theta}
\rangle\right\vert ^{2}+\langle\partial_{\theta}\phi_{\theta}|\partial
_{\theta}\phi_{\theta}\rangle+\left(  \langle\partial_{\theta}\phi_{\theta
}|\phi_{\theta}\rangle\right)  ^{2}\\
&  =\left(  \langle\phi_{\theta}|\partial_{\theta}\phi_{\theta}\rangle\right)
^{2}+2\left\vert \langle\phi_{\theta}|\partial_{\theta}\phi_{\theta}
\rangle\right\vert ^{2}+\left(  \langle\partial_{\theta}\phi_{\theta}
|\phi_{\theta}\rangle\right)  ^{2}+\langle\partial_{\theta}\phi_{\theta
}|\partial_{\theta}\phi_{\theta}\rangle-\left\vert \langle\phi_{\theta
}|\partial_{\theta}\phi_{\theta}\rangle\right\vert ^{2}\\
&  =\left\vert \langle\phi_{\theta}|\partial_{\theta}\phi_{\theta}
\rangle+\langle\partial_{\theta}\phi_{\theta}|\phi_{\theta}\rangle\right\vert
^{2}+\langle\partial_{\theta}\phi_{\theta}|\partial_{\theta}\phi_{\theta
}\rangle-\left\vert \langle\phi_{\theta}|\partial_{\theta}\phi_{\theta}
\rangle\right\vert ^{2}\\
&  =\langle\partial_{\theta}\phi_{\theta}|\partial_{\theta}\phi_{\theta
}\rangle-\left\vert \langle\phi_{\theta}|\partial_{\theta}\phi_{\theta}
\rangle\right\vert ^{2},
\end{align}
where we again applied \eqref{eq:real-part-pure-states-zero}\ to get the last
line. Substituting into \eqref{eq:fisher-SLD-pure-states-almost-done-app}, we
arrive at the statement of the proposition.
\end{proof}

\begin{proposition}
Let $\{|\phi_{\theta}\rangle\!\langle\phi_{\theta}|\}_{\theta}$ be a
differentiable family of pure states. If the family is constant, so that
$|\varphi_{\theta}\rangle=|\varphi\rangle$ for all $\theta$, then the
RLD\ Fisher information is equal to zero. Otherwise, the RLD\ Fisher
information is infinite.
\end{proposition}

\begin{proof}
The RLD\ Fisher information is finite if and only if the finiteness condition
in \eqref{eq:support-condition-RLD}\ is satisfied. This condition is
equivalent to the following: $0=\operatorname{Tr}[\Pi_{\phi_{\theta}}^{\perp
}(\partial_{\theta}(|\phi_{\theta}\rangle\!\langle\phi_{\theta}|))^{2}]$. Now
consider that
\begin{align}
&  \operatorname{Tr}[\Pi_{\phi_{\theta}}^{\perp}(\partial_{\theta}
(|\phi_{\theta}\rangle\!\langle\phi_{\theta}|))^{2}]\nonumber\\
&  =\operatorname{Tr}[\Pi_{\phi_{\theta}}^{\perp}(|\partial_{\theta}
\phi_{\theta}\rangle\!\langle\phi_{\theta}|+|\phi_{\theta}\rangle\!\langle
\partial_{\theta}\phi_{\theta}|)((|\partial_{\theta}\phi_{\theta}
\rangle\!\langle\phi_{\theta}|+|\phi_{\theta}\rangle\!\langle\partial_{\theta}
\phi_{\theta}|))]\\
&  =\operatorname{Tr}[\Pi_{\phi_{\theta}}^{\perp}|\partial_{\theta}
\phi_{\theta}\rangle\!\langle\phi_{\theta}|\partial_{\theta}\phi_{\theta}
\rangle\!\langle\phi_{\theta}|]+\operatorname{Tr}[\Pi_{\phi_{\theta}}^{\perp
}|\partial_{\theta}\phi_{\theta}\rangle\!\langle\phi_{\theta}|\phi_{\theta
}\rangle\!\langle\partial_{\theta}\phi_{\theta}|]\nonumber\\
&  \qquad+\operatorname{Tr}[\Pi_{\phi_{\theta}}^{\perp}|\phi_{\theta}
\rangle\!\langle\partial_{\theta}\phi_{\theta}|\partial_{\theta}\phi_{\theta
}\rangle\!\langle\phi_{\theta}|]+\operatorname{Tr}[\Pi_{\phi_{\theta}}^{\perp
}|\phi_{\theta}\rangle\!\langle\partial_{\theta}\phi_{\theta}|\phi_{\theta
}\rangle\!\langle\partial_{\theta}\phi_{\theta}|]\\
&  =\langle\partial_{\theta}\phi_{\theta}|\Pi_{\phi_{\theta}}^{\perp}
|\partial_{\theta}\phi_{\theta}\rangle\\
&  =\langle\partial_{\theta}\phi_{\theta}|\partial_{\theta}\phi_{\theta
}\rangle-\left\vert \langle\partial_{\theta}\phi_{\theta}|\phi_{\theta}
\rangle\right\vert ^{2}.
\end{align}
From Proposition~\ref{prop:SLD-Fish-pure-state}, it follows that
$\langle\partial_{\theta}\phi_{\theta}|\partial_{\theta}\phi_{\theta}
\rangle-\left\vert \langle\partial_{\theta}\phi_{\theta}|\phi_{\theta}
\rangle\right\vert ^{2}=I_{F}(\theta;\{|\phi_{\theta}\rangle\!\langle
\phi_{\theta}|\}_{\theta})$. Then, by the faithfulness of SLD\ Fisher
information from Proposition~\ref{prop:faithfulness-SLD-RLD-Fish}, it follows
that $\{|\phi_{\theta}\rangle\!\langle\phi_{\theta}|\}_{\theta}$ is a constant family.
\end{proof}

\section{Additivity of SLD\ and RLD\ Fisher informations}

\label{app:additivity-SLD-RLD}

\begin{proof}
[Proof of Proposition~\ref{prop:additivity-SLD-RLD-states}]Let us begin with
the SLD\ Fisher information. We are trying to prove the following
statement:\ Let $\{\rho_{A}^{\theta}\}_{\theta}$ and $\{\sigma_{B}^{\theta
}\}_{\theta}$ be differentiable families of quantum states. The SLD\ Fisher
information is additive in the following sense:
\begin{equation}
I_{F}(\theta;\{\rho_{A}^{\theta}\otimes\sigma_{B}^{\theta}\}_{\theta}
)=I_{F}(\theta;\{\rho_{A}^{\theta}\}_{\theta})+I_{F}(\theta;\{\sigma
_{B}^{\theta}\}_{\theta}). \label{eq:additivity-SLD-Fish-app}
\end{equation}
Let us first consider the finiteness condition in
\eqref{eq:finiteness-condition-SLD-Fish-states}. For the quantities on the
right-hand side of \eqref{eq:additivity-SLD-Fish-app}, the finiteness
conditions are
\begin{equation}
\Pi_{\rho_{A}^{\theta}}^{\perp}(\partial_{\theta}\rho_{A}^{\theta})\Pi
_{\rho_{A}^{\theta}}^{\perp}=0\qquad\wedge\qquad\Pi_{\sigma_{B}^{\theta}
}^{\perp}(\partial_{\theta}\sigma_{B}^{\theta})\Pi_{\sigma_{B}^{\theta}
}^{\perp}=0. \label{eq:finiteness-conditions-additivity}
\end{equation}
For the quantity on the left-hand side of \eqref{eq:additivity-SLD-Fish-app},
the finiteness condition is
\begin{equation}
\Pi_{\rho_{A}^{\theta}\otimes\sigma_{B}^{\theta}}^{\perp}(\partial_{\theta
}(\rho_{A}^{\theta}\otimes\sigma_{B}^{\theta}))\Pi_{\rho_{A}^{\theta}
\otimes\sigma_{B}^{\theta}}^{\perp}=0.
\label{eq:finiteness-conditions-add-joint-state}
\end{equation}
We now show that these conditions are equivalent. Consider that
\begin{equation}
\Pi_{\rho_{A}^{\theta}\otimes\sigma_{B}^{\theta}}=\Pi_{\rho_{A}^{\theta}
}\otimes\Pi_{\sigma_{B}^{\theta}}.
\end{equation}
This implies that
\begin{align}
\Pi_{\rho_{A}^{\theta}\otimes\sigma_{B}^{\theta}}^{\perp}  &  =I_{AB}
-\Pi_{\rho_{A}^{\theta}}\otimes\Pi_{\sigma_{B}^{\theta}}\\
&  =\Pi_{\rho_{A}^{\theta}}^{\perp}\otimes\Pi_{\sigma_{B}^{\theta}}^{\perp
}+\Pi_{\rho_{A}^{\theta}}^{\perp}\otimes\Pi_{\sigma_{B}^{\theta}}+\Pi
_{\rho_{A}^{\theta}}\otimes\Pi_{\sigma_{B}^{\theta}}^{\perp}.
\end{align}
Consider that
\begin{equation}
\partial_{\theta}(\rho_{A}^{\theta}\otimes\sigma_{B}^{\theta})=(\partial
_{\theta}\rho_{A}^{\theta})\otimes\sigma_{B}^{\theta}+\rho_{A}^{\theta}
\otimes(\partial_{\theta}\sigma_{B}^{\theta}).
\end{equation}
Then
\begin{align}
&  \Pi_{\rho_{A}^{\theta}\otimes\sigma_{B}^{\theta}}^{\perp}(\partial_{\theta
}(\rho_{A}^{\theta}\otimes\sigma_{B}^{\theta}))\Pi_{\rho_{A}^{\theta}
\otimes\sigma_{B}^{\theta}}^{\perp}\nonumber\\
&  =\left(  \Pi_{\rho_{A}^{\theta}}^{\perp}\otimes\Pi_{\sigma_{B}^{\theta}
}^{\perp}+\Pi_{\rho_{A}^{\theta}}^{\perp}\otimes\Pi_{\sigma_{B}^{\theta}}
+\Pi_{\rho_{A}^{\theta}}\otimes\Pi_{\sigma_{B}^{\theta}}^{\perp}\right)
\left(  (\partial_{\theta}\rho_{A}^{\theta})\otimes\sigma_{B}^{\theta}
+\rho_{A}^{\theta}\otimes(\partial_{\theta}\sigma_{B}^{\theta})\right)
\nonumber\\
&  \qquad\times\left(  \Pi_{\rho_{A}^{\theta}}^{\perp}\otimes\Pi_{\sigma
_{B}^{\theta}}^{\perp}+\Pi_{\rho_{A}^{\theta}}^{\perp}\otimes\Pi_{\sigma
_{B}^{\theta}}+\Pi_{\rho_{A}^{\theta}}\otimes\Pi_{\sigma_{B}^{\theta}}^{\perp
}\right) \\
&  =\left(  \Pi_{\rho_{A}^{\theta}}^{\perp}\otimes\Pi_{\sigma_{B}^{\theta}
}\right)  \left(  (\partial_{\theta}\rho_{A}^{\theta})\otimes\sigma
_{B}^{\theta}\right)  \left(  \Pi_{\rho_{A}^{\theta}}^{\perp}\otimes
\Pi_{\sigma_{B}^{\theta}}\right) \nonumber\\
&  \qquad+\left(  \Pi_{\rho_{A}^{\theta}}\otimes\Pi_{\sigma_{B}^{\theta}
}^{\perp}\right)  \left(  \rho_{A}^{\theta}\otimes(\partial_{\theta}\sigma
_{B}^{\theta})\right)  \left(  \Pi_{\rho_{A}^{\theta}}\otimes\Pi_{\sigma
_{B}^{\theta}}^{\perp}\right) \\
&  =\Pi_{\rho_{A}^{\theta}}^{\perp}(\partial_{\theta}\rho_{A}^{\theta}
)\Pi_{\rho_{A}^{\theta}}^{\perp}\otimes\sigma_{B}^{\theta}+\rho_{A}^{\theta
}\otimes\Pi_{\sigma_{B}^{\theta}}^{\perp}(\partial_{\theta}\sigma_{B}^{\theta
})\Pi_{\sigma_{B}^{\theta}}^{\perp}.
\end{align}
From this we see that $\Pi_{\rho_{A}^{\theta}\otimes\sigma_{B}^{\theta}
}^{\perp}(\partial_{\theta}(\rho_{A}^{\theta}\otimes\sigma_{B}^{\theta}
))\Pi_{\rho_{A}^{\theta}\otimes\sigma_{B}^{\theta}}^{\perp}=0$ if
\eqref{eq:finiteness-conditions-additivity} holds. Now suppose that $\Pi
_{\rho_{A}^{\theta}\otimes\sigma_{B}^{\theta}}^{\perp}(\partial_{\theta}
(\rho_{A}^{\theta}\otimes\sigma_{B}^{\theta}))\Pi_{\rho_{A}^{\theta}
\otimes\sigma_{B}^{\theta}}^{\perp}=0$ holds.\ Then we can sandwich this
equation by $I_{A}\otimes\Pi_{\sigma_{B}^{\theta}}$ and perform a partial
trace over $B$ to conclude that $\Pi_{\rho_{A}^{\theta}}^{\perp}
(\partial_{\theta}\rho_{A}^{\theta})\Pi_{\rho_{A}^{\theta}}^{\perp}=0$, i.e.,
\begin{equation}
\left(  I\otimes\Pi_{\sigma_{B}^{\theta}}\right)  \left[  \Pi_{\rho
_{A}^{\theta}\otimes\sigma_{B}^{\theta}}^{\perp}(\partial_{\theta}(\rho
_{A}^{\theta}\otimes\sigma_{B}^{\theta}))\Pi_{\rho_{A}^{\theta}\otimes
\sigma_{B}^{\theta}}^{\perp}\right]  \left(  I\otimes\Pi_{\sigma_{B}^{\theta}
}\right)  =\Pi_{\rho_{A}^{\theta}}^{\perp}(\partial_{\theta}\rho_{A}^{\theta
})\Pi_{\rho_{A}^{\theta}}^{\perp}\otimes\sigma_{B}^{\theta}.
\end{equation}
Similarly, we can sandwich by $\Pi_{\rho_{A}^{\theta}}\otimes I_{B}$ and
perform a partial trace over $A$ to conclude that $\Pi_{\sigma_{B}^{\theta}
}^{\perp}(\partial_{\theta}\sigma_{B}^{\theta})\Pi_{\sigma_{B}^{\theta}
}^{\perp}=0$.

Due to the equivalence of the conditions in
\eqref{eq:finiteness-conditions-additivity} and
\eqref{eq:finiteness-conditions-add-joint-state}, it follows that the
left-hand side of \eqref{eq:additivity-SLD-Fish-app} is infinite if and only
if the right-hand side of \eqref{eq:additivity-SLD-Fish-app} is infinite. So
we can analyze the case in which the quantities are finite by making use of
the explicit formula in \eqref{eq:SLD-Fish-info-formula}.

Consider the following spectral decompositions of $\rho_{A}^{\theta}$ and
$\sigma_{B}^{\theta}$:
\begin{equation}
\rho_{A}^{\theta}=\sum_{x}\lambda_{x}^{\theta}|\psi_{x}^{\theta}\rangle
\langle\psi_{x}^{\theta}|,\qquad\sigma_{B}^{\theta}=\sum_{y}\mu_{y}^{\theta
}|\varphi_{y}^{\theta}\rangle\!\langle\varphi_{y}^{\theta}|.
\end{equation}
Plugging into the formula for SLD\ Fisher information from
\eqref{eq:SLD-Fish-info-formula}, while observing that
\begin{equation}
\partial_{\theta}(\rho_{A}^{\theta}\otimes\sigma_{B}^{\theta})=(\partial
_{\theta}\rho_{A}^{\theta})\otimes\sigma_{B}^{\theta}+\rho_{A}^{\theta}
\otimes(\partial_{\theta}\sigma_{B}^{\theta}),
\end{equation}
we find that
\begin{align}
&  I_{F}(\theta;\{\rho_{A}^{\theta}\otimes\sigma_{B}^{\theta}\}_{\theta
})\nonumber\\
&  =2\sum_{\substack{x,y,x^{\prime},y^{\prime}:\\\lambda_{x}^{\theta}\mu
_{y}^{\theta}+\lambda_{x^{\prime}}^{\theta}\mu_{y^{\prime}}^{\theta}>0}
}\frac{|\langle\psi_{x}^{\theta}|_{A}\langle\varphi_{y}^{\theta}|_{B}\left(
\partial_{\theta}(\rho_{A}^{\theta}\otimes\sigma_{B}^{\theta})\right)
|\psi_{x^{\prime}}^{\theta}\rangle|\varphi_{y^{\prime}}^{\theta}\rangle|^{2}
}{\lambda_{x}^{\theta}\mu_{y}^{\theta}+\lambda_{x^{\prime}}^{\theta}
\mu_{y^{\prime}}^{\theta}}\\
&  =2\sum_{\substack{x,y,x^{\prime},y^{\prime}:\\\lambda_{x}^{\theta}\mu
_{y}^{\theta}+\lambda_{x^{\prime}}^{\theta}\mu_{y^{\prime}}^{\theta}>0}
}\frac{|\langle\psi_{x}^{\theta}|_{A}\langle\varphi_{y}^{\theta}|_{B}\left(
(\partial_{\theta}\rho_{A}^{\theta})\otimes\sigma_{B}^{\theta}+\rho
_{A}^{\theta}\otimes(\partial_{\theta}\sigma_{B}^{\theta})\right)
|\psi_{x^{\prime}}^{\theta}\rangle|\varphi_{y^{\prime}}^{\theta}\rangle|^{2}
}{\lambda_{x}^{\theta}\mu_{y}^{\theta}+\lambda_{x^{\prime}}^{\theta}
\mu_{y^{\prime}}^{\theta}} \label{eq:fish-add-SLD-states}
\end{align}
Then consider that
\begin{align}
&  \left\vert \langle\psi_{x}^{\theta}|_{A}\langle\varphi_{y}^{\theta}
|_{B}\left(  (\partial_{\theta}\rho_{A}^{\theta})\otimes\sigma_{B}^{\theta
}+\rho_{A}^{\theta}\otimes(\partial_{\theta}\sigma_{B}^{\theta})\right)
|\psi_{x^{\prime}}^{\theta}\rangle|\varphi_{y^{\prime}}^{\theta}
\rangle\right\vert ^{2}\nonumber\\
&  =\left\vert \mu_{y}^{\theta}\delta_{y,y^{\prime}}\langle\psi_{x}^{\theta
}|_{A}(\partial_{\theta}\rho_{A}^{\theta})|\psi_{x^{\prime}}^{\theta}
\rangle+\lambda_{x}^{\theta}\delta_{x,x^{\prime}}\langle\varphi_{y}^{\theta
}|_{B}(\partial_{\theta}\sigma_{B}^{\theta})|\varphi_{y^{\prime}}^{\theta
}\rangle\right\vert ^{2}\\
&  =\left(  \mu_{y}^{\theta}\right)  ^{2}\delta_{y,y^{\prime}}\left\vert
\langle\psi_{x}^{\theta}|_{A}(\partial_{\theta}\rho_{A}^{\theta}
)|\psi_{x^{\prime}}^{\theta}\rangle\right\vert ^{2}\nonumber\\
&  \qquad+\mu_{y}^{\theta}\lambda_{x}^{\theta}\delta_{y,y^{\prime}}
\delta_{x,x^{\prime}}2\operatorname{Re}\left[  \langle\psi_{x}^{\theta}
|_{A}(\partial_{\theta}\rho_{A}^{\theta})|\psi_{x^{\prime}}^{\theta}
\rangle\!\langle\varphi_{y}^{\theta}|_{B}(\partial_{\theta}\sigma_{B}^{\theta
})|\varphi_{y^{\prime}}^{\theta}\rangle\right] \nonumber\\
&  \qquad+\left(  \lambda_{x}^{\theta}\right)  ^{2}\delta_{x,x^{\prime}
}\left\vert \langle\varphi_{y}^{\theta}|_{B}(\partial_{\theta}\sigma
_{B}^{\theta})|\varphi_{y^{\prime}}^{\theta}\rangle\right\vert ^{2}.
\end{align}
Plugging back into \eqref{eq:fish-add-SLD-states}\ and evaluating each of the
three terms separately, we find that
\begin{align}
&  2\sum_{\substack{x,y,x^{\prime},y^{\prime}:\\\lambda_{x}^{\theta}\mu
_{y}^{\theta}+\lambda_{x^{\prime}}^{\theta}\mu_{y^{\prime}}^{\theta}>0}
}\frac{\left(  \mu_{y}^{\theta}\right)  ^{2}\delta_{y,y^{\prime}}\left\vert
\langle\psi_{x}^{\theta}|_{A}(\partial_{\theta}\rho_{A}^{\theta}
)|\psi_{x^{\prime}}^{\theta}\rangle\right\vert ^{2}}{\lambda_{x}^{\theta}
\mu_{y}^{\theta}+\lambda_{x^{\prime}}^{\theta}\mu_{y^{\prime}}^{\theta}
}\nonumber\\
&  =2\sum_{\substack{x,y,x^{\prime},y^{\prime}:\\\mu_{y}^{\theta}\left(
\lambda_{x}^{\theta}+\lambda_{x^{\prime}}^{\theta}\right)  >0}}\frac{\left(
\mu_{y}^{\theta}\right)  ^{2}\left\vert \langle\psi_{x}^{\theta}|_{A}
(\partial_{\theta}\rho_{A}^{\theta})|\psi_{x^{\prime}}^{\theta}\rangle
\right\vert ^{2}}{\mu_{y}^{\theta}\left(  \lambda_{x}^{\theta}+\lambda
_{x^{\prime}}^{\theta}\right)  }\\
&  =2\sum_{\substack{x,y,x^{\prime},y^{\prime}:\\\lambda_{x}^{\theta}
+\lambda_{x^{\prime}}^{\theta}>0,\mu_{y}^{\theta}>0}}\frac{\left(  \mu
_{y}^{\theta}\right)  ^{2}\left\vert \langle\psi_{x}^{\theta}|_{A}
(\partial_{\theta}\rho_{A}^{\theta})|\psi_{x^{\prime}}^{\theta}\rangle
\right\vert ^{2}}{\mu_{y}^{\theta}\left(  \lambda_{x}^{\theta}+\lambda
_{x^{\prime}}^{\theta}\right)  }\\
&  =2\sum_{x,y,x^{\prime}:\lambda_{x}^{\theta}+\lambda_{x^{\prime}}^{\theta
}>0,\mu_{y}^{\theta}>0}\frac{\mu_{y}^{\theta}\left\vert \langle\psi
_{x}^{\theta}|_{A}(\partial_{\theta}\rho_{A}^{\theta})|\psi_{x^{\prime}
}^{\theta}\rangle\right\vert ^{2}}{\lambda_{x}^{\theta}+\lambda_{x^{\prime}
}^{\theta}}\\
&  =2\sum_{x,x^{\prime}:\lambda_{x}^{\theta}+\lambda_{x^{\prime}}^{\theta}
>0}\frac{\left\vert \langle\psi_{x}^{\theta}|_{A}(\partial_{\theta}\rho
_{A}^{\theta})|\psi_{x^{\prime}}^{\theta}\rangle\right\vert ^{2}}{\lambda
_{x}^{\theta}+\lambda_{x^{\prime}}^{\theta}}\sum_{y:\mu_{y}^{\theta}>0}\mu
_{y}^{\theta}\\
&  =2\sum_{x,x^{\prime}:\lambda_{x}^{\theta}+\lambda_{x^{\prime}}^{\theta}
>0}\frac{\left\vert \langle\psi_{x}^{\theta}|_{A}(\partial_{\theta}\rho
_{A}^{\theta})|\psi_{x^{\prime}}^{\theta}\rangle\right\vert ^{2}}{\lambda
_{x}^{\theta}+\lambda_{x^{\prime}}^{\theta}}\\
&  =I_{F}(\theta;\{\rho_{A}^{\theta}\}_{\theta}).
\end{align}
For the second term:
\begin{align}
&  2\sum_{\substack{x,y,x^{\prime},y^{\prime}:\\\lambda_{x}^{\theta}\mu
_{y}^{\theta}+\lambda_{x^{\prime}}^{\theta}\mu_{y^{\prime}}^{\theta}>0}
}\frac{\mu_{y}^{\theta}\lambda_{x}^{\theta}\delta_{y,y^{\prime}}
\delta_{x,x^{\prime}}2\operatorname{Re}\left[  \langle\psi_{x}^{\theta}
|_{A}(\partial_{\theta}\rho_{A}^{\theta})|\psi_{x^{\prime}}^{\theta}
\rangle\!\langle\varphi_{y}^{\theta}|_{B}(\partial_{\theta}\sigma_{B}^{\theta
})|\varphi_{y^{\prime}}^{\theta}\rangle\right]  }{\lambda_{x}^{\theta}\mu
_{y}^{\theta}+\lambda_{x^{\prime}}^{\theta}\mu_{y^{\prime}}^{\theta}
}\nonumber\\
&  =\sum_{\substack{x,y:\\\lambda_{x}^{\theta}\mu_{y}^{\theta}>0}}\frac
{\mu_{y}^{\theta}\lambda_{x}^{\theta}2\operatorname{Re}\left[  \langle\psi
_{x}^{\theta}|_{A}(\partial_{\theta}\rho_{A}^{\theta})|\psi_{x}^{\theta
}\rangle\!\langle\varphi_{y}^{\theta}|_{B}(\partial_{\theta}\sigma_{B}^{\theta
})|\varphi_{y}^{\theta}\rangle\right]  }{\lambda_{x}^{\theta}\mu_{y}^{\theta}
}\\
&  =\sum_{\substack{x,y:\\\lambda_{x}^{\theta}>0,\mu_{y}^{\theta}>0}}\frac
{\mu_{y}^{\theta}\lambda_{x}^{\theta}2\operatorname{Re}\left[  \langle\psi
_{x}^{\theta}|_{A}(\partial_{\theta}\rho_{A}^{\theta})|\psi_{x}^{\theta
}\rangle\!\langle\varphi_{y}^{\theta}|_{B}(\partial_{\theta}\sigma_{B}^{\theta
})|\varphi_{y}^{\theta}\rangle\right]  }{\lambda_{x}^{\theta}\mu_{y}^{\theta}
}\\
&  =2\sum_{\substack{x,y:\\\lambda_{x}^{\theta}>0,\mu_{y}^{\theta}
>0}}\operatorname{Re}\left[  \langle\psi_{x}^{\theta}|_{A}(\partial_{\theta
}\rho_{A}^{\theta})|\psi_{x}^{\theta}\rangle\!\langle\varphi_{y}^{\theta}
|_{B}(\partial_{\theta}\sigma_{B}^{\theta})|\varphi_{y}^{\theta}\rangle\right]
\\
&  =2\operatorname{Re}\left[  \sum_{x:\lambda_{x}^{\theta}>0}\langle\psi
_{x}^{\theta}|_{A}(\partial_{\theta}\rho_{A}^{\theta})|\psi_{x}^{\theta
}\rangle\sum_{y:\mu_{y}^{\theta}>0}\langle\varphi_{y}^{\theta}|_{B}
(\partial_{\theta}\sigma_{B}^{\theta})|\varphi_{y}^{\theta}\rangle\right] \\
&  =2\operatorname{Re}\left[  \sum_{x}\langle\psi_{x}^{\theta}|_{A}
(\partial_{\theta}\rho_{A}^{\theta})|\psi_{x}^{\theta}\rangle\sum_{y}
\langle\varphi_{y}^{\theta}|_{B}(\partial_{\theta}\sigma_{B}^{\theta}
)|\varphi_{y}^{\theta}\rangle\right] \\
&  =2\operatorname{Re}\left[  \operatorname{Tr}[\partial_{\theta}\rho
_{A}^{\theta}]\operatorname{Tr}[\partial_{\theta}\sigma_{B}^{\theta}]\right]
\\
&  =0.
\end{align}
The third-to-last equality follows because
\eqref{eq:finiteness-conditions-additivity} holds, so that we can add these to
the sums to complete the basis for the trace. The last equality follows
because $\operatorname{Tr}[\partial_{\theta}\rho_{A}^{\theta}
]=\operatorname{Tr}[\partial_{\theta}\sigma_{B}^{\theta}]=0$. The analysis
involving the last term $\left(  \lambda_{x}^{\theta}\right)  ^{2}
\delta_{x,x^{\prime}}\left\vert \langle\varphi_{y}^{\theta}|_{B}
(\partial_{\theta}\sigma_{B}^{\theta})|\varphi_{y^{\prime}}^{\theta}
\rangle\right\vert ^{2}$ is similar to that of the first term, and it
evaluates to $I_{F}(\theta;\{\sigma_{B}^{\theta}\}_{\theta})$.

Now let us turn to the RLD\ Fisher information. We are trying to prove the
following statement: Let $\{\rho_{A}^{\theta}\}_{\theta}$ and $\{\sigma
_{B}^{\theta}\}_{\theta}$ be differentiable families of quantum states. The
RLD\ Fisher information is additive in the following sense:
\begin{equation}
\widehat{I}_{F}(\theta;\{\rho_{A}^{\theta}\otimes\sigma_{B}^{\theta}
\}_{\theta})=\widehat{I}_{F}(\theta;\{\rho_{A}^{\theta}\}_{\theta}
)+\widehat{I}_{F}(\theta;\{\sigma_{B}^{\theta}\}_{\theta}).
\label{eq:additivity-RLD-Fish-app}
\end{equation}
Let us begin by considering the finiteness condition in
\eqref{eq:support-condition-RLD} for RLD\ Fisher information. For the
quantities on the right-hand side of \eqref{eq:additivity-RLD-Fish-app}, the
finiteness conditions are
\begin{equation}
\Pi_{\rho_{A}^{\theta}}^{\perp}(\partial_{\theta}\rho_{A}^{\theta}
)=0\qquad\wedge\qquad\Pi_{\sigma_{B}^{\theta}}^{\perp}(\partial_{\theta}
\sigma_{B}^{\theta})=0. \label{eq:finiteness-conditions-additivity-RLD}
\end{equation}
For the quantity on the left-hand side of \eqref{eq:additivity-RLD-Fish-app},
the finiteness condition is
\begin{equation}
\Pi_{\rho_{A}^{\theta}\otimes\sigma_{B}^{\theta}}^{\perp}(\partial_{\theta
}(\rho_{A}^{\theta}\otimes\sigma_{B}^{\theta}))=0.
\label{eq:finiteness-conditions-add-joint-state-RLD}
\end{equation}
We now show that these conditions are equivalent. Consider that
\begin{equation}
\Pi_{\rho_{A}^{\theta}\otimes\sigma_{B}^{\theta}}=\Pi_{\rho_{A}^{\theta}
}\otimes\Pi_{\sigma_{B}^{\theta}}.
\end{equation}
This implies that
\begin{align}
\Pi_{\rho_{A}^{\theta}\otimes\sigma_{B}^{\theta}}^{\perp}  &  =I_{AB}
-\Pi_{\rho_{A}^{\theta}}\otimes\Pi_{\sigma_{B}^{\theta}}\\
&  =\Pi_{\rho_{A}^{\theta}}^{\perp}\otimes\Pi_{\sigma_{B}^{\theta}}^{\perp
}+\Pi_{\rho_{A}^{\theta}}^{\perp}\otimes\Pi_{\sigma_{B}^{\theta}}+\Pi
_{\rho_{A}^{\theta}}\otimes\Pi_{\sigma_{B}^{\theta}}^{\perp}.
\end{align}
Consider that
\begin{equation}
\partial_{\theta}(\rho_{A}^{\theta}\otimes\sigma_{B}^{\theta})=(\partial
_{\theta}\rho_{A}^{\theta})\otimes\sigma_{B}^{\theta}+\rho_{A}^{\theta}
\otimes(\partial_{\theta}\sigma_{B}^{\theta}).
\end{equation}
Then we find that
\begin{align}
&  \Pi_{\rho_{A}^{\theta}\otimes\sigma_{B}^{\theta}}^{\perp}\partial_{\theta
}(\rho_{A}^{\theta}\otimes\sigma_{B}^{\theta})\nonumber\\
&  =\left(  \Pi_{\rho_{A}^{\theta}}^{\perp}\otimes\Pi_{\sigma_{B}^{\theta}
}^{\perp}+\Pi_{\rho_{A}^{\theta}}^{\perp}\otimes\Pi_{\sigma_{B}^{\theta}}
+\Pi_{\rho_{A}^{\theta}}\otimes\Pi_{\sigma_{B}^{\theta}}^{\perp}\right)
\left(  (\partial_{\theta}\rho_{A}^{\theta})\otimes\sigma_{B}^{\theta}
+\rho_{A}^{\theta}\otimes(\partial_{\theta}\sigma_{B}^{\theta})\right) \\
&  =\left(  \Pi_{\rho_{A}^{\theta}}^{\perp}\otimes\Pi_{\sigma_{B}^{\theta}
}\right)  \left(  (\partial_{\theta}\rho_{A}^{\theta})\otimes\sigma
_{B}^{\theta}\right)  +\left(  \Pi_{\rho_{A}^{\theta}}\otimes\Pi_{\sigma
_{B}^{\theta}}^{\perp}\right)  \left(  \rho_{A}^{\theta}\otimes(\partial
_{\theta}\sigma_{B}^{\theta})\right) \\
&  =\Pi_{\rho_{A}^{\theta}}^{\perp}(\partial_{\theta}\rho_{A}^{\theta}
)\otimes\sigma_{B}^{\theta}+\rho_{A}^{\theta}\otimes\Pi_{\sigma_{B}^{\theta}
}^{\perp}(\partial_{\theta}\sigma_{B}^{\theta}).
\end{align}
From this we see that $\Pi_{\rho_{A}^{\theta}\otimes\sigma_{B}^{\theta}
}^{\perp}(\partial_{\theta}(\rho_{A}^{\theta}\otimes\sigma_{B}^{\theta}))=0$
if \eqref{eq:finiteness-conditions-additivity-RLD} holds. Now suppose that
$\Pi_{\rho_{A}^{\theta}\otimes\sigma_{B}^{\theta}}^{\perp}(\partial_{\theta
}(\rho_{A}^{\theta}\otimes\sigma_{B}^{\theta}))=0$ holds.\ Then we can
left-multiply this equation by $I_{A}\otimes\Pi_{\sigma_{B}^{\theta}}$ and
perform a partial trace over $B$ to conclude that $\Pi_{\rho_{A}^{\theta}
}^{\perp}(\partial_{\theta}\rho_{A}^{\theta})=0$, i.e.,
\begin{equation}
\left(  I\otimes\Pi_{\sigma_{B}^{\theta}}\right)  \left[  \Pi_{\rho
_{A}^{\theta}\otimes\sigma_{B}^{\theta}}^{\perp}(\partial_{\theta}(\rho
_{A}^{\theta}\otimes\sigma_{B}^{\theta}))\right]  =\Pi_{\rho_{A}^{\theta}
}^{\perp}(\partial_{\theta}\rho_{A}^{\theta})\otimes\sigma_{B}^{\theta}.
\end{equation}
Similarly, we can left-multiply by $\Pi_{\rho_{A}^{\theta}}\otimes I_{B}$ and
perform a partial trace over $A$ to conclude that $\Pi_{\sigma_{B}^{\theta}
}^{\perp}(\partial_{\theta}\sigma_{B}^{\theta})=0$.

Due to the equivalence of the conditions in
\eqref{eq:finiteness-conditions-additivity-RLD} and
\eqref{eq:finiteness-conditions-add-joint-state-RLD}, it follows that the
left-hand side of \eqref{eq:additivity-RLD-Fish-app} is infinite if and only
if the right-hand side of \eqref{eq:additivity-RLD-Fish-app} is infinite. So
we can analyze the case in which the quantities are finite by making use of
the explicit formula in Definition~\ref{def:RLD-Fish-info-states}.

Observe that
\begin{align}
(\partial_{\theta}(\rho_{A}^{\theta}\otimes\sigma_{B}^{\theta}))^{2}  &
=((\partial_{\theta}\rho_{A}^{\theta})\otimes\sigma_{B}^{\theta}+\rho
_{A}^{\theta}\otimes(\partial_{\theta}\sigma_{B}^{\theta}))^{2}\\
&  =(\partial_{\theta}\rho_{A}^{\theta})^{2}\otimes(\sigma_{B}^{\theta}
)^{2}+(\partial_{\theta}\rho_{A}^{\theta})\rho_{A}^{\theta}\otimes\sigma
_{B}^{\theta}(\partial_{\theta}\sigma_{B}^{\theta})\nonumber\\
&  \qquad+\rho_{A}^{\theta}(\partial_{\theta}\rho_{A}^{\theta})\otimes
(\partial_{\theta}\sigma_{B}^{\theta})\sigma_{B}^{\theta}+(\rho_{A}^{\theta
})^{2}\otimes(\partial_{\theta}\sigma_{B}^{\theta})^{2}.
\end{align}
Then consider that
\begin{align}
&  \widehat{I}_{F}(\theta;\{\rho_{A}^{\theta}\otimes\sigma_{B}^{\theta
}\}_{\theta})\nonumber\\
&  =\operatorname{Tr}[(\partial_{\theta}(\rho_{A}^{\theta}\otimes\sigma
_{B}^{\theta}))^{2}(\rho_{A}^{\theta}\otimes\sigma_{B}^{\theta})^{-1}]\\
&  =\operatorname{Tr}[((\partial_{\theta}\rho_{A}^{\theta})^{2}\otimes
(\sigma_{B}^{\theta})^{2})((\rho_{A}^{\theta})^{-1}\otimes(\sigma_{B}^{\theta
})^{-1})]\nonumber\\
&  \qquad+\operatorname{Tr}[((\partial_{\theta}\rho_{A}^{\theta})\rho
_{A}^{\theta}\otimes\sigma_{B}^{\theta}(\partial_{\theta}\sigma_{B}^{\theta
}))((\rho_{A}^{\theta})^{-1}\otimes(\sigma_{B}^{\theta})^{-1})]\nonumber\\
&  \qquad+\operatorname{Tr}[(\rho_{A}^{\theta}(\partial_{\theta}\rho
_{A}^{\theta})\otimes(\partial_{\theta}\sigma_{B}^{\theta})\sigma_{B}^{\theta
})((\rho_{A}^{\theta})^{-1}\otimes(\sigma_{B}^{\theta})^{-1})]\nonumber\\
&  \qquad+\operatorname{Tr}[((\rho_{A}^{\theta})^{2}\otimes(\partial_{\theta
}\sigma_{B}^{\theta})^{2})((\rho_{A}^{\theta})^{-1}\otimes(\sigma_{B}^{\theta
})^{-1})]\\
&  =\operatorname{Tr}[(\partial_{\theta}\rho_{A}^{\theta})^{2}(\rho
_{A}^{\theta})^{-1}]\operatorname{Tr}[\sigma_{B}^{\theta}]+2\operatorname{Tr}
[\Pi_{\rho_{A}^{\theta}}(\partial_{\theta}\rho_{A}^{\theta})]\operatorname{Tr}
[\Pi_{\sigma_{B}^{\theta}}(\partial_{\theta}\sigma_{B}^{\theta})]\nonumber\\
&  \qquad+\operatorname{Tr}[\rho_{A}^{\theta}]\operatorname{Tr}[(\partial
_{\theta}\sigma_{B}^{\theta})^{2}(\sigma_{B}^{\theta})^{-1})]\\
&  =\operatorname{Tr}[(\partial_{\theta}\rho_{A}^{\theta})^{2}(\rho
_{A}^{\theta})^{-1}]+2\operatorname{Tr}[(\partial_{\theta}\rho_{A}^{\theta
})]\operatorname{Tr}[(\partial_{\theta}\sigma_{B}^{\theta})]+\operatorname{Tr}
[(\partial_{\theta}\sigma_{B}^{\theta})^{2}(\sigma_{B}^{\theta})^{-1})]\\
&  =\widehat{I}_{F}(\theta;\{\rho_{A}^{\theta}\}_{\theta})+\widehat{I}
_{F}(\theta;\{\sigma_{B}^{\theta}\}_{\theta}).
\end{align}
The second-to-last equality follows because $\Pi_{\rho_{A}^{\theta}}^{\perp
}(\partial_{\theta}\rho_{A}^{\theta})=0$ and $\Pi_{\sigma_{B}^{\theta}}
^{\perp}(\partial_{\theta}\sigma_{B}^{\theta})=0$, so that $\operatorname{Tr}
[\Pi_{\rho_{A}^{\theta}}(\partial_{\theta}\rho_{A}^{\theta}
)]=\operatorname{Tr}[(\partial_{\theta}\rho_{A}^{\theta})]$ and
$\operatorname{Tr}[\Pi_{\sigma_{B}^{\theta}}(\partial_{\theta}\sigma
_{B}^{\theta})]=\operatorname{Tr}[(\partial_{\theta}\sigma_{B}^{\theta})]$.
The final equality follows because $\operatorname{Tr}[(\partial_{\theta}
\rho_{A}^{\theta})]=\operatorname{Tr}[(\partial_{\theta}\sigma_{B}^{\theta
})]=0$.
\end{proof}

\section{SLD and RLD\ Fisher informations for classical--quantum
states\label{app:decomp-fisher-cqq}}

\begin{proof}
[Proof of Proposition~\ref{prop:cq-decomp-SLD-RLD}]We begin with the
SLD\ Fisher information, with the goal being to prove the following
statement:\ For a differentiable family of classical--quantum states:
\begin{equation}
\left\{  \sum_{x}p_{\theta}(x)|x\rangle\!\langle x|_{X}\otimes\rho_{\theta}
^{x}\right\}  _{\theta},
\end{equation}
the SLD Fisher information can be evaluated as follows:
\begin{equation}
I_{F}\!\left(  \theta;\left\{  \sum_{x}p_{\theta}(x)|x\rangle\!\langle
x|_{X}\otimes\rho_{\theta}^{x}\right\}  _{\theta}\right)  =I_{F}
(\theta;\{p_{\theta}\}_{\theta})+\sum_{x:p_{\theta}(x)>0}p_{\theta}
(x)I_{F}(\theta;\{\rho_{\theta}^{x}\}_{\theta}). \label{eq:QFI-breakdown-cqq}
\end{equation}
We first consider the finiteness conditions for the left- and right-hand sides
of \eqref{eq:QFI-breakdown-cqq} and show that they are equivalent. For the
right-hand side, the finiteness conditions are
\begin{equation}
\operatorname{supp}(\partial_{\theta}p_{\theta})\subseteq\operatorname{supp}
(p_{\theta})\quad\wedge\quad\Pi_{\rho_{\theta}^{x}}^{\perp}(\partial_{\theta
}\rho_{\theta}^{x})\Pi_{\rho_{\theta}^{x}}^{\perp}=0\quad\forall x:p_{\theta
}(x)>0, \label{eq:cq-finiteness-cond-SLD-right-side}
\end{equation}
while for the left-hand side, the finiteness condition is
\begin{equation}
\Pi_{\rho_{XB}^{\theta}}^{\perp}(\partial_{\theta}\rho_{XB}^{\theta})\Pi
_{\rho_{XB}^{\theta}}^{\perp}=0, \label{eq:cq-finiteness-cond-SLD-left-side}
\end{equation}
where
\begin{equation}
\rho_{XB}^{\theta}:=\sum_{x}p_{\theta}(x)|x\rangle\!\langle x|_{X}\otimes
\rho_{\theta}^{x}.
\end{equation}
Consider that
\begin{equation}
\Pi_{\rho_{XB}^{\theta}}=\sum_{x:p_{\theta}(x)>0}|x\rangle\!\langle
x|_{X}\otimes\Pi_{\rho_{\theta}^{x}},
\end{equation}
which implies that
\begin{align}
\Pi_{\rho_{XB}^{\theta}}^{\perp}  &  =I_{XB}-\Pi_{\rho_{XB}^{\theta}}\\
&  =\sum_{x:p_{\theta}(x)=0}|x\rangle\!\langle x|_{X}\otimes I+\sum
_{x:p_{\theta}(x)>0}|x\rangle\!\langle x|_{X}\otimes\Pi_{\rho_{\theta}^{x}
}^{\perp}.
\end{align}
Also, we have that
\begin{align}
\partial_{\theta}\rho_{XB}^{\theta}  &  =\partial_{\theta}\left(  \sum
_{x}p_{\theta}(x)|x\rangle\!\langle x|_{X}\otimes\rho_{\theta}^{x}\right) \\
&  =\partial_{\theta}\left(  \sum_{x}|x\rangle\!\langle x|_{X}\otimes p_{\theta
}(x)\rho_{\theta}^{x}\right) \\
&  =\sum_{x}|x\rangle\!\langle x|_{X}\otimes\partial_{\theta}(p_{\theta}
(x)\rho_{\theta}^{x})\\
&  =\sum_{x}|x\rangle\!\langle x|_{X}\otimes\left[  \partial_{\theta}(p_{\theta
}(x))\rho_{\theta}^{x}+p_{\theta}(x)(\partial_{\theta}\rho_{\theta}
^{x})\right]  .\\
&  =\sum_{x}\partial_{\theta}(p_{\theta}(x))|x\rangle\!\langle x|_{X}\otimes
\rho_{\theta}^{x}+\sum_{x}p_{\theta}(x)|x\rangle\!\langle x|_{X}\otimes
(\partial_{\theta}\rho_{\theta}^{x}).
\end{align}
We then find that
\begin{align}
0  &  =\Pi_{\rho_{XB}^{\theta}}^{\perp}(\partial_{\theta}\rho_{XB}^{\theta
})\Pi_{\rho_{XB}^{\theta}}^{\perp}\nonumber\\
&  =\left(  \sum_{x:p_{\theta}(x)=0}|x\rangle\!\langle x|_{X}\otimes
I+\sum_{x:p_{\theta}(x)>0}|x\rangle\!\langle x|_{X}\otimes\Pi_{\rho_{\theta}
^{x}}^{\perp}\right)  \times\nonumber\\
&  \left(  \sum_{x}\partial_{\theta}(p_{\theta}(x))|x\rangle\!\langle
x|_{X}\otimes\rho_{\theta}^{x}+\sum_{x}p_{\theta}(x)|x\rangle\!\langle
x|_{X}\otimes(\partial_{\theta}\rho_{\theta}^{x})\right)  \times\nonumber\\
&  \left(  \sum_{x:p_{\theta}(x)=0}|x\rangle\!\langle x|_{X}\otimes
I+\sum_{x:p_{\theta}(x)>0}|x\rangle\!\langle x|_{X}\otimes\Pi_{\rho_{\theta}
^{x}}^{\perp}\right) \\
&  =\sum_{x:p_{\theta}(x)=0}\partial_{\theta}(p_{\theta}(x))|x\rangle\!\langle
x|_{X}\otimes\rho_{\theta}^{x}+\sum_{x:p_{\theta}(x)>0}p_{\theta}
(x)|x\rangle\!\langle x|_{X}\otimes\Pi_{\rho_{\theta}^{x}}^{\perp}
(\partial_{\theta}\rho_{\theta}^{x})\Pi_{\rho_{\theta}^{x}}^{\perp}.
\label{eq:cq-finiteness-cond-SLD-proof-step}
\end{align}
Now sandwiching by $\sum_{x:p_{\theta}(x)=0}|x\rangle\!\langle x|_{X}\otimes I$
on both sides (which projects out the second sum above) and tracing over the
second system, we conclude that
\begin{equation}
\sum_{x:p_{\theta}(x)=0}\partial_{\theta}(p_{\theta}(x))|x\rangle\!\langle
x|_{X}=0.
\end{equation}
This is the same as $\operatorname{supp}(\partial_{\theta}p_{\theta}
)\subseteq\operatorname{supp}(p_{\theta})$. Instead sandwiching by
$\sum_{x:p_{\theta}(x)>0}|x\rangle\!\langle x|_{X}\otimes I$, we are left with
the following conditions:
\begin{equation}
\Pi_{\rho_{\theta}^{x}}^{\perp}(\partial_{\theta}\rho_{\theta}^{x})\Pi
_{\rho_{\theta}^{x}}^{\perp}=0\quad\forall x:p_{\theta}(x)>0.
\end{equation}
Thus, the finiteness condition in \eqref{eq:cq-finiteness-cond-SLD-left-side}
implies the finiteness condition in
\eqref{eq:cq-finiteness-cond-SLD-right-side}. The other implication follows
from plugging \eqref{eq:cq-finiteness-cond-SLD-right-side}\ into \eqref{eq:cq-finiteness-cond-SLD-proof-step}.

Since the finiteness of the left-hand side of \eqref{eq:QFI-breakdown-cqq} is
equivalent to the finiteness of the right-hand side of
\eqref{eq:QFI-breakdown-cqq}, we can now focus on establishing the equality
under these conditions. For the state
\begin{equation}
\sum_{x}p_{\theta}(x)|x\rangle\!\langle x|_{X}\otimes\rho_{\theta}^{x},
\end{equation}
let its spectral decomposition be as follows:
\begin{equation}
\sum_{x}p_{\theta}(x)|x\rangle\!\langle x|_{X}\otimes\sum_{y}\lambda_{\theta
}^{x,y}|\psi_{\theta}^{x,y}\rangle\!\langle\psi_{\theta}^{x,y}|=\sum
_{x,y}p_{\theta}(x)\lambda_{\theta}^{x,y}|x\rangle\!\langle x|_{X}\otimes
|\psi_{\theta}^{x,y}\rangle\!\langle\psi_{\theta}^{x,y}|.
\end{equation}
Plugging into the SLD Fisher information formula in
\eqref{eq:SLD-Fish-info-formula}, we find that
\begin{align}
&  I_{F}\!\left(  \theta;\left\{  \sum_{x}p_{\theta}(x)|x\rangle\!\langle
x|_{X}\otimes\rho_{\theta}^{x}\right\}  _{\theta}\right) \nonumber\\
&  =2\sum_{\substack{x,y,x^{\prime},y^{\prime}:\\p_{\theta}(x)\lambda_{\theta
}^{x,y}+p_{\theta}(x^{\prime})\lambda_{\theta}^{x^{\prime},y^{\prime}}
>0}}\frac{|\langle x|_{X}\langle\psi_{\theta}^{x,y}|\left(  \partial_{\theta
}\sum_{x^{\prime\prime}}p_{\theta}(x^{\prime\prime})|x^{\prime\prime}
\rangle\!\langle x^{\prime\prime}|_{X}\otimes\rho_{\theta}^{x^{\prime\prime}
}\right)  |x^{\prime}\rangle_{X}|\psi_{\theta}^{x^{\prime},y^{\prime}}
\rangle|^{2}}{p_{\theta}(x)\lambda_{\theta}^{x,y}+p_{\theta}(x^{\prime
})\lambda_{\theta}^{x^{\prime},y^{\prime}}}\\
&  =2\sum_{\substack{x,y,x^{\prime},y^{\prime}:\\p_{\theta}(x)\lambda_{\theta
}^{x,y}+p_{\theta}(x^{\prime})\lambda_{\theta}^{x^{\prime},y^{\prime}}
>0}}\frac{|\langle x|_{X}\langle\psi_{\theta}^{x,y}|\left(  \sum
_{x^{\prime\prime}}|x^{\prime\prime}\rangle\!\langle x^{\prime\prime}
|_{X}\otimes\partial_{\theta}(p_{\theta}(x^{\prime\prime})\rho_{\theta
}^{x^{\prime\prime}})\right)  |x^{\prime}\rangle_{X}|\psi_{\theta}^{x^{\prime
},y^{\prime}}\rangle|^{2}}{p_{\theta}(x)\lambda_{\theta}^{x,y}+p_{\theta
}(x^{\prime})\lambda_{\theta}^{x^{\prime},y^{\prime}}}\\
&  =2\sum_{\substack{x,y,y^{\prime}:\\p_{\theta}(x)(\lambda_{\theta}
^{x,y}+\lambda_{\theta}^{x,y^{\prime}})>0}}\frac{|\langle\psi_{\theta}
^{x,y}|\partial_{\theta}(p_{\theta}(x)\rho_{\theta}^{x})|\psi_{\theta
}^{x,y^{\prime}}\rangle|^{2}}{p_{\theta}(x)\left(  \lambda_{\theta}
^{x,y}+\lambda_{\theta}^{x,y^{\prime}}\right)  }\\
&  =2\sum_{\substack{x,y,y^{\prime}:\\p_{\theta}(x)>0,\lambda_{\theta}
^{x,y}+\lambda_{\theta}^{x,y^{\prime}}>0}}\frac{|\langle\psi_{\theta}
^{x,y}|\partial_{\theta}(p_{\theta}(x)\rho_{\theta}^{x})|\psi_{\theta
}^{x,y^{\prime}}\rangle|^{2}}{p_{\theta}(x)\left(  \lambda_{\theta}
^{x,y}+\lambda_{\theta}^{x,y^{\prime}}\right)  }. \label{eq:numerator-QFI-cqq}
\end{align}
Now consider that
\begin{equation}
\partial_{\theta}(p_{\theta}(x)\rho_{\theta}^{x})=(\partial_{\theta}p_{\theta
}(x))\rho_{\theta}^{x}+p_{\theta}(x)(\partial_{\theta}\rho_{\theta}^{x}).
\end{equation}
Plugging into the numerator in \eqref{eq:numerator-QFI-cqq}, we find that
\begin{align}
&  \langle\psi_{\theta}^{x,y}|\partial_{\theta}(p_{\theta}(x)\rho_{\theta}
^{x})|\psi_{\theta}^{x,y^{\prime}}\rangle\nonumber\\
&  =\langle\psi_{\theta}^{x,y}|(\partial_{\theta}p_{\theta}(x))\rho_{\theta
}^{x}|\psi_{\theta}^{x,y^{\prime}}\rangle+\langle\psi_{\theta}^{x,y}
|p_{\theta}(x)(\partial_{\theta}\rho_{\theta}^{x})|\psi_{\theta}^{x,y^{\prime
}}\rangle\\
&  =\delta_{y,y^{\prime}}\lambda_{\theta}^{x,y}(\partial_{\theta}p_{\theta
}(x))+p_{\theta}(x)\langle\psi_{\theta}^{x,y}|(\partial_{\theta}\rho_{\theta
}^{x})|\psi_{\theta}^{x,y^{\prime}}\rangle.
\end{align}
Then we can evaluate the numerator in \eqref{eq:numerator-QFI-cqq}\ as
follows:
\begin{align}
&  |\langle\psi_{\theta}^{x,y}|\partial_{\theta}(p_{\theta}(x)\rho_{\theta
}^{x})|\psi_{\theta}^{x,y^{\prime}}\rangle|^{2}\nonumber\\
&  =\left\vert \delta_{y,y^{\prime}}\lambda_{\theta}^{x,y}(\partial_{\theta
}p_{\theta}(x))+p_{\theta}(x)\langle\psi_{\theta}^{x,y}|(\partial_{\theta}
\rho_{\theta}^{x})|\psi_{\theta}^{x,y^{\prime}}\rangle\right\vert ^{2}\\
&  =\delta_{y,y^{\prime}}(\lambda_{\theta}^{x,y})^{2}(\partial_{\theta
}p_{\theta}(x))^{2}+2\delta_{y,y^{\prime}}\lambda_{\theta}^{x,y}p_{\theta
}(x)(\partial_{\theta}p_{\theta}(x))\operatorname{Re}[\langle\psi_{\theta
}^{x,y}|(\partial_{\theta}\rho_{\theta}^{x})|\psi_{\theta}^{x,y^{\prime}
}\rangle]\nonumber\\
&  \qquad+\left[  p_{\theta}(x)\right]  ^{2}\left\vert \langle\psi_{\theta
}^{x,y}|(\partial_{\theta}\rho_{\theta}^{x})|\psi_{\theta}^{x,y^{\prime}
}\rangle\right\vert ^{2}.
\end{align}
We can then evaluate the sum in \eqref{eq:numerator-QFI-cqq}\ for each of the
three terms above, starting with the first one:
\begin{align}
&  2\sum_{\substack{x,y,y^{\prime}:\\p_{\theta}(x)>0,\lambda_{\theta}
^{x,y}+\lambda_{\theta}^{x,y^{\prime}}>0}}\frac{\left[  \partial_{\theta
}p_{\theta}(x)\right]  ^{2}\delta_{y,y^{\prime}}\left(  \lambda_{\theta}
^{x,y}\right)  ^{2}}{p_{\theta}(x)\left(  \lambda_{\theta}^{x,y}
+\lambda_{\theta}^{x,y^{\prime}}\right)  }\nonumber \\
&  =\sum_{x,y:p_{\theta}(x)>0,\lambda_{\theta}^{x,y}>0}\frac{\left[
\partial_{\theta}p_{\theta}(x)\right]  ^{2}\left(  \lambda_{\theta}
^{x,y}\right)  ^{2}}{p_{\theta}(x)\lambda_{\theta}^{x,y}}\\
&  =\sum_{x:p_{\theta}(x)>0}\frac{\left[  \partial_{\theta}p_{\theta
}(x)\right]  ^{2}}{p_{\theta}(x)}\sum_{y:\lambda_{\theta}^{x,y}>0}\left(
\lambda_{\theta}^{x,y}\right) \\
&  =\sum_{x:p_{\theta}(x)>0}\frac{\left[  \partial_{\theta}p_{\theta
}(x)\right]  ^{2}}{p_{\theta}(x)}\\
&  =I_{F}(\theta;\{p_{\theta}\}_{\theta}).
\end{align}
Consider the next term:
\begin{align}
&  2\sum_{\substack{x,y,y^{\prime}:\\p_{\theta}(x)>0,\lambda_{\theta}
^{x,y}+\lambda_{\theta}^{x,y^{\prime}}>0}}\frac{2\delta_{y,y^{\prime}}
\lambda_{\theta}^{x,y}p_{\theta}(x)(\partial_{\theta}p_{\theta}
(x))\operatorname{Re}\left[  \langle\psi_{\theta}^{x,y}|(\partial_{\theta}
\rho_{\theta}^{x})|\psi_{\theta}^{x,y^{\prime}}\rangle\right]  }{p_{\theta
}(x)\left(  \lambda_{\theta}^{x,y}+\lambda_{\theta}^{x,y^{\prime}}\right)
}\nonumber\\
&  =2\sum_{\substack{x,y:\\p_{\theta}(x)>0,\lambda_{\theta}^{x,y}>0}
}\frac{\lambda_{\theta}^{x,y}p_{\theta}(x)(\partial_{\theta}p_{\theta
}(x))\operatorname{Re}\left[  \langle\psi_{\theta}^{x,y}|(\partial_{\theta
}\rho_{\theta}^{x})|\psi_{\theta}^{x,y}\rangle\right]  }{p_{\theta}
(x)\lambda_{\theta}^{x,y}}\\
&  =2\sum_{x,y:p_{\theta}(x)>0,\lambda_{\theta}^{x,y}>0}(\partial_{\theta
}p_{\theta}(x))\operatorname{Re}\left[  \langle\psi_{\theta}^{x,y}
|(\partial_{\theta}\rho_{\theta}^{x})|\psi_{\theta}^{x,y}\rangle\right] \\
&  =2\sum_{x:p_{\theta}(x)>0}(\partial_{\theta}p_{\theta}(x))\operatorname{Re}
\left[  \sum_{y:\lambda_{\theta}^{x,y}>0}\langle\psi_{\theta}^{x,y}
|(\partial_{\theta}\rho_{\theta}^{x})|\psi_{\theta}^{x,y}\rangle\right] \\
&  =2\sum_{x:p_{\theta}(x)>0}(\partial_{\theta}p_{\theta}(x))\operatorname{Re}
\left[  \sum_{y}\langle\psi_{\theta}^{x,y}|(\partial_{\theta}\rho_{\theta}
^{x})|\psi_{\theta}^{x,y}\rangle\right] \\
&  =2\sum_{x:p_{\theta}(x)>0}(\partial_{\theta}p_{\theta}(x))\operatorname{Re}
\left[  \operatorname{Tr}[\partial_{\theta}\rho_{\theta}^{x}]\right] \\
&  =0.
\end{align}
The third-to-last equality holds because $\Pi_{\rho_{\theta}^{x}}^{\perp
}(\partial_{\theta}\rho_{\theta}^{x})\Pi_{\rho_{\theta}^{x}}^{\perp}
=0\quad\forall x:p_{\theta}(x)>0$, implying that we can add these terms to the
sum to get the full trace in the next line. The last line follows because
$\operatorname{Tr}[\partial_{\theta}\rho_{\theta}^{x}]=\partial_{\theta
}\operatorname{Tr}[\rho_{\theta}^{x}]=0$. Now consider the final term:
\begin{align}
&  2\sum_{\substack{x,y,y^{\prime}:\\p_{\theta}(x)>0,\lambda_{\theta}
^{x,y}+\lambda_{\theta}^{x,y^{\prime}}>0}}\frac{\left[  p_{\theta}(x)\right]
^{2}\left\vert \langle\psi_{\theta}^{x,y}|(\partial_{\theta}\rho_{\theta}
^{x})|\psi_{\theta}^{x,y^{\prime}}\rangle\right\vert ^{2}}{p_{\theta
}(x)\left(  \lambda_{\theta}^{x,y}+\lambda_{\theta}^{x,y^{\prime}}\right)
}\nonumber\\
&  =2\sum_{\substack{x,y,y^{\prime}:\\p_{\theta}(x)>0,\lambda_{\theta}
^{x,y}+\lambda_{\theta}^{x,y^{\prime}}>0}}\frac{p_{\theta}(x)\left\vert
\langle\psi_{\theta}^{x,y}|(\partial_{\theta}\rho_{\theta}^{x})|\psi_{\theta
}^{x,y^{\prime}}\rangle\right\vert ^{2}}{  \lambda_{\theta}
^{x,y}+\lambda_{\theta}^{x,y^{\prime}}  }\\
&  =\sum_{x:p_{\theta}(x)>0}p_{\theta}(x)\left(  2\sum_{y,y^{\prime}
:\lambda_{\theta}^{x,y}+\lambda_{\theta}^{x,y^{\prime}}>0}\frac{\left\vert
\langle\psi_{\theta}^{x,y}|(\partial_{\theta}\rho_{\theta}^{x})|\psi_{\theta
}^{x,y^{\prime}}\rangle\right\vert ^{2}}{  \lambda_{\theta}
^{x,y}+\lambda_{\theta}^{x,y^{\prime}}  }\right) \\
&  =\sum_{x:p_{\theta}(x)>0}p_{\theta}(x)I_{F}(\theta;\{\rho_{\theta}
^{x}\}_{\theta}).
\end{align}
So we conclude the formula in \eqref{eq:QFI-breakdown-cqq}\ after putting all
of the above together.

We now turn to the RLD\ Fisher information, with the goal being to prove the
following statement: For a differentiable family of classical--quantum states:
\begin{equation}
\left\{  \sum_{x}p_{\theta}(x)|x\rangle\!\langle x|_{X}\otimes\rho_{\theta}
^{x}\right\}  _{\theta},
\end{equation}
the RLD Fisher information can be evaluated as follows:
\begin{equation}
\widehat{I}_{F}\!\left(  \theta;\left\{  \sum_{x}p_{\theta}(x)|x\rangle\!\langle
x|_{X}\otimes\rho_{\theta}^{x}\right\}  _{\theta}\right)  =I_{F}
(\theta;\{p_{\theta}\}_{\theta})+\sum_{x:p_{\theta}(x)>0}p_{\theta}
(x)\widehat{I}_{F}(\theta;\{\rho_{\theta}^{x}\}_{\theta}).
\label{eq:RLD-QFI-breakdown-cqq}
\end{equation}
The beginning of the proof is similar to the previous proof for SLD\ Fisher
information, and so we use the same notation used there. We first consider the
finiteness conditions for the left- and right-hand sides of
\eqref{eq:RLD-QFI-breakdown-cqq} and show that they are equivalent. For the
right-hand side, the finiteness conditions are
\begin{equation}
\operatorname{supp}(\partial_{\theta}p_{\theta})\subseteq\operatorname{supp}
(p_{\theta})\quad\wedge\quad(\partial_{\theta}\rho_{\theta}^{x})\Pi
_{\rho_{\theta}^{x}}^{\perp}=0\quad\forall x:p_{\theta}(x)>0,
\label{eq:cq-finiteness-cond-RLD-right-side}
\end{equation}
while for the left-hand side, the finiteness condition is
\begin{equation}
(\partial_{\theta}\rho_{XB}^{\theta})\Pi_{\rho_{XB}^{\theta}}^{\perp}=0.
\label{eq:cq-finiteness-cond-RLD-left-side}
\end{equation}
We find that
\begin{align}
0  &  =(\partial_{\theta}\rho_{XB}^{\theta})\Pi_{\rho_{XB}^{\theta}}^{\perp
}\nonumber\\
&  =\left(  \sum_{x}\partial_{\theta}(p_{\theta}(x))|x\rangle\!\langle
x|_{X}\otimes\rho_{\theta}^{x}+\sum_{x}p_{\theta}(x)|x\rangle\!\langle
x|_{X}\otimes(\partial_{\theta}\rho_{\theta}^{x})\right)  \times\nonumber\\
&  \left(  \sum_{x:p_{\theta}(x)=0}|x\rangle\!\langle x|_{X}\otimes
I+\sum_{x:p_{\theta}(x)>0}|x\rangle\!\langle x|_{X}\otimes\Pi_{\rho_{\theta}
^{x}}^{\perp}\right) \\
&  =\sum_{x:p_{\theta}(x)=0}\partial_{\theta}(p_{\theta}(x))|x\rangle\!\langle
x|_{X}\otimes\rho_{\theta}^{x}+\sum_{x:p_{\theta}(x)>0}p_{\theta}
(x)|x\rangle\!\langle x|_{X}\otimes(\partial_{\theta}\rho_{\theta}^{x})\Pi
_{\rho_{\theta}^{x}}^{\perp}. \label{eq:cq-finiteness-cond-RLD-proof-step}
\end{align}
Now sandwiching by $\sum_{x:p_{\theta}(x)=0}|x\rangle\!\langle x|_{X}\otimes I$
on both sides (which projects out the second sum above) and tracing over the
second system, we conclude that
\begin{equation}
\sum_{x:p_{\theta}(x)=0}\partial_{\theta}(p_{\theta}(x))|x\rangle\!\langle
x|_{X}=0.
\end{equation}
This is the same as $\operatorname{supp}(\partial_{\theta}p_{\theta}
)\subseteq\operatorname{supp}(p_{\theta})$. Instead sandwiching by
$\sum_{x:p_{\theta}(x)>0}|x\rangle\!\langle x|_{X}\otimes I$, we are left with
the following conditions:
\begin{equation}
(\partial_{\theta}\rho_{\theta}^{x})\Pi_{\rho_{\theta}^{x}}^{\perp}
=0\quad\forall x:p_{\theta}(x)>0.
\end{equation}
Thus, the finiteness condition in \eqref{eq:cq-finiteness-cond-RLD-left-side}
implies the finiteness condition in
\eqref{eq:cq-finiteness-cond-RLD-right-side}. The other implication follows
from plugging \eqref{eq:cq-finiteness-cond-RLD-right-side}\ into \eqref{eq:cq-finiteness-cond-RLD-proof-step}.

Since the finiteness of the left-hand side of \eqref{eq:RLD-QFI-breakdown-cqq}
is equivalent to the finiteness of the right-hand side of
\eqref{eq:RLD-QFI-breakdown-cqq}, we can now focus on establishing the
equality under these conditions. Consider that
\begin{align}
&  \left(  \partial_{\theta}\left(  \sum_{x}p_{\theta}(x)|x\rangle\!\langle
x|_{X}\otimes\rho_{\theta}^{x}\right)  \right)  ^{2}\\
&  =\left(  \sum_{x}|x\rangle\!\langle x|_{X}\otimes\left[  \partial_{\theta
}(p_{\theta}(x))\rho_{\theta}^{x}+p_{\theta}(x)(\partial_{\theta}\rho_{\theta
}^{x})\right]  \right)  ^{2}\\
&  =\left(  \sum_{x}|x\rangle\!\langle x|_{X}\otimes\left[  \left[
\partial_{\theta}(p_{\theta}(x))\right]  ^{2}[\rho_{\theta}^{x}]^{2}
+p_{\theta}(x)\partial_{\theta}(p_{\theta}(x))\left\{  \rho_{\theta}
^{x},(\partial_{\theta}\rho_{\theta}^{x})\right\}  +[p_{\theta}(x)]^{2}
(\partial_{\theta}\rho_{\theta}^{x})^{2}\right]  \right)  .
\end{align}
Then we find that
\begin{align}
&  \widehat{I}_{F}\!\left(  \theta;\left\{  \sum_{x}p_{\theta}(x)|x\rangle
\langle x|_{X}\otimes\rho_{\theta}^{x}\right\}  _{\theta}\right) \nonumber\\
&  =\operatorname{Tr}\left[  \left(  \partial_{\theta}\left(  \sum
_{x}p_{\theta}(x)|x\rangle\!\langle x|_{X}\otimes\rho_{\theta}^{x}\right)
\right)  ^{2}\left(  \sum_{x:p_{\theta}(x)>0}p_{\theta}(x)|x\rangle\!\langle
x|_{X}\otimes\rho_{\theta}^{x}\right)  ^{-1}\right] \\
&  =\operatorname{Tr}\left[  \left(  \partial_{\theta}\left(  \sum
_{x}p_{\theta}(x)|x\rangle\!\langle x|_{X}\otimes\rho_{\theta}^{x}\right)
\right)  ^{2}\left(  \sum_{x:p_{\theta}(x)>0}|x\rangle\!\langle x|_{X}
\otimes\lbrack p_{\theta}(x)]^{-1}[\rho_{\theta}^{x}]^{-1}\right)  \right] \\
&  =\sum_{x:p_{\theta}(x)>0}\operatorname{Tr}[\left[  \partial_{\theta
}(p_{\theta}(x))\right]  ^{2}[\rho_{\theta}^{x}]^{2}[p_{\theta}(x)]^{-1}
[\rho_{\theta}^{x}]^{-1}]\nonumber\\
&  \qquad+\sum_{x:p_{\theta}(x)>0}\operatorname{Tr}[p_{\theta}(x)\partial
_{\theta}(p_{\theta}(x))\left\{  \rho_{\theta}^{x},(\partial_{\theta}
\rho_{\theta}^{x})\right\}  [p_{\theta}(x)]^{-1}[\rho_{\theta}^{x}
]^{-1}]\nonumber\\
&  \qquad+\sum_{x:p_{\theta}(x)>0}\operatorname{Tr}[[p_{\theta}(x)]^{2}
(\partial_{\theta}\rho_{\theta}^{x})^{2}[p_{\theta}(x)]^{-1}[\rho_{\theta}
^{x}]^{-1}]\\
&  =\sum_{x:p_{\theta}(x)>0}\left[  \frac{\left[  \partial_{\theta}(p_{\theta
}(x))\right]  ^{2}}{p_{\theta}(x)}\operatorname{Tr}[\rho_{\theta}
^{x}]+2\partial_{\theta}(p_{\theta}(x))\operatorname{Tr}[(\partial_{\theta
}\rho_{\theta}^{x})\Pi_{\rho_{\theta}^{x}}]+p_{\theta}(x)\operatorname{Tr}
[(\partial_{\theta}\rho_{\theta}^{x})^{2}[\rho_{\theta}^{x}]^{-1}]\right] \\
&  =I_{F}(\theta;\{p_{\theta}\}_{\theta})+2\sum_{x:p_{\theta}(x)>0}
\partial_{\theta}(p_{\theta}(x))\operatorname{Tr}[\partial_{\theta}
\rho_{\theta}^{x}]+\sum_{x:p_{\theta}(x)>0}p_{\theta}(x)\widehat{I}_{F}
(\theta;\{\rho_{\theta}^{x}\}_{\theta})\\
&  =I_{F}(\theta;\{p_{\theta}\}_{\theta})+\sum_{x:p_{\theta}(x)>0}p_{\theta
}(x)\widehat{I}_{F}(\theta;\{\rho_{\theta}^{x}\}_{\theta}).
\end{align}
The second-to-last equality follows because $\operatorname{Tr}[(\partial
_{\theta}\rho_{\theta}^{x})\Pi_{\rho_{\theta}^{x}}^{\perp}]=0$ and so we can
add this term to the sum. The last equality follows because $\operatorname{Tr}
[\partial_{\theta}\rho_{\theta}^{x}]=\partial_{\theta}\operatorname{Tr}
[\rho_{\theta}^{x}]=0$.
\end{proof}

\section{Proof of Proposition~\ref{prop:SDP-SLD-Fish} (Bilinear program\ for
SLD\ Fisher information of quantum channels)}

\label{app:SDP-SLD-Fish-ch}

Recall that the Fisher information of channels is defined as the following
optimization over pure state inputs:
\begin{equation}
I_{F}(\theta;\{\mathcal{N}_{A\rightarrow B}^{\theta}\})=\sup_{\psi_{RA}}
I_{F}(\theta;\{\mathcal{N}_{A\rightarrow B}^{\theta}(\psi_{RA})\}).
\end{equation}
It suffices to optimize over pure state inputs $\psi_{RA}$ such that the
reduced state $\psi_{R}>0$, because this set is dense in the set of all pure
bipartite states. Now consider a fixed input state $\psi_{RA}$, and recall
that it can be written as follows:
\begin{equation}
\psi_{RA}=Z_{R}\Gamma_{RA}Z_{R}^{\dag},
\end{equation}
where $Z_{R}$ is an invertible operator satisfying $\operatorname{Tr}
[Z_{R}^{\dag}Z_{R}]=1$. Then the output state is as follows:
\begin{equation}
\omega_{RB}^{\theta}:=\mathcal{N}_{A\rightarrow B}^{\theta}(\psi_{RA}
)=Z_{R}\Gamma_{RB}^{\mathcal{N}^{\theta}}Z_{R}^{\dag},
\end{equation}
and we find that
\begin{multline}
\frac{1}{2}I_{F}(\theta;\{\mathcal{N}_{A\rightarrow B}^{\theta}(\psi
_{RA})\})\label{eq:SDP-SLD-Fisher-fixed-state}\\
=\inf\left\{  \mu:
\begin{bmatrix}
\mu & \langle\Gamma|_{RR^{\prime}BB^{\prime}}\left(  \partial_{\theta}
\omega_{RB}^{\theta}\otimes I_{R^{\prime}B^{\prime}}\right) \\
\left(  \partial_{\theta}\omega_{RB}^{\theta}\otimes I_{R^{\prime}B^{\prime}
}\right)  |\Gamma\rangle_{RR^{\prime}BB^{\prime}} & \omega_{RB}^{\theta
}\otimes I_{R^{\prime}B^{\prime}}+I_{RB}\otimes(\omega_{R^{\prime}B^{\prime}
}^{\theta})^{T}
\end{bmatrix}
\geq0\right\}  ,
\end{multline}
by applying Proposition~\ref{prop:SLD-Fish-states-SDP}. Now consider that
\begin{align}
&
\begin{bmatrix}
\mu & \langle\Gamma|_{RR^{\prime}BB^{\prime}}\left(  \partial_{\theta}
\omega_{RB}^{\theta}\otimes I_{R^{\prime}B^{\prime}}\right) \\
\left(  \partial_{\theta}\omega_{RB}^{\theta}\otimes I_{R^{\prime}B^{\prime}
}\right)  |\Gamma\rangle_{RR^{\prime}BB^{\prime}} & \omega_{RB}^{\theta
}\otimes I_{R^{\prime}B^{\prime}}+I_{RB}\otimes(\omega_{R^{\prime}B^{\prime}
}^{\theta})^{T}
\end{bmatrix}
\nonumber\\
&  =
\begin{bmatrix}
\mu & \langle\Gamma|_{RR^{\prime}BB^{\prime}}\left(  Z_{R}(\partial_{\theta
}\Gamma_{RB}^{\mathcal{N}^{\theta}})Z_{R}^{\dag}\otimes I_{R^{\prime}
B^{\prime}}\right) \\
\left(  Z_{R}(\partial_{\theta}\Gamma_{RB}^{\mathcal{N}^{\theta}})Z_{R}^{\dag
}\otimes I_{R^{\prime}B^{\prime}}\right)  |\Gamma\rangle_{RR^{\prime
}BB^{\prime}} & Z_{R}\Gamma_{RB}^{\mathcal{N}^{\theta}}Z_{R}^{\dag}\otimes
I_{R^{\prime}B^{\prime}}+I_{RB}\otimes\overline{Z}_{R^{\prime}}(\Gamma
_{R^{\prime}B^{\prime}}^{\mathcal{N}^{\theta}})^{T}Z_{R^{\prime}}^{T}
\end{bmatrix}
\label{eq:bilinear-SLD-Fish-ch-proof-1}\\
&  =
\begin{bmatrix}
1 & 0\\
0 & Z_{R}\otimes I_{B}\otimes\overline{Z}_{R^{\prime}}\otimes I_{B^{\prime}}
\end{bmatrix}
\times\nonumber\\
&  \qquad
\begin{bmatrix}
\mu & \langle\Gamma|_{RR^{\prime}BB^{\prime}}\left(  (\partial_{\theta}
\Gamma_{RB}^{\mathcal{N}^{\theta}})\otimes I_{R^{\prime}B^{\prime}}\right) \\
\left(  (\partial_{\theta}\Gamma_{RB}^{\mathcal{N}^{\theta}})\otimes
I_{R^{\prime}B^{\prime}}\right)  |\Gamma\rangle_{RR^{\prime}BB^{\prime}} &
\Gamma_{RB}^{\mathcal{N}^{\theta}}\otimes\sigma_{R^{\prime}}^{-T}\otimes
I_{B^{\prime}}+\sigma_{R}^{-1}\otimes I_{B}\otimes(\Gamma_{R^{\prime}
B^{\prime}}^{\mathcal{N}^{\theta}})^{T}
\end{bmatrix}
\times\nonumber\\
&  \qquad
\begin{bmatrix}
1 & 0\\
0 & Z_{R}\otimes I_{B}\otimes\overline{Z}_{R^{\prime}}\otimes I_{B^{\prime}}
\end{bmatrix}
^{\dag}, \label{eq:bilinear-SLD-Fish-ch-proof-2}
\end{align}
where we define
\begin{equation}
\sigma_{R}:=Z_{R}^{\dag}Z_{R},
\end{equation}
and we applied the following observations:
\begin{align}
&  \left(  Z_{R}(\partial_{\theta}\Gamma_{RB}^{\mathcal{N}^{\theta}}
)Z_{R}^{\dag}\otimes I_{R^{\prime}B^{\prime}}\right)  |\Gamma\rangle
_{RR^{\prime}BB^{\prime}}\nonumber\\
&  =\left(  Z_{R}(\partial_{\theta}\Gamma_{RB}^{\mathcal{N}^{\theta}}
)\otimes\overline{Z}_{R^{\prime}}\otimes I_{B^{\prime}}\right)  |\Gamma
\rangle_{RR^{\prime}BB^{\prime}}\\
&  =\left(  Z_{R}\otimes\overline{Z}_{R^{\prime}}\right)  \left(
(\partial_{\theta}\Gamma_{RB}^{\mathcal{N}^{\theta}})\otimes I_{R^{\prime
}B^{\prime}}\right)  |\Gamma\rangle_{RR^{\prime}BB^{\prime}},
\end{align}
\begin{align}
&  Z_{R}\Gamma_{RB}^{\mathcal{N}^{\theta}}Z_{R}^{\dag}\otimes I_{R^{\prime
}B^{\prime}}+I_{RB}\otimes\overline{Z}_{R^{\prime}}(\Gamma_{R^{\prime
}B^{\prime}}^{\mathcal{N}^{\theta}})^{T}Z_{R^{\prime}}^{T}\nonumber\\
&  =Z_{R}\Gamma_{RB}^{\mathcal{N}^{\theta}}Z_{R}^{\dag}\otimes\overline
{Z}_{R^{\prime}}\left(  \overline{Z}_{R^{\prime}}\right)  ^{-1}\left(
Z_{R^{\prime}}^{T}\right)  ^{-1}Z_{R^{\prime}}^{T}\otimes I_{B^{\prime}
}\nonumber\\
&  \quad+Z_{R}\left(  Z_{R}\right)  ^{-1}\left(  Z_{R}^{\dag}\right)
^{-1}Z_{R}^{\dag}\otimes I_{B}\otimes\overline{Z}_{R^{\prime}}(\Gamma
_{R^{\prime}B^{\prime}}^{\mathcal{N}^{\theta}})^{T}Z_{R^{\prime}}^{T}\\
&  =Z_{R}\Gamma_{RB}^{\mathcal{N}^{\theta}}Z_{R}^{\dag}\otimes\overline
{Z}_{R^{\prime}}\sigma_{R}^{-T}Z_{R^{\prime}}^{T}\otimes I_{B^{\prime}}
+Z_{R}\sigma_{R}^{-1}Z_{R}^{\dag}\otimes I_{B}\otimes\overline{Z}_{R^{\prime}
}(\Gamma_{R^{\prime}B^{\prime}}^{\mathcal{N}^{\theta}})^{T}Z_{R^{\prime}}
^{T}\\
&  =\left(  Z_{R}\otimes I_{B}\otimes\overline{Z}_{R^{\prime}}\otimes
I_{B^{\prime}}\right)  \left(  \Gamma_{RB}^{\mathcal{N}^{\theta}}\otimes
\sigma_{R}^{-T}\otimes I_{B^{\prime}}\right)  \left(  Z_{R}\otimes
I_{B}\otimes\overline{Z}_{R^{\prime}}\otimes I_{B^{\prime}}\right)  ^{\dag
}\nonumber\\
&  \quad+\left(  Z_{R}\otimes I_{B}\otimes\overline{Z}_{R^{\prime}}\otimes
I_{B^{\prime}}\right)  \left(  \sigma_{R}^{-1}\otimes I_{B}\otimes
(\Gamma_{R^{\prime}B^{\prime}}^{\mathcal{N}^{\theta}})^{T}\right)  \left(
Z_{R}\otimes I_{B}\otimes\overline{Z}_{R^{\prime}}\otimes I_{B^{\prime}
}\right)  ^{\dag}\\
&  =\left(  Z_{R}\otimes I_{B}\otimes\overline{Z}_{R^{\prime}}\otimes
I_{B^{\prime}}\right)  \left(  \Gamma_{RB}^{\mathcal{N}^{\theta}}\otimes
\sigma_{R}^{-T}\otimes I_{B^{\prime}}+\sigma_{R}^{-1}\otimes I_{B}
\otimes(\Gamma_{R^{\prime}B^{\prime}}^{\mathcal{N}^{\theta}})^{T}\right)
\nonumber\\
&  \quad\times\left(  Z_{R}\otimes I_{B}\otimes\overline{Z}_{R^{\prime}
}\otimes I_{B^{\prime}}\right)  ^{\dag}.
\end{align}
Since the first matrix in
\eqref{eq:bilinear-SLD-Fish-ch-proof-1}--\eqref{eq:bilinear-SLD-Fish-ch-proof-2}
above is positive semi-definite if and only if the last one is, the
semi-definite program in \eqref{eq:SDP-SLD-Fisher-fixed-state} becomes as
follows:
\begin{equation}
\inf\left\{  \mu:
\begin{bmatrix}
\mu & \langle\Gamma|_{RR^{\prime}BB^{\prime}}\left(  (\partial_{\theta}
\Gamma_{RB}^{\mathcal{N}^{\theta}})\otimes I_{R^{\prime}B^{\prime}}\right) \\
\left(  (\partial_{\theta}\Gamma_{RB}^{\mathcal{N}^{\theta}})\otimes
I_{R^{\prime}B^{\prime}}\right)  |\Gamma\rangle_{RR^{\prime}BB^{\prime}} &
\Gamma_{RB}^{\mathcal{N}^{\theta}}\otimes\sigma_{R^{\prime}}^{-T}\otimes
I_{B^{\prime}}+\sigma_{R}^{-1}\otimes I_{B}\otimes(\Gamma_{R^{\prime}
B^{\prime}}^{\mathcal{N}^{\theta}})^{T}
\end{bmatrix}
\geq0\right\}  . \label{eq:SDP-first-primal-SLD-Fish-ch}
\end{equation}
By invoking Lemma~\ref{lem:freq-used-SDP-primal-dual}, the dual of this
program is given by
\begin{multline}
\sup_{\lambda,|\varphi\rangle_{RBR^{\prime}B^{\prime}},W_{RBR^{\prime
}B^{\prime}}}2\operatorname{Re}[\langle\varphi|_{RBR^{\prime}B^{\prime}
}(\partial_{\theta}\Gamma_{RB}^{\mathcal{N}^{\theta}})|\Gamma\rangle
_{RR^{\prime}BB^{\prime}}]\label{eq:SDP-first-dual-SLD-Fish-ch}\\
-\operatorname{Tr}[(\Gamma_{RB}^{\mathcal{N}^{\theta}}\otimes\sigma
_{R^{\prime}}^{-T}\otimes I_{B^{\prime}}+\sigma_{R}^{-1}\otimes I_{B}
\otimes(\Gamma_{R^{\prime}B^{\prime}}^{\mathcal{N}^{\theta}})^{T}
)W_{RBR^{\prime}B^{\prime}}]
\end{multline}
subject to
\begin{equation}
\lambda\leq1,\qquad
\begin{bmatrix}
\lambda & \langle\varphi|_{RBR^{\prime}B^{\prime}}\\
|\varphi\rangle_{RBR^{\prime}B^{\prime}} & W_{RBR^{\prime}B^{\prime}}
\end{bmatrix}
\geq0. \label{eq:SDP-first-dual-SLD-Fish-ch-constraints}
\end{equation}
Strong duality holds, so that \eqref{eq:SDP-first-dual-SLD-Fish-ch} is equal
to \eqref{eq:SDP-first-primal-SLD-Fish-ch}, because we are free to choose
values $\lambda$, $|\varphi\rangle_{RBR^{\prime}B^{\prime}}$, and
$W_{RBR^{\prime}B^{\prime}}$ such that the constraints in
\eqref{eq:SDP-first-dual-SLD-Fish-ch-constraints}\ are strict. Employing the
unitary swap operators $F_{RR^{\prime}}$ and $F_{BB^{\prime}}$, we can rewrite
the second term in the objective function as follows:
\begin{align}
&  \operatorname{Tr}[(\Gamma_{RB}^{\mathcal{N}^{\theta}}\otimes\sigma
_{R^{\prime}}^{-T}\otimes I_{B^{\prime}}+\sigma_{R}^{-1}\otimes I_{B}
\otimes(\Gamma_{R^{\prime}B^{\prime}}^{\mathcal{N}^{\theta}})^{T}
)W_{RBR^{\prime}B^{\prime}}]\nonumber\\
&  =\operatorname{Tr}[(\Gamma_{RB}^{\mathcal{N}^{\theta}}\otimes
\sigma_{R^{\prime}}^{-T}\otimes I_{B^{\prime}})W_{RBR^{\prime}B^{\prime}
}]+\operatorname{Tr}[(\sigma_{R}^{-1}\otimes I_{B}\otimes(\Gamma_{R^{\prime
}B^{\prime}}^{\mathcal{N}^{\theta}})^{T})W_{RBR^{\prime}B^{\prime}}]\\
&  =\operatorname{Tr}[(\left(  F_{RR^{\prime}}\otimes F_{BB^{\prime}}\right)
(\sigma_{R}^{-T}\otimes I_{B}\otimes\Gamma_{R^{\prime}B^{\prime}}
^{\mathcal{N}^{\theta}})\left(  F_{RR^{\prime}}\otimes F_{BB^{\prime}}\right)
W_{RBR^{\prime}B^{\prime}}]\nonumber\\
&  \qquad+\operatorname{Tr}[\sigma_{R}^{-1}\operatorname{Tr}_{BR^{\prime
}B^{\prime}}[(\Gamma_{R^{\prime}B^{\prime}}^{\mathcal{N}^{\theta}}
)^{T}W_{RBR^{\prime}B^{\prime}}]]\\
&  =\operatorname{Tr}[((\sigma_{R}^{-T}\otimes I_{B}\otimes\Gamma_{R^{\prime
}B^{\prime}}^{\mathcal{N}^{\theta}})\left(  F_{RR^{\prime}}\otimes
F_{BB^{\prime}}\right)  W_{RBR^{\prime}B^{\prime}}\left(  F_{RR^{\prime}
}\otimes F_{BB^{\prime}}\right)  ]\nonumber\\
&  \qquad+\operatorname{Tr}[\sigma_{R}^{-1}\operatorname{Tr}_{BR^{\prime
}B^{\prime}}[(\Gamma_{R^{\prime}B^{\prime}}^{\mathcal{N}^{\theta}}
)^{T}W_{RBR^{\prime}B^{\prime}}]]\\
&  =\operatorname{Tr}[\sigma_{R}^{-T}\operatorname{Tr}_{BR^{\prime}B^{\prime}
}[\Gamma_{R^{\prime}B^{\prime}}^{\mathcal{N}^{\theta}}\left(  F_{RR^{\prime}
}\otimes F_{BB^{\prime}}\right)  W_{RBR^{\prime}B^{\prime}}\left(
F_{RR^{\prime}}\otimes F_{BB^{\prime}}\right)  ]]\nonumber\\
&  \qquad+\operatorname{Tr}[\sigma_{R}^{-1}\operatorname{Tr}_{BR^{\prime
}B^{\prime}}[(\Gamma_{R^{\prime}B^{\prime}}^{\mathcal{N}^{\theta}}
)^{T}W_{RBR^{\prime}B^{\prime}}]]\\
&  =\operatorname{Tr}[\sigma_{R}^{-1}(\operatorname{Tr}_{BR^{\prime}B^{\prime
}}[\Gamma_{R^{\prime}B^{\prime}}^{\mathcal{N}^{\theta}}\left(  F_{RR^{\prime}
}\otimes F_{BB^{\prime}}\right)  W_{RBR^{\prime}B^{\prime}}\left(
F_{RR^{\prime}}\otimes F_{BB^{\prime}}\right)  ])^{T}]\nonumber\\
&  \qquad+\operatorname{Tr}[\sigma_{R}^{-1}\operatorname{Tr}_{BR^{\prime
}B^{\prime}}[(\Gamma_{R^{\prime}B^{\prime}}^{\mathcal{N}^{\theta}}
)^{T}W_{RBR^{\prime}B^{\prime}}]]\\
&  =\operatorname{Tr}[\sigma_{R}^{-1}K_{R}],
\end{align}
where
\begin{multline}
K_{R}=(\operatorname{Tr}_{BR^{\prime}B^{\prime}}[\Gamma_{R^{\prime}B^{\prime}
}^{\mathcal{N}^{\theta}}\left(  F_{RR^{\prime}}\otimes F_{BB^{\prime}}\right)
W_{RBR^{\prime}B^{\prime}}\left(  F_{RR^{\prime}}\otimes F_{BB^{\prime}
}\right)  ])^{T}\\
+\operatorname{Tr}_{BR^{\prime}B^{\prime}}[(\Gamma_{R^{\prime}B^{\prime}
}^{\mathcal{N}^{\theta}})^{T}W_{RBR^{\prime}B^{\prime}}].
\end{multline}
So the SDP\ in \eqref{eq:SDP-first-dual-SLD-Fish-ch}\ can be written as
\begin{equation}
\sup_{\lambda,|\varphi\rangle_{RBR^{\prime}B^{\prime}},W_{RBR^{\prime
}B^{\prime}}}2\operatorname{Re}[\langle\varphi|_{RBR^{\prime}B^{\prime}
}(\partial_{\theta}\Gamma_{RB}^{\mathcal{N}^{\theta}})|\Gamma\rangle
_{RR^{\prime}BB^{\prime}}]-\operatorname{Tr}[\sigma_{R}^{-1}K_{R}]
\label{eq:SDP-first-dual-SLD-Fish-ch-2}
\end{equation}
subject to
\begin{equation}
\lambda\leq1,\qquad
\begin{bmatrix}
\lambda & \langle\varphi|_{RBR^{\prime}B^{\prime}}\\
|\varphi\rangle_{RBR^{\prime}B^{\prime}} & W_{RBR^{\prime}B^{\prime}}
\end{bmatrix}
\geq0.
\end{equation}
Now noting from Lemma~\ref{lem:min-XYinvX} that
\begin{equation}
\sigma_{R}^{-1}=\inf\left\{  Y_{R}:
\begin{bmatrix}
\sigma_{R} & I_{R}\\
I_{R} & Y_{R}
\end{bmatrix}
\geq0\right\}  ,
\end{equation}
and that $\sigma_{R}^{-1}$ and $K_{R}$ are positive semi-definite, we can
rewrite the SDP in \eqref{eq:SDP-first-dual-SLD-Fish-ch-2}\ as
\begin{multline}
\sup_{\lambda,|\varphi\rangle_{RBR^{\prime}B^{\prime}},W_{RBR^{\prime
}B^{\prime}}}\left(  2\operatorname{Re}[\langle\varphi|_{RBR^{\prime}
B^{\prime}}(\partial_{\theta}\Gamma_{RB}^{\mathcal{N}^{\theta}})|\Gamma
\rangle_{RR^{\prime}BB^{\prime}}]-\inf_{Y_{R}}\operatorname{Tr}[Y_{R}
K_{R}]\right) \label{eq:almost-there-SLD-Fisher-channels-SDP}\\
=\sup_{\lambda,|\varphi\rangle_{RBR^{\prime}B^{\prime}},W_{RBR^{\prime
}B^{\prime}},Y_{R}}\left(  2\operatorname{Re}[\langle\varphi|_{RBR^{\prime
}B^{\prime}}(\partial_{\theta}\Gamma_{RB}^{\mathcal{N}^{\theta}}
)|\Gamma\rangle_{RR^{\prime}BB^{\prime}}]-\operatorname{Tr}[Y_{R}
K_{R}]\right)
\end{multline}
subject to
\begin{equation}
\lambda\leq1,\quad
\begin{bmatrix}
\lambda & \langle\varphi|_{RBR^{\prime}B^{\prime}}\\
|\varphi\rangle_{RBR^{\prime}B^{\prime}} & W_{RBR^{\prime}B^{\prime}}
\end{bmatrix}
\geq0,\quad
\begin{bmatrix}
\sigma_{R} & I_{R}\\
I_{R} & Y_{R}
\end{bmatrix}
\geq0.
\end{equation}
Then we can finally include the maximization over input states $\sigma_{R}$
(satisfying $\sigma_{R}\geq0$ and $\operatorname{Tr}[\sigma_{R}]=1$) to arrive
at the form given in \eqref{eq:SDP-SLD-ch-Fish}.

\section{Proof of Propositions \ref{prop:geo-fish-explicit-formula}\ and
\ref{prop:add-RLD-Fish-ch} (Formula for RLD\ Fisher information of quantum
channels and its additivity)}

\label{app:proofs-RLD-fish-props}

\begin{proof}
[Proof of Proposition~\ref{prop:geo-fish-explicit-formula}]From
\eqref{eq:finiteness-condition-RLD-fish-ch}, the finiteness condition for the
RLD Fisher information $I_{F}(\theta;\{\mathcal{N}_{A\rightarrow B}^{\theta
}\}_{\theta})$ of the family $\{\mathcal{N}_{A\rightarrow B}^{\theta
}\}_{\theta}$ of channels is that $\Pi_{\Gamma^{\mathcal{N}^{\theta}}}^{\perp
}(\partial_{\theta}\Gamma_{RB}^{\mathcal{N}^{\theta}})=0$, where $\Gamma
_{RB}^{\mathcal{N}^{\theta}}$ is the Choi state of the channel $\mathcal{N}
_{A\rightarrow B}^{\theta}$. So we suppose that this condition holds. This
condition implies that $(\partial_{\theta}\Gamma_{RB}^{\mathcal{N}^{\theta}
})(\Gamma_{RB}^{\mathcal{N}^{\theta}})^{-1}(\partial_{\theta}\Gamma
_{RB}^{\mathcal{N}^{\theta}})$ is a well-defined operator with the inverse
taken on the support of $(\Gamma_{RB}^{\mathcal{N}^{\theta}})^{-1}$. Recall
that any pure state $\psi_{RA}$ can be written as
\begin{equation}
\psi_{RA}=Z_{R}\Gamma_{RA}Z_{R}^{\dag},
\end{equation}
where
\begin{align}
\Gamma_{RA}  &  =|\Gamma\rangle\!\langle\Gamma|_{RA},\\
|\Gamma\rangle_{RA}  &  =\sum_{i}|i\rangle_{R}|i\rangle_{A},
\end{align}
and $Z_{R}$ is a square operator satisfying $\operatorname{Tr}[Z_{R}^{\dag
}Z_{R}]=1$. This implies that
\begin{equation}
\mathcal{N}_{A\rightarrow B}^{\theta}(\psi_{RA})=\mathcal{N}_{A\rightarrow
B}^{\theta}(Z_{R}\Gamma_{RA}Z_{R}^{\dag})=Z_{R}\mathcal{N}_{A\rightarrow
B}^{\theta}(\Gamma_{RA})Z_{R}^{\dag}=Z_{R}\Gamma_{RB}^{\mathcal{N}^{\theta}
}Z_{R}^{\dag}.
\end{equation}
It suffices to optimize over pure states $\psi_{RA}$ such that $\psi_{A}>0$
because these states are dense in the set of all pure bipartite states. Then
consider that
\begin{align}
&  \sup_{\psi_{RA}}\widehat{I}_{F}(\theta;\{\mathcal{N}_{A\rightarrow
B}^{\theta}(\psi_{RA})\}_{\theta})\nonumber\\
&  =\sup_{\psi_{RA}}\operatorname{Tr}[(\partial_{\theta}\mathcal{N}
_{A\rightarrow B}^{\theta}(\psi_{RA}))^{2}(\mathcal{N}_{A\rightarrow
B}^{\theta}(\psi_{RA}))^{-1}]\\
&  =\sup_{Z_{R}:\operatorname{Tr}[Z_{R}^{\dag}Z_{R}]=1}\operatorname{Tr}
[(\partial_{\theta}Z_{R}\Gamma_{RB}^{\mathcal{N}^{\theta}}Z_{R}^{\dag}
)^{2}(Z_{R}\Gamma_{RB}^{\mathcal{N}^{\theta}}Z_{R}^{\dag})^{-1}]\\
&  =\sup_{Z_{R}:\operatorname{Tr}[Z_{R}^{\dag}Z_{R}]=1}\operatorname{Tr}
[(Z_{R}(\partial_{\theta}\Gamma_{RB}^{\mathcal{N}^{\theta}})Z_{R}^{\dag}
)^{2}(Z_{R}\Gamma_{RB}^{\mathcal{N}^{\theta}}Z_{R}^{\dag})^{-1}]\\
&  =\sup_{Z_{R}:\operatorname{Tr}[Z_{R}^{\dag}Z_{R}]=1}\operatorname{Tr}
[(Z_{R}(\partial_{\theta}\Gamma_{RB}^{\mathcal{N}^{\theta}})Z_{R}^{\dag
})(Z_{R}\Gamma_{RB}^{\mathcal{N}^{\theta}}Z_{R}^{\dag})^{-1}(Z_{R}
(\partial_{\theta}\Gamma_{RB}^{\mathcal{N}^{\theta}})Z_{R}^{\dag})]\\
&  =\sup_{Z_{R}:\operatorname{Tr}[Z_{R}^{\dag}Z_{R}]=1}\operatorname{Tr}
[Z_{R}(\partial_{\theta}\Gamma_{RB}^{\mathcal{N}^{\theta}})(\Gamma
_{RB}^{\mathcal{N}^{\theta}})^{-1}(\partial_{\theta}\Gamma_{RB}^{\mathcal{N}
^{\theta}})Z_{R}^{\dag}]\\
&  =\sup_{Z_{R}:\operatorname{Tr}[Z_{R}^{\dag}Z_{R}]=1}\operatorname{Tr}
[Z_{R}^{\dag}Z_{R}\operatorname{Tr}_{B}[(\partial_{\theta}\Gamma
_{RB}^{\mathcal{N}^{\theta}})(\Gamma_{RB}^{\mathcal{N}^{\theta}}
)^{-1}(\partial_{\theta}\Gamma_{RB}^{\mathcal{N}^{\theta}})]]\\
&  =\left\Vert \operatorname{Tr}_{B}[(\partial_{\theta}\Gamma_{RB}
^{\mathcal{N}^{\theta}})(\Gamma_{RB}^{\mathcal{N}^{\theta}})^{-1}
(\partial_{\theta}\Gamma_{RB}^{\mathcal{N}^{\theta}})]\right\Vert _{\infty}.
\end{align}
The fifth equality is a consequence of the transformer equality in
Lemma~\ref{lem:transformer-ineq-basic}, with $L=Z_{R}$, $X=\partial_{\theta
}\Gamma_{RB}^{\mathcal{N}^{\theta}}$, and $Y=\Gamma_{RB}^{\mathcal{N}^{\theta
}}$. The last equality is a consequence of the characterization of the
infinity norm of a positive semi-definite operator $Y$ as $\left\Vert
Y\right\Vert _{\infty}=\sup_{\rho>0,\operatorname{Tr}[\rho]=1}
\operatorname{Tr}[Y\rho]$.
\end{proof}

\bigskip

\begin{proof}
[Proof of Proposition~\ref{prop:add-RLD-Fish-ch}]The proof begins by
considering the finiteness condition in
\eqref{eq:finiteness-condition-RLD-fish-ch} and showing that finiteness of the
left-hand side is equivalent to finiteness of the right-hand side. The
manipulations are the same as given in the proof of
Proposition~\ref{prop:additivity-SLD-RLD-states}, and so we omit showing them
again. So we can focus on the case when the quantities are finite and exploit
the explicit formula from Proposition~\ref{prop:geo-fish-explicit-formula} to
evaluate the left-hand side directly. Consider that
\begin{multline}
\widehat{I}_{F}(\theta;\{\mathcal{N}_{A\rightarrow B}^{\theta}\otimes
\mathcal{M}_{C\rightarrow D}^{\theta}\}_{\theta})\\
=\left\Vert \operatorname{Tr}_{BD}[(\partial_{\theta}(\Gamma_{RB}
^{\mathcal{N}^{\theta}}\otimes\Gamma_{SD}^{\mathcal{M}^{\theta}}))(\Gamma
_{RB}^{\mathcal{N}^{\theta}}\otimes\Gamma_{SD}^{\mathcal{M}^{\theta}}
)^{-1}(\partial_{\theta}(\Gamma_{RB}^{\mathcal{N}^{\theta}}\otimes\Gamma
_{SD}^{\mathcal{M}^{\theta}}))]\right\Vert _{\infty},
\end{multline}
because the Choi operator of the tensor-product channel $\mathcal{N}
_{A\rightarrow B}^{\theta}\otimes\mathcal{M}_{C\rightarrow D}^{\theta}$ is
$\Gamma_{RB}^{\mathcal{N}^{\theta}}\otimes\Gamma_{SD}^{\mathcal{M}^{\theta}}$.
Then
\begin{equation}
\partial_{\theta}(\Gamma_{RB}^{\mathcal{N}^{\theta}}\otimes\Gamma
_{SD}^{\mathcal{M}^{\theta}})=(\partial_{\theta}\Gamma_{RB}^{\mathcal{N}
^{\theta}})\otimes\Gamma_{SD}^{\mathcal{M}^{\theta}}+\Gamma_{RB}
^{\mathcal{N}^{\theta}}\otimes\partial_{\theta}(\Gamma_{SD}^{\mathcal{M}
^{\theta}}),
\end{equation}
and right multiplying by $(\Gamma_{RB}^{\mathcal{N}^{\theta}}\otimes
\Gamma_{SD}^{\mathcal{M}^{\theta}})^{-1}$ gives
\begin{align}
&  (\partial_{\theta}(\Gamma_{RB}^{\mathcal{N}^{\theta}}\otimes\Gamma
_{SD}^{\mathcal{M}^{\theta}}))(\Gamma_{RB}^{\mathcal{N}^{\theta}}\otimes
\Gamma_{SD}^{\mathcal{M}^{\theta}})^{-1}\nonumber\\
&  =\left[  (\partial_{\theta}\Gamma_{RB}^{\mathcal{N}^{\theta}})\otimes
\Gamma_{SD}^{\mathcal{M}^{\theta}}+\Gamma_{RB}^{\mathcal{N}^{\theta}}
\otimes\partial_{\theta}(\Gamma_{SD}^{\mathcal{M}^{\theta}})\right]
(\Gamma_{RB}^{\mathcal{N}^{\theta}}\otimes\Gamma_{SD}^{\mathcal{M}^{\theta}
})^{-1}\\
&  =(\partial_{\theta}\Gamma_{RB}^{\mathcal{N}^{\theta}})(\Gamma
_{RB}^{\mathcal{N}^{\theta}})^{-1}\otimes\Gamma_{SD}^{\mathcal{M}^{\theta}
}(\Gamma_{SD}^{\mathcal{M}^{\theta}})^{-1}+\Gamma_{RB}^{\mathcal{N}^{\theta}
}(\Gamma_{RB}^{\mathcal{N}^{\theta}})^{-1}\otimes\partial_{\theta}(\Gamma
_{SD}^{\mathcal{M}^{\theta}})(\Gamma_{SD}^{\mathcal{M}^{\theta}})^{-1}\\
&  =(\partial_{\theta}\Gamma_{RB}^{\mathcal{N}^{\theta}})(\Gamma
_{RB}^{\mathcal{N}^{\theta}})^{-1}\otimes\Pi_{\Gamma^{\mathcal{M}^{\theta}}
}+\Pi_{\Gamma^{\mathcal{N}^{\theta}}}\otimes\partial_{\theta}(\Gamma
_{SD}^{\mathcal{M}^{\theta}})(\Gamma_{SD}^{\mathcal{M}^{\theta}})^{-1}.
\end{align}
Right multiplying again by $(\partial_{\theta}(\Gamma_{RB}^{\mathcal{N}
^{\theta}}\otimes\Gamma_{SD}^{\mathcal{M}^{\theta}}))$ gives
\begin{align}
&  \left[  (\partial_{\theta}\Gamma_{RB}^{\mathcal{N}^{\theta}})(\Gamma
_{RB}^{\mathcal{N}^{\theta}})^{-1}\otimes\Pi_{\Gamma^{\mathcal{M}^{\theta}}
}+\Pi_{\Gamma^{\mathcal{N}^{\theta}}}\otimes\partial_{\theta}(\Gamma
_{SD}^{\mathcal{M}^{\theta}})(\Gamma_{SD}^{\mathcal{M}^{\theta}})^{-1}\right]
(\partial_{\theta}(\Gamma_{RB}^{\mathcal{N}^{\theta}}\otimes\Gamma
_{SD}^{\mathcal{M}^{\theta}}))\nonumber\\
&  =\left[  (\partial_{\theta}\Gamma_{RB}^{\mathcal{N}^{\theta}})(\Gamma
_{RB}^{\mathcal{N}^{\theta}})^{-1}\otimes\Pi_{\Gamma^{\mathcal{M}^{\theta}}
}+\Pi_{\Gamma^{\mathcal{N}^{\theta}}}\otimes\partial_{\theta}(\Gamma
_{SD}^{\mathcal{M}^{\theta}})(\Gamma_{SD}^{\mathcal{M}^{\theta}})^{-1}\right]
\nonumber\\
&  \qquad\times\left[  (\partial_{\theta}\Gamma_{RB}^{\mathcal{N}^{\theta}
})\otimes\Gamma_{SD}^{\mathcal{M}^{\theta}}+\Gamma_{RB}^{\mathcal{N}^{\theta}
}\otimes\partial_{\theta}(\Gamma_{SD}^{\mathcal{M}^{\theta}})\right] \\
&  =(\partial_{\theta}\Gamma_{RB}^{\mathcal{N}^{\theta}})(\Gamma
_{RB}^{\mathcal{N}^{\theta}})^{-1}(\partial_{\theta}\Gamma_{RB}^{\mathcal{N}
^{\theta}})\otimes\Gamma_{SD}^{\mathcal{M}^{\theta}}+(\partial_{\theta}
\Gamma_{RB}^{\mathcal{N}^{\theta}})(\Gamma_{RB}^{\mathcal{N}^{\theta}}
)^{-1}\Gamma_{RB}^{\mathcal{N}^{\theta}}\otimes\Pi_{\Gamma^{\mathcal{M}
^{\theta}}}\partial_{\theta}(\Gamma_{SD}^{\mathcal{M}^{\theta}})\nonumber\\
&  \qquad+\Pi_{\Gamma^{\mathcal{N}^{\theta}}}(\partial_{\theta}\Gamma
_{RB}^{\mathcal{N}^{\theta}})\otimes\partial_{\theta}(\Gamma_{SD}
^{\mathcal{M}^{\theta}})(\Gamma_{SD}^{\mathcal{M}^{\theta}})^{-1}\Gamma
_{SD}^{\mathcal{M}^{\theta}}+\Gamma_{RB}^{\mathcal{N}^{\theta}}\otimes
\partial_{\theta}(\Gamma_{SD}^{\mathcal{M}^{\theta}})(\Gamma_{SD}
^{\mathcal{M}^{\theta}})^{-1}\partial_{\theta}(\Gamma_{SD}^{\mathcal{M}
^{\theta}})\\
&  =(\partial_{\theta}\Gamma_{RB}^{\mathcal{N}^{\theta}})(\Gamma
_{RB}^{\mathcal{N}^{\theta}})^{-1}(\partial_{\theta}\Gamma_{RB}^{\mathcal{N}
^{\theta}})\otimes\Gamma_{SD}^{\mathcal{M}^{\theta}}+(\partial_{\theta}
\Gamma_{RB}^{\mathcal{N}^{\theta}})\Pi_{\Gamma^{\mathcal{N}^{\theta}}}
\otimes\Pi_{\Gamma^{\mathcal{M}^{\theta}}}\partial_{\theta}(\Gamma
_{SD}^{\mathcal{M}^{\theta}})\nonumber\\
&  \qquad+\Pi_{\Gamma^{\mathcal{N}^{\theta}}}(\partial_{\theta}\Gamma
_{RB}^{\mathcal{N}^{\theta}})\otimes\partial_{\theta}(\Gamma_{SD}
^{\mathcal{M}^{\theta}})\Pi_{\Gamma^{\mathcal{M}^{\theta}}}+\Gamma
_{RB}^{\mathcal{N}^{\theta}}\otimes\partial_{\theta}(\Gamma_{SD}
^{\mathcal{M}^{\theta}})(\Gamma_{SD}^{\mathcal{M}^{\theta}})^{-1}
\partial_{\theta}(\Gamma_{SD}^{\mathcal{M}^{\theta}})\\
&  =(\partial_{\theta}\Gamma_{RB}^{\mathcal{N}^{\theta}})(\Gamma
_{RB}^{\mathcal{N}^{\theta}})^{-1}(\partial_{\theta}\Gamma_{RB}^{\mathcal{N}
^{\theta}})\otimes\Gamma_{SD}^{\mathcal{M}^{\theta}}+2(\partial_{\theta}
\Gamma_{RB}^{\mathcal{N}^{\theta}})\otimes\partial_{\theta}(\Gamma
_{SD}^{\mathcal{M}^{\theta}})\nonumber\\
&  \qquad+\Gamma_{RB}^{\mathcal{N}^{\theta}}\otimes\partial_{\theta}
(\Gamma_{SD}^{\mathcal{M}^{\theta}})(\Gamma_{SD}^{\mathcal{M}^{\theta}}
)^{-1}\partial_{\theta}(\Gamma_{SD}^{\mathcal{M}^{\theta}}),
\end{align}
where the last line follows because we can \textquotedblleft add
in\textquotedblright\ zero-valued terms like $(\partial_{\theta}\Gamma
_{RB}^{\mathcal{N}^{\theta}})\Pi_{\Gamma^{\mathcal{N}^{\theta}}}^{\perp}
=\Pi_{\Gamma^{\mathcal{N}^{\theta}}}^{\perp}(\partial_{\theta}\Gamma
_{RB}^{\mathcal{N}^{\theta}})=0$\ and $\Pi_{\Gamma^{\mathcal{M}^{\theta}}
}^{\perp}\partial_{\theta}(\Gamma_{SD}^{\mathcal{M}^{\theta}})=\partial
_{\theta}(\Gamma_{SD}^{\mathcal{M}^{\theta}})\Pi_{\Gamma^{\mathcal{M}^{\theta
}}}^{\perp}=0$, due to the finiteness condition in
\eqref{eq:finiteness-condition-RLD-fish-ch}\ holding. Now taking the trace
over systems $BD$ for each term, we find that
\begin{multline}
\operatorname{Tr}_{BD}[(\partial_{\theta}\Gamma_{RB}^{\mathcal{N}^{\theta}
})(\Gamma_{RB}^{\mathcal{N}^{\theta}})^{-1}(\partial_{\theta}\Gamma
_{RB}^{\mathcal{N}^{\theta}})\otimes\Gamma_{SD}^{\mathcal{M}^{\theta}}]\\
=\operatorname{Tr}_{B}[(\partial_{\theta}\Gamma_{RB}^{\mathcal{N}^{\theta}
})(\Gamma_{RB}^{\mathcal{N}^{\theta}})^{-1}(\partial_{\theta}\Gamma
_{RB}^{\mathcal{N}^{\theta}})]\otimes I_{S},
\end{multline}
\begin{align}
\operatorname{Tr}_{BD}[2(\partial_{\theta}\Gamma_{RB}^{\mathcal{N}^{\theta}
})\otimes\partial_{\theta}(\Gamma_{SD}^{\mathcal{M}^{\theta}})]  &
=2\operatorname{Tr}_{B}[(\partial_{\theta}\Gamma_{RB}^{\mathcal{N}^{\theta}
})]\otimes\operatorname{Tr}_{D}[\partial_{\theta}(\Gamma_{SD}^{\mathcal{M}
^{\theta}})]\\
&  =2(\partial_{\theta}\operatorname{Tr}_{B}[\Gamma_{RB}^{\mathcal{N}^{\theta
}}])\otimes\partial_{\theta}(\operatorname{Tr}_{D}[\Gamma_{SD}^{\mathcal{M}
^{\theta}}])\\
&  =2(\partial_{\theta}I_{R})\otimes(\partial_{\theta}I_{S})\\
&  =0,
\end{align}
\begin{multline}
\operatorname{Tr}_{BD}[\Gamma_{RB}^{\mathcal{N}^{\theta}}\otimes
\partial_{\theta}(\Gamma_{SD}^{\mathcal{M}^{\theta}})(\Gamma_{SD}
^{\mathcal{M}^{\theta}})^{-1}\partial_{\theta}(\Gamma_{SD}^{\mathcal{M}
^{\theta}})]\\
=I_{R}\otimes\operatorname{Tr}_{D}[\partial_{\theta}(\Gamma_{SD}
^{\mathcal{M}^{\theta}})(\Gamma_{SD}^{\mathcal{M}^{\theta}})^{-1}
\partial_{\theta}(\Gamma_{SD}^{\mathcal{M}^{\theta}})].
\end{multline}
So we conclude that
\begin{multline}
\operatorname{Tr}_{BD}[(\partial_{\theta}(\Gamma_{RB}^{\mathcal{N}^{\theta}
}\otimes\Gamma_{SD}^{\mathcal{M}^{\theta}}))(\Gamma_{RB}^{\mathcal{N}^{\theta
}}\otimes\Gamma_{SD}^{\mathcal{M}^{\theta}})^{-1}(\partial_{\theta}
(\Gamma_{RB}^{\mathcal{N}^{\theta}}\otimes\Gamma_{SD}^{\mathcal{M}^{\theta}
}))]\\
=\operatorname{Tr}_{B}[(\partial_{\theta}\Gamma_{RB}^{\mathcal{N}^{\theta}
})(\Gamma_{RB}^{\mathcal{N}^{\theta}})^{-1}(\partial_{\theta}\Gamma
_{RB}^{\mathcal{N}^{\theta}})]\otimes I_{S}\\
+I_{R}\otimes\operatorname{Tr}_{D}[\partial_{\theta}(\Gamma_{SD}
^{\mathcal{M}^{\theta}})(\Gamma_{SD}^{\mathcal{M}^{\theta}})^{-1}
\partial_{\theta}(\Gamma_{SD}^{\mathcal{M}^{\theta}})]
\end{multline}
Consider now from Lemma~\ref{lem:additivity-op-norm} that
\begin{equation}
\left\Vert X\otimes I+I\otimes Y\right\Vert _{\infty}=\left\Vert X\right\Vert
_{\infty}+\left\Vert Y\right\Vert _{\infty},
\label{eq:op-norm-identity-additive}
\end{equation}
for positive semi-definite operators $X$ and $Y$. Now applying
\eqref{eq:op-norm-identity-additive}, we find that
\begin{align}
&  \widehat{I}_{F}(\theta;\{\mathcal{N}_{A\rightarrow B}^{\theta}
\otimes\mathcal{M}_{C\rightarrow D}^{\theta}\}_{\theta})\\
&  =\left\Vert \operatorname{Tr}_{BD}[(\partial_{\theta}(\Gamma_{RB}
^{\mathcal{N}^{\theta}}\otimes\Gamma_{SD}^{\mathcal{M}^{\theta}}))(\Gamma
_{RB}^{\mathcal{N}^{\theta}}\otimes\Gamma_{SD}^{\mathcal{M}^{\theta}}
)^{-1}(\partial_{\theta}(\Gamma_{RB}^{\mathcal{N}^{\theta}}\otimes\Gamma
_{SD}^{\mathcal{M}^{\theta}}))]\right\Vert _{\infty}\\
&  =\left\Vert \operatorname{Tr}_{B}[(\partial_{\theta}\Gamma_{RB}
^{\mathcal{N}^{\theta}})(\Gamma_{RB}^{\mathcal{N}^{\theta}})^{-1}
(\partial_{\theta}\Gamma_{RB}^{\mathcal{N}^{\theta}})]\otimes I_{S}
+I_{R}\otimes\operatorname{Tr}_{D}[\partial_{\theta}(\Gamma_{SD}
^{\mathcal{M}^{\theta}})(\Gamma_{SD}^{\mathcal{M}^{\theta}})^{-1}
\partial_{\theta}(\Gamma_{SD}^{\mathcal{M}^{\theta}})]\right\Vert _{\infty}\\
&  =\left\Vert \operatorname{Tr}_{B}[(\partial_{\theta}\Gamma_{RB}
^{\mathcal{N}^{\theta}})(\Gamma_{RB}^{\mathcal{N}^{\theta}})^{-1}
(\partial_{\theta}\Gamma_{RB}^{\mathcal{N}^{\theta}})]\right\Vert _{\infty
}+\left\Vert \operatorname{Tr}_{D}[\partial_{\theta}(\Gamma_{SD}
^{\mathcal{M}^{\theta}})(\Gamma_{SD}^{\mathcal{M}^{\theta}})^{-1}
\partial_{\theta}(\Gamma_{SD}^{\mathcal{M}^{\theta}})]\right\Vert _{\infty}\\
&  =\widehat{I}_{F}(\theta;\{\mathcal{N}_{A\rightarrow B}^{\theta}\}_{\theta
})+\widehat{I}_{F}(\theta;\{\mathcal{M}_{C\rightarrow D}^{\theta}\}_{\theta}).
\end{align}
This concludes the proof.
\end{proof}

\section{Geometric R\'enyi relative entropy and its properties}

\label{app:geo-ren-props}

Before going into detail for the geometric R\'{e}nyi relative entropy, we
first briefly recall some quantum R\'{e}nyi relative entropies.

The Petz--R\'{e}nyi relative entropy \cite{P85,P86}\ is defined as follows for
a state $\rho$, a positive semi-definite operator $\sigma$, and $\alpha
\in(0,1)\cup(1,\infty)$:
\begin{equation}
D_{\alpha}(\rho\Vert\sigma):=\frac{1}{\alpha-1}\ln Q_{\alpha}(\rho\Vert
\sigma),
\end{equation}
where the Petz--R\'{e}nyi relative quasi-entropy is defined as
\begin{equation}
Q_{\alpha}(\rho\Vert\sigma):=\left\{
\begin{array}
[c]{cc}
\operatorname{Tr}[\rho^{\alpha}\sigma^{1-\alpha}] &
\begin{array}
[c]{c}
\text{if }\alpha\in(0,1)\text{ or}\\
\operatorname{supp}(\rho)\subseteq\operatorname{supp}(\sigma)\text{ and
}\alpha\in(1,\infty)
\end{array}
\\
+\infty & \text{otherwise}
\end{array}
\right.  .
\end{equation}
The full definition with the support condition was given in \cite{TCR09}. The
Petz--R\'{e}nyi relative entropy obeys the data-processing inequality for
$\alpha\in(0,1)\cup(1,2]$:
\begin{equation}
D_{\alpha}(\rho\Vert\sigma)\geq D_{\alpha}(\mathcal{N}(\rho)\Vert
\mathcal{N}(\sigma)),
\end{equation}
where $\mathcal{N}$ is a quantum channel \cite{P85,P86}. Note that the
following limit holds \cite{MH11}
\begin{equation}
D_{\alpha}(\rho\Vert\sigma)=\lim_{\varepsilon\rightarrow0^{+}}D_{\alpha}
(\rho\Vert\sigma_{\varepsilon}),
\end{equation}
where $\sigma_{\varepsilon}:=\sigma+\varepsilon I$.

The sandwiched R\'{e}nyi relative entropy \cite{MDSFT13,WWY14}\ is defined as
follows for a state $\rho$, a positive semi-definite operator $\sigma$, and
$\alpha\in(0,1)\cup(1,\infty)$:
\begin{equation}
\widetilde{D}_{\alpha}(\rho\Vert\sigma):=\frac{1}{\alpha-1}\ln\widetilde
{Q}_{\alpha}(\rho\Vert\sigma),
\end{equation}
where the sandwiched R\'{e}nyi relative quasi-entropy is defined as
\begin{equation}
\widetilde{Q}_{\alpha}(\rho\Vert\sigma):=\left\{
\begin{array}
[c]{cc}
\operatorname{Tr}\!\left[  \left(  \sigma^{\frac{1-\alpha}{2\alpha}}\rho
\sigma^{\frac{1-\alpha}{2\alpha}}\right)  ^{\alpha}\right]  &
\begin{array}
[c]{c}
\text{if }\alpha\in(0,1)\text{ or}\\
\operatorname{supp}(\rho)\subseteq\operatorname{supp}(\sigma)\text{ and
}\alpha\in(1,\infty)
\end{array}
\\
+\infty & \text{otherwise}
\end{array}
\right.  .
\end{equation}
Note that the following limit holds \cite{MDSFT13}
\begin{equation}
\widetilde{D}_{\alpha}(\rho\Vert\sigma)=\lim_{\varepsilon\rightarrow0^{+}
}\widetilde{D}_{\alpha}(\rho\Vert\sigma_{\varepsilon}).
\label{eq:limit-sandwiched-to-def}
\end{equation}

Let us also recall the quantum relative entropy \cite{U62}:
\begin{equation}
D(\rho\Vert\sigma):=\left\{
\begin{array}
[c]{cc}
\operatorname{Tr}[\rho(\ln\rho-\ln\sigma)] & \text{if }\operatorname{supp}
(\rho)\subseteq\operatorname{supp}(\sigma)\\
+\infty & \text{otherwise}
\end{array}
\right.  ,
\end{equation}
and note that the following limit holds (see, e.g., \cite{Wbook17})
\begin{equation}
D(\rho\Vert\sigma)=\lim_{\varepsilon\rightarrow0^{+}}D(\rho\Vert
\sigma_{\varepsilon}). \label{eq:rel-ent-limit-eps-0}
\end{equation}
It is known that the Petz-- \cite{P85,P86} and sandwiched \cite{MDSFT13,WWY14}
R\'{e}nyi relative entropies converge to the quantum relative entropy in the
limit $\alpha\rightarrow1$:
\begin{equation}
\lim_{\alpha\rightarrow1}\widetilde{D}_{\alpha}(\rho\Vert\sigma)=\lim
_{\alpha\rightarrow1}D_{\alpha}(\rho\Vert\sigma)=D(\rho\Vert\sigma).
\label{eq:alpha-1-limit-to-rel-ent}
\end{equation}

The max-relative entropy is defined as \cite{Datta2009b}
\begin{equation}
D_{\max}(\rho\Vert\sigma):=\inf\left\{  \lambda\geq0:\rho\leq e^{\lambda}
\sigma\right\}  ,
\end{equation}
and the following limit is known \cite{MDSFT13}
\begin{equation}
\lim_{\alpha\rightarrow\infty}\widetilde{D}_{\alpha}(\rho\Vert\sigma)=D_{\max
}(\rho\Vert\sigma).
\end{equation}

We now recall the definition of the geometric R\'{e}nyi relative entropy:

\begin{definition}
[Geometric R\'{e}nyi relative entropy]\label{def:geometric-renyi-rel-ent}Let
$\rho$ be a state, $\sigma$ a positive semi-definite operator, and $\alpha
\in(0,1)\cup(1,\infty)$. The geometric R\'{e}nyi relative quasi-entropy is
defined as
\begin{equation}
\widehat{Q}_{\alpha}(\rho\Vert\sigma):=\lim_{\varepsilon\rightarrow0^{+}
}\operatorname{Tr}\!\left[  \sigma_{\varepsilon}\!\left(  \sigma_{\varepsilon
}^{-\frac{1}{2}}\rho\sigma_{\varepsilon}^{-\frac{1}{2}}\right)  ^{\alpha
}\right]  , \label{eq:def-geometric-renyi-rel-quasi-ent}
\end{equation}
where $\sigma_{\varepsilon}:=\sigma+\varepsilon I$, and the geometric
R\'{e}nyi relative entropy is then defined as
\begin{equation}
\widehat{D}_{\alpha}(\rho\Vert\sigma):=\frac{1}{\alpha-1} \ln\widehat
{Q}_{\alpha}(\rho\Vert\sigma). \label{eq:def-geometric-renyi-rel-ent}
\end{equation}

\end{definition}

In Definition~\ref{def:geometric-renyi-rel-ent}, we have defined the geometric
R\'{e}nyi relative entropy as a limit, in contrast to how the Petz--R\'{e}nyi
relative entropy and the sandwiched R\'{e}nyi relative entropy are usually
defined (see, e.g., \cite{Berta2018c}). The geometric R\'{e}nyi relative
entropy is a bit more complicated than these other R\'{e}nyi relative
entropies for $\alpha\in(0,1)$, and so defining it as such gives us a more
compact expression to work with.
Proposition~\ref{prop:explicit-form-geometric-renyi}\ below gives explicit
formulas to work with in all cases for which the geometric R\'{e}nyi relative
entropy is defined.

\begin{proposition}
\label{prop:explicit-form-geometric-renyi}For any state $\rho$, positive
semi-definite operator $\sigma$, and $\alpha\in(0,1)\cup(1,\infty)$, the
following equality holds
\begin{equation}
\widehat{Q}_{\alpha}(\rho\Vert\sigma)=\left\{
\begin{array}
[c]{cc}
\operatorname{Tr}\!\left[  \sigma\!\left(  \sigma^{-\frac{1}{2}}\rho
\sigma^{-\frac{1}{2}}\right)  ^{\alpha}\right]  &
\begin{array}
[c]{c}
\text{if }\alpha\in\left(  0,1\right)  \cup(1,\infty)\\
\text{and }\operatorname{supp}(\rho)\subseteq\operatorname{supp}(\sigma)
\end{array}
\\
& \\
\operatorname{Tr}\!\left[  \sigma\!\left(  \sigma^{-\frac{1}{2}}\tilde{\rho
}\sigma^{-\frac{1}{2}}\right)  ^{\alpha}\right]  &
\begin{array}
[c]{c}
\text{if }\alpha\in\left(  0,1\right) \\
\text{and }\operatorname{supp}(\rho)\not \subseteq \operatorname{supp}(\sigma)
\end{array}
\\
& \\
+\infty &
\begin{array}
[c]{c}
\text{if }\alpha\in(1,\infty)\text{ and}\\
\operatorname{supp}(\rho)\not \subseteq \operatorname{supp}(\sigma)\text{.}
\end{array}
\end{array}
\right.  , \label{eq:geometric-rel-quasi-explicit-1}
\end{equation}
where
\begin{align}
\tilde{\rho}  &  :=\rho_{0,0}-\rho_{0,1}\rho_{1,1}^{-1}\rho_{0,1}^{\dag}
,\quad\rho=
\begin{bmatrix}
\rho_{0,0} & \rho_{0,1}\\
\rho_{0,1}^{\dag} & \rho_{1,1}
\end{bmatrix}
,\\
\rho_{0,0}  &  :=\Pi_{\sigma}\rho\Pi_{\sigma},\quad\rho_{0,1}:=\Pi_{\sigma
}\rho\Pi_{\sigma}^{\perp},\quad\rho_{1,1}:=\Pi_{\sigma}^{\perp}\rho\Pi
_{\sigma}^{\perp},
\end{align}
$\Pi_{\sigma}$ is the projection onto the support of $\sigma$, $\Pi_{\sigma
}^{\perp}$ is the projection onto the kernel of $\sigma$, and the inverses
$\sigma^{-\frac{1}{2}}$ and $\rho_{1,1}^{-1}$ are generalized inverses (taken
on the support of $\sigma$ and $\rho_{1,1}$, respectively). We also have the
alternative expressions below for certain cases:
\begin{equation}
\widehat{Q}_{\alpha}(\rho\Vert\sigma)=\left\{
\begin{array}
[c]{cc}
\operatorname{Tr}\!\left[  \rho\!\left(  \rho^{-\frac{1}{2}}\sigma\rho
^{-\frac{1}{2}}\right)  ^{1-\alpha}\right]  &
\begin{array}
[c]{c}
\text{if }\alpha\in\left(  0,1\right) \\
\text{and }\operatorname{supp}(\sigma)\subseteq\operatorname{supp}(\rho)
\end{array}
\\
& \\
\operatorname{Tr}\!\left[  \rho\!\left(  \rho^{\frac{1}{2}}\sigma^{-1}
\rho^{\frac{1}{2}}\right)  ^{\alpha-1}\right]  &
\begin{array}
[c]{c}
\text{if }\alpha\in(1,\infty)\\
\text{and }\operatorname{supp}(\rho)\subseteq\operatorname{supp}(\sigma)
\end{array}
\end{array}
\right.  , \label{eq:geometric-rel-quasi-explicit-2}
\end{equation}
where the inverses $\rho^{-\frac{1}{2}}$ and $\sigma^{-1}$ are generalized inverses.
\end{proposition}

One should observe that when $\operatorname{supp}(\rho)\subseteq
\operatorname{supp}(\sigma)$ and $\alpha\in(0,1)$, the expression
$\operatorname{Tr}\!\left[  \sigma\!\left(  \sigma^{-\frac{1}{2}}\rho
\sigma^{-\frac{1}{2}}\right)  ^{\alpha}\right]  $ is actually a special case
of $\operatorname{Tr}\!\left[  \sigma\!\left(  \sigma^{-\frac{1}{2}}
\tilde{\rho}\sigma^{-\frac{1}{2}}\right)  ^{\alpha}\right]  $, because the
operators $\rho_{0,1}$ and $\rho_{1,1}$ are both equal to zero in this case,
so that $\Pi_{\sigma}\rho=\rho\Pi_{\sigma}=\rho$ and $\tilde{\rho}=\rho_{0,0}
$. The expression $\operatorname{Tr}\!\left[  \sigma\!\left(  \sigma
^{-\frac{1}{2}}\tilde{\rho}\sigma^{-\frac{1}{2}}\right)  ^{\alpha}\right]  $
for $\alpha=1/2$ and $\operatorname{supp}(\rho)\not \subseteq
\operatorname{supp}(\sigma)$ was identified in \cite[Section~3]{Mat14} and
later generalized to all $\alpha\in(0,1)$ in \cite[Section~2]{Mat14condconv}.

The main intuition behind some of the formulas in
Proposition~\ref{prop:explicit-form-geometric-renyi} is as follows. If $\rho$
and $\sigma$ are positive definite, then the following equalities hold
\begin{align}
\operatorname{Tr}\!\left[  \sigma\!\left(  \sigma^{-\frac{1}{2}}\rho
\sigma^{-\frac{1}{2}}\right)  ^{\alpha}\right]   &  =\operatorname{Tr}
\!\left[  \rho\!\left(  \rho^{-\frac{1}{2}}\sigma\rho^{-\frac{1}{2}}\right)
^{1-\alpha}\right] \\
&  =\operatorname{Tr}\!\left[  \rho\!\left(  \rho^{\frac{1}{2}}\sigma^{-1}
\rho^{\frac{1}{2}}\right)  ^{\alpha-1}\right]  ,
\end{align}
for all $\alpha\in(0,1)\cup(1,\infty)$, as shown in
Proposition~\ref{prop:alt-rep-geometric-renyi-quasi} below. If the support
condition $\operatorname{supp}(\rho)\subseteq\operatorname{supp}(\sigma)$
holds, then we can think of $\operatorname{supp}(\sigma)$ as being the whole
Hilbert space and $\sigma$ being invertible on the whole space. So then
generalized inverses like $\sigma^{-\frac{1}{2}}$ or $\sigma^{-1}$ are true
inverses on $\operatorname{supp}(\sigma)$, and the expression
$\operatorname{Tr}[\sigma(\sigma^{-1/2}\rho\sigma^{-1/2})^{\alpha}]$ is
sensible for $\alpha\in(0,1)\cup(1,\infty)$, with the only inverse in the
expression being $\sigma^{-\frac{1}{2}}$. Similarly, the expression
$\operatorname{Tr}[\rho(\rho^{1/2}\sigma^{-1}\rho^{1/2})^{\alpha-1}]$ is
sensible for $\alpha\in(1,\infty)$, with the only inverse in the expression
being $\sigma^{-1}$. On the other hand, if the support condition
$\operatorname{supp}(\sigma)\subseteq\operatorname{supp}(\rho)$ holds, then we
can think of $\operatorname{supp}(\rho)$ as being the whole Hilbert space and
$\rho$ being invertible on the whole space. So then the generalized inverse
$\rho^{-\frac{1}{2}}$ is a true inverse on $\operatorname{supp}(\rho)$, and
the expression $\operatorname{Tr}[\rho(\rho^{-1/2}\sigma\rho^{-1/2}
)^{1-\alpha}]$ is sensible for $\alpha\in(0,1)$, with the only inverse in the
expression being $\rho^{-\frac{1}{2}}$. After developing a few properties of
the geometric R\'{e}nyi relative entropy, we prove
Proposition~\ref{prop:explicit-form-geometric-renyi}.

Due to the fact that Definition~\ref{def:geometric-renyi-rel-ent} does not
involve an inverse of the state $\rho$, the following equality holds for all
$\alpha\in\left(  0,1\right)  \cup\left(  1,\infty\right)  $:
\begin{equation}
\widehat{Q}_{\alpha}(\rho\Vert\sigma)=\lim_{\varepsilon\rightarrow0^{+}}
\lim_{\delta\rightarrow0^{+}}\operatorname{Tr}\!\left[  \sigma_{\varepsilon
}\!\left(  \sigma_{\varepsilon}^{-\frac{1}{2}}\rho_{\delta}\sigma
_{\varepsilon}^{-\frac{1}{2}}\right)  ^{\alpha}\right]  ,
\label{eq:limit-eq-geometric-renyi}
\end{equation}
where
\begin{equation}
\rho_{\delta}:=\left(  1-\delta\right)  \rho+\delta\pi,
\end{equation}
and $\pi$ is the maximally mixed state. The equality in
\eqref{eq:limit-eq-geometric-renyi}\ is useful for establishing the
data-processing inequality for the geometric R\'{e}nyi relative entropy
(Theorem~\ref{thm:data-proc-geometric-renyi} below), as well as its
monotonicity with respect to $\alpha$
(Proposition~\ref{prop:geometric-renyi-props} below). Note that we can
exchange the order of the limits in \eqref{eq:limit-eq-geometric-renyi}\ for
$\alpha\in(0,1)$, which we show later on in
Lemma~\ref{lem:limit-exchange-geom-renyi-a-0-to1}.

The geometric R\'{e}nyi relative entropy is named as such because it can be
written in terms of the \textit{weighted operator geometric mean}. The
weighted operator geometric mean of two positive definite operators $X$ and
$Y$ is defined as follows:
\begin{equation}
G_{\beta}(X,Y):=X^{\frac{1}{2}}\left(  X^{-\frac{1}{2}}YX^{-\frac{1}{2}
}\right)  ^{\beta}X^{\frac{1}{2}}, \label{eq:weighted-op-geo-mean}
\end{equation}
where $\beta\in\mathbb{R}$ is the weight parameter. We recover the standard
operator geometric mean by setting $\beta=1/2$. By using the definition in
\eqref{eq:weighted-op-geo-mean}, we see that the geometric R\'{e}nyi relative
quasi-entropy can be written in terms of the weighted operator geometric mean
as
\begin{align}
\widehat{Q}_{\alpha}(\rho\Vert\sigma)  &  =\operatorname{Tr}\!\left[
\sigma^{\frac{1}{2}}\left(  \sigma^{-\frac{1}{2}}\rho\sigma^{-\frac{1}{2}
}\right)  ^{\alpha}\sigma^{\frac{1}{2}}\right] \\
&  =\operatorname{Tr}[G_{\alpha}(\sigma,\rho)],
\end{align}
whenever $\operatorname{supp}(\rho)\subseteq\operatorname{supp}(\sigma)$.

Whenever\ $\rho$ and $\sigma$ are positive definite, an alternative way of
writing the geometric R\'{e}nyi relative quasi-entropy is given by the
following proposition:

\begin{proposition}
\label{prop:alt-rep-geometric-renyi-quasi}Let $\rho$ be a positive definite
state and $\sigma$ a positive definite operator. For all $\alpha\in\left(
0,1\right)  \cup\left(  1,\infty\right)  $, the following equalities hold
\begin{align}
\widehat{Q}_{\alpha}(\rho\Vert\sigma)  &  =\operatorname{Tr}\!\left[
\rho\left(  \rho^{-\frac{1}{2}}\sigma\rho^{-\frac{1}{2}}\right)  ^{1-\alpha
}\right] \label{eq:alt-op-geo-mean-4-geo-ent}\\
&  =\operatorname{Tr}[G_{1-\alpha}(\rho,\sigma)]\\
&  =\operatorname{Tr}\!\left[  \rho\left(  \rho^{\frac{1}{2}}\sigma^{-1}
\rho^{\frac{1}{2}}\right)  ^{\alpha-1}\right]  .
\end{align}

\end{proposition}

\begin{proof}
The first two equalities follow from a fundamental property of the weighted
operator geometric mean given in Lemma~\ref{lem:geo-mean-symmetry} below. The
last equality follows because $(\rho^{-1/2}\sigma\rho^{-1/2})^{1-\alpha}
=(\rho^{1/2}\sigma^{-1}\rho^{1/2})^{\alpha-1}$ whenever $\rho$ and $\sigma$
are positive definite.
\end{proof}

\begin{lemma}
\label{lem:geo-mean-symmetry}Let $X$ and $Y$ be positive definite operators
and $\beta\in\mathbb{R}$. Then the following equality holds
\begin{equation}
G_{\beta}(X,Y)=G_{1-\beta}(Y,X), \label{eq:geometric-mean-identity}
\end{equation}
with $G_{\beta}(X,Y)$ defined in \eqref{eq:weighted-op-geo-mean}.
\end{lemma}

\begin{proof}
To see \eqref{eq:geometric-mean-identity}, consider that
\begin{align}
G_{1-\beta}(Y,X)  &  =Y^{\frac{1}{2}}\left(  Y^{-\frac{1}{2}}XY^{-\frac{1}{2}
}\right)  ^{1-\beta}Y^{\frac{1}{2}}\\
&  =Y^{\frac{1}{2}}\left(  Y^{-\frac{1}{2}}XY^{-\frac{1}{2}}\right)  \left(
Y^{-\frac{1}{2}}XY^{-\frac{1}{2}}\right)  ^{-\beta}Y^{\frac{1}{2}}\\
&  =X^{\frac{1}{2}}X^{\frac{1}{2}}Y^{-\frac{1}{2}}\left(  Y^{-\frac{1}{2}
}X^{\frac{1}{2}}X^{\frac{1}{2}}Y^{-\frac{1}{2}}\right)  ^{-\beta}Y^{\frac
{1}{2}}\\
&  =X^{\frac{1}{2}}\left(  X^{\frac{1}{2}}Y^{-\frac{1}{2}}Y^{-\frac{1}{2}
}X^{\frac{1}{2}}\right)  ^{-\beta}X^{\frac{1}{2}}Y^{-\frac{1}{2}}Y^{\frac
{1}{2}}\\
&  =X^{\frac{1}{2}}\left(  X^{-\frac{1}{2}}YX^{-\frac{1}{2}}\right)  ^{\beta
}X^{\frac{1}{2}}\\
&  =G_{\beta}(X,Y).
\end{align}
The fourth equality follows from Lemma~\ref{lem:sing-val-lemma-pseudo-commute}, by setting $L=X^{\frac{1}{2}}Y^{-\frac{1}{2}}$ and $f(x)=x^{-\beta}$ therein.
\end{proof}

\bigskip

We now show that the order of limits in \eqref{eq:limit-eq-geometric-renyi}
does not matter when $\alpha\in(0,1)$:

\begin{lemma}
\label{lem:limit-exchange-geom-renyi-a-0-to1}Let $\rho$ be a state and
$\sigma$ a positive semi-definite operator. For $\alpha\in(0,1)$, the
following equality holds
\begin{align}
\widehat{Q}_{\alpha}(\rho\Vert\sigma)  &  =\lim_{\varepsilon\rightarrow0^{+}
}\lim_{\delta\rightarrow0^{+}}\operatorname{Tr}\!\left[  \sigma_{\varepsilon
}\!\left(  \sigma_{\varepsilon}^{-\frac{1}{2}}\rho_{\delta}\sigma
_{\varepsilon}^{-\frac{1}{2}}\right)  ^{\alpha}\right] \\
&  =\inf_{\varepsilon,\delta>0}\operatorname{Tr}\!\left[  \sigma_{\varepsilon
}\!\left(  \sigma_{\varepsilon}^{-\frac{1}{2}}\rho_{\delta}\sigma
_{\varepsilon}^{-\frac{1}{2}}\right)  ^{\alpha}\right] \\
&  =\lim_{\delta\rightarrow0^{+}}\lim_{\varepsilon\rightarrow0^{+}
}\operatorname{Tr}\!\left[  \sigma_{\varepsilon}\!\left(  \sigma_{\varepsilon
}^{-\frac{1}{2}}\rho_{\delta}\sigma_{\varepsilon}^{-\frac{1}{2}}\right)
^{\alpha}\right]  ,
\end{align}
where $\rho_{\delta}:=\left(  1-\delta\right)  \rho+\delta\pi$, $\delta
\in\left(  0,1\right)  $, $\pi$ is the maximally mixed state, $\sigma
_{\varepsilon}:=\sigma+\varepsilon I$, and $\varepsilon>0$.
\end{lemma}

\begin{proof}
First consider that
\begin{equation}
\left(  1-\delta\right)  \rho_{\delta}^{\prime}\leq\rho_{\delta}\leq
\rho_{\delta}^{\prime},
\end{equation}
where
\begin{equation}
\rho_{\delta}^{\prime}:=\rho+\delta\pi.
\end{equation}
By operator monotonicity of $x^{\alpha}$ for $\alpha\in\left(  0,1\right)  $,
we conclude that
\begin{align}
\left(  1-\delta\right)  ^{\alpha}\operatorname{Tr}\!\left[  \sigma
_{\varepsilon}\!\left(  \sigma_{\varepsilon}^{-\frac{1}{2}}\rho_{\delta
}^{\prime}\sigma_{\varepsilon}^{-\frac{1}{2}}\right)  ^{\alpha}\right]   &
\leq\operatorname{Tr}\!\left[  \sigma_{\varepsilon}\!\left(  \sigma
_{\varepsilon}^{-\frac{1}{2}}\rho_{\delta}\sigma_{\varepsilon}^{-\frac{1}{2}
}\right)  ^{\alpha}\right] \\
&  \leq\operatorname{Tr}\!\left[  \sigma_{\varepsilon}\!\left(  \sigma
_{\varepsilon}^{-\frac{1}{2}}\rho_{\delta}^{\prime}\sigma_{\varepsilon
}^{-\frac{1}{2}}\right)  ^{\alpha}\right]  .
\end{align}
These bounds are uniform and independent of $\varepsilon$, and so it follows
that
\begin{align}
\lim_{\varepsilon\rightarrow0^{+}}\lim_{\delta\rightarrow0^{+}}
\operatorname{Tr}\!\left[  \sigma_{\varepsilon}\!\left(  \sigma_{\varepsilon
}^{-\frac{1}{2}}\rho_{\delta}\sigma_{\varepsilon}^{-\frac{1}{2}}\right)
^{\alpha}\right]   &  =\lim_{\varepsilon\rightarrow0^{+}}\lim_{\delta
\rightarrow0^{+}}\operatorname{Tr}\!\left[  \sigma_{\varepsilon}\!\left(
\sigma_{\varepsilon}^{-\frac{1}{2}}\rho_{\delta}^{\prime}\sigma_{\varepsilon
}^{-\frac{1}{2}}\right)  ^{\alpha}\right]  ,\\
\lim_{\delta\rightarrow0^{+}}\lim_{\varepsilon\rightarrow0^{+}}
\operatorname{Tr}\!\left[  \sigma_{\varepsilon}\!\left(  \sigma_{\varepsilon
}^{-\frac{1}{2}}\rho_{\delta}\sigma_{\varepsilon}^{-\frac{1}{2}}\right)
^{\alpha}\right]   &  =\lim_{\delta\rightarrow0^{+}}\lim_{\varepsilon
\rightarrow0^{+}}\operatorname{Tr}\!\left[  \sigma_{\varepsilon}\!\left(
\sigma_{\varepsilon}^{-\frac{1}{2}}\rho_{\delta}^{\prime}\sigma_{\varepsilon
}^{-\frac{1}{2}}\right)  ^{\alpha}\right]  .
\end{align}
Again from the operator monotonicity of $x^{\alpha}$ for $\alpha\in\left(
0,1\right)  $, we conclude for fixed $\varepsilon>0$ that
\begin{equation}
\delta_{1}\leq\delta_{2}\qquad\Rightarrow\qquad\operatorname{Tr}\!\left[
\sigma_{\varepsilon}\!\left(  \sigma_{\varepsilon}^{-\frac{1}{2}}\rho
_{\delta_{1}}^{\prime}\sigma_{\varepsilon}^{-\frac{1}{2}}\right)  ^{\alpha
}\right]  \leq\operatorname{Tr}\!\left[  \sigma_{\varepsilon}\!\left(
\sigma_{\varepsilon}^{-\frac{1}{2}}\rho_{\delta_{2}}^{\prime}\sigma
_{\varepsilon}^{-\frac{1}{2}}\right)  ^{\alpha}\right]  ,
\end{equation}
where $\delta_{1}>0$. By exploiting the identity
\begin{equation}
\operatorname{Tr}\!\left[  \sigma_{\varepsilon}\!\left(  \sigma_{\varepsilon
}^{-\frac{1}{2}}\rho_{\delta}^{\prime}\sigma_{\varepsilon}^{-\frac{1}{2}
}\right)  ^{\alpha}\right]  =\operatorname{Tr}\!\left[  \rho_{\delta}^{\prime
}\!\left(  \left(  \rho_{\delta}^{\prime}\right)  ^{-\frac{1}{2}}
\sigma_{\varepsilon}\left(  \rho_{\delta}^{\prime}\right)  ^{-\frac{1}{2}
}\right)  ^{1-\alpha}\right]
\end{equation}
from Proposition~\ref{prop:alt-rep-geometric-renyi-quasi} and operator
monotonicity of $x^{1-\alpha}$ for $\alpha\in\left(  0,1\right)  $, we
conclude for fixed $\delta>0$ that
\begin{equation}
\varepsilon_{1}\leq\varepsilon_{2}\qquad\Rightarrow\qquad\operatorname{Tr}
\!\left[  \sigma_{\varepsilon_{1}}\!\left(  \sigma_{\varepsilon_{1}}
^{-\frac{1}{2}}\rho_{\delta}\sigma_{\varepsilon_{1}}^{-\frac{1}{2}}\right)
^{\alpha}\right]  \leq\operatorname{Tr}\!\left[  \sigma_{\varepsilon_{2}
}\!\left(  \sigma_{\varepsilon_{2}}^{-\frac{1}{2}}\rho_{\delta}^{\prime}
\sigma_{\varepsilon_{2}}^{-\frac{1}{2}}\right)  ^{\alpha}\right]  ,
\end{equation}
where $\varepsilon_{1}>0$. Thus, we find that
\begin{align}
\lim_{\varepsilon\rightarrow0^{+}}\lim_{\delta\rightarrow0^{+}}
\operatorname{Tr}\!\left[  \sigma_{\varepsilon}\!\left(  \sigma_{\varepsilon
}^{-\frac{1}{2}}\rho_{\delta}^{\prime}\sigma_{\varepsilon}^{-\frac{1}{2}
}\right)  ^{\alpha}\right]   &  =\inf_{\varepsilon>0}\inf_{\delta
>0}\operatorname{Tr}\!\left[  \sigma_{\varepsilon}\!\left(  \sigma
_{\varepsilon}^{-\frac{1}{2}}\rho_{\delta}^{\prime}\sigma_{\varepsilon
}^{-\frac{1}{2}}\right)  ^{\alpha}\right]  ,\\
\lim_{\delta\rightarrow0^{+}}\lim_{\varepsilon\rightarrow0^{+}}
\operatorname{Tr}\!\left[  \sigma_{\varepsilon}\!\left(  \sigma_{\varepsilon
}^{-\frac{1}{2}}\rho_{\delta}^{\prime}\sigma_{\varepsilon}^{-\frac{1}{2}
}\right)  ^{\alpha}\right]   &  =\inf_{\delta>0}\inf_{\varepsilon
>0}\operatorname{Tr}\!\left[  \sigma_{\varepsilon}\!\left(  \sigma
_{\varepsilon}^{-\frac{1}{2}}\rho_{\delta}^{\prime}\sigma_{\varepsilon
}^{-\frac{1}{2}}\right)  ^{\alpha}\right]  .
\end{align}
Since infima can be exchanged, we conclude the statement of the proposition.
\end{proof}

\bigskip

A first property of the geometric R\'{e}nyi relative entropy that we recall is
its relation to the sandwiched R\'{e}nyi relative entropy \cite{MDSFT13,WWY14}
. The inequality below was established for the interval $\alpha\in
(0,1)\cup(1,2]$ in \cite{T15book} (by making use of a general result in
\cite{Mat13,Matsumoto2018})\ and for the full interval $\alpha\in(1,\infty)$
in \cite{WWW19}. Below we follow the approach of \cite{WWW19} and offer a
unified proof in terms of the Araki--Lieb--Thirring inequality
\cite{Araki1990,LT76}.

\begin{proposition}
\label{prop:geometric-to-sandwiched}Let $\rho$ be a state and $\sigma$ a
positive semi-definite operator. The geometric R\'{e}nyi relative entropy is
not smaller than the sandwiched R\'{e}nyi relative entropy for all $\alpha
\in\left(  0,1\right)  \cup\left(  1,\infty\right)  $:
\begin{equation}
\widetilde{D}_{\alpha}(\rho\Vert\sigma)\leq\widehat{D}_{\alpha}(\rho
\Vert\sigma). \label{eq:geometric-renyi-to-sandwiched}
\end{equation}

\end{proposition}

\begin{proof}
This is a direct consequence of the Araki--Lieb--Thirring inequality
\cite{Araki1990,LT76}. For positive semi-definite operators $X$ and $Y$,
$q\geq0$, and $r\in\left[  0,1\right]  $, the following inequality holds
\begin{equation}
\operatorname{Tr}\left[  \left(  Y^{\frac{1}{2}}XY^{\frac{1}{2}}\right)
^{rq}\right]  \geq\operatorname{Tr}\left[  \left(  Y^{\frac{r}{2}}
X^{r}Y^{\frac{r}{2}}\right)  ^{q}\right]  . \label{eq:ALT-1}
\end{equation}
For $r\geq1$, the following inequality holds
\begin{equation}
\operatorname{Tr}\left[  \left(  Y^{\frac{1}{2}}XY^{\frac{1}{2}}\right)
^{rq}\right]  \leq\operatorname{Tr}\left[  \left(  Y^{\frac{r}{2}}
X^{r}Y^{\frac{r}{2}}\right)  ^{q}\right]  . \label{eq:ALT-2}
\end{equation}
By employing it with $q=1$, $r=\alpha\in(0,1)$, $Y=\sigma_{\varepsilon}
^{\frac{1}{\alpha}}$, and $X=\sigma_{\varepsilon}^{-\frac{1}{2}}\rho
\sigma_{\varepsilon}^{-\frac{1}{2}}$, and recalling that $\sigma_{\varepsilon
}:=\sigma+\varepsilon I$, we find that
\begin{align}
\widehat{Q}_{\alpha}(\rho\Vert\sigma_{\varepsilon})  &  =\operatorname{Tr}
\!\left[  \sigma_{\varepsilon}\!\left(  \sigma_{\varepsilon}^{-\frac{1}{2}
}\rho\sigma_{\varepsilon}^{-\frac{1}{2}}\right)  ^{\alpha}\right] \\
&  =\operatorname{Tr}\!\left[  \left(  \sigma_{\varepsilon}^{\frac{1}{2\alpha
}}\right)  ^{\alpha}\left(  \sigma_{\varepsilon}^{-\frac{1}{2}}\rho
\sigma_{\varepsilon}^{-\frac{1}{2}}\right)  ^{\alpha}\left(  \sigma
_{\varepsilon}^{\frac{1}{2\alpha}}\right)  ^{\alpha}\right] \\
&  \leq\operatorname{Tr}\!\left[  \left(  \sigma_{\varepsilon}^{\frac
{1}{2\alpha}}\sigma_{\varepsilon}^{-\frac{1}{2}}\rho\sigma_{\varepsilon
}^{-\frac{1}{2}}\sigma_{\varepsilon}^{\frac{1}{2\alpha}}\right)  ^{\alpha
}\right] \\
&  =\operatorname{Tr}\!\left[  \left(  \sigma_{\varepsilon}^{\frac{1-\alpha
}{2\alpha}}\rho\sigma_{\varepsilon}^{\frac{1-\alpha}{2\alpha}}\right)
^{\alpha}\right] \\
&  =\widetilde{Q}_{\alpha}(\rho\Vert\sigma_{\varepsilon}),
\end{align}
which implies for $\alpha\in(0,1)$, by using definitions, that
\begin{equation}
\widetilde{D}_{\alpha}(\rho\Vert\sigma_{\varepsilon})\leq\widehat{D}_{\alpha
}(\rho\Vert\sigma_{\varepsilon}).
\end{equation}
Now taking the limit as $\varepsilon\rightarrow0^{+}$, employing
\eqref{eq:limit-sandwiched-to-def}\ and
Definition~\ref{def:geometric-renyi-rel-ent}, we arrive at the inequality in \eqref{eq:geometric-renyi-to-sandwiched}.

Since the Araki--Lieb--Thirring inequality is reversed for $r=\alpha\in\left(
1,\infty\right)  $, we can employ similar reasoning as above and definitions
to arrive at \eqref{eq:geometric-renyi-to-sandwiched} for $\alpha\in
(1,\infty)$.
\end{proof}

\bigskip

We are now ready to provide a proof of
Proposition~\ref{prop:explicit-form-geometric-renyi}.

\bigskip

\begin{proof}
[Proof of Proposition~\ref{prop:explicit-form-geometric-renyi}]First suppose
that $\alpha\in(1,\infty)$ and $\operatorname{supp}(\rho)\not \subseteq
\operatorname{supp}(\sigma)$. Then from \eqref{eq:limit-sandwiched-to-def} and
Proposition~\ref{prop:geometric-to-sandwiched}\ and the fact that the
sandwiched R\'{e}nyi relative quasi-entropy $\widetilde{Q}_{\alpha}(\rho
\Vert\sigma)=+\infty$ in this case, it follows that $\widehat{Q}_{\alpha}
(\rho\Vert\sigma)=+\infty$, thus establishing the third expression in \eqref{eq:geometric-rel-quasi-explicit-1}.

Now suppose that $\alpha\in\left(  0,1\right)  \cup(1,\infty)$ and
$\operatorname{supp}(\rho)\subseteq\operatorname{supp}(\sigma)$. Let us employ
the decomposition of the Hilbert space $\mathcal{H}$ as $\mathcal{H}
=\operatorname{supp}(\sigma)\oplus\ker(\sigma)$.\ Then we can write $\rho$ as
\begin{equation}
\rho=
\begin{pmatrix}
\rho_{0,0} & \rho_{0,1}\\
\rho_{0,1}^{\dag} & \rho_{1,1}
\end{pmatrix}
,\qquad\sigma=
\begin{pmatrix}
\sigma & 0\\
0 & 0
\end{pmatrix}
. \label{eq:rho-sig-decompose-for-geometric-formula}
\end{equation}
Writing $I=\Pi_{\sigma}+\Pi_{\sigma}^{\perp}$, where $\Pi_{\sigma}$ is the
projection onto the support of $\sigma$ and $\Pi_{\sigma}^{\perp}$ is the
projection onto the orthogonal complement of $\operatorname{supp}(\sigma)$, we
find that
\begin{equation}
\sigma_{\varepsilon}=
\begin{pmatrix}
\sigma+\varepsilon\Pi_{\sigma} & 0\\
0 & \varepsilon\Pi_{\sigma}^{\perp}
\end{pmatrix}
,
\end{equation}
which implies that
\begin{equation}
\sigma_{\varepsilon}^{-\frac{1}{2}}=
\begin{pmatrix}
\left(  \sigma+\varepsilon\Pi_{\sigma}\right)  ^{-\frac{1}{2}} & 0\\
0 & \varepsilon^{-\frac{1}{2}}\Pi_{\sigma}^{\perp}
\end{pmatrix}
. \label{eq:rho-sig-decompose-for-geometric-formula-3}
\end{equation}

The condition $\operatorname{supp}(\rho)\subseteq\operatorname{supp}(\sigma)$
implies that $\rho_{0,1}=0$ and $\rho_{1,1}=0$. Then
\begin{equation}
\sigma_{\varepsilon}^{-\frac{1}{2}}\rho\sigma_{\varepsilon}^{-\frac{1}{2}}=
\begin{pmatrix}
\left(  \sigma+\varepsilon\Pi_{\sigma}\right)  ^{-\frac{1}{2}}\rho
_{0,0}\left(  \sigma+\varepsilon\Pi_{\sigma}\right)  ^{-\frac{1}{2}} & 0\\
0 & 0
\end{pmatrix}
,
\end{equation}
so that
\begin{align}
&  \operatorname{Tr}\!\left[  \sigma_{\varepsilon}\!\left(  \sigma
_{\varepsilon}^{-\frac{1}{2}}\rho\sigma_{\varepsilon}^{-\frac{1}{2}}\right)
^{\alpha}\right] \nonumber\\
&  =\operatorname{Tr}\!\left[
\begin{pmatrix}
\sigma+\varepsilon\Pi_{\sigma} & 0\\
0 & \varepsilon\Pi_{\sigma}^{\perp}
\end{pmatrix}
\begin{pmatrix}
\left[  \left(  \sigma+\varepsilon\Pi_{\sigma}\right)  ^{-\frac{1}{2}}
\rho_{0,0}\left(  \sigma+\varepsilon\Pi_{\sigma}\right)  ^{-\frac{1}{2}
}\right]  ^{\alpha} & 0\\
0 & 0
\end{pmatrix}
\right] \\
&  =\operatorname{Tr}\!\left[  \left(  \sigma+\varepsilon\Pi_{\sigma}\right)
\left[  \left(  \sigma+\varepsilon\Pi_{\sigma}\right)  ^{-\frac{1}{2}}
\rho_{0,0}\left(  \sigma+\varepsilon\Pi_{\sigma}\right)  ^{-\frac{1}{2}
}\right]  ^{\alpha}\right]  .
\end{align}
Taking the limit $\varepsilon\rightarrow0^{+}$ then leads to
\begin{align}
\lim_{\varepsilon\rightarrow0^{+}}\operatorname{Tr}\!\left[  \sigma
_{\varepsilon}\!\left(  \sigma_{\varepsilon}^{-\frac{1}{2}}\rho\sigma
_{\varepsilon}^{-\frac{1}{2}}\right)  ^{\alpha}\right]   &  =\operatorname{Tr}
\!\left[  \sigma\!\left(  \sigma^{-\frac{1}{2}}\rho_{0,0}\sigma^{-\frac{1}{2}
}\right)  ^{\alpha}\right] \\
&  =\operatorname{Tr}\!\left[  \sigma\!\left(  \sigma^{-\frac{1}{2}}\rho
\sigma^{-\frac{1}{2}}\right)  ^{\alpha}\right]  ,
\label{eq:rho-sig-decompose-for-geometric-formula-last}
\end{align}
thus establishing the first expression in \eqref{eq:geometric-rel-quasi-explicit-1}.

We now establish \eqref{eq:geometric-rel-quasi-explicit-2}. For $\alpha
\in(1,\infty)$ and $\operatorname{supp}(\rho)\subseteq\operatorname{supp}
(\sigma)$, the same analysis implies that
\begin{equation}
\operatorname{Tr}\!\left[  \sigma_{\varepsilon}\!\left(  \sigma_{\varepsilon
}^{-\frac{1}{2}}\rho\sigma_{\varepsilon}^{-\frac{1}{2}}\right)  ^{\alpha
}\right]  =\operatorname{Tr}\!\left[  \hat{\sigma}_{\varepsilon}\!\left(
\hat{\sigma}_{\varepsilon}^{-\frac{1}{2}}\rho_{0,0}\hat{\sigma}_{\varepsilon
}^{-\frac{1}{2}}\right)  ^{\alpha}\right]  ,
\end{equation}
where
\begin{equation}
\hat{\sigma}_{\varepsilon}:=\sigma+\varepsilon\Pi_{\sigma}.
\end{equation}
Since
\begin{equation}
\left(  \hat{\sigma}_{\varepsilon}^{-\frac{1}{2}}\rho_{0,0}\hat{\sigma
}_{\varepsilon}^{-\frac{1}{2}}\right)  ^{\alpha}=\hat{\sigma}_{\varepsilon
}^{-\frac{1}{2}}\rho_{0,0}\hat{\sigma}_{\varepsilon}^{-\frac{1}{2}}\left(
\hat{\sigma}_{\varepsilon}^{-\frac{1}{2}}\rho_{0,0}\hat{\sigma}_{\varepsilon
}^{-\frac{1}{2}}\right)  ^{\alpha-1}
\end{equation}
for $\alpha>1$, we have that
\begin{align}
&  \operatorname{Tr}\!\left[  \hat{\sigma}_{\varepsilon}\hat{\sigma
}_{\varepsilon}^{-\frac{1}{2}}\rho_{0,0}\hat{\sigma}_{\varepsilon}^{-\frac
{1}{2}}\!\left(  \hat{\sigma}_{\varepsilon}^{-\frac{1}{2}}\rho_{0,0}
\hat{\sigma}_{\varepsilon}^{-\frac{1}{2}}\right)  ^{\alpha-1}\right]
\nonumber\\
&  =\operatorname{Tr}\!\left[  \hat{\sigma}_{\varepsilon}^{\frac{1}{2}}
\rho_{0,0}^{\frac{1}{2}}\rho_{0,0}^{\frac{1}{2}}\hat{\sigma}_{\varepsilon
}^{-\frac{1}{2}}\!\left(  \hat{\sigma}_{\varepsilon}^{-\frac{1}{2}}\rho
_{0,0}^{\frac{1}{2}}\rho_{0,0}^{\frac{1}{2}}\hat{\sigma}_{\varepsilon}
^{-\frac{1}{2}}\right)  ^{\alpha-1}\right] \\
&  =\operatorname{Tr}\!\left[  \hat{\sigma}_{\varepsilon}^{\frac{1}{2}}
\rho_{0,0}^{\frac{1}{2}}\!\left(  \rho_{0,0}^{\frac{1}{2}}\hat{\sigma
}_{\varepsilon}^{-\frac{1}{2}}\hat{\sigma}_{\varepsilon}^{-\frac{1}{2}}
\rho_{0,0}^{\frac{1}{2}}\right)  ^{\alpha-1}\rho_{0,0}^{\frac{1}{2}}
\hat{\sigma}_{\varepsilon}^{-\frac{1}{2}}\right] \\
&  =\operatorname{Tr}\!\left[  \rho_{0,0}^{\frac{1}{2}}\hat{\sigma
}_{\varepsilon}^{-\frac{1}{2}}\hat{\sigma}_{\varepsilon}^{\frac{1}{2}}
\rho_{0,0}^{\frac{1}{2}}\!\left(  \rho_{0,0}^{\frac{1}{2}}\hat{\sigma
}_{\varepsilon}^{-1}\rho_{0,0}^{\frac{1}{2}}\right)  ^{\alpha-1}\right] \\
&  =\operatorname{Tr}\!\left[  \rho_{0,0}\left(  \rho_{0,0}^{\frac{1}{2}}
\hat{\sigma}_{\varepsilon}^{-1}\rho_{0,0}^{\frac{1}{2}}\right)  ^{\alpha
-1}\right]  ,
\end{align}
where we applied Lemma~\ref{lem:sing-val-lemma-pseudo-commute}\ with
$f(x)=x^{\alpha-1}$ and $L=\rho_{0,0}^{\frac{1}{2}}\hat{\sigma}_{\varepsilon
}^{-\frac{1}{2}}$. Now taking the limit $\varepsilon\rightarrow0^{+}$, we
conclude that
\begin{align}
\lim_{\varepsilon\rightarrow0^{+}}\operatorname{Tr}\!\left[  \sigma
_{\varepsilon}\!\left(  \sigma_{\varepsilon}^{-\frac{1}{2}}\rho\sigma
_{\varepsilon}^{-\frac{1}{2}}\right)  ^{\alpha}\right]   &  =\lim
_{\varepsilon\rightarrow0^{+}}\operatorname{Tr}\!\left[  \rho_{0,0}\left(
\rho_{0,0}^{\frac{1}{2}}\hat{\sigma}_{\varepsilon}^{-1}\rho_{0,0}^{\frac{1}
{2}}\right)  ^{\alpha-1}\right] \\
&  =\operatorname{Tr}\!\left[  \rho_{0,0}\left(  \rho_{0,0}^{\frac{1}{2}
}\sigma^{-1}\rho_{0,0}^{\frac{1}{2}}\right)  ^{\alpha-1}\right] \\
&  =\operatorname{Tr}\!\left[  \rho\left(  \rho^{\frac{1}{2}}\sigma^{-1}
\rho^{\frac{1}{2}}\right)  ^{\alpha-1}\right]  ,
\end{align}
for the case $\alpha\in(1,\infty)$ and $\operatorname{supp}(\rho
)\subseteq\operatorname{supp}(\sigma)$, thus establishing \eqref{eq:geometric-rel-quasi-explicit-2}.

For the case that $\alpha\in(0,1)$ and $\operatorname{supp}(\sigma
)\subseteq\operatorname{supp}(\rho)$, we can employ the limit exchange from
Lemma~\ref{lem:limit-exchange-geom-renyi-a-0-to1}\ and a similar argument as
in
\eqref{eq:rho-sig-decompose-for-geometric-formula}--\eqref{eq:rho-sig-decompose-for-geometric-formula-last},
but with respect to the decomposition $\mathcal{H}=\operatorname{supp}
(\rho)\oplus\ker(\rho)$, to conclude that
\begin{equation}
\widehat{Q}_{\alpha}(\rho\Vert\sigma)=\operatorname{Tr}\!\left[  \rho\!\left(
\rho^{-\frac{1}{2}}\sigma\rho^{-\frac{1}{2}}\right)  ^{1-\alpha}\right]  ,
\end{equation}
thus establishing the second expression in
\eqref{eq:geometric-rel-quasi-explicit-1}. This case amounts to the exchange
$\rho\leftrightarrow\sigma$ and $\alpha\leftrightarrow1-\alpha$.

We finally consider the case $\alpha\in(0,1)$ and $\operatorname{supp}
(\rho)\not \subseteq \operatorname{supp}(\sigma)$, which is the most involved
case. Consider that
\begin{equation}
\sigma_{\varepsilon}:=\sigma+\varepsilon I=
\begin{bmatrix}
\hat{\sigma}_{\varepsilon} & 0\\
0 & \varepsilon\Pi_{\sigma}^{\perp}
\end{bmatrix}
,
\end{equation}
where $\hat{\sigma}_{\varepsilon}:=\sigma+\varepsilon\Pi_{\sigma}$. Let us
define
\begin{equation}
\rho_{\delta}:=\left(  1-\delta\right)  \rho+\delta\pi,
\end{equation}
with $\delta\in(0,1)$ and $\pi$ the maximally mixed state. By invoking
Lemma~\ref{lem:limit-exchange-geom-renyi-a-0-to1}, we conclude that the
following exchange of limits is possible for $\alpha\in(0,1)$:
\begin{equation}
\lim_{\varepsilon\rightarrow0^{+}}D_{\alpha}(\rho\Vert\sigma_{\varepsilon
})=\lim_{\varepsilon\rightarrow0^{+}}\lim_{\delta\rightarrow0^{+}}D_{\alpha
}(\rho_{\delta}\Vert\sigma_{\varepsilon})=\lim_{\delta\rightarrow0^{+}}
\lim_{\varepsilon\rightarrow0^{+}}D_{\alpha}(\rho_{\delta}\Vert\sigma
_{\varepsilon}).
\end{equation}
Now define
\begin{equation}
\rho_{0,0}^{\delta}:=\Pi_{\sigma}\rho_{\delta}\Pi_{\sigma},\quad\rho
_{0,1}^{\delta}:=\Pi_{\sigma}\rho_{\delta}\Pi_{\sigma}^{\perp},\quad\rho
_{1,1}^{\delta}:=\Pi_{\sigma}^{\perp}\rho_{\delta}\Pi_{\sigma}^{\perp},
\end{equation}
so that
\begin{equation}
\rho_{\delta}=
\begin{bmatrix}
\rho_{0,0}^{\delta} & \rho_{0,1}^{\delta}\\
(\rho_{0,1}^{\delta})^{\dag} & \rho_{1,1}^{\delta}
\end{bmatrix}
.
\end{equation}
Then
\begin{equation}
D_{\alpha}(\rho_{\delta}\Vert\sigma_{\varepsilon})=\frac{1}{\alpha-1}
\ln\operatorname{Tr}\!\left[  \sigma_{\varepsilon}\!\left(  \sigma
_{\varepsilon}^{-\frac{1}{2}}\rho_{\delta}\sigma_{\varepsilon}^{-\frac{1}{2}
}\right)  ^{\alpha}\right]  .
\end{equation}
Consider that
\begin{align}
\sigma_{\varepsilon}^{-\frac{1}{2}}\rho_{\delta}\sigma_{\varepsilon}
^{-\frac{1}{2}}  &  =
\begin{bmatrix}
\hat{\sigma}_{\varepsilon} & 0\\
0 & \varepsilon\Pi_{\sigma}^{\perp}
\end{bmatrix}
^{-\frac{1}{2}}
\begin{bmatrix}
\rho_{0,0}^{\delta} & \rho_{0,1}^{\delta}\\
(\rho_{0,1}^{\delta})^{\dag} & \rho_{1,1}^{\delta}
\end{bmatrix}
\begin{bmatrix}
\hat{\sigma}_{\varepsilon} & 0\\
0 & \varepsilon\Pi_{\sigma}^{\perp}
\end{bmatrix}
^{-\frac{1}{2}}\\
&  =
\begin{bmatrix}
\hat{\sigma}_{\varepsilon}^{-\frac{1}{2}} & 0\\
0 & \varepsilon^{-\frac{1}{2}}\Pi_{\sigma}^{\perp}
\end{bmatrix}
\begin{bmatrix}
\rho_{0,0}^{\delta} & \rho_{0,1}^{\delta}\\
(\rho_{0,1}^{\delta})^{\dag} & \rho_{1,1}^{\delta}
\end{bmatrix}
\begin{bmatrix}
\hat{\sigma}_{\varepsilon}^{-\frac{1}{2}} & 0\\
0 & \varepsilon^{-\frac{1}{2}}\Pi_{\sigma}^{\perp}
\end{bmatrix}
\\
&  =
\begin{bmatrix}
\hat{\sigma}_{\varepsilon}^{-\frac{1}{2}}\rho_{0,0}^{\delta}\hat{\sigma
}_{\varepsilon}^{-\frac{1}{2}} & \varepsilon^{-\frac{1}{2}}\hat{\sigma
}_{\varepsilon}^{-\frac{1}{2}}\rho_{0,1}^{\delta}\Pi_{\sigma}^{\perp}\\
\varepsilon^{-\frac{1}{2}}\Pi_{\sigma}^{\perp}(\rho_{0,1}^{\delta})^{\dag}
\hat{\sigma}_{\varepsilon}^{-\frac{1}{2}} & \varepsilon^{-1}\Pi_{\sigma
}^{\perp}\rho_{1,1}^{\delta}\Pi_{\sigma}^{\perp}
\end{bmatrix}
\\
&  =
\begin{bmatrix}
\hat{\sigma}_{\varepsilon}^{-\frac{1}{2}}\rho_{0,0}^{\delta}\hat{\sigma
}_{\varepsilon}^{-\frac{1}{2}} & \varepsilon^{-\frac{1}{2}}\hat{\sigma
}_{\varepsilon}^{-\frac{1}{2}}\rho_{0,1}^{\delta}\\
\varepsilon^{-\frac{1}{2}}(\rho_{0,1}^{\delta})^{\dag}\hat{\sigma
}_{\varepsilon}^{-\frac{1}{2}} & \varepsilon^{-1}\rho_{1,1}^{\delta}
\end{bmatrix}
.
\end{align}
So then
\begin{align}
&  \operatorname{Tr}\!\left[  \sigma_{\varepsilon}\!\left(  \sigma
_{\varepsilon}^{-\frac{1}{2}}\rho_{\delta}\sigma_{\varepsilon}^{-\frac{1}{2}
}\right)  ^{\alpha}\right] \nonumber\\
&  =\operatorname{Tr}\!\left[
\begin{bmatrix}
\hat{\sigma}_{\varepsilon} & 0\\
0 & \varepsilon\Pi_{\sigma}^{\perp}
\end{bmatrix}
\left(
\begin{bmatrix}
\hat{\sigma}_{\varepsilon}^{-\frac{1}{2}}\rho_{0,0}^{\delta}\hat{\sigma
}_{\varepsilon}^{-\frac{1}{2}} & \varepsilon^{-\frac{1}{2}}\hat{\sigma
}_{\varepsilon}^{-\frac{1}{2}}\rho_{0,1}^{\delta}\\
\varepsilon^{-\frac{1}{2}}(\rho_{0,1}^{\delta})^{\dag}\hat{\sigma
}_{\varepsilon}^{-\frac{1}{2}} & \varepsilon^{-1}\rho_{1,1}^{\delta}
\end{bmatrix}
\right)  ^{\alpha}\right] \\
&  =\operatorname{Tr}\!\left[
\begin{bmatrix}
\hat{\sigma}_{\varepsilon} & 0\\
0 & \varepsilon\Pi_{\sigma}^{\perp}
\end{bmatrix}
\left(  \varepsilon^{-1}
\begin{bmatrix}
\varepsilon\hat{\sigma}_{\varepsilon}^{-\frac{1}{2}}\rho_{0,0}^{\delta}
\hat{\sigma}_{\varepsilon}^{-\frac{1}{2}} & \varepsilon^{\frac{1}{2}}
\hat{\sigma}_{\varepsilon}^{-\frac{1}{2}}\rho_{0,1}^{\delta}\\
\varepsilon^{\frac{1}{2}}(\rho_{0,1}^{\delta})^{\dag}\hat{\sigma}
_{\varepsilon}^{-\frac{1}{2}} & \rho_{1,1}^{\delta}
\end{bmatrix}
\right)  ^{\alpha}\right] \\
&  =\operatorname{Tr}\!\left[
\begin{bmatrix}
\varepsilon^{-\alpha}\hat{\sigma}_{\varepsilon} & 0\\
0 & \varepsilon^{1-\alpha}\Pi_{\sigma}^{\perp}
\end{bmatrix}
\begin{bmatrix}
\varepsilon\hat{\sigma}_{\varepsilon}^{-\frac{1}{2}}\rho_{0,0}^{\delta}
\hat{\sigma}_{\varepsilon}^{-\frac{1}{2}} & \varepsilon^{\frac{1}{2}}
\hat{\sigma}_{\varepsilon}^{-\frac{1}{2}}\rho_{0,1}^{\delta}\\
\varepsilon^{\frac{1}{2}}(\rho_{0,1}^{\delta})^{\dag}\hat{\sigma}
_{\varepsilon}^{-\frac{1}{2}} & \rho_{1,1}^{\delta}
\end{bmatrix}
^{\alpha}\right]
\end{align}
Let us define
\begin{equation}
K(\varepsilon):=
\begin{bmatrix}
\varepsilon\hat{\sigma}_{\varepsilon}^{-\frac{1}{2}}\rho_{0,0}^{\delta}
\hat{\sigma}_{\varepsilon}^{-\frac{1}{2}} & \varepsilon^{\frac{1}{2}}
\hat{\sigma}_{\varepsilon}^{-\frac{1}{2}}\rho_{0,1}^{\delta}\\
\varepsilon^{\frac{1}{2}}(\rho_{0,1}^{\delta})^{\dag}\hat{\sigma}
_{\varepsilon}^{-\frac{1}{2}} & \rho_{1,1}^{\delta}
\end{bmatrix}
,
\end{equation}
so that we can write
\begin{equation}
\operatorname{Tr}\!\left[  \sigma_{\varepsilon}\!\left(  \sigma_{\varepsilon
}^{-\frac{1}{2}}\rho_{\delta}\sigma_{\varepsilon}^{-\frac{1}{2}}\right)
^{\alpha}\right]  =\operatorname{Tr}\!\left[
\begin{bmatrix}
\varepsilon^{-\alpha}\hat{\sigma}_{\varepsilon} & 0\\
0 & \varepsilon^{1-\alpha}\Pi_{\sigma}^{\perp}
\end{bmatrix}
\left(  K(\varepsilon)\right)  ^{\alpha}\right]  .
\end{equation}
Now let us invoke Lemma~\ref{lem:sisi-zhou-lem} with the substitutions
\begin{align}
A  &  \leftrightarrow\rho_{1,1}^{\delta},\\
B  &  \leftrightarrow(\rho_{0,1}^{\delta})^{\dag}\hat{\sigma}_{\varepsilon
}^{-\frac{1}{2}},\\
C  &  \leftrightarrow\hat{\sigma}_{\varepsilon}^{-\frac{1}{2}}\rho
_{0,0}^{\delta}\hat{\sigma}_{\varepsilon}^{-\frac{1}{2}},\\
\varepsilon &  \leftrightarrow\varepsilon^{\frac{1}{2}}.
\end{align}
Defining
\begin{align}
L(\varepsilon)  &  :=
\begin{bmatrix}
\varepsilon S(\rho^{\delta},\hat{\sigma}_{\varepsilon}) & 0\\
0 & \rho_{1,1}^{\delta}+\varepsilon R
\end{bmatrix}
,\\
S(\rho^{\delta},\hat{\sigma}_{\varepsilon})  &  :=\hat{\sigma}_{\varepsilon
}^{-\frac{1}{2}}\left(  \rho_{0,0}^{\delta}-\rho_{0,1}^{\delta}(\rho
_{1,1}^{\delta})^{-1}(\rho_{0,1}^{\delta})^{\dag}\right)  \hat{\sigma
}_{\varepsilon}^{-\frac{1}{2}},\\
R  &  :=\operatorname{Re}[(\rho_{1,1}^{\delta})^{-1}(\rho_{0,1}^{\delta
})^{\dag}(\hat{\sigma}_{\varepsilon})^{-1}(\rho_{0,1}^{\delta})],
\end{align}
we conclude from Lemma~\ref{lem:sisi-zhou-lem}\ that
\begin{equation}
\left\Vert K(\varepsilon)-e^{-i\sqrt{\varepsilon}G}L(\varepsilon
)e^{i\sqrt{\varepsilon}G}\right\Vert _{\infty}\leq o(\varepsilon),
\label{eq:op-norm-bound-geo-support}
\end{equation}
where $G$ in Lemma~\ref{lem:sisi-zhou-lem}\ is defined from $A$ and $B$ above.
The inequality in \eqref{eq:op-norm-bound-geo-support} in turn implies the
following operator inequalities:
\begin{equation}
e^{-i\sqrt{\varepsilon}G}L(\varepsilon)e^{i\sqrt{\varepsilon}G}-o(\varepsilon
)I\leq K(\varepsilon)\leq e^{-i\sqrt{\varepsilon}G}L(\varepsilon
)e^{i\sqrt{\varepsilon}G}+o(\varepsilon)I.
\label{eq:op-ineq-geo-supp-renyi-bnd}
\end{equation}
Observe that
\begin{equation}
e^{-i\sqrt{\varepsilon}G}L(\varepsilon)e^{i\sqrt{\varepsilon}G}+o(\varepsilon
)I=e^{-i\sqrt{\varepsilon}G}\left[  L(\varepsilon)+o(\varepsilon)I\right]
e^{i\sqrt{\varepsilon}G}.
\end{equation}
Now invoking these and the operator monotonicity of the function $x^{\alpha}$
for $\alpha\in(0,1)$, we find that
\begin{align}
&  \operatorname{Tr}\!\left[  \sigma_{\varepsilon}\left(  \sigma_{\varepsilon
}^{-\frac{1}{2}}\rho_{\delta}\sigma_{\varepsilon}^{-\frac{1}{2}}\right)
^{\alpha}\right] \label{eq:geo-renyi-sup-arg-exp-to-bnd}\\
&  =\operatorname{Tr}\!\left[
\begin{bmatrix}
\varepsilon^{-\alpha}\hat{\sigma}_{\varepsilon} & 0\\
0 & \varepsilon^{1-\alpha}\Pi_{\sigma}^{\perp}
\end{bmatrix}
\left(  K(\varepsilon)\right)  ^{\alpha}\right] \\
&  \leq\operatorname{Tr}\!\left[
\begin{bmatrix}
\varepsilon^{-\alpha}\hat{\sigma}_{\varepsilon} & 0\\
0 & \varepsilon^{1-\alpha}\Pi_{\sigma}^{\perp}
\end{bmatrix}
\left(  e^{-i\sqrt{\varepsilon}G}\left[  L(\varepsilon)+o(\varepsilon
)I\right]  e^{i\sqrt{\varepsilon}G}\right)  ^{\alpha}\right] \\
&  =\operatorname{Tr}\!\left[
\begin{bmatrix}
\varepsilon^{-\alpha}\hat{\sigma}_{\varepsilon} & 0\\
0 & \varepsilon^{1-\alpha}\Pi_{\sigma}^{\perp}
\end{bmatrix}
e^{-i\sqrt{\varepsilon}G}\left(  L(\varepsilon)+o(\varepsilon)I\right)
^{\alpha}e^{i\sqrt{\varepsilon}G}\right]  . \label{eq:up-op-bnd-suppo-geo-ren}
\end{align}
Consider that
\begin{align}
&  \left(  L(\varepsilon)+o(\varepsilon)I\right)  ^{\alpha}\nonumber\\
&  =
\begin{bmatrix}
\varepsilon S(\rho^{\delta},\hat{\sigma}_{\varepsilon})+o(\varepsilon)I & 0\\
0 & \rho_{1,1}^{\delta}+\varepsilon R+o(\varepsilon)I
\end{bmatrix}
^{\alpha}\\
&  =
\begin{bmatrix}
\left(  \varepsilon S(\rho^{\delta},\hat{\sigma}_{\varepsilon})+o(\varepsilon
)I\right)  ^{\alpha} & 0\\
0 & \left(  \rho_{1,1}^{\delta}+\varepsilon R+o(\varepsilon)I\right)
^{\alpha}
\end{bmatrix}
\\
&  =
\begin{bmatrix}
\varepsilon^{\alpha}\left(  S(\rho^{\delta},\hat{\sigma}_{\varepsilon
})+o(1)I\right)  ^{\alpha} & 0\\
0 & \left(  \rho_{1,1}^{\delta}+\varepsilon R+o(\varepsilon)I\right)
^{\alpha}
\end{bmatrix}
.
\end{align}
Now expanding $e^{i\sqrt{\varepsilon}G}$ to first order in order to evaluate
\eqref{eq:up-op-bnd-suppo-geo-ren}\ (higher order terms will end up being
irrelevant), we find that
\begin{align}
&  \operatorname{Tr}\!\left[
\begin{bmatrix}
\varepsilon^{-\alpha}\hat{\sigma}_{\varepsilon} & 0\\
0 & \varepsilon^{1-\alpha}\Pi_{\sigma}^{\perp}
\end{bmatrix}
e^{-i\sqrt{\varepsilon}G}\left(  L(\varepsilon)+o(\varepsilon)I\right)
^{\alpha}e^{i\sqrt{\varepsilon}G}\right] \nonumber\\
&  =\operatorname{Tr}\!\left[
\begin{bmatrix}
\varepsilon^{-\alpha}\hat{\sigma}_{\varepsilon} & 0\\
0 & \varepsilon^{1-\alpha}\Pi_{\sigma}^{\perp}
\end{bmatrix}
\left(  L(\varepsilon)+o(\varepsilon)I\right)  ^{\alpha}\right] \nonumber\\
&  \quad+\operatorname{Tr}\!\left[
\begin{bmatrix}
\varepsilon^{-\alpha}\hat{\sigma}_{\varepsilon} & 0\\
0 & \varepsilon^{1-\alpha}\Pi_{\sigma}^{\perp}
\end{bmatrix}
\left(  -i\sqrt{\varepsilon}G\right)  \left(  L(\varepsilon)+o(\varepsilon
)I\right)  ^{\alpha}\right] \nonumber\\
&  \quad+\operatorname{Tr}\!\left[
\begin{bmatrix}
\varepsilon^{-\alpha}\hat{\sigma}_{\varepsilon} & 0\\
0 & \varepsilon^{1-\alpha}\Pi_{\sigma}^{\perp}
\end{bmatrix}
\left(  L(\varepsilon)+o(\varepsilon)I\right)  ^{\alpha}\left(  i\sqrt
{\varepsilon}G\right)  \right]  +o(1)\\
&  =\operatorname{Tr}\!\left[
\begin{bmatrix}
\hat{\sigma}_{\varepsilon}\left(  S(\rho^{\delta},\hat{\sigma}_{\varepsilon
})+o(1)I\right)  ^{\alpha} & 0\\
0 & \varepsilon^{1-\alpha}\Pi_{\sigma}^{\perp}\left(  \rho_{1,1}^{\delta
}+\varepsilon R+o(\varepsilon)I\right)  ^{\alpha}
\end{bmatrix}
\right] \nonumber\\
&  \quad-i\sqrt{\varepsilon}\operatorname{Tr}\!\left[
\begin{bmatrix}
\left(  S(\rho^{\delta},\hat{\sigma}_{\varepsilon})+o(1)I\right)  ^{\alpha
}\hat{\sigma}_{\varepsilon} & 0\\
0 & \varepsilon^{1-\alpha}\left(  \rho_{1,1}^{\delta}+\varepsilon
R+o(\varepsilon)I\right)  ^{\alpha}\Pi_{\sigma}^{\perp}
\end{bmatrix}
G\right] \nonumber\\
&  \quad+i\sqrt{\varepsilon}\operatorname{Tr}\!\left[
\begin{bmatrix}
\hat{\sigma}_{\varepsilon}\left(  S(\rho^{\delta},\hat{\sigma}_{\varepsilon
})+o(1)I\right)  ^{\alpha} & 0\\
0 & \varepsilon^{1-\alpha}\Pi_{\sigma}^{\perp}\left(  \rho_{1,1}^{\delta
}+\varepsilon R+o(\varepsilon)I\right)  ^{\alpha}
\end{bmatrix}
G\right]  +o(1)\\
&  =\operatorname{Tr}\!\left[
\begin{bmatrix}
\hat{\sigma}_{\varepsilon}\left(  S(\rho^{\delta},\hat{\sigma}_{\varepsilon
})+o(1)I\right)  ^{\alpha} & 0\\
0 & \varepsilon^{1-\alpha}\Pi_{\sigma}^{\perp}\left(  \rho_{1,1}^{\delta
}+\varepsilon R+o(\varepsilon)I\right)  ^{\alpha}
\end{bmatrix}
\right]  +o(1)\\
&  =\operatorname{Tr}\!\left[  \hat{\sigma}_{\varepsilon}\left(
S(\rho^{\delta},\hat{\sigma}_{\varepsilon})+o(1)I\right)  ^{\alpha}\right]
+\varepsilon^{1-\alpha}\operatorname{Tr}[\Pi_{\sigma}^{\perp}\left(
\rho_{1,1}^{\delta}+\varepsilon R+o(\varepsilon)I\right)  ^{\alpha}]+o(1).
\end{align}
By observing the last line, we see that higher order terms for $e^{i\sqrt
{\varepsilon}G}$ include prefactors of $\varepsilon$ (or higher powers), which
vanish in the $\varepsilon\rightarrow0^{+}$ limit. Now taking the limit
$\varepsilon\rightarrow0^{+}$, we find that
\begin{multline}
\lim_{\varepsilon\rightarrow0^{+}}\operatorname{Tr}\!\left[
\begin{bmatrix}
\varepsilon^{-\alpha}\hat{\sigma}_{\varepsilon} & 0\\
0 & \varepsilon^{1-\alpha}\Pi_{\sigma}^{\perp}
\end{bmatrix}
e^{-i\sqrt{\varepsilon}G}\left(  L(\varepsilon)+o(\varepsilon)I\right)
^{\alpha}e^{i\sqrt{\varepsilon}G}\right]
\label{eq:eps-final-limit-geo-ren-supp}\\
=\operatorname{Tr}\!\left[  \sigma\left(  \sigma^{-\frac{1}{2}}\left(
\rho_{0,0}^{\delta}-\rho_{0,1}^{\delta}(\rho_{1,1}^{\delta})^{-1}(\rho
_{0,1}^{\delta})^{\dag}\right)  \sigma^{-\frac{1}{2}}\right)  ^{\alpha
}\right]  ,
\end{multline}
where the inverses are taken on the support of $\sigma$. By proceeding in a
similar way, but using the lower bound in
\eqref{eq:op-ineq-geo-supp-renyi-bnd}, we find the following lower bound on
\eqref{eq:geo-renyi-sup-arg-exp-to-bnd}:
\begin{equation}
\operatorname{Tr}\!\left[
\begin{bmatrix}
\varepsilon^{-\alpha}\hat{\sigma}_{\varepsilon} & 0\\
0 & \varepsilon^{1-\alpha}\Pi_{\sigma}^{\perp}
\end{bmatrix}
e^{-i\sqrt{\varepsilon}G}\left(  L(\varepsilon)-o(\varepsilon)I\right)
^{\alpha}e^{i\sqrt{\varepsilon}G}\right]  .
\end{equation}
Then by the same argument above, the lower bound on
\eqref{eq:geo-renyi-sup-arg-exp-to-bnd} after taking the limit $\varepsilon
\rightarrow0^{+}$ is the same as in \eqref{eq:eps-final-limit-geo-ren-supp}.
So we conclude that
\begin{equation}
\lim_{\varepsilon\rightarrow0^{+}}\operatorname{Tr}\!\left[  \sigma
_{\varepsilon}\left(  \sigma_{\varepsilon}^{-\frac{1}{2}}\rho_{\delta}
\sigma_{\varepsilon}^{-\frac{1}{2}}\right)  ^{\alpha}\right]
=\operatorname{Tr}\!\left[  \sigma\left(  \sigma^{-\frac{1}{2}}\left(
\rho_{0,0}^{\delta}-\rho_{0,1}^{\delta}(\rho_{1,1}^{\delta})^{-1}(\rho
_{0,1}^{\delta})^{\dag}\right)  \sigma^{-\frac{1}{2}}\right)  ^{\alpha
}\right]  .
\end{equation}
Now consider that
\begin{equation}
\lim_{\delta\rightarrow0^{+}}\rho_{0,0}^{\delta}-\rho_{0,1}^{\delta}
(\rho_{1,1}^{\delta})^{-1}(\rho_{0,1}^{\delta})^{\dag}=\rho_{0,0}-\rho
_{0,1}\rho_{1,1}^{-1}\rho_{0,1}^{\dag},
\end{equation}
where the inverse on the right is taken on the support of $\rho_{1,1}$. This
follows because the image of $\rho_{0,1}^{\dag}$ is contained in the support
of $\rho_{1,1}$. Thus, we take the limit $\delta\rightarrow0^{+}$, and find
that
\begin{equation}
\lim_{\delta\rightarrow0^{+}}\lim_{\varepsilon\rightarrow0^{+}}
\operatorname{Tr}\!\left[  \sigma_{\varepsilon}\left(  \sigma_{\varepsilon
}^{-\frac{1}{2}}\rho_{\delta}\sigma_{\varepsilon}^{-\frac{1}{2}}\right)
^{\alpha}\right]  =\operatorname{Tr}\!\left[  \sigma\left(  \sigma^{-\frac
{1}{2}}\left(  \rho_{0,0}-\rho_{0,1}\rho_{1,1}^{-1}\rho_{0,1}^{\dag}\right)
\sigma^{-\frac{1}{2}}\right)  ^{\alpha}\right]  ,
\end{equation}
where all inverses are taken on the support. This concludes the proof.
\end{proof}

\bigskip

If the state $\rho$ is pure, then the geometric R\'{e}nyi relative entropy
simplifies as follows, such that it is independent of $\alpha$:

\begin{proposition}
Let $\rho=|\psi\rangle\!\langle\psi|$ be a pure state and $\sigma$ a positive
semi-definite operator.\ Then the following equality holds for all $\alpha
\in(0,1)\cup(1,\infty)$:
\begin{equation}
\widehat{D}_{\alpha}(\rho\Vert\sigma)=\left\{
\begin{array}
[c]{cc}
\ln\langle\psi|\sigma^{-1}|\psi\rangle & \text{if }\operatorname{supp}
(|\psi\rangle\!\langle\psi|)\subseteq\operatorname{supp}(\sigma)\\
+\infty & \text{otherwise}
\end{array}
\right.  ,
\end{equation}
where $\sigma^{-1}$ is understood as a generalized inverse. If $\sigma$ is
also a rank-one operator, so that $\sigma=|\phi\rangle\!\langle\phi|$ and
$\left\Vert |\phi\rangle\right\Vert _{2}>0$, then the following equality holds
for all $\alpha\in(0,1)\cup(1,\infty)$:
\begin{equation}
\widehat{D}_{\alpha}(\rho\Vert\sigma)=\left\{
\begin{array}
[c]{cc}
-\ln\left\Vert |\phi\rangle\right\Vert _{2}^{2} & \text{if }\exists
c\in\mathbb{C}\text{ such that }|\psi\rangle=c|\phi\rangle\\
+\infty & \text{otherwise}
\end{array}
\right.  . \label{eq:geometric-renyi-pure-states-collapse}
\end{equation}
In particular, if $\sigma=|\phi\rangle\!\langle\phi|$ is a state so that
$\left\Vert |\phi\rangle\right\Vert _{2}^{2}=1$, then
\begin{equation}
\widehat{D}_{\alpha}(\rho\Vert\sigma)=\left\{
\begin{array}
[c]{cc}
0 & \text{if }|\psi\rangle=|\phi\rangle\\
+\infty & \text{otherwise}
\end{array}
\right.  .
\end{equation}

\end{proposition}

\begin{proof}
Defining $\sigma_{\varepsilon}:=\sigma+\varepsilon I$, consider that
\begin{align}
\operatorname{Tr}\!\left[  \sigma_{\varepsilon}\!\left(  \sigma_{\varepsilon
}^{-\frac{1}{2}}\rho\sigma_{\varepsilon}^{-\frac{1}{2}}\right)  ^{\alpha
}\right]   &  =\operatorname{Tr}\!\left[  \sigma_{\varepsilon}\!\left(
\sigma_{\varepsilon}^{-\frac{1}{2}}|\psi\rangle\!\langle\psi|\sigma
_{\varepsilon}^{-\frac{1}{2}}\right)  ^{\alpha}\right] \\
&  =\left(  \left\Vert \sigma_{\varepsilon}^{-\frac{1}{2}}|\psi\rangle
\right\Vert _{2}^{2}\right)  ^{\alpha}\operatorname{Tr}\!\left[
\sigma_{\varepsilon}\left(  \frac{\sigma_{\varepsilon}^{-\frac{1}{2}}
|\psi\rangle\!\langle\psi|\sigma_{\varepsilon}^{-\frac{1}{2}}}{\left\Vert
\sigma_{\varepsilon}^{-\frac{1}{2}}|\psi\rangle\right\Vert _{2}^{2}}\right)
^{\alpha}\right] \\
&  =\left(  \left\Vert \sigma_{\varepsilon}^{-\frac{1}{2}}|\psi\rangle
\right\Vert _{2}^{2}\right)  ^{\alpha}\operatorname{Tr}\!\left[
\sigma_{\varepsilon}\frac{\sigma_{\varepsilon}^{-\frac{1}{2}}|\psi
\rangle\!\langle\psi|\sigma_{\varepsilon}^{-\frac{1}{2}}}{\left\Vert
\sigma_{\varepsilon}^{-\frac{1}{2}}|\psi\rangle\right\Vert _{2}^{2}}\right] \\
&  =\left(  \left\Vert \sigma_{\varepsilon}^{-\frac{1}{2}}|\psi\rangle
\right\Vert _{2}^{2}\right)  ^{\alpha-1}\operatorname{Tr}\!\left[
\sigma_{\varepsilon}\sigma_{\varepsilon}^{-\frac{1}{2}}|\psi\rangle\!\langle
\psi|\sigma_{\varepsilon}^{-\frac{1}{2}}\right] \\
&  =\left(  \left\Vert \sigma_{\varepsilon}^{-\frac{1}{2}}|\psi\rangle
\right\Vert _{2}^{2}\right)  ^{\alpha-1}\operatorname{Tr}[|\psi\rangle
\langle\psi|]\\
&  =\left[  \langle\psi|\sigma_{\varepsilon}^{-1}|\psi\rangle\right]
^{\alpha-1}.
\end{align}
The third equality follows because $|\varphi\rangle\!\langle\varphi|^{\alpha
}=|\varphi\rangle\!\langle\varphi|$ for all $\alpha\in\left(  0,1\right)
\cup\left(  1,\infty\right)  $ when $\left\Vert |\varphi\rangle\right\Vert
_{2}=1$. Applying the above chain of equalities, we find that
\begin{align}
\frac{1}{\alpha-1}\ln\operatorname{Tr}\!\left[  \sigma_{\varepsilon}\left(
\sigma_{\varepsilon}^{-\frac{1}{2}}\rho\sigma_{\varepsilon}^{-\frac{1}{2}
}\right)  ^{\alpha}\right]   &  =\frac{1}{\alpha-1}\log_{2}\!\left[
\langle\psi|\sigma_{\varepsilon}^{-1}|\psi\rangle\right]  ^{\alpha-1}\\
&  =\ln\langle\psi|\sigma_{\varepsilon}^{-1}|\psi\rangle.
\end{align}

Now let a spectral decomposition of $\sigma$ be given by
\begin{equation}
\sigma=\sum_{y}\mu_{y}Q_{y},
\end{equation}
where $\mu_{y}$ are the non-negative eigenvalues and $Q_{y}$ are the
eigenprojections. In this decomposition, we are including values of $\mu_{y}$
for which $\mu_{y}=0$. Then it follows that
\begin{equation}
\sigma_{\varepsilon}=\sigma+\varepsilon I=\sum_{y}\left(  \mu_{y}
+\varepsilon\right)  Q_{y},
\end{equation}
and we find that
\begin{equation}
\sigma_{\varepsilon}^{-1}=\sum_{y}\left(  \mu_{y}+\varepsilon\right)
^{-1}Q_{y}.
\end{equation}
We can then conclude that
\begin{align}
\ln\langle\psi|\sigma_{\varepsilon}^{-1}|\psi\rangle &  =\ln\!\left[
\langle\psi|\sum_{y}\left(  \mu_{y}+\varepsilon\right)  ^{-1}Q_{y}|\psi
\rangle\right] \\
&  =\ln\!\left[  \sum_{y}\left(  \mu_{y}+\varepsilon\right)  ^{-1}\langle
\psi|Q_{y}|\psi\rangle\right] \\
&  =\ln\!\left[  \sum_{y:\mu_{y}\neq0}\left(  \mu_{y}+\varepsilon\right)
^{-1}\langle\psi|Q_{y}|\psi\rangle+\varepsilon^{-1}\langle\psi|Q_{y_{0}}
|\psi\rangle\right]  ,
\end{align}
where $y_{0}$ is the value of $y$ for which $\mu_{y}=0$ (if no such value of
$y$ exists, then $Q_{y_{0}}$ is equal to the zero operator). Thus, if
$\langle\psi|Q_{y_{0}}|\psi\rangle\neq0$ (equivalent to $|\psi\rangle$ being
outside the support of $\sigma$), then it follows that
\begin{equation}
\lim_{\varepsilon\rightarrow0^{+}}\ln\langle\psi|\sigma_{\varepsilon}
^{-1}|\psi\rangle=+\infty.
\end{equation}
Otherwise the expression converges as claimed.

Now suppose that $\sigma$ is a rank-one operator, so that $\sigma=|\phi
\rangle\!\langle\phi|$ and $\left\Vert |\phi\rangle\right\Vert _{2}>0$. By
defining
\begin{align}
|\phi^{\prime}\rangle &  :=\frac{|\phi\rangle}{\sqrt{\left\Vert |\phi
\rangle\right\Vert _{2}}},\\
N  &  :=\left\Vert |\phi\rangle\right\Vert _{2}^{2},
\end{align}
we find that
\begin{align}
\sigma_{\varepsilon}  &  =|\phi\rangle\!\langle\phi|+\varepsilon I\\
&  =N\ |\phi^{\prime}\rangle\!\langle\phi^{\prime}|+\varepsilon\left(
I-|\phi^{\prime}\rangle\!\langle\phi^{\prime}|+|\phi^{\prime}\rangle\!\langle
\phi^{\prime}|\right) \\
&  =\left(  N+\varepsilon\right)  \ |\phi^{\prime}\rangle\!\langle\phi^{\prime
}|+\varepsilon\left(  I-|\phi^{\prime}\rangle\!\langle\phi^{\prime}|\right)  ,
\end{align}
so that
\begin{align}
\sigma_{\varepsilon}^{-1}  &  =\left(  N+\varepsilon\right)  ^{-1}
|\phi^{\prime}\rangle\!\langle\phi^{\prime}|+\varepsilon^{-1}\left(
I-|\phi^{\prime}\rangle\!\langle\phi^{\prime}|\right) \\
&  =\left(  \left(  N+\varepsilon\right)  ^{-1}-\varepsilon^{-1}\right)
|\phi^{\prime}\rangle\!\langle\phi^{\prime}|+\varepsilon^{-1}I
\end{align}
and then
\begin{align}
\ln\left[  \langle\psi|\sigma_{\varepsilon}^{-1}|\psi\rangle\right]   &
=\ln\!\left[  \langle\psi|\left[  \left(  \left(  N+\varepsilon\right)
^{-1}-\varepsilon^{-1}\right)  |\phi^{\prime}\rangle\!\langle\phi^{\prime
}|+\varepsilon^{-1}I\right]  |\psi\rangle\right] \\
&  =\ln\!\left[  \left(  \left(  N+\varepsilon\right)  ^{-1}-\varepsilon
^{-1}\right)  \left\vert \langle\psi|\phi^{\prime}\rangle\right\vert
^{2}+\varepsilon^{-1}\right] \\
&  =\ln\!\left[  \frac{\left\vert \langle\psi|\phi^{\prime}\rangle\right\vert
^{2}}{N+\varepsilon}+\frac{1-\left\vert \langle\psi|\phi^{\prime}
\rangle\right\vert ^{2}}{\varepsilon}\right]  .
\end{align}
Note that we always have $\left\vert \langle\psi|\phi^{\prime}\rangle
\right\vert ^{2}\in\left[  0,1\right]  $ because $|\psi\rangle$ and
$|\phi^{\prime}\rangle$ are unit vectors. In the case that $\left\vert
\langle\psi|\phi^{\prime}\rangle\right\vert ^{2}\in\lbrack0,1)$, then we find
that
\begin{equation}
\lim_{\varepsilon\rightarrow0^{+}}\ln\!\left[  \frac{\left\vert \langle
\psi|\phi^{\prime}\rangle\right\vert ^{2}}{N+\varepsilon}+\frac{1-\left\vert
\langle\psi|\phi^{\prime}\rangle\right\vert ^{2}}{\varepsilon}\right]
=+\infty.
\end{equation}
Otherwise, if $\left\vert \langle\psi|\phi^{\prime}\rangle\right\vert ^{2}=1$,
then
\begin{align}
\lim_{\varepsilon\rightarrow0^{+}}\ln\!\left[  \frac{\left\vert \langle
\psi|\phi^{\prime}\rangle\right\vert ^{2}}{N+\varepsilon}+\frac{1-\left\vert
\langle\psi|\phi^{\prime}\rangle\right\vert ^{2}}{\varepsilon}\right]   &
=\lim_{\varepsilon\rightarrow0^{+}}\ln\!\left[  \frac{1}{N+\varepsilon}\right]
\\
&  =-\ln N,
\end{align}
concluding the proof.
\end{proof}

\bigskip

We note here that, for pure states $\rho$ and $\sigma$, the geometric
R\'{e}nyi relative entropy is either equal to zero or $+\infty$, depending on
whether $\rho=\sigma$. This behavior of the geometric R\'{e}nyi relative
entropy for pure states $\rho$ and $\sigma$ is very different from that of the
Petz--R\'{e}nyi and sandwiched R\'{e}nyi relative entropies. The latter
quantities always evaluate to a finite value if the pure states are non-orthogonal.

The geometric R\'{e}nyi relative entropy possesses a number of useful
properties, which we list in the proposition below.

\begin{proposition}
[Properties of the geometric R\'{e}nyi relative entropy]
\label{prop:geometric-renyi-props}For all states $\rho$, $\rho_{1}$, $\rho
_{2}$ and positive semi-definite operators $\sigma$, $\sigma_{1}$, $\sigma
_{2}$, the geometric R\'{e}nyi relative entropy satisfies the following properties.

\begin{enumerate}
\item Isometric invariance: For all $\alpha\in(0,1)\cup(1,\infty)$ and for all
isometries $V$,
\begin{equation}
\widehat{D}_{\alpha}(\rho\Vert\sigma)=\widehat{D}_{\alpha}(V\rho V^{\dag}\Vert
V\sigma V^{\dag}).
\end{equation}

\item Monotonicity in $\alpha$: For all $\alpha\in(0,1)\cup(1,\infty)$, the
geometric R\'{e}nyi relative entropy $\widehat{D}_{\alpha}$ is monotonically
increasing in $\alpha$; i.e., $\alpha<\beta$ implies $\widehat{D}_{\alpha
}(\rho\Vert\sigma)\leq\widehat{D}_{\beta}(\rho\Vert\sigma)$.

\item Additivity:\ For all $\alpha\in(0,1)\cup(1,\infty)$,
\begin{equation}
\widehat{D}_{\alpha}(\rho_{1}\otimes\rho_{2}\Vert\sigma_{1}\otimes\sigma
_{2})=\widehat{D}_{\alpha}(\rho_{1}\Vert\sigma_{1})+\widehat{D}_{\alpha}
(\rho_{2}\Vert\sigma_{2}). \label{eq:additivity-geometric-renyi}
\end{equation}

\item Direct-sum property: Let $p:\mathcal{X}\rightarrow\left[  0,1\right]  $
be a probability distribution over a finite alphabet $\mathcal{X}$ with
associated $\left\vert \mathcal{X}\right\vert $-dimensional system $X$, and
let $q:\mathcal{X}\rightarrow(0,\infty)$ be a positive function on
$\mathcal{X}$. Let $\left\{  \rho_{A}^{x}:x\in\mathcal{X}\right\}  $ be a set
of states on a system $A$, and let $\left\{  \sigma_{A}^{x}:x\in
\mathcal{X}\right\}  $ be a set of positive semi-definite operators on $A$.
Then,
\begin{equation}
\widehat{Q}_{\alpha}(\rho_{XA}\Vert\sigma_{XA})=\sum_{x\in\mathcal{X}
}p(x)^{\alpha}q(x)^{1-\alpha}\widehat{Q}_{\alpha}(\rho_{A}^{x}\Vert\sigma
_{A}^{x}), \label{eq:direct-sum-prop-geometric-renyi}
\end{equation}
where
\begin{align}
\rho_{XA}  &  :=\sum_{x\in\mathcal{X}}p(x)|x\rangle\!\langle x|_{X}\otimes
\rho_{A}^{x},\\
\sigma_{XA}  &  :=\sum_{x\in\mathcal{X}}q(x)|x\rangle\!\langle x|_{X}
\otimes\sigma_{A}^{x}.
\end{align}

\end{enumerate}
\end{proposition}

\begin{proof}
\begin{enumerate}
\item \textit{Proof of isometric invariance}: Let us start by writing
$\widehat{D}_{\alpha}(\rho\Vert\sigma)$ as in
\eqref{eq:def-geometric-renyi-rel-quasi-ent}--\eqref{eq:def-geometric-renyi-rel-ent}:
\begin{equation}
\widehat{D}_{\alpha}(\rho\Vert\sigma)=\lim_{\varepsilon\rightarrow0^{+}}
\frac{1}{\alpha-1} \ln\operatorname{Tr}\!\left[  \sigma_{\varepsilon}\!\left(
\sigma_{\varepsilon}^{-\frac{1}{2}}\rho\sigma_{\varepsilon}^{-\frac{1}{2}
}\right)  ^{\alpha}\right]  .
\end{equation}
where
\begin{equation}
\sigma_{\varepsilon}:=\sigma+\varepsilon I.
\end{equation}
Let $V$ be an isometry. Then, defining
\begin{equation}
\omega_{\varepsilon}:=V\sigma V^{\dag}+\varepsilon I,
\end{equation}
we find that
\begin{equation}
\widehat{D}_{\alpha}(V\rho V^{\dag}\Vert V\sigma V^{\dag})=\lim_{\varepsilon
\rightarrow0^{+}}\frac{1}{\alpha-1} \ln\operatorname{Tr}\!\left[
\omega_{\varepsilon}\!\left(  \omega_{\varepsilon}^{-\frac{1}{2}}V\rho
V^{\dag}\omega_{\varepsilon}^{-\frac{1}{2}}\right)  ^{\alpha}\right]  .
\end{equation}
Now let $\Pi:=VV^{\dag}$ be the projection onto the image of $V$, so that $\Pi
V=V$, and let $\hat{\Pi}:= I-\Pi$. Then, we can write
\begin{equation}
\omega_{\varepsilon}=V\sigma V^{\dag}+\varepsilon\Pi+\varepsilon\hat{\Pi
}=V\sigma_{\varepsilon}V^{\dag}+\varepsilon\hat{\Pi}.
\end{equation}
Since $V\sigma_{\varepsilon}V^{\dag}$ and $\varepsilon\hat{\Pi}$ are supported
on orthogonal subspaces, we obtain
\begin{equation}
\omega_{\varepsilon}^{-\frac{1}{2}}=V\sigma_{\varepsilon}^{^{-\frac{1}{2}}
}V^{\dag}+\varepsilon^{-\frac{1}{2}}\hat{\Pi}.
\end{equation}
Consider then that
\begin{align}
\omega_{\varepsilon}^{-\frac{1}{2}}V\rho V^{\dag}\omega_{\varepsilon}
^{-\frac{1}{2}}  &  =\left(  V\sigma_{\varepsilon}^{^{-\frac{1}{2}}}V^{\dag
}+\varepsilon^{-\frac{1}{2}}\hat{\Pi}\right)  \Pi V\rho V^{\dag}\Pi\left(
V\sigma_{\varepsilon}^{^{-\frac{1}{2}}}V^{\dag}+\varepsilon^{-\frac{1}{2}}
\hat{\Pi}\right) \\
&  =\left(  V\sigma_{\varepsilon}^{^{-\frac{1}{2}}}V^{\dag}\right)  \Pi V\rho
V^{\dag}\Pi\left(  V\sigma_{\varepsilon}^{^{-\frac{1}{2}}}V^{\dag}\right) \\
&  =V\sigma_{\varepsilon}^{^{-\frac{1}{2}}}\rho\sigma_{\varepsilon}
^{^{-\frac{1}{2}}}V^{\dag},
\end{align}
where the second equality follows because $\hat{\Pi}\Pi=\Pi\hat{\Pi}=0$. Thus,
\begin{equation}
\left(  \omega_{\varepsilon}^{-\frac{1}{2}}V\rho V^{\dag}\omega_{\varepsilon
}^{-\frac{1}{2}}\right)  ^{\alpha}=V\left(  \sigma_{\varepsilon}^{^{-\frac
{1}{2}}}\rho\sigma_{\varepsilon}^{^{-\frac{1}{2}}}\right)  ^{\alpha}V^{\dag},
\end{equation}
and we find that
\begin{align}
\operatorname{Tr}\!\left[  \omega_{\varepsilon}\!\left(  \omega_{\varepsilon
}^{-\frac{1}{2}}V\rho V^{\dag}\omega_{\varepsilon}^{-\frac{1}{2}}\right)
^{\alpha}\right]   &  =\operatorname{Tr}\!\left[  \left(  V\sigma
_{\varepsilon}V^{\dag}+\varepsilon\hat{\Pi}\right)  V\left(  \sigma
_{\varepsilon}^{-\frac{1}{2}}\rho\sigma_{\varepsilon}^{-\frac{1}{2}}\right)
^{\alpha}V^{\dag}\right] \\
&  =\operatorname{Tr}\!\left[  \sigma_{\varepsilon}\!\left(  \sigma
_{\varepsilon}^{^{-\frac{1}{2}}}\rho\sigma_{\varepsilon}^{^{-\frac{1}{2}}
}\right)  ^{\alpha}\right]  .
\end{align}
Since the equality
\begin{equation}
\operatorname{Tr}\!\left[  \omega_{\varepsilon}\!\left(  \omega_{\varepsilon
}^{-\frac{1}{2}}V\rho V^{\dag}\omega_{\varepsilon}^{-\frac{1}{2}}\right)
^{\alpha}\right]  =\operatorname{Tr}\!\left[  \sigma_{\varepsilon}\!\left(
\sigma_{\varepsilon}^{^{-\frac{1}{2}}}\rho\sigma_{\varepsilon}^{^{-\frac{1}
{2}}}\right)  ^{\alpha}\right]
\end{equation}
holds for all $\varepsilon>0$, we conclude the proof of isometric invariance
by taking the limit $\varepsilon\rightarrow0^{+}$.

\item \textit{Proof of monotonicity in }$\alpha$: We prove this by showing
that the derivative is non-negative for all $\alpha>0$. By applying
\eqref{eq:limit-eq-geometric-renyi}, we can consider $\rho$ and $\sigma$ to be
positive definite without loss of generality. By applying
\eqref{eq:alt-op-geo-mean-4-geo-ent}, consider that
\begin{align}
\widehat{Q}_{\alpha}(\rho\Vert\sigma)  &  =\operatorname{Tr}\!\left[
\rho\!\left(  \rho^{-\frac{1}{2}}\sigma\rho^{-\frac{1}{2}}\right)  ^{1-\alpha
}\right] \\
&  =\operatorname{Tr}\!\left[  \rho\!\left(  \rho^{\frac{1}{2}}\sigma^{-1}
\rho^{\frac{1}{2}}\right)  ^{\alpha-1}\right]  .
\label{eq:alt-way-for-geometric-in-mono}
\end{align}
Now defining $|\varphi^{\rho}\rangle=(\rho^{\frac{1}{2}}\otimes I)|\Gamma
\rangle$ as a purification of $\rho$, and setting
\begin{align}
\gamma &  =\alpha-1,\\
X  &  =\rho^{\frac{1}{2}}\sigma^{-1}\rho^{\frac{1}{2}},
\end{align}
we can write the geometric R\'{e}nyi relative entropy as
\begin{equation}
\widehat{D}_{\alpha}(\rho\Vert\sigma)=\frac{1}{\gamma}\ln\langle\varphi^{\rho
}|X^{\gamma}\otimes I|\varphi^{\rho}\rangle,
\end{equation}
where we made use of \eqref{eq:alt-way-for-geometric-in-mono}. Then
$\frac{\partial}{\partial\alpha}=\frac{\partial}{\partial\gamma}\frac
{\partial\gamma}{\partial\alpha}=\frac{\partial}{\partial\gamma}$, and so we
find that
\begin{align}
&  \frac{\partial}{\partial\alpha}\widehat{D}_{\alpha}(\rho\Vert
\sigma)\nonumber\\
&  =\frac{\partial}{\partial\gamma}\left[  \frac{1}{\gamma}\ln\langle
\varphi^{\rho}|X^{\gamma}\otimes I|\varphi^{\rho}\rangle\right] \\
&  =\left[  -\frac{1}{\gamma^{2}}\ln\langle\varphi^{\rho}|X^{\gamma}\otimes
I|\varphi^{\rho}\rangle+\frac{1}{\gamma}\frac{\partial}{\partial\gamma}
\ln\langle\varphi^{\rho}|X^{\gamma}\otimes I|\varphi^{\rho}\rangle\right] \\
&  =\left[  -\frac{1}{\gamma^{2}}\ln\langle\varphi^{\rho}|X^{\gamma}\otimes
I|\varphi^{\rho}\rangle+\frac{1}{\gamma}\frac{\langle\varphi^{\rho}|X^{\gamma
}\ln X\otimes I|\varphi^{\rho}\rangle}{\langle\varphi^{\rho}|X^{\gamma}\otimes
I|\varphi^{\rho}\rangle}\right] \\
&  =\left[  \frac{-\langle\varphi^{\rho}|X^{\gamma}\otimes I|\varphi^{\rho
}\rangle\ln\langle\varphi^{\rho}|X^{\gamma}\otimes I|\varphi^{\rho}
\rangle+\gamma\langle\varphi^{\rho}|X^{\gamma}\ln X\otimes I|\varphi^{\rho
}\rangle}{\gamma^{2}\langle\varphi^{\rho}|X^{\gamma}\otimes I|\varphi^{\rho
}\rangle}\right] \\
&  =\left[  \frac{-\langle\varphi^{\rho}|X^{\gamma}\otimes I|\varphi^{\rho
}\rangle\ln\langle\varphi^{\rho}|X^{\gamma}\otimes I|\varphi^{\rho}
\rangle+\langle\varphi^{\rho}|X^{\gamma}\ln X^{\gamma}\otimes I|\varphi^{\rho
}\rangle}{\gamma^{2}\langle\varphi^{\rho}|X^{\gamma}\otimes I|\varphi^{\rho
}\rangle}\right]  .
\end{align}
Letting $g(x):=x\ln x$, we write
\begin{equation}
\frac{\partial}{\partial\alpha}\widehat{D}_{\alpha}(\rho\Vert\sigma
)=\frac{\langle\varphi^{\rho}|g(X^{\gamma}\otimes I)|\varphi^{\rho}
\rangle-g(\langle\varphi^{\rho}|(X^{\gamma}\otimes I)|\varphi^{\rho}\rangle
)}{\gamma^{2}\langle\varphi^{\rho}|X^{\gamma}\otimes I|\varphi^{\rho}\rangle}.
\end{equation}
Then, since $g(x)$ is operator convex, by the operator Jensen inequality
\cite{HP03}, we conclude that
\begin{equation}
\langle\varphi^{\rho}|g(X^{\gamma}\otimes I)|\varphi^{\rho}\rangle\geq
g(\langle\varphi^{\rho}|(X^{\gamma}\otimes I)|\varphi^{\rho}\rangle),
\end{equation}
which means that $\frac{\partial}{\partial\alpha}\widehat{D}_{\alpha}
(\rho\Vert\sigma)\geq0$. Therefore, $\widehat{D}_{\alpha}(\rho\Vert\sigma)$ is
monotonically increasing in $\alpha$, as required.

\item \textit{Proof of additivity}:\ The proof of
\eqref{eq:additivity-geometric-renyi}\ is found by direct evaluation.

\item \textit{Proof of direct-sum property}: The proof of
\eqref{eq:direct-sum-prop-geometric-renyi}\ is found by direct evaluation.
\end{enumerate}
\end{proof}

\bigskip

We now recall the data-processing inequality for the geometric R\'{e}nyi
relative entropy for $\alpha\in\left(  0,1\right)  \cup(1,2]$. This was
established by an operator-theoretic approach in \cite{PR98} and by an
operational method in \cite{Mat13,Matsumoto2018}. The operator-theoretic
method has its roots in \cite[Proposition~2.5]{HP91} and was reviewed in
\cite[Corollary 3.31]{HM17}. We follow the operator-theoretic approach here.

\begin{theorem}
[Data-processing inequality for geometric R\'{e}nyi relative entropy]
\label{thm:data-proc-geometric-renyi}Let $\rho$ be a state, $\sigma$ a
positive semi-definite operator, and $\mathcal{N}$ a quantum channel. Then,
for all $\alpha\in\left(  0,1\right)  \cup(1,2]$, the following inequality
holds
\begin{equation}
\widehat{D}_{\alpha}(\rho\Vert\sigma)\geq\widehat{D}_{\alpha}(\mathcal{N}
(\rho)\Vert\mathcal{N}(\sigma)).
\end{equation}

\end{theorem}

\begin{proof}
From Stinespring's dilation theorem \cite{Sti55}, we know that the action of a
quantum channel $\mathcal{N}$ on any linear operator $X$ can be written as
\begin{equation}
\mathcal{N}(X)=\operatorname{Tr}_{E}[VXV^{\dag}],
\end{equation}
where $V$ is an isometry and $E$ is an auxiliary system with dimension
$d_{E}\geq\ $rank$(\Gamma_{AB}^{\mathcal{N}})$, with $\Gamma_{AB}
^{\mathcal{N}}$ the Choi operator for the channel $\mathcal{N}$. As stated in
Proposition~\ref{prop:geometric-renyi-props}, the geometric R\'{e}nyi relative
entropy $\widehat{D}_{\alpha}$ is isometrically invariant. Therefore, it
suffices to establish the data-processing inequality for $\widehat{D}_{\alpha
}$ under partial trace; i.e., it suffices to show that for any state
$\rho_{AB}$ and any positive semi-definite operator $\sigma_{AB}$,
\begin{equation}
\widehat{D}_{\alpha}(\rho_{AB}\Vert\sigma_{AB})\geq\widehat{D}_{\alpha}
(\rho_{A}\Vert\sigma_{A})\qquad\forall\alpha\in\left(  0,1\right)  \cup(1,2].
\end{equation}
We now proceed to prove this inequality. We prove it for $\rho_{AB}$, and
hence $\rho_{A}$, invertible, as well as for $\sigma_{AB}$ and $\sigma_{A}$
invertible. The result follows in the general case of $\rho_{AB}$ and/or
$\rho_{A}$ non-invertible, as well as $\sigma_{AB}$ and/or $\sigma_{A}$
non-invertible, by applying the result to the invertible operators $\left(
1-\delta\right)  \rho_{AB}+\delta\pi_{AB}$ and $\sigma_{AB}+\varepsilon
I_{AB}$, with $\delta\in(0,1)$ and $\varepsilon>0$, and taking the limit
$\delta\rightarrow0^{+}$ followed by $\varepsilon\rightarrow0^{+}$, because
\begin{align}
\widehat{D}_{\alpha}(\rho_{AB}\Vert\sigma_{AB})  &  =\lim_{\varepsilon
\rightarrow0^{+}}\lim_{\delta\rightarrow0^{+}}\widehat{D}_{\alpha}(\left(
1-\delta\right)  \rho_{AB}+\delta\pi_{AB}\Vert\sigma_{AB}+\varepsilon
I_{AB}),\\
\widehat{D}_{\alpha}(\rho_{A}\Vert\sigma_{A})  &  =\lim_{\varepsilon
\rightarrow0^{+}}\lim_{\delta\rightarrow0^{+}}\widehat{D}_{\alpha}(\left(
1-\delta\right)  \rho_{A}+\delta\pi_{A}\Vert\sigma_{A}+d_{B}\varepsilon
I_{A}),
\end{align}
which follows from \eqref{eq:limit-eq-geometric-renyi} and the fact that the
dimensional factor $d_{B}$ does not affect the limit in the second quantity above.

To establish the data-processing inequality, we make use of the Petz recovery
channel for partial trace \cite{Petz1986,Petz1988}, as well as the operator
Jensen inequality \cite{HP03}. Recall that the Petz recovery channel
$\mathcal{P}_{\sigma_{AB},\operatorname{Tr}_{B}}$ for partial trace is defined
as
\begin{equation}
\mathcal{P}_{\sigma_{AB},\operatorname{Tr}_{B}}(X_{A})\equiv\mathcal{P}
(X_{A}):=\sigma_{AB}^{\frac{1}{2}}\left(  \sigma_{A}^{-\frac{1}{2}}X_{A}
\sigma_{A}^{-\frac{1}{2}}\otimes I_{B}\right)  \sigma_{AB}^{\frac{1}{2}}.
\end{equation}
The Petz recovery channel has the following property:
\begin{equation}
\mathcal{P}(\sigma_{A})=\sigma_{AB}, \label{eq:petz-map-recovers-sig-geo-DP}
\end{equation}
which can be verified by inspection. Since $\mathcal{P}_{\sigma_{AB}
,\operatorname{Tr}_{B}}$ is completely positive and trace preserving, it
follows that its adjoint
\begin{equation}
\mathcal{P}^{\dag}(Y_{AB}):=\sigma_{A}^{-\frac{1}{2}}\operatorname{Tr}
_{B}[\sigma_{AB}^{\frac{1}{2}}Y_{AB}\sigma_{AB}^{\frac{1}{2}}]\sigma
_{A}^{-\frac{1}{2}},
\end{equation}
is completely positive and unital. Observe that
\begin{equation}
\mathcal{P}^{\dag}(\sigma_{AB}^{-\frac{1}{2}}\rho_{AB}\sigma_{AB}^{-\frac
{1}{2}})=\sigma_{A}^{-\frac{1}{2}}\rho_{A}\sigma_{A}^{-\frac{1}{2}}.
\label{eq:adjoint-petz-funny-prop-geo-DP}
\end{equation}
We then find for $\alpha\in(1,2]$ that
\begin{align}
\widehat{Q}_{\alpha}(\rho_{AB}\Vert\sigma_{AB})  &  =\operatorname{Tr}
\!\left[  \sigma_{AB}\!\left(  \sigma_{AB}^{-\frac{1}{2}}\rho_{AB}\sigma
_{AB}^{-\frac{1}{2}}\right)  ^{\alpha}\right] \\
&  =\operatorname{Tr}\!\left[  \mathcal{P}(\sigma_{A})\!\left(  \sigma
_{AB}^{-\frac{1}{2}}\rho_{AB}\sigma_{AB}^{-\frac{1}{2}}\right)  ^{\alpha
}\right] \\
&  =\operatorname{Tr}\!\left[  \sigma_{A}\mathcal{P}^{\dag}\!\left(
\sigma_{AB}^{-\frac{1}{2}}\rho_{AB}\sigma_{AB}^{-\frac{1}{2}}\right)
^{\alpha}\right] \\
&  \geq\operatorname{Tr}\!\left[  \sigma_{A}\!\left(  \mathcal{P}^{\dag
}\!\left(  \sigma_{AB}^{-\frac{1}{2}}\rho_{AB}\sigma_{AB}^{-\frac{1}{2}
}\right)  \right)  ^{\alpha}\right] \\
&  =\operatorname{Tr}\!\left[  \sigma_{A}\!\left(  \sigma_{A}^{-\frac{1}{2}
}\rho_{A}\sigma_{A}^{-\frac{1}{2}}\right)  ^{\alpha}\right] \\
&  =\widehat{Q}_{\alpha}(\rho_{A}\Vert\sigma_{A}).
\end{align}
The second equality follows from \eqref{eq:petz-map-recovers-sig-geo-DP}. The
sole inequality is a consequence of the operator Jensen inequality and the
fact that $x^{\alpha}$ is operator convex for $\alpha\in(1,2]$. Indeed, for
$\mathcal{M}$ a completely positive unital map, it follows from the operator
Jensen inequality that
\begin{equation}
f(\mathcal{M}(X))\leq\mathcal{M}(f(X)) \label{eq:Choi-thm-ext-op-jensen}
\end{equation}
for Hermitian $X$ and an operator convex function $f$. The second-to-last
equality follows from \eqref{eq:adjoint-petz-funny-prop-geo-DP}.

Applying the same reasoning as above, but using the fact that $x^{\alpha}$ is
operator concave for $\alpha\in(0,1)$, we find for $\alpha\in(0,1)$ that
\begin{equation}
\widehat{Q}_{\alpha}(\rho_{A}\Vert\sigma_{A})\geq\widehat{Q}_{\alpha}
(\rho_{AB}\Vert\sigma_{AB}).
\end{equation}
Putting together the above and employing definitions, we find that the
following inequality holds for $\alpha\in(0,1)\cup(1,2]$:
\begin{equation}
\widehat{D}_{\alpha}(\rho_{AB}\Vert\sigma_{AB})\geq\widehat{D}_{\alpha}
(\rho_{A}\Vert\sigma_{A}),
\end{equation}
concluding the proof.
\end{proof}

\bigskip

With the data-processing inequality for the geometric R\'{e}nyi relative
entropy in hand, we can easily establish some additional properties.

\begin{proposition}
[Additional Properties of the Geometric R\'enyi Relative Entropy]The geometric
R\'enyi relative entropy $\widehat{D}_{\alpha}$ satisfies the following
properties for all states $\rho$ and positive semi-definite operators $\sigma$
for $\alpha\in\left(  0,1\right)  \cup(1,2]$.

\begin{enumerate}
\item If $\operatorname{Tr}[\sigma]\leq\operatorname{Tr}[\rho]=1$, then
$\widehat{D}_{\alpha}(\rho\Vert\sigma)\geq0$.

\item Faithfulness: Suppose that $\operatorname{Tr}[\sigma]\leq
\operatorname{Tr}[\rho]=1$ and let $\alpha\in(0,1)\cup(1,\infty)$. Then
$\widehat{D}_{\alpha}(\rho\Vert\sigma)=0$ if and only if $\rho=\sigma$.

\item If $\rho\leq\sigma$, then $\widehat{D}_{\alpha}(\rho\Vert\sigma)\leq0$.

\item For any positive semi-definite operator $\sigma^{\prime}$ such that
$\sigma^{\prime}\geq\sigma$, the following inequality holds $\widehat
{D}_{\alpha}(\rho\Vert\sigma^{\prime})\leq\widehat{D}_{\alpha}(\rho\Vert
\sigma)$.
\end{enumerate}
\end{proposition}

\begin{proof}
\begin{enumerate}
\item Apply the data processing inequality with the channel being the full
trace-out channel:
\begin{align}
\widehat{D}_{\alpha}(\rho\Vert\sigma)  &  \geq\widehat{D}_{\alpha
}(\operatorname{Tr}[\rho]\Vert\operatorname{Tr}[\sigma])\\
&  =\frac{1}{\alpha-1}\ln\left[  \left(  \operatorname{Tr}[\rho]\right)
^{\alpha}\left(  \operatorname{Tr}[\sigma]\right)  ^{1-\alpha}\right] \\
&  =-\ln\operatorname{Tr}[\sigma]\\
&  \geq0.
\end{align}

\item If $\rho=\sigma$, then it follows by direct evaluation that $\widehat
{D}_{\alpha}(\rho\Vert\sigma)=0$. Suppose first that $\left(  0,1\right)
\cup(1,2]$. Then $\widehat{D}_{\alpha}(\rho\Vert\sigma)=0$ implies that
$\widehat{D}_{\alpha}(\mathcal{M}(\rho)\Vert\mathcal{M}(\sigma))=0$ for all
measurement channels $\mathcal{M}$. This includes informationally complete
measurements \cite{P77,B91,RBSC04}. By applying the faithfulness of the
classical R\'{e}nyi relative entropy and the informationally completeness
property, we conclude that $\rho=\sigma$. To get the range outside the
data-processing interval of $\left(  0,1\right)  \cup(1,2]$, note that
$\widehat{D}_{\alpha}(\rho\Vert\sigma)=0$ for $\alpha>2$ implies by
monotonicity (Property~2\ of Proposition~\ref{prop:geometric-renyi-props})
that $\widehat{D}_{\alpha}(\rho\Vert\sigma)=0$ for $\alpha\leq2$. Then it
follows that $\rho=\sigma$. The other implication follows for $\alpha
\in(0,1)\cup(1,\infty)$ by direct evaluation.

\item Consider that $\rho\leq\sigma$ implies that $\sigma-\rho\geq0$. Then
define the following positive semi-definite operators:
\begin{align}
\hat{\rho}  &  :=|0\rangle\!\langle0|\otimes\rho,\\
\hat{\sigma}  &  :=|0\rangle\!\langle0|\otimes\rho+|1\rangle\!\langle
1|\otimes\left(  \sigma-\rho\right)  .
\end{align}
By exploiting the direct-sum property of geometric R\'{e}nyi relative entropy
(Proposition~\ref{prop:geometric-renyi-props}) and the data-processing
inequality (Theorem~\ref{thm:data-proc-geometric-renyi}), we find that
\begin{equation}
0=\widehat{D}_{\alpha}(\rho\Vert\rho)=\widehat{D}_{\alpha}(\hat{\rho}\Vert
\hat{\sigma})\geq\widehat{D}_{\alpha}(\rho\Vert\sigma),
\end{equation}
where the inequality follows from data processing with respect to partial
trace over the classical register.

\item Similar to the above proof, the condition $\sigma^{\prime}\geq\sigma$
implies that $\sigma^{\prime}-\sigma\geq0$. Then define the following positive
semi-definite operators:
\begin{align}
\hat{\rho}  &  :=|0\rangle\!\langle0|\otimes\rho,\\
\hat{\sigma}  &  :=|0\rangle\!\langle0|\otimes\sigma+|1\rangle\!\langle
1|\otimes\left(  \sigma^{\prime}-\sigma\right)  .
\end{align}
By exploiting the direct-sum property of geometric R\'{e}nyi relative entropy
(Proposition~\ref{prop:geometric-renyi-props}) and the data-processing
inequality (Theorem~\ref{thm:data-proc-geometric-renyi}), we find that
\begin{equation}
\widehat{D}_{\alpha}(\rho\Vert\sigma)=\widehat{D}_{\alpha}(\hat{\rho}\Vert
\hat{\sigma})\geq\widehat{D}_{\alpha}(\rho\Vert\sigma^{\prime}),
\end{equation}
where the inequality follows from data processing with respect to partial
trace over the classical register.
\end{enumerate}
\end{proof}

\bigskip

The data-processing inequality for the geometric R\'{e}nyi relative entropy
can be written using the geometric R\'{e}nyi relative quasi-entropy
$\widehat{Q}_{\alpha}(\rho\Vert\sigma)$ as
\begin{equation}
\frac{1}{\alpha-1} \ln\widehat{Q}_{\alpha}(\rho\Vert\sigma)\geq\frac{1}
{\alpha-1} \ln\widehat{Q}_{\alpha}(\mathcal{N}(\rho)\Vert\mathcal{N}(\sigma)).
\end{equation}
Since $\alpha-1$ is negative for $\alpha\in(0,1)$, we can use the monotonicity
of the function $\ln$ to obtain
\begin{align}
\widehat{Q}_{\alpha}(\rho\Vert\sigma)  &  \geq\widehat{Q}_{\alpha}
(\mathcal{N}(\rho)\Vert\mathcal{N}(\sigma)),\qquad\text{for }\alpha\in(1,2],\\
\widehat{Q}_{\alpha}(\rho\Vert\sigma)  &  \leq\widehat{Q}_{\alpha}
(\mathcal{N}(\rho)\Vert\mathcal{N}(\sigma)),\qquad\text{for }\alpha\in(0,1).
\end{align}
We can use this to establish some convexity statements for the geometric
R\'{e}nyi relative entropy.

\begin{proposition}
Let $p:\mathcal{X}\rightarrow\left[  0,1\right]  $ be a probability
distribution over a finite alphabet $\mathcal{X}$ with associated $\left\vert
\mathcal{X}\right\vert $-dimensional system $X$, let $\left\{  \rho_{A}
^{x}:x\in\mathcal{X}\right\}  $ be a set of states on system $A$, and let
$\left\{  \sigma_{A}^{x}:x\in\mathcal{X}\right\}  $ be a set of positive
semi-definite operators on $A$. Then, for $\alpha\in(1,2]$,
\begin{equation}
\widehat{Q}_{\alpha}\!\left(  \sum_{x\in\mathcal{X}}p(x)\rho_{A}
^{x}\middle\Vert\sum_{x\in\mathcal{X}}p(x)\sigma_{A}^{x}\right)  \leq
\sum_{x\in\mathcal{X}}p(x)\widehat{Q}_{\alpha}(\rho_{A}^{x}\Vert\sigma_{A}
^{x}), \label{eq:convex-quasi-geometric-renyi-rel-alpha-higher-1}
\end{equation}
and for $\alpha\in(0,1)$,
\begin{equation}
\widehat{Q}_{\alpha}\!\left(  \sum_{x\in\mathcal{X}}p(x)\rho_{A}
^{x}\middle\Vert\sum_{x\in\mathcal{X}}p(x)\sigma_{A}^{x}\right)  \geq
\sum_{x\in\mathcal{X}}p(x)\widehat{Q}_{\alpha}(\rho_{A}^{x}\Vert\sigma_{A}
^{x}).
\end{equation}
Consequently, the geometric R\'{e}nyi relative entropy $\widehat{D}_{\alpha}$
is jointly convex for $\alpha\in(0,1)$:
\begin{equation}
\widehat{D}_{\alpha}\!\left(  \sum_{x\in\mathcal{X}}p(x)\rho_{A}
^{x}\middle\Vert\sum_{x\in\mathcal{X}}p(x)\sigma_{A}^{x}\right)  \leq
\sum_{x\in\mathcal{X}}p(x)\widehat{D}_{\alpha}(\rho_{A}^{x}\Vert\sigma_{A}
^{x}). \label{eq:quasi-convex-quasi-geometric-renyi-rel-alpha-higher-1}
\end{equation}

\end{proposition}

\begin{proof}
The first two inequalities follow directly from the direct-sum property of
geometric R\'{e}nyi relative entropy
(Proposition~\ref{prop:geometric-renyi-props}) and the data-processing
inequality (Theorem~\ref{thm:data-proc-geometric-renyi}). The last inequality
follows from the first by applying the logarithm and scaling by $1/\left(
\alpha-1\right)  $ and taking a maximum.
\end{proof}

\bigskip

Although the geometric R\'{e}nyi relative entropy is not jointly convex for
$\alpha\in(1,2]\,$, it is jointly quasi-convex, in the sense that
\begin{equation}
\widehat{D}_{\alpha}\!\left(  \sum_{x\in\mathcal{X}}p(x)\rho_{A}
^{x}\middle\Vert\sum_{x\in\mathcal{X}}p(x)\sigma_{A}^{x}\right)  \leq
\max_{x\in\mathcal{X}}\widehat{D}_{\alpha}(\rho_{A}^{x}\Vert\sigma_{A}^{x}),
\end{equation}
for any finite alphabet $\mathcal{X}$, probability distribution $p:\mathcal{X}
\rightarrow\left[  0,1\right]  $, set $\left\{  \rho_{A}^{x}:x\in
\mathcal{X}\right\}  $ of states, and set $\left\{  \sigma_{A}^{x}
:x\in\mathcal{X}\right\}  $ of positive semi-definite operators. Indeed, from
\eqref{eq:convex-quasi-geometric-renyi-rel-alpha-higher-1}, we immediately
obtain
\begin{equation}
\widehat{Q}_{\alpha}\!\left(  \sum_{x\in\mathcal{X}}p(x)\rho_{A}
^{x}\middle\Vert\sum_{x\in\mathcal{X}}p(x)\sigma_{A}^{x}\right)  \leq
\max_{x\in\mathcal{X}}\widehat{Q}_{\alpha}(\rho_{A}^{x}\Vert\sigma_{A}^{x}).
\end{equation}
Taking the logarithm and multiplying by $\frac{1}{\alpha-1}$ on both sides of
this inequality leads to \eqref{eq:quasi-convex-quasi-geometric-renyi-rel-alpha-higher-1}.

The geometric R\'{e}nyi relative entropy has another interpretation, which was
discovered in \cite{Mat13,Matsumoto2018} and is worthwhile to mention.

\begin{proposition}
[Geometric R\'{e}nyi relative entropy from classical preparations]
\label{prop:geometric-renyi-from-classical-preps}Let $\rho$ be a state and
$\sigma$ a positive semi-definite operator satisfying $\operatorname{supp}
(\rho)\subseteq\operatorname{supp}(\sigma)$. For all $\alpha\in(0,1)\cup
(1,2]$, the geometric R\'{e}nyi relative entropy is equal to the smallest
value that the classical R\'{e}nyi relative entropy can take by minimizing
over classical--quantum channels that realize the state $\rho$ and the
positive semi-definite operator $\sigma$. That is, the following equality
holds
\begin{equation}
\widehat{D}_{\alpha}(\rho\Vert\sigma)=\inf_{\left\{  p,q,\mathcal{P}\right\}
}\left\{  D_{\alpha}(p\Vert q):\mathcal{P}(p)=\rho,\mathcal{P}(q)=\sigma
\right\}  , \label{eq:matsumoto-geo-renyi-preps}
\end{equation}
where the classical R\'{e}nyi relative entropy is defined as
\begin{equation}
D_{\alpha}(p\Vert q) := \frac{1}{\alpha- 1} \sum_{x \in\mathcal{X}}
p(x)^{\alpha} q(x)^{1-\alpha},
\end{equation}
the channel $\mathcal{P}$ is a classical--quantum channel, $p:\mathcal{X}
\rightarrow\left[  0,1\right]  $ is a probability distribution over a finite
alphabet $\mathcal{X}$, and $q:\mathcal{X}\rightarrow(0,\infty)$ is a positive
function on $\mathcal{X}$.
\end{proposition}

\begin{proof}
First, let us define the classical (diagonal)\ state $\omega(p)$ and diagonal
positive semi-definite operator $\omega(q)$ as an embedding of the respective
probability distribution $p$ and positive function $q$:
\begin{equation}
\omega(p):=\sum_{x\in\mathcal{X}}p(x)|x\rangle\!\langle x|,\qquad\omega
(q):=\sum_{x\in\mathcal{X}}q(x)|x\rangle\!\langle x|,
\end{equation}
and suppose that there exists a quantum channel $\mathcal{P}$ such that
\begin{equation}
\mathcal{P}(\omega(p))=\rho,\qquad\mathcal{P}(\omega(q))=\sigma.
\label{eq:constraint-p-q-geo-renyi}
\end{equation}
Then consider the following chain of inequalities:
\begin{align}
D_{\alpha}(p\Vert q)  &  =\widehat{D}_{\alpha}(\omega(p)\Vert\omega(q))\\
&  \geq\widehat{D}_{\alpha}(\mathcal{P}(\omega(p))\Vert\mathcal{P}
(\omega(q)))\\
&  =\widehat{D}_{\alpha}(\rho\Vert\sigma).
\end{align}
The first equality follows because the geometric R\'{e}nyi relative entropy
reduces to the classical R\'{e}nyi relative entropy for commuting operators.
The inequality is a consequence of the data-processing inequality for the
geometric R\'{e}nyi relative entropy
(Theorem~\ref{thm:data-proc-geometric-renyi}). The final equality follows from
the constraint in \eqref{eq:constraint-p-q-geo-renyi}. Since the inequality
holds for arbitrary $p$, $q$, and $\mathcal{P}$ satisfying
\eqref{eq:constraint-p-q-geo-renyi}, we conclude that
\begin{equation}
\inf_{\left\{  p,q,\mathcal{P}\right\}  }\left\{  D_{\alpha}(p\Vert
q):\mathcal{P}(p)=\rho,\mathcal{P}(q)=\sigma\right\}  \geq\widehat{D}_{\alpha
}(\rho\Vert\sigma). \label{eq:ineq-classical-prep-bigger-geo-ren}
\end{equation}

The equality in \eqref{eq:matsumoto-geo-renyi-preps}\ then follows by
demonstrating a specific distribution $p$, positive function $q$, and
preparation channel $\mathcal{P}$ that saturate the inequality in
\eqref{eq:ineq-classical-prep-bigger-geo-ren}. The optimal choices of $p$,
$q$, and $\mathcal{P}$ are given by
\begin{align}
p(x)  &  :=\lambda_{x}q(x),\label{eq:optimal-choices-classical-preps-1}\\
q(x)  &  :=\operatorname{Tr}[\Pi_{x}\sigma],\\
\mathcal{P}(\cdot)  &  :=\sum_{x}\langle x|(\cdot)|x\rangle\frac{\sigma
^{\frac{1}{2}}\Pi_{x}\sigma^{\frac{1}{2}}}{q(x)},
\label{eq:optimal-choices-classical-preps-3}
\end{align}
where the spectral decomposition of the positive semi-definite operator
$\sigma^{-\frac{1}{2}}\rho\sigma^{-\frac{1}{2}}$ is given by
\begin{equation}
\sigma^{-\frac{1}{2}}\rho\sigma^{-\frac{1}{2}}=\sum_{x}\lambda_{x}\Pi_{x}.
\label{eq:def-Delta-geo-renyi}
\end{equation}
The choice of $p(x)$ above is a probability distribution because
\begin{align}
\sum_{x}p(x)  &  =\sum_{x}\lambda_{x}q(x) =\sum_{x}\lambda_{x}
\operatorname{Tr}[\Pi_{x}\sigma] =\operatorname{Tr}[\sigma^{-\frac{1}{2}}
\rho\sigma^{-\frac{1}{2}}\sigma] =\operatorname{Tr}[\Pi_{\sigma}\rho] =1.
\end{align}
The preparation channel $\mathcal{P}$ is a classical--quantum channel that
measures the input in the basis $\{|x\rangle\}_{x}$ and prepares the state
$\frac{\sigma^{\frac{1}{2}}\Pi_{x}\sigma^{\frac{1}{2}}}{q(x)}$ if the
measurement outcome is $x$. We find that
\begin{align}
\mathcal{P}(\omega(p))  &  =\sum_{x}\frac{p(x)}{q(x)}\sigma^{\frac{1}{2}}
\Pi_{x}\sigma^{\frac{1}{2}} =\sum_{x}\frac{\lambda_{x}q(x)}{q(x)}\sigma
^{\frac{1}{2}}\Pi_{x}\sigma^{\frac{1}{2}} =\sigma^{\frac{1}{2}}\left(
\sum_{x}\lambda_{x}\Pi_{x}\right)  \sigma^{\frac{1}{2}}\\
&  =\sigma^{\frac{1}{2}}\left(  \sigma^{-\frac{1}{2}}\rho\sigma^{-\frac{1}{2}
}\right)  \sigma^{\frac{1}{2}} =\Pi_{\sigma}\rho\Pi_{\sigma} =\rho,
\end{align}
and
\begin{align}
\mathcal{P}(\omega(q))  &  =\sum_{x}\frac{q(x)}{q(x)}\sigma^{\frac{1}{2}}
\Pi_{x}\sigma^{\frac{1}{2}} =\sigma^{\frac{1}{2}}\left(  \sum_{x}\Pi
_{x}\right)  \sigma^{\frac{1}{2}} =\sigma.
\end{align}
Finally, consider the classical R\'{e}nyi relative quasi-entropy:
\begin{align}
\sum_{x}p(x)^{\alpha}q(x)^{1-\alpha}  &  =\sum_{x}\left(  \lambda
_{x}q(x)\right)  ^{\alpha}q(x)^{1-\alpha} =\sum_{x}\lambda_{x}^{\alpha}q(x)
=\sum_{x}\lambda_{x}^{\alpha}\operatorname{Tr}[\Pi_{x}\sigma]\\
&  =\operatorname{Tr}\!\left[  \sigma\!\left(  \sum_{x}\lambda_{x}^{\alpha}
\Pi_{x}\right)  \right]  =\operatorname{Tr}\!\left[  \sigma\!\left(
\sigma^{-\frac{1}{2}}\rho\sigma^{-\frac{1}{2}}\right)  ^{\alpha}\right]
=\widehat{Q}_{\alpha}(\rho\Vert\sigma),
\end{align}
where the second-to-last equality follows from the spectral decomposition in
\eqref{eq:def-Delta-geo-renyi} and the form of the geometric R\'{e}nyi
relative quasi-entropy from
Proposition~\ref{prop:explicit-form-geometric-renyi}. As a consequence of the
equality
\begin{equation}
\sum_{x}p(x)^{\alpha}q(x)^{1-\alpha}=\widehat{Q}_{\alpha}(\rho\Vert\sigma),
\end{equation}
and the fact that these choices of $p$, $q$, and $\mathcal{P}$ satisfy the
constraints $\mathcal{P}(p)=\rho$ and $\mathcal{P}(q)=\sigma$, we conclude
that
\begin{equation}
D_{\alpha}(p\Vert q)=\widehat{D}_{\alpha}(\rho\Vert\sigma).
\end{equation}
Combining this equality with \eqref{eq:ineq-classical-prep-bigger-geo-ren}, we
conclude the equality in \eqref{eq:matsumoto-geo-renyi-preps}.
\end{proof}

\bigskip

The following proposition recalls the ordering between the sandwiched, Petz--,
and geometric R\'{e}nyi relative entropies for the interval $\alpha
\in(0,1)\cup(1,2]$. The first inequality in
Proposition~\ref{prop:sand-Petz-geo-ineqs} was established for $\alpha
\in(1,2]$ in \cite{WWY14}\ and for $\alpha\in(0,1)$ in \cite{DL14limit}, by
employing the Araki--Lieb--Thirring inequality \cite{Araki1990,LT76}. The
second inequality was established by \cite{Mat13,Matsumoto2018} and reviewed
in \cite{T15book}. It follows by applying similar reasoning as in the proof of
Proposition~\ref{prop:geometric-renyi-from-classical-preps}.

\begin{proposition}
\label{prop:sand-Petz-geo-ineqs}Let $\rho$ be a state and $\sigma$ a positive
semi-definite operator. For $\alpha\in(0,1)\cup(1,2]$, the following
inequalities hold
\begin{equation}
\widetilde{D}_{\alpha}(\rho\Vert\sigma)\leq D_{\alpha}(\rho\Vert\sigma
)\leq\widehat{D}_{\alpha}(\rho\Vert\sigma),
\label{eq:sandwiched-Petz-geometric-ineqs}
\end{equation}
for the sandwiched ($\widetilde{D}_{\alpha}$), Petz ($D_{\alpha}$), and
geometric ($\widehat{D}_{\alpha}$) R\'{e}nyi relative entropies.
\end{proposition}

\begin{proof}
As stated above, the first inequality follows from the Araki--Lieb--Thirring
inequalities in \eqref{eq:ALT-1}--\eqref{eq:ALT-2}\ by picking $q=1$,
$r=\alpha$, $X=\rho$, and $Y=\sigma^{\frac{1-\alpha}{\alpha}}$. So we recall
the proof of the second inequality here. Suppose that $\mathcal{P}$ is a
classical--quantum channel, $p:\mathcal{X}\rightarrow\left[  0,1\right]  $ is
a probability distribution over a finite alphabet $\mathcal{X}$, and
$q:\mathcal{X}\rightarrow(0,\infty)$ is a positive function on $\mathcal{X}$
satisfying
\begin{equation}
\mathcal{P}(\omega(p))=\rho,\qquad\mathcal{P}(\omega(q))=\sigma,
\label{eq:geo-constraint-up-bnd-Petz-R}
\end{equation}
where
\begin{equation}
\omega(p):=\sum_{x\in\mathcal{X}}p(x)|x\rangle\!\langle x|,\qquad\omega
(q):=\sum_{x\in\mathcal{X}}q(x)|x\rangle\!\langle x|.
\end{equation}
Then consider the following chain of inequalities:
\begin{align}
D_{\alpha}(p\Vert q)  &  =D_{\alpha}(\omega(p)\Vert\omega(q))\\
&  \geq D_{\alpha}(\mathcal{P}(\omega(p))\Vert\mathcal{P}(\omega(q)))\\
&  =D_{\alpha}(\rho\Vert\sigma).
\end{align}
The first equality follows because the Petz--R\'{e}nyi relative entropy
reduces to the classical R\'{e}nyi relative entropy for commuting operators.
The inequality follows from the data-processing inequality for the
Petz--R\'{e}nyi relative entropy for $\alpha\in(0,1)\cup(1,2]$ \cite{P85,P86}.
The final equality follows from the constraint in
\eqref{eq:geo-constraint-up-bnd-Petz-R}. Since the inequality above holds for
all $p$, $q$, and $\mathcal{P}$ satisfying
\eqref{eq:geo-constraint-up-bnd-Petz-R}, we conclude that
\begin{equation}
\inf_{\left\{  p,q,\mathcal{P}\right\}  }\left\{  D_{\alpha}(p\Vert
q):\mathcal{P}(p)=\rho,\mathcal{P}(q)=\sigma\right\}  \geq D_{\alpha}
(\rho\Vert\sigma).
\end{equation}
Now applying Proposition~\ref{prop:geometric-renyi-from-classical-preps}, we
conclude the second inequality in \eqref{eq:sandwiched-Petz-geometric-ineqs}.
\end{proof}

\subsection{Belavkin--Staszewski relative entropy}

A different quantum generalization of the classical relative entropy is given
by the Belavkin--Staszewski\footnote{The name Staszewski is pronounced
Stah$\cdot$shev$\cdot$ski, with emphasis on the second syllable.} relative
entropy \cite{Belavkin1982}:

\begin{definition}
[Belavkin--Staszewski relative entropy]\label{def:belavkin-sta-rel-ent}The
Belavkin--Staszewski relative entropy of a quantum state $\rho$ and a positive
semi-definite operator $\sigma$ is defined as
\begin{equation}
\widehat{D}(\rho\Vert\sigma):=\left\{
\begin{array}
[c]{cc}
\operatorname{Tr}\!\left[  \rho\ln\!\left(  \rho^{\frac{1}{2}}\sigma^{-1}
\rho^{\frac{1}{2}}\right)  \right]  & \text{if }\operatorname{supp}
(\rho)\subseteq\operatorname{supp}(\sigma)\\
+\infty & \text{otherwise}
\end{array}
\right.  ,
\end{equation}
where the inverse $\sigma^{-1}$ is understood in the generalized sense and the
logarithm is evaluated on the support of $\rho$.
\end{definition}

This quantum generalization of classical relative entropy is not known to
possess an information-theoretic meaning. However, it is quite useful for obtaining
upper bounds on quantum channel capacities and quantum channel discrimination
rates \cite{Fang2019a}.

An important property of the Belavkin--Staszewski relative entropy is that it
is the limit of the geometric R\'{e}nyi relative entropy as $\alpha
\rightarrow1$ \cite{Mat13,Matsumoto2018}. The proposition below was known for
positive definite operators, but it is not clear to us whether it has been
established in the general case.

\begin{proposition}
\label{prop:BS-rel-ent-to-geometric}Let $\rho$ be a state and $\sigma$ a
positive semi-definite operator.\ Then, in the limit $\alpha\rightarrow1$, the
geometric R\'{e}nyi relative entropy converges to the Belavkin--Staszewski
relative entropy:
\begin{equation}
\lim_{\alpha\rightarrow1}\widehat{D}_{\alpha}(\rho\Vert\sigma)=\widehat
{D}(\rho\Vert\sigma).
\end{equation}

\end{proposition}

\begin{proof}
Suppose at first that $\operatorname{supp}(\rho)\subseteq\operatorname{supp}
(\sigma)$. Then $\widehat{D}_{\alpha}(\rho\Vert\sigma)$ is finite for all
$\alpha\in(0,1)\cup(1,\infty)$, and we can write the following explicit
formula for the geometric R\'{e}nyi relative entropy by employing
Proposition~\ref{prop:explicit-form-geometric-renyi}:
\begin{align}
\widehat{D}_{\alpha}(\rho\Vert\sigma)  &  =\frac{1}{\alpha-1} \ln\widehat
{Q}_{\alpha}(\rho\Vert\sigma)\\
&  =\frac{1}{\alpha-1} \ln\operatorname{Tr}\!\left[  \sigma\left(
\sigma^{-\frac{1}{2}}\rho\sigma^{-\frac{1}{2}}\right)  ^{\alpha}\right]  .
\end{align}
Our assumption implies that $\operatorname{Tr}[\Pi_{\sigma}\rho]=1$, and we
find that
\begin{align}
\widehat{Q}_{1}(\rho\Vert\sigma)  &  =\operatorname{Tr}\!\left[  \sigma\left(
\sigma^{-\frac{1}{2}}\rho\sigma^{-\frac{1}{2}}\right)  \right] \\
&  =\operatorname{Tr}[\Pi_{\sigma}\rho]\\
&  =1.
\end{align}
Since $\ln1=0$, we can write
\begin{equation}
\widehat{D}_{\alpha}(\rho\Vert\sigma)=\frac{ \ln\widehat{Q}_{\alpha}(\rho
\Vert\sigma)- \ln\widehat{Q}_{1}(\rho\Vert\sigma)}{\alpha-1},
\end{equation}
so that
\begin{align}
\lim_{\alpha\rightarrow1}\widehat{D}_{\alpha}(\rho\Vert\sigma)  &
=\lim_{\alpha\rightarrow1}\frac{ \ln\widehat{Q}_{\alpha}(\rho\Vert\sigma)-
\ln\widehat{Q}_{1}(\rho\Vert\sigma)}{\alpha-1}\\
&  =\left.  \frac{d}{d\alpha} \ln\widehat{Q}_{\alpha}(\rho\Vert\sigma
)\right\vert _{\alpha=1}\\
&  =\frac{\left.  \frac{d}{d\alpha}\widehat{Q}_{\alpha}
(\rho\Vert\sigma)\right\vert _{\alpha=1}}{\widehat{Q}_{1}(\rho\Vert\sigma)}\\
&  =\left.  \frac{d}{d\alpha}\widehat{Q}_{\alpha}(\rho
\Vert\sigma)\right\vert _{\alpha=1}.
\end{align}
Then
\begin{align*}
\left.  \frac{d}{d\alpha}\widehat{Q}_{\alpha}(\rho\Vert\sigma)\right\vert
_{\alpha=1}  &  =\left.  \frac{d}{d\alpha}\operatorname{Tr}\left[
\sigma\left(  \sigma^{-\frac{1}{2}}\rho\sigma^{-\frac{1}{2}}\right)  ^{\alpha
}\right]  \right\vert _{\alpha=1}\\
&  =\left.  \operatorname{Tr}\left[  \sigma\frac{d}{d\alpha}\left(
\sigma^{-\frac{1}{2}}\rho\sigma^{-\frac{1}{2}}\right)  ^{\alpha}\right]
\right\vert _{\alpha=1}.
\end{align*}
For a positive semi-definite operator $X$ with spectral decomposition
\begin{equation}
X=\sum_{z}\nu_{z}\Pi_{z},
\end{equation}
it follows that
\begin{align}
\left.  \frac{d}{d\alpha}X^{\alpha}\right\vert _{\alpha=1}  &  =\left.
\frac{d}{d\alpha}\sum_{z}\nu_{z}^{\alpha}\Pi_{z}\right\vert _{\alpha=1}\\
&  =\sum_{z}\left(  \left.  \frac{d}{d\alpha}\nu_{z}^{\alpha}\right\vert
_{\alpha=1}\right)  \Pi_{z}\\
&  =\sum_{z}\left(  \left.  \nu_{z}^{\alpha}\ln\nu_{z}^{\alpha}\right\vert
_{\alpha=1}\right)  \Pi_{z}\\
&  =\sum_{z}\left(  \nu_{z}\ln\nu_{z}\right)  \Pi_{z}\\
&  =X\ln_{\ast}X,
\end{align}
where
\begin{equation}
\ln_{\ast}(x):=\left\{
\begin{array}
[c]{cc}
\ln(x) & x>0\\
0 & x=0
\end{array}
\right.  . \label{eq:log-star-fnc-geometric-renyi}
\end{equation}
Thus we find that
\begin{align}
&  \left.  \operatorname{Tr}\left[  \sigma\frac{d}{d\alpha}\left(
\sigma^{-\frac{1}{2}}\rho\sigma^{-\frac{1}{2}}\right)  ^{\alpha}\right]
\right\vert _{\alpha=1}\nonumber\\
&  =\operatorname{Tr}\!\left[  \sigma\!\left(  \sigma^{-\frac{1}{2}}\rho
\sigma^{-\frac{1}{2}}\right)  \ln_{\ast}\!\left(  \sigma^{-\frac{1}{2}}
\rho\sigma^{-\frac{1}{2}}\right)  \right]
\label{eq:special-log-func-steps-BS-1}\\
&  =\operatorname{Tr}\!\left[  \sigma^{\frac{1}{2}}\rho^{\frac{1}{2}}
\rho^{\frac{1}{2}}\sigma^{-\frac{1}{2}}\ln_{\ast}\!\left(  \sigma^{-\frac
{1}{2}}\rho^{\frac{1}{2}}\rho^{\frac{1}{2}}\sigma^{-\frac{1}{2}}\right)
\right] \\
&  =\operatorname{Tr}\!\left[  \sigma^{\frac{1}{2}}\rho^{\frac{1}{2}}\ln
_{\ast}\!\left(  \rho^{\frac{1}{2}}\sigma^{-\frac{1}{2}}\sigma^{-\frac{1}{2}
}\rho^{\frac{1}{2}}\right)  \rho^{\frac{1}{2}}\sigma^{-\frac{1}{2}}\right] \\
&  =\operatorname{Tr}\!\left[  \rho^{\frac{1}{2}}\Pi_{\sigma}\rho^{\frac{1}
{2}}\ln_{\ast}\!\left(  \rho^{\frac{1}{2}}\sigma^{-\frac{1}{2}}\sigma
^{-\frac{1}{2}}\rho^{\frac{1}{2}}\right)  \right] \\
&  =\operatorname{Tr}\!\left[  \rho\ln\!\left(  \rho^{\frac{1}{2}}\sigma
^{-1}\rho^{\frac{1}{2}}\right)  \right]  .
\label{eq:special-log-func-steps-BS-last}
\end{align}
The third equality follows from Lemma~\ref{lem:sing-val-lemma-pseudo-commute}.
The final equality follows from the assumption $\operatorname{supp}
(\rho)\subseteq\operatorname{supp}(\sigma)$ and by applying the interpretation
of the logarithm exactly as stated in
Definition~\ref{def:belavkin-sta-rel-ent}. Then we find that
\begin{align}
\lim_{\alpha\rightarrow1}\widehat{D}_{\alpha}(\rho\Vert\sigma)  &  = \operatorname{Tr}\!\left[  \rho\ln\!\left(  \rho^{\frac{1}{2}}
\sigma^{-1}\rho^{\frac{1}{2}}\right)  \right] ,
\end{align}
for the case in which $\operatorname{supp}(\rho)\subseteq\operatorname{supp}
(\sigma)$.

Now suppose that $\alpha\in(1,\infty)$ and $\operatorname{supp}(\rho
)\not \subseteq \operatorname{supp}(\sigma)$. Then $\widehat{D}_{\alpha}
(\rho\Vert\sigma)=+\infty$, so that $\lim_{\alpha\rightarrow1^{+}}\widehat
{D}_{\alpha}(\rho\Vert\sigma)=+\infty$, consistent with the definition of the
Belavkin--Staszewski relative entropy in this case (see
Definition~\ref{def:belavkin-sta-rel-ent}).

Suppose that $\alpha\in(0,1)$ and $\operatorname{supp}(\rho)\not \subseteq
\operatorname{supp}(\sigma)$. Employing
Proposition~\ref{prop:geometric-to-sandwiched}, we have that $\widehat
{D}_{\alpha}(\rho\Vert\sigma)\geq\widetilde{D}_{\alpha}(\rho\Vert\sigma)$ for
all $\alpha\in(0,1)$. Since $\lim_{\alpha\rightarrow1^{-}}\widetilde
{D}_{\alpha}(\rho\Vert\sigma)=+\infty$ in this case \cite[Corollary~III.2]
{MO15}, it follows that $\lim_{\alpha\rightarrow1^{-}}\widehat{D}_{\alpha
}(\rho\Vert\sigma)=+\infty$.

Therefore,
\begin{align}
&  \lim_{\alpha\rightarrow1^{-}}\widehat{D}_{\alpha}(\rho\Vert\sigma
)\nonumber\\
&  =\left\{
\begin{array}
[c]{cc}
\operatorname{Tr}\!\left[  \rho\ln\!\left(  \rho^{\frac{1}{2}}\sigma^{-1}
\rho^{\frac{1}{2}}\right)  \right]  & \text{if }\operatorname{supp}
(\rho)\subseteq\operatorname{supp}(\sigma)\\
+\infty & \text{otherwise}
\end{array}
\right. \\
&  =\widehat{D}(\rho\Vert\sigma).
\end{align}
To conclude, we have established that $\lim_{\alpha\rightarrow1^{+}}
\widehat{D}_{\alpha}(\rho\Vert\sigma)=\lim_{\alpha\rightarrow1^{-}}\widehat
{D}_{\alpha}(\rho\Vert\sigma)=\widehat{D}(\rho\Vert\sigma)$, which means that
\begin{equation}
\lim_{\alpha\rightarrow1}\widehat{D}_{\alpha}(\rho\Vert\sigma)=\widehat
{D}(\rho\Vert\sigma),
\end{equation}
as required.
\end{proof}

\bigskip

The following inequality relates the quantum relative entropy to the
Belavkin--Staszewski relative entropy \cite{HP91}:

\begin{proposition}
\label{cor:BS-to-q-rel-ent}Let $\rho$ be a state and $\sigma$ a positive
semi-definite operator. Then the quantum relative entropy is never larger than
the Belavkin--Staszewski relative entropy:
\begin{equation}
D(\rho\Vert\sigma)\leq\widehat{D}(\rho\Vert\sigma).
\label{eq:BS-rel-ent-to-usual-rel-ent}
\end{equation}

\end{proposition}

\begin{proof}
If $\operatorname{supp}(\rho)\not \subseteq \operatorname{supp}(\sigma)$, then
there is nothing to prove in this case because both
\begin{equation}
D(\rho\Vert\sigma)=\widehat{D}(\rho\Vert\sigma)=+\infty,
\end{equation}
and so the inequality in \eqref{eq:BS-rel-ent-to-usual-rel-ent} holds
trivially in this case. So let us suppose instead that $\operatorname{supp}
(\rho)\subseteq\operatorname{supp}(\sigma)$. From
Propositions~\ref{prop:geometric-to-sandwiched} and
\ref{prop:explicit-form-geometric-renyi}, we conclude for all $\alpha
\in\left(  0,1\right)  \cup\left(  1,\infty\right)  $ that
\begin{equation}
\widetilde{D}_{\alpha}(\rho\Vert\sigma)\leq\widehat{D}_{\alpha}(\rho
\Vert\sigma). \label{eq:progress-BS-to-rel-ent-ineq}
\end{equation}
From \eqref{eq:alpha-1-limit-to-rel-ent}, we know that
\begin{equation}
\lim_{\alpha\rightarrow1}\widetilde{D}_{\alpha}(\rho\Vert\sigma)=D(\rho
\Vert\sigma).
\end{equation}
While from Proposition~\ref{prop:BS-rel-ent-to-geometric}, we know that
\begin{equation}
\lim_{\alpha\rightarrow1}\widehat{D}_{\alpha}(\rho\Vert\sigma)=\widehat
{D}(\rho\Vert\sigma).
\end{equation}
Thus, applying the limit $\alpha\rightarrow1$ to
\eqref{eq:progress-BS-to-rel-ent-ineq} and the two equalities above, we
conclude \eqref{eq:BS-rel-ent-to-usual-rel-ent}.
\end{proof}

\bigskip

Similar to \eqref{eq:rel-ent-limit-eps-0},
Definition~\ref{def:belavkin-sta-rel-ent}\ is consistent with the following limit:

\begin{proposition}
For any state $\rho$ and positive semi-definite operator $\sigma$, the
following limit holds
\begin{equation}
\widehat{D}(\rho\Vert\sigma)=\lim_{\varepsilon\rightarrow0^{+}}\lim
_{\delta\rightarrow0^{+}}\operatorname{Tr}\!\left[  \rho_{\delta}\log
_{2}\!\left(  \rho_{\delta}^{\frac{1}{2}}\sigma_{\varepsilon}^{-1}\rho
_{\delta}^{\frac{1}{2}}\right)  \right]  , \label{eq:limit-formula-BS-entropy}
\end{equation}
where $\delta\in\left(  0,1\right)  $ and
\begin{equation}
\rho_{\delta}:=\left(  1-\delta\right)  \rho+\delta\pi,\qquad\sigma
_{\varepsilon}:=\sigma+\varepsilon I,
\end{equation}
with $\pi$ the maximally mixed state.
\end{proposition}

\begin{proof}
Suppose first that $\operatorname{supp}(\rho)\subseteq\operatorname{supp}
(\sigma)$. We follow an approach similar to that given in the proof of
Proposition~\ref{prop:explicit-form-geometric-renyi}. Let us employ the
decomposition of the Hilbert space into $\operatorname{supp}(\sigma)\oplus
\ker(\sigma)$.\ Then we can write $\rho$ and $\sigma$ as in
\eqref{eq:rho-sig-decompose-for-geometric-formula}, so that
\begin{equation}
\sigma_{\varepsilon}^{-1}=
\begin{pmatrix}
\left(  \sigma+\varepsilon\Pi_{\sigma}\right)  ^{-1} & 0\\
0 & \varepsilon^{-1}\Pi_{\sigma}^{\perp}
\end{pmatrix}
,
\end{equation}
where we have followed the developments in
\eqref{eq:rho-sig-decompose-for-geometric-formula}--\eqref{eq:rho-sig-decompose-for-geometric-formula-3}.
The condition $\operatorname{supp}(\rho)\subseteq\operatorname{supp}(\sigma)$
implies that $\rho_{0,1}=0$ and $\rho_{1,1}=0$. It thus follows that
$\lim_{\delta\rightarrow0^{+}}\rho_{\delta}=\rho_{0,0}$. We then find that
\begin{align}
\operatorname{Tr}\!\left[  \rho_{\delta}\ln\!\left(  \rho_{\delta}^{\frac
{1}{2}}\sigma_{\varepsilon}^{-1}\rho_{\delta}^{\frac{1}{2}}\right)  \right]
&  =\operatorname{Tr}\!\left[  \rho_{\delta}^{\frac{1}{2}}\sigma_{\varepsilon
}^{\frac{1}{2}}\sigma_{\varepsilon}^{-\frac{1}{2}}\rho_{\delta}^{\frac{1}{2}
}\ln\!\left(  \rho_{\delta}^{\frac{1}{2}}\sigma_{\varepsilon}^{-\frac{1}{2}
}\sigma_{\varepsilon}^{-\frac{1}{2}}\rho_{\delta}^{\frac{1}{2}}\right)
\right] \\
&  =\operatorname{Tr}\!\left[  \rho_{\delta}^{\frac{1}{2}}\sigma_{\varepsilon
}^{\frac{1}{2}}\ln\!\left(  \sigma_{\varepsilon}^{-\frac{1}{2}}\rho_{\delta
}^{\frac{1}{2}}\rho_{\delta}^{\frac{1}{2}}\sigma_{\varepsilon}^{-\frac{1}{2}
}\right)  \sigma_{\varepsilon}^{-\frac{1}{2}}\rho_{\delta}^{\frac{1}{2}
}\right] \\
&  =\operatorname{Tr}\!\left[  \ln\!\left(  \sigma_{\varepsilon}^{-\frac{1}
{2}}\rho_{\delta}\sigma_{\varepsilon}^{-\frac{1}{2}}\right)  \left(
\sigma_{\varepsilon}^{-\frac{1}{2}}\rho_{\delta}\sigma_{\varepsilon}
^{-\frac{1}{2}}\right)  \sigma_{\varepsilon}\right] \\
&  =\operatorname{Tr}\!\left[  \sigma_{\varepsilon}\left(  \sigma
_{\varepsilon}^{-\frac{1}{2}}\rho_{\delta}\sigma_{\varepsilon}^{-\frac{1}{2}
}\right)  \ln\!\left(  \sigma_{\varepsilon}^{-\frac{1}{2}}\rho_{\delta}
\sigma_{\varepsilon}^{-\frac{1}{2}}\right)  \right] \\
&  =\operatorname{Tr}\!\left[  \sigma_{\varepsilon}\eta\!\left(
\sigma_{\varepsilon}^{-\frac{1}{2}}\rho_{\delta}\sigma_{\varepsilon}
^{-\frac{1}{2}}\right)  \right]  ,
\end{align}
where the second equality follows from applying
Lemma~\ref{lem:sing-val-lemma-pseudo-commute}\ with $f=\ln$ and $L=\rho
_{\delta}^{\frac{1}{2}}\sigma_{\varepsilon}^{-\frac{1}{2}}$. The
second-to-last equality follows because $\sigma_{\varepsilon}^{-\frac{1}{2}
}\rho_{\delta}\sigma_{\varepsilon}^{-\frac{1}{2}}$ commutes with $\ln
(\sigma_{\varepsilon}^{-\frac{1}{2}}\rho_{\delta}\sigma_{\varepsilon}
^{-\frac{1}{2}})$, and by employing cyclicity of trace. In the last line, we
made use of the following function:
\begin{equation}
\eta(x):=x\ln x,
\end{equation}
defined for all $x\in\lbrack0,\infty)$ with $\eta(0)=0$. By appealing to the
continuity of the function $\eta(x)$ on $x\in\lbrack0,\infty)$ and the fact
that $\lim_{\delta\rightarrow0^{+}}\rho_{\delta}=\rho_{0,0}$, we find that
\begin{equation}
\lim_{\delta\rightarrow0^{+}}\operatorname{Tr}\!\left[  \sigma_{\varepsilon
}\eta\!\left(  \sigma_{\varepsilon}^{-\frac{1}{2}}\rho_{\delta}\sigma
_{\varepsilon}^{-\frac{1}{2}}\right)  \right]  =\operatorname{Tr}\!\left[
\sigma_{\varepsilon}\eta\!\left(  \sigma_{\varepsilon}^{-\frac{1}{2}}
\rho_{0,0}\sigma_{\varepsilon}^{-\frac{1}{2}}\right)  \right]  .
\end{equation}
Now recall the function $\ln_{\ast}$ defined in
\eqref{eq:log-star-fnc-geometric-renyi}. Using it, we can write
\begin{align}
&  \operatorname{Tr}\!\left[  \sigma_{\varepsilon}\eta\!\left(  \sigma
_{\varepsilon}^{-\frac{1}{2}}\rho_{0,0}\sigma_{\varepsilon}^{-\frac{1}{2}
}\right)  \right] \nonumber\\
&  =\operatorname{Tr}\!\left[  \sigma_{\varepsilon}\sigma_{\varepsilon
}^{-\frac{1}{2}}\rho_{0,0}\sigma_{\varepsilon}^{-\frac{1}{2}}\ln_{\ast
}\!\left(  \sigma_{\varepsilon}^{-\frac{1}{2}}\rho_{0,0}\sigma_{\varepsilon
}^{-\frac{1}{2}}\right)  \right] \\
&  =\operatorname{Tr}\!\left[  \sigma_{\varepsilon}^{\frac{1}{2}}\rho
_{0,0}^{\frac{1}{2}}\rho_{0,0}^{\frac{1}{2}}\sigma_{\varepsilon}^{-\frac{1}
{2}}\ln_{\ast}\!\left(  \sigma_{\varepsilon}^{-\frac{1}{2}}\rho_{0,0}
^{\frac{1}{2}}\rho_{0,0}^{\frac{1}{2}}\sigma_{\varepsilon}^{-\frac{1}{2}
}\right)  \right] \\
&  =\operatorname{Tr}\!\left[  \sigma_{\varepsilon}^{\frac{1}{2}}\rho
_{0,0}^{\frac{1}{2}}\ln_{\ast}\!\left(  \rho_{0,0}^{\frac{1}{2}}
\sigma_{\varepsilon}^{-\frac{1}{2}}\sigma_{\varepsilon}^{-\frac{1}{2}}
\rho_{0,0}^{\frac{1}{2}}\right)  \rho_{0,0}^{\frac{1}{2}}\sigma_{\varepsilon
}^{-\frac{1}{2}}\right] \\
&  =\operatorname{Tr}\!\left[  \rho_{0,0}\ln_{\ast}\!\left(  \rho_{0,0}
^{\frac{1}{2}}\sigma_{\varepsilon}^{-1}\rho_{0,0}^{\frac{1}{2}}\right)
\right] \\
&  =\operatorname{Tr}\!\left[  \rho_{0,0}\ln_{\ast}\!\left(  \rho_{0,0}
^{\frac{1}{2}}\left(  \sigma+\varepsilon\Pi_{\sigma}\right)  ^{-1}\rho
_{0,0}^{\frac{1}{2}}\right)  \right]  ,
\end{align}
where the last line follows because
\begin{align}
&  \rho_{0,0}^{\frac{1}{2}}\left(  \sigma+\varepsilon\Pi_{\sigma}\right)
^{-1}\rho_{0,0}^{\frac{1}{2}}\nonumber\\
&  =
\begin{pmatrix}
\rho_{0,0}^{\frac{1}{2}} & 0\\
0 & 0
\end{pmatrix}
\begin{pmatrix}
\left(  \sigma+\varepsilon\Pi_{\sigma}\right)  ^{-1} & 0\\
0 & \varepsilon^{-1}\Pi_{\sigma}^{\perp}
\end{pmatrix}
\begin{pmatrix}
\rho_{0,0}^{\frac{1}{2}} & 0\\
0 & 0
\end{pmatrix}
\\
&  =
\begin{pmatrix}
\rho_{0,0}^{\frac{1}{2}}\left(  \sigma+\varepsilon\Pi_{\sigma}\right)
^{-1}\rho_{0,0}^{\frac{1}{2}} & 0\\
0 & 0
\end{pmatrix}
.
\end{align}
Now taking the limit as $\varepsilon\rightarrow0^{+}$, and appealing to
continuity of $\ln_{\ast}(x)$ and $x^{-1}$ for $x>0$, we find that
\begin{align}
&  \lim_{\varepsilon\rightarrow0^{+}}\operatorname{Tr}\!\left[  \rho_{0,0}
\ln_{\ast}\!\left(  \rho_{0,0}^{\frac{1}{2}}\left(  \sigma+\varepsilon
\Pi_{\sigma}\right)  ^{-1}\rho_{0,0}^{\frac{1}{2}}\right)  \right] \nonumber\\
&  =\operatorname{Tr}\!\left[  \rho_{0,0}\ln_{\ast}\!\left(  \rho_{0,0}
^{\frac{1}{2}}\sigma^{-1}\rho_{0,0}^{\frac{1}{2}}\right)  \right] \\
&  =\operatorname{Tr}\!\left[  \rho\ln\!\left(  \rho^{\frac{1}{2}}\sigma
^{-1}\rho^{\frac{1}{2}}\right)  \right]
\end{align}
where the formula in the last line is interpreted exactly as stated in
Definition~\ref{def:belavkin-sta-rel-ent}. Thus, we conclude that
\begin{equation}
\lim_{\varepsilon\rightarrow0^{+}}\lim_{\delta\rightarrow0^{+}}
\operatorname{Tr}\!\left[  \rho_{\delta}\ln\!\left(  \rho_{\delta}^{\frac
{1}{2}}\sigma_{\varepsilon}^{-1}\rho_{\delta}^{\frac{1}{2}}\right)  \right]
=\operatorname{Tr}\!\left[  \rho\ln\!\left(  \rho^{\frac{1}{2}}\sigma^{-1}
\rho^{\frac{1}{2}}\right)  \right]  .
\end{equation}

Now suppose that $\operatorname{supp}(\rho)\not \subseteq \operatorname{supp}
(\sigma)$. Then applying Proposition~\ref{cor:BS-to-q-rel-ent}, we find that
the following inequality holds for all $\delta\in(0,1)$ and $\varepsilon>0$:
\begin{equation}
\widehat{D}(\rho_{\delta}\Vert\sigma_{\varepsilon})\geq D(\rho_{\delta}
\Vert\sigma_{\varepsilon}).
\end{equation}
Now taking limits and applying \eqref{eq:rel-ent-limit-eps-0}, we find that
\begin{align}
\lim_{\varepsilon\rightarrow0^{+}}\lim_{\delta\rightarrow0^{+}}\widehat
{D}(\rho_{\delta}\Vert\sigma_{\varepsilon})  &  \geq\lim_{\varepsilon
\rightarrow0^{+}}\lim_{\delta\rightarrow0^{+}}D(\rho_{\delta}\Vert
\sigma_{\varepsilon})\\
&  =\lim_{\varepsilon\rightarrow0^{+}}D(\rho\Vert\sigma_{\varepsilon})\\
&  =+\infty.
\end{align}
This concludes the proof.
\end{proof}

\bigskip

By taking the limit $\alpha\rightarrow1$ in the statement of the
data-processing inequality for $\widehat{D}_{\alpha}$, and applying
Proposition~\ref{prop:BS-rel-ent-to-geometric}, we immediately obtain the
data-processing inequality for the Belavkin--Staszewski relative entropy. This
was shown by a different method in \cite{HP91}.

\begin{corollary}
[Data-Processing Inequality for Belavkin--Staszewski Relative Entropy]
\label{cor:DP-BS-rel-ent}Let $\rho$ be a state, $\sigma$ a positive
semi-definite operator, and $\mathcal{N}$ a quantum channel. Then
\begin{equation}
\widehat{D}(\rho\Vert\sigma)\geq\widehat{D}(\mathcal{N}(\rho)\Vert
\mathcal{N}(\sigma)).
\end{equation}

\end{corollary}

Some basic properties of the Belavkin--Staszewski relative entropy are as follows:

\begin{proposition}
[Basic Properties of Belavkin--Staszewski Relative Entropy]The
Belavkin--Staszewski relative entropy satisfies the following properties for
states $\rho,\rho_{1},\rho_{2}$ and positive semi-definite operators
$\sigma,\sigma_{1},\sigma_{2}$.

\begin{enumerate}
\item \textit{Isometric invariance}: For any isometry $V$,
\begin{equation}
\widehat{D}(V\rho V^{\dagger}\Vert V\sigma V^{\dagger})=\widehat{D}(\rho
\Vert\sigma).
\end{equation}

\item
\begin{enumerate}
\item If $\operatorname{Tr}[\sigma]\leq1$, then $\widehat{D}(\rho\Vert
\sigma)\geq0$.

\item \textit{Faithfulness}: Suppose that $\operatorname{Tr}[\sigma
]\leq\operatorname{Tr}[\rho]=1$. Then $\widehat{D}(\rho\Vert\sigma)=0$ if and
only if $\rho=\sigma$.

\item If $\rho\leq\sigma$, then $\widehat{D}(\rho\Vert\sigma)\leq0$.

\item If $\sigma\leq\sigma^{\prime}$, then $\widehat{D}(\rho\Vert\sigma
)\geq\widehat{D}(\rho\Vert\sigma^{\prime})$.
\end{enumerate}

\item \textit{Additivity}:
\begin{equation}
\widehat{D}(\rho_{1}\otimes\rho_{2}\Vert\sigma_{1}\otimes\sigma_{2}
)=\widehat{D}(\rho_{1}\Vert\sigma_{1})+D(\rho_{2}\Vert\sigma_{2}).
\label{eq-BS_rel_ent_additivity}
\end{equation}
As a special case, for any $\beta\in(0,\infty)$,
\begin{equation}
\widehat{D}(\rho\Vert\beta\sigma)=\widehat{D}(\rho\Vert\sigma)+\log
_{2}\!\left(  \frac{1}{\beta}\right)  . \label{eq-BS_rel_ent_scalar_mult}
\end{equation}

\item \textit{Direct-sum property}: Let $p:\mathcal{X}\rightarrow\lbrack0,1]$
be a probability distribution over a finite alphabet $\mathcal{X}$ with
associated $|\mathcal{X}|$-dimensional system $X$, and let $q:\mathcal{X}
\rightarrow\lbrack0,\infty)$ be a positive function on $\mathcal{X}$. Let
$\{\rho_{A}^{x}:x\in\mathcal{X}\}$ be a set of states on a system $A$, and let
$\{\sigma_{A}^{x}:x\in\mathcal{X}\}$ be a set of positive semi-definite
operators on $A$. Then,
\begin{equation}
\widehat{D}(\rho_{XA}\Vert\sigma_{XA})=\widehat{D}(p\Vert q)+\sum
_{x\in\mathcal{X}}p(x)\widehat{D}(\rho_{A}^{x}\Vert\sigma_{A}^{x}).
\label{eq-BS_rel_ent_direct_sum}
\end{equation}
where
\begin{align}
\rho_{XA}  &  :=\sum_{x\in\mathcal{X}}p(x)|x\rangle\!\langle x|_{X}\otimes
\rho_{A}^{x},\\
\sigma_{XA}  &  :=\sum_{x\in\mathcal{X}}q(x)|x\rangle\!\langle x|_{X}
\otimes\sigma_{A}^{x}.
\end{align}

\end{enumerate}
\end{proposition}

\begin{proof}
\begin{enumerate}
    \item Isometric invariance is a direct consequence of Propositions
\ref{prop:geometric-renyi-props}\ and \ref{prop:BS-rel-ent-to-geometric}.

\item All of the properties in the second item follow from data processing
(Corollary~\ref{cor:DP-BS-rel-ent}). Applying the trace-out channel, we find
that
\begin{align}
\widehat{D}(\rho\Vert\sigma)  &  \geq\widehat{D}(\operatorname{Tr}[\rho
]\Vert\operatorname{Tr}[\sigma])\\
&  =\operatorname{Tr}[\rho] \ln(\operatorname{Tr}[\rho]/\operatorname{Tr}
[\sigma])\\
&  =- \ln\operatorname{Tr}[\sigma]\\
&  \geq0.
\end{align}

If $\rho=\sigma$, then it follows by direct evalution that $\widehat{D}
(\rho\Vert\sigma)=0$. If $\widehat{D}(\rho\Vert\sigma)=0$ and
$\operatorname{Tr}[\sigma]\leq1$, then $D(\rho\Vert\sigma)=0$ by
Proposition~\ref{cor:BS-to-q-rel-ent}\ and we conclude that $\rho=\sigma$ from
faithfulness of the quantum relative entropy (see, e.g., \cite[Theorem~11.8.2]
{Wbook17}).

If $\rho\leq\sigma$, then $\sigma-\rho$ is positive semi-definite, and the
following operator is positive semi-definite:
\begin{equation}
\hat{\sigma}:=|0\rangle\!\langle0|\otimes\rho+|1\rangle\!\langle1|\otimes\left(
\sigma-\rho\right)  .
\end{equation}
Defining $\hat{\rho}:=|0\rangle\!\langle0|\otimes\rho$, we find from the
direct-sum property that
\begin{equation}
0=\widehat{D}(\rho\Vert\rho)=\widehat{D}(\hat{\rho}\Vert\hat{\sigma}
)\geq\widehat{D}(\rho\Vert\sigma),
\end{equation}
where the inequality follows from data processing by tracing out the first
classical register of $\hat{\rho}$ and $\hat{\sigma}$.

If $\sigma\leq\sigma^{\prime}$, then the operator $\sigma^{\prime}-\sigma$ is
positive semi-definite and so is the following one:
\begin{equation}
\hat{\sigma}:=|0\rangle\!\langle0|\otimes\sigma+|1\rangle\!\langle1|\otimes\left(
\sigma^{\prime}-\sigma\right)  .
\end{equation}
Defining $\hat{\rho}:=|0\rangle\!\langle0|\otimes\rho$, we find from the
direct-sum property that
\begin{equation}
\widehat{D}(\rho\Vert\sigma)=\widehat{D}(\hat{\rho}\Vert\hat{\sigma}
)\geq\widehat{D}(\rho\Vert\sigma^{\prime}),
\end{equation}
where the inequality follows from data processing by tracing out the first
classical register of $\hat{\rho}$ and $\hat{\sigma}$.

\item Additivity follows by direct evaluation.

\item The direct-sum property follows by direct evaluation.

\end{enumerate}
\end{proof}

\bigskip

A statement similar to that made by
Proposition~\ref{prop:geometric-renyi-from-classical-preps} holds for the
Belavkin--Staszewski relative entropy \cite{Mat13,Matsumoto2018}:

\begin{proposition}
[Belavkin--Staszewski Relative Entropy from Classical Preparations]Let $\rho$
be a state and $\sigma$ a positive semi-definite operator satisfying
$\operatorname{supp}(\rho)\subseteq\operatorname{supp}(\sigma)$. The
Belavkin--Staszewski relative entropy is equal to the smallest value that the
classical relative entropy can take by minimizing over classical--quantum
channels that realize the state $\rho$ and the positive semi-definite operator
$\sigma$. That is, the following equality holds
\begin{equation}
\widehat{D}(\rho\Vert\sigma)=\inf_{\left\{  p,q,\mathcal{P}\right\}  }\left\{
D(p\Vert q):\mathcal{P}(p)=\rho,\mathcal{P}(q)=\sigma\right\}  ,
\label{eq:BS-rel-ent-equality-classical-preps}
\end{equation}
where the classical relative entropy is defined as
\begin{equation}
D(p\Vert q):=\sum_{x}p(x)\ln\!\left(  \frac{p(x)}{q(x)}\right)  ,
\end{equation}
the channel $\mathcal{P}$ is a classical--quantum channel, $p:\mathcal{X}
\rightarrow\left[  0,1\right]  $ is a probability distribution over a finite
alphabet $\mathcal{X}$, and $q:\mathcal{X}\rightarrow(0,\infty)$ is a positive
function on $\mathcal{X}$.
\end{proposition}

\begin{proof}
The proof is very similar to the proof of
Proposition~\ref{prop:geometric-renyi-from-classical-preps}, and so we use the
same notation to provide a brief proof. By following the same reasoning that
leads to \eqref{eq:ineq-classical-prep-bigger-geo-ren}, it follows that
\begin{equation}
\inf_{\left\{  p,q,\mathcal{P}\right\}  }\left\{  D(p\Vert q):\mathcal{P}
(p)=\rho,\mathcal{P}(q)=\sigma\right\}  \geq\widehat{D}(\rho\Vert\sigma).
\label{eq:BS-rel-ent-lower-bound-classical-preps}
\end{equation}
The optimal choices of $p$, $q$, and $\mathcal{P}$ saturating the inequality
in \eqref{eq:BS-rel-ent-lower-bound-classical-preps} are again given by
\eqref{eq:optimal-choices-classical-preps-1}--\eqref{eq:optimal-choices-classical-preps-3}.
Consider for those choices that
\begin{align}
\sum_{x}p(x) \ln\!\left(  \frac{p(x)}{q(x)}\right)   &  =\sum_{x}p(x)
\ln\!\left(  \lambda_{x}\right) \\
&  =\sum_{x}\lambda_{x}q(x) \ln\!\left(  \lambda_{x}\right) \\
&  =\sum_{x}\lambda_{x}\operatorname{Tr}[\Pi_{x}\sigma] \ln\!\left(
\lambda_{x}\right) \\
&  =\operatorname{Tr}\!\left[  \sigma\left(  \sum_{x}\lambda_{x}\log
_{2}\!\left(  \lambda_{x}\right)  \Pi_{x}\right)  \right] \\
&  =\operatorname{Tr}\!\left[  \sigma\left(  \sigma^{-\frac{1}{2}}\rho
\sigma^{-\frac{1}{2}}\right)  \ln\!\left(  \sigma^{-\frac{1}{2}}\rho
\sigma^{-\frac{1}{2}}\right)  \right] \\
&  =\operatorname{Tr}\!\left[  \rho\ln\!\left(  \rho^{\frac{1}{2}}\sigma
^{-1}\rho^{\frac{1}{2}}\right)  \right]  ,
\end{align}
where the last equality follows from reasoning similar to that used to justify
\eqref{eq:special-log-func-steps-BS-1}--\eqref{eq:special-log-func-steps-BS-last}.
Then by following the reasoning at the end of the proof of
Proposition~\ref{prop:geometric-renyi-from-classical-preps}, we conclude \eqref{eq:BS-rel-ent-equality-classical-preps}.
\end{proof}

\subsection{Convergence of geometric R\'enyi relative entropy to max-relative
entropy}

\begin{proposition}
The geometric R\'enyi relative entropy converges to the max-relative entropy
in the limit as $\alpha\rightarrow\infty$:
\begin{equation}
\lim_{\alpha\rightarrow\infty}\widehat{D}_{\alpha}(\rho\Vert\sigma)=D_{\max
}(\rho\Vert\sigma). \label{eq:geometric-renyi-limit-to-max}
\end{equation}

\end{proposition}

\begin{proof}
We only consider the case in which $\operatorname{supp}(\rho)\subseteq
\operatorname{supp}(\sigma)$. Otherwise, we trivially have $\widehat
{D}_{\alpha}(\rho\Vert\sigma)=+\infty$ for all $\alpha>1$. In the case that
$\operatorname{supp}(\rho)\subseteq\operatorname{supp}(\sigma)$, we can
consider, without loss of generality, that $\operatorname{supp}(\sigma
)=\mathcal{H}$, which implies that $\lambda_{\min}(\sigma)>0$. Since we have
that
\begin{equation}
\lambda_{\min}(\sigma)I\leq\sigma\leq\lambda_{\max}(\sigma)I
\end{equation}
it follows that
\begin{align}
\lambda_{\min}(\sigma)\operatorname{Tr}\left[  \left(  \sigma^{-\frac{1}{2}
}\rho\sigma^{-\frac{1}{2}}\right)  ^{\alpha}\right]   &  \leq\operatorname{Tr}
\left[  \sigma\left(  \sigma^{-\frac{1}{2}}\rho\sigma^{-\frac{1}{2}}\right)
^{\alpha}\right] \\
&  \leq\lambda_{\max}(\sigma)\operatorname{Tr}\left[  \left(  \sigma
^{-\frac{1}{2}}\rho\sigma^{-\frac{1}{2}}\right)  ^{\alpha}\right]  .
\end{align}
Now taking a logarithm, dividing by $\alpha-1$, and applying definitions, we
find that the following inequalities hold for $\alpha>1$:
\begin{align}
&  \frac{1}{\alpha-1}\ln\lambda_{\min}(\sigma)+\frac{1}{\alpha-1}\log
_{2}\operatorname{Tr}\left[  \left(  \sigma^{-\frac{1}{2}}\rho\sigma
^{-\frac{1}{2}}\right)  ^{\alpha}\right] \nonumber\\
&  \leq\widehat{D}_{\alpha}(\rho\Vert\sigma
)\label{eq:geometric-renyi-to-max-1}\\
&  \leq\frac{1}{\alpha-1}\ln\lambda_{\max}(\sigma)+\frac{1}{\alpha-1}
\ln\operatorname{Tr}\left[  \left(  \sigma^{-\frac{1}{2}}\rho\sigma^{-\frac
{1}{2}}\right)  ^{\alpha}\right]  . \label{eq:geometric-renyi-to-max-2}
\end{align}
Rewriting
\begin{align}
\frac{1}{\alpha-1}\ln\operatorname{Tr}\left[  \left(  \sigma^{-\frac{1}{2}
}\rho\sigma^{-\frac{1}{2}}\right)  ^{\alpha}\right]   &  =\frac{\alpha}
{\alpha-1}\ln\left(  \operatorname{Tr}\left[  \left(  \sigma^{-\frac{1}{2}
}\rho\sigma^{-\frac{1}{2}}\right)  ^{\alpha}\right]  \right)  ^{\frac
{1}{\alpha}}\\
&  =\frac{\alpha}{\alpha-1}\ln\left\Vert \sigma^{-\frac{1}{2}}\rho
\sigma^{-\frac{1}{2}}\right\Vert _{\alpha}.
\end{align}
Then by applying $\lim_{\alpha\rightarrow\infty}\left\Vert X\right\Vert
_{\alpha}=\left\Vert X\right\Vert _{\infty}$, it follows that
\begin{equation}
\lim_{\alpha\rightarrow\infty}\frac{1}{\alpha-1}\ln\operatorname{Tr}\left[
\left(  \sigma^{-\frac{1}{2}}\rho\sigma^{-\frac{1}{2}}\right)  ^{\alpha
}\right]  =D_{\max}(\rho\Vert\sigma).
\end{equation}
Combining this limit with the inequalities in
\eqref{eq:geometric-renyi-to-max-1} and \eqref{eq:geometric-renyi-to-max-2},
we arrive at the equality in \eqref{eq:geometric-renyi-limit-to-max}.
\end{proof}

\section{Geometric R\'{e}nyi relative entropy of quantum channels}

\label{app:geo-renyi-channels-app}Here we prove the explicit form for the
geometric R\'{e}nyi relative entropy of quantum channels from
Proposition~\ref{prop:explicit-formula-geo-renyi-ch}, as well as the chain
rule from Proposition~\ref{prop:chain-rule-geo-renyi-ch}. We first begin by
recalling the transformer inequality from \cite{KA80} and \cite[Lemma~47]
{Fang2019a}.

\begin{lemma}
Let $X$ and $Y$ be positive semi-definite such that $\operatorname{supp}
(Y)\subseteq\operatorname{supp}(X)$, and let $L$ be a linear operator. Then
for $\alpha\in(1,2]$, the following inequality holds
\begin{equation}
G_{\alpha}(LXL^{\dag},LYL^{\dag})\leq LG_{\alpha}(X,Y)L^{\dag},
\label{eq:transform-ineq-alpha>1}
\end{equation}
where $G_{\alpha}$ is defined in \eqref{eq:weighted-op-geo-mean}. For
$\alpha\in(0,1)$, the following inequality holds
\begin{equation}
LG_{\alpha}(X,Y)L^{\dag}\leq G_{\alpha}(LXL^{\dag},LYL^{\dag}),
\label{eq:transform-ineq-alpha<1}
\end{equation}
In both of the above inequalities, the inverses $\left(  LXL^{\dag}\right)
^{-1}$ are taken on the support of $LXL^{\dag}$. If $L$ is invertible, then
the inequalities hold with equality.
\end{lemma}

\begin{proof}
For positive definite $X$ and $Y$ and $\alpha\in(1,2]$, we have that
\begin{align}
G_{\alpha}(X,Y)  &  =G_{1-\alpha}(Y,X),\\
G_{1-\alpha}(LYL^{\dag},LXL^{\dag})  &  \leq LG_{1-\alpha}(Y,X)L^{\dag},
\end{align}
where the equality follows from Lemma~\ref{lem:geo-mean-symmetry}\ and the
inequality from \cite[Lemma~47]{Fang2019a} (the special case of $\alpha=2$ was
established in \cite[Proposition~4.1]{Choi80}). Then by defining
$Y_{\varepsilon}=Y+\varepsilon I$ for $\varepsilon>0$, we conclude that
\begin{equation}
G_{\alpha}(LXL^{\dag},LY_{\varepsilon}L^{\dag})\leq LG_{\alpha}
(X,Y_{\varepsilon})L^{\dag}.
\end{equation}
By taking the limit $\varepsilon\rightarrow0^{+}$, we conclude
\eqref{eq:transform-ineq-alpha>1}, holding for $X$ and $Y$ positive
semi-definite such that $\operatorname{supp}(Y)\subseteq\operatorname{supp}
(X)$.

The inequality in \eqref{eq:transform-ineq-alpha<1} is known from
\cite{KA80}\ for positive definite $X$ and $Y$. Then we get
\eqref{eq:transform-ineq-alpha<1} by employing $Y_{\varepsilon}$ again and
taking the limit $\varepsilon\rightarrow0^{+}$.

For invertible $L$, the equalities follow by applying the inequality again, as
shown in \cite{KA80} and \cite[Lemma~47]{Fang2019a}. For $\alpha\in(1,2]$, we
have the following for invertible $L$:
\begin{align}
G_{\alpha}(LXL^{\dag},LYL^{\dag})  &  \leq LG_{\alpha}(X,Y)L^{\dag}\\
&  =LG_{\alpha}(L^{-1}LXL^{\dag}L^{-\dag},L^{-1}LYL^{\dag}L^{-\dag})L^{\dag}\\
&  \leq LL^{-1}G_{\alpha}(LXL^{\dag},LYL^{\dag})L^{-\dag}L^{\dag}\\
&  =G_{\alpha}(LXL^{\dag},LYL^{\dag}).
\end{align}
The same argument applies for $\alpha\in(0,1)$, but the inequalities flip.
\end{proof}

\bigskip

\begin{proof}
[Proof of Proposition~\ref{prop:explicit-formula-geo-renyi-ch}]First, suppose
that $\alpha\in(1,2]$ and $\operatorname{supp}(\Gamma_{RB}^{\mathcal{N}
})\not \subseteq \operatorname{supp}(\Gamma_{RB}^{\mathcal{M}})$. Then we can
take the maximally entangled state $\Phi_{RA}$ (normalized version of
$\Gamma_{RA}$) as input, and it follows that $\widehat{D}_{\alpha}
(\mathcal{N}\Vert\mathcal{M})=+\infty$.

So let us suppose that $\alpha\in(1,2]$ and $\operatorname{supp}(\Gamma
_{RB}^{\mathcal{N}})\subseteq\operatorname{supp}(\Gamma_{RB}^{\mathcal{M}})$.
Let $\psi_{RA}$ be an arbitrary pure bipartite input state. We can write such
a state as follows:
\begin{equation}
\psi_{RA}=Z_{R}\Gamma_{RA}Z_{R}^{\dag},
\end{equation}
where $Z_{R}$ is an operator satisfying $\operatorname{Tr}[Z_{R}^{\dag}
Z_{R}]=1$. Then it follows that
\begin{equation}
\mathcal{N}_{A\rightarrow B}(\psi_{RA})=Z_{R}\Gamma_{RB}^{\mathcal{N}}
Z_{R}^{\dag}.
\end{equation}
Due to the fact that the set of states with $Z_{R}$ invertible is dense in the
set of all pure bipartite states, it suffices to optimize with respect to this
set:
\begin{align}
&  \widehat{D}_{\alpha}(\mathcal{N}\Vert\mathcal{M})\\
&  =\sup_{\psi_{RA}}\widehat{D}_{\alpha}(\mathcal{N}_{A\rightarrow B}
(\psi_{RA})\Vert\mathcal{M}_{A\rightarrow B}(\psi_{RA}))\\
&  =\sup_{\substack{Z_{R}:\left\vert Z_{R}\right\vert >0,\\\operatorname{Tr}
[Z_{R}^{\dag}Z_{R}]=1}}\widehat{D}_{\alpha}(Z_{R}\Gamma_{RB}^{\mathcal{N}
}Z_{R}^{\dag}\Vert Z_{R}\Gamma_{RB}^{\mathcal{M}}Z_{R}^{\dag})\\
&  =\sup_{\substack{Z_{R}:\left\vert Z_{R}\right\vert >0,\\\operatorname{Tr}
[Z_{R}^{\dag}Z_{R}]=1}}\frac{1}{\alpha-1}\ln\operatorname{Tr}[G_{\alpha}
(Z_{R}\Gamma_{RB}^{\mathcal{M}}Z_{R}^{\dag},Z_{R}\Gamma_{RB}^{\mathcal{N}
}Z_{R}^{\dag})]\\
&  =\sup_{\substack{Z_{R}:\left\vert Z_{R}\right\vert >0,\\\operatorname{Tr}
[Z_{R}^{\dag}Z_{R}]=1}}\frac{1}{\alpha-1}\ln\operatorname{Tr}[Z_{R}G_{\alpha
}(\Gamma_{RB}^{\mathcal{M}},\Gamma_{RB}^{\mathcal{N}})Z_{R}^{\dag}]\\
&  =\sup_{\substack{Z_{R}:\left\vert Z_{R}\right\vert >0,\\\operatorname{Tr}
[Z_{R}^{\dag}Z_{R}]=1}}\frac{1}{\alpha-1}\ln\operatorname{Tr}[Z_{R}^{\dag
}Z_{R}G_{\alpha}(\Gamma_{RB}^{\mathcal{M}},\Gamma_{RB}^{\mathcal{N}})]\\
&  =\frac{1}{\alpha-1}\ln\sup_{\substack{Z_{R}:\left\vert Z_{R}\right\vert
>0,\\\operatorname{Tr}[Z_{R}^{\dag}Z_{R}]=1}}\operatorname{Tr}[Z_{R}^{\dag
}Z_{R}G_{\alpha}(\Gamma_{RB}^{\mathcal{M}},\Gamma_{RB}^{\mathcal{N}})]\\
&  =\frac{1}{\alpha-1}\ln\left\Vert \operatorname{Tr}_{B}\left[  [\Gamma
_{RB}^{\mathcal{M}}]^{1/2}(\left[  \Gamma_{RB}^{\mathcal{M}}\right]
^{-1/2}\Gamma_{RB}^{\mathcal{N}}\left[  \Gamma_{RB}^{\mathcal{M}}\right]
^{-1/2})^{\alpha}[\Gamma_{RB}^{\mathcal{M}}]^{1/2}\right]  \right\Vert
_{\infty}.
\end{align}
The critical equality is the fourth one, which follows from the transformer
equality of \cite[Lemma~47]{Fang2019a}.

Now suppose that $\alpha\in(0,1)$. Then proceeding by similar reasoning, but
taking care with various limits and the sign flip due to the prefactor
$\frac{1}{\alpha-1}$, we find the following:
\begin{align}
&  \widehat{D}_{\alpha}(\mathcal{N}\Vert\mathcal{M})\nonumber\\
&  =\sup_{\psi_{RA}}\widehat{D}_{\alpha}(\mathcal{N}_{A\rightarrow B}
(\psi_{RA})\Vert\mathcal{M}_{A\rightarrow B}(\psi_{RA}))\\
&  =\sup_{Z_{R}:\left\vert Z_{R}\right\vert >0}\widehat{D}_{\alpha}
(Z_{R}\Gamma_{RB}^{\mathcal{N}}Z_{R}^{\dag}\Vert Z_{R}\Gamma_{RB}
^{\mathcal{M}}Z_{R}^{\dag})\\
&  =\sup_{Z_{R}:\left\vert Z_{R}\right\vert >0}\lim_{\varepsilon
\rightarrow0^{+}}\widehat{D}_{\alpha}(Z_{R}\Gamma_{RB}^{\mathcal{N}}
Z_{R}^{\dag}\Vert Z_{R}\Gamma_{RB}^{\mathcal{M}_{\varepsilon}}Z_{R}^{\dag})\\
&  =\sup_{Z_{R}:\left\vert Z_{R}\right\vert >0}\lim_{\varepsilon
\rightarrow0^{+}}\frac{1}{\alpha-1}\ln\operatorname{Tr}[G_{\alpha}(Z_{R}
\Gamma_{RB}^{\mathcal{M}_{\varepsilon}}Z_{R}^{\dag},Z_{R}\Gamma_{RB}
^{\mathcal{N}}Z_{R}^{\dag})]\\
&  =\frac{1}{\alpha-1}\ln\inf_{Z_{R}:\left\vert Z_{R}\right\vert >0}
\lim_{\varepsilon\rightarrow0^{+}}\operatorname{Tr}[G_{\alpha}(Z_{R}
\Gamma_{RB}^{\mathcal{M}_{\varepsilon}}Z_{R}^{\dag},Z_{R}\Gamma_{RB}
^{\mathcal{N}}Z_{R}^{\dag})]\\
&  =\frac{1}{\alpha-1}\ln\inf_{Z_{R}:\left\vert Z_{R}\right\vert >0}
\inf_{\varepsilon>0}\operatorname{Tr}[G_{\alpha}(Z_{R}\Gamma_{RB}
^{\mathcal{M}_{\varepsilon}}Z_{R}^{\dag},Z_{R}\Gamma_{RB}^{\mathcal{N}}
Z_{R}^{\dag})].
\end{align}
The last equality follows from reasoning similar to that in
Lemma~\ref{lem:limit-exchange-geom-renyi-a-0-to1}, that the limit
$\varepsilon\rightarrow0^{+}$ is the same as the infimum over $\varepsilon>0$.
Continuing, we find that
\begin{align}
&  \widehat{D}_{\alpha}(\mathcal{N}\Vert\mathcal{M})\nonumber\\
&  =\frac{1}{\alpha-1}\ln\inf_{\varepsilon>0}\inf_{Z_{R}:\left\vert
Z_{R}\right\vert >0}\operatorname{Tr}[G_{\alpha}(Z_{R}\Gamma_{RB}
^{\mathcal{M}_{\varepsilon}}Z_{R}^{\dag},Z_{R}\Gamma_{RB}^{\mathcal{N}}
Z_{R}^{\dag})]\\
&  =\frac{1}{\alpha-1}\ln\inf_{\varepsilon>0}\inf_{Z_{R}:\left\vert
Z_{R}\right\vert >0}\operatorname{Tr}[Z_{R}G_{\alpha}(\Gamma_{RB}
^{\mathcal{M}},\Gamma_{RB}^{\mathcal{N}})Z_{R}^{\dag}]\\
&  =\frac{1}{\alpha-1}\ln\inf_{\varepsilon>0}\lambda_{\min}\left(
\operatorname{Tr}_{B}\left(  G_{\alpha}(\Gamma_{RB}^{\mathcal{M}},\Gamma
_{RB}^{\mathcal{N}})\right)  \right) \\
&  =\frac{1}{\alpha-1}\ln\lim_{\varepsilon\rightarrow0^{+}}\lambda_{\min
}\left(  \operatorname{Tr}_{B}\left(  G_{\alpha}(\Gamma_{RB}^{\mathcal{M}
},\Gamma_{RB}^{\mathcal{N}})\right)  \right) \\
&  =\lim_{\varepsilon\rightarrow0^{+}}\frac{1}{\alpha-1}\ln\lambda_{\min
}\left(  \operatorname{Tr}_{B}\left(  [\Gamma_{RB}^{\mathcal{M}_{\varepsilon}
}]^{1/2}(\left[  \Gamma_{RB}^{\mathcal{M}_{\varepsilon}}\right]  ^{-1/2}
\Gamma_{RB}^{\mathcal{N}}\left[  \Gamma_{RB}^{\mathcal{M}_{\varepsilon}
}\right]  ^{-1/2})^{\alpha}[\Gamma_{RB}^{\mathcal{M}_{\varepsilon}}
]^{1/2}\right)  \right)  .
\end{align}

Now we establish the formula in \eqref{eq:geo-ren-ch-exp-form-all-cases} for
$\alpha\in(0,1)$. If $\operatorname{supp}(\Gamma_{RB}^{\mathcal{N}}
)\subseteq\operatorname{supp}(\Gamma_{RB}^{\mathcal{M}})$, then taking the
limit $\varepsilon\rightarrow0^{+}$ leads to the formula
\begin{equation}
\widehat{D}_{\alpha}(\mathcal{N}\Vert\mathcal{M})=\frac{1}{\alpha-1}\ln
\lambda_{\min}\left(  \operatorname{Tr}_{B}\left(  [\Gamma_{RB}^{\mathcal{M}
}]^{1/2}(\left[  \Gamma_{RB}^{\mathcal{M}}\right]  ^{-1/2}\Gamma
_{RB}^{\mathcal{N}}\left[  \Gamma_{RB}^{\mathcal{M}}\right]  ^{-1/2})^{\alpha
}[\Gamma_{RB}^{\mathcal{M}}]^{1/2}\right)  \right)  .
\end{equation}

If $\operatorname{supp}(\Gamma_{RB}^{\mathcal{N}})\not \subseteq
\operatorname{supp}(\Gamma_{RB}^{\mathcal{M}})$, then the proof is similar to
the proof of \eqref{eq:geo-ren-exp-form-weird-case}, but more involved. We
need to evaluate the following limit for $\alpha\in(0,1)$:
\begin{equation}
\lim_{\varepsilon\rightarrow0^{+}}\lambda_{\min}\left(  \operatorname{Tr}
_{B}\left[  [\Gamma_{RB}^{\mathcal{M}_{\varepsilon}}]^{1/2}(\left[
\Gamma_{RB}^{\mathcal{M}_{\varepsilon}}\right]  ^{-1/2}\Gamma_{RB}
^{\mathcal{N}}\left[  \Gamma_{RB}^{\mathcal{M}_{\varepsilon}}\right]
^{-1/2})^{\alpha}[\Gamma_{RB}^{\mathcal{M}_{\varepsilon}}]^{1/2}\right]
\right)  .
\end{equation}
For $\varepsilon>0$ and $\delta\in(0,1)$, let us write
\begin{align}
\Gamma_{RB}^{\mathcal{M}_{\varepsilon}}  &  =
\begin{bmatrix}
\hat{\Gamma}_{RB}^{\mathcal{M}_{\varepsilon}} & 0\\
0 & \varepsilon\Pi_{\Gamma^{\mathcal{M}}}^{\perp}
\end{bmatrix}
,\\
\hat{\Gamma}_{RB}^{\mathcal{M}_{\varepsilon}}  &  :=\Gamma_{RB}^{\mathcal{M}
}+\varepsilon\Pi_{\Gamma^{\mathcal{M}}},\\
\Gamma_{RB}^{\mathcal{N}_{\delta}}  &  :=\left(  1-\delta\right)  \Gamma
_{RB}^{\mathcal{N}_{\delta}}+\delta I_{R}\otimes\pi_{B},\\
\Gamma_{RB}^{\mathcal{N}_{\delta}}  &  =
\begin{bmatrix}
(\Gamma_{RB}^{\mathcal{N}_{\delta}})_{0,0} & (\Gamma_{RB}^{\mathcal{N}
_{\delta}})_{0,1}\\
(\Gamma_{RB}^{\mathcal{N}_{\delta}})_{0,1}^{\dag} & (\Gamma_{RB}
^{\mathcal{N}_{\delta}})_{1,1}
\end{bmatrix}
,\\
(\Gamma_{RB}^{\mathcal{N}_{\delta}})_{0,0}  &  :=\Pi_{\Gamma^{\mathcal{M}}
}\Gamma_{RB}^{\mathcal{N}_{\delta}}\Pi_{\Gamma^{\mathcal{M}}},\\
(\Gamma_{RB}^{\mathcal{N}_{\delta}})_{0,1}  &  :=\Pi_{\Gamma^{\mathcal{M}}
}\Gamma_{RB}^{\mathcal{N}_{\delta}}\Pi_{\Gamma^{\mathcal{M}}}^{\perp},\\
(\Gamma_{RB}^{\mathcal{N}_{\delta}})_{1,1}  &  :=\Pi_{\Gamma^{\mathcal{M}}
}^{\perp}\Gamma_{RB}^{\mathcal{N}_{\delta}}\Pi_{\Gamma^{\mathcal{M}}}^{\perp}.
\end{align}
Consider that
\begin{multline}
\lim_{\varepsilon\rightarrow0^{+}}\lambda_{\min}\left(  \operatorname{Tr}
_{B}\left[  [\Gamma_{RB}^{\mathcal{M}_{\varepsilon}}]^{1/2}(\left[
\Gamma_{RB}^{\mathcal{M}_{\varepsilon}}\right]  ^{-1/2}\Gamma_{RB}
^{\mathcal{N}}\left[  \Gamma_{RB}^{\mathcal{M}_{\varepsilon}}\right]
^{-1/2})^{\alpha}[\Gamma_{RB}^{\mathcal{M}_{\varepsilon}}]^{1/2}\right]
\right) \label{eq:lim-exch-geo-ren-ch-exp-form-1}\\
=\lim_{\varepsilon\rightarrow0^{+}}\lim_{\delta\rightarrow0^{+}}\lambda_{\min
}\left(  \operatorname{Tr}_{B}\left[  [\Gamma_{RB}^{\mathcal{M}_{\varepsilon}
}]^{1/2}(\left[  \Gamma_{RB}^{\mathcal{M}_{\varepsilon}}\right]  ^{-1/2}
\Gamma_{RB}^{\mathcal{N}_{\delta}}\left[  \Gamma_{RB}^{\mathcal{M}
_{\varepsilon}}\right]  ^{-1/2})^{\alpha}[\Gamma_{RB}^{\mathcal{M}
_{\varepsilon}}]^{1/2}\right]  \right)  .
\end{multline}
We note by similar reasoning given to establish
Lemma~\ref{lem:limit-exchange-geom-renyi-a-0-to1}, it follows for $\alpha
\in(0,1)$ that
\begin{multline}
\lim_{\varepsilon\rightarrow0^{+}}\lim_{\delta\rightarrow0^{+}}\lambda_{\min
}\left(  \operatorname{Tr}_{B}\left[  [\Gamma_{RB}^{\mathcal{M}_{\varepsilon}
}]^{1/2}(\left[  \Gamma_{RB}^{\mathcal{M}_{\varepsilon}}\right]  ^{-1/2}
\Gamma_{RB}^{\mathcal{N}_{\delta}}\left[  \Gamma_{RB}^{\mathcal{M}
_{\varepsilon}}\right]  ^{-1/2})^{\alpha}[\Gamma_{RB}^{\mathcal{M}
_{\varepsilon}}]^{1/2}\right]  \right)
\label{eq:lim-exch-geo-ren-ch-exp-form-2}\\
=\lim_{\delta\rightarrow0^{+}}\lim_{\varepsilon\rightarrow0^{+}}\lambda_{\min
}\left(  \operatorname{Tr}_{B}\left[  [\Gamma_{RB}^{\mathcal{M}_{\varepsilon}
}]^{1/2}(\left[  \Gamma_{RB}^{\mathcal{M}_{\varepsilon}}\right]  ^{-1/2}
\Gamma_{RB}^{\mathcal{N}_{\delta}}\left[  \Gamma_{RB}^{\mathcal{M}
_{\varepsilon}}\right]  ^{-1/2})^{\alpha}[\Gamma_{RB}^{\mathcal{M}
_{\varepsilon}}]^{1/2}\right]  \right)  .
\end{multline}
Then
\begin{align}
&  \left[  \Gamma_{RB}^{\mathcal{M}_{\varepsilon}}\right]  ^{-1/2}\Gamma
_{RB}^{\mathcal{N}_{\delta}}\left[  \Gamma_{RB}^{\mathcal{M}_{\varepsilon}
}\right]  ^{-1/2}\nonumber\\
&  =
\begin{bmatrix}
\hat{\Gamma}_{RB}^{\mathcal{M}_{\varepsilon}} & 0\\
0 & \varepsilon\Pi_{\Gamma^{\mathcal{M}}}^{\perp}
\end{bmatrix}
^{-\frac{1}{2}}
\begin{bmatrix}
(\Gamma_{RB}^{\mathcal{N}_{\delta}})_{0,0} & (\Gamma_{RB}^{\mathcal{N}
_{\delta}})_{0,1}\\
(\Gamma_{RB}^{\mathcal{N}_{\delta}})_{1,0} & (\Gamma_{RB}^{\mathcal{N}
_{\delta}})_{1,1}
\end{bmatrix}
\begin{bmatrix}
\hat{\Gamma}_{RB}^{\mathcal{M}_{\varepsilon}} & 0\\
0 & \varepsilon\Pi_{\Gamma^{\mathcal{M}}}^{\perp}
\end{bmatrix}
^{-\frac{1}{2}}\\
&  =
\begin{bmatrix}
(\hat{\Gamma}_{RB}^{\mathcal{M}_{\varepsilon}})^{-\frac{1}{2}} & 0\\
0 & \varepsilon^{-\frac{1}{2}}\Pi_{\Gamma^{\mathcal{M}}}^{\perp}
\end{bmatrix}
\begin{bmatrix}
(\Gamma_{RB}^{\mathcal{N}_{\delta}})_{0,0} & (\Gamma_{RB}^{\mathcal{N}
_{\delta}})_{0,1}\\
(\Gamma_{RB}^{\mathcal{N}_{\delta}})_{1,0} & (\Gamma_{RB}^{\mathcal{N}
_{\delta}})_{1,1}
\end{bmatrix}
\begin{bmatrix}
(\hat{\Gamma}_{RB}^{\mathcal{M}_{\varepsilon}})^{-\frac{1}{2}} & 0\\
0 & \varepsilon^{-\frac{1}{2}}\Pi_{\Gamma^{\mathcal{M}}}^{\perp}
\end{bmatrix}
\\
&  =
\begin{bmatrix}
(\hat{\Gamma}_{RB}^{\mathcal{M}_{\varepsilon}})^{-\frac{1}{2}}(\Gamma
_{RB}^{\mathcal{N}_{\delta}})_{0,0}(\hat{\Gamma}_{RB}^{\mathcal{M}
_{\varepsilon}})^{-\frac{1}{2}} & \varepsilon^{-\frac{1}{2}}(\hat{\Gamma}
_{RB}^{\mathcal{M}_{\varepsilon}})^{-\frac{1}{2}}(\Gamma_{RB}^{\mathcal{N}
_{\delta}})_{0,1}\Pi_{\Gamma^{\mathcal{M}}}^{\perp}\\
\varepsilon^{-\frac{1}{2}}\Pi_{\Gamma^{\mathcal{M}}}^{\perp}(\Gamma
_{RB}^{\mathcal{N}_{\delta}})_{1,0}(\hat{\Gamma}_{RB}^{\mathcal{M}
_{\varepsilon}})^{-\frac{1}{2}} & \varepsilon^{-1}\Pi_{\Gamma^{\mathcal{M}}
}^{\perp}(\Gamma_{RB}^{\mathcal{N}_{\delta}})_{1,1}\Pi_{\Gamma^{\mathcal{M}}
}^{\perp}
\end{bmatrix}
\\
&  =
\begin{bmatrix}
(\hat{\Gamma}_{RB}^{\mathcal{M}_{\varepsilon}})^{-\frac{1}{2}}(\Gamma
_{RB}^{\mathcal{N}_{\delta}})_{0,0}(\hat{\Gamma}_{RB}^{\mathcal{M}
_{\varepsilon}})^{-\frac{1}{2}} & \varepsilon^{-\frac{1}{2}}(\hat{\Gamma}
_{RB}^{\mathcal{M}_{\varepsilon}})^{-\frac{1}{2}}(\Gamma_{RB}^{\mathcal{N}
_{\delta}})_{0,1}\\
\varepsilon^{-\frac{1}{2}}(\Gamma_{RB}^{\mathcal{N}_{\delta}})_{1,0}
(\hat{\Gamma}_{RB}^{\mathcal{M}_{\varepsilon}})^{-\frac{1}{2}} &
\varepsilon^{-1}(\Gamma_{RB}^{\mathcal{N}_{\delta}})_{1,1}
\end{bmatrix}
,
\end{align}
so that
\begin{align}
&  \lbrack\Gamma_{RB}^{\mathcal{M}_{\varepsilon}}]^{1/2}(\left[  \Gamma
_{RB}^{\mathcal{M}_{\varepsilon}}\right]  ^{-1/2}\Gamma_{RB}^{\mathcal{N}
_{\delta}}\left[  \Gamma_{RB}^{\mathcal{M}_{\varepsilon}}\right]
^{-1/2})^{\alpha}[\Gamma_{RB}^{\mathcal{M}_{\varepsilon}}]^{1/2}\nonumber\\
&  =
\begin{bmatrix}
\hat{\Gamma}_{RB}^{\mathcal{M}_{\varepsilon}} & 0\\
0 & \varepsilon\Pi_{\Gamma^{\mathcal{M}}}^{\perp}
\end{bmatrix}
^{\frac{1}{2}}
\begin{bmatrix}
(\hat{\Gamma}_{RB}^{\mathcal{M}_{\varepsilon}})^{-\frac{1}{2}}(\Gamma
_{RB}^{\mathcal{N}_{\delta}})_{0,0}(\hat{\Gamma}_{RB}^{\mathcal{M}
_{\varepsilon}})^{-\frac{1}{2}} & \varepsilon^{-\frac{1}{2}}(\hat{\Gamma}
_{RB}^{\mathcal{M}_{\varepsilon}})^{-\frac{1}{2}}(\Gamma_{RB}^{\mathcal{N}
_{\delta}})_{0,1}\\
\varepsilon^{-\frac{1}{2}}(\Gamma_{RB}^{\mathcal{N}_{\delta}})_{0,1}^{\dag
}(\hat{\Gamma}_{RB}^{\mathcal{M}_{\varepsilon}})^{-\frac{1}{2}} &
\varepsilon^{-1}(\Gamma_{RB}^{\mathcal{N}_{\delta}})_{1,1}
\end{bmatrix}
^{\alpha}\nonumber\\
&  \qquad\times
\begin{bmatrix}
\hat{\Gamma}_{RB}^{\mathcal{M}_{\varepsilon}} & 0\\
0 & \varepsilon\Pi_{\Gamma^{\mathcal{M}}}^{\perp}
\end{bmatrix}
^{\frac{1}{2}}\\
&  =
\begin{bmatrix}
(\hat{\Gamma}_{RB}^{\mathcal{M}_{\varepsilon}})^{\frac{1}{2}} & 0\\
0 & \varepsilon^{\frac{1}{2}}\Pi_{\Gamma^{\mathcal{M}}}^{\perp}
\end{bmatrix}
\left(  \varepsilon^{-1}
\begin{bmatrix}
\varepsilon(\hat{\Gamma}_{RB}^{\mathcal{M}_{\varepsilon}})^{-\frac{1}{2}
}(\Gamma_{RB}^{\mathcal{N}_{\delta}})_{0,0}(\hat{\Gamma}_{RB}^{\mathcal{M}
_{\varepsilon}})^{-\frac{1}{2}} & \varepsilon^{\frac{1}{2}}(\hat{\Gamma}
_{RB}^{\mathcal{M}_{\varepsilon}})^{-\frac{1}{2}}(\Gamma_{RB}^{\mathcal{N}
_{\delta}})_{0,1}\\
\varepsilon^{\frac{1}{2}}(\Gamma_{RB}^{\mathcal{N}_{\delta}})_{0,1}^{\dag
}(\hat{\Gamma}_{RB}^{\mathcal{M}_{\varepsilon}})^{-\frac{1}{2}} & (\Gamma
_{RB}^{\mathcal{N}_{\delta}})_{1,1}
\end{bmatrix}
\right)  ^{\alpha}\nonumber\\
&  \qquad\times
\begin{bmatrix}
(\hat{\Gamma}_{RB}^{\mathcal{M}_{\varepsilon}})^{\frac{1}{2}} & 0\\
0 & \varepsilon^{\frac{1}{2}}\Pi_{\Gamma^{\mathcal{M}}}^{\perp}
\end{bmatrix}
\\
&  =\varepsilon^{-\frac{\alpha}{2}}
\begin{bmatrix}
(\hat{\Gamma}_{RB}^{\mathcal{M}_{\varepsilon}})^{\frac{1}{2}} & 0\\
0 & \varepsilon^{\frac{1}{2}}\Pi_{\Gamma^{\mathcal{M}}}^{\perp}
\end{bmatrix}
\begin{bmatrix}
\varepsilon(\hat{\Gamma}_{RB}^{\mathcal{M}_{\varepsilon}})^{-\frac{1}{2}
}(\Gamma_{RB}^{\mathcal{N}_{\delta}})_{0,0}(\hat{\Gamma}_{RB}^{\mathcal{M}
_{\varepsilon}})^{-\frac{1}{2}} & \varepsilon^{\frac{1}{2}}(\hat{\Gamma}
_{RB}^{\mathcal{M}_{\varepsilon}})^{-\frac{1}{2}}(\Gamma_{RB}^{\mathcal{N}
_{\delta}})_{0,1}\\
\varepsilon^{\frac{1}{2}}(\Gamma_{RB}^{\mathcal{N}_{\delta}})_{0,1}^{\dag
}(\hat{\Gamma}_{RB}^{\mathcal{M}_{\varepsilon}})^{-\frac{1}{2}} & (\Gamma
_{RB}^{\mathcal{N}_{\delta}})_{1,1}
\end{bmatrix}
^{\alpha}\nonumber\\
&  \qquad\times\varepsilon^{-\frac{\alpha}{2}}
\begin{bmatrix}
(\hat{\Gamma}_{RB}^{\mathcal{M}_{\varepsilon}})^{\frac{1}{2}} & 0\\
0 & \varepsilon^{\frac{1}{2}}\Pi_{\Gamma^{\mathcal{M}}}^{\perp}
\end{bmatrix}
\\
&  =
\begin{bmatrix}
\varepsilon^{-\frac{\alpha}{2}}(\hat{\Gamma}_{RB}^{\mathcal{M}_{\varepsilon}
})^{\frac{1}{2}} & 0\\
0 & \varepsilon^{\frac{1-\alpha}{2}}\Pi_{\Gamma^{\mathcal{M}}}^{\perp}
\end{bmatrix}
\begin{bmatrix}
\varepsilon(\hat{\Gamma}_{RB}^{\mathcal{M}_{\varepsilon}})^{-\frac{1}{2}
}(\Gamma_{RB}^{\mathcal{N}_{\delta}})_{0,0}(\hat{\Gamma}_{RB}^{\mathcal{M}
_{\varepsilon}})^{-\frac{1}{2}} & \varepsilon^{\frac{1}{2}}(\hat{\Gamma}
_{RB}^{\mathcal{M}_{\varepsilon}})^{-\frac{1}{2}}(\Gamma_{RB}^{\mathcal{N}
_{\delta}})_{0,1}\\
\varepsilon^{\frac{1}{2}}(\Gamma_{RB}^{\mathcal{N}_{\delta}})_{0,1}^{\dag
}(\hat{\Gamma}_{RB}^{\mathcal{M}_{\varepsilon}})^{-\frac{1}{2}} & (\Gamma
_{RB}^{\mathcal{N}_{\delta}})_{1,1}
\end{bmatrix}
^{\alpha}\nonumber\\
&  \qquad\times
\begin{bmatrix}
\varepsilon^{-\frac{\alpha}{2}}(\hat{\Gamma}_{RB}^{\mathcal{M}_{\varepsilon}
})^{\frac{1}{2}} & 0\\
0 & \varepsilon^{\frac{1-\alpha}{2}}\Pi_{\Gamma^{\mathcal{M}}}^{\perp}
\end{bmatrix}
.
\end{align}
Let us define
\begin{equation}
K(\varepsilon):=
\begin{bmatrix}
\varepsilon(\hat{\Gamma}_{RB}^{\mathcal{M}_{\varepsilon}})^{-\frac{1}{2}
}(\Gamma_{RB}^{\mathcal{N}_{\delta}})_{0,0}(\hat{\Gamma}_{RB}^{\mathcal{M}
_{\varepsilon}})^{-\frac{1}{2}} & \varepsilon^{\frac{1}{2}}(\hat{\Gamma}
_{RB}^{\mathcal{M}_{\varepsilon}})^{-\frac{1}{2}}(\Gamma_{RB}^{\mathcal{N}
_{\delta}})_{0,1}\\
\varepsilon^{\frac{1}{2}}(\Gamma_{RB}^{\mathcal{N}_{\delta}})_{0,1}^{\dag
}(\hat{\Gamma}_{RB}^{\mathcal{M}_{\varepsilon}})^{-\frac{1}{2}} & (\Gamma
_{RB}^{\mathcal{N}_{\delta}})_{1,1}
\end{bmatrix}
,
\end{equation}
so that we can write
\begin{multline}
\lbrack\Gamma_{RB}^{\mathcal{M}_{\varepsilon}}]^{1/2}(\left[  \Gamma
_{RB}^{\mathcal{M}_{\varepsilon}}\right]  ^{-1/2}\Gamma_{RB}^{\mathcal{N}
}\left[  \Gamma_{RB}^{\mathcal{M}_{\varepsilon}}\right]  ^{-1/2})^{\alpha
}[\Gamma_{RB}^{\mathcal{M}_{\varepsilon}}]^{1/2}\\
=
\begin{bmatrix}
\varepsilon^{-\frac{\alpha}{2}}(\hat{\Gamma}_{RB}^{\mathcal{M}_{\varepsilon}
})^{\frac{1}{2}} & 0\\
0 & \varepsilon^{\frac{1-\alpha}{2}}\Pi_{\Gamma^{\mathcal{M}}}^{\perp}
\end{bmatrix}
\left[  K(\varepsilon)\right]  ^{\alpha}
\begin{bmatrix}
\varepsilon^{-\frac{\alpha}{2}}(\hat{\Gamma}_{RB}^{\mathcal{M}_{\varepsilon}
})^{\frac{1}{2}} & 0\\
0 & \varepsilon^{\frac{1-\alpha}{2}}\Pi_{\Gamma^{\mathcal{M}}}^{\perp}
\end{bmatrix}
.
\end{multline}
Now let us invoke Lemma~\ref{lem:sisi-zhou-lem}\ with the substitutions
\begin{align}
A  &  \leftrightarrow(\Gamma_{RB}^{\mathcal{N}_{\delta}})_{1,1},\\
B  &  \leftrightarrow(\Gamma_{RB}^{\mathcal{N}_{\delta}})_{0,1}^{\dag}
(\hat{\Gamma}_{RB}^{\mathcal{M}_{\varepsilon}})^{-\frac{1}{2}},\\
C  &  \leftrightarrow\varepsilon(\hat{\Gamma}_{RB}^{\mathcal{M}_{\varepsilon}
})^{-\frac{1}{2}}(\Gamma_{RB}^{\mathcal{N}_{\delta}})_{0,0}(\hat{\Gamma}
_{RB}^{\mathcal{M}_{\varepsilon}})^{-\frac{1}{2}},\\
\varepsilon &  \leftrightarrow\varepsilon^{\frac{1}{2}}.
\end{align}
Defining
\begin{align}
L(\varepsilon)  &  :=
\begin{bmatrix}
\varepsilon S_{\delta} & 0\\
0 & (\Gamma_{RB}^{\mathcal{N}_{\delta}})_{1,1}+\varepsilon R
\end{bmatrix}
,\\
S_{\delta}  &  :=(\hat{\Gamma}_{RB}^{\mathcal{M}_{\varepsilon}})^{-\frac{1}
{2}}\left(  (\Gamma_{RB}^{\mathcal{N}_{\delta}})_{0,0}-(\Gamma_{RB}
^{\mathcal{N}_{\delta}})_{0,1}(\Gamma_{RB}^{\mathcal{N}_{\delta}})_{1,1}
^{-1}(\Gamma_{RB}^{\mathcal{N}_{\delta}})_{0,1}^{\dag}\right)  (\hat{\Gamma
}_{RB}^{\mathcal{M}_{\varepsilon}})^{-\frac{1}{2}},\\
R  &  :=\operatorname{Re}[(\Gamma_{RB}^{\mathcal{N}_{\delta}})_{1,1}
^{-1}(\Gamma_{RB}^{\mathcal{N}_{\delta}})_{0,1}^{\dag}(\hat{\Gamma}
_{RB}^{\mathcal{M}_{\varepsilon}})^{-1}(\Gamma_{RB}^{\mathcal{N}_{\delta}
})_{0,1}],
\end{align}
we conclude from Lemma~\ref{lem:sisi-zhou-lem}\ that
\begin{equation}
\left\Vert K(\varepsilon)-e^{-i\sqrt{\varepsilon}G}L(\varepsilon
)e^{i\sqrt{\varepsilon}G}\right\Vert _{\infty}\leq o(\varepsilon),
\label{eq:approx-diag-geo-ch-form-proof}
\end{equation}
where $G$ in Lemma~\ref{lem:sisi-zhou-lem}\ is defined from $A$ and $B$ above.
The inequality in \eqref{eq:approx-diag-geo-ch-form-proof} in turn implies the
following operator inequalities:
\begin{equation}
e^{-i\sqrt{\varepsilon}G}L(\varepsilon)e^{i\sqrt{\varepsilon}G}-o(\varepsilon
)I\leq K(\varepsilon)\leq e^{-i\sqrt{\varepsilon}G}L(\varepsilon
)e^{i\sqrt{\varepsilon}G}+o(\varepsilon)I.
\label{eq:op-ineq-geo-ren-ch-exp-form}
\end{equation}
Observe that
\begin{equation}
e^{-i\sqrt{\varepsilon}G}L(\varepsilon)e^{i\sqrt{\varepsilon}G}+o(\varepsilon
)I=e^{-i\sqrt{\varepsilon}G}\left[  L(\varepsilon)+o(\varepsilon)I\right]
e^{i\sqrt{\varepsilon}G}.
\end{equation}
Now invoking these and the operator monotonicity of the function $x^{\alpha}$
for $\alpha\in(0,1)$, we find that
\begin{align}
&  \lbrack\Gamma_{RB}^{\mathcal{M}_{\varepsilon}}]^{1/2}(\left[  \Gamma
_{RB}^{\mathcal{M}_{\varepsilon}}\right]  ^{-1/2}\Gamma_{RB}^{\mathcal{N}
}\left[  \Gamma_{RB}^{\mathcal{M}_{\varepsilon}}\right]  ^{-1/2})^{\alpha
}[\Gamma_{RB}^{\mathcal{M}_{\varepsilon}}]^{1/2}\nonumber\\
&  =
\begin{bmatrix}
\varepsilon^{-\frac{\alpha}{2}}(\hat{\Gamma}_{RB}^{\mathcal{M}_{\varepsilon}
})^{\frac{1}{2}} & 0\\
0 & \varepsilon^{\frac{1-\alpha}{2}}\Pi_{\Gamma^{\mathcal{M}}}^{\perp}
\end{bmatrix}
\left[  K(\varepsilon)\right]  ^{\alpha}
\begin{bmatrix}
\varepsilon^{-\frac{\alpha}{2}}(\hat{\Gamma}_{RB}^{\mathcal{M}_{\varepsilon}
})^{\frac{1}{2}} & 0\\
0 & \varepsilon^{\frac{1-\alpha}{2}}\Pi_{\Gamma^{\mathcal{M}}}^{\perp}
\end{bmatrix}
\\
&  \leq
\begin{bmatrix}
\varepsilon^{-\frac{\alpha}{2}}(\hat{\Gamma}_{RB}^{\mathcal{M}_{\varepsilon}
})^{\frac{1}{2}} & 0\\
0 & \varepsilon^{\frac{1-\alpha}{2}}\Pi_{\Gamma^{\mathcal{M}}}^{\perp}
\end{bmatrix}
\left[  e^{-i\sqrt{\varepsilon}G}\left[  L(\varepsilon)+o(\varepsilon
)I\right]  e^{i\sqrt{\varepsilon}G}\right]  ^{\alpha}
\begin{bmatrix}
\varepsilon^{-\frac{\alpha}{2}}(\hat{\Gamma}_{RB}^{\mathcal{M}_{\varepsilon}
})^{\frac{1}{2}} & 0\\
0 & \varepsilon^{\frac{1-\alpha}{2}}\Pi_{\Gamma^{\mathcal{M}}}^{\perp}
\end{bmatrix}
\\
&  =
\begin{bmatrix}
\varepsilon^{-\frac{\alpha}{2}}(\hat{\Gamma}_{RB}^{\mathcal{M}_{\varepsilon}
})^{\frac{1}{2}} & 0\\
0 & \varepsilon^{\frac{1-\alpha}{2}}\Pi_{\Gamma^{\mathcal{M}}}^{\perp}
\end{bmatrix}
e^{-i\sqrt{\varepsilon}G}\left[  L(\varepsilon)+o(\varepsilon)I\right]
^{\alpha}e^{i\sqrt{\varepsilon}G}
\begin{bmatrix}
\varepsilon^{-\frac{\alpha}{2}}(\hat{\Gamma}_{RB}^{\mathcal{M}_{\varepsilon}
})^{\frac{1}{2}} & 0\\
0 & \varepsilon^{\frac{1-\alpha}{2}}\Pi_{\Gamma^{\mathcal{M}}}^{\perp}
\end{bmatrix}
. \label{eq:up-bnd-geo-exp-form-pf}
\end{align}
Defining
\begin{equation}
Q(\varepsilon):=(\Gamma_{RB}^{\mathcal{N}_{\delta}})_{1,1}+\varepsilon
R+o(\varepsilon)I,
\end{equation}
consider that
\begin{align}
\left[  L(\varepsilon)+o(\varepsilon)I\right]  ^{\alpha}  &  =
\begin{bmatrix}
\varepsilon S_{\delta}+o(\varepsilon)I & 0\\
0 & (\Gamma_{RB}^{\mathcal{N}_{\delta}})_{1,1}+\varepsilon R+o(\varepsilon)I
\end{bmatrix}
^{\alpha}\\
&  =
\begin{bmatrix}
\left(  \varepsilon S_{\delta}+o(\varepsilon)I\right)  ^{\alpha} & 0\\
0 & \left(  (\Gamma_{RB}^{\mathcal{N}_{\delta}})_{1,1}+\varepsilon
R+o(\varepsilon)I\right)  ^{\alpha}
\end{bmatrix}
\\
&  =
\begin{bmatrix}
\varepsilon^{\alpha}\left(  S_{\delta}+o(1)I\right)  ^{\alpha} & 0\\
0 & \left(  Q(\varepsilon)\right)  ^{\alpha}
\end{bmatrix}
.
\end{align}
Now expanding $e^{-i\sqrt{\varepsilon}G}$ and $e^{i\sqrt{\varepsilon}G}$ to
first order to evaluate \eqref{eq:up-bnd-geo-exp-form-pf}, we find that
\begin{align}
&
\begin{bmatrix}
\varepsilon^{-\frac{\alpha}{2}}(\hat{\Gamma}_{RB}^{\mathcal{M}_{\varepsilon}
})^{\frac{1}{2}} & 0\\
0 & \varepsilon^{\frac{1-\alpha}{2}}\Pi_{\Gamma^{\mathcal{M}}}^{\perp}
\end{bmatrix}
e^{-i\sqrt{\varepsilon}G}\left[  L(\varepsilon)+o(\varepsilon)I\right]
^{\alpha}e^{i\sqrt{\varepsilon}G}
\begin{bmatrix}
\varepsilon^{-\frac{\alpha}{2}}(\hat{\Gamma}_{RB}^{\mathcal{M}_{\varepsilon}
})^{\frac{1}{2}} & 0\\
0 & \varepsilon^{\frac{1-\alpha}{2}}\Pi_{\Gamma^{\mathcal{M}}}^{\perp}
\end{bmatrix}
\nonumber\\
&  =
\begin{bmatrix}
\varepsilon^{-\frac{\alpha}{2}}(\hat{\Gamma}_{RB}^{\mathcal{M}_{\varepsilon}
})^{\frac{1}{2}} & 0\\
0 & \varepsilon^{\frac{1-\alpha}{2}}\Pi_{\Gamma^{\mathcal{M}}}^{\perp}
\end{bmatrix}
\left[  L(\varepsilon)+o(\varepsilon)I\right]  ^{\alpha}
\begin{bmatrix}
\varepsilon^{-\frac{\alpha}{2}}(\hat{\Gamma}_{RB}^{\mathcal{M}_{\varepsilon}
})^{\frac{1}{2}} & 0\\
0 & \varepsilon^{\frac{1-\alpha}{2}}\Pi_{\Gamma^{\mathcal{M}}}^{\perp}
\end{bmatrix}
\nonumber\\
&  \quad-i\varepsilon^{\frac{1}{2}}
\begin{bmatrix}
\varepsilon^{-\frac{\alpha}{2}}(\hat{\Gamma}_{RB}^{\mathcal{M}_{\varepsilon}
})^{\frac{1}{2}} & 0\\
0 & \varepsilon^{\frac{1-\alpha}{2}}\Pi_{\Gamma^{\mathcal{M}}}^{\perp}
\end{bmatrix}
G\left[  L(\varepsilon)+o(\varepsilon)I\right]  ^{\alpha}
\begin{bmatrix}
\varepsilon^{-\frac{\alpha}{2}}(\hat{\Gamma}_{RB}^{\mathcal{M}_{\varepsilon}
})^{\frac{1}{2}} & 0\\
0 & \varepsilon^{\frac{1-\alpha}{2}}\Pi_{\Gamma^{\mathcal{M}}}^{\perp}
\end{bmatrix}
\nonumber\\
&  \quad+i\varepsilon^{\frac{1}{2}}
\begin{bmatrix}
\varepsilon^{-\frac{\alpha}{2}}(\hat{\Gamma}_{RB}^{\mathcal{M}_{\varepsilon}
})^{\frac{1}{2}} & 0\\
0 & \varepsilon^{\frac{1-\alpha}{2}}\Pi_{\Gamma^{\mathcal{M}}}^{\perp}
\end{bmatrix}
\left[  L(\varepsilon)+o(\varepsilon)I\right]  ^{\alpha}G
\begin{bmatrix}
\varepsilon^{-\frac{\alpha}{2}}(\hat{\Gamma}_{RB}^{\mathcal{M}_{\varepsilon}
})^{\frac{1}{2}} & 0\\
0 & \varepsilon^{\frac{1-\alpha}{2}}\Pi_{\Gamma^{\mathcal{M}}}^{\perp}
\end{bmatrix}
+o(1)\\
&  =
\begin{bmatrix}
(\hat{\Gamma}_{RB}^{\mathcal{M}_{\varepsilon}})^{\frac{1}{2}}\left(
S_{\delta}+o(1)I\right)  ^{\alpha}(\hat{\Gamma}_{RB}^{\mathcal{M}
_{\varepsilon}})^{\frac{1}{2}} & 0\\
0 & \varepsilon^{1-\alpha}\Pi_{\Gamma^{\mathcal{M}}}^{\perp}\left(
Q(\varepsilon)\right)  ^{\alpha}\Pi_{\Gamma^{\mathcal{M}}}^{\perp}
\end{bmatrix}
\nonumber\\
&  \quad-i
\begin{bmatrix}
\varepsilon^{\frac{1-\alpha}{2}}(\hat{\Gamma}_{RB}^{\mathcal{M}_{\varepsilon}
})^{\frac{1}{2}} & 0\\
0 & \varepsilon^{\frac{2-\alpha}{2}}\Pi_{\Gamma^{\mathcal{M}}}^{\perp}
\end{bmatrix}
G
\begin{bmatrix}
\varepsilon^{\frac{\alpha}{2}}\left(  S_{\delta}+o(1)I\right)  ^{\alpha}
(\hat{\Gamma}_{RB}^{\mathcal{M}_{\varepsilon}})^{\frac{1}{2}} & 0\\
0 & \varepsilon^{\frac{1-\alpha}{2}}\left(  Q(\varepsilon)\right)  ^{\alpha
}\Pi_{\Gamma^{\mathcal{M}}}^{\perp}
\end{bmatrix}
\nonumber\\
&  \quad+i
\begin{bmatrix}
\varepsilon^{\frac{\alpha}{2}}(\hat{\Gamma}_{RB}^{\mathcal{M}_{\varepsilon}
})^{\frac{1}{2}}\left(  S_{\delta}+o(1)I\right)  ^{\alpha} & 0\\
0 & \varepsilon^{\frac{1-\alpha}{2}}\Pi_{\Gamma^{\mathcal{M}}}^{\perp}\left(
Q(\varepsilon)\right)  ^{\alpha}
\end{bmatrix}
G
\begin{bmatrix}
\varepsilon^{\frac{1-\alpha}{2}}(\hat{\Gamma}_{RB}^{\mathcal{M}_{\varepsilon}
})^{\frac{1}{2}} & 0\\
0 & \varepsilon^{\frac{2-\alpha}{2}}\Pi_{\Gamma^{\mathcal{M}}}^{\perp}
\end{bmatrix}
+o(1)\\
&  =
\begin{bmatrix}
(\hat{\Gamma}_{RB}^{\mathcal{M}_{\varepsilon}})^{\frac{1}{2}}\left(
S_{\delta}+o(1)I\right)  ^{\alpha}(\hat{\Gamma}_{RB}^{\mathcal{M}
_{\varepsilon}})^{\frac{1}{2}} & 0\\
0 & \varepsilon^{1-\alpha}\Pi_{\Gamma^{\mathcal{M}}}^{\perp}\left(
Q(\varepsilon)\right)  ^{\alpha}\Pi_{\Gamma^{\mathcal{M}}}^{\perp}
\end{bmatrix}
+o(1).
\end{align}
Thus, we have established the following operator inequality:
\begin{multline}
\lbrack\Gamma_{RB}^{\mathcal{M}_{\varepsilon}}]^{1/2}(\left[  \Gamma
_{RB}^{\mathcal{M}_{\varepsilon}}\right]  ^{-1/2}\Gamma_{RB}^{\mathcal{N}
}\left[  \Gamma_{RB}^{\mathcal{M}_{\varepsilon}}\right]  ^{-1/2})^{\alpha
}[\Gamma_{RB}^{\mathcal{M}_{\varepsilon}}]^{1/2}\\
\leq
\begin{bmatrix}
(\hat{\Gamma}_{RB}^{\mathcal{M}_{\varepsilon}})^{\frac{1}{2}}\left(
S_{\delta}+o(1)I\right)  ^{\alpha}(\hat{\Gamma}_{RB}^{\mathcal{M}
_{\varepsilon}})^{\frac{1}{2}} & 0\\
0 & \varepsilon^{1-\alpha}\Pi_{\Gamma^{\mathcal{M}}}^{\perp}\left(
Q(\varepsilon)\right)  ^{\alpha}\Pi_{\Gamma^{\mathcal{M}}}^{\perp}
\end{bmatrix}
+o(1)
\end{multline}
By similar reasoning, but applying the lower bound in
\eqref{eq:op-ineq-geo-ren-ch-exp-form}, we also establish the following
operator inequality lower bound:
\begin{multline}
\begin{bmatrix}
(\hat{\Gamma}_{RB}^{\mathcal{M}_{\varepsilon}})^{\frac{1}{2}}\left(
S_{\delta}-o(1)I\right)  ^{\alpha}(\hat{\Gamma}_{RB}^{\mathcal{M}
_{\varepsilon}})^{\frac{1}{2}} & 0\\
0 & \varepsilon^{1-\alpha}\Pi_{\Gamma^{\mathcal{M}}}^{\perp}\left(
Q(\varepsilon)\right)  ^{\alpha}\Pi_{\Gamma^{\mathcal{M}}}^{\perp}
\end{bmatrix}
+o(1)\\
\leq\lbrack\Gamma_{RB}^{\mathcal{M}_{\varepsilon}}]^{1/2}(\left[  \Gamma
_{RB}^{\mathcal{M}_{\varepsilon}}\right]  ^{-1/2}\Gamma_{RB}^{\mathcal{N}
}\left[  \Gamma_{RB}^{\mathcal{M}_{\varepsilon}}\right]  ^{-1/2})^{\alpha
}[\Gamma_{RB}^{\mathcal{M}_{\varepsilon}}]^{1/2}.
\end{multline}
Now taking the partial trace, evaluating the minimum eigenvalue, and the limit
$\varepsilon\rightarrow0^{+}$, we conclude that
\begin{multline}
\lim_{\varepsilon\rightarrow0^{+}}\lambda_{\min}\left(  \operatorname{Tr}
_{B}\left[  [\Gamma_{RB}^{\mathcal{M}_{\varepsilon}}]^{1/2}(\left[
\Gamma_{RB}^{\mathcal{M}_{\varepsilon}}\right]  ^{-1/2}\Gamma_{RB}
^{\mathcal{N}_{\delta}}\left[  \Gamma_{RB}^{\mathcal{M}_{\varepsilon}}\right]
^{-1/2})^{\alpha}[\Gamma_{RB}^{\mathcal{M}_{\varepsilon}}]^{1/2}\right]
\right) \\
=\lambda_{\min}\left(  \operatorname{Tr}_{B}\left[  [\Gamma_{RB}^{\mathcal{M}
}]^{1/2}(\left[  \Gamma_{RB}^{\mathcal{M}}\right]  ^{-1/2}\widetilde{\Gamma
}_{RB}^{\mathcal{N}_{\delta}}\left[  \Gamma_{RB}^{\mathcal{M}}\right]
^{-1/2})^{\alpha}[\Gamma_{RB}^{\mathcal{M}}]^{1/2}\right]  \right)  ,
\end{multline}
where
\begin{equation}
\widetilde{\Gamma}_{RB}^{\mathcal{N}_{\delta}}:=\left(  (\Gamma_{RB}
^{\mathcal{N}_{\delta}})_{0,0}-(\Gamma_{RB}^{\mathcal{N}_{\delta}}
)_{0,1}(\Gamma_{RB}^{\mathcal{N}_{\delta}})_{1,1}^{-1}(\Gamma_{RB}
^{\mathcal{N}_{\delta}})_{0,1}^{\dag}\right)  (\hat{\Gamma}_{RB}
^{\mathcal{M}_{\varepsilon}})^{-\frac{1}{2}}.
\end{equation}
Noting that
\begin{equation}
\lim_{\delta\rightarrow0^{+}}\widetilde{\Gamma}_{RB}^{\mathcal{N}_{\delta}
}=\widetilde{\Gamma}_{RB}^{\mathcal{N}},
\end{equation}
where $\widetilde{\Gamma}_{RB}^{\mathcal{N}}:=\widetilde{\Gamma}
_{RB}^{\mathcal{N}_{\delta=0}}$, because the image of $(\Gamma_{RB}
^{\mathcal{N}})_{0,1}^{\dag}$ is contained in the support of $(\Gamma
_{RB}^{\mathcal{N}})_{1,1}$, we conclude that
\begin{multline}
\lim_{\delta\rightarrow0^{+}}\lambda_{\min}\left(  \operatorname{Tr}
_{B}\left[  [\Gamma_{RB}^{\mathcal{M}}]^{1/2}(\left[  \Gamma_{RB}
^{\mathcal{M}}\right]  ^{-1/2}\widetilde{\Gamma}_{RB}^{\mathcal{N}_{\delta}
}\left[  \Gamma_{RB}^{\mathcal{M}}\right]  ^{-1/2})^{\alpha}[\Gamma
_{RB}^{\mathcal{M}}]^{1/2}\right]  \right) \\
=\lambda_{\min}\left(  \operatorname{Tr}_{B}\left[  [\Gamma_{RB}^{\mathcal{M}
}]^{1/2}(\left[  \Gamma_{RB}^{\mathcal{M}}\right]  ^{-1/2}\widetilde{\Gamma
}_{RB}^{\mathcal{N}}\left[  \Gamma_{RB}^{\mathcal{M}}\right]  ^{-1/2}
)^{\alpha}[\Gamma_{RB}^{\mathcal{M}}]^{1/2}\right]  \right)  .
\end{multline}
Combining with \eqref{eq:lim-exch-geo-ren-ch-exp-form-1} and
\eqref{eq:lim-exch-geo-ren-ch-exp-form-2}, this concludes the proof.
\end{proof}

\bigskip

\begin{proof}
[Proof of Proposition~\ref{prop:chain-rule-geo-renyi-ch}]Let us first consider
the case $\alpha\in(1,2]$ and $\operatorname{supp}(\rho_{RA})\not \subseteq
\operatorname{supp}(\sigma_{RA})$ or $\operatorname{supp}(\Gamma
_{RB}^{\mathcal{N}})\not \subseteq \operatorname{supp}(\Gamma_{RB}
^{\mathcal{M}})$. In this case, the sum on the right-hand side is equal to
$+\infty$, so that the inequality trivially holds.

Let us then consider the case $\alpha\in(1,2]$ and $\operatorname{supp}
(\rho_{RA})\subseteq\operatorname{supp}(\sigma_{RA})$ and $\operatorname{supp}
(\Gamma_{RB}^{\mathcal{N}})\subseteq\operatorname{supp}(\Gamma_{RB}
^{\mathcal{M}})$. The postselected teleportation identity implies that
\begin{align}
\mathcal{N}_{A\rightarrow B}(\rho_{RA})  &  =\langle\Gamma|_{AS}\rho
_{RA}\otimes\Gamma_{SB}^{\mathcal{N}}|\Gamma\rangle_{AS},\\
\mathcal{M}_{A\rightarrow B}(\sigma_{RA})  &  =\langle\Gamma|_{AS}\sigma
_{RA}\otimes\Gamma_{SB}^{\mathcal{M}}|\Gamma\rangle_{AS}\text{.}
\end{align}
Consider that
\begin{align}
&  \operatorname{Tr}[G_{\alpha}(\mathcal{M}_{A\rightarrow B}(\sigma
_{RA}),\mathcal{N}_{A\rightarrow B}(\rho_{RA}))]\nonumber\\
&  =\operatorname{Tr}[G_{\alpha}(\langle\Gamma|_{AS}\sigma_{RA}\otimes
\Gamma_{SB}^{\mathcal{M}}|\Gamma\rangle_{AS},\langle\Gamma|_{AS}\rho
_{RA}\otimes\Gamma_{SB}^{\mathcal{N}}|\Gamma\rangle_{AS})]\\
&  \leq\operatorname{Tr}[\langle\Gamma|_{AS}G_{\alpha}(\sigma_{RA}
\otimes\Gamma_{SB}^{\mathcal{M}},\rho_{RA}\otimes\Gamma_{SB}^{\mathcal{N}
})|\Gamma\rangle_{AS}]\\
&  =\operatorname{Tr}[\langle\Gamma|_{AS}G_{\alpha}(\sigma_{RA},\rho
_{RA})\otimes G_{\alpha}(\Gamma_{SB}^{\mathcal{M}},\Gamma_{SB}^{\mathcal{N}
})|\Gamma\rangle_{AS}]\\
&  =\operatorname{Tr}_{RB}[\langle\Gamma|_{AS}G_{\alpha}(\sigma_{RA},\rho
_{RA})\otimes G_{\alpha}(\Gamma_{SB}^{\mathcal{M}},\Gamma_{SB}^{\mathcal{N}
})|\Gamma\rangle_{AS}]\\
&  =\langle\Gamma|_{AS}\operatorname{Tr}_{R}[G_{\alpha}(\sigma_{RA},\rho
_{RA})]\otimes\operatorname{Tr}_{B}[G_{\alpha}(\Gamma_{SB}^{\mathcal{M}
},\Gamma_{SB}^{\mathcal{N}})]|\Gamma\rangle_{AS}\\
&  \leq\left\Vert \operatorname{Tr}_{B}[G_{\alpha}(\Gamma_{SB}^{\mathcal{M}
},\Gamma_{SB}^{\mathcal{N}})]\right\Vert _{\infty}\cdot\langle\Gamma
|_{AS}\operatorname{Tr}_{R}[G_{\alpha}(\sigma_{RA},\rho_{RA})]\otimes
I_{S}|\Gamma\rangle_{AS}\\
&  =\left\Vert \operatorname{Tr}_{B}[G_{\alpha}(\Gamma_{SB}^{\mathcal{M}
},\Gamma_{SB}^{\mathcal{N}})]\right\Vert _{\infty}\cdot\operatorname{Tr}
_{RA}[G_{\alpha}(\sigma_{RA},\rho_{RA})]\\
&  =\left\Vert \operatorname{Tr}_{B}[G_{\alpha}(\Gamma_{SB}^{\mathcal{M}
},\Gamma_{SB}^{\mathcal{N}})]\right\Vert _{\infty}\cdot\operatorname{Tr}
[G_{\alpha}(\sigma_{RA},\rho_{RA})].
\end{align}
Now applying a logarithm and dividing by $\alpha-1$, we conclude the chain
rule:
\begin{multline}
\widehat{D}_{\alpha}(\mathcal{N}_{A\rightarrow B}(\rho_{RA}))\Vert
\mathcal{M}_{A\rightarrow B}(\sigma_{RA}))\leq
\frac{1}{\alpha-1}\ln\left\Vert   \operatorname{Tr}
_{B}[G_{\alpha}(\Gamma_{SB}^{\mathcal{M}},\Gamma_{SB}
^{\mathcal{N}})]\right\Vert_\infty \\
+\widehat{D}_{\alpha}(\rho_{RA}\Vert\sigma_{RA}).
\end{multline}

The argument for $\alpha\in(0,1)$ is similar, but we should be careful with
limits and we exploit the minimum eigenvalue instead of the maximum
eigenvalue. Fix $\varepsilon>0$, $\delta\in(0,1)$, and consider that
\begin{align}
&  \operatorname{Tr}[G_{\alpha}(\mathcal{M}_{A\rightarrow B}^{\varepsilon
}(\sigma_{RA}^{\varepsilon}),\mathcal{N}_{A\rightarrow B}^{\delta}(\rho
_{RA}^{\delta}))]\nonumber\\
&  =\operatorname{Tr}[G_{\alpha}(\langle\Gamma|_{AS}\sigma_{RA}^{\varepsilon
}\otimes\Gamma_{SB}^{\mathcal{M}_{\varepsilon}}|\Gamma\rangle_{AS}
,\langle\Gamma|_{AS}\rho_{RA}^{\delta}\otimes\Gamma_{SB}^{\mathcal{N}_{\delta
}}|\Gamma\rangle_{AS})]\\
&  \geq\operatorname{Tr}[\langle\Gamma|_{AS}G_{\alpha}(\sigma_{RA}
^{\varepsilon}\otimes\Gamma_{SB}^{\mathcal{M}_{\varepsilon}},\rho_{RA}
^{\delta}\otimes\Gamma_{SB}^{\mathcal{N}_{\delta}})|\Gamma\rangle_{AS}]\\
&  =\operatorname{Tr}[\langle\Gamma|_{AS}G_{\alpha}(\sigma_{RA}^{\varepsilon
},\rho_{RA}^{\delta})\otimes G_{\alpha}(\Gamma_{SB}^{\mathcal{M}_{\varepsilon
}},\Gamma_{SB}^{\mathcal{N}_{\delta}})|\Gamma\rangle_{AS}]\\
&  =\operatorname{Tr}_{RB}[\langle\Gamma|_{AS}G_{\alpha}(\sigma_{RA}
^{\varepsilon},\rho_{RA}^{\delta})\otimes G_{\alpha}(\Gamma_{SB}
^{\mathcal{M}_{\varepsilon}},\Gamma_{SB}^{\mathcal{N}_{\delta}})|\Gamma
\rangle_{AS}]\\
&  =\langle\Gamma|_{AS}\operatorname{Tr}_{R}[G_{\alpha}(\sigma_{RA}
^{\varepsilon},\rho_{RA}^{\delta})]\otimes\operatorname{Tr}_{B}[G_{\alpha
}(\Gamma_{SB}^{\mathcal{M}_{\varepsilon}},\Gamma_{SB}^{\mathcal{N}_{\delta}
})]|\Gamma\rangle_{AS}\\
&  \geq\lambda_{\min}\left(  \operatorname{Tr}_{B}[G_{\alpha}(\Gamma
_{SB}^{\mathcal{M}_{\varepsilon}},\Gamma_{SB}^{\mathcal{N}_{\delta}})]\right)
\cdot\langle\Gamma|_{AS}\operatorname{Tr}_{R}[G_{\alpha}(\sigma_{RA}
^{\varepsilon},\rho_{RA}^{\delta})]\otimes I_{S}|\Gamma\rangle_{AS}\\
&  =\lambda_{\min}\left(  \operatorname{Tr}_{B}[G_{\alpha}(\Gamma
_{SB}^{\mathcal{M}_{\varepsilon}},\Gamma_{SB}^{\mathcal{N}_{\delta}})]\right)
\cdot\operatorname{Tr}_{RA}[G_{\alpha}(\sigma_{RA}^{\varepsilon},\rho
_{RA}^{\delta})]\\
&  =\lambda_{\min}\left(  \operatorname{Tr}_{B}[G_{\alpha}(\Gamma
_{SB}^{\mathcal{M}_{\varepsilon}},\Gamma_{SB}^{\mathcal{N}_{\delta}})]\right)
\cdot\operatorname{Tr}[G_{\alpha}(\sigma_{RA}^{\varepsilon},\rho_{RA}^{\delta
})].
\end{align}
Now taking a logarithm and dividing by $\alpha-1$, we arrive at the following
inequality:
\begin{multline}
\widehat{D}_{\alpha}(\mathcal{N}_{A\rightarrow B}^{\delta}(\rho_{RA}^{\delta
}))\Vert\mathcal{M}_{A\rightarrow B}^{\varepsilon}(\sigma_{RA}^{\varepsilon
}))\leq\frac{1}{\alpha-1}\ln\lambda_{\min}\left(  \operatorname{Tr}
_{B}[G_{\alpha}(\Gamma_{SB}^{\mathcal{M}_{\varepsilon}},\Gamma_{SB}
^{\mathcal{N}_{\delta}})]\right) \\
+\widehat{D}_{\alpha}(\rho_{RA}^{\delta}\Vert\sigma_{RA}^{\varepsilon}).
\end{multline}
Taking the limit as $\delta\rightarrow0^{+}$, we find that
\begin{multline}
\widehat{D}_{\alpha}(\mathcal{N}_{A\rightarrow B}(\rho_{RA}))\Vert
\mathcal{M}_{A\rightarrow B}^{\varepsilon}(\sigma_{RA}^{\varepsilon}
))\leq\frac{1}{\alpha-1}\ln\lambda_{\min}\left(  \operatorname{Tr}
_{B}[G_{\alpha}(\Gamma_{SB}^{\mathcal{M}_{\varepsilon}},\Gamma_{SB}
^{\mathcal{N}})]\right) \\
+\widehat{D}_{\alpha}(\rho_{RA}\Vert\sigma_{RA}^{\varepsilon}).
\end{multline}
where we used the fact that the operations of evaluating the minimum
eigenvalue and the limit $\delta\rightarrow0^{+}$ commute. Then taking the
limit as $\varepsilon\rightarrow0^{+}$, we conclude that
\begin{multline}
\widehat{D}_{\alpha}(\mathcal{N}_{A\rightarrow B}(\rho_{RA}))\Vert
\mathcal{M}_{A\rightarrow B}(\sigma_{RA}))\leq\lim_{\varepsilon\rightarrow
0^{+}}\frac{1}{\alpha-1}\ln\lambda_{\min}\left(  \operatorname{Tr}
_{B}[G_{\alpha}(\Gamma_{SB}^{\mathcal{M}_{\varepsilon}},\Gamma_{SB}
^{\mathcal{N}})]\right) \\
+\widehat{D}_{\alpha}(\rho_{RA}\Vert\sigma_{RA}).
\end{multline}
This concludes the proof.
\end{proof}

\section{SLD and RLD Fisher informations as limits of R\'enyi relative
entropies}

\label{app:fish-info-fid-limit}

\begin{lemma}
\label{lem:different-exps-fisher-info}For a second-order differentiable family
$\{\rho_{\theta}\}_{\theta}$ of quantum states, the expressions in
\eqref{eq:SLD-fish-bures-connection} and
\eqref{eq:SLD-fish-bures-connection-2} are equal.
\end{lemma}

\begin{proof}
This follows from the linear approximation of the logarithm around one. Set
\begin{equation}
I_{F}\equiv I_{F}(\theta;\{\mathcal{N}_{\theta}\}_{\theta}):=\lim
_{\delta\rightarrow0}\frac{8}{\delta^{2}}f(\theta,\delta),
\label{eq:I_F-alt-exp}
\end{equation}
where
\begin{equation}
f(\theta,\delta):=1-\sqrt{F}(\rho_{\theta}\Vert\rho_{\theta+\delta}),
\end{equation}
and suppose that the limit in \eqref{eq:I_F-alt-exp} exists and is a finite
number. Then for sufficiently small $\delta>0$, the following inequalities
hold
\begin{equation}
\frac{8}{\delta^{2}}f(\theta,\delta)<\left\vert I_{F}\right\vert
+1,\qquad\left\vert f(\theta,\delta)\right\vert <1/2.
\label{eq:for-small-delta}
\end{equation}
Using the following expansion for $x\in\lbrack0,1)$
\begin{equation}
-\ln(1-x)=\sum_{n=1}^{\infty}\frac{x^{n}}{n},
\end{equation}
we find that
\begin{align}
-\frac{4}{\delta^{2}}\ln F(\rho_{\theta}\Vert\rho_{\theta+\delta})  &
=-\frac{8}{\delta^{2}}\ln\sqrt{F}(\rho_{\theta}\Vert\rho_{\theta+\delta})\\
&  =-\frac{8}{\delta^{2}}\ln(1-f(\theta,\delta))\\
&  =\frac{8}{\delta^{2}}\sum_{n=1}^{\infty}\frac{f(\theta,\delta)^{n}}{n}\\
&  =\frac{8}{\delta^{2}}\left[  f(\theta,\delta)+\frac{f(\theta,\delta)^{2}
}{2}+\sum_{n=1}^{\infty}\frac{f(\theta,\delta)^{n+2}}{n+2}\right] \\
&  =\frac{8}{\delta^{2}}\left[  f(\theta,\delta)+f(\theta,\delta)^{2}\left(
\frac{1}{2}+\sum_{n=1}^{\infty}\frac{f(\theta,\delta)^{n}}{n+2}\right)
\right] \\
&  =\frac{8}{\delta^{2}}f(\theta,\delta)+\frac{\delta^{2}}{8}\left[
\frac{8f(\theta,\delta)}{\delta^{2}}\right]  ^{2}\left(  \frac{1}{2}
+\sum_{n=1}^{\infty}\frac{f(\theta,\delta)^{n}}{n+2}\right) \\
&  =\frac{8}{\delta^{2}}f(\theta,\delta)+\frac{\delta^{2}}{8}g(\theta,\delta),
\end{align}
where
\begin{equation}
g(\theta,\delta):=\left[  \frac{8f(\theta,\delta)}{\delta^{2}}\right]
^{2}\left(  \frac{1}{2}+\sum_{n=1}^{\infty}\frac{f(\theta,\delta)^{n}}
{n+2}\right)  .
\end{equation}
For sufficiently small $\delta$, it follows from \eqref{eq:for-small-delta}
that
\[
\left\vert g(\theta,\delta)\right\vert \leq\left[  \left\vert I_{F}\right\vert
+1\right]  ^{2}\sum_{n=0}^{\infty}\left(  \frac{1}{2}\right)  ^{n}=2\left(
\left\vert I_{F}\right\vert +1\right)  ^{2}.
\]
Then we find that
\begin{equation}
\lim_{\delta\rightarrow0}-\frac{4}{\delta^{2}}\ln F(\rho_{\theta}\Vert
\rho_{\theta+\delta})=\lim_{\delta\rightarrow0}\frac{8}{\delta^{2}}
f(\theta,\delta),
\end{equation}
concluding the proof.
\end{proof}

\begin{lemma}
\label{lem:different-exps-RLD-fisher-info}For a second-order differentiable
family $\{\rho_{\theta}\}_{\theta}$ of quantum states, the expressions in
\eqref{eq:geo-renyi-limit-to-RLD} and \eqref{eq:geo-renyi-limit-to-RLD-2} are equal.
\end{lemma}

\begin{proof}
This again follows from the linear approximation of the logarithm around one.
Suppose $\alpha\in(0,1)\cup(1,\infty)$. Set
\begin{equation}
\widehat{I}_{F}\equiv\widehat{I}_{F}(\theta;\{\mathcal{N}_{\theta}\}_{\theta
}):=\lim_{\delta\rightarrow0}\frac{2q_{\alpha}(\theta,\delta)}{\alpha\left(
1-\alpha\right)  \delta^{2}},
\end{equation}
where
\begin{equation}
q_{\alpha}(\theta,\delta):=1-\widehat{Q}_{\alpha}(\rho_{\theta+\delta}
\Vert\rho_{\theta}),
\end{equation}
and suppose that the limit in \eqref{eq:I_F-alt-exp} exists and is a finite
number. Then, for sufficiently small $\delta>0$, the following inequalities
hold
\begin{equation}
\frac{2}{\alpha\left(  1-\alpha\right)  \delta^{2}}q_{\alpha}(\theta
,\delta)<\left\vert \widehat{I}_{F}\right\vert +1,\qquad\left\vert q_{\alpha
}(\theta,\delta)\right\vert <\frac{1}{2}.
\end{equation}
Using the following expansion for $x\in\lbrack0,1)$
\begin{equation}
-\ln(1-x)=\sum_{n=1}^{\infty}\frac{x^{n}}{n},
\end{equation}
and taking sufficiently small $\delta>0$ as stated above, we find that
\begin{align}
&  \frac{2}{\alpha\delta^{2}}\widehat{D}_{\alpha}(\rho_{\theta+\delta}
\Vert\rho_{\theta})\\
&  =-\frac{2}{\alpha\left(  1-\alpha\right)  \delta^{2}}\ln\widehat{Q}
_{\alpha}(\rho_{\theta+\delta}\Vert\rho_{\theta})\\
&  =-\frac{2}{\alpha\left(  1-\alpha\right)  \delta^{2}}\ln(1-q_{\alpha
}(\theta,\delta))\\
&  =\frac{2}{\alpha\left(  1-\alpha\right)  \delta^{2}}\sum_{n=1}^{\infty
}\frac{q_{\alpha}(\theta,\delta)^{n}}{n}\\
&  =\frac{2}{\alpha\left(  1-\alpha\right)  \delta^{2}}\left[  q_{\alpha
}(\theta,\delta)+\frac{q_{\alpha}(\theta,\delta)^{2}}{2}+\sum_{n=1}^{\infty
}\frac{q_{\alpha}(\theta,\delta)^{n+2}}{n+2}\right] \\
&  =\frac{2}{\alpha\left(  1-\alpha\right)  \delta^{2}}\left[  q_{\alpha
}(\theta,\delta)+q_{\alpha}(\theta,\delta)^{2}\left(  \frac{1}{2}+\sum
_{n=1}^{\infty}\frac{q_{\alpha}(\theta,\delta)^{n}}{n+2}\right)  \right] \\
&  =\frac{2}{\alpha\left(  1-\alpha\right)  \delta^{2}}q_{\alpha}
(\theta,\delta)+\frac{\alpha\left(  1-\alpha\right)  }{2}\delta^{2}\left[
\frac{2q_{\alpha}(\theta,\delta)}{\alpha\left(  1-\alpha\right)  \delta^{2}
}\right]  ^{2}\left(  \frac{1}{2}+\sum_{n=1}^{\infty}\frac{q_{\alpha}
(\theta,\delta)^{n}}{n+2}\right) \\
&  =\frac{2}{\alpha\left(  1-\alpha\right)  \delta^{2}}q_{\alpha}
(\theta,\delta)+\frac{\alpha\left(  1-\alpha\right)  }{2}\delta^{2}g_{\alpha
}(\theta,\delta),
\end{align}
where
\begin{equation}
g_{\alpha}(\theta,\delta):=\left[  \frac{2q_{\alpha}(\theta,\delta)}
{\alpha\left(  1-\alpha\right)  \delta^{2}}\right]  ^{2}\left(  \frac{1}
{2}+\sum_{n=1}^{\infty}\frac{q_{\alpha}(\theta,\delta)^{n}}{n+2}\right)  .
\end{equation}
For sufficiently small $\delta$, it follows from \eqref{eq:for-small-delta}
that
\[
\left\vert g_{\alpha}(\theta,\delta)\right\vert \leq\left(  \left\vert
\widehat{I}_{F}\right\vert +1\right)  ^{2}\sum_{n=0}^{\infty}\left(  \frac
{1}{2}\right)  ^{n}=2\left(  \left\vert I_{F}\right\vert +1\right)  ^{2}.
\]
Then we find that
\begin{equation}
\lim_{\delta\rightarrow0}\frac{2}{\alpha\delta^{2}}\widehat{D}_{\alpha}
(\rho_{\theta+\delta}\Vert\rho_{\theta})=\lim_{\delta\rightarrow0}
\frac{2q_{\alpha}(\theta,\delta)}{\alpha\left(  1-\alpha\right)  \delta^{2}},
\end{equation}
concluding the proof.
\end{proof}

\section{RLD Fisher information of quantum channels as a limit of geometric
R\'{e}nyi relative entropy}

\label{app:geo-renyi-limit-to-RLD-fish}

\begin{proposition}
Let $\{\mathcal{N}_{\theta}\}_{\theta}$ be a second-order differentiable
family of channels such that the support condition in
\eqref{eq:finiteness-condition-RLD-fish-ch}\ holds. Then for all $\alpha
\in(0,1)\cup(1,\infty)$, the RLD Fisher information of channels can be written
as
\begin{align}
\widehat{I}_{F}(\theta;\{\mathcal{N}_{\theta}\}_{\theta})  &  =\lim
_{\varepsilon\rightarrow0}\lim_{\delta\rightarrow0}\frac{2}{\alpha\left(
1-\alpha\right)  \delta^{2}}\left(  1-\widehat{Q}_{\alpha}(\mathcal{N}
_{\theta+\delta}^{\varepsilon}\Vert\mathcal{N}_{\theta}^{\varepsilon})\right)
\label{eq:RLD-to-geo-ren-linear-ch}\\
&  =\lim_{\varepsilon\rightarrow0}\lim_{\delta\rightarrow0}\frac{2}
{\alpha\delta^{2}}\widehat{D}_{\alpha}(\mathcal{N}_{\theta+\delta
}^{\varepsilon}\Vert\mathcal{N}_{\theta}^{\varepsilon}),
\label{eq:RLD-to-geo-ren-log-ch}
\end{align}
where $\mathcal{N}_{\theta}^{\varepsilon}(\rho):=(1-\varepsilon)\mathcal{N}
_{\theta}(\rho)+\varepsilon\operatorname{Tr}[\rho]\pi$. Additionally, we have
that
\begin{equation}
\widehat{I}_{F}(\theta;\{\mathcal{N}_{\theta}\}_{\theta})=\lim_{\varepsilon
\rightarrow0}\lim_{\delta\rightarrow0}\frac{2}{\delta^{2}}\widehat
{D}(\mathcal{N}_{\theta+\delta}^{\varepsilon}\Vert\mathcal{N}_{\theta
}^{\varepsilon}). \label{eq:RLD-Fish-limit-of-BS-ent-ch}
\end{equation}

\end{proposition}

\begin{proof}
We focus on the case when $\alpha\in(0,1)$ and for full-rank channels, due to
the order of limits given above and the fact that $\mathcal{N}_{\theta
}^{\varepsilon}(\rho)$ is a full-rank channel for all $\varepsilon\in(0,1)$.
Let $\Gamma_{RB}^{\mathcal{N}_{\theta}}$ denote the Choi operator of the
channel $\mathcal{N}_{\theta}$, and let $\Gamma_{RB}^{\mathcal{N}
_{\theta+\delta}}$ denote the Choi operator of the channel $\mathcal{N}
_{\theta+\delta}$. Let us define
\begin{equation}
d\Gamma_{RB}^{\mathcal{N}_{\theta}}:=\Gamma_{RB}^{\mathcal{N}_{\theta+\delta}
}-\Gamma_{RB}^{\mathcal{N}_{\theta}},
\end{equation}
and observe that
\begin{equation}
\operatorname{Tr}_{B}[d\Gamma_{RB}^{\mathcal{N}_{\theta}}]=0,
\label{eq:PT-differential-equal-zero}
\end{equation}
because $\operatorname{Tr}_{B}[\Gamma_{RB}^{\mathcal{N}_{\theta}
}]=\operatorname{Tr}_{B}[\Gamma_{RB}^{\mathcal{N}_{\theta+\delta}}]=I_{R}
$.\ Then by plugging into \eqref{eq:geo-ren-ch-exp-form-all-cases}, we find
that
\begin{multline}
\widehat{Q}_{\alpha}(\mathcal{N}_{\theta+\delta}\Vert\mathcal{N}_{\theta
})=\label{eq:geometric-fidelity-fisher-explicit}\\
\lambda_{\min}\left(  \operatorname{Tr}_{B}\left[  \left(  \Gamma
_{RB}^{\mathcal{N}_{\theta}}\right)  ^{1/2}\left(  \left(  \Gamma
_{RB}^{\mathcal{N}_{\theta}}\right)  ^{-1/2}\Gamma_{RB}^{\mathcal{N}
_{\theta+\delta}}\left(  \Gamma_{RB}^{\mathcal{N}_{\theta}}\right)
^{-1/2}\right)  ^{\alpha}\left(  \Gamma_{RB}^{\mathcal{N}_{\theta}}\right)
^{1/2}\right]  \right)  .
\end{multline}
Now, by using the expansion
\begin{equation}
\left(  1+x\right)  ^{\alpha}=1+\alpha x+\frac{\alpha\left(  \alpha-1\right)
}{2}x^{2}+O(x^{3}),
\end{equation}
we evaluate the innermost expression of
\eqref{eq:geometric-fidelity-fisher-explicit}:
\begin{align}
&  \left(  \left(  \Gamma_{RB}^{\mathcal{N}_{\theta}}\right)  ^{-1/2}
\Gamma_{RB}^{\mathcal{N}_{\theta+\delta}}\left(  \Gamma_{RB}^{\mathcal{N}
_{\theta}}\right)  ^{-1/2}\right)  ^{\alpha}\nonumber\\
&  =\left(  \left(  \Gamma_{RB}^{\mathcal{N}_{\theta}}\right)  ^{-1/2}\left(
\Gamma_{RB}^{\mathcal{N}_{\theta}}+d\Gamma_{RB}^{\mathcal{N}_{\theta}}\right)
\left(  \Gamma_{RB}^{\mathcal{N}_{\theta}}\right)  ^{-1/2}\right)  ^{\alpha}\\
&  =\left(  I_{RB}+\left(  \Gamma_{RB}^{\mathcal{N}_{\theta}}\right)
^{-1/2}d\Gamma_{RB}^{\mathcal{N}_{\theta}}\left(  \Gamma_{RB}^{\mathcal{N}
_{\theta}}\right)  ^{-1/2}\right)  ^{\alpha}\\
&  =I_{RB}+\alpha\left(  \Gamma_{RB}^{\mathcal{N}_{\theta}}\right)
^{-1/2}d\Gamma_{RB}^{\mathcal{N}_{\theta}}\left(  \Gamma_{RB}^{\mathcal{N}
_{\theta}}\right)  ^{-1/2}\nonumber\\
&  \qquad+\frac{\alpha\left(  \alpha-1\right)  }{2}\left(  \left(  \Gamma
_{RB}^{\mathcal{N}_{\theta}}\right)  ^{-1/2}d\Gamma_{RB}^{\mathcal{N}_{\theta
}}\left(  \Gamma_{RB}^{\mathcal{N}_{\theta}}\right)  ^{-1/2}\right)
^{2}+O\!\left(  \left(  d\Gamma_{RB}^{\mathcal{N}_{\theta}}\right)
^{3}\right)  .
\end{align}
Sandwiching the last expression by $\left(  \Gamma_{RB}^{\mathcal{N}_{\theta}
}\right)  ^{1/2}$ on both sides, we arrive at
\begin{multline}
\left(  \Gamma_{RB}^{\mathcal{N}_{\theta}}\right)  ^{1/2}\left(  \left(
\Gamma_{RB}^{\mathcal{N}_{\theta}}\right)  ^{-1/2}\Gamma_{RB}^{\mathcal{N}
_{\theta+\delta}}\left(  \Gamma_{RB}^{\mathcal{N}_{\theta}}\right)
^{-1/2}\right)  ^{\alpha}\left(  \Gamma_{RB}^{\mathcal{N}_{\theta}}\right)
^{1/2}\\
=\Gamma_{RB}^{\mathcal{N}_{\theta}}+\alpha d\Gamma_{RB}^{\mathcal{N}_{\theta}
}+\frac{\alpha\left(  \alpha-1\right)  }{2}d\Gamma_{RB}^{\mathcal{N}_{\theta}
}\left(  \Gamma_{RB}^{\mathcal{N}_{\theta}}\right)  ^{-1}d\Gamma
_{RB}^{\mathcal{N}_{\theta}}+O\!\left(  \left(  d\Gamma_{RB}^{\mathcal{N}
_{\theta}}\right)  ^{3}\right) .
\end{multline}
Then it follows that the partial trace $\operatorname{Tr}_{B}$ is given by
\begin{multline}
\operatorname{Tr}_{B}\!\left[  \Gamma_{RB}^{\mathcal{N}_{\theta}}+\frac{1}
{2}d\Gamma_{RB}^{\mathcal{N}_{\theta}}-\frac{1}{8}d\Gamma_{RB}^{\mathcal{N}
_{\theta}}\left(  \Gamma_{RB}^{\mathcal{N}_{\theta}}\right)  ^{-1}d\Gamma
_{RB}^{\mathcal{N}_{\theta}}+O\!\left(  \left(  d\Gamma_{RB}^{\mathcal{N}
_{\theta}}\right)  ^{3}\right)  \right] \\
=I_{R}+\frac{\alpha\left(  \alpha-1\right)  }{2}\operatorname{Tr}_{B}\!\left[
d\Gamma_{RB}^{\mathcal{N}_{\theta}}\left(  \Gamma_{RB}^{\mathcal{N}_{\theta}
}\right)  ^{-1}d\Gamma_{RB}^{\mathcal{N}_{\theta}}\right]  +O\!\left(
\operatorname{Tr}_{B}\!\left[  \left(  d\Gamma_{RB}^{\mathcal{N}_{\theta}
}\right)  ^{3}\right]  \right)  ,
\end{multline}
where we used \eqref{eq:PT-differential-equal-zero}. Observe that all higher
order terms correspond to a positive semi-definite operator (each term being
$\left(  \Gamma_{RB}^{\mathcal{N}_{\theta}}\right)  ^{-1}$ sandwiched by other
operators). Supposing that $\delta$ is sufficiently small so that
\begin{equation}
I_{R}+\frac{\alpha\left(  \alpha-1\right)  }{2}\operatorname{Tr}_{B}\!\left[
d\Gamma_{RB}^{\mathcal{N}_{\theta}}\left(  \Gamma_{RB}^{\mathcal{N}_{\theta}
}\right)  ^{-1}d\Gamma_{RB}^{\mathcal{N}_{\theta}}\right]
\end{equation}
is a positive definite operator, we then have the bounds
\begin{align}
&  \lambda_{\min}\!\left(  I_{R}+\frac{\alpha\left(  \alpha-1\right)  }
{2}\operatorname{Tr}_{B}\!\left[  d\Gamma_{RB}^{\mathcal{N}_{\theta}}\left(
\Gamma_{RB}^{\mathcal{N}_{\theta}}\right)  ^{-1}d\Gamma_{RB}^{\mathcal{N}
_{\theta}}\right]  \right) \nonumber\\
&  \leq\lambda_{\min}\!\left(  I_{R}+\frac{\alpha\left(  \alpha-1\right)  }
{2}\operatorname{Tr}_{B}\left[  d\Gamma_{RB}^{\mathcal{N}_{\theta}}\left(
\Gamma_{RB}^{\mathcal{N}_{\theta}}\right)  ^{-1}d\Gamma_{RB}^{\mathcal{N}
_{\theta}}\right]  +O\!\left(  \operatorname{Tr}_{B}\!\left[  \left(
d\Gamma_{RB}^{\mathcal{N}_{\theta}}\right)  ^{3}\right]  \right)  \right) \\
&  \leq\lambda_{\min}\!\left(  I_{R}+\frac{\alpha\left(  \alpha-1\right)  }
{2}\operatorname{Tr}_{B}\left[  d\Gamma_{RB}^{\mathcal{N}_{\theta}}\left(
\Gamma_{RB}^{\mathcal{N}_{\theta}}\right)  ^{-1}d\Gamma_{RB}^{\mathcal{N}
_{\theta}}\right]  \right)  +\lambda_{\max}\!\left(  O\!\left(
\operatorname{Tr}_{B}\!\left[  \left(  d\Gamma_{RB}^{\mathcal{N}_{\theta}
}\right)  ^{3}\right]  \right)  \right) \\
&  =\lambda_{\min}\!\left(  I_{R}+\frac{\alpha\left(  \alpha-1\right)  }
{2}\operatorname{Tr}_{B}\left[  d\Gamma_{RB}^{\mathcal{N}_{\theta}}\left(
\Gamma_{RB}^{\mathcal{N}_{\theta}}\right)  ^{-1}d\Gamma_{RB}^{\mathcal{N}
_{\theta}}\right]  \right)  +\left\Vert \left(  O\!\left(  \operatorname{Tr}
_{B}\!\left[  \left(  d\Gamma_{RB}^{\mathcal{N}_{\theta}}\right)  ^{3}\right]
\right)  \right)  \right\Vert _{\infty}\\
&  \leq\lambda_{\min}\!\left(  I_{R}+\frac{\alpha\left(  \alpha-1\right)  }
{2}\operatorname{Tr}_{B}\left[  d\Gamma_{RB}^{\mathcal{N}_{\theta}}\left(
\Gamma_{RB}^{\mathcal{N}_{\theta}}\right)  ^{-1}d\Gamma_{RB}^{\mathcal{N}
_{\theta}}\right]  \right)  +O\!\left(  \left\Vert \operatorname{Tr}
_{B}\!\left[  \left(  d\Gamma_{RB}^{\mathcal{N}_{\theta}}\right)  ^{3}\right]
\right\Vert _{\infty}\right) \\
&  =\lambda_{\min}\!\left(  I_{R}+\frac{\alpha\left(  \alpha-1\right)  }
{2}\operatorname{Tr}_{B}\left[  d\Gamma_{RB}^{\mathcal{N}_{\theta}}\left(
\Gamma_{RB}^{\mathcal{N}_{\theta}}\right)  ^{-1}d\Gamma_{RB}^{\mathcal{N}
_{\theta}}\right]  \right)  +O\!\left(  \left\Vert d\Gamma_{RB}^{\mathcal{N}
_{\theta}}\right\Vert _{\infty}^{3}\right)  .
\label{eq:triangle-ineq-extra-term}
\end{align}
The first two inequalities are a consequence of the following inequalities
that hold for positive definite operators $X$ and $Y$:
\begin{equation}
\lambda_{\min}(X)\leq\lambda_{\min}(X+Y)\leq\lambda_{\min}(X)+\lambda_{\max
}(Y).
\end{equation}
In the second-to-last last line we employed the submultiplicavity of the
infinity norm, and in the last line the bound
\begin{equation}
\left\Vert \operatorname{Tr}_{B}[X_{AB}]\right\Vert _{\infty}\leq
d_{B}\left\Vert X_{AB}\right\Vert _{\infty},
\end{equation}
where $d_{B}$ is the dimension of the channel output system $B$, as well as
the fact that $d_{B}<\infty$ is a constant. For a second-order differentiable
family, the following limit holds
\begin{equation}
\lim_{\delta\rightarrow0}\frac{1}{\delta^{2}}O\!\left(  \left\Vert
d\Gamma_{RB}^{\mathcal{N}_{\theta}}\right\Vert _{\infty}^{3}\right)
=\lim_{\delta\rightarrow0}\delta\ O\!\left(  \left[  \left\Vert d\Gamma
_{RB}^{\mathcal{N}_{\theta}}/\delta\right\Vert _{\infty}\right]  ^{3}\right)
=0. \label{eq:third-order-assumption}
\end{equation}
This means that we can then focus on the term
\begin{equation}
\lambda_{\min}\!\left(  I_{R}+\frac{\alpha\left(  \alpha-1\right)  }
{2}\operatorname{Tr}_{B}\!\left[  d\Gamma_{RB}^{\mathcal{N}_{\theta}}\left(
\Gamma_{RB}^{\mathcal{N}_{\theta}}\right)  ^{-1}d\Gamma_{RB}^{\mathcal{N}
_{\theta}}\right]  \right)  ,
\end{equation}
because the last term in \eqref{eq:triangle-ineq-extra-term}\ will vanish when
we divide by $\delta^{2}$ and take the final limit as $\delta\rightarrow0$.
For any positive semi-definite operator $A$ with sufficiently small
eigenvalues all strictly less than one, it follows that
\begin{equation}
\lambda_{\min}\left(  I-A\right)  =1-\lambda_{\max}(A)=1-\left\Vert
A\right\Vert _{\infty}.
\end{equation}
We can apply this reasoning to the operator
\begin{equation}
-\frac{\alpha\left(  \alpha-1\right)  }{2}\operatorname{Tr}_{B}\!\left[
d\Gamma_{RB}^{\mathcal{N}_{\theta}}\left(  \Gamma_{RB}^{\mathcal{N}_{\theta}
}\right)  ^{-1}d\Gamma_{RB}^{\mathcal{N}_{\theta}}\right]
\end{equation}
because it is positive semi-definite and its eigenvalues can be made
arbitrarily close to zero for $\delta$ small enough. Then by employing the
expression in \eqref{eq:RLD-to-geo-ren-linear-ch}, we find that
\begin{align}
&  \widehat{I}_{F}(\theta;\{\mathcal{N}_{\theta}\}_{\theta})\nonumber\\
&  =\lim_{\delta\rightarrow0}\frac{2}{\alpha\left(  \alpha-1\right)
\delta^{2}}\left(  \widehat{Q}_{\alpha}(\mathcal{N}_{\theta+\delta}
\Vert\mathcal{N}_{\theta})-1\right) \\
&  =\lim_{\delta\rightarrow0}\frac{2}{\alpha\left(  \alpha-1\right)
\delta^{2}}\left(  \left[  1+\frac{\alpha\left(  \alpha-1\right)  }
{2}\left\Vert \operatorname{Tr}_{B}\!\left[  d\Gamma_{RB}^{\mathcal{N}
_{\theta}}\left(  \Gamma_{RB}^{\mathcal{N}_{\theta}}\right)  ^{-1}d\Gamma
_{RB}^{\mathcal{N}_{\theta}}\right]  \right\Vert _{\infty}\right]  -1 \right) \\
&  =\lim_{\delta\rightarrow0}\frac{1}{\delta^{2}}\left\Vert \operatorname{Tr}
_{B}\!\left[  d\Gamma_{RB}^{\mathcal{N}_{\theta}}\left(  \Gamma_{RB}
^{\mathcal{N}_{\theta}}\right)  ^{-1}d\Gamma_{RB}^{\mathcal{N}_{\theta}
}\right]  \right\Vert _{\infty}\\
&  =\lim_{\delta\rightarrow0}\left\Vert \operatorname{Tr}_{B}\!\left[
\frac{d\Gamma_{RB}^{\mathcal{N}_{\theta}}}{\delta}\left(  \Gamma
_{RB}^{\mathcal{N}_{\theta}}\right)  ^{-1}\frac{d\Gamma_{RB}^{\mathcal{N}
_{\theta}}}{\delta}\right]  \right\Vert _{\infty}\\
&  =\left\Vert \lim_{\delta\rightarrow0}\operatorname{Tr}_{B}\!\left[
\frac{d\Gamma_{RB}^{\mathcal{N}_{\theta}}}{\delta}\left(  \Gamma
_{RB}^{\mathcal{N}_{\theta}}\right)  ^{-1}\frac{d\Gamma_{RB}^{\mathcal{N}
_{\theta}}}{\delta}\right]  \right\Vert _{\infty}\\
&  =\left\Vert \operatorname{Tr}_{B}\!\left[  \left(  \partial_{\theta}
\Gamma_{RB}^{\mathcal{N}_{\theta}}\right)  \left(  \Gamma_{RB}^{\mathcal{N}
_{\theta}}\right)  ^{-1}\left(  \partial_{\theta}\Gamma_{RB}^{\mathcal{N}
_{\theta}}\right)  \right]  \right\Vert _{\infty}.
\label{eq:final-expression-GFI}
\end{align}
The third-to-last line follows because $c\lambda_{\max}(A)=\lambda_{\max}(cA)$
for a positive semi-definite operator $A$ and scaling parameter $c>0$. The
second-to-last line follows because the maximum and limit commute. The last
line follows by evaluating the limit.

The proof for $\alpha \in (1,\infty)$ is similar, except that we work with $\lambda_{\max}$ instead of $\lambda_{\min}$.

The proof that \eqref{eq:RLD-to-geo-ren-linear-ch} is equal to
\eqref{eq:RLD-to-geo-ren-log-ch} is similar to the proof of
Lemma~\ref{lem:different-exps-RLD-fisher-info}.

The proof of \eqref{eq:RLD-Fish-limit-of-BS-ent-ch} is similar to the proof of
\eqref{eq:RLD-from-BS-rel-ent}. Consider that
\begin{equation}
\widehat{D}(\mathcal{N}_{\theta+\delta}\Vert\mathcal{N}_{\theta})=\left\Vert
\operatorname{Tr}_{B}\!\left[  \left(  \Gamma^{\mathcal{N}_{\theta+\delta}
}\right)  ^{\frac{1}{2}}\ln\left(  \left(  \Gamma^{\mathcal{N}_{\theta+\delta
}}\right)  ^{\frac{1}{2}}(\Gamma^{\mathcal{N}_{\theta}})^{-1}\left(
\Gamma^{\mathcal{N}_{\theta+\delta}}\right)  ^{\frac{1}{2}}\right)  \left(
\Gamma^{\mathcal{N}_{\theta+\delta}}\right)  ^{\frac{1}{2}}\right]
\right\Vert _{\infty}.
\end{equation}
So we focus on the operator in the middle. It suffices to consider a full-rank
channel family and consider for $\eta(x)=x\ln x$ that
\begin{align}
&  \left(  \Gamma^{\mathcal{N}_{\theta+\delta}}\right)  ^{\frac{1}{2}}
\ln\left(  \left(  \Gamma^{\mathcal{N}_{\theta+\delta}}\right)  ^{\frac{1}{2}
}(\Gamma^{\mathcal{N}_{\theta}})^{-1}\left(  \Gamma^{\mathcal{N}
_{\theta+\delta}}\right)  ^{\frac{1}{2}}\right)  \left(  \Gamma^{\mathcal{N}
_{\theta+\delta}}\right)  ^{\frac{1}{2}}\nonumber\\
&  =(\Gamma^{\mathcal{N}_{\theta}})^{\frac{1}{2}}(\Gamma^{\mathcal{N}_{\theta
}})^{-\frac{1}{2}}\left(  \Gamma^{\mathcal{N}_{\theta+\delta}}\right)
^{\frac{1}{2}}\ln\left(  \left(  \Gamma^{\mathcal{N}_{\theta+\delta}}\right)
^{\frac{1}{2}}(\Gamma^{\mathcal{N}_{\theta}})^{-\frac{1}{2}}(\Gamma
^{\mathcal{N}_{\theta}})^{-\frac{1}{2}}\left(  \Gamma^{\mathcal{N}
_{\theta+\delta}}\right)  ^{\frac{1}{2}}\right)  \left(  \Gamma^{\mathcal{N}
_{\theta+\delta}}\right)  ^{\frac{1}{2}}\\
&  =(\Gamma^{\mathcal{N}_{\theta}})^{\frac{1}{2}}\ln\left(  (\Gamma
^{\mathcal{N}_{\theta}})^{-\frac{1}{2}}\left(  \Gamma^{\mathcal{N}
_{\theta+\delta}}\right)  ^{\frac{1}{2}}\left(  \Gamma^{\mathcal{N}
_{\theta+\delta}}\right)  ^{\frac{1}{2}}(\Gamma^{\mathcal{N}_{\theta}
})^{-\frac{1}{2}}\right)  (\Gamma^{\mathcal{N}_{\theta}})^{-\frac{1}{2}
}\left(  \Gamma^{\mathcal{N}_{\theta+\delta}}\right)  ^{\frac{1}{2}}\left(
\Gamma^{\mathcal{N}_{\theta+\delta}}\right)  ^{\frac{1}{2}}\\
&  =(\Gamma^{\mathcal{N}_{\theta}})^{\frac{1}{2}}\ln\left(  (\Gamma
^{\mathcal{N}_{\theta}})^{-\frac{1}{2}}\left(  \Gamma^{\mathcal{N}
_{\theta+\delta}}\right)  (\Gamma^{\mathcal{N}_{\theta}})^{-\frac{1}{2}
}\right)  \left(  (\Gamma^{\mathcal{N}_{\theta}})^{-\frac{1}{2}}\left(
\Gamma^{\mathcal{N}_{\theta+\delta}}\right)  (\Gamma^{\mathcal{N}_{\theta}
})^{-\frac{1}{2}}\right)  (\Gamma^{\mathcal{N}_{\theta}})^{\frac{1}{2}}\\
&  =(\Gamma^{\mathcal{N}_{\theta}})^{\frac{1}{2}}\eta\left(  (\Gamma
^{\mathcal{N}_{\theta}})^{-\frac{1}{2}}\left(  \Gamma^{\mathcal{N}
_{\theta+\delta}}\right)  (\Gamma^{\mathcal{N}_{\theta}})^{-\frac{1}{2}
}\right)  (\Gamma^{\mathcal{N}_{\theta}})^{\frac{1}{2}}\\
&  =(\Gamma^{\mathcal{N}_{\theta}})^{\frac{1}{2}}\eta\left(  (\Gamma
^{\mathcal{N}_{\theta}})^{-\frac{1}{2}}\left(  \Gamma^{\mathcal{N}_{\theta}
}+d\Gamma^{\mathcal{N}_{\theta}}\right)  (\Gamma^{\mathcal{N}_{\theta}
})^{-\frac{1}{2}}\right)  (\Gamma^{\mathcal{N}_{\theta}})^{\frac{1}{2}}\\
&  =(\Gamma^{\mathcal{N}_{\theta}})^{\frac{1}{2}}\eta\left(  I+(\Gamma
^{\mathcal{N}_{\theta}})^{-\frac{1}{2}}d\Gamma^{\mathcal{N}_{\theta}}
(\Gamma^{\mathcal{N}_{\theta}})^{-\frac{1}{2}}\right)  (\Gamma^{\mathcal{N}
_{\theta}})^{\frac{1}{2}}\\
&  =(\Gamma^{\mathcal{N}_{\theta}})^{\frac{1}{2}}\left(  (\Gamma
^{\mathcal{N}_{\theta}})^{-\frac{1}{2}}d\Gamma^{\mathcal{N}_{\theta}}
(\Gamma^{\mathcal{N}_{\theta}})^{-\frac{1}{2}}+\left[  (\Gamma^{\mathcal{N}
_{\theta}})^{-\frac{1}{2}}d\Gamma^{\mathcal{N}_{\theta}}(\Gamma^{\mathcal{N}
_{\theta}})^{-\frac{1}{2}}\right]  ^{2}/2+O((d\Gamma^{\mathcal{N}_{\theta}
})^{3})\right)  (\Gamma^{\mathcal{N}_{\theta}})^{\frac{1}{2}}\\
&  =d\Gamma^{\mathcal{N}_{\theta}}+d\Gamma^{\mathcal{N}_{\theta}}
(\Gamma^{\mathcal{N}_{\theta}})^{-1}d\Gamma^{\mathcal{N}_{\theta}
}/2+O((d\Gamma^{\mathcal{N}_{\theta}})^{3}).
\end{align}
The second equality follows from Lemma~\ref{lem:sing-val-lemma-pseudo-commute}
, with $f=\ln$ and $L=\left(  \Gamma^{\mathcal{N}_{\theta+\delta}}\right)
^{\frac{1}{2}}(\Gamma^{\mathcal{N}_{\theta}})^{-\frac{1}{2}}$. The
second-to-last equality follows because $\eta(1+x)=x+x^{2}/2+O(x^{3})$. Now
evaluating the partial trace over $B$, we find that
\begin{multline}
\operatorname{Tr}_{B}\left[  d\Gamma^{\mathcal{N}_{\theta}}+d\Gamma
^{\mathcal{N}_{\theta}}(\Gamma^{\mathcal{N}_{\theta}})^{-1}d\Gamma
^{\mathcal{N}_{\theta}}/2+O((d\Gamma^{\mathcal{N}_{\theta}})^{3})\right] \\
=\frac{1}{2}\operatorname{Tr}_{B}\left[  d\Gamma^{\mathcal{N}_{\theta}}
(\Gamma^{\mathcal{N}_{\theta}})^{-1}d\Gamma^{\mathcal{N}_{\theta}}\right]
+O((d\Gamma^{\mathcal{N}_{\theta}})^{3}),
\end{multline}
which follows because $\operatorname{Tr}_{B}\left[  d\Gamma^{\mathcal{N}
_{\theta}}\right]  =0$. Then finally
\begin{align}
&  \lim_{\delta\rightarrow0}\frac{2}{\delta^{2}}\widehat{D}(\mathcal{N}
_{\theta+\delta}\Vert\mathcal{N}_{\theta})\nonumber\\
&  =\lim_{\delta\rightarrow0}\frac{2}{\delta^{2}}\left\Vert \frac{1}
{2}\operatorname{Tr}_{B}\!\left[  d\Gamma^{\mathcal{N}_{\theta}}
(\Gamma^{\mathcal{N}_{\theta}})^{-1}d\Gamma^{\mathcal{N}_{\theta}}\right]
+O((d\Gamma^{\mathcal{N}_{\theta}})^{3})\right\Vert _{\infty}\\
&  =\lim_{\delta\rightarrow0}\left\Vert \operatorname{Tr}_{B}\!\left[
\frac{d\Gamma^{\mathcal{N}_{\theta}}}{\delta}(\Gamma^{\mathcal{N}_{\theta}
})^{-1}\frac{d\Gamma^{\mathcal{N}_{\theta}}}{\delta}\right]  +O(\delta
(d\Gamma^{\mathcal{N}_{\theta}}/\delta)^{3})\right\Vert _{\infty}\\
&  =\left\Vert \operatorname{Tr}_{B}\!\left[  \left(  \partial_{\theta}
\Gamma_{RB}^{\mathcal{N}_{\theta}}\right)  \left(  \Gamma_{RB}^{\mathcal{N}
_{\theta}}\right)  ^{-1}\left(  \partial_{\theta}\Gamma_{RB}^{\mathcal{N}
_{\theta}}\right)  \right]  \right\Vert _{\infty}.
\end{align}
This concludes the proof.
\end{proof}

\section{Semi-definite program for the root fidelity of quantum
channels\label{app:SDP-root-fid-ch}}

\begin{proof}
[Proof of Proposition~\ref{prop:SDP-ch-fid}]For a pure bipartite state
$\psi_{RA}$, we use the fact that
\begin{equation}
\psi_{RA}=X_{R}\Gamma_{RA}X_{R}^{\dag},
\end{equation}
where $\operatorname{Tr}[X_{R}^{\dag}X_{R}]=1$ to see that
\begin{equation}
\mathcal{N}_{A\rightarrow B}(\psi_{RA})=X_{R}\Gamma_{RB}^{\mathcal{N}}
X_{R}^{\dag},\quad\mathcal{M}_{A\rightarrow B}(\psi_{RA})=X_{R}\Gamma
_{RB}^{\mathcal{M}}X_{R}^{\dag},
\end{equation}
and then plug in to \eqref{eq:dual-root-fid-states-SDP}\ to get that
\begin{equation}
\sqrt{F}(\mathcal{N}_{A\rightarrow B},\mathcal{M}_{A\rightarrow B})=\frac
{1}{2}\inf_{W_{RB},Z_{RB}}\operatorname{Tr}[X_{R}\Gamma_{RB}^{\mathcal{N}
}X_{R}^{\dag}W_{RB}]+\operatorname{Tr}[X_{R}\Gamma_{RB}^{\mathcal{M}}
X_{R}^{\dag}Z_{RB}]
\end{equation}
subject to
\begin{equation}
\begin{bmatrix}
W_{RB} & I_{RB}\\
I_{RB} & Z_{RB}
\end{bmatrix}
\geq0. \label{eq:ineq-const-dual-fid-states-SDP}
\end{equation}
Consider that the objective function can be written as
\begin{equation}
\operatorname{Tr}[\Gamma_{RB}^{\mathcal{N}}W_{RB}^{\prime}]+\operatorname{Tr}
[\Gamma_{RB}^{\mathcal{M}}Z_{RB}^{\prime}],
\end{equation}
with
\begin{equation}
W_{RB}^{\prime}:=X_{R}^{\dag}W_{RB}X_{R},\quad Z_{RB}^{\prime}:=X_{R}^{\dag
}Z_{RB}X_{R}
\end{equation}
Now consider that the inequality in
\eqref{eq:ineq-const-dual-fid-states-SDP}\ is equivalent to
\begin{equation}
\begin{bmatrix}
X_{R} & 0\\
0 & X_{R}
\end{bmatrix}
^{\dag}
\begin{bmatrix}
W_{RB} & I_{RB}\\
I_{RB} & Z_{RB}
\end{bmatrix}
\begin{bmatrix}
X_{R} & 0\\
0 & X_{R}
\end{bmatrix}
\geq0
\end{equation}
Multiplying out the last matrix we find that
\begin{align}
&
\begin{bmatrix}
X_{R} & 0\\
0 & X_{R}
\end{bmatrix}
^{\dag}
\begin{bmatrix}
W_{RB} & I_{RB}\\
I_{RB} & Z_{RB}
\end{bmatrix}
\begin{bmatrix}
X_{R} & 0\\
0 & X_{R}
\end{bmatrix}
\nonumber\\
&  =
\begin{bmatrix}
X_{R}^{\dag}W_{RB}X_{R} & X_{R}^{\dag}X_{R}\otimes I_{B}\\
X_{R}^{\dag}X_{R}\otimes I_{B} & X_{R}^{\dag}Z_{RB}X_{R}
\end{bmatrix}
\\
&  =
\begin{bmatrix}
W_{RB}^{\prime} & \rho_{R}\otimes I_{B}\\
\rho_{R}\otimes I_{B} & Z_{RB}^{\prime}
\end{bmatrix}
,
\end{align}
where we defined $\rho_{R}=X_{R}^{\dag}X_{R}$. Observing that $\rho_{R}\geq0$
and $\operatorname{Tr}[\rho_{R}]=1$, we can write the final SDP\ as follows:
\begin{equation}
\sqrt{F}(\mathcal{N}_{A\rightarrow B},\mathcal{M}_{A\rightarrow B})=\frac
{1}{2}\inf_{\rho_{R},W_{RB},Z_{RB}}\operatorname{Tr}[\Gamma_{RB}^{\mathcal{N}
}W_{RB}]+\operatorname{Tr}[\Gamma_{RB}^{\mathcal{M}}Z_{RB}],
\end{equation}
subject to
\begin{equation}
\rho_{R}\geq0,\quad\operatorname{Tr}[\rho_{R}]=1,\quad
\begin{bmatrix}
W_{RB} & \rho_{R}\otimes I_{B}\\
\rho_{R}\otimes I_{B} & Z_{RB}
\end{bmatrix}
\geq0. \label{eq:ch-fid-SDP-constr}
\end{equation}

Now let us calculate the dual SDP\ to this, using the following standard forms
for primal and dual SDPs, with Hermitian operators $A$ and $B$ and a
Hermiticity-preserving map $\Phi$ \cite{Wat18}:
\begin{equation}
\sup_{X\geq0}\left\{  \operatorname{Tr}[AX]:\Phi(X)\leq B\right\}  ,\qquad
\inf_{Y\geq0}\left\{  \operatorname{Tr}[BY]:\Phi^{\dag}(Y)\geq A\right\}  .
\end{equation}
Consider that the constraint in \eqref{eq:ch-fid-SDP-constr} implies
$W_{RB}\geq0$ and $Z_{RB}\geq0$, so that we can set
\begin{align}
Y  &  =
\begin{bmatrix}
W_{RB} & 0 & 0\\
0 & Z_{RB} & 0\\
0 & 0 & \rho_{R}
\end{bmatrix}
,\quad B=
\begin{bmatrix}
\Gamma_{RB}^{\mathcal{N}} & 0 & 0\\
0 & \Gamma_{RB}^{\mathcal{M}} & 0\\
0 & 0 & 0
\end{bmatrix}
,\\
\Phi^{\dag}(Y)  &  =
\begin{bmatrix}
W_{RB} & \rho_{R}\otimes I_{B} & 0 & 0\\
\rho_{R}\otimes I_{B} & Z_{RB} & 0 & 0\\
0 & 0 & \operatorname{Tr}[\rho_{R}] & 0\\
0 & 0 & 0 & -\operatorname{Tr}[\rho_{R}]
\end{bmatrix}
,\\
A  &  =
\begin{bmatrix}
0 & 0 & 0 & 0\\
0 & 0 & 0 & 0\\
0 & 0 & 1 & 0\\
0 & 0 & 0 & -1
\end{bmatrix}
.
\end{align}
Then with
\begin{equation}
X=
\begin{bmatrix}
P_{RB} & Q_{RB}^{\dag} & 0 & 0\\
Q_{RB} & S_{RB} & 0 & 0\\
0 & 0 & \lambda & 0\\
0 & 0 & 0 & \mu
\end{bmatrix}
\end{equation}
the map $\Phi$ is given by
\begin{align}
&  \operatorname{Tr}[X\Phi^{\dag}(Y)]\nonumber\\
&  =\operatorname{Tr}\left[
\begin{bmatrix}
P_{RB} & Q_{RB}^{\dag} & 0 & 0\\
Q_{RB} & S_{RB} & 0 & 0\\
0 & 0 & \lambda & 0\\
0 & 0 & 0 & \mu
\end{bmatrix}
\begin{bmatrix}
W_{RB} & \rho_{R}\otimes I_{B} & 0 & 0\\
\rho_{R}\otimes I_{B} & Z_{RB} & 0 & 0\\
0 & 0 & \operatorname{Tr}[\rho_{R}] & 0\\
0 & 0 & 0 & -\operatorname{Tr}[\rho_{R}]
\end{bmatrix}
\right] \\
&  =\operatorname{Tr}[P_{RB}W_{RB}]+\operatorname{Tr}[Q_{RB}^{\dag}\left(
\rho_{R}\otimes I_{B}\right)  ]+\operatorname{Tr}[Q_{RB}(\rho_{R}\otimes
I_{B})]\nonumber\\
&  \qquad+\operatorname{Tr}[S_{RB}Z_{RB}]+\left(  \lambda-\mu\right)
\operatorname{Tr}[\rho_{R}]\\
&  =\operatorname{Tr}[P_{RB}W_{RB}]+\operatorname{Tr}[S_{RB}Z_{RB}
]+\operatorname{Tr}[(\operatorname{Tr}_{B}[Q_{RB}+Q_{RB}^{\dag}]+\left(
\lambda-\mu\right)  I_{R})\rho_{R}]\\
&  =\operatorname{Tr}\left[
\begin{bmatrix}
P_{RB} & 0 & 0\\
0 & S_{RB} & 0\\
0 & 0 & \operatorname{Tr}_{B}[Q_{RB}+Q_{RB}^{\dag}]+\left(  \lambda
-\mu\right)  I_{R}
\end{bmatrix}
\begin{bmatrix}
W_{RB} & 0 & 0\\
0 & Z_{RB} & 0\\
0 & 0 & \rho_{R}
\end{bmatrix}
\right]  .
\end{align}
So then
\begin{equation}
\Phi(X)=
\begin{bmatrix}
P_{RB} & 0 & 0\\
0 & S_{RB} & 0\\
0 & 0 & \operatorname{Tr}_{B}[Q_{RB}+Q_{RB}^{\dag}]+\left(  \lambda
-\mu\right)  I_{R}
\end{bmatrix}
.
\end{equation}
The primal is then given by
\begin{equation}
\frac{1}{2}\sup\operatorname{Tr}\left[
\begin{bmatrix}
0 & 0 & 0 & 0\\
0 & 0 & 0 & 0\\
0 & 0 & 1 & 0\\
0 & 0 & 0 & -1
\end{bmatrix}
\begin{bmatrix}
P_{RB} & Q_{RB}^{\dag} & 0 & 0\\
Q_{RB} & S_{RB} & 0 & 0\\
0 & 0 & \lambda & 0\\
0 & 0 & 0 & \mu
\end{bmatrix}
\right]  ,
\end{equation}
subject to
\begin{align}
\begin{bmatrix}
P_{RB} & 0 & 0\\
0 & S_{RB} & 0\\
0 & 0 & \operatorname{Tr}_{B}[Q_{RB}+Q_{RB}^{\dag}]+\left(  \lambda
-\mu\right)  I_{R}
\end{bmatrix}
&  \leq
\begin{bmatrix}
\Gamma_{RB}^{\mathcal{N}} & 0 & 0\\
0 & \Gamma_{RB}^{\mathcal{M}} & 0\\
0 & 0 & 0
\end{bmatrix}
,\\
\begin{bmatrix}
P_{RB} & Q_{RB}^{\dag} & 0 & 0\\
Q_{RB} & S_{RB} & 0 & 0\\
0 & 0 & \lambda & 0\\
0 & 0 & 0 & \mu
\end{bmatrix}
&  \geq0,
\end{align}
which simplifies to
\begin{equation}
\frac{1}{2}\sup\left(  \lambda-\mu\right)
\end{equation}
subject to
\begin{align}
P_{RB}  &  \leq\Gamma_{RB}^{\mathcal{N}},\\
S_{RB}  &  \leq\Gamma_{RB}^{\mathcal{M}},\\
\operatorname{Tr}_{B}[Q_{RB}+Q_{RB}^{\dag}]+\left(  \lambda-\mu\right)  I_{R}
&  \leq0,\\
\begin{bmatrix}
P_{RB} & Q_{RB}^{\dag}\\
Q_{RB} & S_{RB}
\end{bmatrix}
&  \geq0,\\
\lambda,\mu &  \geq0.
\end{align}
We can simplify this even more. We can set $\lambda^{\prime}=\lambda-\mu
\in\mathbb{R}$, and we can substitute $Q_{RB}$ with $-Q_{RB}$ without changing
the value, so then it becomes
\begin{equation}
\frac{1}{2}\sup\lambda^{\prime}
\end{equation}
subject to
\begin{align}
P_{RB}  &  \leq\Gamma_{RB}^{\mathcal{N}},\\
S_{RB}  &  \leq\Gamma_{RB}^{\mathcal{M}},\\
\lambda^{\prime}I_{R}  &  \leq\operatorname{Tr}_{B}[Q_{RB}+Q_{RB}^{\dag}],\\
\begin{bmatrix}
P_{RB} & -Q_{RB}^{\dag}\\
-Q_{RB} & S_{RB}
\end{bmatrix}
&  \geq0,\\
\lambda^{\prime}  &  \in\mathbb{R}.
\end{align}
We can rewrite
\begin{align}
\begin{bmatrix}
P_{RB} & -Q_{RB}^{\dag}\\
-Q_{RB} & S_{RB}
\end{bmatrix}
\geq0\quad &  \Longleftrightarrow\quad
\begin{bmatrix}
P_{RB} & Q_{RB}^{\dag}\\
Q_{RB} & S_{RB}
\end{bmatrix}
\geq0\\
&  \Longleftrightarrow\quad
\begin{bmatrix}
P_{RB} & 0\\
0 & S_{RB}
\end{bmatrix}
\geq
\begin{bmatrix}
0 & -Q_{RB}^{\dag}\\
-Q_{RB} & 0
\end{bmatrix}
\end{align}
We then have the simplified condition
\begin{equation}
\begin{bmatrix}
0 & -Q_{RB}^{\dag}\\
-Q_{RB} & 0
\end{bmatrix}
\leq
\begin{bmatrix}
P_{RB} & 0\\
0 & S_{RB}
\end{bmatrix}
\leq
\begin{bmatrix}
\Gamma_{RB}^{\mathcal{N}} & 0\\
0 & \Gamma_{RB}^{\mathcal{M}}
\end{bmatrix}
.
\end{equation}
Since $P_{RB}$ and $S_{RB}$ do not appear in the objective function, we can
set them to their largest value and obtain the following simplification
\begin{equation}
\frac{1}{2}\sup\lambda^{\prime}
\end{equation}
subject to
\begin{equation}
\lambda^{\prime}I_{R}\leq\operatorname{Tr}_{B}[Q_{RB}+Q_{RB}^{\dag}],\quad
\begin{bmatrix}
\Gamma_{RB}^{\mathcal{N}} & Q_{RB}^{\dag}\\
Q_{RB} & \Gamma_{RB}^{\mathcal{M}}
\end{bmatrix}
\geq0,\quad\lambda^{\prime}\in\mathbb{R}
\end{equation}
Since a feasible solution is $\lambda^{\prime}=0$ and $Q_{RB}=0$, it is clear
that we can restrict to $\lambda^{\prime}\geq0$. After a relabeling, this
becomes
\begin{multline}
\frac{1}{2}\sup_{\lambda\geq0,Q_{RB}}\left\{  \lambda:\lambda I_{R}
\leq\operatorname{Tr}_{B}[Q_{RB}+Q_{RB}^{\dag}],\quad
\begin{bmatrix}
\Gamma_{RB}^{\mathcal{N}} & Q_{RB}^{\dag}\\
Q_{RB} & \Gamma_{RB}^{\mathcal{M}}
\end{bmatrix}
\geq0\right\} \\
=\sup_{\lambda\geq0,Q_{RB}}\left\{  \lambda:\lambda I_{R}\leq\operatorname{Re}
[\operatorname{Tr}_{B}[Q_{RB}]],\quad
\begin{bmatrix}
\Gamma_{RB}^{\mathcal{N}} & Q_{RB}^{\dag}\\
Q_{RB} & \Gamma_{RB}^{\mathcal{M}}
\end{bmatrix}
\geq0\right\}  .
\end{multline}
This is equivalent to
\begin{equation}
\sup_{Q_{RB}}\left\{  \lambda_{\min}\left(  \operatorname{Re}
[\operatorname{Tr}_{B}[Q_{RB}]]\right)  :
\begin{bmatrix}
\Gamma_{RB}^{\mathcal{N}} & Q_{RB}^{\dag}\\
Q_{RB} & \Gamma_{RB}^{\mathcal{M}}
\end{bmatrix}
\geq0\right\}  .
\end{equation}
This concludes the proof.
\end{proof}

\end{document}